\documentclass[a4paper, 11pt, french, english]{memoir}
\usepackage{myhdr}
\pdfoutput=1

\date{2022}

\begin{document}

\includepdf{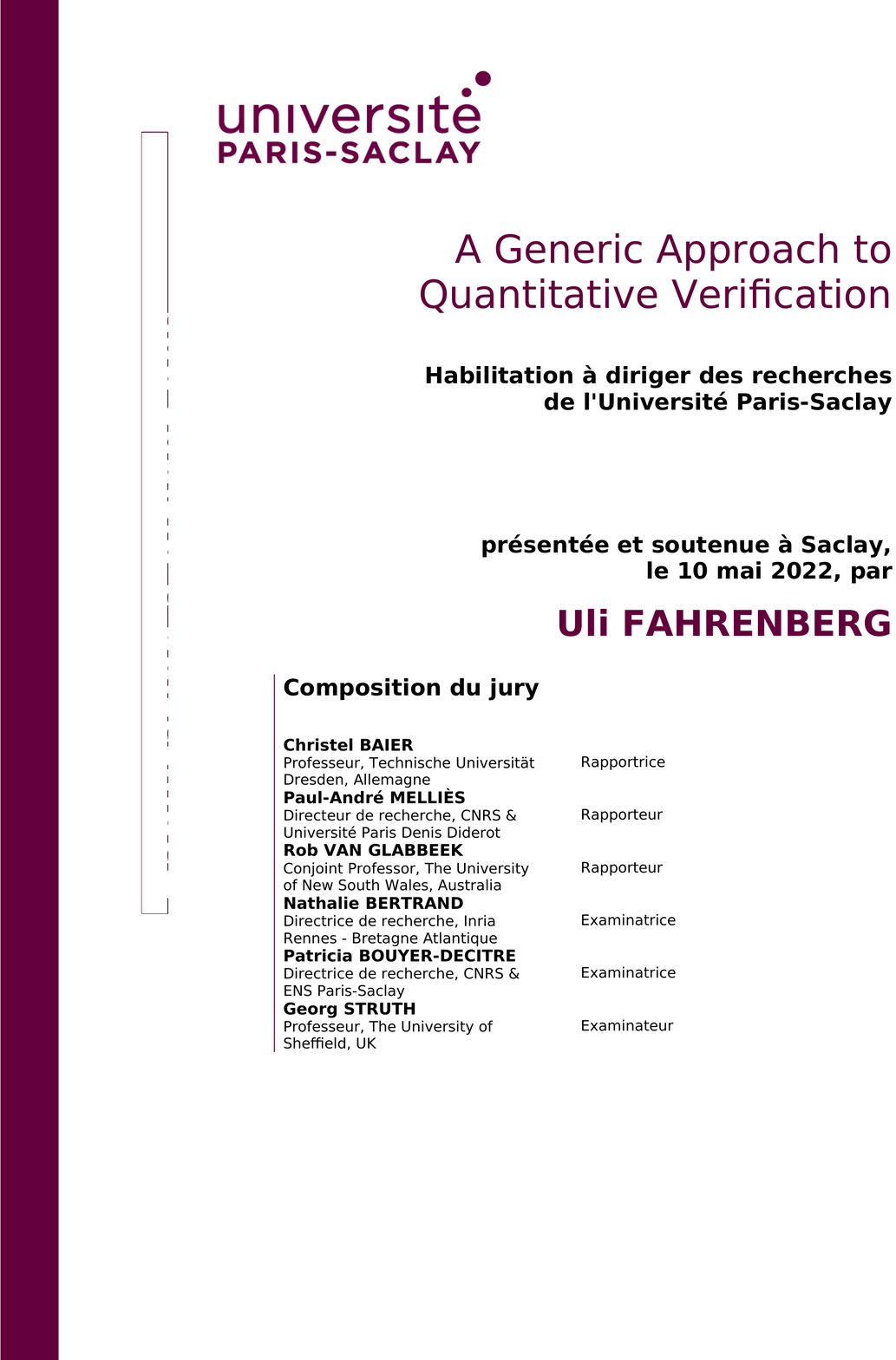}

\thispagestyle{empty}

\hspace*{-3.5cm}%
\begin{minipage}{1.3\linewidth}
  \vspace*{-2cm}%
  \selectlanguage{french}

  \paragraph*{Titre :} Une approche générique à la vérification quantitative

  \medskip

  \paragraph*{Mots clés :} vérification quantitative,
  compositionnalité, incrémentalité, robustesse

  \medskip

  \begin{multicols}{2}
    \paragraph*{Résumé :}

    Ce mémoire porte sur la vérification quantitative, c'est-à-dire la
    vérification des propriétés quantitatives des systèmes
    quantitatifs.  Ces systèmes se retrouvent dans de nombreuses
    applications, et leur vérification quantitative est importante,
    mais aussi assez complexe.  En particulier, étant donné que la
    plupart des systèmes trouvés dans les applications sont plutôt
    larges, il est alors essentiel que les méthodes soient
    compositionnelles et incrémentielles. 

    Afin d'assurer la robustesse de la vérification, nous remplaçons
    les réponses booléennes de la vérification standard par des
    distances.  Selon le contexte de l'application, de nombreux types
    de distances différentes sont utilisées dans la vérification
    quantitative.  Par conséquent, il est nécessaire d'avoir une
    théorie générale des distances de systèmes qui puisse s'abstraire
    des distances concrètes, et de développer une vérification
    quantitative qui est indépendante de la distance.  Nous sommes de
    l'avis que dans une théorie de la vérification quantitative, les
    aspects quantitatifs devraient être traités, tout autant que les
    aspects qualitatifs, comme des éléments d'entrée d'un problème de
    vérification.

    Dans ce travail, nous développons de la sorte une théorie générale
    de la vérification quantitative.  Nous supposons comme entrée une
    distance entre traces, ou exécutions, puis utilisons la théorie
    des jeux à objectifs quantitatifs pour définir des distances entre
    systèmes quantitatifs.  Différentes versions du jeu de
    bisimulation ( quantitatif ) donnent lieu à différents types de
    distances : distance de bisimulation, distance de simulation,
    distance d'équivalence de trace, etc., permettant de construire
    une généralisation quantitative du spectre temps linéaire--temps
    de branchement de van Glabbeek.

    Nous étendons notre théorie générale de la vérification
    quantitative à une théorie des spécifications quantitatives.  Pour
    cela nous utilisons des systèmes de transitions modaux, et nous
    développons les propriétés quantitatives des opérateurs usuels
    pour les théories de spécifications. Tout cela est indépendant de
    la distance concrète entre les traces utilisée.
  \end{multicols}

  \vspace{2cm}

  \selectlanguage{english}

  \paragraph*{Title:} \projecttitle

  \medskip

  \paragraph*{Keywords:} quantitative verification, compositionality,
  incrementality, robustness

  \medskip

  \begin{multicols}{2}
    \paragraph*{Abstract:}

    This thesis is concerned with quantitative verification, that is,
    the verification of quantitative properties of quantitative
    systems.  These systems are found in numerous applications, and
    their quantitative verification is important, but also rather
    challenging.  In particular, given that most systems found in
    applications are rather big, compositionality and incrementality
    of verification methods are essential.

    In order to ensure robustness of verification, we replace the
    Boolean yes-no answers of standard verification with distances.
    Depending on the application context, many different types of
    distances are being employed in quantitative verification.
    Consequently, there is a need for a general theory of system
    distances which abstracts away from the concrete distances and
    develops quantitative verification at a level independent of the
    distance.  It is our view that in a theory of quantitative
    verification, the quantitative aspects should be treated just as
    much as input to a verification problem as the qualitative aspects
    are.

    In this work we develop such a general theory of quantitative
    verification.  We assume as input a distance between traces, or
    executions, and then employ the theory of games with quantitative
    objectives to define distances between quantitative systems.
    Different versions of the quantitative bisimulation game give rise
    to different types of distances, viz.~bisimulation distance,
    simulation distance, trace equivalence distance, etc., enabling us
    to construct a quantitative generalization of van Glabbeek's
    linear-time--branching-time spectrum.

    We also extend our general theory of quantitative verification to
    a theory of quantitative specifications.  For this we use modal
    transition systems, and we develop the quantitative properties of
    the usual operators for behavioral specification theories.  All
    this is independent of the concrete distance between traces which
    is utilized.
  \end{multicols}
\end{minipage}

\frontmatter


\tableofcontents*

\mainmatter

\chapter{Introduction}

This thesis is concerned with quantitative verification, that is, the
verification of quantitative properties of quantitative systems.
These systems are found in numerous applications, and their
quantitative verification is important, but also rather challenging.
In particular, given that most systems found in applications are
rather big, compositionality and incrementality are essential.  That
is, quantitative verification should be applied as much as possible to
subsystems and at as high a level as possible, and then verified
partial specifications should be composed and refined into an
implementation.

Much work has been done in the area of compositional and incremental
design, but robust quantitative frameworks are lacking.  This thesis
presents work published between 2009 and 2020 by the author and
various co-authors which attempts to introduce such a framework.  Much
remains to be done, in particular in applications to real-time and
hybrid systems, but we believe that the foundations laid out here will
be useful in this endeavor.

\section{Motivation}

\subsection{Quantitative Verification}

Motivated by applications in real-time systems, hybrid systems,
embedded systems, and other areas, formal verification has seen a
trend towards modeling and analyzing systems which contain
quantitative information.  Quantitative information can thus be a
variety of things: probabilities, time, tank pressure, energy intake,
\etc

A number of quantitative models have been developed: probabilistic
automata~\cite{DBLP:conf/concur/SegalaL94}; stochastic process
algebras~\cite{book/Hillston96}; timed
automata~\cite{DBLP:journals/tcs/AlurD94}; hybrid
automata~\cite{DBLP:journals/tcs/AlurCHHHNOSY95}; timed variants of
Petri nets~\cite{journals/transcom/MerlinF76, DBLP:conf/apn/Hanisch93};
con\-tinuous-time Markov chains~\cite{book/Stewart94}; \etc Similarly,
there is a number of specification formalisms for expressing
quantitative properties: timed computation tree
logic~\cite{DBLP:journals/iandc/HenzingerNSY94}; probabilistic
computation tree logic~\cite{DBLP:journals/fac/HanssonJ94}; metric
temporal logic~\cite{DBLP:journals/rts/Koymans90}; stochastic
continuous logic~\cite{DBLP:journals/tocl/AzizSSB00}; \etc
Quantitative model checking, the verification of quantitative
properties for quantitative systems, has also seen rapid development:
for probabilistic systems in
PRISM~\cite{DBLP:conf/tacas/KwiatkowskaNP02} and
PEPA~\cite{DBLP:conf/cpe/GilmoreH94}; for real-time systems in
Uppaal~\cite{DBLP:journals/sttt/LarsenPY97},
RED~\cite{DBLP:conf/fm/WangME93}, TAPAAL~\cite{DBLP:conf/atva/BygJS09}
and Romeo~\cite{DBLP:conf/cav/GardeyLMR05}; and for hybrid systems in
HyTech~\cite{DBLP:journals/sttt/HenzingerHW97},
SpaceEx~\cite{DBLP:conf/cav/FrehseGDCRLRGDM11} and
HySAT~\cite{DBLP:journals/fmsd/FranzleH07}, to name but a few.

Quantitative model checking has, however, a problem of
\emph{robustness}.  When the answers to model checking problems are
Boolean---either a system meets its specification or it does not---then
small perturbations in the system's parameters may invalidate the
result.  This means that, from a model checking point of view, small,
perhaps unimportant, deviations in quantities are indistinguishable from
larger ones which may be critical.

As an example, Figure~\ref{fi:tatrain} shows three simple
timed-automaton models of a train crossing, each modeling that once
the gates are closed, some time will pass before the train arrives.
Now assume that the specification of the system is
\begin{equation*}
  \textit{The gates have to be closed 60 seconds before the train
    arrives.}
\end{equation*}
Model $A$ does guarantee this property, hence satisfies the
specification.  Model $B$ only guarantees that the gates are closed 58
seconds before the train arrives, and in model $C$, only one second
may pass between the gates closing and the train.

Neither model $B$ or $C$ satisfy the specification, so this is the
result which a model checker like for example Uppaal would output.
What this does not tell us, however, is that model $C$ is dangerously
far away from the specification, whereas model $B$ only violates it
slightly and may be acceptable given other engineering constraints, or
may be more easily amenable to satisfy the specification than model
$C$.

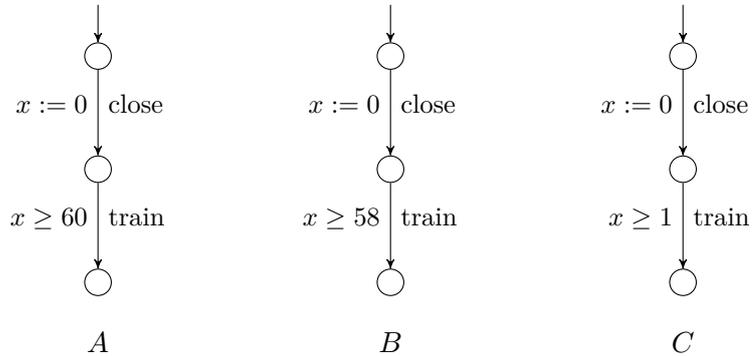
\begin{figure}[tbp]
  \centering
  \begin{tikzpicture}[->,>=stealth',auto,initial text=]
    \tikzstyle{every node}=[font=\small]
    \tikzstyle{every state}=[inner sep=.5mm,minimum size=3.5mm]
    \begin{scope}
      \node[state,initial above] (0) at (0,0) {};
      \node[state] (1) at (0,-1.5) {};
      \node[state] (2) at (0,-3) {};
      \path (0) edge node [left,anchor=base east] {$x:= 0$} node
      [right,anchor=base west] {close} (1);
      \path (1) edge node [left,anchor=base east] {$x\ge 60$} node
      [right,anchor=base west] {train} (2);
      \node[font=\normalsize] at (0,-3.8) {$A$};
    \end{scope}
    \begin{scope}[xshift=10em]
      \node[state,initial above] (0) at (0,0) {};
      \node[state] (1) at (0,-1.5) {};
      \node[state] (2) at (0,-3) {};
      \path (0) edge node [left,anchor=base east] {$x:= 0$} node
      [right,anchor=base west] {close} (1);
      \path (1) edge node [left,anchor=base east] {$x\ge 58$} node
      [right,anchor=base west] {train} (2);
      \node[font=\normalsize] at (0,-3.8) {$B$};
    \end{scope}
    \begin{scope}[xshift=20em]
      \node[state,initial above] (0) at (0,0) {};
      \node[state] (1) at (0,-1.5) {};
      \node[state] (2) at (0,-3) {};
      \path (0) edge node [left,anchor=base east] {$x:= 0$} node
      [right,anchor=base west] {close} (1);
      \path (1) edge node [left,anchor=base east] {$x\ge 1$} node
      [right,anchor=base west] {train} (2);
      \node[font=\normalsize] at (0,-3.8) {$C$};
    \end{scope}
  \end{tikzpicture}
  \caption{%
    \label{fi:tatrain}
    Three timed automata modeling a train crossing.
  }
\end{figure}

In order to address the robustness problem, our approach is to replace
the Boolean yes-no answers of standard verification with
\emph{distances}.  That is, the Boolean co-domain of model checking is
replaced by the non-negative real numbers.  In this setting, the
Boolean \texttt{true} corresponds to a distance of zero, and
\texttt{false} corresponds to any non-zero number, so that
quantitative model checking can now tell us not only that a
specification is violated, but also \emph{how much} it is violated, or
\emph{how far} the system is from corresponding to its specification.

In the example of Figure~\ref{fi:tatrain} and for a simple definition of
system distances, the distance from $A$ to our specification would be $0$,
whereas the distances from $B$ and $C$ to the specification would be $2$
and $59$, respectively.  The precise interpretation of distance values
will be application-dependent; but in any case, it is clear that $C$
is much farther away from the specification than $B$ is.

\subsection{Specification Theories}

One of the major current challenges to rigorous design of software
systems is that these systems are becoming increasingly complex and
difficult to reason about~\cite{Sifakis11}.  As an example, an
integrated communication system in a modern airplane can have more than
$10^{ 900}$ distinct states~\cite{DBLP:conf/forte/BasuBBCDL10}, and
state-of-the-art tools offer no possibility to reason about, and model
check, the system as a whole.
One promising approach to overcome such problems is the one of
\emph{compositional and incremental design}.  Here the reasoning is
done as much as possible at higher \emph{specification} levels rather
than with \emph{implementations}; partial specifications are proven
correct and then composed and refined until one arrives at an
implementation model.  Practical experience indicates that this is a
viable approach~\cite{COMBEST,SPEEDS}.

Specifications of system requirements are high-level finite abstractions
of possibly infinite sets of implementations. A model of a system is
considered an implementation of a given specification if the behavior
defined by the implementation is implied by the description provided by
the specification. 

Any practical specification formalism comes equipped with a number of
operations which permit compositional and incremental reasoning.  The
first of these is a \emph{refinement} relation which allows to
successively distill specifications into more detailed ones and
eventually into implementations.  In an implementation, all optional
behavior defined in the specification has been decided upon in
compliance with the specification.
Also needed is an operation of \emph{logical conjunction} which allows
to combine specifications so that the systems which refine the
conjunction of two specifications are precisely the ones which satisfy
both of them.  Refinement and conjunction together permit incremental
reasoning as specifications are successively refined and conjoined.

For compositional reasoning, one needs another operation of
\emph{structural composition} which allows to infer specifications
from sub-specifications of independent requirements, mimicking at the
implementation level for example the interaction of components in a
distributed system.  A partial inverse of this operation is given by a
\emph{quotient} operation which allows to synthesize a specification
of missing components from an overall specification and an
implementation which realizes a part of that specification.

Over the years, there have been a series of advances on specification
theories~\cite{inbook/natosec/AlfaroH95, ChakrabartiAHM02,
  DBLP:conf/hybrid/DavidLLNW10, thesis/Delahaye10,
  journals/cwi/LynchT89, thesis/Nyman08, thesis/Thrane11}.  The
predominant approaches are based on modal logics and process algebras
but have the drawback that they cannot naturally embed both logical
and structural composition within the same
formalism~\cite{DBLP:conf/avmfss/Larsen89}.  Hence such formalisms do
not permit to reason incrementally through refinement.

In order to leverage these problems, the concept of \emph{modal
  transition systems} was
introduced~\cite{DBLP:conf/avmfss/Larsen89}. In short, modal
transition systems are labeled transition systems equipped with two
types of transitions: \textit{must}~transitions which are mandatory
for any implementation, and \textit{may}~transitions which are
optional for implementations.  It is well established that modal
transition systems match all the requirements of a reasonable
specification theory,
and much progress has been made in this area, see for example
\cite{thesis/Nyman08, DBLP:conf/fmoods/GrulerLS08,
  DBLP:conf/concur/GodefroidHJ01, DBLP:conf/vmcai/GrumbergLLS05}
or~\cite{DBLP:journals/eatcs/AntonikHLNW08} for an overview.  Also,
practical experience shows that the formalism is expressive enough to
handle complex industrial problems~\cite{COMBEST,SPEEDS}.

As an example, consider the modal transition system shown in
Figure~\ref{fi:mtsintro} which models the requirements of a simple email
system in which emails are first received and then delivered.  Before
delivering the email, the system may check or process the email, for
example for en- or decryption, filtering of spam emails, or generating
automatic answers using an auto-reply feature (see
also~\cite{DBLP:conf/fiw/Hall00}).  \textit{Must} transitions,
representing obligatory behavior, are drawn as solid arrows, whereas
\textit{may} transitions, modeling optional behavior, are shown as
dashed arrows: hence any implementation of this email system
specification \emph{must} be able to receive and deliver email, and it
\emph{may} also be able to check arriving email before delivering it.
No other behavior is allowed.

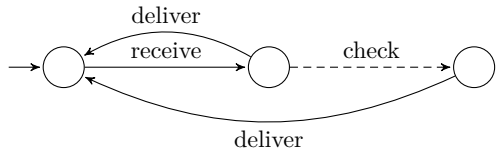
\begin{figure}[tp]
  \centering
  \begin{tikzpicture}[->,>=stealth',shorten >=1pt,auto,node
    distance=2.0cm,initial text=,scale=0.9,transform shape]
    \tikzstyle{every node}=[font=\small] \tikzstyle{every
      state}=[fill=white,shape=circle,inner sep=.5mm,minimum size=6mm]
    \node[state,initial] (s0) at (0,0) {};
    \node[state] (s1) at (3,0) {};
    \node[state] (s2) at (6,0) {};
    \path (s0) edge [solid] node [above,sloped] {receive} (s1);
    \path (s1) edge [solid,bend right] node [above,sloped] {deliver} (s0);
    \path (s1) edge [densely dashed] node [above] {check} (s2);
    \path (s2) edge [solid,bend left=25] node [below,sloped] {deliver} (s0);
  \end{tikzpicture} 
  \caption{Modal transition system modeling a simple email system, with
    an optional behavior: Once an email is received it may \eg~be
    scanned for containing viruses, or automatically decrypted, before
    it is delivered to the receiver.}
  \label{fi:mtsintro}
\end{figure}

Implementations can also be represented within the modal transition
system formalism, simply as specifications without \textit{may}
transitions.  Here, any implementation choice has been resolved, so
that implementations are (isomorphic to) plain labeled transition
systems.  Formally, for a labeled transition system to be an
implementation of a given specification, we require that the states of
the two objects are related by a refinement relation with the property
that all behavior required by the specification has been implemented,
and that any implementation behavior is permitted in the
specification.  Figure~\ref{fi:intro_impl} shows an implementation of
our email specification with two different checks, leading to distinct
processing states.

\begin{figure}[tp]
  \centering
  \begin{tikzpicture}[->,>=stealth',shorten >=1pt,auto,node
    distance=2.0cm,initial text=,scale=0.9,transform shape]
    \tikzstyle{every node}=[font=\small] \tikzstyle{every
      state}=[fill=white,shape=circle,inner sep=.5mm,minimum size=6mm]
    \node[state,initial] (s0) at (0,0) {};
    \node[state] (s1) at (3,0) {};
    \node[state] (s2) at (6,-.7) {};
    \node[state] (s3) at (6,.7) {};
    \path (s0) edge [solid] node [above,sloped] {receive} (s1);
    \path (s1) edge [solid,bend right] node [near start,above]
    {deliver} (s0);
    \path (s1) edge [solid] node [above,sloped] {check} (s2);
    \path (s1) edge [solid] node [above,sloped] {check} (s3);
    \path (s2) edge [solid,bend left=25] node [below] {deliver} (s0);
    \path (s3) edge [solid,bend right=35] node [above] {deliver} (s0);
  \end{tikzpicture} 
  \caption{An implementation of the simple email system in
    Figure~\ref{fi:mtsintro} in which we explicitly model two distinct
    types of email pre-processing.}
  \label{fi:intro_impl}
\end{figure}
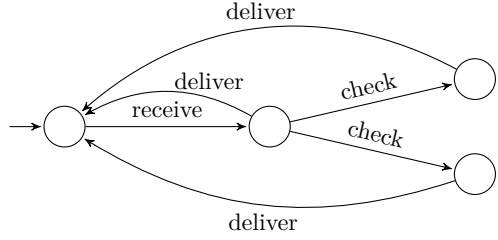

\subsection{Quantitative Specification Theories}

In recent work~\cite{DBLP:journals/jlp/JuhlLS12,
  DBLP:journals/mscs/BauerJLLS12, DBLP:conf/tase/BauerJLSL12,
  DBLP:conf/lpar/BenesKLMS12}, modal transition systems have been
extended by adding richer information to the usual discrete label set
of transition systems, permitting to reason about \emph{quantitative}
aspects of models and specifications.  These quantitative labels can
be used to model and analyze for example timing
behavior~\cite{DBLP:conf/formats/HenzingerMP05,
  DBLP:conf/hybrid/DavidLLNW10}, resource
usage~\cite{DBLP:journals/fmsd/RasmussenLS06,
  DBLP:conf/tase/BauerJLSL12}, or energy
consumption~\cite{DBLP:journals/cacm/BouyerFLM11,
  DBLP:conf/ictac/FahrenbergJLS11}.

In particular, \cite{DBLP:journals/jlp/JuhlLS12}~extends modal
transition systems with integer intervals and introduces corresponding
extensions of the above operations which observe the added quantitative
information, and \cite{DBLP:journals/mscs/BauerJLLS12}~generalizes this
theory to general \emph{structured labels}.  Both theories are, however,
\emph{fragile} in the sense that they rely on Boolean notions of
satisfaction and refinement: as refinement either holds or does not,
they are unable to \emph{quantify} the impact of small variations in
quantities.

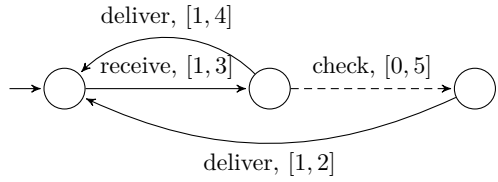
\begin{figure}[tp]
  \centering
  \begin{tikzpicture}[->,>=stealth',shorten >=1pt,auto,node
    distance=2.0cm,initial text=,scale=0.9,transform shape]
    \tikzstyle{every node}=[font=\small] \tikzstyle{every
      state}=[fill=white,shape=circle,inner sep=.5mm,minimum size=6mm]
    \node[state,initial] (s0) at (0,0) {};
    \node[state] (s1) at (3,0) {};
    \node[state] (s2) at (6,0) {};
    \path (s0) edge [solid] node [above,sloped] {receive, $[1,3]$} (s1);
    \path (s1) edge [solid,bend right=45] node [above,sloped] {deliver,
      $[1,4]$} (s0); 
    \path (s1) edge [densely dashed] node [above] {check, $[0,5]$} (s2);
    \path (s2) edge [solid,bend left=25] node [below,sloped] {deliver,
      $[1,2]$} (s0); 
  \end{tikzpicture} 
  \caption{%
    \label{fi:mtsquantities}
    Specification of a simple email system, with integer intervals
    modeling time constraints for performing the corresponding actions.}
\end{figure}

An example of a quantitative specification is shown in
Figure~\ref{fi:mtsquantities}.  The intuition is that any concrete
implementation \emph{must} be able to receive and deliver email,
within one to three and one to four time units, respectively; but it
also \emph{may} be able to check incoming email, \eg~for viruses,
before delivering it.  No other behavior is permitted.

\begin{figure}[tp]
  \centering
\subbottom[\label{fi:introimpl.1}]{
  \begin{tikzpicture}[->,>=stealth',shorten >=1pt,auto,node
    distance=2.0cm,initial text=,scale=0.9,transform shape]
    \tikzstyle{every node}=[font=\small] \tikzstyle{every
      state}=[fill=white,shape=circle,inner sep=.5mm,minimum size=6mm]
    \path[use as bounding box] (-1,-.5) rectangle (5.8,1.5); 
    \node[state,initial] (s0) at (0,0) {};
    \node[state] (s1) at (3,0) {};
    \path (s0) edge [solid] node [above,sloped] {receive, $2$} (s1);
    \path (s1) edge [solid,bend right=45] node [above,sloped] {deliver, $3$} (s0);
    \path (s1) edge [solid,loop right] node [right] {check, $1$} (s1);
  \end{tikzpicture}
  }
\subbottom[\label{fi:introimpl.2}]{
  \begin{tikzpicture}[->,>=stealth',shorten >=1pt,auto,node
    distance=2.0cm,initial text=,scale=0.9,transform shape]
    \tikzstyle{every node}=[font=\small] \tikzstyle{every
      state}=[fill=white,shape=circle,inner sep=.5mm,minimum size=6mm]
    \path[use as bounding box] (-1,-0.5) rectangle (4,1.5); 
    \node[state,initial] (s0) at (0,0) {};
    \node[state] (s1) at (3,0) {};
    \path (s0) edge [solid] node [above,sloped] {receive, $4$} (s1);
    \path (s1) edge [solid,bend right=45] node [above,sloped] {deliver, $3$} (s0);
  \end{tikzpicture}}
\subbottom[\label{fi:introimpl.3}]{
  \begin{tikzpicture}[->,>=stealth',shorten >=1pt,auto,node
    distance=2.0cm,initial text=,scale=0.9,transform shape]
    \tikzstyle{every node}=[font=\small] \tikzstyle{every
      state}=[fill=white,shape=circle,inner sep=.5mm,minimum size=6mm]
    \path[use as bounding box] (-1,-1.2) rectangle (5.8,1.5); 
    \node[state,initial] (s0) at (0,0) {};
    \node[state] (s1) at (3,0) {};
    \node[state] (s2) at (5,0) {};
    \path (s0) edge [solid] node [above,sloped] {receive, $3$} (s1);
    \path (s1) edge [solid,bend right=45] node [above,sloped] {deliver, $3$} (s0);
    \path (s1) edge [solid] node [above] {check, $1$} (s2);
    \path (s2) edge [solid,bend left=25] node [below,sloped] {deliver, $3$} (s0);
  \end{tikzpicture}}
  \subbottom[\label{fi:introimpl.4}]{
  \begin{tikzpicture}[->,>=stealth',shorten >=1pt,auto,node
    distance=2.0cm,initial text=,scale=0.9,transform shape]
    \tikzstyle{every node}=[font=\small] \tikzstyle{every
      state}=[fill=white,shape=circle,inner sep=.5mm,minimum size=6mm]
    \path[use as bounding box] (-1,-1.2) rectangle (4,1.5); 
    \node[state,initial] (s0) at (0,0) {};
    \node[state] (s1) at (3,0) {};
    \path (s0) edge [solid] node [above,sloped] {receive, $2$} (s1);
    \path (s1) edge [solid,bend right=45] node [above,sloped] {deliver, $3$} (s0);
  \end{tikzpicture}}
\caption{%
  \label{fi:introimpl}
  Four implementations of the simple email system in
  Figure~\ref{fi:mtsquantities}.}
\end{figure}
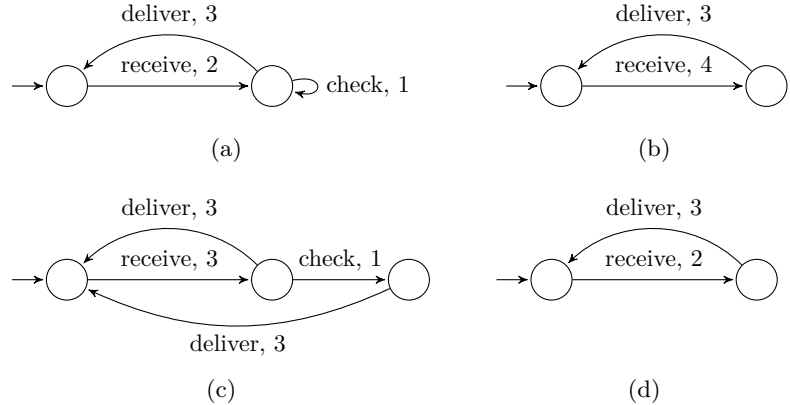

Figure~\ref{fi:introimpl} shows four different implementation
candidates for the specification of Figure~\ref{fi:mtsquantities}.  The
first candidate, in Figure~\ref{fi:introimpl.1}, has an error in the
discrete structure: after receiving an email, it may check the email
indefinitely.  Hence it does not satisfy the specification.  The
second candidate, in Figure~\ref{fi:introimpl.2}, is also problematic:
not implementing the checking part of the specification is entirely
permissible, but it takes too long to receive email.  Thus, if the
timing constraints are abstracted away, it is a perfectly good
implementation; but the quantitative timing constraints are off.  The
implementation candidate in Figure~\ref{fi:introimpl.3} has similar
problems, as it takes too long to deliver emails after checking them.
The transition system in Figure~\ref{fi:introimpl.4} is, finally, a true
implementation of the specification.

An important observation is, now, that even though the systems in
Figs.\ \ref{fi:introimpl.2} and~\ref{fi:introimpl.3} strictly are not
implementations of the email system specification, they conform much
better to it than the system in Figure~\ref{fi:introimpl.1}.
Intuitively, they ``almost'' comply with the specification; given some
other engineering constraints, they might indeed be considered ``good
enough'' compared to the specification.  It is, then, this ``almost''
and ``good enough'' which we shall attempt to formalize in this work.

Our point of view is, more generally speaking, that \emph{any}
quantitative specification formalism falls short with a Boolean notion
of satisfaction and refinement.  If the specification formalism is
intended to model quantitative properties, then it is of little use to
know that a proposed implementation does not precisely adhere to a
specification; much more useful information is obtained by knowing
\emph{how well} it implements the specification, or \emph{how far} it
is deviating.  Of course, the answer to this ``how far'' question
might be $\infty$, due to discrete errors as in
Figure~\ref{fi:introimpl.1}; but in case it is finite, useful knowledge
may be gained, for example as to how much more implementation effort
is needed, or whether one can satisfy oneself with this slightly
imperfect implementation.  Our approach will hence again be to replace
satisfaction and refinement relations by satisfaction and refinement
\emph{distances}.

\subsection{Related Work}

The distance-based approach to quantitative verification has been
developed the furthest for probabilistic and stochastic systems.
Panangaden and Desharnais \etal have worked with distances for Markov
processes in~\cite{DBLP:journals/tcs/DesharnaisGJP04,
  DBLP:conf/uai/FernsPP05, DBLP:conf/concur/DesharnaisGJP99,
  DBLP:conf/qest/DesharnaisLT08, DBLP:conf/lics/DesharnaisJGP02,
  DBLP:books/daglib/0023536, DBLP:conf/mfcs/LarsenMP12,
  DBLP:conf/tacas/BacciBLM13} and other papers, and van~Breugel and
Worrell \etal have developed distances for probabilistic transition
systems in~\cite{DBLP:journals/tcs/BreugelW05,
  DBLP:conf/concur/BreugelW01, DBLP:journals/tcs/BreugelW06}.
De~Alfaro and Stoelinga \etal have worked on distances between
probabilistic systems and specifications
in~\cite{DBLP:journals/tcs/AlfaroFHMS05, DBLP:conf/icalp/AlfaroHM03,
  DBLP:conf/lics/AlfaroMRS07, DBLP:conf/qest/ChatterjeeAFHMS06,
  DBLP:journals/corr/abs-0809-4326, DBLP:journals/lmcs/AlfaroMRS08,
  DBLP:conf/icalp/AlfaroFS04} and other papers.

For real-time and hybrid systems, some explicit work on distances is
available in~\cite{DBLP:conf/formats/HenzingerMP05,
  DBLP:journals/corr/abs-1011-0688, DBLP:conf/formats/QueselFD11}.
Otherwise, distances have been used in approaches to robust
verification~\cite{DBLP:conf/formats/LarsenLTW11,
  DBLP:conf/concur/BouyerLMST11}, and Girard \etal have developed a
theory of approximate bisimulation for robust
control~\cite{DBLP:conf/hybrid/ZhengG09, DBLP:journals/tac/GirardP07}.

Also general work on distances for quantitative systems where the
precise meaning of the quantities remains unspecified has been done.
Van~Breugel has developed a general theory of behavioral
pseudometrics~\cite{DBLP:journals/tcs/Breugel01,
  DBLP:journals/tcs/BonsangueBR98, journals/anyas/Breugel96,
  DBLP:conf/concur/Breugel05}.  Henzinger \etal have employed
distances in a software engineering context
in~\cite{DBLP:journals/tcs/CernyHR12, DBLP:conf/concur/CernyHR10} and
for abstraction refinement and synthesis
in~\cite{DBLP:conf/popl/CernyHR13, DBLP:conf/emsoft/CernyH11,
  DBLP:conf/cav/CernyCHRS11, DBLP:journals/tcs/CernyCHR14,
  DBLP:conf/qest/ChatterjeeAFHMS06}.

Common to all the above distance-based approaches is that they
introduce distances between systems, or between systems and
specifications, and then employ these for approximate or quantitative
verification.  However, depending on the application context, a
plethora of different distances are being used, motivating the need
for a general theory.  This is a point of view which is also argued
in~\cite{DBLP:conf/popl/CernyHR13, DBLP:conf/qest/ChatterjeeAFHMS06}.

To be more specific, most of the above approaches can be classified
according to the way they measure distances between \emph{executions},
or system traces.  The perhaps easiest such way is the
\emph{point-wise} distance, which measures the greatest individual
distance between corresponding points in the traces.  Theory for this
specific distance has been developed
in~\cite{DBLP:journals/tse/AlfaroFS09,
  DBLP:journals/tcs/AlfaroFHMS05, DBLP:conf/icalp/AlfaroFS04,
  DBLP:conf/concur/BouyerLMST11} and other papers.  Sometimes
\emph{discounting} is applied to diminish the influence of individual
distances far in the future, for example
in~\cite{DBLP:journals/tse/AlfaroFS09, DBLP:journals/tcs/AlfaroFHMS05,
  DBLP:conf/icalp/AlfaroFS04}.

Another distance which has been used is the \emph{accumulating} one,
which sums individual distances along executions.  Two major types
have been considered here: the \emph{discounted} accumulating distance
\eg~in~\cite{DBLP:conf/concur/CernyHR10, DBLP:conf/icalp/AlfaroHM03,
  DBLP:conf/emsoft/AlurT11, DBLP:journals/tcs/AlfaroFHMS05} and the
\emph{limit-average} accumulating distance
\eg~in~\cite{DBLP:conf/concur/CernyHR10, DBLP:conf/emsoft/AlurT11}.
Both are well-known from the theory of discounted and mean-payoff
games~\cite{journals/gameth/EhrenfeuchtM79,
  DBLP:journals/tcs/ZwickP96}.

For real-time systems, a useful distance is the \emph{maximum-lead}
distance of~\cite{DBLP:conf/formats/HenzingerMP05} which measures the
maximum difference between accumulated time delays along traces.  For
hybrid systems, things are more complicated, as distances between
hybrid traces have to take into account both spatial and timing
differences, see for example~\cite{DBLP:conf/formats/QueselFD11,
  DBLP:conf/cdc/Girard10, DBLP:conf/hybrid/ZhengG09,
  DBLP:journals/tac/GirardP07}.

\subsection{A General Theory of Quantitative Verification}

Depending on the application context, many different types of
distances are being employed in quantitative verification.
Consequently, there is a need for a general theory of system distances
which abstracts away from the concrete distances and develops
quantitative verification at a level independent of the distance.  It
is our view that in a theory of quantitative verification, the
quantitative aspects should be treated just as much as input to a
verification problem as the qualitative aspects are.

In this work we develop such a general theory of quantitative
verification.  We assume as input a distance between traces, or
executions, and then employ the theory of games with quantitative
objectives to define distances between quantitative systems.
Different versions of the (quantitative) bisimulation game give rise
to different types of distances, \viz~bisimulation distance,
simulation distance, trace equivalence distance, \etc, enabling us to
construct a quantitative generalization of the
linear-time--branching-time spectrum.

We also extend our general theory of quantitative verification to a
theory of quantitative specifications.  For this we use modal
transition systems, and we develop the quantitative properties of the
usual operators for behavioral specification theories.  All this is
independent of the concrete distance between traces which is utilized.

\section{Contributions}

In the following chapters we present work based on eight papers,
published between 2009 and 2020 by the author of this thesis with
different co-authors, on quantitative verification and quantitative
specification theories.  The first three, Chapters~\ref{ch:wtsjlap}
to~\ref{ch:simdistax}, are each concerned with properties of three
specific system distances: the point-wise distance, the discounted
accumulating distance, and the maximum-lead distance.  The next
Chapter~\ref{ch:qltbt} develops our general theory of quantitative
verification and shows basic properties.
Chapters~\ref{ch:weightedmodal} and~\ref{ch:wm2} then extend this
theory to specification theories, first for the discounted
accumulating distance in Chapter~\ref{ch:weightedmodal} and then for
the general setting in Chapter~\ref{ch:wm2}.  In Chapter~\ref{ch:dmts}
we take a break from the quantitative setting in order to introduce an
extension of modal transition systems and show that the so-obtained
specification theory is closely related to other popular specification
formalisms.  The final Chapter~\ref{ch:dmts2} extends these results to
general quantitative and develops their properties.

Compared to their sources, all chapters have been heavily redacted in
order to correct errors, unify notation, and smoothen the
presentation.  Any remaining errors are the sole responsibility of the
author of this thesis.

\subsection{Geometric Preliminaries}

Before we can give an overview of our contributions, we recall a few
standard notions from geometry and topology which we will use
throughout.  Let $\Realnn\cup\{ \infty\}$ denote the extended
non-negative reals.

A \emph{hemimetric} on a set $X$ is a function
$d: X\times X\to \Realnn\cup\{ \infty\}$ which satisfies $d( x, x)= 0$
and $d( x, y)+ d( y, z)\ge d( x, z)$ (the \emph{triangle inequality})
for all $x, y, z\in X$.  The hemimetric is said to be \emph{symmetric}
if also $d( x, y)= d( y, x)$ for all $x, y\in X$; it is said to be
\emph{separating} if $d( x, y)= 0$ implies $x= y$.

A symmetric hemimetric is generally called a \emph{pseudometric}, and
a hemimetric which is both symmetric and separating is simply a
\emph{metric}.  The tuple $( X, d)$ is called a (hemi/pseudo)metric
\emph{space}.

Note that our hemimetrics are \emph{extended} in that they can take the
value $\infty$.  This is convenient for several reasons,
\cf~\cite{journals/rsmfm/Lawvere73}, one of them being that it allows
for a disjoint union, or coproduct, of hemimetric spaces: the disjoint
union of $( X_1, d_1)$ and $( X_2, d_2)$ is the hemimetric space $( X_1,
d_1)\cupplus( X_1, d_2)=( X_1\cupplus X_2, d)$ where points from
different components are infinitely far away from each other, \ie~with
$d$ defined by
\begin{equation*}
  d( x, y)=
  \begin{cases}
    d_1( x, y) &\text{if } x, y\in X_1, \\
    d_2( x, y) &\text{if } x, y\in X_2, \\
    \infty &\text{otherwise}.
  \end{cases}
\end{equation*}
The \emph{product} of two hemimetric spaces $( X_1, d_1)$ and $( X_2,
d_2)$ is the hemimetric space $( X_1, d_1)\times( X_2, d_2)=( X_1\times
X_2, d)$ with $d$ given by $d(( x_1, x_2),( y_1, y_2))= \max( d_1( x_1,
y_1), d_2( x_2, y_2))$.

The \emph{symmetrization} of a hemimetric $d$ on $X$ is the symmetric
hemimetric $\bar d: X\times X\to \Realnn\cup\{ \infty\}$ defined by
$\bar d( x, y)= \max( d( x, y), d( y, x))$; this is the smallest among
all pseudometrics $d'$ on $X$ for which $d\le d'$.  The
\emph{topology} generated by a hemimetric $d$ on $X$ is defined to be
the same as the one generated by its symmetrization $\bar d$; it has
as open sets all unions of open balls
$B( x; r)=\{ y\in X\mid \bar d( x, y)< r\}$, for $x\in X$ and $r> 0$.

A continuous function $f: X \to X$ on a pseudometric space $( X, d)$
is called a \emph{contraction} if there exists $0\le \alpha< 1$ (its
\emph{Lipschitz constant}) such that $d(f(x),f(y)) \le \alpha d(x,y)$
for all $x,y\in X$.

Two pseudometrics $d_1$, $d_2$ on $X$ are said to be
\begin{itemize}
\item \emph{topologically equivalent} provided that for all $x\in X$
  and all $\epsilon\in \Realp$, there exists $\delta\in \Realp$ such
  that $d_1( x, y)< \delta$ implies $d_2( x, y)< \epsilon$ and
  $d_2( x, y)< \delta$ implies $d_1( x, y)< \epsilon$ for all
  $y\in X$,
\item \emph{Lipschitz equivalent} if there exist $m, M\in \Realp$ such
  that $m d_1( x, y)\le d_2( x, y)\le M d_1( x, y)$ for all
  $x, y\in X$.
\end{itemize}
Hemimetrics are topologically or Lipschitz equivalent if their
symmetrizations are.

Topological equivalence is the same as asking the identity function
$\id:( X, d_1)\to( X, d_2)$ to be a homeomorphism, and Lipschitz
equivalence implies topological equivalence.

Topological equivalence of $d_1$ and $d_2$ is also the same as
requiring the topologies generated by $d_1$ and $d_2$ to coincide.
Topological equivalence hence preserves topological notions such as
convergence of sequences: If a sequence $( x_j)$ of points in $X$
converges in one pseudometric, then it also converges in the other.
As a consequence, topological equivalence of hemimetrics $d_1$ and
$d_2$ implies that for all $x, y\in X$, $d_1( x, y)= 0$ if, and only
if, $d_2( x, y)= 0$.

Topological equivalence is the weakest of the common notions of
equivalence for metrics; it does not preserve geometric properties
such as distances or angles.  We are hence mainly interested in
topological equivalence as a tool for showing \emph{negative}
properties; we will later prove a number of results on topological
\emph{inequivalence} of hemimetrics which imply that any other
reasonable metric equivalence, such as Lipschitz equivalence, also
fails for these cases.

The \emph{Hausdorff hemimetric} associated with a hemimetric $d: X\times
X\to \Realnn\cup\{ \infty\}$ is the function $d^\textup{H}: 2^X\times
2^X\to \Realnn\cup\{ \infty\}$ given for subsets $A, B\subseteq X$ by
\begin{equation*}
  d^\textup{H}( A, B)= \adjustlimits \sup_{ x\in A} \inf_{ y\in B} d( x, y).
\end{equation*}
This is a well-known construction for metric spaces,
\cf~\cite{book/Munkres00, book/AliprantisB07}; there it is usually
symmetrized and defined only for \emph{closed} subsets, in which case
it is a metric.  The following alternative formulation follows
straight from the definition:

\begin{proposition}
  \label{wtsjlap.prop:haus}
  For a hemimetric $d$ on $X$, $A, B\subseteq X$, and $\epsilon\in
  \Realp$, we have $d( A, B)\le \epsilon$ if and only if for any $x\in
  A$ there exists $y\in B$ for which \mbox{$d( x, y)\le \epsilon$}.
\end{proposition}

A sequence $(x_j)$ in a metric space $X$ is a \emph{Cauchy sequence}
if it holds that for all $\epsilon > 0$ there exist $N\in\Nat$ such
that $d(x_m,x_n) < \epsilon$ for all $n,m\geq N$. $X$ is said to be
\emph{complete} if every Cauchy sequence in $X$ converges in $X$.

Finally, we recall the Banach fixed-point theorem: Any contraction on
a complete metric space has precisely one fixed point.

\subsection{Chapter~\ref{ch:wtsjlap}, ``Quantitative Analysis of
  Weighted Transition Systems''}

In Chapter~\ref{ch:wtsjlap} we introduce the point-wise, accumulating
and maximum-lead trace distances, all in a discounted version which
allows to diminish the influence of future differences.  In a notation
which is simpler than the one used in Chapter~\ref{ch:wtsjlap} and
follows the one of later chapters, these are given as follows.  Let
$\KK$ be a set of symbols together with an extended metric
$d_\KK: \KK\times \KK\to \Realnn\cup\{ \infty\}$.  A \emph{trace}
$\sigma=( \sigma_0, \sigma_1,\dotsc)$ is an infinite sequence of
symbols in $\KK$.  Let $\lambda\in \Realnn$ with $\lambda\le 1$ be a
\emph{discounting factor} and $\sigma$ and $\tau$ traces.
\begin{itemize}
\item The \emph{point-wise trace distance} between $\sigma$ and $\tau$ is
  \begin{equation*}
    \pwtd{ \sigma, \tau}= \sup_{ i\ge 0} \lambda^i d_\KK( \sigma_i, \tau_i)\,.
  \end{equation*}
\item The \emph{accumulating trace distance} between $\sigma$ and $\tau$ is
  \begin{equation*}
    \awtd{ \sigma, \tau}= \sum_{ i\ge 0} \lambda^i d_\KK( \sigma_i, \tau_i)\,.
  \end{equation*}
\item The \emph{maximum-lead trace distance} between $\sigma$ and $\tau$ is
  \begin{equation*}
    \mltd{ \sigma, \tau}= \sup_{ i\ge 0} \lambda^i\Bigl| \sum_{ 0\le j\le i}(
  \sigma_i- \tau_i)\Bigr|\,.
  \end{equation*}
\end{itemize}
Note that the definition of the last distance requires extra structure
of addition and subtraction on $\KK$; generally this is only used for
$\KK= \Int$ or $\KK= \Real$.

We then use these trace distances to define point-wise, accumulating
and maximum-lead \emph{linear} distances between states in weighted
transition systems.  If $d^\textup{\textsf{T}}$ is any of the above
trace distances, then the linear distance between two states $s$ and
$t$ of a transition system is defined to be
\begin{equation*}
  \distgeneric{ s, t}=
  \adjustlimits \sup_{ \sigma\in \tracesfrom s} \inf_{ \tau\in
    \tracesfrom t} \tdgeneric{ \sigma, \tau}\,,
\end{equation*}
where $\tracesfrom s$ denotes the set of traces emanating from $s$,
similarly for $\tracesfrom t$.  This is thus the Hausdorff distance
from $\tracesfrom s$ to $\tracesfrom t$; note that the definition is
independent of which particular trace distance is used.

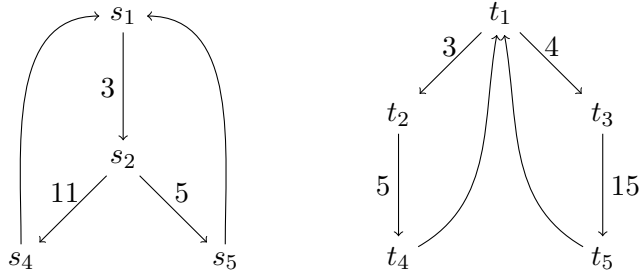
\begin{figure}[tbp]
  \centering
  \begin{tikzpicture}[->, shorten >=1pt, auto, initial text=]
    \begin{scope}[outer sep=1pt,minimum size=5pt,inner sep=2pt, node
      distance=1.9cm]
      \node (v0) {$s_1$}; 
      \node (v1) [below of=v0] {$s_2$};
      \node (v2) [below right of=v1] {$s_5$}; 
      \node (v3) [below left of=v1] {$s_4$};
      \draw (v0) -- node[left]  {$3$} (v1);
      \draw (v1) -- node[above right, xshift=-.1cm] {$5$} (v2);
      \draw (v1) -- node[above left, xshift=.2cm]  {$11$} (v3);
      \draw (v3) to [out=90, in=180] (v0);
      \draw (v2) to [in=0, out=90] (v0);
    \end{scope}
  \end{tikzpicture}     
  \qquad\qquad%
  \begin{tikzpicture}[->, shorten >=1pt, auto, node distance=.75cm,
    initial text=]
    \begin{scope}[outer sep=1pt,minimum
      size=5pt,inner sep=2pt, node distance=1.9cm] 
      \node (u0) {$t_1$};        
      \node (u1) [below left of=u0]{$t_2$};
      \node (u2) [below right of=u0]{$t_3$};
      \node (u3) [below of=u1] {$t_4$};    
      \node (u4) [below of=u2] {$t_5$};    
      \draw (u0) -- node[above]  {$3$} (u1);
      \draw (u0) -- node[above] {$4$} (u2);
      \draw (u1) -- node[left] {$5$} (u3);
      \draw (u2) -- node[right] {$15$} (u4);
      \draw (u3) to[out=30, in=260] (u0);
      \draw (u4) to[out=150, in=280] (u0);
    \end{scope}
  \end{tikzpicture}
  \caption{%
    \label{wtsjlap.fig:example_sim}
    Two example transition systems}
\end{figure}

\begin{example}
  \label{wtsjlap.ex:distances}
  We show a computation of the different distances between the states
  $s_1$ and $t_1$ in the transition system in
  Figure~\ref{wtsjlap.fig:example_sim}.  Edges without specified weight
  have weight $0$, and the discounting factor is $\lambda= .90$.

  It is easy to see that supremum trace distance is obtained for the
  path from $s_1$ which always turns left at $s_2$, \ie takes the
  transition $s_2\tto{11} s_4$, and then for the point-wise and
  accumulating trace distances, that the matching trace from $t_1$
  giving infimum trace distance in turn is obtained for the path which
  always takes the transition $t_1\tto 4 t_3$.  Hence we can compute
  \begin{align*}
    \apwdist{ s_1, t_1} &= \sup_i\big\{ \max( 1, 4\cdot .90)\cdot
    .90^{3i}\big\}= 3.60\,, \\
    \aacdist{ s_1, t_1} &= \sum_i ( 1+ 4\cdot .90)\cdot
    .90^{3i}\approx 17.0\,.
  \end{align*}

  For maximum-lead trace distance the situation is more involved.  It
  can be shown that for this distance, an infimum trace $\tau$ from
  $t_1$ follows the path which takes $t_1\tto{4}t_3$, followed by $t_1
  \tto{3} t_2$ three times, and then repeats $t_1\tto{4}t_3$
  indefinitely.  Using this trace, we obtain
  \begin{align*}
    \apmdist{s_1,t_1}= 13\cdot .90^{ 10}\approx 4.53\,.
  \end{align*}
  \qedhere
\end{example}

The definition of branching distance is \emph{not} independent of
which trace distance is being used.  We only give the definitions for
the first two of our example trace distances.  They are defined as
least fixed points to the following equations:
\begin{align*}
  \pwbd{ s, t} &= \adjustlimits \sup_{ s\tto x s'} \inf_{ t\tto y t'}
  \max\big( d_\KK( x, y), \lambda\, \pwbd{ s', t'}\big) \\
  \awbd{ s, t} &= \adjustlimits \sup_{ s\tto x s'} \inf_{ t\tto y t'}
  d_\KK( x, y)+ \lambda\, \pwbd{ s', t'} \\
\end{align*}

\begin{example}
  Continuing the previous example, repeated application of the
  definition yields the following fixed-point equation for
  $\pwbd{ s_1, t_1}$ (note that there is only one transition from
  $s_1$, $t_2$ and $t_3$, respectively):
  \begin{align*}
    \pwbd{ s_1, t_1} &= \inf \left\{
      \begin{aligned}
        & \max\big(| 3- 3|, .90\, \pwbd{ s_2, t_2}\big) \\
        & \max\big(| 3- 4|, .90\, \pwbd{ s_2, t_3}\big)
      \end{aligned}
    \right. \\
    &= \inf \left\{
      \begin{aligned}
        & \max\big( 0, .90\,| 11- 5|, .90^2 \pwbd{ s_4, t_4}, .90\,| 5-
        5|, .90^2 \pwbd{ s_5, t_4}\big) \\
        & \max\big( 1, .90\,| 11- 15|, .90^2 \pwbd{ s_4, t_5}, .90\,|
        5- 15|, .90^2 \pwbd{ s_5, t_5}\big)
      \end{aligned}
    \right. \\
    &= \inf \left\{
      \begin{aligned}
        & \max\big( 5.4, .90^3 \pwbd{ s_1, t_1}\big) \\
        & \max\big( 9, .90^3 \pwbd{ s_1, t_1}\big)
      \end{aligned}
    \right.
  \end{align*}
  which has least fixed point $\pwbd{ s_1, t_1}= 5.4$.  For the
  accumulating distance, we calculate:
  \begin{align*}
    \awbd{ s_1, t_1} &= \inf \left\{
      \begin{aligned}
        & | 3- 3|+ .90\, \awbd{ s_2, t_2} \\
        & | 3- 4|+ .90\, \awbd{ s_2, t_3}
      \end{aligned}
    \right. \\
    &= \inf \left\{
      \begin{aligned}
        & .90 \sup \left\{
          \begin{aligned}
            & | 11- 5|+ .90\, \awbd{ s_4, t_4} \\
            & | 5- 5|+ .90\, \awbd{ s_5, t_4}
          \end{aligned}
        \right. \\
        & 1+ .90 \sup \left\{
          \begin{aligned}
            & | 11- 15|+ .90\, \awbd{ s_4, t_5} \\
            & | 5- 15|+ .90\, \awbd{ s_5, t_5}
          \end{aligned}
        \right.
      \end{aligned}
    \right. \\
    &= \inf \left\{
      \begin{aligned}
        & .90 \sup \left\{
          \begin{aligned}
            & 6+ .90^2 \awbd{ s_1, t_1} \\
            & .90^2 \awbd{ s_1, t_1}
          \end{aligned}
        \right. \\
        & 1+ .90 \sup \left\{
          \begin{aligned}
            & 4+ .90^2 \awbd{ s_1, t_1} \\
            & 10+ .90^2 \awbd{ s_1, t_1}
          \end{aligned}
        \right.
      \end{aligned}
    \right. \\
    &= \inf\big( .90\big( 6+ .90^2 \awbd{ s_1, t_1}\big), 1+
    .90\big( 10+ .90^2 \awbd{ s_1, t_1}\big)\big) \\
    &= 5.4+ .90^3 \awbd{ s_1, t_1}
  \end{align*}
  Hence $\awbd{ s_1, t_1}\approx 19.9$. \qedhere
\end{example}
  
Linear distances generalize \emph{trace inclusion} for transition
systems, whereas branching distances generalize \emph{simulation}.  We
show that the linear distance between two states is always bounded
above by the corresponding branching distance
(Theorem~\ref{th:wtsjlap.linvsbra}), a generalization of the fact that
simulation implies trace inclusion.

We also show that the point-wise linear and the point-wise branching
distances are \emph{topologically inequivalent}, that is, one may be
zero while the other is infinite.  This is a quantitative
generalization of the fact that trace inclusion and simulation are not
equivalent.  We show the same topological inequivalence for the
accumulating and maximum-lead distances
(Theorem~\ref{th:wtsjlap.topineq}).

When discounting is applied, then the point-wise, accumulating and
max\-imum-lead linear distances are \emph{Lipschitz equivalent}
(Theorem~\ref{wtsjlap.th:trace-dist-eq-disc<1}); similarly, the three
branching distances are Lipschitz equivalent
(Theorem~\ref{th:wtsjlap.topeqbra}).  Without discounting, the distances
are topologically inequivalent.  Lipschitz equivalence means that one
distance is bounded by the other, multiplied by a scaling factor;
hence properties of one distance may be transferred to the other.

Chapter~\ref{ch:wtsjlap} is based on work by the author's PhD~student
Claus Thrane, Kim G.~Larsen, and the author, which has been presented
at the 20th Nordic Workshop on Programming Theory (NWPT)
\cite{UF-Thrane-08-WTS} and subsequently published in the Journal of
Logic and Algebraic Programming (now the Journal of Logical and
Algebraic Methods in Programming) in
2010~\cite{DBLP:journals/jlp/ThraneFL10}.

\subsection{Chapter~\ref{ch:wtscai}, ``A Quantitative Characterization
  of Weighted Kripke Structures in Temporal Logic''}

In Chapter~\ref{ch:wtscai} we consider the discounted point-wise and
accumulating distances and introduce corresponding semantics for
\emph{weighted CTL}.  In these semantics, the evaluation of a formula
in a state is not a Boolean \texttt{true} or \texttt{false}, but
instead a non-negative real number (or infinity) which, intuitively,
characterizes how well the state satisfies the formula.  Our syntax
for WCTL extends the one of CTL~\cite{DBLP:conf/lop/ClarkeE81} as
follows:
\begin{align*}    
  \Phi &::= p \mid \lnot p \mid \Phi_1 \land \Phi_2 \mid
         \Phi_1 \lor \Phi_2 \mid \sfE \Psi \mid \sfA \Psi\\
  \Psi &::= \sfX_c \Phi\mid \sfG_c \Phi \mid
         \sfF_c \Phi\mid [\Phi_1 \sfU_c \Phi_2]
\end{align*}
Here, as usual, $\Phi$ generates state formulae whereas $\Psi$
generates path formulae, $p$ is an atomic proposition, and
$c\in \Realnn$ is any non-negative real number.

Semantically, formulae are interpreted in states of a Kripke structure
with labels in $\KK$, and the result of such an interpretation is a
non-negative real number.  First, the semantics of state formulae is
given as follows:
\begin{gather*}
  \lsem{p}(s) =
  \begin{cases}
    0 &\text{if }p \in L(s) \\
    \infty & \text{otherwise}
  \end{cases}
  \qquad%
  \lsem{\lnot p}(s) =
  \begin{cases}
    0 &\text{if }p \in \Props \setminus L(s) \\
    \infty & \text{otherwise}
  \end{cases}
  \\
  \lsem{\varphi_1 \lor \varphi_2}(s) = \inf\big\{\lsem{\varphi_1}(s),
  \lsem{\varphi_2}(s)\big\} \qquad%
  \lsem{\varphi_1 \land \varphi_2}(s) = \sup\big\{\lsem{\varphi_1}(s),
  \lsem{\varphi_2}(s)\big\} \\[1ex]
  \lsem{\sfE \psi}(s) = \inf\big\{ \lsem{\psi}(\sigma) \mid {\sigma
    \in \tracesfrom{s}}\big\} \qquad%
  \lsem{\sfA \psi}(s) = \sup\big\{ \lsem{\psi}(\sigma) \mid {\sigma
    \in \tracesfrom{s}}\big\}
\end{gather*}
In the last two formulae, $\lsem{\psi}(\sigma)$ is the semantics of
the trace $\sigma$ with respect to $\psi$, which depends on whether
the point-wise or the accumulating distance is used.  For example, the
point-wise path semantics is given as follows:
\begin{align*}
    \lsem{\varphi}(\sigma) &{}={} \lsem{\varphi}(\sigma_0)\\
    \lsem{\sfX_c \varphi}(\sigma) &{}={} \max\big( d_\KK( \sigma_0, c),
    \lambda \lsem{\varphi}(\sigma^1)\Big\} \\
    \lsem{\sfF_c \varphi}(\sigma) &{}={}
    \inf_{ k\ge 0}\Big( \max\big( \max_{0\le j < k} \big( \lambda^j d_\KK(
    \sigma_j, c)\big),
    \lambda^k\lsem{\varphi}(\sigma^k)\big)\Big) \\
    \lsem{\sfG_c \varphi}(\sigma) &{}={}
    \sup_{ k\ge 0}\Big( \max\big( \max_{0\le j < k} \big( \lambda^j d_\KK(
    \sigma_j, c)\big),
    \lambda^k\lsem{\varphi}(\sigma^k)\big)\Big) \\
    \lsem{\varphi_1 \sfU_c \varphi_2}(\sigma) &{}={} 
    \inf_{ k\ge0} \Big( \max\big( \max_{0\le j < k} \big( \lambda^j d_\KK(
    \lsem{\varphi_1}(\sigma^j), c) \big),
    \lambda^k\lsem{\varphi_2}(\sigma^k)\big)\Big)
\end{align*}
Here $\sigma^k=( \sigma_k, \sigma_{k+1},\dotsc)$ denotes the $k$-shift
of the trace $\sigma=( \sigma_0, \sigma_1,\dotsc)$.

We then show in Theorems~\ref{th:wtscai.adeq-acc}
and~\ref{th:wtscai.adeq-pw} that with these semantics, WCTL is
\emph{adequate} for the corresponding bisimulation distances.  That
is, the bisimulation distance between two states is precisely the
supremum, over all WCTL formulae, of the absolute value of the
difference of the formula's evaluation in these two states.

We also show, in Theorems~\ref{wtscai.th:express}
and~\ref{wtscai.th:express_pt}, that with the corresponding semantics,
WCTL is \emph{expressive} for the discounted point-wise and
accumulating distances.  This means that given a
state in a Kripke structure, there exists a WCTL formula which
characterizes the state in the sense that the bisimulation distance to
any other state is precisely the evaluation of the formula in that
state.

Our notions of adequacy and expressiveness are standard quantitative
generalizations of Hennessy and Milner's definitions
from~\cite{DBLP:journals/jacm/HennessyM85}.

Chapter~\ref{ch:wtscai} is based on work by the author's PhD~student
Claus Thrane, Kim G.~Larsen, and the author, which has been presented
at the 5th Conference on Mathematical and Engineering Methods in
Computer Science \linebreak
(MEMICS; best paper award) \cite{FLT10-Kripke-MEMICS} and subsequently
published in the Journal of Computing and
Informatics~\cite{journals/cai/FahrenbergLT10}.

\subsection{Chapter~\ref{ch:simdistax}, ``Metrics for Weighted
  Transition Systems: Axiomatization''}

In Chapter~\ref{ch:simdistax} we develop sound and complete
axiomatizations of the point-wise and the discounted accumulating
distances for finite and for regular weighted processes.  In this
context, a finite weighted process is given using the grammar
\begin{equation*}
  E::= \nil\mid n.E\mid E+ E \mid \qquad n\in \KK\,,
\end{equation*}
where $\KK$ is a finite set of weights, with a metric $d_\KK$, and
$\nil$ denotes the empty process.  A regular weighted process is
given using the grammar
\begin{equation*}
  E::= \U \mid X\mid n.E\mid E+ E\mid \mu X. E \qquad n\in \KK,
  X\in V\,,
\end{equation*}
where $V$ is a set of variables, $\U$ is the universal process, and
$\mu X. E$ denotes a minimal fixed point.

\newcommand{\rulename}[1]{\LeftLabel{\small (#1)}}
\newcommand{\condition}[1]{\RightLabel{\small \quad #1}}
\begin{figure}[tb]
  \centering
  \AxiomC{}
  \rulename{A1}
  \condition{$0\bowtie r$}
  \UnaryInfC{$\vdash \dist{\nil,E}\bowtie r$}
  \DisplayProof\qquad
  \AxiomC{}
  \rulename{A2}
  \condition{$\infty\bowtie r$}
  \UnaryInfC{$\vdash \dist{n.E,\nil}\bowtie r$}
  \DisplayProof
  \begin{prooftree} 
    \AxiomC{$\vdash \dist{E,F} \bowtie r_1$} 
    \condition{$\max(d_\KK( n, m), \lambda r_1)\bowtie r$}
    \rulename{R1$^\bullet$}
    \UnaryInfC{$\vdash \dist{n.E,m.F} \bowtie r$} 
  \end{prooftree}
  \begin{prooftree} 
    \AxiomC{$\vdash \dist{E,F} \bowtie r_1$} 
    \condition{$d_\KK( n, m) +\lambda r_1 \bowtie r$}
    \rulename{R1$^+$}
    \UnaryInfC{$\vdash \dist{n.E,m.F} \bowtie r$} 
  \end{prooftree}
  \begin{prooftree} 
    \AxiomC{$\vdash \dist{E_1, F}\bowtie r_1$} 
    \AxiomC{$\vdash \dist{E_2, F}\bowtie r_2$} 
    \condition{$\max(r_1, r_2)\bowtie r$}
    \rulename{R2}
    \BinaryInfC{$\vdash \dist{E_1+ E_2, F}\bowtie r$} 
  \end{prooftree}
  \begin{prooftree} 
    \AxiomC{$\vdash \dist{n.E, F_1} \bowtie r_1$} 
    \AxiomC{$\vdash \dist{n.E, F_2} \bowtie r_2$} 
    \condition{$\min(r_1,r_2)\bowtie r$}
    \rulename{R3}\label{intro:simdistax.rule:R3}
    \BinaryInfC{$\vdash \dist{n.E, F_1+ F_2} \bowtie r$} 
  \end{prooftree}
  \caption{Axiomatization of point-wise and discounted accumulating
    distance for finite weighted processes}
  \label{fig:intro.F}
\end{figure}

We then give axiomatizations of the point-wise and the discounted
accumulating distances for finite weighted processes, as shown in
Figure~\ref{fig:intro.F}.  These differ only in one proof rule: for the
point-wise distance, rule (R1$^\bullet$) applies, for the discounted
accumulating distance, rule (R1$^+$).  We show the axiomatizations to
be sound and complete in Theorems~\ref{th:simdistax.sound-acc},
\ref{simdistax.th:finite-complete}
and~\ref{th:simdistax.sound-comp-pw}.

For regular weighted processes, we develop similar axiomatizations,
which we then show to be sound and \emph{$\epsilon$-complete} in
Theorems~\ref{simdistax.th:regular-sound}, \ref{simdistax.th:rec_comp}
and~\ref{th:simdistax.sound-comp-pw.reg}.  Here,
$\epsilon$-completeness means that a distance of $d$ can be proven
within an interval $[ d- \epsilon, d+ \epsilon]$, for any positive
real $\epsilon$.

Chapter~\ref{ch:simdistax} is based on work by the author's
PhD~student Claus Thrane, Kim G.~Larsen, and the author, which has
been published in Theoretical Computer
Science~\cite{DBLP:journals/tcs/LarsenFT11}.

\subsection{Chapter~\ref{ch:qltbt}, ``The Quantitative
  Linear-\!Time--Branching-\!Time Spectrum''}

Chapter~\ref{ch:qltbt} presents a generalization of the work in
Chapter~\ref{ch:wtsjlap} along several directions.  Instead of
developing theory separately for different trace distances, we treat
the trace distance as an input and develop a general theory of linear
and branching distances pertaining to a given, but unspecified, trace
distance.

Let again $\KK$ be a set of symbols, and denote by
$\KK^\infty= \KK^*\cup \KK^\omega$ the set of finite and infinite
traces in $\KK$.  A \emph{trace distance} is, then, a function
$d: \KK^\infty\times \KK^\infty\to \Realnn\cup\{ \infty\}$ which
satisfies $d( \sigma, \sigma)= 0$ and
$d( \sigma, \tau)+ d( \tau, \chi)\ge d( \sigma, \chi)$ for all
$\sigma, \tau, \chi\in \KK^\infty$, and additionally, $d( \sigma,
\tau)= \infty$ if $\sigma$ and $\tau$ have different length.

Given such a general trace distance $d$, we can define the linear
distance between two states $s$ and $t$ of a transition system by
\begin{equation*}
  \ntrace 1 d( s, t)= \adjustlimits \sup_{ \sigma\in \tracesfrom s}
  \inf_{ \tau\in \tracesfrom t} d( \sigma, \tau)
\end{equation*}
as before.  As this generalizes the standard trace inclusion preorder,
we now call this the ($1$-nested) \emph{trace inclusion distance}.

Using a \emph{quantitative Ehrenfeucht-Fra{\"\i}ss{\'e} game}, we can
then define a corresponding ($1$-nested) \emph{simulation distance}
$\nsim 1 d$ and show that $\ntrace 1 d( s, t)\le \nsim 1 d( s, t)$ for
all states $s$, $t$.  Similarly, we can define the ($1$-nested)
\emph{trace equivalence distance} between $s$ and $t$ by
\begin{equation*}
  \ntraceeq 1 d( s, t)= \max\Big( \adjustlimits \sup_{ \sigma\in
    \tracesfrom s} \inf_{ \tau\in \tracesfrom t} d( \sigma, \tau),
  \adjustlimits \sup_{ \tau\in \tracesfrom t} \inf_{ \sigma\in 
    \tracesfrom s} d( \sigma, \tau)\Big)
\end{equation*}
and use a different quantitative Ehrenfeucht-Fra{\"\i}ss{\'e} game to
define the \emph{bisimulation distance} $\bisim d$, with the property
that $\ntraceeq 1 d( s, t)\le \bisim d( s, t)$ for all $s$, $t$.

In Chapter~\ref{ch:qltbt}, we generalize these considerations to
define linear and branching distances for most of the preorders and
equivalences in van~Glabbeek's linear-time--branching-time
spectrum~\cite{inbook/hpa/Glabbeek01}.  Hence we can define nested simulation
distances, ready simulation distances, possible-futures distances,
readiness distances, and others, all parameterized by the
given-but-unspecified trace distance.  The resulting
\emph{quantitative linear-time--branching-time spectrum} is depicted
in Figure~\ref{intro.fi:spectrum}.

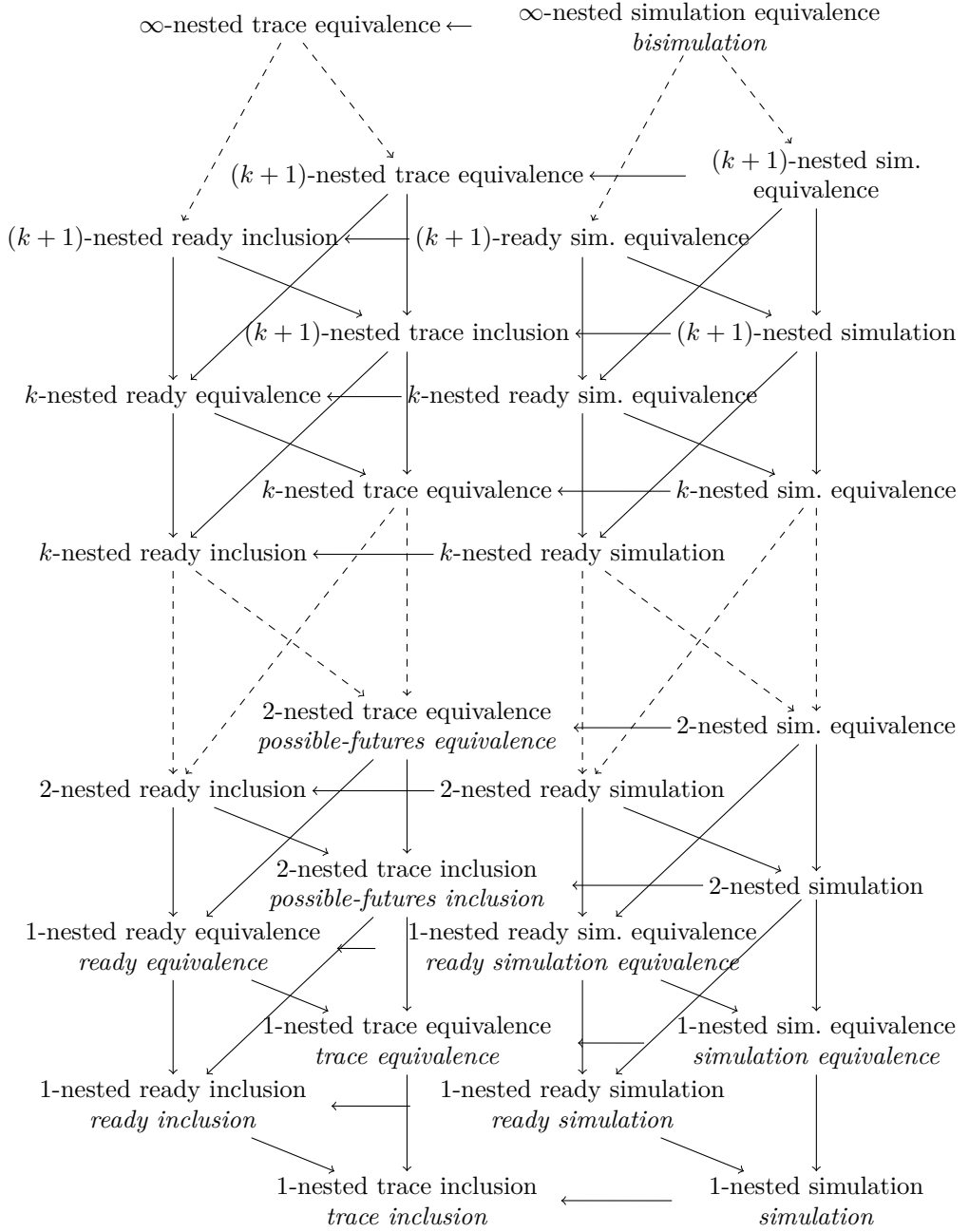
\begin{figure}[tbp]
  \centering
  \begin{tikzpicture}[->,xscale=.81,yscale=1.1]
    \tikzstyle{every node}=[font=\small,text badly centered, inner sep=2pt]
    \node (traceeq) at (0,.3) {$\infty$-nested trace equivalence};
    \node (k+1-r-trace) at (-2,-2.4) {$( k+ 1)$-nested ready
      inclusion};
    \node (k+1-traceeq) at (2,-1.6) {$( k+ 1)$-nested trace equivalence};
    \node (k-r-traceeq) at (-2,-4.4) {$k$-nested ready equivalence};
    \node (k+1-trace) at (2,-3.6) {$( k+ 1)$-nested trace inclusion};
    \node (k-r-trace) at (-2,-6.4) {$k$-nested ready inclusion};
    \node (k-traceeq) at (2,-5.6) {$k$-nested trace equivalence};
    \node (2-r-trace) at (-2,-9.4) {$2$-nested ready inclusion};
    \node [text width=11.6em] (2-traceeq) at (2,-8.6) {$2$-nested trace
      equivalence \\ \emph{possible-futures equivalence}};
    \node [text width=11.4em] (1-r-traceeq) at (-2,-11.4) {$1$-nested ready
      equivalence \\ \emph{ready equivalence}};
    \node [text width=11.5em] (2-trace) at (2,-10.6) {$2$-nested trace
      inclusion \\ \emph{possible-futures inclusion}};
    \node [text width=11em] (1-r-trace) at (-2,-13.4) {$1$-nested ready
      inclusion \\ \emph{ready inclusion}};
    \node [text width=11.9em] (1-traceeq) at (2,-12.6) {$1$-nested trace
      equivalence \\ \emph{trace equivalence}};
    \node [text width=10.7em] (1-trace) at (2,-14.6) {$1$-nested trace
      inclusion \\ \emph{trace inclusion}};
    \node [text width=16em] (bisim) at (7,.3) {$\infty$-nested
      simulation equivalence \\ \emph{bisimulation}};
    \node (k+1-r-sim) at (5,-2.4) {$( k+ 1)$-ready sim.~equivalence};
    \node [text width=9em] (k+1-simeq) at (9,-1.6) {$( k+ 1)$-nested
      sim. equivalence};
    \node (k-r-simeq) at (5,-4.4) {$k$-nested ready sim.~equivalence};
    \node (k+1-sim) at (9,-3.6) {$( k+ 1)$-nested simulation};
    \node (k-r-sim) at (5,-6.4) {$k$-nested ready simulation};
    \node (k-simeq) at (9,-5.6) {$k$-nested sim.~equivalence};
    \node (2-r-sim) at (5,-9.4) {$2$-nested ready simulation};
    \node (2-simeq) at (9,-8.6) {$2$-nested sim.~equivalence};
    \node [text width=14.5em] (1-r-simeq) at (5,-11.4) {$1$-nested ready
      sim.~equivalence \\ \emph{ready simulation equivalence}};
    \node (2-sim) at (9,-10.6) {$2$-nested simulation};
    \node [text width=12em] (1-r-sim) at (5,-13.4) {$1$-nested ready
      simulation \\ \emph{ready simulation}};
    \node [text width=12em] (1-simeq) at (9,-12.6) {$1$-nested
      sim.~equivalence \\ \emph{simulation equivalence}};
    \node [text width=10em] (1-sim) at (9,-14.6) {$1$-nested
      simulation \\ \emph{simulation}};
    %
    \path (bisim) edge (traceeq);
    \path (k+1-r-sim) edge (k+1-r-trace);
    \path (k+1-simeq) edge (k+1-traceeq);
    \path (k-r-simeq) edge (k-r-traceeq);
    \path (k+1-sim) edge (k+1-trace);
    \path (k-r-sim) edge (k-r-trace);
    \path (k-simeq) edge (k-traceeq);
    \path (2-simeq) edge (2-traceeq);
    \path (2-r-sim) edge (2-r-trace);
    \path (2-sim) edge (2-trace);
    \path (1-r-simeq) edge (1-r-traceeq);
    \path (1-simeq) edge (1-traceeq);
    \path (1-r-sim) edge (1-r-trace);
    \path (1-sim) edge (1-trace);
    \path [dashed] (traceeq) edge (k+1-r-trace);
    \path [dashed] (traceeq) edge (k+1-traceeq);
    \path (k+1-r-trace) edge (k-r-traceeq);
    \path (k+1-r-trace) edge (k+1-trace);
    \path (k+1-traceeq) edge (k-r-traceeq);
    \path (k+1-traceeq) edge (k+1-trace);
    \path (k-r-traceeq) edge (k-r-trace);
    \path (k-r-traceeq) edge (k-traceeq);
    \path (k+1-trace) edge (k-r-trace);
    \path (k+1-trace) edge (k-traceeq);
    \path [dashed] (k-r-trace) edge (2-r-trace);
    \path [dashed] (k-r-trace) edge (2-traceeq);
    \path [dashed] (k-traceeq) edge (2-r-trace);
    \path [dashed] (k-traceeq) edge (2-traceeq);
    \path (2-r-trace) edge (1-r-traceeq);
    \path (2-r-trace) edge (2-trace);
    \path (2-traceeq) edge (1-r-traceeq);
    \path (2-traceeq) edge (2-trace);
    \path (1-r-traceeq) edge (1-r-trace);
    \path (1-r-traceeq) edge (1-traceeq);
    \path (2-trace) edge (1-r-trace);
    \path (2-trace) edge (1-traceeq);
    \path (1-r-trace) edge (1-trace);
    \path (1-traceeq) edge (1-trace);
    \path [dashed] (bisim) edge (k+1-r-sim);
    \path [dashed] (bisim) edge (k+1-simeq);
    \path (k+1-r-sim) edge (k-r-simeq);
    \path (k+1-r-sim) edge (k+1-sim);
    \path (k+1-simeq) edge (k-r-simeq);
    \path (k+1-simeq) edge (k+1-sim);
    \path (k-r-simeq) edge (k-r-sim);
    \path (k-r-simeq) edge (k-simeq);
    \path (k+1-sim) edge (k-r-sim);
    \path (k+1-sim) edge (k-simeq);
    \path [dashed] (k-r-sim) edge (2-r-sim);
    \path [dashed] (k-r-sim) edge (2-simeq);
    \path [dashed] (k-simeq) edge (2-r-sim);
    \path [dashed] (k-simeq) edge (2-simeq);
    \path (2-r-sim) edge (1-r-simeq);
    \path (2-r-sim) edge (2-sim);
    \path (2-simeq) edge (1-r-simeq);
    \path (2-simeq) edge (2-sim);
    \path (1-r-simeq) edge (1-r-sim);
    \path (1-r-simeq) edge (1-simeq);
    \path (2-sim) edge (1-r-sim);
    \path (2-sim) edge (1-simeq);
    \path (1-r-sim) edge (1-sim);
    \path (1-simeq) edge (1-sim);
  \end{tikzpicture}
  \caption{%
    \label{intro.fi:spectrum}
    The quantitative linear-time--branching-time spectrum.  The nodes
    are different system distances, and an edge
    $d_1\longrightarrow d_2$ or $d_1\dashrightarrow d_2$ indicates
    that $d_1( s, t)\ge d_2( s, t)$ for all states $s$, $t$, and that
    $d_1$ and $d_2$ in general are topologically inequivalent.}
\end{figure}

We also show that if the trace distance has a recursive
characterization in a lattice $L$ above $\Realnn\cup\{ \infty\}$, then
all distances in the quantitative linear-time--branching-time spectrum
have a \emph{fixed-point characterization} over $L$.  As an example,
if $d: \KK^\infty\times \KK^\infty\to \Realnn\cup\{ \infty\}$ is the
point-wise distance, then
\begin{equation*}
  d( \sigma, \tau)= F\big( \sigma_0, \tau_0, d( \sigma^1, \tau^1)\big)
\end{equation*}
for all $\sigma, \tau\in \KK^\infty$ (recall that $\sigma_0$ denotes
the head of $\sigma$ and $\sigma^1$ its tail), where
$F: \KK\times \KK\times( \Realnn\cup\{ \infty\})\to \Realnn\cup\{
\infty\}$ is given by
$F( x, y, \alpha)= \max\big( d_\KK( x, y), \lambda \alpha\big)$.  (In
this case, the lattice $L= \Realnn\cup\{\infty\}$.)

The simulation distance is then the least fixed point to the equations
\begin{equation*}
  \nsim 1 d( s, t)= \adjustlimits \sup_{ s\tto x s'} \inf_{ t\tto y
    t'} F\big( x, y, \nsim 1 d( s', t')\big)\,.
\end{equation*}
If, instead, $d$ is the discounted accumulating distance, then the
above equations hold for $F$ replaced by
$F( x, y, \alpha)= d_\KK( x, y)+ \lambda \alpha$.

Chapter~\ref{ch:qltbt} is based on work by the author's PhD~student
Claus Thrane, Kim G.~Larsen, Axel Legay, and the author, which has
been presented at the 9th Workshop on Quantitative Aspects of
Programming Languages (QAPL) \cite{conf/qapl/FahrenbergTL11} and the
31st IARCS Conference on Foundations of Software Technology and
Theoretical Computer Science (FSTTCS)
\cite{DBLP:conf/fsttcs/FahrenbergLT11} and subsequently published in
Theoretical Computer Science~\cite{DBLP:journals/tcs/FahrenbergL14}.

\subsection{Chapter~\ref{ch:weightedmodal}, ``Weighted Modal
  Transition Systems''}
\label{se:intro.weightedmodal}

Chapter~\ref{ch:weightedmodal} presents a lifting of our work on
quantitative verification to \emph{quantitative specification
  theories}.  Fundamental to specification theories is the
\emph{refinement relation} which permits to successively refine
specifications until an implementation is reached.  Here
implementations are the models with which previous chapters were
concerned, \ie~transition systems.  In the context of quantitative
verification, we have in previous chapters replaced equivalence
relations and preorders between models by linear and branching
distances.  Similarly in spirit, we replace in this chapter the
refinement relation with a \emph{refinement distance}, in order to be
able to reason quantitatively about quantitative specifications.

In Chapter~\ref{ch:weightedmodal} we treat a special case of
quantitative specification theory, using models which are transition
systems whose transitions are labeled with symbols from a discrete
alphabet $\Sigma$ and with integers.  We also use one particular
distance, the discounted accumulating one.  In the following
Chapter~\ref{ch:wm2}, we generalize this setting to arbitrary models
and specifications and arbitrary distances.

In the specifications of Chapter~\ref{ch:weightedmodal}, integer
weights are relaxed to \emph{integer intervals} and, as usual in modal
specifications, transitions can be of type \must or of type \may.
Hence we define a \emph{weighted modal transition system} (WMTS) to be
a structure $\cal S=( S, s^0, \mmayto, \mmustto)$ consisting of a set
of states $S$ with an initial state $s^0\in S$ and transition
relations $\mmustto, \mmayto\subseteq S\times \Spec\times S$ such that
for every $(s,k,s') \in \mmustto$ there is $(s,\ell,s') \in \mmayto$
where $k \labpre \ell$.  Here
\begin{equation*}
  \Spec= \Sigma\times \big\{[ x, y]\bigmid x\in \Int\cup\{ -\infty\},
  y\in \Int\cup\{ \infty\}, x\le y\big\}
\end{equation*}
is the set of (weighted) \emph{specification labels}, and the partial
order $\labpre$ on $\Spec$ is defined by
$( a, I)\labpre( a', I')$ if $a= a'$ and $I\subseteq I'$.

A WMTS $\cal S$ as above is an \emph{implementation} if
$\mmustto= \mmayto\subseteq S\times \Impl \Spec\times S$, where
\begin{equation*}
  \Impl \Spec= \Sigma\times\big\{[ x, x]\bigmid x\in \Int\big\}\approx
  \Sigma\times \Int
\end{equation*}
is the set of (weighted) \emph{implementation labels} in $\Spec$: the
minimal elements of $\Spec$ with respect to $\labpre$.

Now in a standard \emph{modal refinement} $\cal S_1\mr \cal S_2$,
\must-transitions in $\cal S_2$ must be preserved in $\cal S_1$,
whereas \may-transitions in $\cal S_1$ must correspond to
\may-transitions in $\cal S_2$.  Using the accumulating distance with
a discounting factor $\lambda< 1$, and our work in
Chapter~\ref{ch:qltbt}, we extend this to a \emph{modal refinement
  distance} which is defined as follows.  First, a distance on
specification labels is introduced by
$d_\Spec(( a, I),( a', I'))= \infty$ if $a\ne a'$ and
\begin{align*}
  d_\Spec\big(( a,[ x, y]),( a,[ x', y'])\big)
  &= \adjustlimits \sup_{ z\in[ x, y]} \inf_{ z'\in[ x', y']}| z- z'|
  \\
  &= \max( x'- x, y- y', 0)\,.
\end{align*}
The modal refinement distance between the states of weighted modal
transition systems $\cal S_1=( S_1, s^0_1, \mmayto_1, \mmustto_1)$ and
$\cal S_2=( S_2, s^0_2, \mmayto_2, \mmustto_2)$ is then defined to be
the least fixed point of the equations
\begin{equation*}
  \md( s_1, s_2)= \max\left\{
    \begin{aligned}
      &\adjustlimits \sup_{ s_1\mayto{ k_1}_1 t_1} \inf_{ s_2\mayto{
          k_2}_2 t_2} d_\Spec( k_1, k_2)+ \lambda \md( t_1, t_2)\,, \\
      &\adjustlimits \sup_{ s_2\mustto{ k_2}_2 t_2} \inf_{
        s_1\mustto{ k_1}_1 t_1} d_\Spec( k_1, k_2)+ \lambda \md(
      t_1, t_2)\,,
    \end{aligned}
  \right.
\end{equation*}
and then $\md( \cal S_1, \cal S_2)= \md( s^0_1, s^0_2)$.

We show in Theorem~\ref{weightedmodal.th:dtledm} that the modal
refinement distance bounds the so-called \emph{thorough refinement
  distance}: for any implementation $\cal I_1\mr \cal S_1$, there is
an implementation $\cal I_2\mr \cal S_2$ such that
$d( \cal I_1, \cal I_2)\le \md( \cal S_1, \cal S_2)$.  Hence the modal
refinement distance between two specifications can serve as an
over-approximation of how far respective implementations can deviate
from each other.

Modal specifications come equipped with a logical operation of
\emph{conjunction} and with structural operations of
\emph{composition} and \emph{quotient}.  Conjunction is the greatest
lower bound in the modal refinement preorder.  We show that such a
conjunction exists in our formalism, but that it does \emph{not}
satisfy a natural quantitative generalization of the greatest lower
bound property; in fact, Theorem~\ref{weightedmodal.th:no-conj} shows
that \emph{there is no} operation $\land$ on WMTS which satisfies that
for any $\epsilon\ge 0$, there exist $\epsilon_1, \epsilon_2\ge 0$
such that whenever $\md( S, S_1)\le \epsilon_1$ and
$\md( S, S_2)\le \epsilon_2$ for some WMTS $S, S_1, S_2$, then
$\md( S, S_1\land S_2)\le \epsilon$.  Conjunction is thus, in this
sense, \emph{discontinuous}; we shall see in the following
Chapter~\ref{ch:wm2} that this is a fundamental problem with
\emph{any} quantitative specification theory.

For structural composition, we use CSP-style synchronization on labels
and addition of intervals.  That is, synchronization
$( a, I)\oplus( a', I')$ on $\Spec$ is undefined if $a\ne a'$, and
otherwise
\begin{equation*}
  ( a,[ x, y])\oplus( a,[ x', y'])=( a,[ x+ x', y+ y'])\,.
\end{equation*}
Using this label operation, we show in
Theorem~\ref{weightedmodal.th:indepimp} that there is a structural
composition operator $\|$ for WMTS which satisfies
$\md( \cal S_1\| \cal S_3, \cal S_2\| \cal S_4)\le \md( \cal S_1, \cal
S_2)+ \md( \cal S_3, \cal S_4)$ for all WMTS
$\cal S_1, \cal S_2, \cal S_3, \cal S_4$.  This property of
\emph{independent implementability} ensures that composition preserves
distances.  We also show in
Theorem~\ref{weightedmodal.th:soundmaxquot} that structural
composition admits a partial inverse, a quotient operation $/$ such
that
$\md( \cal S_3, \cal S_1/ \cal S_2)= \md( \cal S_2\| \cal S_3, \cal
S_1)$ for all WMTS $\cal S_1, \cal S_2, \cal S_3$ whenever $\cal S_2$
is \emph{deterministic}.  The quotient operation can hence be used to
synthesize partial specifications also in this quantitative context.

Chapter~\ref{ch:weightedmodal} is based on work by Sebastian S.~Bauer,
Line Juhl, Claus Thrane, Kim G.~Larsen, Axel Legay, and the author,
which has been presented at the 36th International Symposium on
Mathematical Foundations of Computer Science (MFCS)
\cite{DBLP:conf/mfcs/BauerFJLLT11} and subsequently published in
Formal Methods in System
Design~\cite{DBLP:journals/fmsd/BauerFJLLT13}.

\subsection{Chapter~\ref{ch:wm2}, ``General Quantitative Specification
  Theories with Modal Transition Systems''}
\label{se:intro.wm2}

In Chapter~\ref{ch:wm2} we develop a general setting for quantitative
specification theories.  Combining the work in Chapters~\ref{ch:qltbt}
and~\ref{ch:weightedmodal}, we work in a setting of modal transition
systems which are labeled with elements in a partially ordered set
$\Spec$ of specification labels.  The set of implementation labels is
then $\Imp=\{ k\in \Spec\mid k'\labpre k\limpl k'= k\}$, and
implementations are $\Imp$-labeled transition systems.

We assume given an abstract trace distance $d: \Spec^\infty\times
\Spec^\infty\to \LL$, where $\LL=( \Realnn\cup\{ \infty\})^M$, for an
arbitrary set $M$, is the lattice of functions from $M$ to
$\Realnn\cup\{ \infty\}$.  We also assume that there exists a
\emph{distance iterator} function $F: \Spec\times \Spec\times \LL\to
\LL$ such that $d( \sigma, \tau)= F( \sigma_0, \tau_0, d( \sigma^1,
\tau^1))$, similarly to the recursive characterization developed in
Chapter~\ref{ch:qltbt}.  

We can then introduce an abstract \emph{modal refinement distance}
between the states of two such \emph{structured} modal transition
systems (or SMTS) $\cal S=( S, s_0, \mmayto_S, \mmustto_S)$ and
$\cal T=( T, t_0, \mmayto_T, \mmustto_T)$ to be the least fixed point,
in $\LL$, to the equations
\begin{equation*}
  \md( s, t)= \max\left\{
    \begin{aligned}
      & \adjustlimits \sup_{ s\,\mayto{ k}_S\, s'\,} \inf_{ \,
        t\,\mayto{ \ell}_T\, t'} F( k, \ell, \md( s', t'))\,, \\
      & \adjustlimits \sup_{ t\,\mustto{ \ell}_T\, t'\,} \inf_{ \,
        s\,\mustto{ k}_S\, s'} F( k, \ell, \md( s', t'))\,,
    \end{aligned}
  \right.
\end{equation*}
and let $\md( \cal S, \cal T)= \md( s_0, t_0)$.

For conjunction of SMTS, we introduce a property of \emph{conjunctive
  boundedness} on labels in $\Spec$ which implies,
\cf~Theorem~\ref{wm2.th:conj}, that conjunction of SMTS is uniformly
bounded in the sense that the modal refinement distance from an SMTS
$\cal U$ to a conjunction $\cal S\land \cal T$ is bounded above by a
uniform function of the distances $\md( \cal U, \cal S)$ and
$\md( \cal U, \cal T)$.  Unfortunately, it turns out that common label
conjunction operators are \emph{not} conjunctively bounded, hence we
propose another property of \emph{relaxed conjunctive boundedness}
which does hold for common label operators, and show in
Theorem~\ref{wm2.th:relax_conj} that it implies a similar property for
SMTS conjunctions.

For structural composition, we generalize the work in
Chapter~\ref{ch:weightedmodal} by introducing an abstract notion of
(partial) label composition $\obar$ on $\Spec$.  Assuming that this
operator is recursively \emph{uniformly bounded} in the sense that
there exists a function $P: \LL\times \LL\to \LL$ such that
\begin{equation*}
  F( k\obar k', \ell\obar \ell', P( \alpha, \alpha'))\sqsubseteq_\LL P(
  F( k, \ell, \alpha), F( k', \ell', \alpha'))
\end{equation*}
for all $k, \ell, k', \ell'\in \Spec$ and $\alpha, \alpha'\in \LL$ for
which $k\obar k'$ and $\ell\obar \ell'$ are defined, we can then show
in Theorem~\ref{wm2.th:struct} that independent implementability holds,
\viz
\begin{equation*}
  \md( \cal S\| \cal S', \cal T\| \cal T')\sqsubseteq_\LL P( \md( \cal
  S, \cal T), \md( \cal S', \cal T'))
\end{equation*}
for all SMTS $\cal S$, $\cal T$, $\cal S'$, $\cal T'$.  In examples,
we expose several different label composition operators and show that
they are uniformly bounded.  We show that the quotient operator
$/$ from Chapter~\ref{ch:weightedmodal} has a similar
generalization to SMTS.

We also show in Chapter~\ref{ch:wm2} that quantitative refinement
admits a \emph{logical characterization}, generalizing the work in
Chapter~\ref{ch:wtscai}.  We use standard Hennessy-Milner logic, with
formulae generated by the syntax
\begin{equation*}
  \phi,\phi_1,\phi_2 := \ltrue \mid \lfalse
  \mid  \langle \ell\rangle \phi \mid [ \ell] \phi \mid
  \phi_1 \wedge \phi_2
  \mid \phi_1 \vee \phi_2 \qquad (\ell\in \Spec)
\end{equation*}
and with quantitative semantics $S\to \LL$, for an SMTS
$\cal S=( S, s_0, \mmayto_S, \mmustto_S)$ given as follows:
\begin{gather*}
  \begin{aligned}
    \wsem \ltrue s &= \bot &\qquad\qquad \wsem \lfalse s &= \top \\
    \wsem{( \phi_1\wedge \phi_2)} s &= \max( \wsem{ \phi_1} s, \wsem \phi_2
    s) &\qquad \wsem{( \phi_1\vee \phi_2)} s &= \min( \wsem{ \phi_1} s,
    \wsem{ \phi_2} s)
  \end{aligned}
  \\
  \begin{aligned}
    \wsem{ \langle \ell\rangle \phi} s &= \inf\{ F( k, \ell, \wsem \phi
    t)\mid s\mustto{ k} t, d( k, \ell)\ne \top_\LL\} \\
    \wsem{ [ \ell] \phi} s &= \sup\{ F( k, \ell, \wsem \phi t)\mid
    s\mayto{ k} t, d( k, \ell)\ne \top_\LL\}
  \end{aligned}
\end{gather*}

Writing $\wsem \phi \cal S= \wsem \phi s_0$, we can then show in
Theorem~\ref{wm2.th:l-sound} that this logic is \emph{quantitatively
  sound} for the modal refinement distance, in the sense that
$\wsem \phi \cal S\sqsubseteq_\LL \wsem \phi \cal T\oplus_\LL \md(
\cal S, \cal T)$
for all formulae $\phi$ and all SMTS $\cal S$, $\cal T$.  For
disjunction-free formulae, we can show a complementary completeness
result in Theorem~\ref{wm2.th:l-complete}, namely that
$\wsem \phi \cal S= \sup_{ \cal I\in \llbracket \cal S\rrbracket}
\wsem \phi \cal I$ for all disjunction-free $\phi$ and all SMTS $\cal
S$, where $\llbracket \cal S\rrbracket$ denotes the set of
implementations of $\cal S$.

Chapter~\ref{ch:wm2} is based on work by Sebastian S.~Bauer, Claus
Thrane, Axel Legay, and the author, which has been presented at the
7th International Computer Science Symposium in Russia (CSR)
\cite{DBLP:conf/csr/BauerFLT12} and subsequently published in Acta
Informatica~\cite{DBLP:journals/acta/FahrenbergL14}.

\subsection{Chapter~\ref{ch:dmts}, ``Logical vs.\ Behavioral
  Specifications''}
\label{se:intro.dmts}

Chapter~\ref{ch:dmts} departs from the quantitative setting of this
thesis in order to introduce a generalization of modal transition
systems which turns out to be somewhat more well-behaved both in a
qualitative setting and also in the quantitative setting of the
subsequent Chapter~\ref{ch:dmts2}.  Extending Larsen and
Xinxin~\cite{DBLP:conf/lics/LarsenX90}, we define a \emph{disjunctive
  modal transition system} (DMTS) to be a structure
$\mcalD=( S, S^0, \omay, \omust)$ consisting of finite sets
$S\supseteq S^0$ of states and initial states, a \emph{may}-transition
relation $\omay\subseteq S\times \Sigma\times S$, and a
\emph{disjunctive must}-transition relation
$\omust\subseteq S\times 2^{ \Sigma\times S}$.

\begin{figure}[tbp]
  \centering
  \begin{tikzpicture}[x=4cm,->,>=stealth',
    state/.style={shape=circle,draw,font=\scriptsize,inner sep=.5mm,outer
      sep=0.8mm,minimum size=0.3cm,initial text=,initial
      distance=3ex}]
    \node[state,initial] (X) {};
    \node[state] (Y) at (1,0) {};
    \path[densely dashed,->] (X) edge node[above]{$\req$} (Y);
    \path[densely dashed,->] (X) edge[loop above]
    node[above]{$\grant,\work,\idle$} (X);
    \path (Y) edge[loop,in=-115,out=-165,looseness=15] node[inner
    sep=0,outer sep=0,minimum size=0,pos=0.2,name=YY] {} node[below] {$\
      \ \work$} (Y);
    \path (YY) edge[bend left] node[below] {$\grant$} (X);
  \end{tikzpicture}
  \caption{%
    \label{fig:intrexample1}
    DMTS corresponding to the CTL property $\text{AG}(\req\Rightarrow
    \text{AX}(\work$ AW $\grant))$}
\end{figure}
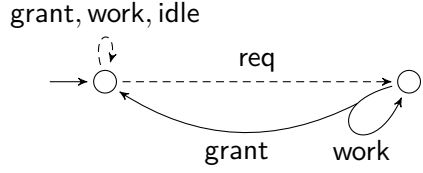

DMTS hence generalize MTS in that they allow for multiple (or zero)
initial states and permit \must transitions to branch to a disjunction
of destination states.  As an example, Figure~\ref{fig:intrexample1}
shows a DMTS expressing the CTL property
\begin{equation*}
  \text{AG}(\req\Rightarrow \text{AX}(\work \text{ AW } \grant))
\end{equation*}
(here ``AW'' denotes the weak-until operator): \emph{``at all time
  points after executing $\req$, no $\idle$ nor further requests but
  only $\work$ is allowed until $\grant$ is executed''}.  The same
property may be expressed as a recursive system of equations in
Hennessy-Milner logic~\cite{DBLP:journals/tcs/Larsen90} as
\begin{align*}
X &= [\grant,\idle,\work]X \wedge [\req] Y\\
Y &= (\langle\work\rangle Y\vee \langle\grant\rangle X)\wedge[\idle,\req]\lfalse
\end{align*}
where the solution is given by the maximal fixed point.

In Chapter~\ref{ch:dmts} we exhibit an equivalence between DMTS and
Hennessy-Milner logic with maximal fixed points (the \emph{modal
  $\nu$-calculus}) and also with a non-deterministic extension of the
\emph{acceptance automata} of~\cite{DBLP:journals/jacm/Hennessy85,
  DBLP:journals/entcs/Raclet08}.  This allows one to freely switch
between formalisms and, more importantly, to generalize the logical
and structural operations on specifications and expose their algebraic
properties.

We thus show in Theorem~\ref{th:condis} that DMTS (and hence also
acceptance automata and the modal $\nu$-calculus) admit notions of
conjunction and disjunction which are greatest lower, respectively
least upper bounds, in the modal refinement order.  That is, DMTS form
a \emph{bounded distributive lattice} up to modal equivalence.

We also generalize structural composition and quotient to DMTS and
further introduce quotients $\mcalS_1 / \mcalS_2$ also for the cases
where $\mcalS_2$ is not deterministic.  Theorem~\ref{th:quotient} then
shows that quotient is a \emph{residual} to structural composition,
with defining property
\begin{equation*}
  \mcalS_1\| \mcalS_2\mr \mcalS_3 \,\Longrightarrow\, \mcalS_2\mr
  \mcalS_3/ \mcalS_1
\end{equation*}
as before, but now without any restrictions on the involved
specifications.

Combining the four operations, DMTS form a \emph{commutative
  residuated lattice}~\cite{JipsenT02} up to modal equivalence.  As an
example, this immediately entails the following properties which may
be used in a calculus of specifications:
\begin{align*}
  \mcal S_1\|( \mcal S_2/ \mcal S_3) &\mr ( \mcal S_1\| \mcal
  S_2)/ \mcal S_3 &
  \mcal S_1/ \mcal S_2 &\mr ( \mcal S_1\| \mcal S_3)/( \mcal
  S_2\| \mcal S_3) \\
  ( \mcal S_1/ \mcal S_2)\|( \mcal S_2/ \mcal S_3) &\mr \mcal
  S_1/ \mcal S_3 &
  ( \mcal S_1/ \mcal S_2)/ \mcal S_3 &\mreq ( \mcal S_1/
  \mcal S_3)/ \mcal S_2 \\
  \mcal S_1/( \mcal S_2\| \mcal S_3) &\mreq ( \mcal S_1/ \mcal
  S_2)/ \mcal S_3 &
  \mcal S\|( \mcal S/ \mcal S) &\mreq \mcal S
\end{align*}

Chapter~\ref{ch:dmts} is based on work by Nikola Bene{\v s}, Jan K{\v
  r}et{\'\i}nsk{\'y}, Axel Legay, Louis-Marie Traonouez, and the
author, which has been presented at the 24th International Conference
on Concurrency Theory (CONCUR) \cite{DBLP:conf/concur/BenesDFKL13} and
at the 11th International Colloquium on Theoretical Aspects of
Computing (ICTAC) \cite{DBLP:conf/ictac/FahrenbergLT14} and
subsequently published in Information and
Computation~\cite{DBLP:journals/iandc/BenesFKLT20}.

\subsection{Chapter~\ref{ch:dmts2}, ``Compositionality for
  Quantitative Specifications''}
\label{se:intro.dmts2}

The final Chapter~\ref{ch:dmts2} combines the work of
Chapters~\ref{ch:wm2} and~\ref{ch:dmts}.  It introduces general
quantitative specification theories based on disjunctive modal
transition systems~\cite{DBLP:conf/lics/LarsenX90} and
(non-deterministic) acceptance
automata~\cite{DBLP:journals/jacm/Hennessy85,
  DBLP:journals/entcs/Raclet08} on the one hand and abstract trace
distances on the other hand.

As in Chapter~\ref{ch:wm2}, specification labels are partially ordered
by a label refinement relation $\preceq$, and implementation labels
are those specification labels which cannot be further refined.  We
also assume partial conjunction and synchronization operators on
labels and work with specification-labeled disjunctive modal
transition systems and acceptance automata.

Also as in Chapter~\ref{ch:wm2}, we assume given a recursively
specified distance on specification traces, which takes values in a
\emph{commutative quantale}: a complete lattice $\LL$ together with a
commutative operation $\oplus_\LL$ which distributes over arbitrary
suprema.  We then generalize the translations between DMTS, acceptance
automata, and the modal $\nu$-calculus from Chapter~\ref{ch:dmts} to
our general setting and show in Theorem~\ref{th:trans-moddist} that
they respect modal refinement distances: denoting the translations by
$\da$, $\ddh$ etc.,
\begin{align*}
  \md( \mcalD_1, \mcalD_2) &= \md( \da( \mcalD_1), \da( \mcalD_2)), \\
  \md( \mcalA_1, \mcalA_2) &= \md( \ad( \mcalA_1), \ad( \mcalA_2)), \\
  \md( \mcalD_1, \mcalD_2) &= \md( \ddh( \mcalD_1), \ddh( \mcalD_2)), \\
  \md( \mcalN_1, \mcalN_2) &= \md( \hd( \mcalN_1), \hd( \mcalN_2)).
\end{align*}

We then turn to the quantitative properties of the operations and show
in Theorem~\ref{th:condis-q} that disjunction is quantitatively sound
and complete in the sense that
$\md( \mcalS_1\lor \mcalS_2, \mcalS_3)= \max( \md( \mcalS_1,
\mcalS_3), \md( \mcalS_2, \mcalS_3))$ for all specifications
$\mcalS_1$, $\mcalS_2$ and $\mcalS_3$.  Conjunction on the other hand
is only quantitatively sound, for the same reasons as exposed in
Chapter~\ref{ch:weightedmodal}.  Assuming a uniform bound on label
synchronization, we again derive a quantitative version of independent
implementability in Theorem~\ref{th:indimp-q}.  We also show in
Theorem~\ref{th:quot-q} that with our new generalized definition of
DMTS quotient, it holds that
$\md( \mcalS_1\| \mcalS_2, \mcalS_3)= \md( \mcalS_2, \mcalS_3/
\mcalS_1)$ for all specifications $\mcalS_1$, $\mcalS_2$ and
$\mcalS_3$.

Chapter~\ref{ch:dmts2} is based on work by Jan K{\v
  r}et{\'\i}nsk{\'y}, Axel Legay, Louis-Marie Traonouez, and the
author, which has been presented at the 11th International Symposium
on Formal Aspects of Component Software (FACS)
\cite{DBLP:conf/facs2/FahrenbergKLT14} and subsequently published in
Soft Computing~\cite{DBLP:journals/soco/FahrenbergKLT18}.

\section{Applications}

Our theory of quantitative specification and verification has found
applications in robustness of real-time systems, feature interactions
in software product lines, compatibility of service interfaces, text
separation, and other areas.  We present four such applications here.

\subsection{A Robust Specification Theory for Modal Event-Clock
  Automata}
\label{se:mecs}

The paper~\cite{conf/fit/FahrenbergL12}, written by Axel Legay and the
author and presented at the Fourth Workshop on Foundations of
Interface Technologies, contains an application of the general
quantitative framework of this thesis in the area of real-time
specifications.  We define a notion of robustness for the modal
event-clock specifications (MECS)
of~\cite{DBLP:conf/icfem/BertrandLPR09,
  DBLP:journals/scp/BertrandLPR12}.

We propose a new version of refinement for MECS which is adequate to
reason on MECS in a robust manner.  We then proceed to exhibit the
properties of the standard operations of specification theories:
conjunction, structural composition and quotient, with respect to this
quantitative refinement.  We show that structural composition and
quotient have properties which are useful generalizations of their
standard Boolean properties, hence they can be employed for robust
reasoning on MECS.  Conjunction, on the other hand, is generally not
robust, but together with the new operator of quantitative widening
can be used in a robust manner.

MECS are modal transition systems in which \may- and \must-transitions
are labeled with symbols from a set $\Sigma$ and annotated with
constraints which are used to enable or disable transitions depending
on the values of real variables.  In the language of
Section~\ref{se:intro.wm2}, their semantics is given as SMTS over the
set
\begin{equation*}
  \Spec=( \Sigma\times\{[ 0, 0]\})\cup(\{ \delta\}\times \I)\subseteq(
  \Sigma\cup\{ \delta\})\times \I
\end{equation*}
of specification labels.  Here
$\I=\{[ x, y]\mid x\in \Realnn, y\in \Realnn\cup\{ \infty\}, x\le y\}$
is the set of closed extended non-negative real intervals, and
$\delta\notin \Sigma$ denotes a special symbol which signifies passage
of time.

The partial order on $\Spec$ is given by
$( a,[ l, r])\labpre( a',[ l', r'])$ iff $a= a'$, $l\ge l'$, and
$r\le r'$ (hence $[ l, r]\subseteq[ l', r']$).
Thus the implementation labels are
$\Imp= \Sigma\times\{ 0\}\cup\{ \delta\}\times \Realnn$, so that
implementations are usual timed transition systems with discrete
transitions $s\mustto{ a, 0} s'$ and delay transitions
$s\mustto{ \delta, d} s'$.

We use the maximum-lead distance to measure differences between timed
traces.  For structural composition, we employ CSP-style label
synchronization and intersection of timing intervals; hence in a
composition, the timing constraints are conjunctions of the
components' constraints.  We then show that structural composition is
bounded and that conjunction is relaxed bounded; the quotient operator
is similarly well-behaved.

\subsection{Measuring Global Similarity between Texts}

The paper~\cite{DBLP:conf/slsp/FahrenbergBCJKL14}, written by Fabrizio
Biondi, Kevin Corre, Cyrille J{\'e}gourel, Simon Kongsh{\o}j, Axel
Legay, and the author and presented at the Second International
Conference on Statistical Language and Speech Processing, contains an
application of some of the theory presented here to a problem in
statistical natural-language processing.  We introduce a new type of
distance between texts and show that it can be used to separate
different classes in corpuses of scientific papers.

We measure the similarity of two texts using a discounted accumulating
distance.  Given two texts $A=( a_1, a_2,\dots, a_{ N_A})$ and
$B=( b_1, b_2,\dots, b_{ N_B})$, seen as finite sequences of words
(and hence stripped of punctuation), we first define an indicator
function $\delta_{ i, j}$, for $i, j\ge 0$, by
\begin{equation*}
  \delta_{ i, j}=
  \begin{cases}
    0 &\text{if } i\le N_A, j\le N_B\text{ and } a_i= b_j\,,\\
    1 &\text{otherwise}\,,
  \end{cases}
\end{equation*}
and then
\begin{equation*}
  \dphrase( i, j, \lambda)= \sum_{ k= 0}^\infty \lambda^k \delta_{ i+
    k, j+ k}\,,
\end{equation*}
for a discounting factor $\lambda\in \Realnn$ with $\lambda< 1$.  This
measures how much the texts $A$ and $B$ ``look alike'' when starting
with the tokens $a_i$ in $A$ and $b_j$ in $B$.  This \emph{position
  match} distance is then summarized and symmetrized as follows:
\begin{gather*}
  d'( A, B, \lambda)= \frac1{ N_A} \smash[t]{\sum_{ i= 1}^{ N_A}}
  \min_{ j= 1,\dotsc, N_B} \dphrase( i, j, \lambda) \\
  d( A, B, \lambda)= \max( d'( A, B, \lambda), d'( B, A, \lambda))
\end{gather*}

We have implemented this computation and then used this implementation
to statistically separate different types of scientific papers.  In a
first experiment, we successfully separate $42$ scientific papers from
$8$ automatically generated ``fake'' scientific papers (using the tool
SCIgen\footnote{\url{http://pdos.csail.mit.edu/scigen/}}).  With very
high discounting, we also achieve a classification where papers which
share authors or are otherwise similar are classified as such.  In a
second experiment, we compare 97 scientific papers with 100 ``fake''
ones generated by different methods.  Also here we achieve a complete
classification.  For high discounting factors, our classifications are
better than those achieved by other work using bag-of-words distances.

\subsection{Measuring Behavior Interactions between Product-Line
  Features}

The paper~\cite{DBLP:conf/icse/AtleeFL15}, written by Joanne M.~Atlee,
Axel Legay and the author and presented at the 3rd IEEE/ACM FME
Workshop on Formal Methods in Software Engineering, suggests a new
method for measuring the degree to which features interact in software
product lines.

The paper first introduces a distance between labeled transition
systems which is similar to the (undiscounted) accumulating simulation
distance, except that every pair of states is only treated once.  That
is, the function computing $d( s, s')$ tries to match every transition
$s\tto a t$ in the first system $\cal S$ with a transition
$s'\tto a t'$ in the second system $\cal S'$.  If no such exists, a
missing behavior is detected and $1$ is added to the score; if there
are transitions $s'\tto a t'$, then distance is recursively computed
for the pair $t, t'$ with the best match.  Once a pair of states has
been checked for behavior mismatches in this way, it is added to a
\textit{Passed} list of states which need not be checked again.

We model software product lines using featured transition systems,
which are transition systems in which transitions are conditioned on
the presence or absence of distinct features.  A product is then
simply a set of features, and a product $p$ has a behavior interaction
with a feature $f$ in a given featured transition system $\cal S$ if
the projection onto $p$ of $\cal S$ and the projection onto $p$ of the
projection onto $p\cup\{ f\}$ of $\cal S$ are not bisimilar.

We then generalize this notion to a behavior interaction distance,
using the above distance between transition systems.  We give evidence
that this is a useful notion to assess the degree of feature
interactions and show that it can be efficiently computed on the given
featured transition system, without resorting to the projections.

\subsection{Compatibility Flooding: Measuring Interaction of
  Behavioral Models}

The paper~\cite{DBLP:conf/sac/OuederniFLS17}, written by Meriem
Ouederni, Axel Legay, Gwen Sala{\"u}n, and the author and presented at
the 32nd ACM SIGAPP Symposium on Applied Computing, deals with
compatibility verification of service interfaces, focusing on the
interaction protocol level.

Checking the compatibility of interaction protocols is a tedious and
hard task, even though it is of utmost importance to avoid run-time
errors, \eg~deadlock situations or unmatched messages.  Most of the
existing approaches return a ``True'' or ``False'' result to detect
whether services are compatible or not, but for many issues such a
Boolean answer is not very helpful.  In real world situations, there
will seldom be a perfect match, and when service protocols are not
compatible, it is useful to differentiate between services that are
slightly incompatible and those that are totally incompatible.  Our
paper aims at quantifying the compatibility degree of service
interfaces, taking a semantic point of view.

Incompatibilities are measured between transition systems modeling
service interfaces, using a version of discounted accumulating
bisimulation distance where differences are propagated both forward
and backwards.  The distance takes into account the compatibility of
parameters and labels and is defined for two different scenarios, one
in which all sent and received messages must be matched, and an
asymmetric one where one of the components may send and receive other
messages which are irrelevant for the composition.

\section{Conclusion and Perspectives}

We have developed a general theory of quantitative verification and
quantitative specification theories.  The theory is independent of how
precisely quantitative differences are measured and applicable to a
large class of distances used in practice.  The quantative spefication
formalism introduced in the last Chapter \ref{ch:dmts2} is also rather
robust, admitting translations between several different specification
formalisms, and has good algebraic and geometric properties.

On a theoretical level, the above is motivation to concern oneself
with the question what precisely \emph{is} a specification theory.
While there is some agreement to this at the qualitative / Boolean
level, it is not clear how to extend this to the quantitative world.
This question is important not only theoretically, but also in
applications, given that the algebraic properties of a formalism
determine how precisely it can be used in practice.

Somewhat related to the question above is the problem of how to treat
\emph{silent} or spontaneous transitions.  In applications it is
common to model uncertainty or ambiguity with silent transitions, and
these are rather well-understood in the qualitative setting; but
again it is unclear how to lift them to the quantitative world.

Further, and taking a more applied view, it is somewhat problematic
that all the formalisms treated here are based on discrete
\emph{transition systems}.  When considering applications in real-time
or hybrid systems, discreteness is not sufficient and some treatment
of continuous time is required.  There is some work on specification
theories for real-time systems, but for hybrid systems these are
lacking, and in any case it is unclear how to relate them to the
quantitative specification theories we have exposed here.

Below we treat the questions and problems above in some more detail
and try to show some avenues for further work on these subjects.

\subsection{Specification Theories}

The work presented here has led to more fundamental questions as to
what precisely \emph{is}, or \emph{should be}, a specification theory.
This is what we set out to answer, together with Axel Legay,
in~\cite{DBLP:conf/sofsem/FahrenbergL17}, presented at the 43rd
International Conference on Current Trends in Theory and Practice of
Computer Science (SOFSEM) and subsequently published in the Journal of
Logical and Algebraic Methods in
Programming~\cite{journals/jlamp/FahrenbergL20}, and
in~\cite{DBLP:conf/isola/FahrenbergL20}, to be presented at the 2021
ISoLA Symposium.

We propose here that a specification theory for a set of models
$\mfra M$ consists of the following ingredients:
\begin{itemize}
\item a set $\mfra S$ of specifications;
\item a mapping $\chi: \mfra M\to \mfra S$; and
\item a refinement preorder $\le$ on $\mfra S$ which is an equivalence
  relation on the image of $\chi$ in $\mfra S$.
\end{itemize}
It then follows that for all $M\in \mfra M$, $\chi( M)$ is the
\emph{characteristic formula}~\cite{DBLP:conf/icalp/Pnueli85} for $M$.

Logical operations on specifications are then obtained by asserting
that $\mfra S$ forms a \emph{bounded distributive lattice} up to
$\equiv$, the equivalence on $\mfra S$ defined by $S_1\equiv S_2$ iff
$S_1\le S_2$ and $S_2\le S_1$.  Structural composition and quotient
are defined by an extra operation $\|$ on specifications which turns
$\mfra S$ into a (bounded distributed) \emph{commutative residuated
  lattice}.  This puts specification theories into a well-understood
algebraic context, see for example~\cite{JipsenT02}, which also
appears in linear logic~\cite{DBLP:journals/tcs/Girard87} and other
areas.

It is an open question how to transfer this algebraic point of view to
the \emph{quantitative} setting.  It is clear that the refinement
order above should be replaced by a hemimetric $d$ on $\mfra S$, and
also that $d$ should be symmetric on the image of $\chi$ in $\mfra S$;
but we do not know how to correctly introduce characteristic formulae
into this setting.

\subsection{Silent Transitions}

Another open question is how to deal with \emph{silent transitions} in
the quantitative setting.  Van~Glabbeek defines a
linear-time--branching-time spectrum for ``processes with silent
moves'' in~\cite{DBLP:conf/concur/Glabbeek93}, but it is unclear how
to translate this into a game framework in order to replicate the work
contained in Chapter~\ref{ch:qltbt}.

Silent transitions are obtained whenever two processes synchronize on
internal actions, so they are important from an application point of
view.  Recent advances on coalgebraic approaches to silent
transitions~\cite{DBLP:journals/corr/Brengos13} and on codensity
games~\cite{DBLP:conf/lics/KomoridaKHKH19} appear to provide a way
forward.

\subsection{Applications}

Much work is to be done in order to apply our work to real-time,
hybrid, or embedded systems.  Specification theories for real-time and
probabilistic systems do exist, see below, but they all have problems
with robustness.  For hybrid systems, no work on compositional
specification theories seems to be available.  Generally speaking, the
problem with real-time and hybrid systems is that time itself provides
an implicit synchronization mechanism, so compositionality is
difficult to achieve for real-time and hybrid systems.

\subsubsection{Modal event-clock specifications}

We have already mentioned modal event-clock specifications (MECS) in
Section~\ref{se:mecs}.  Introduced
in~\cite{DBLP:conf/icfem/BertrandLPR09,
  DBLP:journals/scp/BertrandLPR12}, these form a specification theory
for event-clock automata~\cite{DBLP:journals/tcs/AlurFH99}, a
determinizable subclass of timed
automata~\cite{DBLP:journals/tcs/AlurD94}, under timed bisimilarity.
Models and specifications are assume to be deterministic, thus
$\mcal S_1\le \mcal S_2$ iff
$\Mod{ \mcal S_1}\subseteq \Mod{ \mcal S_2}$ in this case.

In~\cite{DBLP:journals/scp/BertrandLPR12} it is shown that MECS admit
a conjunction, thus forming a meet-semilattice up to~$\equiv$.  The
authors also introduce composition and quotient; but computation of
quotient incurs an exponential blow-up.  Using the maximum-lead
distance to measure differences between timed traces, we have shown
in~\cite{conf/fit/FahrenbergL12} how to develop a framework for robust
quantitative reasoning.

\subsubsection{Timed input/output automata}

\cite{DBLP:journals/sttt/DavidLLNTW15,
  DBLP:journals/sttt/DavidLLMNRSW12}~introduce a specification theory
based on a variant of the timed input/output automata (TIOA)
of~\cite{DBLP:series/synthesis/2010Kaynar,
  DBLP:conf/rtss/KaynarLSV03}.  Both models and specifications are
TIOA which are action-deterministic and input-enabled; but models are
further restricted using conditions of output urgency and independent
progress.  The equivalence on models being specified is timed
bisimilarity.

In~\cite{DBLP:journals/sttt/DavidLLNTW15} it is shown that TIOA
admit a conjunction.  The paper also introduces a composition
operation and a quotient, but the quotient is only shown to satisfy
the property that
\begin{equation*}
  \mcal S_1\| \mcal M\le \mcal S_3 \liff \mcal M\le \mcal S_3/ \mcal
  S_1
\end{equation*}
for all specifications $\mcal S_1, \mcal S_3$ and all \emph{models}
$\mcal M$.  No robust quantitative specification theories for TIOA are
available.

\subsubsection{Abstract probabilistic automata}

Abstract probabilistic automata (APA), introduced
in~\cite{DBLP:journals/iandc/DelahayeKLLPSW13,
  journals/lmcs/DelahayeFLL14}, form a specification theory for
probabilistic automata~\cite{DBLP:journals/njc/SegalaL95} under
probabilistic bisimilarity.  They build on earlier models of interval
Markov chains (IMC)~\cite{DBLP:journals/jlp/DelahayeLLPW12}, see
also~\cite{DBLP:journals/tcs/BartDFLMT18,
  DBLP:conf/vmcai/DelahayeLP16} for a related line of work.

In~\cite{DBLP:journals/iandc/DelahayeKLLPSW13} it is shown that APA
admit a conjunction, but that IMC do not.  Also a composition is
introduced in~\cite{DBLP:journals/iandc/DelahayeKLLPSW13}, and it is
shown that composing two APA with interval constraints (hence, IMCs)
may yield an APA with \emph{polynomial} constraints (not an IMC); but
APA with polynomial constraints are closed under composition.  No
robust quantitative specification theories for APA are available.

\section{About the Author}

Ulrich (Uli) Fahrenberg holds a PhD in mathematics from Aalborg
University, Denmark.  For his thesis, which he defended in 2005, he
worked in algebraic topology and its applications in concurrency
theory.  His work was supervised by Lisbeth Fajstrup and Martin
Raussen, and his thesis bore the title ``Higher-Dimensional Automata
from a Topological Viewpoint''.

After his PhD, Fahrenberg started a career in computer science as an
assistant professor at Aalborg University.  During this time, he
worked with Kim G.~Larsen and others on weighted timed automata and
quantitative verification.  The work reported in this thesis was
started together with Kim G.~Larsen and Fahrenberg's PhD student Claus
Thrane while Fahrenberg was at this position.  In 2010, Fahrenberg
passed his University Teacher Education for Assistant Professors, the
prerequisite for holding a permanent post at a Danish university.

From 2010 to 2016, Fahrenberg has worked as a postdoc at Inria Rennes,
France, in the group of Axel Legay.  During this time, he deepened his
work in quantitative analysis and verification and started his work in
quantitative specification theories.  Between 2016 and 2021 Fahrenberg
was a researcher at the computer science lab at {\'E}cole
polytechnique in Palaiseau, France, and since 2021 he is associate
professor at EPITA Rennes.  His current research centers on the theory
of concurrent, distributed, and hybrid systems.

Since 2001, Fahrenberg has published 96 scientific contributions,
among which 24 papers in peer-reviewed international journals and 47
papers in peer-reviewed international conference or workshop
proceedings.  He has been a member of numerous program committees, and
since 2016 he is a reviewer for AMS Mathematical Reviews.  He is a
member of the Steering Committee of the RAMiCS international
conferences and has been PC co-chair of RAMiCS-2020 and RAMiCS-2021.

Fahrenberg has been co-supervisor for one PhD student and three
Masters students.  He has supervised two PhD students' internships and
two Masters students' internships and also taught a number of courses
both in mathematics and computer science.

\section*{Bibliography}

\subsection*{Refereed Journal Papers}

\begin{bibunit}[plainyr-rev]
  \nocite{journals/jhrs/FahrenbergR07, DBLP:journals/jlp/ThraneFL10,
    DBLP:journals/cacm/BouyerFLM11, journals/cai/FahrenbergLT10,
    DBLP:journals/tcs/LarsenFT11, DBLP:journals/jlp/LuMMRFL12,
    DBLP:journals/fmsd/BauerFJLLT13, DBLP:journals/tcs/FahrenbergL14,
    DBLP:journals/acta/FahrenbergL14, journals/ijac/AllamigeonFGKL13,
    journals/lmcs/DelahayeFLL14, DBLP:journals/scp/LePFL16,
    DBLP:journals/tecs/LePFL16, DBLP:journals/actaC/EsikFLQ17,
    DBLP:journals/actaC/EsikFLQ17a, DBLP:journals/fac/BacciBFLMR21,
    DBLP:journals/soco/FahrenbergKLT18,
    DBLP:journals/lmcs/CacheraFL19, journals/jlamp/FahrenbergL20,
    DBLP:journals/tcs/FahrenbergLQ20, journals/sttt/FahrenbergL19,
    DBLP:journals/iandc/BenesFKLT20,
    DBLP:journals/lmcs/FahrenbergJTZ21,
    journals/actacyb/FahrenbergJTZ21, journals/mscs/FahrenbergJSZ21}

  \putbib[bib]
\end{bibunit}

\subsection*{Refereed Conference and Workshop Publications}

\begin{bibunit}[plainyr-rev]
  \nocite{fah02-getco, fah03-getco, fah05-hda,
    DBLP:journals/entcs/FahrenbergL09, DBLP:conf/formats/BouyerFLMS08,
    FLT10-Kripke-MEMICS, FL09-QAPL, FLT09-FSEN,
    DBLP:conf/hybrid/BouyerFLM10, conf/qapl/FahrenbergTL11,
    DBLP:conf/models/FahrenbergLW11, DBLP:conf/mfcs/BauerFJLLT11,
    DBLP:conf/ictac/FahrenbergJLS11, DBLP:conf/fsttcs/FahrenbergLT11,
    DBLP:conf/csr/BauerFLT12, conf/fit/FahrenbergL12,
    DBLP:conf/forte/DelahayeFHLN12, conf/mfps/FahrenbergL13,
    DBLP:conf/qest/DelahayeFLL13, conf/acsd/LePFL13,
    DBLP:conf/concur/BenesDFKL13, DBLP:conf/atva/EsikFLQ13,
    DBLP:conf/aplas/FahrenbergL13, DBLP:conf/fase/FahrenbergALW14,
    conf/foclasa/LePFL13b, DBLP:conf/facs2/FahrenbergKLT14,
    DBLP:conf/ictac/FahrenbergLT14, DBLP:conf/slsp/FahrenbergBCJKL14,
    fah04-getco-pre, DBLP:conf/models/FahrenbergL14,
    DBLP:conf/calco/FahrenbergL15, DBLP:conf/icse/AtleeFL15,
    DBLP:conf/models/AtleeBFL15, DBLP:conf/fsttcs/CacheraFL15,
    DBLP:journals/corr/EsikFL15a, DBLP:conf/dlt/EsikFL15,
    DBLP:conf/splc/OlaecheaFAL16, DBLP:conf/sofsem/FahrenbergL17,
    DBLP:conf/icse/FahrenbergL17, DBLP:conf/sac/OuederniFLS17,
    DBLP:conf/adhs/Fahrenberg18, conf/fm/BacciLMBFR18,
    conf/nik/FahrenbergL18, DBLP:conf/icse/OlaecheaALF18,
    DBLP:conf/ictac/FahrenbergLQ19, DBLP:conf/RelMiCS/FahrenbergJST20,
    DBLP:conf/isola/FahrenbergL20, DBLP:conf/tase/FahrenbergL21,
    DBLP:conf/RelMiCS/CalkFJSZ21}

  \putbib[bib]
\end{bibunit}

\subsection*{Conference Abstracts}

\begin{bibunit}[plainyr-rev]
  \nocite{fah01-nwpt, UF-08-ACCAT-Ext, conf/wata/FahrenbergLQ12,
    conf/nwpt/DyhrbergLMRF10, conf/wata/Fahrenberg16,
    conf/wata/FahrenbergLQ16, conf/wata/FahrenbergLQ14,
    conf/nwpt/Fahrenberg16, conf/nwpt/Fahrenberg18, UF-Thrane-08-WTS,
    conf/wata/FahrenbergJST21}

  \putbib[bib]
\end{bibunit}

\subsection*{Books}

\begin{bibunit}[plainyr-rev]
  \nocite{getco03proc, DBLP:conf/formats/2011, conf/qfm/2012,
    DBLP:conf/RelMiCS/2020, DBLP:conf/RelMiCS/2021}

  \putbib[bib]
\end{bibunit}

\subsection*{Book Chapters}

\begin{bibunit}[plainyr-rev]
  \nocite{FLT09-Markto, FLT11-Markto, book/FahrenbergLL13,
    series/natosec/FahrenbergLLT13, conf/sifakis/FahrenbergLLT14,
    conf/sifakis/FahrenbergLT14, DBLP:reference/mc/BouyerFLMO018,
    DBLP:series/natosec/LarsenFL17}
  \putbib[bib]
\end{bibunit}

\subsection*{Theses}

\begin{bibunit}[plainyr-rev]
  \nocite{thesis/Fahrenberg05, thesis/Fahrenberg02}
  \putbib[bib]
\end{bibunit}

\section{Acknowledgments}

The author would like to thank all coauthors involved in the papers
which form the basis for this thesis.  In alphabetical order, these are
\begin{inparablank}
\item Sebastian S.~Bauer, Munich, Germany;
\item Nikola Bene{\v s}, Brno, Czechia;
\item Line Juhl, Aalborg, Denmark;
\item Jan K{\v r}et{\'\i}nsk{\'y}, Munich, Germany;
\item Kim G.~Larsen, Aalborg, Denmark;
\item Axel Legay, Louvain-la-Neuve, Belgium;
\item Claus Thrane, Copenhagen, Denmark;
\item Louis-Marie Traonouez, Rennes, France.
\end{inparablank}

\chapter[Quantitative Analysis of Weighted Transition
Systems][Quantitative Analysis of Weighted Transition
Systems]{Quantitative Analysis of Weighted Transition
  Systems\footnote{This chapter is based on the journal
    paper~\cite{DBLP:journals/jlp/ThraneFL10} published in the Journal
    of Logic and Algebraic Programming.}}
\label{ch:wtsjlap}

This chapter introduces a notion of \emph{weighted transition system}
(WTS), essentially an extension of the standard concept of (labeled)
transition system~\cite{report/daimi/Plotkin81} which has been used to
introduce operational semantics for a wide range of systems.  It then
proceeds to define what is meant by \emph{linear} and \emph{branching
  distances} between such WTS, and introduces three examples of such
distances: the \emph{point-wise}, \emph{accumulated}, and
\emph{maximum-lead} distances.  Finally, it is shown that linear
distances are bounded by branching distances, and some results on
topological inequivalence are provided.

\section{Weighted transition systems}
\label{wtsjlap.sec:wts}

The intention of WTS is to describe a system's behavior as well as
quantitative properties in terms of \emph{transition weights}.  Recall
that a transition system is a quadruple $(S, s_{0}, \Sigma, R)$
consisting of a set $S$ of states with initial state $s_0\in S$, a
finite set $\Sigma$ of labels, and a set of transitions
$R\subseteq S\times \Sigma\times S$.

\begin{definition}
  \label{wtsjlap.def:wts}
  A \emph{weighted transition system} is a tuple $(S, s_{0}, \Sigma,
  R, w)$, where
  \begin{itemize}
  \item $( S, s_{0}, \Sigma, R)$ is a transition system, and
  \item $w : R \to \Realnn$ assigns weights to transitions.
  \end{itemize}
\end{definition}

We write $s \xrightarrow{\alpha, w} s'$ whenever $(s,\alpha,s') \in
R$ and $w(s, \alpha, s') = w$,
and $s \not \rightarrow$ if there is no transition $(s,\alpha,s')$ in
$R$ for any $\alpha$ and $s'$.

We lift the standard notions of path and trace to WTS:

\begin{definition}
  \label{wtsjlap.def:wts_traces}
  Let $\mathcal{S}=( S, s_0,\Sigma, R, w)$ be a WTS and
  $s\in S$.  A \emph{path from $s$} in $\mathcal{S}$ is a (possibly
  infinite) sequence $((s_0,\alpha_0,s_1), (s_1 \alpha_1,s_2), \dots)$ of
  transitions $(s_i,\alpha_i,s_{i+1}) \in R$ with $s_0= s$.  A
  \emph{(weighted) trace from $s$} is a sequence $(( \alpha_0,
  w_0),( \alpha_1, w_1), \dots)$ of pairs $( \alpha_i,
  w_i)\in \Sigma\times \Realnn$ for which there exists a path
  $((s_0,\alpha_0,s_1), (s_1 \alpha_1,s_2), \dots)$ from $s$ for which
  $w_i= w( s_i, \alpha_i, s_{ i+ 1})$.
\end{definition}

The set of traces from a state $s$ is denoted $\tracesfrom s$.  Given
a trace $\sigma$, we denote by $U( \sigma)\in \Sigma^\omega$ its
label sequence (\ie the associated unweighted trace), 
and by $\sigma_i$ its $i$'th label-weight pair.

\section{Quantitative Analysis}
\label{wtsjlap.sec:analysis}

In this section we introduce our quantitative analysis of WTS, both in
a linear and in a branching setting.  For ease of exposition we
concentrate on \emph{trace inclusion} and \emph{simulation} here and
defer treatment of both trace equivalence and bisimulation to other
work.  We shall introduce three different quantitative notions of
trace inclusion and of simulation, all filling in the gap between the
\emph{unweighted} and the \emph{weighted} relations, which we recall
below:

\begin{definition}
  \label{wtsjlap.de:simulation}
  Let $( S, s_{0}, \Sigma, R, w)$ be a WTS.  A relation
  $R \subseteq S \times S$ is
  \begin{itemize}
  \item an \emph{unweighted simulation} provided that for all $( s,
    t)\in R$ and $s\tto{ \alpha, c} s'$, also $t\tto{ \alpha, d} t'$ for
    some $d\in \Realnn$ and $( s', t')\in R$,
  \item a \emph{(weighted) simulation} provided that for all $( s,
    t)\in R$ and $s\tto{ \alpha, c} s'$, also $t\tto{ \alpha, c} t'$ for
    some $( s', t')\in R$.
  \end{itemize}
  We write
  \begin{itemize}
  \item $s\uwsim t$ if $( s, t)\in R$ for some unweighted simulation
    $R$,
  \item $s\wsim t$ if $( s, t)\in R$ for some weighted simulation $R$.
  \end{itemize}
  Also, we write
  \begin{itemize}
  \item $s\le^u t$ if $U( \tracesfrom s)\subseteq U( \tracesfrom t)$,
  \item $s\le t$ if $\tracesfrom s\subseteq \tracesfrom t$.
  \end{itemize}
\end{definition}

We shall fill in the gap between unweighted and weighted relations
using (asymmetric) \emph{distance functions}
$d: S\times S\to \Realnn\cup\{ \infty\}$.  Any of the distances
defined below will obey the properties given in the following
definition.

\begin{definition}
  \label{wtsjlap.def:dist-require}
  A hemimetric $d: S\times S\to \Realnn\cup\{ \infty\}$ defined on the
  states of a WTS $( S, s_{0}, \Sigma, R, w)$ is called
  \begin{itemize}
  \item a \emph{linear distance} if $s\le t$ implies $d( s, t)= 0$
    and $s\not\le^u t$ implies $d( s, t)= \infty$,
  \item a \emph{branching distance} if $s\wsim t$ implies
    $d( s, t)= 0$ and $s\not\uwsim t$ implies $d( s, t)= \infty$.
  \end{itemize}
\end{definition}

As usual, we can generalize distances between states of a single WTS
to distances between two different WTS by taking their disjoint
union.

Our distance functions are essentially based on three different
metrics on the set of sequences of real numbers.  Throughout this
work, these are referred to as
\emph{point-wise}~\eqref{wtsjlap.eq:sequence_pw},
\emph{accumulated}~\eqref{wtsjlap.eq:sequence_ac}, and
\emph{maximum-lead}~\eqref{wtsjlap.eq:sequence_pm} distances,
respectively.  For sequences $a=( a_i)$, $b=( b_i)$ these are defined
as follows:
\begin{align}
  d_\bullet( a, b) &= \sup_i\big\{ |a_i-
  b_i|\big\} \label{wtsjlap.eq:sequence_pw}\\
  d_+( a, b) &= \sum_i |a_i- b_i| \label{wtsjlap.eq:sequence_ac}\\
  d_\pm( a, b) &= \sup_i \Big\{ \Bigl| \sum^i_{j=0} a_j- \sum^i_{j=0}
  b_j\Bigr|\Big\} \label{wtsjlap.eq:sequence_pm}
\end{align}
The intuition behind these metrics is that $d_\bullet$ measures the
largest individual difference of sequence entries, $d_+$ measures the
accumulated sum of (the absolute values of) the entries' differences,
and $d_\pm$ measures the largest \emph{lead} of one sequence over the
other, \ie the maximum difference in accumulated values.  Hence the
maximum-lead distance of two sequences is the same as the point-wise
distance of their \emph{partial-sum} sequences.

Besides the above three, other metrics on sequences of reals are also
of interest, and we will see in Chapter~\ref{ch:qltbt} that linear and
branching distances of WTS based on these other metrics can be
developed similarly to the ones we introduce in this chapter.

In the following we will consider \textit{discounted} distances, where
the contribution of each step is decreased exponentially over
time.  To this end, we fix a discounting factor $\lambda\in[ 0, 1]$;
as extreme cases, $\lambda= 1$ means that the future is
undiscounted, and $\lambda= 0$ means that only the present is
considered.

Also, we fix a WTS $\mathcal{S}=( S, s_{0}, \Sigma, R, w)$.

\subsection{Linear distances}
\label{wtsjlap.sec:tracedistance}

We will now introduce our quantitative trace distances.

\begin{definition}
  \label{wtsjlap.de:dist-of-traces}
  For traces $\sigma$, $\tau$, the point-wise, accumulating, and
  maximum-lead \emph{trace distances} are given by
  $\pwtd{ \sigma, \tau}= \awtd{ \sigma, \tau}= \mltd{
    \sigma, \tau}= \infty$ if $U( \sigma)\ne U( \tau)$, and for
  $U( \sigma)= U( \tau)$,
  \begin{align*}
    \pwtd{\sigma,\tau} &= \sup_{i} \big\{ \lambda^i\,
    |\cost{\sigma_{i}} -\cost{\tau_{i}}|\big\}\,, \\
    \awtd{\sigma,\tau} &= \sum_{i} \lambda^i\,
    |\cost{\sigma_{i}} -\cost{\tau_{i}}|\,, \\
    \mltd{\sigma,\tau} &= \sup_{i} \Big\{
    \lambda^i\, \Bigl| \sum^{i}_{j=0}\cost{\sigma_{j}} -
    \sum^{i}_{j=0}\cost{\tau_{j}}\Bigr|\Big\}\,.
  \end{align*}
\end{definition}

Observe that the above distances on traces are symmetric; they are
indeed \emph{metrics} on the set of traces.  This is not the case when
lifted to states:

\begin{definition}
  \label{wtsjlap.def:trace_distances}
  For states $s, t\in S$, the point-wise, accumulating and
  maximum-lead \emph{linear distances} are given as follows:
  \begin{align*}
    \apwdist{s,t} &= \adjustlimits \sup_{\sigma \in \tracesfrom{s}}
    \inf_{\tau \in \tracesfrom{t}} \pwtd{\sigma,\tau}
    \\
    \aacdist{s,t} &= \adjustlimits \sup_{\sigma \in \tracesfrom{s}}
    \inf_{\tau \in \tracesfrom{t}} \awtd{\sigma,\tau} \\
    \apmdist{s,t} &= \adjustlimits \sup_{\sigma \in \tracesfrom{s}}
    \inf_{\tau \in \tracesfrom{t}} \mltd{\sigma,\tau}
  \end{align*}
\end{definition}

Note that this is precisely the Hausdorff-hemimetric construction,
hence it can be generalized to other distances between traces.  Also,
it is quite natural, \cf~Proposition~\ref{wtsjlap.prop:haus}.  It can
easily be shown that the distances defined above are indeed linear
distances in the sense of Definition~\ref{wtsjlap.def:dist-require}.

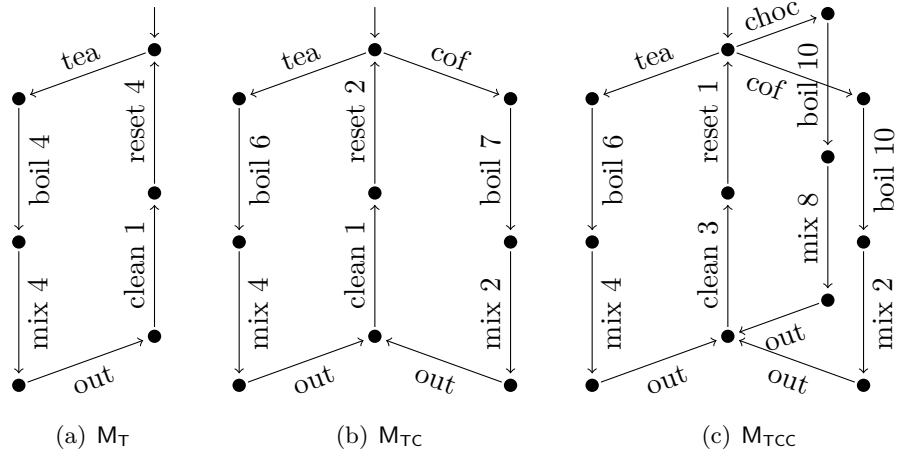
\begin{figure}[tbp]
  \centering
  \subbottom[$\T$]{%
    \begin{tikzpicture}[->,shorten >=1pt, auto, node distance=.75cm,
      initial text=]
      
      \begin{scope}[shape=circle, outer sep=1pt,minimum
        size=5pt,inner sep=0pt, node distance=1.9cm] 
        \node[fill] (t0) {};        
        \node[fill] (t1) at (200:1.9){};
        \node[fill] (t2) [below of=t1]{};
        \node[fill] (t3) [below of=t2] {};
        \node[fill] (t5) [below of=t0] {};
        \node[fill] (t4) [below of=t5] {};
      \end{scope}
      
      \node (start) at (90:.7){};
      \draw (start) -- (t0);
      
      \begin{scope}
        \draw (t0) -- node[above, rotate=20] {tea} (t1);
        \draw (t3) -- node[below, rotate=20] {out} (t4);
        \tikzstyle{every node}=[rotate=90]
        \draw (t1) -- node[below] {boil 4}  (t2);
        \draw (t2) -- node[below] {mix 4}   (t3);
        \draw (t4) -- node[above] {clean 1} (t5);
        \draw (t5) -- node[above] {reset 4} (t0);          
      \end{scope}
      
    \end{tikzpicture}
  }\qquad%
  \subbottom[$\TC$]{%
    \begin{tikzpicture}[->,shorten >=1pt, auto, node distance=.75cm,
      initial text=]
      
      \begin{scope}[shape=circle, outer sep=1pt,minimum
        size=5pt,inner sep=0pt, node distance=1.9cm] 
        \node[fill] (t0) {};        
        \node[fill] (t1) at (200:1.9){};
        \node[fill] (t2) [below of=t1]{};
        \node[fill] (t3) [below of=t2] {};
        \node[fill] (t5) [below of=t0] {};
        \node[fill] (t4) [below of=t5] {};
        \node[fill] (t6) at (340:1.9){};
        \node[fill] (t7) [below of=t6]{};
        \node[fill] (t8) [below of=t7]{};
      \end{scope}
      
      \node (start) at (90:.7){};
      \draw (start) -- (t0);

      \begin{scope}
        \draw (t0) -- node[above, rotate=20]  {tea} (t1);
        \draw (t3) -- node[below, rotate=20]  {out} (t4);
        \draw (t0) -- node[above, rotate=-20] {cof} (t6);
        \draw (t8) -- node[below, rotate=-20] {out} (t4);
        \tikzstyle{every node}=[rotate=90]
        \draw (t1) -- node[below] {boil 6}  (t2);
        \draw (t2) -- node[below] {mix 4}   (t3);
        \draw (t4) -- node[above] {clean 1} (t5);
        \draw (t5) -- node[above] {reset 2} (t0);
        
        \draw (t6) -- node[above] {boil 7} (t7);
        \draw (t7) -- node[above] {mix 2}  (t8);
      \end{scope}
      
    \end{tikzpicture}
  }\qquad%
  \subbottom[$\TCC$]{%
    
    \begin{tikzpicture}[->,shorten >=1pt, auto, node distance=.75cm,
      initial text=]
      
      \begin{scope}[shape=circle, outer sep=1pt,minimum
        size=5pt,inner sep=0pt, node distance=1.9cm] 
        \node[fill] (t0) {};
        \node[fill] (t1) at (200:1.9){};
        \node[fill] (t2) [below of=t1]{};
        \node[fill] (t3) [below of=t2] {};
        \node[fill] (t5) [below of=t0] {};
        \node[fill] (t4) [below of=t5] {};
        \node[fill] (t6) at (340:1.9){};
        \node[fill] (t7) [below of=t6]{};
        \node[fill] (t8) [below of=t7]{};
        \node[fill] (t9) at (20:1.4){};
        \node[fill] (t10) [below of=t9]{};
        \node[fill] (t11) [below of=t10]{};
      \end{scope}
      
      \node (start) at (90:.7){};
      \draw (start) -- (t0);
      
      \begin{scope}
        \tikzstyle{every node}=[rotate=20]
        \draw (t0) -- node[above]  {tea}  (t1);
        \draw (t3) -- node[below]  {out}  (t4);
        \draw (t0) -- node[above]  {choc} (t9);
        \draw (t11) -- node[below] {out}  (t4);
        \tikzstyle{every node}=[rotate=-20]
        \draw (t0) -- node[below, pos=.3]  {cof} (t6);
        \draw (t8) -- node[below]          {out} (t4);
        \tikzstyle{every node}=[rotate=90]
        \draw (t1) -- node[below] {boil 6}  (t2);
        \draw (t2) -- node[below] {mix 4}   (t3);
        \draw (t4) -- node[above] {clean 3} (t5);
        \draw (t5) -- node[above] {reset 1} (t0);
        
        \draw (t6) -- node[below] {boil 10}(t7);
        \draw (t7) -- node[below] {mix 2}   (t8);
        
        \draw (t9) -- node[above] {boil 10}  (t10);
        \draw (t10) -- node[above] {mix 8}  (t11);
      \end{scope}	            
    \end{tikzpicture}
  }%
  \caption{%
    \label{wtsjlap.fig:example2}
    Three beverage machines}
\end{figure}

\begin{example}
  To illustrate differences between the three linear distances
  introduced above, consider the three WTS models of beverage machines
  depicted in Figure~\ref{wtsjlap.fig:example2}; a Tea maker $\T$, a
  Tea and Coffee maker $\TC$ and a Tea, Coffee and Chocolate maker
  $\TCC$.  In the figure, edges without specified weight have weight
  $0$.

  The production of a beverage consists of six operations: Selecting
  the drink, boiling the water, mixing the beverage, outputting the
  finished product, self cleaning, and resetting. Each operation
  consumes a certain amount of power depending on its implementation
  by electrical components.  Weights thus model power consumption, and
  are given in such a way that in more powerful machines, some
  operations, for example boiling, require more power, whereas some
  other, for example resetting, require less.

  By design of the beverage machines, there are unweighted trace
  inclusions $\T \le^u \TC \le^u \TCC$; any behavior of a ``lesser''
  machine can be emulated qualitatively by a ``better'' one.  What is
  less obvious is how they compare in power consumption.

  Noting that any infinite behavior in the beverage machines consists
  of loops of width $6$, we can introduce some ad-hoc notation to
  simplify calculations.  Let $\apwdist{\T, \TC}^6$ denote point-wise
  distance from $\T$ to $\TC$ when only traces of length at most $6$
  are considered, and similarly for the other machines and distances.
  For a (realistic) discounting factor of $\lambda = .90$, the
  point-wise distances can be computed as follows:
  \begin{alignat*}{2}
    \apwdist{ \T, \TC} &= \sup_i \big\{ \apwdist{ \T, \TC}^6\cdot
    \lambda^{6i}\big\}= \apwdist{ \T, \TC}^6 & &= 1.80 \\
    \apwdist{ \T, \TCC} &= \sup_i \big\{ \apwdist{ \T, \TCC}^6\cdot
    \lambda^{6i}\big\}= \apwdist{ \T, \TCC}^6 & &= 1.80 \\
    \apwdist{ \TC, \TCC} &= \sup_i \big\{ \apwdist{ \TC, \TCC}^6\cdot
    \lambda^{6i}\big\}= \apwdist{ \TC, \TCC}^6 & &= 2.70
  \end{alignat*}
  For the accumulating distances:
  \begin{alignat*}{2}
    \aacdist{ \T, \TC} &= \sum_i \apwdist{ \T, \TC}^6\cdot
    \lambda^{6i} = \apwdist{ \T, \TC}^6 \frac1{1-\lambda^6} &
    &\approx 2.52\\
    \aacdist{ \T, \TCC} &= \sum_i \apwdist{ \T, \TCC}^6\cdot
    \lambda^{6i} = \apwdist{ \T, \TCC}^6 \frac1{1-\lambda^6} &
    &\approx 8.80\\
    \aacdist{ \TC, \TCC} &= \sum_i \apwdist{ \TC, \TCC}^6\cdot
    \lambda^{6i} = \apwdist{ \TC, \TCC}^6 \frac1{1-\lambda^6} &
    &\approx 7.41
  \end{alignat*}
  Similarly, the maximum-lead distances can be computed as follows:
  \begin{align*}
    \apmdist{\T,\TC} &\approx
    1.62\\
    \apmdist{\T,\TCC} &\approx
    2.62\\
    \apmdist{\TC,\TCC} &\approx 3.34
  \end{align*}
  \qedhere
\end{example}

The following lemma provides recursive bounds on linear distances and
will be useful as motivation for the definition of branching distance
below.  For the bound on the maximum-lead distance, we introduce a
generalization of $\apmdists$ by
\begin{align*}
  \apmdist{s,t}( \delta) &= \adjustlimits \sup_{\sigma \in \tracesfrom{s}}
  \inf_{\tau \in \tracesfrom{t}} \apmdist{\sigma,\tau}(
  \delta)\,, \\
  \apmdist{\sigma,\tau}( \delta) &= \sup_{i} \Big\{
  \lambda^i\, \Bigl| \delta+ \sum^{i}_{j=0}\cost{\sigma_{j}} -
  \sum^{i}_{j=0}\cost{\tau_{j}}\Bigr|\Big\}\,.
\end{align*}
Here $\delta\in \Real$ is the \emph{lead} which $\sigma$ has already
acquired over $\tau$; hence
$\apmdist{\sigma,\tau}( 0)= \apmdist{\sigma,\tau}$ and
$\apmdist{s,t}( 0)= \apmdist{s,t}$.

\begin{lemma}
  \label{wtsjlap.le:dist-fix}
  For states $s, t\in S$,
  \begin{align*}
    \apwdist{s,t} &\le \adjustlimits \sup_{s\tto{ \alpha, c} s'} \inf_{t\tto{
        \alpha, d} t'} \max\big(| c- d|, \lambda\, \apwdist{ s', t'}\big)\,,
    \\
    \aacdist{s,t} &\le \adjustlimits \sup_{s\tto{ \alpha, c} s'} \inf_{t\tto{
        \alpha, d} t'} |c- d|+ \lambda\,
    \aacdist{s', t'}\,,
    \\
    \apmdist{s,t}( \delta) &\le \adjustlimits \sup_{s\tto{ \alpha,
        c} s'} \inf_{t\tto{ \alpha, d} t'} \max\Big( | \delta|,
    \lambda\, \apmdist{ s', t'}\big( \tfrac{ \delta+ c-
        d}{\lambda}\big)\Big)\,.
  \end{align*}
\end{lemma}

\begin{proof}
  We only show the proof for accumulated distance; the others are
  similar.  If $\tracesfrom s= \emptyset$, then $\aacdist{ s, t}= 0$
  and we are done.  Otherwise, let $\sigma\in \tracesfrom s$; we need
  to show that
  \begin{equation*}
    \inf_{\tau\in \tracesfrom t} \aacdist{ \sigma,
      \tau}\le \adjustlimits \sup_{s\tto{ \alpha, c} s'}
    \inf_{t\tto{ \alpha, d} t'} |c- d| +
    \lambda\, \aacdist{s', t'}\,.
  \end{equation*}

  Let $\pi$ be a path from $s$ which realizes $\sigma$, write $\pi=
  s\tto{ \alpha_1, c_1} s_1\to \dots$, and let $\sigma_1$ be the trace
  generated by the suffix of $\pi$ starting in $s_1$.  If $t\not\tto{
    \alpha_1}$, then the infimum on the right hand side of the equation
  is $\infty$, and we are done. 

  Assume that the infimum is finite and let $\epsilon\in \Realp$.
  There exists $t\tto{ \alpha_1, d_1} t_1$ for which
  \begin{equation*}
    | c_1- d_1|+ \lambda\, \aacdist{
      s_1, t_1}\le \inf_{t\tto{ \alpha_1, d} t'} |c_1- d| +
    \lambda\, \aacdist{ s_1, t'}+
    \tfrac \epsilon 2\,.
  \end{equation*}
  Let $\tau_1\in \tracesfrom{ t_1}$ be such that
  $\aacdist{ \sigma_1, \tau_1}\le \aacdist{ s_1, t_1}+ \tfrac \epsilon
  2$.  Let $\tau=( \alpha_1, d_1). \tau_1$, the concatenation, then
  \begin{align}
    \aacdist{ \sigma, \tau} &= | c_1- d_1|+
    \lambda\, \aacdist{ \sigma_1, \tau_1} \notag\\
    &\le | c_1- d_1|+ \lambda\,
    \aacdist{
      s_1, t_1}+ \tfrac\epsilon2 \notag\\
    &\le \inf_{t\tto{ \alpha_1, d} t'} |c_1- d| +
    \lambda\, \aacdist{s_1, t'}+
    \epsilon\,. \label{wtsjlap.eq:sigma'}
  \end{align}

  We have shown that for all $\epsilon\in \Realp$, there exists
  $\tau\in \tracesfrom t$ for which
  Equation~\eqref{wtsjlap.eq:sigma'} holds, hence
  \begin{equation*}
    \inf_{\tau\in
      \tracesfrom t} \aacdist{ \sigma, \tau}\le \inf_{t\tto{ \alpha_1, d}
      t'} |c_1- d| + \lambda\,
    \aacdist{ s_1, t'}
  \end{equation*}
  and the claim follows. \qed
\end{proof}

\subsection{Simulation distances}
\label{wtsjlap.sec:simulationmetrics}

In the following we use parametrized families $\{ R_\epsilon\subseteq
S\times S\}$ and $\{ R_{ \epsilon, \delta}\subseteq S\times
S\}$, \ie~functions $\Realnn\to 2^{ S\times S}$ and
$\Realnn\times \Realnn\to 2^{ S\times S}$, respectively; we
shall show how these give rise to distances in
Section~\ref{wtsjlap.sec:branchingdistances}.

\begin{definition}
  \label{wtsjlap.def:simu}
  A family of relations $\cal R = \{ R_\epsilon \subseteq S \times
  S \mid \epsilon \geq 0\}$ is
  \begin{itemize}
  \item a \emph{point-wise simulation family} provided that for all
    $( s, t)\in R_\epsilon\in \cal R$ and $s\tto{ \alpha, c} s'$, also
    $t\tto{ \alpha, d} t'$ with $| c- d|\le \epsilon$ for some
    $d\in \Realnn$ and $( s', t')\in R'_\epsilon\in \cal R$ with
    $\epsilon' \le \tfrac \epsilon \lambda$,
  \item an \emph{accumulating simulation family} provided that for all
    $( s, t)\in R_\epsilon\in \cal R$ and $s\tto{ \alpha, c} s'$, also
    $t\tto{ \alpha, d} t'$ with $| c- d|\le \epsilon$ for
    some $d\in \Realnn$ and $( s', t')\in R'_\epsilon\in \cal R$ with
    $\epsilon' \le \frac{\epsilon-| c- d|}{\lambda}$.
  \end{itemize}
  A family of relations $\cal R = \{ R_{ \epsilon, \delta} \subseteq S \times
  S \mid \epsilon \geq 0, -\epsilon \le \delta \le
  \epsilon\}$ is
  \begin{itemize}
  \item a \emph{maximum-lead simulation family} provided that for all
    $( s, t)\in R_{ \epsilon, \delta}\in \cal R$ and
    $s\tto{ \alpha, c} s'$, also $t\tto{ \alpha, d} t'$ with
    $|\delta+ c- d|\le \epsilon$ for some $d\in \Realnn$ and
    $( s', t')\in R_{ \epsilon', \delta'}\in \cal R$ with
    $\epsilon' \le \frac{\epsilon}{\lambda}$ and
    $\delta'\le \frac{ \delta+ c- d}{ \lambda}$.
  \end{itemize}
  We write
  \begin{itemize}
  \item $s\pwsim t$ if $( s, t)\in R_\epsilon\in \cal R$ for some
    point-wise simulation family $\cal R$,
  \item $s\awsim t$ if $( s, t)\in R_\epsilon\in \cal R$ for some
    accumulating simulation family $\cal R$,
  \item $s\pmsim t$ if $( s, t)\in R_{ \epsilon, 0}\in \cal R$ for some
    maximum-lead simulation family $\cal R$.
  \end{itemize}
\end{definition}

Note that the relations defined in the last part above again can be
collected into families $\mathord{\pwsim[]}=\{ \pwsim\mid
\epsilon\ge 0\}$, $\mathord{\awsim[]}=\{ \awsim\mid \epsilon\ge
0\}$, and $\mathord{\pmsim[]}=\{ \pmsim[ \epsilon, \delta]\mid
\epsilon, \delta\ge 0\}$.

Some explanatory remarks regarding these definitions are in order.
For point-wise simulation, $( s, t)\in R_\epsilon$ means that any
computation from $s$ can be matched by one from $t$ with the same
labels and a point-wise weight difference of at most $\epsilon$.
Hence the requirement that $s\tto{ \alpha, c} s'$ imply
$t\tto{ \alpha, d} t'$ with weight difference $| c- d|\le \epsilon$,
and that computations from the target states $s'$, $t'$ be matched
with some (inversely) discounted point-wise distance
$\epsilon'\le \frac{ \epsilon}{ \lambda}$.

For accumulated simulation, $( s, t)\in R_\epsilon$ is interpreted so
that any computation from $s$ can be matched by one from $t$ with the
same labels and accumulated absolute-value weight difference at most
$\epsilon$.  Hence we again require that $| c- d|\le \epsilon$, but
now computations from the target states have to be matched by what is
left of $\epsilon$ after $| c- d|$ has been used (and inverse
discounting applied).

Maximum-lead simulation is slightly more complicated, because we need
to keep track of the lead $\delta$ which one computation has
accomplished over the other.  Hence $( s, t)\in R_{ \epsilon, \delta}$
is to mean that any computation from $s$ \emph{which starts with a
  lead of $\delta$ over $t$} can be matched by a computation from $t$
with accumulated weight difference at most $\epsilon$.  Thus we
require that lead plus weight difference, $\delta+ c- d$, be
in-between $-\epsilon$ and $\epsilon$, and the new lead for
computations from the target states is set to that value (again with
inverse discounting applied).

For later use we collect the following easy facts about the above
simulations:
\begin{lemma}
  \label{wtsjlap.fact}
  \mbox{}
  \begin{enumerate}
  \item The families $\pwsim[]$, $\awsim[]$ and $\pmsim[]$
    are the largest respective simulation families.
  \item For $\epsilon\le \epsilon'$ and
    $R_\epsilon, R'_\epsilon \in \cal R$ a point-wise or accumulating
    simulation family, $R_\epsilon\subseteq R'_\epsilon$.  For
    $\epsilon\le \epsilon'$, $-\epsilon\le \delta\le \epsilon$ and
    $R_{ \epsilon, \delta}, R_{ \epsilon', \delta}\in \cal R$ a
    maximum-lead simulation family,
    $R_{ \epsilon, \delta}\subseteq R_{ \epsilon', \delta}$.
  \item For states $s, t\in S$ and $\epsilon\le \epsilon'$,
    $s\pwsim t$ implies $s\pwsim[\epsilon'] t$, $s\awsim t$ implies
    $s\awsim[\epsilon'] t$, and $s\pmsim t$ implies
    $s\pmsim[\epsilon'] t$.
  \item For states $s, t\in S$, $s\wsim t$ implies $s\pwsim[0] t$,
    $s\awsim[0] t$, and $s\pmsim[0] t$.
  \item For states $s, t\in S$, $s\not\uwsim t$ implies
    $s\not\pwsim t$, $s\not\awsim t$, and $s\not\pmsim t$ for any
    $\epsilon$.
  \end{enumerate}
\end{lemma}

\subsection{Branching distances}
\label{wtsjlap.sec:branchingdistances}

We present an alternative characterization of the above simulation
relations in form of recursive equations; note that these closely
resemble the inequalities of Lemma~\ref{wtsjlap.le:dist-fix}:

\begin{definition}\label{wtsjlap.def:branchingdistances}
  For states $s, t\in S$, the \emph{point-wise}, \emph{accumulated},
  and \emph{maximum-lead branching distances} are the respective
  minimal fixed points to the following recursive equations:
  \begin{align*}
    \pwbd{s,t} &= \adjustlimits \sup_{s\tto{ \alpha, c} s'} \inf_{t\tto{
        \alpha, d} t'} \max\big(| c- d|, \lambda\, \pwbd{ s', t'}\big)
    \\
    \awbd{s,t} &= \adjustlimits \sup_{s\tto{ \alpha, c} s'} \inf_{t\tto{
        \alpha, d} t'} |c- d| + \lambda\, \awbd{s', t'}
    \\
    \mlbd{ s, t} &= \mlbd{ s, t}( 0) \\
    &\hspace*{1em}\text{with } \mlbd{s,t}( \delta)= \adjustlimits \sup_{s\tto{ \alpha,
        c} s'} \inf_{t\tto{ \alpha, d} t'} \max\Big( | \delta|,
    \lambda\, \mlbd{ s', t'}\big( \tfrac{ \delta+ c-
      d}{\lambda}\big)\Big)
  \end{align*}
\end{definition}

Again, some remarks regarding these definitions will be in order.
First note that sup and inf are taken over the complete lattice
$\Realnn\cup\{ \infty\}$ here, whence $\inf \emptyset= \infty$ and
$\sup \emptyset= 0$.  Thus $\pwbd{ s, t}= 0$ in case $s\not\to$ and
$\pwbd{ s, t}= \infty$ in case $s\tto{ \alpha, c}$ but $t\not\tto{
  \alpha}$ for some $\alpha$, and similarly for the other distances.

The functionals defined by the first two equations above are
endofunctions on the complete lattice of functions
$S\times S\to \Realnn\cup\{ \infty\}$; they are easily shown to be
monotone, hence the minimal fixed points exist.  For the last
equation, the functional is an endofunction on the complete lattice
$\Real\to\big( S\times S\to \Realnn\cup\{ \infty\}\big)$, mapping each
lead $\delta\in \Real$ to a function $\mlbd{\cdot,\cdot}( \delta)$.
Also this functional can be shown to be monotone and hence to have a
minimal fixed point.

It is not difficult to see that the distances defined above are
branching distances in the sense of
Definition~\ref{wtsjlap.def:dist-require}.  Below we show that they
are closely related to the simulations of
Definition~\ref{wtsjlap.def:simu}:

\begin{proposition}\label{wtsjlap.lem:characterization}
  For states $s, t\in S$ and $\epsilon\in \Realnn$, we have
  \begin{itemize}
  \item $s\pwsim t$ if and only if $\pwbd{ s, t}\le \epsilon$,
  \item $s\awsim t$ if and only if $\awbd{ s, t}\le \epsilon$,
  \item $s\pmsim t$ if and only if $\mlbd{ s, t}\le \epsilon$.
  \end{itemize}
\end{proposition}

\begin{proof}
  Each of the six implications involved can be shown using standard
  structural-induction arguments.
\end{proof}

\section{Properties of distances}

In this section we present a number of properties of the six distances
introduced above.

\subsection{Branching versus linear distance}

For the qualitative relations, simulation implies trace inclusion, \ie
$s\uwsim t$ implies $s\le^u t$, and $s\wsim t$ implies $s\le t$.
Below we show a natural generalization of this to our quantitative
setting, where implications translate to inequalities; note that an
equivalent statement of the theorem is that for any $\epsilon$,
$\bdgeneric{ s, t}\le \epsilon$ implies $\distgeneric{ s, t}\le
\epsilon$ for all three distances considered.

\begin{theorem}
  \label{th:wtsjlap.linvsbra}
  For all states $s, t\in S$, we have
  \begin{equation*}
    \apwdist{ s, t}\le \pwbd{ s, t}\,, \qquad \aacdist{ s, t}\le \awbd{ s,
      t}\,, \qquad \apmdist{ s, t}\le \mlbd{ s, t}\,.
  \end{equation*}
\end{theorem}

\begin{proof}
  This follows from Lemma~\ref{wtsjlap.le:dist-fix} by an easy
  structural-induction argument.
\end{proof}

Note that Example~\ref{wtsjlap.ex:distances} shows that indeed, all distances
in the equations above can be finite.  Other, standard examples show
however that WTS exist for which $s\not\wsim t$ and yet $s\le t$,
hence $\bdgeneric{ s, t}= \infty$ and $\distgeneric{ s, t}= 0$ for all
three distances, showing the following theorem:

\begin{theorem}
  \label{th:wtsjlap.topineq}
  The distances $\apwdists$ and $\pwbds$ are topologically
  inequivalent.  Similarly, $\aacdists$ and $\awbds$, and also
  $\apmdists$ and $\mlbds$, are topologically inequivalent.
\end{theorem}

\subsection{Relationship between distances}

The theorems below sum up the relationship between our three linear
distances; note that the results depend heavily on whether or not
discounting is applied.  The following lemma is useful and easily
shown:

\begin{lemma}
  \label{wtsjlap.le:tracedist-ineq}
  For states $s, t\in S$, we have
  \begin{alignat*}{3}
    \apwdist{ s, t} &\le \aacdist{ s, t} &\qquad \apmdist{ s, t} &\le
    \aacdist{ s, t} &\qquad \apwdist{ s, t} &\le 2\, \apmdist{ s, t} \\
    \pwbd{ s, t} &\le \awbd{ s, t} & \mlbd{ s, t} &\le \awbd{ s, t} &
    \pwbd{ s, t} &\le 2\, \mlbd{ s, t}
  \end{alignat*}
\end{lemma}

The restrictions on traces mentioned below are understood to be
applied to the sets $\tracesfrom s$, $\tracesfrom t$ in
Definition~\ref{wtsjlap.def:trace_distances}.

\begin{theorem}
  \label{wtsjlap.th:trace-dist-eq-disc=1}
  Assume the discounting factor $\lambda= 1$.
  \begin{enumerate}
  \item When restricted to traces of bounded length, the three linear
    distances $\apwdists$, $\aacdists$ and $\apmdists$ are Lipschitz
    equivalent.
  \item For traces of unbounded length, the linear distances are
    mutually topologically inequivalent.
  \end{enumerate}
\end{theorem}

\begin{proof}
  If the length of traces is bounded above by $N\in \Nat$, then
  $\aacdist{ s, t}\le N\apwdist{ s, t}$ for all $s, t\in S$, and
  the result follows with Lemma~\ref{wtsjlap.le:tracedist-ineq}.

  For traces of unbounded length, topological inequivalence of
  $\apwdists$ and $\aacdists$, and of $\apwdists$ and $\apmdists$, can be
  shown by the following infinite WTS:
  \begin{equation*}
    \begin{tikzpicture}[shorten >=1pt, auto, initial text=]
      \begin{scope}[outer sep=1pt,minimum size=5pt,inner sep=2pt, node
        distance=1.4cm]
        \node (s0) {$s$}; 
        \node (s1) [right of=s0]{$s_1$}; 
        \node (s2) [right of=s1]{$s_2$}; 
        \node (dots) [right of=s2]{}; 
        \node (dots1) [right of=s2, yshift=-0.35cm,
        xshift=-.2cm]{$\cdots$};
        \node (sn) [right of=dots,
        xshift=-.1cm]{$s_n$}; 
        \node (dots2) [right of=sn, yshift=-0.35cm,
        xshift=-.2cm]{$\cdots$}; 
      \end{scope}
      \begin{scope}[->]
        \draw[] (s0) .. controls +(230:1cm) and +(310:1cm) ..
        node[above left, xshift=-.2cm] {$0$} (s0);
        \draw[] (s1) .. controls +(230:1cm) and +(310:1cm) ..
        node[above left, xshift=-.2cm, yshift=-.1cm] {$\frac12$} (s1);
        \draw[] (s2) .. controls +(230:1cm) and +(310:1cm) ..
        node[above left, xshift=-.2cm, yshift=-.1cm] {$\frac14$} (s2);
        \draw[] (sn) .. controls +(230:1cm) and +(310:1cm) ..
        node[above left, xshift=-.2cm, yshift=-.1cm] {$\frac1{2^n}$}
        (sn); 
      \end{scope}
    \end{tikzpicture} 
  \end{equation*} 
  
  Here we have $\aacdist{ s, s_n}= \apmdist{ s, s_n}= \infty$ for all
  $n$, but for any $\delta\in \Realp$ there is an $n$ for which
  $\apwdist{ s, s_n}< \delta$.  Similarly, topological inequivalence
  of $\aacdists$ and $\apmdists$ is shown by the infinite WTS below:
  \begin{equation*}
    \begin{tikzpicture}[shorten >=1pt, auto, initial text=]
      \begin{scope}[outer sep=0pt,minimum size=5pt,inner sep=2pt, node
        distance=2.7cm]
        \node (s0) {$s$}; 
        \node (s1) [right of=s0, xshift=-.5cm]{$s_1$};
        \node (s2) [right of=s1]{$s_2$}; 
        \node (dots) [right of=s2, xshift=-.5cm]{}; 
        \node (sn) [right of=dots, xshift=-1cm]{$s_n$};
        \node (dots2) [right of=sn, xshift=-.5cm]{}; 
      \end{scope}
      \begin{scope}[outer sep=0pt,minimum size=5pt,inner sep=2pt, node
        distance=1.6cm]
        \node (s') [below of=s0]{$s'$}; 
        \node (s1') [below of=s1]{$s_1'$}; 
        \node (s2') [below of=s2]{$s_2'$}; 
        \node (dots') [below of=dots, yshift=0.8cm]{$\cdots$}; 
        \node (sn') [below of=sn]{$s_n'$};  
        \node (dots2') [below of=dots2 , yshift=0.8cm]{$\cdots$}; 
      \end{scope}
      \begin{scope}[->]
        \draw[] (s0) .. controls +(230:.7cm) and +(130:.7cm) ..
        node[left] {$0$} (s');
        \draw[] (s1) .. controls +(230:.7cm) and +(130:.7cm) ..
        node[left] {$\frac12$} (s1');
        \draw[] (s2) .. controls +(230:.7cm) and +(130:.7cm) ..
        node[left] {$\frac14$} (s2');
        \draw[] (sn) .. controls +(230:.7cm) and +(130:.7cm) ..
        node[left] {$\frac1{2^n}$} (sn');

        \draw[] (s') .. controls +(50:.7cm) and +(310:.7cm) ..
        node[right] {$1$} (s0);
        \draw[] (s1') .. controls +(50:.7cm) and +(310:.7cm) ..
        node[right] {$1- \frac12$} (s1);
        \draw[] (s2') .. controls +(50:.7cm) and +(310:.7cm) ..
        node[right] {$1- \frac14$} (s2);
        \draw[] (sn') .. controls +(50:.7cm) and +(310:.7cm) ..
        node[right] {$1-\frac1{2^n}$} (sn);
      \end{scope}
    \end{tikzpicture} 
  \end{equation*}
\end{proof}

\begin{theorem}
  \label{wtsjlap.th:trace-dist-eq-disc<1}
  For discounting factor $\lambda< 1$, the three linear distances
  $\apwdists$, $\aacdists$ and $\apmdists$ are Lipschitz
  equivalent.
\end{theorem}

\begin{proof}
  This is similar to the first claim of the previous theorem: For all
  states $s, t\in S$, we have $\aacdist{ s, t}\le \frac1{1-
    \lambda} \apwdist{ s, t}$, and the result follows with
  Lemma~\ref{wtsjlap.le:tracedist-ineq}.
\end{proof}

\begin{theorem}
  \label{th:wtsjlap.topeqbra}
  For discounting factor $\lambda= 1$, the three branching distances
  $\pwbds$, $\awbds$ and $\mlbds$ are mutually topologically
  inequivalent.  For $\lambda< 1$, they are Lipschitz equivalent.
\end{theorem}

\begin{proof}
  The first claim can be shown using the same example WTS as for the
  second part of the proof of
  Theorem~\ref{wtsjlap.th:trace-dist-eq-disc=1}, and for the second
  claim we have $\awbd{ s, t}\le \frac1{1- \lambda} \pwbd{ s, t}$ and
  can apply Lemma~\ref{wtsjlap.le:tracedist-ineq}.
\end{proof}

\section{Conclusion}

We have argued above that our proposed extension of the qualitative
notion of trace inclusion and simulation to a quantitative setting is
reasonable.

For the three types of distances considered in this chapter, we have
seen that linear distances can easily be introduced, whereas
definition of branching distances requires more work and involves
fixed-point computations.  Our Lemma~\ref{wtsjlap.le:dist-fix}
remedies some of these difficulties, and we expect this remedy to also
be applicable for other interesting trace distances.  We will show in
Chapter~\ref{ch:qltbt} that a general procedure for obtaining
branching distances from linear distances is available.

We have shown that all our three linear distances are topologically
inequivalent to their corresponding branching distance, thus measure
inherently different properties.  Still, and analogously to the
qualitative setting, the branching distance can be used as an
over-approximation of the linear distance.  Also, and perhaps more
surprisingly, whether different linear or branching distances are
mutually equivalent depends on the usage of discounting.  We expect
most of these results to also hold for other kinds of trace distances,
see again Chapter~\ref{ch:qltbt}.

We have mentioned earlier that in this work we concentrate on trace
inclusion and simulation (asymmetric) distances, and of course similar
treatment should be given to trace equivalence and bisimulation
distances.  Symmetric linear distances are easily defined as
symmetrizations of the linear distances introduced here, but for the
branching distances there are subtle differences between symmetrized
simulation distances on the one hand and bisimulation distances on the
other hand.  The next chapter will be concerned with bisimulation
distances.

\chapter[A Quantitative Characterization of Weighted Kripke Structures
in Temporal Logic][Weighted Kripke Structures]{A Quantitative
  Characterization of Weighted Kripke Structures in Temporal
  Logic\footnote{This chapter is based on the journal
    paper~\cite{journals/cai/FahrenbergLT10} published in Computing
    and Informatics.}}
\label{ch:wtscai}

This chapter is concerned with \emph{weighted Kripke structures}
(WKS), which represent a straight-forward extension of Kripke
structures with a weighted transition relation labeling each
transition.  It then proceeds to define a weighted version of
Computation Tree Logic (WCTL) together with two different semantics
which mirror the point-wise and accumulating distances of the previous
chapter.  It is then shown that WCTL is \emph{adequate} and
\emph{expressive} for the corresponding bisimulation distances.

\section{Preliminaries}
\label{wtscai.sec:prim}

As in Chapter~\ref{ch:wtsjlap}, the results presented in this chapter
are based on metrics on sequences of real numbers. Let $a = (a_i)$ and
$b = (b_i)$ be such sequences, we then define for
$\lambda \in \mathopen]0,1\mathclose[$ the following basic distances:
\begin{align}
  d_+(a,b) &= \sum_i \lambda^i | a_i - b_i| \label{wtscai.eq:d+}\\
  d_\bullet(a,b) &= \sup_i\, \{\lambda^i|a_i - b_i|\} \label{wtscai.eq:d}
\end{align}
Throughout the chapter we will refer to \eqref{wtscai.eq:d+} and
\eqref{wtscai.eq:d}, as well other distances based on these, as an
\emph{accumulating distance} and as a \emph{point-wise distance},
respectively. For the remainder of this chapter we fix a discounting
factor $\lambda \in \mathopen]0,1\mathclose[$; note that contrary to
the previous chapter, we here assume $0<\lambda<1$.

We proceed to introduce WKS.  A natural interpretation is to view the
labellings as the cost of taking transitions in the structure. This
extension is similar to the one presented in Chapter~\ref{ch:wtsjlap}
for labeled transition systems, thus the results presented in the
previous chapter are transferable to the current setting.

\begin{definition}
  For a finite set $\Props$ of atomic propositions, a \emph{weighted
    Kripke structure} is a quadruple $M = ( S, T, L, w)$ where
  \begin{itemize}
  \item $S$ is a finite set of states,
  \item $T \subseteq S \times S$ is a transition
    relation
  \item $L:S \to 2^{\Props}$ is the proposition
    labeling, and
  \item $w:T \to \Realnn$ assigns positive real-valued weights to
    transitions.
  \end{itemize}
  We write $s \to s'$ instead of $(s,s')\in T$ and $s \tto{w}
  s'$ to indicate $w(s,s') = w$.
\end{definition}

A \emph{(weighted) path} in a WKS $M = ( S, T, L, w)$ is a
(possibly infinite) sequence $ \sigma = ((s_0,w_0), (s_1,w_1),
(s_2,w_2),\ldots)$ with $(s_i,w_i)\in S \times \Realnn$ and such that
$s_i \to s_{i+1}$ and $w_i = w(s_i,s_{i+1})$ for all $i$.  We denote
by $\tracesfrom{s}$ the set of paths in $M$ starting at state $s$,
and by $\tracesfrom{M}$ the set of all paths in $M$. Given a path
$\sigma$, we write $\sigma(i) = (\sigma(i)_s,\sigma(i)_w)$ for its
$i$'th state-weight pair, and $\sigma^i$ for the suffix starting at
$\sigma(i)$.

Notice that we have restricted ourselves to \emph{finite} weighted
Kripke structures here, \ie structures with a finite set of states and
finitely many atomic propositions.  Our characterization results in
Section~\ref{wtscai.se:char} only hold for such finite structures.

\begin{example}
  Figure~\ref{wtscai.fig:example} gives a model of a simple printer as
  a WKS which we shall come back to again later. Resource usage is
  modeled as atomic propositions, and transition weights model the
  combined cost of the operations. Turning on the machine, it moves
  from the state \textbf{Off} to \textbf{Ready}, from where it can
  \textbf{Suspend} and wake up at a much lower cost. Input is
  processed in the \textbf{Receiving} state, and the chosen output
  form incurs different costs related to resource usage, clean-up and
  reset.
\end{example}

\begin{figure}
  \centering
  \begin{tikzpicture}[->,>=stealth',node distance=2.7cm,state/.style={
    rectangle,
    rounded corners,
    draw=black, very thick,
    minimum height=2em,
    inner sep=2pt,
    text centered,
  }]
    \scriptsize
    \node[state] (OFF) 
    {
      \begin{tabular}{l} 
        \textbf{Off}\\
        \parbox{1.3cm}{Power/off}
      \end{tabular}
    };
    \node[state, right of=OFF, xshift=.6cm] (READY) 
    {
      \begin{tabular}{l} 
        \textbf{Ready}\\
        \parbox{1.3cm}{Power/on}
      \end{tabular}
    };
    \node[state, right of=READY, xshift=.6cm] (SUSPEND) 
    {
      \begin{tabular}{l} 
        \textbf{Suspended}\\
        \parbox{1.3cm}{ Power/on}
      \end{tabular}
    };
    \node[state, below of=READY, yshift=.7cm] (DATA) 
    {
      \begin{tabular}{l} 
        \textbf{Receiving}\\
        \parbox{2.6cm}{ Power/on \\Running PostScript}
      \end{tabular}
    };
    \node[state, below right of=DATA, yshift=-.1cm, xshift=-.5cm] (PRINT4) 
    {
      \begin{tabular}{l} 
        \textbf{Printing}\\
        \parbox{1.3cm}{ Power/on \\ A4}
      \end{tabular}
    };
    \node[state, left of=PRINT4] (COLOR) 
    {
      \begin{tabular}{l} 
        \textbf{Color}\\
        \parbox{1.3cm}{ Power/on \\ A4}
      \end{tabular}
    };
    \node[state, left of=COLOR] (FAX) 
    {
      \begin{tabular}{l} 
        \textbf{Faxing}\\
        \parbox{1.4cm}{ Power/on\\ Phone-line}
      \end{tabular}
    };
    \node[state, right of=PRINT4] (PRINT3) 
    {
      \begin{tabular}{l} 
        \textbf{Large}\\
        \parbox{1.3cm}{ Power/on \\ A3}
      \end{tabular}
    };
    \path 
    (OFF)  	edge[bend left=20]  node[anchor=south,above]{$100$} (READY)
    (READY)     	edge[bend left=20] node[anchor=south,above]{$0.5$}   (SUSPEND)
    (READY)  	edge  node[anchor=south,below]{$0$}                 (OFF)
    (SUSPEND)     	edge  node[anchor=south,below]{$25$}                (READY)
    (READY)       	edge                node[anchor=left,right]{$10$}   (DATA)
    (DATA)       	edge[bend right]   node[anchor=left,below]{$30$}     (FAX)
    (DATA)       	edge[bend right]   node[anchor=left,right]{$50$}     (COLOR)
    (DATA)  	edge[loop below]    node[anchor=north,below]{$4$}   (DATA)
    (DATA)  	edge[bend left]    node[anchor=left,right]{$35$}     (PRINT4)
    (DATA)  	edge[bend left]    node[anchor=left,below]{$40$}     (PRINT3)
    (PRINT3)       edge[bend right]   node[anchor=left,right]{$0.5$}     (READY)
    (PRINT4)       edge[bend right=50]   node[anchor=left,right]{$0.5$}  (READY)
    (COLOR)        edge[bend left=50]   node[anchor=left,left]{$0.7$}    (READY)
    (FAX)          edge[bend left]   node[anchor=right,left]{$0.2$}      (READY);
  \end{tikzpicture}
  \caption{The behavior and cost and resource usage of a simple
    printer.\label{wtscai.fig:example}}
\end{figure}
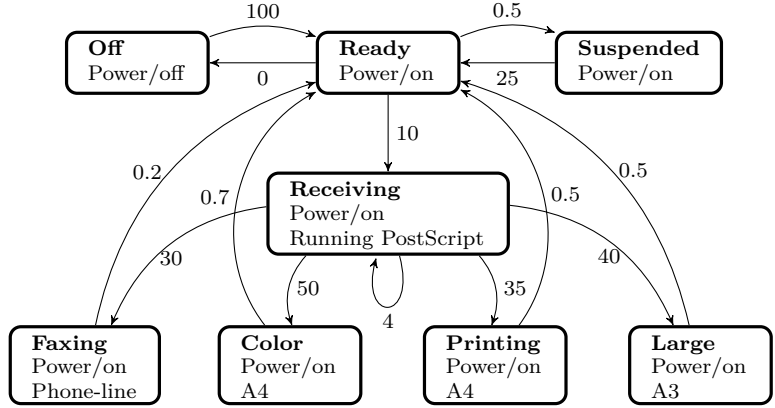

\section{Weighted CTL}
\label{wtscai.sec:ctl}

We now consider two interpretations of \emph{weighted CTL} (WCTL),
based on \eqref{wtscai.eq:d+} and \eqref{wtscai.eq:d}, which will
encompass quantitative information by two means. First, as with TCTL
and PCTL, a syntactic extension of path operators, by annotation with
real weights, models requirements on path weights (the exact meaning
of these are deferred to the choice of semantics). Second,
satisfaction of a formula by a system is no longer interpreted in the
Boolean domain $\{\top,\bot\}$, but rather assigns to a state a truth
value in the domain $\Realnn\cup\{ \infty\}$. We will interpret $0$ as
an exact match, whereas $\infty$ indicates an incompatibility between
the system and the specified atomic propositions of a formula.  Any
intermediate value is interpreted as a \emph{distance} from an exact
match. That is, a smaller distance means a closer (better) match of
the specified weights in the formula. We denote by
$\lsem{\varphi}(s) \in \Realnn\cup\{ \infty\}$ the value obtained by
evaluating formula $\varphi$ at state $s$.

From here on, we fix a set $\Props$ of atomic propositions and a WKS
$( S, T, L, w)$.  All definitions and results below will be given for
the states of one single WKS, but we note that to relate states of
different WKS, one can simply form the disjoint union.

\begin{definition}
  For $p\in \Props$, $\Phi$ generates the set of state formulae, and
  $\Psi$ the set of path formulae, annotated by weights $c \in
  \Realnn$, according to the following abstract syntax:
  \begin{align*}    
    \Phi &::= p \mid \lnot p \mid \Phi_1 \land \Phi_2 \mid
    \Phi_1 \lor \Phi_2 \mid \sfE \Psi \mid \sfA \Psi\\
    \Psi &::= \sfX_c \Phi\mid \sfG_c \Phi \mid
    \sfF_c \Phi\mid [\Phi_1 \sfU_c \Phi_2]
  \end{align*}
  The logic WCTL is the set of state formulae, written $\whml(
  \Props)$ or simply $\whml$.
\end{definition}

Before presenting the formal semantics, let us consider the usual
meaning of the CTL modalities, as well as how these may be generalized
to ensure adherence to bisimulation variants considered in the
following section:

Given CTL propositions of the form $M,s \models \sfE\psi$ and $M,s
\models \sfA\psi$, we may interpret these as infinite
\emph{existential}, respectively \emph{universal}, quantifications
over paths in $M$ from $s$ satisfying $\psi$. Similarly, $M,\sigma
\models \sfF\varphi$ and $M,\sigma \models \sfG\varphi$ may be interpreted
as an infinite \emph{disjunction}, respectively \emph{conjunction},
over propositions on the form $M,s_i \models \varphi$ for $i \geq 0$,
where $s_i$ is a state on $\sigma$.

This observation is in line with some arguments given in
\cite{DBLP:journals/fuin/KonikowkaP04}, and we expect that a generic
approach to defining quantitative (or multi-valued) semantics for WCTL
over the truth domain $\Realnn\cup\{ \infty\}$ is obtainable.  To this
end, the standard $\sup$ and $\inf$ operators are reasonable
generalization of $\sfE$, $\sfA$, $\sfF$ and $\sfG$ (interpreted as
disjunction and conjunction over the standard Boolean domain) to the
complete lattice $\Realnn\cup\{ \infty\}$.

Furthermore, this approach requires only modification to the
evaluation (\ie semantics) of path formulae.  Our semantics
specializes to the usual one in two different ways: either by mapping
to the designated set of truth values (\ie to $\top$), all $\epsilon<
\infty$ and $\infty$ to $\bot$, or by mapping only $0$ to $\top$ and
all $\epsilon> 0$ to $\bot$.

\subsection{Semantics}

In the following we present two \emph{discounted semantics}, derived
from the distances $d_+$ from \eqref{wtscai.eq:d+}, and $d_\bullet$
from \eqref{wtscai.eq:d} where weights of transition are accumulated
or considered point-wise, respectively.  Formally, the semantics of
$\varphi \in \whml$ defines a map from the set of states $S$ to the
set $\Realnn\cup\{ \infty\}$.  The first definition gives a general
weighted semantics to state formulae:

\begin{definition}[State semantics]
  \label{wtscai.def:state-sem}
  The semantics of state formulae is defined inductively as follows:
 \begin{align*}
    \lsem{p}(s) &{}={} \begin{cases}0 &\text{if }p \in L(s) \\
      \infty & \text{otherwise}\end{cases} &
    \lsem{\lnot p}(s) &{}={} \begin{cases}0 &\text{if }p \in \Props 
      \setminus L(s)\\
      \infty & \text{otherwise}\end{cases}\\
    \lsem{\varphi_1 \lor \varphi_2}(s) &{}={} \inf\big\{\lsem{\varphi_1}(s),
    \lsem{\varphi_2}(s)\big\}
    & \lsem{\varphi_1 \land \varphi_2}(s) &{}={} \sup\big\{\lsem{\varphi_1}(s),
    \lsem{\varphi_2}(s)\big\}\mbox{}\\[7pt]
    \lsem{\sfE \psi}(s) &{}={} \inf\big\{ \lsem{\psi}(\sigma) \mid {\sigma
      \in \tracesfrom{s}}\big\} & 
    \lsem{\sfA \psi}(s) &{}={} \sup\big\{ \lsem{\psi}(\sigma) \mid {\sigma
      \in \tracesfrom{s}}\big\}
  \end{align*}
  In the last two formulae, $\lsem{\psi}(\sigma)$ is the accumulating
  or point-wise semantics of $\sigma$ with respect to $\psi$ as
  appropriate, see below.
\end{definition}

In the next definition, we give the two different weighted semantics
to path formulae; an accumulated and a point-wise one.  Note that the
only difference between the two is an interchange of maximum and sum,
which supports the findings in~\cite{DBLP:journals/fuin/KonikowkaP04,
  DBLP:journals/tcs/Lluch-LafuenteM05} which advocate abstracting
away from concrete operators and interpreting the semantics over
general algebraic structures.

\begin{definition}[Path semantics] 
  \label{wtscai.def:pw-sem}
  The \emph{accumulating} semantics of path formulae is defined
  inductively as follows:
  \begin{align*}
    \lsem{\varphi}_+(\sigma) &{}={} \lsem{\varphi}(\sigma(0)_s)\\[10pt]
    \lsem{\sfX_c \varphi}_+(\sigma) &{}={} |\sigma(0)_w- c| +
    \lambda \lsem{\varphi}_+(\sigma^1) \\
    \lsem{\sfF_c \varphi}_+(\sigma) &{}={}
    \inf_k\bigg(\sum^{k-1}_{j=0} \lambda^j \Bigl|
    \sigma(j)_w -c\Bigr| + \lambda^k\lsem{\varphi}_+(\sigma^k)\bigg)\\
    \lsem{\sfG_c \varphi}_+(\sigma) &{}={}
    \sup_k\bigg(\sum^{k-1}_{j=0} \lambda^j \Bigl|
    \sigma(j)_w -c\Bigr| + \lambda^k\lsem{\varphi}_+(\sigma^k)\bigg)\\
    \lsem{\varphi_1 \sfU_c \varphi_2}_+(\sigma) &{}={} 
    \inf_k \bigg(
    \sum^{k-1}_{j=0} \lambda^j \Bigl| \lsem{\varphi_1}_+(\sigma^j) -c \Bigr|+
    \lambda^k\lsem{\varphi_2}_+(\sigma^k)\bigg)
  \end{align*}
  The \emph{point-wise} semantics of path formulae is defined
  inductively as follows:
  \begin{align*}
    \lsem{\varphi}_\bullet(\sigma) &{}={} \lsem{\varphi}(\sigma(0)_s)\\[10pt]
    \lsem{\sfX_c \varphi}_\bullet(\sigma) &{}={} \max\Big\{|
    \sigma(0)_w- c| ,\
    \lambda \lsem{\varphi}_\bullet(\sigma^1)\Big\} \\
    \lsem{\sfF_c \varphi}_\bullet(\sigma) &{}={}
    \inf_k\bigg(\max\Big\{ \max_{0\leq j < k}\big\{ \lambda^j |
    \sigma(j)_w -c|\big\} , \
    \lambda^k\lsem{\varphi}_\bullet(\sigma^k)\Big\}\bigg)\\
    \lsem{\sfG_c \varphi}_\bullet(\sigma) &{}={}
    \sup_k\bigg(\max\Big\{ \max_{0\leq j < k}\big\{ \lambda^j |
    \sigma(j)_w -c|\big\} , \
    \lambda^k\lsem{\varphi}_\bullet(\sigma^k)\Big\}\bigg)\\
    \lsem{\varphi_1 \sfU_c \varphi_2}_\bullet(\sigma) &{}={} 
    \inf_k \bigg(\max\Big\{ \max_{0\leq j < k}\big\{ \lambda^j \big|
    \lsem{\varphi_1}_\bullet(\sigma^j) -c \big| \big\} , \
    \lambda^k\lsem{\varphi_2}_\bullet(\sigma^k)\Big\}\bigg)
  \end{align*}
\end{definition}

Note that as usual, $\sfF_c$ can also be derived from $\sfU_c$ by $\sfF_c
\varphi \triangleq \ltrue \sfU_c \varphi$ (where $\ltrue$ is some
tautology).

Compared to, for example, TCTL, the annotated operators specify an
expected value, hence $\sfX_c\varphi$ evaluated on $\sigma$ means that
$c$ is expected of the first transition in $\sigma$. The difference is
then added to (or the maximum is taken of it and) the value of
$\varphi$ over the remaining path $\sigma^1$.

\begin{example}
  In the context of the example from Figure~\ref{wtscai.fig:example} we
  consider a useful property of printers, that of \emph{having
    received a job, the printer cannot suspend before completing the
    job}. The formula $\varphi = \sfA(\lnot \textbf{Suspended}
  \mathrel{\sfU_{10}} \textbf{Ready})$ formalizes this qualitative property
  and also states that we expect to reach the $\textbf{Ready}$ state
  using transitions with cost $10$.  With $\lambda = .9$, the
  point-wise interpretation $\lsem{\varphi}_\bullet(\textbf{Receiving}) =
  40$ is the cost (minus $10$) of the transition in the computation
  tree which is furthest from $10$.  In the accumulating
  interpretation, $\lsem{\varphi}_+(\textbf{Receiving}) = 48.37$ yields the
  sum of all such differences.
\end{example}

\section{Bisimulation}

We now consider extensions of \emph{strong bisimulation}
\cite{book/Milner89} over WKS, based on \eqref{wtscai.eq:d+} and
\eqref{wtscai.eq:d}.  These are filling the gap between
\emph{unweighted} and \emph{weighted} strong bisimulation as defined
below; \cf~also Def.~\ref{wtsjlap.de:simulation}.

\begin{definition}
  Let $( S, T, L, w)$ be a WKS on a set $\Props$ of atomic
  propositions. A relation $R \subseteq S \times S$ is
  \begin{itemize}
  \item an \emph{unweighted bisimulation} provided that for all $(s,t)
    \in R$, $L(s) = L(t)$ and
    \begin{itemize}
    \item if $s \to s'$, then also $t \to t'$ and $(s',t')\in R$
      for some $t' \in S'$,
    \item if $t \to t'$, then also $s \to s'$ and $(s',t')\in R$
      for some $s' \in S$;
    \end{itemize}
  \item a \emph{(weighted) bisimulation} provided that for all $(s,t)
    \in R$, $L(s) = L(t)$ and
    \begin{itemize}
    \item if $s \tto{c} s'$, then also $t \tto{c} t'$ and
      $(s',t')\in R$ for some $t' \in S'$,
    \item if $t \tto{c} t'$, then also $s \tto{c} s'$ and
      $(s',t')\in R$ for some $s' \in S$.
    \end{itemize}
  \end{itemize}
  We write $s \uwbisim t$ if $(s,t)\in R$ for some unweighted
  bisimulation $R$, and $s \wbisim t$ if $(s,t)\in R$ for some
  weighted bisimulation $R$.
\end{definition}

The motivation for the variants defined below is that, in order to
relate structures, we do not always need perfect matching of
transition weights; rather we would like to know how accurately
weights are matched.  As with the simulation distances of
Chapter~\ref{ch:wtsjlap}, we call a \emph{bisimulation distance} any
pseudometric on the states of a WKS which mediates between unweighted
and weighted bisimilarity:

\begin{definition}\label{wtscai.def:distance-mix}
  A \emph{bisimulation distance} on a WKS $( S, T, L, w)$ is a
  function $d: S \times S\to \Realnn\cup\{ \infty\}$ which satisfies the
  following for all $s_1, s_2, s_3\in S$:
  \begin{itemize}
  \item $d( s_1, s_1)= 0$,
  \item $d( s_1, s_2)+ d( s_2, s_3)\ge d( s_1, s_3)$,
  \item $d( s_1, s_2)= d( s_2, s_1)$,
  \item $s_1\wbisim s_2$ implies $d( s_1, s_2)= 0$ and
  \item $d( s_1, s_2)\ne \infty$ implies $s_1\uwbisim s_2$
  \end{itemize}
\end{definition}

Our distances are based on distances of (infinite) sequences of real
\linebreak
numbers, which is appropriate, as for ($s,t$) in $\uwbisim$ (or in
$\wbisim$), any path \linebreak
$(s, a, s_1 , a_1, s_2, \dots )\in \tracesfrom{s}$ must be matched by
an equal-length path \linebreak
$(t, b, t_1 , b_1, t_2, \dots )\in \tracesfrom{t}$ with $(s_i,t_i)$ in
$\uwbisim$ (respectively $\wbisim$).

By extending bisimulation with the $d_+$ and $d_\bullet$ distances, we
collect a family of relations $\{R_\epsilon \subseteq S \times
S\}$ (\ie a map $\Realnn \to 2^{S \times S}$)
since, due to discounting, for each step the distance between each
successor pair may grow:

\begin{definition}
  \label{wtscai.def:awbisim}
  A family of relations $\cal R = \{ R_\epsilon \subseteq S \times
  S \mid \epsilon > 0 \}$ is
  \begin{itemize}
  \item an \emph{accumulating bisimulation family} provided that for
    all $(s,t)\in R_\epsilon \in \cal R$, $L(s) = L(t)$ and
    \begin{itemize}
    \item if $s\tto{c}s'$, then also $t \tto{d} t'$ with
      $|c-d| \leq \epsilon$ for some $d\in \Realnn$ and
      $(s',t')\in R_{ \epsilon'} \in \cal R$ with
      $\epsilon' \lambda \leq \epsilon-|c-d|$, and
    \item if $t\tto{c}t'$, then also $s \tto{d} s'$ with
      $|c-d| \leq \epsilon$ for some $d\in \Realnn$ and
      $(s',t')\in R_{ \epsilon'} \in \cal R$ with
      $\epsilon' \lambda \leq \epsilon-|c-d|$.
    \end{itemize}
  \item a \emph{point-wise bisimulation family} provided that for all
    $(s,t)\in R_\epsilon \in \cal R$, $L(s) = L(t)$ and
    \begin{itemize}
    \item if $s\tto{c}s'$, then also $t \tto{d} t'$ with
      $|c-d| \leq \epsilon$ for some $d\in \Realnn$ and
      $(s',t')\in R_{ \epsilon'} \in \cal R$ with
      $\epsilon'\lambda \leq \epsilon$, and
    \item if $t\tto{c}t'$, then also $s \tto{d} s'$ with
      $|c-d| \leq \epsilon$ for some $d\in \Realnn$ and
      $(s',t')\in R_{ \epsilon'} \in \cal R$ with
      $\epsilon'\lambda \leq \epsilon$ .
    \end{itemize}
  \end{itemize}
  We write $s \awbisim t$ and $s \pwbisim t$, if
  $(s,t) \in R_\epsilon \in \cal R$ for an accumulating, respectively
  point-wise, bisimulation family $\cal R$.
\end{definition}

Both variants of bisimulation families give raise to a bisimulation
distance in the sense of Definition~\ref{wtscai.def:distance-mix} by $d_+( s,
t)= \inf\{ \epsilon\mid s\awbisim t\}$ and $d_\bullet( s, t)= \inf\{
\epsilon\mid s\pwbisim t\}$.  Observe the following easy facts:

\begin{lemma}
  \label{wtscai.lem:facts}
  \mbox{}
  \begin{enumerate}
  \item For $\epsilon \leq \epsilon'$ and members
    $R_\epsilon, R_{ \epsilon'}\in \cal R$ of an accumulating or point-wise
    bisimulation family, $R_\epsilon \subseteq R_{ \epsilon'}$.
  \item Given $s\awbisim t$, then every path $\sigma = (s_0, w_0, s_1
    , w_1 s_2, \dots )\in \tracesfrom{s}$ has a corresponding path
    $\sigma' = (t_0, w_0', t_1 , w'_1 t_2, \dots )\in \tracesfrom{t}$
    such that $\epsilon = \epsilon_0$ and $s_i
    \awbisim[\epsilon_i] t_i$ for all $i$, where
    $\epsilon_{i+1}\lambda = \epsilon_{i} - |w_i - w'_i|$.
  \item Given $s\pwbisim t$, then every path $\sigma = (s_0, w_0, s_1
    , w_1 s_2, \dots )\in \tracesfrom{s}$ has a corresponding path
    $\sigma' = (t_0, w_0', t_1 , w'_1 t_2, \dots )\in \tracesfrom{t}$
    such that $\epsilon = \epsilon_0$ and $s_i
    \pwbisim[\epsilon_i] t_i$ for all $i$, where
    $\epsilon_{i+1}\lambda = \epsilon_{i}$.
  \end{enumerate}
\end{lemma}

Note that as we only consider finite WKS, all $R_\epsilon$
relations are finite. Also, we shall speak of \emph{corresponding
  paths} when referring to the second and third properties of the
above lemma.

\section{Characterization}
\label{wtscai.se:char}

In this section we show that the presented WCTL interpretations are
adequate and expressive with respect to the appropriate bisimilarity
variant.

\subsection{Adequacy}

The link between accumulating bisimilarity and our accumulating
semantics for WCTL is as follows:

\begin{theorem}
  \label{th:wtscai.adeq-acc}
  For $s, t\in S$, $s\awbisim t$ iff\/ $\forall
  \varphi\in \whml : \bigl|\lsem{\varphi}_+(s) -
  \lsem{\varphi}_+(t)\bigr|\le \epsilon$.
\end{theorem}

The proof follows from Lemmas \ref{wtscai.lem:1} and
\ref{wtscai.lem:2} below. Observe that this provides us with the
following corollary which is precisely the standard notion of
adequacy, see~\cite{DBLP:journals/jacm/HennessyM85}:

\begin{corollary}
  For $s, t\in S$, $s\awbisim[0] t$ iff\/
  $\lsem{\varphi}_+(s) = \lsem{\varphi}_+(t)$ for all
  $\varphi\in \whml$.
\end{corollary}

We obtain an equivalent result for the point-wise semantics:

\begin{theorem}
  \label{th:wtscai.adeq-pw}
  For $s, t\in S$, $s\pwbisim t$ iff\/
  $\forall \varphi\in \whml : \bigl|\lsem{\varphi}_\bullet(s) -
  \lsem{\varphi}_\bullet(t)\bigr|\le \epsilon$.
\end{theorem}

\begin{example}
  We consider again the printer from Figure~\ref{wtscai.fig:example}.
  When ignoring \textbf{Color} and \textbf{Printing} as atomic
  propositions, we have $\textbf{Color} \awbisim[.2]
  \textbf{Printing}$, as the two initial transition are the only
  difference.  As a formula which realizes this bisimulation distance
  one can take $\varphi = \text{power/on} \land \text{A4} \land
  \sfA\sfX_{0.5}\text{Ready}$; then $\lsem{\varphi}_+(\textbf{Printing}) =
  0$ and $\lsem{\varphi}_+(\textbf{Color}) = .2$.
\end{example}

The proofs of adequacy, and also of expressivity below, for the
accumulating and point-wise cases are similar, hence we concentrate on
the accumulating case.  In the proof we will repeatedly make use of
the lesser-known little brother of the triangle inequality
\begin{equation*}
  \bigl||x-y|-|x-z|\bigr|\le\bigl|y-z\bigr|\,.
\end{equation*}

\begin{lemma}\label{wtscai.lem:1}
  Let $s, t\in S$ with $s\awbisim t$, and let $\sigma = (s,u,s_1, u_1,
  \ldots) \in \tracesfrom{s}$, $\tau = (t,v,t_1, v_1, \ldots) \in
  \tracesfrom{t}$ be corresponding paths.  Then
  $\bigl|\lsem{\varphi}_+(s) - \lsem{\varphi}_+(t)\bigr| \leq
  \epsilon$ for all state formulae $\varphi$, and
  $\bigl|\lsem{\varphi}_+(\sigma) - \lsem{\varphi}_+(\tau)\bigr| \le
  \epsilon$ for all path formulae $\varphi$.
\end{lemma}

\begin{proof*}
  We prove the lemma by structural induction in $\varphi$.  The
  induction base is clear, as $s\awbisim t$ implies that $p\in
  L( s)$ if and only if $p\in L( t)$, hence
  $\lsem{\varphi}_+(s)= \lsem{\varphi}_+(t)$ for $\varphi = p$ or
  $\varphi = \lnot p$.  For the inductive step, we examine each
  syntactic construction in turn:
  \begin{enumerate}
  \item $\varphi = \varphi_1 \lor \varphi_2$
      
    There are four cases to consider, corresponding to whether
    $\lsem{\varphi_1}_+( s)\le \lsem{\varphi_2}_+( s)$ or
    $\lsem{\varphi_1}_+( s)> \lsem{\varphi_2}_+( s)$ and similarly for
    $\lsem{\varphi_1}_+( t)$ and $\lsem{\varphi_2}_+( t)$.  We show the
    proof for one of the ``mixed'' cases; the other three are similar
    or easier:

    Assume $\lsem{\varphi_1}_+( s)\le \lsem{\varphi_2}_+( s)$ and
    $\lsem{\varphi_1}_+( t)> \lsem{\varphi_2}_+( t)$.  Then $\lsem{
      \varphi_1\lor \varphi_2}_+( s)- \lsem{ \varphi_1\lor \varphi_2}_+( t)=
    \lsem{\varphi_1}_+( s)- \lsem{ \varphi_2}_+( t)$, and
    $\lsem{\varphi_1}_+( s)- \lsem{ \varphi_1}_+( t)\le \lsem{\varphi_1}_+(
    s)- \lsem{ \varphi_2}_+( t)\le \lsem{\varphi_2}_+( s)- \lsem{
      \varphi_2}_+( t)$, and by induction hypothesis, $-\epsilon\le
    \lsem{\varphi_1}_+( s)- \lsem{ \varphi_1}_+( t)$ and
    $\lsem{\varphi_2}_+( s)- \lsem{ \varphi_2}_+( t)\le \epsilon$.

  \item $\varphi = \varphi_1 \land \varphi_2$.  This is similar to the
    previous case.

  \item $\varphi = \sfE \varphi_1$
      
    By definition of $\lsem{ \sfE \varphi_1}_+$ there is a path $\sigma
    \in \tracesfrom{s}$ for which $\lsem{\varphi_1}_+(\sigma) =
    \lsem{\varphi}_+(s)$.  By Lemma \ref{wtscai.lem:facts} there is a
    corresponding path $\tau \in \tracesfrom{t}$, and from the
    induction hypothesis we know that $|\lsem{\varphi_1}_+(\sigma) -
    \lsem{\varphi_1}_+(\tau)| \leq \epsilon$. Thus $|\lsem{\varphi}_+(s)
    - \lsem{\varphi}_+(t)| \leq \epsilon$.

  \item $\varphi = \sfA \varphi_1$.  This is similar to the previous case.

  \item $\varphi = \sfX_c \varphi_1$

    By definition, $\lsem{\varphi}_+(\sigma) =
    \lambda\lsem{\varphi_1}_+(\sigma^1) + |c - u|$ and
    $\lsem{\varphi}_+(\tau) = \lambda\lsem{\varphi_1}_+(\tau^1) + |c
    - v|$, where $\sigma= s \tto{u} \sigma^1$ and $\tau= t \tto{v}
    \tau^1$. Since $s\awbisim t$ and $\sigma$ and $\tau$ correspond,
    we have $\sigma(1) \awbisim[\epsilon'] \tau(1)$ with
    $\epsilon'\lambda \leq \epsilon - |u - v|$, and by
    induction hypothesis $|\lsem{\varphi_1}_+(\sigma^1) -
    \lsem{\varphi_1}_+(\tau^1)| \leq \epsilon'$.  Hence
    $\bigl|\lsem{ \varphi}_+( \sigma)- \lsem{ \varphi}_+(
    \tau)\bigr|\le \bigl|| c- u|-| c- v|\bigr|+ \lambda\bigl| \lsem{
      \varphi_1}_+( \sigma^1)- \lsem{ \varphi_2}_+( \tau^1)\bigr|\le|
    u- v|+ \epsilon-| u- v|= \epsilon$.

  \item $\varphi = \sfF_c \varphi_1$

    Pick any $\delta> 0$, then there is $k\in \Nat$ for which $S_k=
    \sum_{ j= 0}^{ k- 1} \lambda^j| \sigma( j)_w- c|+ \lambda^k \lsem{
      \phi}_+( \sigma^k)\le \lsem{ \phi}_+( \sigma)+ \delta$.  As the
    paths $\sigma$ and $\tau$ correspond, we also have $T_k= \sum_{ j=
      0}^{ k- 1} \lambda^j| \tau( j)_w- c|+ \lambda^k \lsem{ \phi}_+(
    \tau^k)\le \lsem{ \phi}_+( \tau)+ \delta$.  Repeated use of the
    definition of $\awbisim$ yields $\sigma( k)\awbisim[\epsilon']
    \tau(k)$ with $\epsilon'\lambda^k\le \epsilon- \sum_{ j=
      0}^{ k- 1} \lambda^j\bigl| \sigma( j)_w- \tau( j)_w\bigr|$,
    hence by induction hypothesis, $\bigl| \lsem{ \phi}_+( \sigma^k)-
    \lsem{ \phi}_+( \tau^k)\bigr|\le \epsilon'$.  Thus $\bigl| \lsem{
      \phi}_+( \sigma)- \lsem{ \phi}_+( \tau)\bigr|\le\bigl| S_k-
    T_k\bigr|+ \delta\le \sum_{ j= 0}^{ k- 1} \lambda^j\bigl|| \sigma(
    j)_w- c|-| \tau( j)_w- c|\bigr|+ \lambda^k\bigl| \lsem{ \phi}_+(
    \sigma^k)- \lsem{ \phi}_+( \tau^k)\bigr|+ \delta\le \epsilon+
    \delta$.  As these considerations hold for any $\delta> 0$, we
    must have $\bigl| \lsem{ \phi}_+( \sigma)- \lsem{ \phi}_+(
    \tau)\bigr|\le \epsilon$.

  \item $\varphi = \sfG_c \varphi_1$; $\varphi = \varphi_1 \sfU_c
    \varphi_2$.  These are similar to the previous
    case. \hfill$\square$

  \end{enumerate}
\end{proof*}

\begin{lemma}
  \label{wtscai.lem:2}
  Let $s, t\in S$ and assume that $\bigl|\lsem{\varphi}_+(s) -
  \lsem{\varphi}_+(t)\bigr| \leq \epsilon$ for all state formulae
  $\varphi \in \whml$.  Then $\smash{s\awbisim t}$.
\end{lemma}

\begin{proof}
  This follows directly from Theorem~\ref{wtscai.th:express} below,
  but one can also observe that the accumulating family $\cal R=\{
  R_\epsilon\}$ defined by
  \begin{equation*}
    R_\epsilon = \big\{(s,t) \mid s,t \in S, \forall \varphi \in
    \whml : \bigl|\lsem{\varphi}_+(s) - \lsem{\varphi}_+(t)\bigr| \leq
    \epsilon\big\}
  \end{equation*}
  is indeed an accumulating bisimulation in terms of Definition
  \ref{wtscai.def:awbisim}.
\end{proof}

\subsection{Expressivity}

We show that WCTL with accumulating semantics is expressive with
respect to accumulating bisimulation in the following sense:

\begin{theorem}
  \label{wtscai.th:express}
  For each $s\in S$ and every $\gamma\in \Realp$, there exists a
  state formula $\varphi^s_\gamma\in \whml$, interpreted over the
  accumulating semantics, which characterizes $s$ up to accumulating
  bisimulation and up to $\gamma$, \ie such that for all $s'\in
  S$, $s\awbisim s'$ if and only if $\lsem{ \varphi^s_\gamma}_+(
  s')\in[ \epsilon- \gamma, \epsilon+ \gamma]$ for all $\gamma$.
\end{theorem}

\begin{proof}
  We define characteristic formulae of unfoldings, as follows: For
  each $s\in S$ and $n\in \Nat$, denote $L(s) = \{p_1, \dots, p_k\}$
  and $\Props \setminus L(s) = \{q_1, \dots, q_\ell\}$ and let
  $\varphi( s, n)$ be the WCTL formula defined inductively as follows:
  \begin{align*}
    \varphi( s, 0) &= (p_1 \land \dots \land p_k ) \land (\lnot
    q_1 \land \dots \land \lnot q_\ell) \\
    \varphi( s, n+ 1) &= \bigwedge_{s \tto{ w} s'} \sfE\sfX_{
      w} \varphi( s', n)\;\land\!\!\! \bigwedge_{ w:
      s\tto{ w} s'} \sfA\sfX_w\Big( \bigvee_{ s\tto{ w} s'}
    \varphi( s', n)\Big)\land \varphi( s, 0)
  \end{align*}
  It is easy to see that $\lsem{\varphi( s, n)}_+(s)= 0$ for all $n$.

  To complete the proof, one observes that for each $\gamma> 0$, there
  is $n( \gamma)\in \Nat$ such that $\varphi( s, n( \gamma))$ can play
  the role of $\phi^s_\gamma$ in the theorem.  Intuitively this is due
  to discounting: The further the unfolding in $\varphi( s, n)$, the
  higher are the weights discounted, hence from some $n( \gamma)$ on,
  maximum weight difference is below $\gamma$.
\end{proof}

\begin{theorem}
  \label{wtscai.th:express_pt}
  For each $s\in S$ and every $\gamma\in \Realp$, there exists a state
  formula $\varphi^s_\gamma\in \whml$, interpreted over the point-wise
  semantics, which characterizes $s$ up to point-wise bisimulation and
  up to $\gamma$, \ie such that for all $s'\in S$, $s\pwbisim s'$ if
  and only if
  $\lsem{ \varphi^s_\gamma}_\bullet( s')\in[ \epsilon- \gamma,
  \epsilon+ \gamma]$ for all $\gamma$.
\end{theorem}

\section{Conclusion}
\label{wtscai.sec:conclusion}

We have shown in this chapter that weighted CTL with an accumulating
semantics is adequate and expressive for accumulating bisimulation for
weighted Kripke structures.  We have also seen that the same holds for
the point-wise semantics for WCTL with respect to point-wise
bisimulation.

We will see in Chapter~\ref{ch:qltbt} that these results can be lifted
to a common abstract framework which also encompasses other weighted
bisimulations such as the maximum-lead bisimulation
of~\cite{DBLP:conf/formats/HenzingerMP05} presented in the previous
chapter.

\chapter[Metrics for Weighted Transition Systems:
Axiomatization][Axiomatization]{Metrics for Weighted Transition
  Systems: Axiomatization\footnote{This chapter is based on the
    journal paper~\cite{DBLP:journals/tcs/LarsenFT11} published in
    Theoretical Computer Science.}}
\label{ch:simdistax}

In this chapter we turn to \emph{axiomatizations} of the point-wise
and accumulating simulation distances.  We first present
axiomatizations for finite processes and then for regular processes.
We then show that the axiomatizations for finite processes are sound
and complete, whereas the ones for regular processes are sound and
$\epsilon$-complete.

\section{Simulation distances}

Throughout this chapter we fix a \emph{finite} metric space $\KK$ of
\emph{weights} with a metric $d_\KK: \KK\times \KK\to \Real$.  We also
fix a \emph{discounting factor} $\lambda$ with $0\le \lambda< 1$,
which will be used in the definition of accumulating distance below.

\begin{definition}
  A \emph{weighted transition system} is a tuple $( S, T)$, where $S$ is a
  finite set of \emph{states} and $T\subseteq S\times \KK\times S$ is a set
  of (weighted) transitions.
\end{definition}

Note that all transition systems in this chapter are indeed assumed finite,
hence requiring finiteness of the metric space $\KK$ does not add extra
restrictions.

We fix a weighted transition system $( S, T)$ and introduce simulation
distance between states in $(S, T)$.  We concentrate on two types
here, accumulating and point-wise distance, but other kinds may indeed
be defined.

\subsection{Accumulating distance}

\begin{definition}
  For states $s, t\in S$, the \emph{accumulating simulation distance}
  from $s$ to $t$ is defined to be the least fixed point to the set of
  equations
  \begin{equation}
    \label{simdistax.eq:dB}
    \awbdg{s, t}= \adjustlimits \max_{ s\tto n s'} \min_{ t\tto m t'}
    \big( d_\KK( n, m)+ \lambda \awbdg{ s', t'}\big)\,.
  \end{equation}
\end{definition}

To justify this definition, we need to show that the equations~\eqref{simdistax.eq:dB}
indeed have a least solution.  To this end, write $S=\{ s_1,\dots, s_p\}$
and assume for the moment that the transition system $( S, T)$ is
\emph{non-blocking} such that every $s_i\in S$ has an outgoing transition
$s_i\tto n s_k$ for some $s_k\in S$.  Define a function $F: \Realnn^{
  p\times p}\to \Realnn^{ p\times p}$ by
\begin{equation*}
  F( x)_{ i, j}= \adjustlimits \max_{ s_i\tto n s_k} \min_{ s_j\tto
    m s_\ell}\big( d_\KK( n, m)+ \lambda\, x_{ k, \ell}\big)\,.
\end{equation*}
Here we are using the standard linear-algebra notation $\Realnn^{ p\times
  p}$ for $p\times p$-matrices with entries in $\Realnn$ and $x_{ k, \ell}$
for the entry in their $k$'th row and $\ell$'th column.

\begin{lemma}
  \label{simdistax.lem:acc:contraction}
  With metric on $\Realnn^{ p\times p}$ defined by $d( x, y)= \max_{ i, j=
    1}^p | x_{i, j}- y_{ i, j}|$, $F$ is a contraction with Lipschitz
  constant $\lambda$.
\end{lemma}

\begin{proof}
  (\textit{Cf.}~also the proof
  of~\cite[Thm.~5.1]{DBLP:journals/tcs/ZwickP96}.)  We can partition
  $\Realnn^{ p\times p}$ into finitely many (indeed at most
  $2^{ p^2 q^2}$ with $q=| \KK|$) closed polyhedral regions
  $R_{ i, j}$ (some of which may be unbounded) such that for
  $x, y\in R_{ i, j}$ in a common region, the $p^2$ max-min equations
  get resolved to the same transitions.  In more precise terms, there
  are mappings
  $n, m, k, \ell:\{ 1,\dots, p\}\times\{ 1,\dots, p\}\to\{ 1,\dots,
  p\}$ such that
  $F( x)_{ i, j}= d_\KK( n( i, j),m( i, j))+ \lambda x_{ k( i,
    j),\ell( i, j)}$ for all $x\in R_{ i, j}$.

  Now if $x, y\in R_{ i, j}$ are in a common region, then
  \begin{align*}
    d( F( x), F( y)) &\le \lambda \max_{ i, j}| x_{ k( i, j), \ell(
      i, j)}- y_{ k( i, j), \ell( i, j)}| \\
    &\le \lambda \max_{ i, j}| x_{ i, j}- y_{ i, j}|= \lambda d( x,
    y)\,.
  \end{align*}

  If $x\in R_{ i_1, j_1}, y\in R_{ i_2, j_2}$ are in different
  regions, a bit more work is needed.  The straight line segment
  between $x$ and $y$ admits finitely many intersection points with
  the regions $R_{ i, j}$; denote these $x= z_0,\dots, z_q= y$.  We
  have
  \begin{align*}
    d( F( x), F( y)) &\le d( F( z_0), F( z_1))+\dots+ d( F( z_{ q- 1},
    z_q)) \\
    &\le \lambda\big( d( z_0, z_1)+\dots+ d( z_{ q- 1}, z_q)\big)=
    \lambda d( x, y)\,.
  \end{align*}
  Note that the last equality only holds because all $z_i$ are on a
  straight line.
\end{proof}

Using the Banach fixed-point theorem and completeness of $\Realnn^{ p\times
  p}$ we can hence conclude that $F$ has a unique fixed point.  In the
general case, where $( S, T)$ may not be non-blocking, $F$ is a function $[
0, \infty]\to[ 0, \infty]$ with (extra) fixed point $[ \infty,\dots,
\infty]$.  Hence as a function $[ 0, \infty]\to[ 0, \infty]$, $F$ has at
most \emph{two} fixed points.  Now we can write the equation set from the
definition as
\begin{multline*}
  \begin{bmatrix}
    \awbdg{ s_1, s_1} & \awbdg{ s_1, s_2} & \cdots & \awbdg{ s_1, s_p}
    \\
    \awbdg{ s_2, s_1} & \awbdg{ s_2, s_2} & \cdots & \awbdg{ s_2, s_p}
    \\
    \vdots & \vdots & \ddots & \vdots \\
    \awbdg{ s_p, s_1} & \awbdg{ s_p, s_2} & \cdots & \awbdg{ s_p, s_p}
  \end{bmatrix}
  \\[1ex]
  = F
  \begin{bmatrix}
    \awbdg{ s_1, s_1} & \awbdg{ s_1, s_2} & \cdots & \awbdg{ s_1, s_p}
    \\
    \awbdg{ s_2, s_1} & \awbdg{ s_2, s_2} & \cdots & \awbdg{ s_2, s_p}
    \\
    \vdots & \vdots & \ddots & \vdots \\
    \awbdg{ s_p, s_1} & \awbdg{ s_p, s_2} & \cdots & \awbdg{ s_p, s_p}
  \end{bmatrix},
\end{multline*}
hence~\eqref{simdistax.eq:dB} has indeed a unique least fixed point.

\subsection{Point-wise distance}

For point-wise simulation distance we follow a lattice-theoretic rather than
a contraction approach.

\begin{definition}
  For states $s, t\in S$, the \emph{point-wise simulation distance}
  from $s$ to $t$ is defined to be the least fixed point to the set of
  equations
  \begin{equation*}
    \pwbdg{s, t}= \adjustlimits \max_{ s\tto n s'} \min_{ t\tto m t'}
    \max\big( d_\KK( n, m), \pwbdg{ s', t'}\big)\,.
  \end{equation*}
\end{definition}

Note that in this chapter, the point-wise distance is
\emph{undiscounted}.

Let $G: [ 0, \infty]^{ p\times p}\to[ 0, \infty]^{ p\times p}$ be the
function defined by
\begin{equation*}
  G( x)_{ i, j}= \adjustlimits \max_{ s_i\tto n s_k} \min_{ s_j\tto
    m s_\ell} \max\big( d_\KK( n, m), x_{ k, \ell}\big)\,.
\end{equation*}

\begin{lemma}
  With partial order on $[ 0, \infty]^{ p\times p}$ defined by $x\le
  y$ iff $x_{ i, j}\le y_{i, j}$ for all $i, j$, $G$ is (weakly)
  increasing.
\end{lemma}

\begin{proof}
  Trivial.
\end{proof}

Now the Tarski fixed-point theorem allows us to conclude that $G$ has a
unique least fixed point, hence the above definition is justified.

\subsection{Properties}

\begin{proposition}
  The functions $\awbdgs$ and $\pwbdgs$ are hemimetrics on $S$.
\end{proposition}

\begin{proof}
  To show that $\awbdg{ s, s}= \pwbdg{ s, s}= 0$ is trivial.  The
  triangle inequalities can be shown inductively; we prove the one for
  $\awbdgs$: For $s, t, u\in S$, we have
  \begin{align*}
    \awbdg{ s, t}+ \awbdg{ t, u} &= \adjustlimits \max_{ s\tto n s'} \min_{ t\tto m
      t'}\big( d_\KK( n, m)+ \lambda \awbdg{ s', t'}\big) \\
    &\qquad+ \adjustlimits \max_{ t\tto m t'} \min_{
      u\tto z u'}\big( d_\KK( m, z)+ \lambda \awbdg{ t', u'}\big) \\
    &\ge \multiadjustlimits{ \max_{ s\tto n s'} \min_{ t\tto m t'}
      \min_{ u\tto z u'}} \big(
    d_\KK( n, m)+ d_\KK( m, z) \\[-2ex]
    &\hspace*{12.5em} {}+ \lambda\big( \awbdg{ s', t'}+ \awbdg{ t',
      u'}\big)\big) \\
    &\ge \max_{ s\tto n s'} \min_{ u\tto z u'}\big( d_\KK( n, z)+
    \lambda \awbdg{ s', u'}\big)= \awbdg{ s, u}
  \end{align*}
  assuming the triangle inequality has been proven for the triple
  $(s', t', u')$.
\end{proof}

In the next proposition we take the standard liberty of comparing
different (weighted) transition systems by considering their disjoint
union.

\begin{proposition}
  \label{simdistax.prop:minmaxelem}
  The weighted transition systems $\nil$ and $\U$ given as
  $\nil=(\{ s_1\}, \emptyset)$ and
  $\U=(\{ s_1\}, \{( s_1, n, s_1)\mid n\in \KK\})$ are respectively
  minimal and maximal elements with respect to both
  $\awbdg{ \cdot, \cdot}$ and $\pwbdg{ \cdot, \cdot}$, that is,
  $\awbdg{ \nil, A}= \pwbdg{ \nil, A}= \awbdg{ A, \U}= \pwbdg{ A, \U}= 0$
  for any WTS~$A$.
\end{proposition}

\begin{proof}
  For $\awbdg{ \nil, A}$ and $\pwbdg{ \nil, A}$, the maximum $\max_{
    \smash[t]{s_1\tto n s_1'}}$ is taken over the empty set and hence is $0$.
  For $\awbdg{ A, \U}$ and $\pwbdg{ A, \U}$, any transition $s\tto n s'$
  in $A$ can be matched by $s_1\tto n s_1$ in $\U$, hence the distance
  is again $0$.
\end{proof}

\section{Axiomatizations for Finite Weighted Processes}

We now turn to a setting where our weighted transition systems are
generated by finite or regular (weighted) process expressions.  We
construct a sound and complete axiomatization of simulation distance in
a setting without recursion first and show afterwards how this may be
extended to a setting with recursion.

Let $\mathcal P$ be the set of process expressions generated by the
following grammar:
\begin{equation*}
  E::= \nil\mid n.E\mid E+ E \mid \qquad n\in \KK
\end{equation*}
Here $\nil$ is used to denote the \emph{empty process},
\cf~Proposition~\ref{simdistax.prop:minmaxelem}.

The semantics of finite process expressions is a weighted transition
system generated by the following standard SOS rules:
\begin{equation*}
  \AxiomC{\vphantom{$E_1\tto{n} E_1'$}}
  \UnaryInfC{$n.E \tto{n} E$}
  \DisplayProof
  \qquad 
  \AxiomC{$E_1\tto{n} E_1'$}
  \UnaryInfC{$E_1 + E_2 \tto{n} E_1'$}
  \DisplayProof
  \qquad 
  \AxiomC{$E_2\tto{n} E_2'$}
  \UnaryInfC{$E_1 + E_2 \tto{n} E_2'$}
  \DisplayProof
\end{equation*}

We can immediately get the following equalities
\begin{align}
  \awbdg{E+\nil,E} ={}&0 \notag\\
  \awbdg{n.E,m.F} ={}& d_\KK(n,m) + \lambda\, \awbdg{E,F} \label{simdistax.eq:prefix}\\
  \awbdg{E_1+E_2,F} ={}& \max (\awbdg{E_1,F},\awbdg{E_2,F}) \notag\\
  \awbdg{n.E,F_1+F_2} ={}& \min (\awbdg{n.E,F_1},\awbdg{n.E,F_2})\label{simdistax.eq:min}\\
  \intertext{For the point-wise distance, we again need only exchange
    \eqref{simdistax.eq:prefix} with} \pwbdg{n.E,m.F} ={}& \max(d_\KK(n,m),
  \pwbdg{E,F}) \notag
\end{align}

In order to show for example~\eqref{simdistax.eq:min} we simply need
to apply the definitions:
\begin{align*}
  \awbdg{ n. E, F_1+ F_2} &= \inf_{ F_1+ F_2\tto m F'}
  d_\KK( n, m)+ \lambda\, \awbdg{ E, F'} \\
  &= \min
  \begin{cases}
    \inf_{ F_1\tto m F'} d_\KK(n,m)+ \lambda\, \awbdg{ E, F'} \\
    \inf_{ F_2\tto m F'} d_\KK(n,m)+ \lambda\, \awbdg{ E, F'}
  \end{cases} \\
  &= \min \big( \awbdg{ n. E, F_1}, \awbdg{ n. E, F_2}\big)
\end{align*}
For~\eqref{simdistax.eq:prefix}, the sup-inf expression ranges over
singleton sets, hence the result is easy; the remaining equalities may
shown in a similar way.

The inference system $\mk F$ as given in Figure~\ref{simdistax.fig:F}
axiomatizes accumulating simulation distance for finite processes, as
we shall prove below.  Its sentences are inequalities of the form
$\dist{E,F} \bowtie r$ where $\mathord{\bowtie} \in \{=,\leq,\geq\}$
and $0 \leq r \leq \infty$.  Whenever $\dist{E,F} \bowtie r$ may be
concluded from $\mk F$, we write
$\vdash_{ \mk F} \dist{E,F} \bowtie r$.

\begin{figure}[tb]
  \fbox{
    \begin{minipage}[h]{1.0\linewidth}\vspace{5mm}      
      \centering
      \AxiomC{}
      \rulename{A1}
      \condition{$0\bowtie r$}
      \UnaryInfC{$\vdash \dist{\nil,E}\bowtie r$}
      \DisplayProof\qquad
      \AxiomC{}
      \rulename{A2}
      \condition{$\infty\bowtie r$}
      \UnaryInfC{$\vdash \dist{n.E,\nil}\bowtie r$}
      \DisplayProof
      \begin{prooftree} 
        \AxiomC{$\vdash \dist{E,F} \bowtie r_1$} 
        \condition{$d_\KK( n, m) +\lambda r_1 \bowtie r$}
        \rulename{R1}
        \UnaryInfC{$\vdash \dist{n.E,m.F} \bowtie r$} 
      \end{prooftree}
      \begin{prooftree} 
        \AxiomC{$\vdash \dist{E_1, F}\bowtie r_1$} 
        \AxiomC{$\vdash \dist{E_2, F}\bowtie r_2$} 
        \condition{$\max(r_1, r_2)\bowtie r$}
        \rulename{R2}
        \BinaryInfC{$\vdash \dist{E_1+ E_2, F}\bowtie r$} 
      \end{prooftree}
      \begin{prooftree} 
        \AxiomC{$\vdash \dist{n.E, F_1} \bowtie r_1$} 
        \AxiomC{$\vdash \dist{n.E, F_2} \bowtie r_2$} 
        \condition{$\min(r_1,r_2)\bowtie r$}
        \rulename{R3}\label{simdistax.rule:R3}
        \BinaryInfC{$\vdash \dist{n.E, F_1+ F_2} \bowtie r$} 
      \end{prooftree}
    \end{minipage}
  }
  \caption{%
    \label{simdistax.fig:F}
    The $\mk F$ proof system.}
\end{figure}

In addition to reflexivity and transitivity, we will need the
following standard properties of $\bowtie$ in latter proofs of
soundness and completeness: Whenever $a\bowtie b$ then, for all $c$:
$a + c \bowtie b+c$, $a\cdot c \bowtie b\cdot c$, and
$\max\{a,c\}\bowtie \max\{b,c\}$.

We also remark that the left process indeed needs to be guarded in
rule (R3) above, \ie~the following proposed rule (R3$'$) leads to an
unsound inference system:
\begin{prooftree} 
  \AxiomC{$\vdash \dist{E, F_1} \bowtie r_1$} 
  \AxiomC{$\vdash \dist{E, F_2} \bowtie r_2$} 
  \condition{$\min(r_1,r_2)\bowtie r$}
  \rulename{R3$'$}
  \BinaryInfC{$\vdash \dist{E, F_1+ F_2} \bowtie r$} 
\end{prooftree}

Indeed, using this rule we can derive the following (incomplete) proof
tree with a contradictory conclusion; the reason behind is that with
$E= E_1+ E_2$ non-deterministic as below, both $F_1$ and $F_2$ may be
needed to answer the challenge posed by $E$:

\vspace{-2ex}
{\small \begin{prooftree}
  \AxiomC{$\vdash \dist{ 1. \nil, 1. \nil}\ge 0$}
  \AxiomC{$\vdash \dist{ 2. \nil, 1. \nil}\ge 1$}
  \BinaryInfC{$\vdash \dist{ 1. \nil+ 2. \nil, 1. \nil}\ge 1$}
  \AxiomC{$\vdash \dist{ 1. \nil, 2. \nil}\ge 1$}
  \AxiomC{$\vdash \dist{ 2. \nil, 2. \nil}\ge 0$}
  \BinaryInfC{$\vdash \dist{ 1. \nil+ 2. \nil, 2. \nil}\ge 1$}
  \BinaryInfC{$\vdash \dist{ 1. \nil+ 2. \nil, 1. \nil+ 2. \nil}\ge 1$}
\end{prooftree}}

\begin{theorem}[Soundness]
  \label{th:simdistax.sound-acc}
  If $\vdash_{ \mk F} \dist{E,F}\bowtie r$, then $\awbdg{E,F}\bowtie r$.
\end{theorem}

\begin{proof}
  By an easy induction in the proof tree for $\vdash_{ \mk F} \dist{E,F}\bowtie r$,
  with a case analysis for the applied proof rule:
  \begin{itemize}
  \item[(A1)] follows from $\awbdg{\nil,E} = 0$.
  \item[(A2)] follows from $\awbdg{n.E,\nil} = \infty$ which is clear
    by the definition of $\awbdgs$.
  \item[(R1)] By induction hypothesis we have
    $\awbdg{E,F} \bowtie r_1$, and as
    $\awbdg{n.E,m.F} = d_\KK(n,m)+\lambda\awbdg{E,F}$, it follows that
    $\awbdg{n.E,m.F} = d_\KK(n,m)+\lambda r_1$.
  \item[(R2)] By induction hypothesis, $\awbdg{E_1,F} \bowtie r_1$ and
    $\awbdg{E_2,F} \bowtie r_2$, hence $\awbdg{E_1+E_2,F} =
    \max(\awbdg{E_1,F},\max\{E_2,F\}) \bowtie \max(r_1,r_2)$.
  \item[(R3)] By induction hypothesis, $\awbdg{ n.E, F_1}\bowtie r_1$ and
    $\awbdg{ n.E, F_2}\bowtie r_2$, hence $\awbdg{ n.E, F_1+ F_2}= \min(
    \awbdg{ n.E, F_1}, \awbdg{ n.E, F_2})\bowtie \min( r_1, r_2)$.
  \end{itemize}
\end{proof}

\begin{theorem}[Completeness]
  \label{simdistax.th:finite-complete}
  If $\awbdg{E,F}\bowtie r$, then $\vdash_{ \mk F} \dist{E,F}\bowtie r$.
\end{theorem}

\begin{proof}
  By an easy structural induction on $E$:
  \begin{itemize}
  \item[($E=$]$\nil)$\quad We have $\awbdg{ \nil, F}= 0\bowtie r$.  By Axiom (A1),
    also $\vdash \dist{ \nil, F}= 0$.
  \item[($E=$]$n.E')$\quad We use an inner induction on $F$:%
    \paragraph{Case $F= \nil$:} Here $\awbdg{ E, F}= \awbdg{ n. E', \nil}=
    \infty\bowtie r$.  By Axiom (A2), also $\vdash \dist{ n. E', \nil}=
    \infty$.
  
    \paragraph{Case $F= m. F'$:} Here $\awbdg{ E, F}= \awbdg{ n. E', m. F'}=
    d_\KK( n, m)+ \lambda \awbdg{ E', F'}\bowtie r$, hence with $r'= \lambda^{
      -1}(r- d_\KK( ,n m))$, $\awbdg{ E', F'}\bowtie r'$.  By induction
    hypothesis it follows that $\vdash \dist{ E', F'}\bowtie r'$, and we can
    use Axiom (R1) to conclude that $\vdash \dist{ E, F}\bowtie r$.
  
    \paragraph{Case $F= F_1+ F_2$:}
    Using~\eqref{simdistax.eq:min}, we have
    $ \awbdg{ E, F} = \awbdg{ n. E', F_1+
      F_2}=
    \min\big( \awbdg{ n. E', F_1}, \awbdg{ n. E', F_2}\big)$.  Let
    $\awbdg{ n. E', F_1}\bowtie r_1$ and
    $\awbdg{ n. E', F_2}\bowtie r_2$.  By the previous case, we know
    $\vdash \dist{ n. E', F_1}\bowtie r_1$. As
    $\min\{r_1,r_2\}\bowtie r$ it follows using (R3) that
    $\vdash_{ \mk F}[n.E,F_1+F_2] \bowtie r$.
  \item[$(E=$] $E_1+ E_2)$\quad By an argument similar to the one in the
    preceding subcase, we have $\awbdg{ E, F}= \max\big( \awbdg{ E_1, F},
    \awbdg{ E_2, F}\big)$.  If $\awbdg{ E_1, F}\bowtie r_1$ and $\awbdg{ E_2,
      F}\bowtie r_2$ with $\max( r_1, r_2)\bowtie r$, we can use the
    induction hypothesis to conclude $\vdash \dist{ E_1, F}\bowtie r_1$ and
    $\vdash \dist{ E_1, F}\bowtie r_1$, whence $\vdash \dist{ E, F}\bowtie r$
    by Axiom (R2).
  \end{itemize}
\end{proof}

\subsection{Point-wise distance}

We can devise a sound and complete inference system $\mk F^\bullet$
for point-wise distance (instead of accumulating) by replacing
inference rule (R1) in System $\mk F$ by the rule
\begin{prooftree} 
  \AxiomC{$\vdash \dist{E,F} \bowtie r_1$} 
  \condition{$\max(d_\KK( n, m), \lambda r_1)\bowtie r$}
  \rulename{R1$^\bullet$}
  \UnaryInfC{$\vdash \dist{n.E,m.F} \bowtie r$} 
\end{prooftree}
As before, we write $\vdash_{\mk F^{\bullet}}[E,F]\bowtie r$ if
$[E,F]\bowtie r$ can be proven by $\mk F^\bullet$.

\begin{theorem}[Soundness \& Completeness]
  \label{th:simdistax.sound-comp-pw}
  $\vdash_{ \mk F^\bullet}[E,F]\bowtie r$ if and only if
  $\pwbdg{E,F}\bowtie r$
\end{theorem}
\begin{proof}
  The proof is similar to the one for $\mk F$.
\end{proof}

\subsection{Simulation distance zero}

We show here that for distance zero, our inference system $\mk F$
specializes to a sound and complete inference system for simulation.
The inference system $\mk F_0$ is displayed in
Figure~\ref{simdistax.fig:sim}.

\begin{figure}[tb]
  \centering  
  \fbox{
    \begin{minipage}[h]{\linewidth}
      \centering  
      \vspace{1mm}
      
      \AxiomC{$\phantom{hep}$}
      \rulename{$A1_0$}
      \UnaryInfC{$\vdash \nil \wsim E$}
      \DisplayProof
      
      \begin{prooftree}
        \AxiomC{$\vdash E\wsim F$}
        \rulename{$R1_0$}
        \UnaryInfC{$\vdash n.E \wsim n.F$}        
      \end{prooftree}
      
      \begin{prooftree}
        \AxiomC{$\vdash E_1 \wsim F$} 
        \AxiomC{$\vdash E_2 \wsim F$} 
        \rulename{$R2_0$}
        \BinaryInfC{$\vdash E_1 + E_2 \wsim F$} 
      \end{prooftree}
      
      \AxiomC{$\vdash n.E\wsim F_1$}
      \rulename{$R3'_0$}
      \UnaryInfC{$\vdash n.E \wsim F_1 + F_2$}
      \DisplayProof\quad
      \AxiomC{$\vdash n.E\wsim F_2$}
      \rulename{$R3_0$}
      \UnaryInfC{$\vdash n.E \wsim F_1 + F_2$}
      \DisplayProof
      \vspace{1mm}
    \end{minipage}
  }
  \caption{%
    \label{simdistax.fig:sim}
    The $\mk F_0$ proof system}
\end{figure}

\begin{theorem}[Soundness \& Completeness] 
  $\vdash_{ \mk F_0}E \wsim F $ if and only if $E \wsim F$.
\end{theorem}

\begin{proof}
  Soundness follows immediately from the soundness of Proof system
  $\mk F$, and for completeness we note that the arguments one uses in
  the inductive proof of Theorem~\ref{simdistax.th:finite-complete}
  all specialize to distance zero.
\end{proof}

We remark that, contrary to the situation for general distance above,
we may indeed replace the guarded process $n.E$ in $(R3'_0)$ and
$(R3_0)$ by a plain $E$ without invalidating the rules.  Note also
that $\mk F_0$ may similarly be obtained as a specialization
$\mk F_0^\bullet$ of the axiomatization $\mk F^\bullet$ of point-wise
distance above.

\section{Axiomatizations for Regular Weighted Processes}

Let $N= \max\{ d_\KK( n, m)\mid n, m\in \KK\}$; by finiteness of $\KK$,
$N\in \Real$.  Let $V$ be a fixed set of variables, then $\mathcal
P^R$ is the set of process expressions generated by the following
grammar:
\begin{equation*}
  E::= \U \mid X\mid n.E\mid E+ E\mid \mu X. E \qquad n\in \KK,
  X\in V
\end{equation*}
Here we use $\U$ to denote the \emph{universal process} recursively
offering any weight in $\KK$,
\cf~Proposition~\ref{simdistax.prop:minmaxelem}.  Note that we do not
incorporate the empty process $\nil$.  Semantically this will mean that
all processes in $\mathcal P^R$ are non-terminating, and that the
accumulating distance between any pair of processes is finite.  The
reason for the exchange of $\nil$ with $\U$ is precisely this last
property; specifically, completeness of our axiomatization
(Theorem~\ref{simdistax.th:rec_comp}) can only be shown if all
accumulating distances are finite.

The semantics of processes in $\mathcal P^R$ is given as weighted
transition systems which are generated by the following standard SOS
rules:    
\begin{gather*}
  \AxiomC{\vphantom{$E_1\tto{n} E_1'$}} \UnaryInfC{$n.E \tto{n} E$}
  \DisplayProof \qquad \AxiomC{\vphantom{$E_1\tto{n} E_1'$}}
  \UnaryInfC{$\U \tto{n} \U$} \DisplayProof
  \\[2ex]
  \AxiomC{$E_1\tto{n} E_1'$} \UnaryInfC{$E_1 + E_2 \tto{n} E_1'$}
  \DisplayProof \qquad \AxiomC{$E_2\tto{n} E_2'$} \UnaryInfC{$E_1 +
    E_2 \tto{n} E_2'$} \DisplayProof \qquad \AxiomC{$E[\mu X.E/ X]
    \tto{n} F$} \UnaryInfC{$\mu X. E \tto{n} F$} \DisplayProof
\end{gather*}

As usual we say that a variable $X$ is \emph{guarded} in an expression
$E$ if any occurrence of $X$ in $E$ is within a subexpression $n.E'$.
Formally, we define the \emph{guarding depth} $\gd( E, X)$ of variable
$X$ in expression $E$ recursively by
\begin{align*}
  \gd( \U, X) &= \infty \\ \gd( X, X) &= 0 \\ \gd( n. E, X) &= 1+
  \gd( E, X) \\
  \gd( E_1+ E_2, X) &= \min\big( \gd( E_1, X), \gd( E_2, X)\big) \\
  \gd( \mu X. E, Y) &=
  \begin{cases}
    \gd( E, X) &\text{if } X\ne Y \\
    \infty &\text{if } X= Y
  \end{cases}
\end{align*}
and we say that $X$ is guarded in $E$ if $\gd( E, X)\ge 1$.

Also as usual, we denote by $E[ F/ X]$ the expression derived from $E$
by substituting all free occurrences of variable $X$ in $E$ by $F$,
and given tuples $\bar F=( F_1,\dots, F_k)$, $\bar X=( X_1,\dots, X_k)
$, we write $E[ \bar F/ \bar X]= E[ F_1/ X_1,\dots, F_k/ X_k]$ for the
simultaneous substitution.

Our inference system for regular processes consists of the set of
rules $\mk R$ as shown in Figure~\ref{simdistax.fig:R}; whenever
$\dist{E,F} \bowtie r$ may be concluded from $\mk R$, we write
$\vdash_{ \mk R} \dist{E,F} \bowtie r$.

\begin{figure}[tb]
  \fbox{
    \begin{minipage}[h]{1.0\linewidth}\vspace{5mm}      
      \begin{prooftree}
        \AxiomC{}
        \rulename{A3}
        \condition{$\frac{N}{1-\lambda}\le r$}
        \UnaryInfC{$\vdash \dist{E,F} \le r$}
      \end{prooftree}

      \begin{prooftree}
        \AxiomC{}
        \rulename{A4}
        \UnaryInfC{$\vdash \dist{\U,\sum_{n\in \KK} n.\U} = 0$}
      \end{prooftree}

      \begin{prooftree}
        \AxiomC{}
        \rulename{A5}
        \UnaryInfC{$\vdash \dist{\sum_{n\in \KK} n.\U, \U} = 0$}
        \DisplayProof\quad
        \AxiomC{}
        \rulename{A6}
        \UnaryInfC{$\vdash \dist{\mu X. E, E[\mu X.E / X]} = 0$}
      \end{prooftree}

      \centering
      \AxiomC{}
      \rulename{A7}
      \UnaryInfC{$\vdash \dist{E[\mu X.E / X], \mu X. E} = 0$}
      \DisplayProof\quad
      \AxiomC{}
      \rulename{A8}
      \UnaryInfC{$\vdash \dist{ E, \U}= 0$}
      \DisplayProof

      \begin{prooftree} 
        \AxiomC{$\vdash \dist{E,F} \bowtie r_1$} 
        \condition{$d_\KK( n, m)+ \lambda r_1\bowtie r$}
        \rulename{R1}
        \UnaryInfC{$\vdash \dist{n.E,m.F} \bowtie r$} 
      \end{prooftree}
      \begin{prooftree} 
        \AxiomC{$\vdash \dist{E_1,F}\bowtie r_1$} 
        \AxiomC{$\vdash \dist{E_2,F}\bowtie r_2$} 
        \condition{$\max( r_1,r_2)\bowtie r$}
        \rulename{R2}
        \BinaryInfC{$\vdash \dist{E_1+E_2,F}\bowtie r$} 
      \end{prooftree}
      \begin{prooftree} 
        \AxiomC{$\vdash \dist{n.E,F_1} \bowtie r_1$} 
        \AxiomC{$\vdash \dist{n.E,F_2} \bowtie r_2$} 
        \condition{$\min( r_1,r_2)\bowtie r$}
        \rulename{R3}
        \BinaryInfC{$\vdash \dist{n.E,F_1+F_2} \bowtie r$} 
      \end{prooftree}
      \begin{prooftree}
        \AxiomC{$\vdash \dist{E,F} \le r_1$}
        \AxiomC{$\vdash \dist{F,G} \le r_2$}
        \condition{$r_1 + r_2\le r$}
        \rulename{R4}
        \BinaryInfC{$\vdash \dist{E,G}\le r$}
      \end{prooftree}

      \begin{prooftree}
        \AxiomC{$\vdash \dist{E,F} \le r$}
        \rulename{R5}
        \UnaryInfC{$\vdash \dist{E+G,F+G}\le r$}
        \DisplayProof\quad
        \AxiomC{$\vdash \dist{E,F} \le r$}
        \rulename{R6}
        \UnaryInfC{$\vdash \dist{G+ E, G+ F}\le r$}
      \end{prooftree}      
    \end{minipage}
  }
  \caption{The $\mk R$ proof system}
  \label{simdistax.fig:R}
\end{figure}

Compared to inference system $\mk F$ for finite processes, we note
that we have to include the triangle inequality (R4) as an inference
rule.  Also, the precongruence property of simulation distance is
expressed by rules (R1), (R5), and (R6).  We will need all those extra
rules in the proof of Lemma~\ref{simdistax.le:A} which again is
necessary for showing completeness.

\begin{theorem}[Soundness]
  \label{simdistax.th:regular-sound}
  For closed expressions $E, F\in \mathcal P^R$ we have that
  $\vdash_{ \mk R} \dist{E, F}\bowtie r$ implies
  $\awbdg{ E, F}\bowtie r$.
\end{theorem}

\begin{proof}
  By an easy induction in the proof tree for
  $\vdash_{ \mk R} \dist{E,F}\bowtie r$, using the definition of
  $\awbdg{ \cdot, \cdot}$.  In relation to Axiom (A3), we note that
  $N= \max\{ d_\KK( n, m)\mid n, m\in \KK\}$ implies
  $\awbdg{ E, F}\le \sum_{ i= 0}^\infty \lambda^i N= \frac N{ 1-
    \lambda}$.
\end{proof}

Our completeness result for regular processes will be based on the
following lemmas; here we call an expression $E\in \mathcal P^R$
\emph{non-recursive} if it does not contain any subexpressions $\mu
X. E'$:

\begin{lemma}
  \label{simdistax.le:A}
  For all $E\in \mathcal P^R$ and $k\in \Nat$ there exist a
  non-recursive expression $F$ and tuples $\bar E=( E_1,\dots, E_k)$,
  $\bar X=( X_1,\dots, X_k)$ for which $\gd( F, X_i)\ge k$ for all $i$
  and
  \begin{equation*}
    \vdash_{ \mk R} \dist{ E, F[ \bar E/ \bar X]}= 0\,, \qquad
    \vdash_{ \mk R} \dist{ F[ \bar E/ \bar X], E}= 0\,.
  \end{equation*}
\end{lemma}

\begin{proof}
  Repeated use of the unfolding axioms (A6) and (A7), the congruence
  rules (R1), (R5), and (R6) with $r= 0$ and of the triangle
  inequality (R4).
\end{proof}

\begin{lemma}
  \label{simdistax.le:B}
  Let $F$ be a non-recursive expression and $\bar E=( E_1,\dots,
  E_k)$, $\bar X=( X_1,\dots, X_k)$ tuples for which $\gd( F, X_i)\ge
  k$ for all $i$.  Then
  \begin{equation*}
    \vdash_{ \mk R} \dist{ F[ \bar E/ \bar X], F[ \bar \U/ \bar X]}= 0\,, \qquad
    \vdash_{ \mk R} \dist{ F[ \bar \U/ \bar X], F[ \bar E/ \bar X]}= \lambda^k
    \tfrac N{ 1- \lambda}\,.
  \end{equation*}
\end{lemma}

\begin{proof}
  Repeated use of Axioms (A3) and (A8) together with the congruence
  rules (R1), (R5), and (R6) with $r= 0$.
\end{proof}

\begin{lemma}
  \label{simdistax.le:limited-complete}
  For closed non-recursive expressions $E$, $F$, $\awbdg{ E, F}\bowtie
  r$ implies $\vdash_{ \mk R} \dist{ E, F}\bowtie r$.
\end{lemma}

\begin{proof}
  By structural induction similar to the proof of
  Theorem~\ref{simdistax.th:finite-complete}.
\end{proof}

We are now in a position to state our completeness result which enables
arbitrary $\epsilon$-close proofs in the sense below.  The proof uses
unfoldings of recursive expressions as in Lemma~\ref{simdistax.le:A}, and as these
unfoldings are \emph{finite} non-recursive processes, we cannot expect exact
completeness.

\begin{theorem}[Completeness up to $\epsilon$]
  \label{simdistax.th:rec_comp}
  Let $E$ and $F$ be closed expressions of $\mathcal P^R$ and
  $\epsilon> 0$.  Then $\awbdg{ E, F}= r$ implies $\vdash_{ \mk R} \dist{ E,
    F}\le r+ \epsilon$ and $\vdash_{ \mk R} \dist{ E, F}\ge r- \epsilon$.
\end{theorem}

\begin{proof}
  Assume $\awbdg{ E, F}= r$, and choose $k\in \Nat$ such that $2
  \lambda^k \frac N{ 1- \lambda}\le \epsilon$.  By Lemma~\ref{simdistax.le:A} we
  have non-recursive expressions $E'$, $F'$ and tuples $\bar E$, $\bar
  F$, $\bar X$, and $\bar Y$ for which $\gd( E', X_i)\ge k$ and $\gd(
  F', Y_i)\ge k$ for all $i$, and such that
  \begin{align*}
    &\vdash_{ \mk R} \dist{ E, E'[ \bar E/ \bar X]}= 0\,, &
    &\vdash_{ \mk R} \dist{ E'[ \bar E/ \bar X], E}= 0\,, \\
    &\vdash_{ \mk R} \dist{ F, F'[ \bar F/ \bar Y]}= 0\,, &
    &\vdash_{ \mk R} \dist{ F'[ \bar F/ \bar Y], F}= 0\,.
  \end{align*}
  
  From Lemma~\ref{simdistax.le:B} it follows that
  \begin{align*}
    &\vdash_{ \mk R} \dist{ E'[ \bar E/ \bar X], E'[ \bar \U/ \bar X]}= 0\,, \\
    &\vdash_{ \mk R} \dist{ E'[ \bar \U/ \bar X], E'[ \bar E/ \bar X]}\le
    \lambda^k \tfrac N{ 1- \lambda}= \tfrac \epsilon 2\,, \\
    &\vdash_{ \mk R} \dist{ F'[ \bar F/ \bar Y], F'[ \bar \U/ \bar Y]}= 0\,, \\
    &\vdash_{ \mk R} \dist{ F'[ \bar \U/ \bar Y], F'[ \bar F/ \bar Y]}\le
    \lambda^k \tfrac N{ 1- \lambda}= \tfrac \epsilon 2\,.
  \end{align*}

  Using the triangle inequality and
  Theorem~\ref{simdistax.th:regular-sound} we now have
  \begin{align*}
    \awbdg{ E'[\bar \U/ \bar X], F'[ \bar U/ \bar X]} &\le \awbdg{ E'[
      \bar \U/ \bar X], E'[ \bar E/ \bar X]}+ \awbdg{ E'[ \bar E/ \bar
      X], E} \\
    &\quad+ \awbdg{ E, F}+ \awbdg{ F, F'[ \bar F/ \bar Y]} \\
    &\quad+ \awbdg{ F'[
      \bar F/ \bar Y], F'[ \bar \U/ \bar Y]} \\
    &\le \tfrac \epsilon 2+ 0+ r+ 0+ 0= r+ \tfrac \epsilon 2\,.
  \end{align*}
  Only non-recursive expressions are involved here, so that we can
  invoke Lemma~\ref{simdistax.le:limited-complete} to conclude
  \begin{equation*}
    \vdash_{ \mk R} \dist{ E'[\bar \U/ \bar X], F'[ \bar U/ \bar X]}\le r+
    \tfrac \epsilon 2\,.
  \end{equation*}

  Now we can use the triangle inequality axiom (R4) together with the
  eight equations above to arrive at
  \begin{equation*}
    \vdash_{ \mk R} \dist{ E, F}\le r+ \epsilon\,.
  \end{equation*}
  Similar arguments show that also $\vdash_{ \mk R} \dist{ E, F}\ge r-
  \epsilon$,
\end{proof}

\subsection{Point-wise distance}

Again we can easily convert our proof system $\mk R$ into one for
point-wise (instead of accumulating) distance.  In this case, we
obtain $\mk R^\bullet$ by replacing inference rule (R1) by (R1$^\bullet$)
as we did for Proof system $\mk F$, and (A3) needs to be replaced by
\begin{prooftree} 
  \AxiomC{}
  \rulename{A3$^\bullet$}
  \UnaryInfC{$\vdash \dist{E,F} \le N$}
\end{prooftree}
With these replacements we have a sound and $\epsilon$-complete
axiomatization of point-wise simulation distance for recursive
weighted processes:

\begin{theorem}[Soundness \& Completeness up to $\epsilon$]
  \label{th:simdistax.sound-comp-pw.reg}
  Let $E$ and $F$ be closed expressions of $\mathcal P^R$, then
  $\vdash_{\mk R^\bullet} \dist{E, F}\bowtie r$ implies
  $\pwbdg{ E, F}\bowtie r$, and $\pwbdg{ E, F}= r$ implies
  $\vdash_{\mk R^\bullet} \dist{ E, F}\le r+ \epsilon$ and
  $\vdash_{\mk R^\bullet} \dist{ E, F}\ge r- \epsilon$ for any
  $\epsilon> 0$.
\end{theorem}
\begin{proof}
  The proof is similar to that for accumulated distance.
\end{proof}

\chapter[The Quantitative Linear-\!Time--Branching-\!Time
Spectrum][The Quantitative Linear-\!Time--Branching-\!Time
Spectrum]{The Quantitative Linear-\!Time--Branching-\!Time
  Spectrum\footnote{This chapter is based on the journal
    paper~\cite{DBLP:journals/tcs/FahrenbergL14} published in
    Theoretical Computer Science.}}
\label{ch:qltbt}

This chapter generalizes the work presented so far in several ways and
develops a general theory of linear and branching distances depending
on a given, but unspecified, trace distance.  It introduces
quantitative Ehrenfeucht-Fra{\"\i}ss{\'e} games as a central tool for
this generalization and then proceeds to define a spectrum of linear
and branching distances which generalizes the one of van
Glabbeek~\cite{inbook/hpa/Glabbeek01}.

\section{Traces, Trace Distances, and Transition Systems}
\label{qltbt.se:prelim}

For a finite non-empty sequence $a=( a_0,\dots, a_n)$, we write
$\last( a)= a_n$ and $\len( a)= n+ 1$ for the length of $a$; for an
infinite sequence $a$ we let $\len( a)= \infty$.  Concatenation of
finite sequences $a$ and $b$ is denoted $a\cdot b$.  We denote by
$a^k=( a_k, a_{ k+ 1},\dots)$ the $k$-shift and by $a_i$ the
$( i+ 1)$st element of a (finite or infinite) sequence, and by
$\emptyseq$ the empty sequence.

Throughout this chapter we fix a set $\KK$ of labels, and we let
$\KK^\infty= \KK^*\cup \KK^\omega$ denote the set of finite and
infinite traces (\ie~sequences) in $\KK$.  A hemimetric
$\trace d: \KK^\infty\times \KK^\infty\to \Realnn\cup\{ \infty\}$ is
called a \emph{trace distance} if $\len( \sigma)\ne \len( \tau)$
implies $\trace d( \sigma, \tau)= \infty$.

A \emph{labeled transition system} (LTS) is a pair $( S, T)$
consisting of states $S$ and transitions
$T\subseteq S\times \KK\times S$.  We often write $s\tto{ x} t$ to
signify that $( s, x, t)\in T$.  Given $e=( s, x, t)\in T$, we write
$\src( e)= s$ and $\tgt( e)= t$ for the source and target of $e$.  A
\emph{path} in $( S, T)$ is a finite or infinite sequence
$\pi=(( s_0, x_0, t_0),( s_1, x_1, t_1),\dots)$ of transitions
$( s_j, x_j, t_j)\in T$ which satisfy $t_j= s_{ j+ 1}$ for all $j$.
We denote by $\tr \pi=( x_0, x_1,\dots)$ the trace induced by such a
path $\pi$.  For $s\in S$ we denote by $\pathsfrom s$ the set of
(finite or infinite) paths from $s$ and by
$\tracesfrom s=\{ \tr \pi\mid \pi\in \pathsfrom s\}$ the set of traces
from $s$.

\section{Examples of Trace Distances}
\label{qltbt.se:example_distances}

We give a systematic treatment of trace distances with which our
quantitative framework can be instantiated.  Some of them have
appeared in previous chapters; some others are new, but have been used
elsewhere in the literature.

Most of the trace distances one finds in the literature are defined by
giving a hemimetric $d$ on $\KK$ and a method to combine the
so-defined distances on individual symbols to a distance on traces.
Three general methods are used for this combination:
\begin{itemize}
\item The \emph{point-wise} trace distance: $\PWDIS d( \sigma, \tau)=
  \sup_j \lambda^j d( \sigma_j, \tau_j)$;
\item the \emph{accumulating} trace distance: $\ACCDIS d( \sigma, \tau)=
  \sum_j \lambda^j d( \sigma_j, \tau_j)$;
\item The \emph{limit-average} trace distance: $\ACCAVG d(
  \sigma, \tau)= \liminf_j \frac 1{ j+ 1} \smash{\sum_{ i= 0}^j} d(
  \sigma_i, \tau_i)$.
\end{itemize}
Note that the trace distances are parametrized by the label distance $d:
\KK\times \KK\to \Realnn\cup\{ \infty\}$.  Also, $\lambda$ is a
\emph{discounting} factor with $0< \lambda\le 1$, and we assume that the
involved traces have equal length; otherwise any trace distance has
value $\infty$.  The point-wise distance thus measures the (discounted)
greatest individual symbol distance in the traces, whereas accumulating
and limit-average distance accumulate these individual distances along
the traces.

If the distance on $\KK$ is the \emph{discrete} distance given by $\disc
d( x, x)= 0$ and $\disc d( x, y)= \infty$ for $x\ne y$, then all trace
distances above agree, for any $\lambda$.  This defines the
\emph{discrete trace distance} $\disc{ \trace d}= \PWDIS{ \disc d}=
\ACCDIS{ \disc d}= \ACCAVG{ \disc d}$ given by $\disc{ \trace d}(
\sigma, \tau)= 0$ if $\sigma= \tau$ and $\infty$ otherwise.  We will
show below that for the discrete trace distance, our quantitative
linear-time--branching-time spectrum specializes to the qualitative one
of~\cite{inbook/hpa/Glabbeek01}.

If one lets $d( x,x)= 0$ and $d( x, y)= 1$ for $x\ne y$ instead, then
$\ACCDIS[ 1] d$ is \emph{Hamming
  distance}~\cite{journals/bell/Hamming50} for finite traces, and
$\ACCDIS d$ with $\lambda< 1$ and $\ACCAVG d$ are two sensible ways to
define Hamming distance also for infinite traces.  $\PWDIS[ 1] d$ is
topologically equivalent to the discrete distance; indeed, $\PWDIS[ 1]
d( \sigma, \tau)= 1$ iff $\disc{\trace d}( \sigma, \tau)=
\infty$.

A generalization of the above distances may be obtained by equipping
$\KK$ with a preorder $\mathord\preceq\subseteq \KK\times \KK$
indicating that a label $x\in \KK$ may be replaced by any $y\in \KK$
with $x\preceq y$, as for example in~\cite{thesis/auc/Thomsen87}.
If we define $d( x, y)= 0$ if $x\preceq y$ and $d( x, y)= \infty$
otherwise (note that this is a hemimetric which is not necessarily
symmetric), then again $\PWDIS d= \ACCDIS d= \ACCAVG d$ for any
$\lambda$.

Point-wise and accumulating distances have been studied in a number of
papers~\cite{DBLP:journals/tse/AlfaroFS09,
  DBLP:journals/tcs/CernyHR12, DBLP:journals/tocl/ChatterjeeDH10,
  DBLP:conf/concur/Breugel05} and in previous chapters.
$\PWDIS[ 1] d$ is the point-wise distance
from~\cite{DBLP:journals/tse/AlfaroFS09,
  DBLP:conf/qest/DesharnaisLT08}, and $\PWDIS d$ for $\lambda< 1$ is
the discounted distance from~\cite{DBLP:journals/tse/AlfaroFS09,
  DBLP:conf/icalp/AlfaroHM03}.  Accumulating distance $\ACCDIS d$ has
been studied in~\cite{DBLP:journals/tse/AlfaroFS09}, and $\ACCAVG d$
in~\cite{DBLP:journals/tocl/ChatterjeeDH10,
  DBLP:journals/tcs/CernyHR12}.  Both $\ACCDIS d$ and $\ACCAVG d$ are
well-known from the theory of discounted and mean-payoff
games~\cite{journals/gameth/EhrenfeuchtM79,
  DBLP:journals/tcs/ZwickP96}.

All distances above were obtained from distances on individual symbols
in $\KK$.  A trace distance for which this is \emph{not} the case is
the \emph{maximum-lead} distance
from~\cite{DBLP:conf/formats/HenzingerMP05} defined for
$\KK\subseteq \Sigma\times \Real$, where $\Sigma$ is an alphabet.
Writing $x\in \KK$ as $x=( x^\ell, x^w)$, it is given by
\begin{equation*}
  \trace d_\pm( \sigma, \tau)=
  \begin{cases}
    \sup_j \bigl| \sum_{i= 0}^j \sigma^w_i - \sum_{i= 0}^j
    \tau^w_i\bigr| &\text{if }
    \sigma^\ell_j= \tau^\ell_j \text{ for all } j, \\
    \infty &\text{otherwise}.
  \end{cases}  
\end{equation*}
As this measures differences of accumulated labels along runs, it is
especially useful for real-time systems,
\cf~\cite{DBLP:conf/formats/HenzingerMP05, conf/fit/FahrenbergL12}.

As a last example of a trace distance we mention the \emph{Cantor}
distance given by $\trace d_\textup{C}( \sigma, \tau)=( 1+ \inf\{ j\mid
\sigma_j\ne \tau_j\})^{ -1}$.  Cantor distance hence measures the
(inverse of the) length of the common prefix of the sequences and has
been used for verification \eg~in~\cite{DBLP:conf/acsd/DoyenHLN10}.
Both Hamming and Cantor distance have applications in information theory
and pattern matching.

We will return to our example trace distances in
Section~\ref{qltbt.se:examples_rec} to show how our framework may be applied
to yield concrete formulations of distances in the
linear-time--branching-time spectrum relative to these.

\section{Quantitative Ehrenfeucht-Fra{\"\i}ss{\'e} Games}
\label{qltbt.se:game}

To lift the linear-time--branching-time spectrum to the quantitative
setting, we define below a quantitative Ehrenfeucht-Fra{\"\i}ss{\'e}
game~\cite{journals/fundmat/Ehrenfeucht61, journals/alger/Fraisse54}
which is similar to the well-known bisimulation game
of~\cite{DBLP:conf/banff/Stirling95}.

Let $( S, T)$ be a LTS and
$\trace d: \KK^\infty\times \KK^\infty\to \Realnn\cup\{ \infty\}$ a
trace distance.

The intuition of the game is as follows: The two players, with
Player~1 starting the game, alternate to choose transitions, or
\emph{moves}, in $T$, starting with transitions from given start
states $s$ and $t$ and continuing their choices from the targets of
the transitions chosen in the previous step.  At each of his turns,
Player~1 also makes a choice whether to choose a transition from the
target of his own previous choice, or from the target of his
opponent's previous choice (to ``switch paths'').  We use a
\emph{switch counter} to keep track of how often Player~1 has chosen
to switch paths.  Player~2 has then to respond with a transition from
the remaining target.  This game is played for an infinite number of
rounds, or until one player runs out of choices, thus building two
finite or infinite paths.  The value of the game is then the trace
distance of the traces of these two paths.

We proceed to formalize the above intuition.  A Player-1
\emph{configuration} of the game is a tuple $( \pi, \rho, m)\in
T^n\times T^n\times \Nat$, for $n\in \Nat$, such that for all $i\in\{
0,\dots, n- 2\}$, either $\src( \pi_{ i+ 1})= \tgt( \pi_i)$ and $\src(
\rho_{ i+ 1})= \tgt( \rho_i)$, or $\src( \pi_{ i+ 1})= \tgt( \rho_i)$
and $\src( \rho_{ i+ 1})= \tgt( \pi_i)$.  Similarly, a Player-2
configuration is a tuple $( \pi, \rho, m)\in T^{ n+ 1}\times T^n\times
\Nat$ such that for all $i\in\{ 0,\dots, n- 2\}$, either $\src( \pi_{ i+
  1})= \tgt( \pi_i)$ and $\src( \rho_{ i+ 1})= \tgt( \rho_i)$, or $\src(
\pi_{ i+ 1})= \tgt( \rho_i)$ and $\src( \rho_{ i+ 1})= \tgt( \pi_i)$;
and $\src( \pi_n)= \tgt( \pi_{ n- 1})$ or $\src( \pi_n)= \tgt( \rho_{ n-
  1})$.  The set of all Player-$i$ configurations is denoted $\Conf_i$.

Intuitively, the configuration $( \pi, \rho, m)$ keeps track of the
history of the game; $\pi$ stores the choices of Player~1, $\rho$ the
choices of Player~2, and $m$ is the switch counter.  Hence $\pi$ and
$\rho$ are sequences of transitions in $T$ which can be arranged by
suitable swapping to form two paths ($\bar\pi$,$\bar\rho$).  How exactly
these sequences are constructed is determined by a pair of
\emph{strategies} which specify for each player which edge to play from
any configuration.

A Player-1 strategy is hence a partial mapping $\theta_1: \Conf_1\parto
T\times \Nat$ such that for all $( \pi, \rho, m)\in \Conf_1$ for which
$\theta_1( \pi, \rho, m)=( e', m')$ is defined,
\begin{itemize}
\item $\src( e')= \tgt( \last( \pi))$ and $m'= m$ or $m'= m+ 1$, or
\item $\src( e')= \tgt( \last( \rho))$ and $m'= m+ 1$.
\end{itemize}
A Player-2 strategy is a partial mapping $\theta_2: \Conf_2\parto T\times
\Nat$ such that for all $( \pi\cdot e, \rho, m)\in \Conf_2$ for which
$\theta_2( \pi\cdot e, \rho, m)=( e', m')$ is defined, $m'= m$, and
$\src( e')= \tgt( \last( \rho))$ if $\src( e)= \tgt( \last( \pi))$,
$\src( e')= \tgt( \last( \pi))$ if $\src( e)= \tgt( \last( \rho))$.  The
sets of Player-1 and Player-2 strategies are denoted $\Theta_1$ and
$\Theta_2$.

Note that if Player~1 chooses a transition from the end of the previous
choice of Player~2 (case $\src( e')= \tgt( \last( \rho))$ above), then
the switch counter is increased; but Player~1 may also choose to
increase the switch counter without switching paths.  Player~2 does not
touch the switch counter.

We can now define what it means to \emph{update} a configuration
according to a strategy: For $\theta_1\in \Theta_1$ and $( \pi, \rho,
m)\in \Conf_1$, $\upd[ \theta_1]( \pi, \rho, m)$ is defined if
$\theta_1( \pi, \rho, m)=( e', m')$ is defined, and then $\upd[
\theta_1]( \pi, \rho, m)= ( \pi\cdot e', \rho, m')$.  Similarly, for
$\theta_2\in \Theta_2$ and $( \pi\cdot e, \rho, m)\in \Conf_2$, $\upd[
\theta_2]( \pi\cdot e, \rho, m)$ is defined if $\theta_2( \pi\cdot e,
\rho, m)=( e', m')$ is defined, and then $\upd[ \theta_2]( \pi\cdot e,
\rho, m)= ( \pi\cdot e, \rho\cdot e', m')$.

For any pair of states $( s, t)\in S\times S$, a pair of strategies $(
\theta_1, \theta_2)\in \Theta_1\times \Theta_2$ inductively determines a
sequence $( \pi^j, \rho^j, m^j)$ of configurations, by
\begin{align*}
  ( \pi^0, \rho^0, m^0) &= ( s, t, 0); \\
  ( \pi^{ 2j+ 1}, \rho^{ 2j+ 1}, m^{ 2j+ 1}) &=
  \begin{cases}
    \text{undef.} \qquad\text{if } \upd[ \theta_1]( \pi^{ 2j}, \rho^{
      2j}, m^{ 2j}) \text{ is undefined}, \\
    \upd[ \theta_1]( \pi^{ 2j}, \rho^{ 2j}, m^{ 2j}) \qquad\text{otherwise};
  \end{cases} \\
  ( \pi^{ 2j}, \rho^{ 2j}, m^{ 2j}) &=
  \begin{cases}
    \text{undef.} \qquad
    \begin{aligned}[t]
      &\text{if } \upd[ \theta_2]( \pi^{ 2j- 1}, \rho^{ 2j-1 }, m^{
        2j- 1}) \\
      &\hspace*{9.5em}\text{is undefined}, 
    \end{aligned} \\
    \upd[ \theta_2]( \pi^{ 2j- 1}, \rho^{ 2j- 1}, m^{ 2j- 1})
    \qquad\text{otherwise}.
  \end{cases}
\end{align*}
Note that indeed, we are updating configurations by alternating
between the two strategies $\theta_1$, $\theta_2$.

The configurations in this sequence satisfy $\pi^j\le_p \pi^{ j+ 1}$
and $\rho^j\le_p \rho^{ j+ 1}$ for all $j$, where $\le_p$ denotes
prefix ordering, hence the direct limits $\pi= \varinjlim \pi^j$,
$\rho= \varinjlim \rho^j$ exist (as finite or infinite paths).  By our
conditions on configurations, the pair $( \pi, \rho)$ in turn
determines a pair $( \bar \pi, \bar \rho)$ of \emph{paths} in $S$, as
follows:
\begin{align*}
  ( \bar \pi_1, \bar \rho_1) &=
  \begin{cases}
    ( \pi_1, \rho_1) &\text{if } \src( \pi_1)= s
    \\
    ( \rho_1, \pi_1) &\text{if } \src( \pi_1)= t
  \end{cases} \\
  ( \bar \pi_j, \bar \rho_j) &=
  \begin{cases}
    ( \pi_j, \rho_j) &\text{if } \src( \pi_j)=
    \tgt( \bar \pi_{ j- 1}) \\
    ( \rho_j, \pi_j) &\text{if } \src( \pi_j)=
    \tgt( \bar \rho_{ j- 1})
  \end{cases}
\end{align*}

The \emph{outcome} of the game when played from $( s, t)$ according to a
strategy pair $( \theta_1, \theta_2)$ is defined to be $\outcome(
\theta_1, \theta_2)( s, t)=( \bar \pi, \bar \rho)$, and its
\emph{utility} is defined by $\util( \theta_1, \theta_2)( s, t)= \trace
d( \tr{ \outcome( \theta_1, \theta_2)( s, t)})= \trace d( \tr{ \bar
  \pi}, \tr{ \bar \rho})$.

Recall that $\trace d$ is given as a parameter to the game; if we want
to make explicit the parametrization on the trace distance $\trace d$
on which utility depends, we write
$\util[ \trace d]( \theta_1, \theta_2)( s, t)$.

Note that $\util( \theta_1, \theta_2)( s, t)$ is defined both in case
the paths $\bar \pi$ and $\bar \rho$ are finite and in case they are
infinite (the case where one is finite and the other is infinite
cannot occur).  Also, if the paths are finite because
$\theta_1( \pi^j, \rho^j, m^j)$ was undefined for some configuration
$( \pi^j, \rho^j, m^j)$ in the sequence, then
$\len( \bar \pi)= \len( \bar \rho)$; if on the other hand the reason
is that $\theta_2( \pi^j, \rho^j, m^j)$ was undefined, then
$\len( \bar \pi)= \len( \bar \rho)\pm 1$, and
$\util( \theta_1, \theta_2)( s, t)= \infty$.  Hence if the game
reaches a configuration in which Player~2 has no moves available, the
utility is $\infty$.

The objective of Player~1 in the game is to maximize utility, whereas
Player~2 wants to minimize it.  Hence we define the \emph{value} of the
game from $( s, t)$ to be
\begin{equation*}
  v( s, t)= \adjustlimits \sup_{ \theta_1\in \Theta_1} \inf_{
    \theta_2\in \Theta_2} \util( \theta_1, \theta_2)( s, t)\,.
\end{equation*}
For a given subset $\Theta_1'\subseteq \Theta_1$ we will write
\begin{equation*}
  v( \Theta_1')( s, t)= \adjustlimits \sup_{ \theta_1\in
    \Theta_1'} \inf_{ \theta_2\in \Theta_2} \util( \theta_1, \theta_2)( s,
  t)\,,
\end{equation*}
and if we need to emphasize dependency of the value on the given trace
distance, we write $v( \trace d, \Theta_1')$.  The following lemma
states the immediate fact that if Player~1 has fewer strategies
available, the game value decreases.

\begin{lemma}
  \label{qltbt.le:strat-restrict}
  For all $\Theta_1'\subseteq \Theta_1''\subseteq \Theta_1$ and all $s, t\in
  S$, $v( \Theta_1')( s, t)\le v( \Theta_1'')( s, t)$.
\end{lemma}

The above definition of strategies is slightly too general in that
whether or not a strategy is defined in a given configuration should
only depend on the actual part of the configuration on which the
strategy has an effect.  We hence define a notion of \emph{uniformity}
which we will assume from now:

\begin{definition}
  A strategy $\theta_1\in \Theta_1$ is \emph{uniform} if it holds for
  all configurations
  $( \pi, \rho, m),( \pi, \tilde \rho, m),( \tilde \pi, \rho, m)\in
  \Conf_1$ that whenever $\theta_1( \pi, \rho, m)=( e', m')$ is
  defined,
  \begin{itemize}
  \item if $\src( e')= \tgt( \last( \pi))$, then also
    $\theta_1( \pi, \tilde \rho, m)$ is defined, and
  \item if $\src( e')= \tgt( \last( \rho))$, then also
    $\theta_1( \tilde \pi, \rho, m)$ is defined.
  \end{itemize}
  A strategy $\theta_2\in \Theta_2$ is uniform if it holds for all
  configurations
  $( \pi\cdot e, \rho, m),( \tilde \pi\cdot \tilde e, \rho, m),(
  \pi\cdot e, \tilde \rho, m)\in \Conf_2$ that whenever
  $\theta_2( \pi\cdot e, \rho, m)=( e', m')$ is defined,
  \begin{itemize}
  \item if $\src( e')= \tgt( \last( \rho))$, then also $\theta_2(
    \tilde \pi\cdot \tilde e, \rho, m)$ is defined, and
  \item if $\src( e')= \tgt( \last( \pi))$, then also $\theta_2(
    \pi\cdot e, \tilde \rho, m)$ is defined.
  \end{itemize}
\end{definition}

A subset $\Theta_1'\subseteq \Theta_1$ is uniform if all strategies in
$\Theta_1'$ are uniform.  Uniformity of strategies is used to combine
paths built from different starting states in the proof of
Proposition~\ref{qltbt.pr:hemi} below, and it allows us to show a
minimax theorem for our setting.

The concrete strategy subsets we will consider in later sections will
all be uniform, so from now on we only consider the subsets of
$\Theta_1$ and $\Theta_2$ consisting of uniform strategies.  Abusing
notation, we will also denote these by $\Theta_1$ and~$\Theta_2$.

\begin{lemma}
  \label{le:minimax}
  For any uniform $\Theta_1'\subseteq \Theta_1$ and all $s, t\in S$,
  \begin{equation*}
    \adjustlimits \sup_{ \theta_1\in \Theta_1'} \inf_{ \theta_2\in
      \Theta_2} \util( \theta_1, \theta_2)( s, t)= \adjustlimits \inf_{
      \theta_2\in \Theta_2} \sup_{ \theta_1\in \Theta_1'} \util( \theta_1,
    \theta_2)( s, t)\,.
  \end{equation*}
\end{lemma}

\begin{proof}
  By uniformity, neither of the two players has any possibility to
  influence the configurations reachable by the other's strategies.
  Hence it is immaterial which player gets to choose strategy first.
\end{proof}

\section{General Properties}
\label{qltbt.se:properties}

We show here that under the uniformity condition, the game value is
indeed a distance, and that results concerning inequalities in the
qualitative dimension can be transfered to topological inequivalences
in the quantitative setting.  Say that a Player-1 strategy
$\theta_1\in \Theta_1$ is \emph{non-switching} if it holds for all
$( \pi, \rho, m)$ for which $\theta_1( \pi, \rho, m)=( e', m')$ is
defined that $m= m'$, and let $\Theta_1^0$ be the set of non-switching
Player-1 strategies.  We first show a lemma which shows that any pair
of traces can be generated by a non-switching strategy:

\begin{lemma}
  \label{qltbt.le:construct_paths}
  For all $s, t\in S$ and all $\sigma\in \tracesfrom s$,
  $\tau\in \tracesfrom t$ there exist $\theta_1\in \Theta_1^0$ and
  $\theta_2\in \Theta_2$ for which
  $\util( \theta_1, \theta_2)( s, t)= \trace d( \sigma, \tau).$
\end{lemma}

\begin{proof}
  Let $( \pi, \rho, 0)\in \Conf_1$ for finite paths $\pi$, $\rho$ with
  $\len( \pi)= \len( \rho)= k\ge 0$ and $\tr \pi= \sigma_0\dots
  \sigma_{ k- 1}$, $\tr \rho= \tau_0\dots \tau_{ k- 1}$.  If $\len(
  \sigma)\ge k$, then there is $e=( \last( \pi), \sigma_k, s')\in T$,
  and we define $\theta_1( \pi, \rho, 0)=( e, 0)$.  If also $\len(
  \tau)\ge k$, then there is $e'=( \last( \rho), \tau_k, t')\in T$,
  and we let $\theta_2( \pi\cdot e, \rho, 0)=( e', 0)$.

  Let $(\bar \pi, \bar \rho)= \outcome( \theta_1, \theta_2)( s, t)$.
  If both $\sigma$ and $\tau$ are infinite traces, then
  $\tr{ \bar \pi}= \sigma$ and $\tr{ \bar \rho}= \tau$; otherwise,
  $\tr{ \bar \pi}$ and $\tr{ \bar \rho}$ will be finite prefixes of
  $\sigma$ and $\tau$ for which
  $\trace d( \tr{ \bar \pi}, \tr{ \bar \rho})= \trace d( \sigma,
  \tau)$. \qed
\end{proof}

The following proposition shows that the distance defined by our
quantitative game is a hemimetric.  Note that the proof of the
triangle inequality uses uniformity.

\begin{proposition}
  \label{qltbt.pr:hemi}
  For all $\Theta_1'\subseteq \Theta_1$ with
  $\Theta_1^0\subseteq \Theta_1'$, $v( \Theta_1')$ is a hemimetric on
  $S$.
\end{proposition}

\begin{proof}
  We write $v= v( \Theta_1')$ during this proof.  It is clear that
  $v( s, s)= 0$ for all $s\in S$: if the players are making their
  choices from the same state, Player~2 can always answer by choosing
  exactly the same transition as Player~1.  For proving the triangle
  inequality $v( s, u)\le v( s, t)+ v( t, u)$, let $\varepsilon> 0$
  and use Lemma~\ref{le:minimax} to choose Player-2 strategies
  $\theta_2^{ s, t}, \theta_2^{ t, u}\in \Theta_2$ for which
  \begin{equation}
    \label{qltbt.eq:suptheta}
    \begin{aligned}
      \sup_{ \theta_1\in \Theta_1'} \util( \theta_1, \theta_2^{ s, t})(
      s, t) &< v( s, t)+ \tfrac \varepsilon 2\,, \\
      \sup_{ \theta_1\in \Theta_1'} \util( \theta_1, \theta_2^{ t, u})(
      t, u) &< v( t, u)+ \tfrac \varepsilon 2\,.
    \end{aligned}
  \end{equation}
  We define a strategy $\theta_2^{ s, u}\in \Theta_2$ which uses three paths
  and two configurations in $S$ as extra memory.  This is only for
  convenience, as these can be reconstructed by Player~2 at any time; hence
  we do not extend the capabilities of Player~2:
  \begin{multline*}
    \theta_2^{ s, u}( \pi\cdot e, \chi, m; \bar \pi, \bar \rho', \bar \chi,
    \pi', \rho_1', \rho_2', \chi')= \\
    \begin{cases}
      \begin{aligned}[b]
        \Big(%
        & \theta_2^{ t, u}\big( \rho_2'\cdot \theta_{ 2, 1}^{ s, t}(
        \pi'\cdot e, \rho_1', m), \chi', m\big); \\
        & \bar \pi\cdot e, \\
        & \bar \rho'\cdot \theta_{ 2, 1}^{ s, t}( \pi'\cdot e, \rho_1', m), \\
        & \bar \chi\cdot \theta_{ 2, 1}^{ t, u}\big( \rho_2'\cdot \theta_{
          2, 1}^{ s, t}( \pi'\cdot e, \rho_1', m)\big), \\
        & \pi'\cdot e, \\
        & \rho_1'\cdot \theta_{ 2, 1}^{ s, t}( \pi'\cdot e, \rho_1', m), \\
        & \rho_2'\cdot \theta_{ 2, 1}^{ s, t}( \pi'\cdot e, \rho_1', m), \\
        & \chi'\cdot \theta_{ 2, 1}^{ t, u}\big( \rho_2'\cdot \theta_{ 2, 1}^{
          s, t}( \pi'\cdot e, \rho_1', m)\big)%
        \Big)
      \end{aligned}
      &\text{if } \src( e)= \tgt( \last( \bar \pi)), \\[2ex]
      \begin{aligned}[b]
        \Big(%
        & \theta_2^{ s, t}\big( \pi'\cdot \theta_{ 2, 1}^{ t, u}(
        \rho_2'\cdot e, \chi', m), \rho_1', m\big); \\
        & \bar \pi\cdot \theta_{ 2, 1}^{ s, t}\big( \pi'\cdot \theta_{ 2, 1}^{
          t, u}( \rho_2'\cdot e, \chi', m)\big), \\
        & \bar \rho'\cdot \theta_{ 2, 1}^{ t, u}( \rho_2'\cdot e, \chi', m), \\
        & \bar \chi\cdot e, \\
        & \pi'\cdot \theta_{ 2, 1}^{ t, u}( \rho_2'\cdot e, \chi', m), \\
        & \rho_1'\cdot \theta_{ 2, 1}^{ s, t}\big( \pi'\cdot \theta_{ 2, 1}^{ t,
          u}( \rho_2'\cdot e, \chi', m)\big) \\
        & \rho_2'\cdot e, \\
        & \chi'\cdot \theta_{ 2, 1}^{ t, u}( \rho_2'\cdot e, \chi', m)%
        \Big)
      \end{aligned}
      &\text{if } \src( e)= \tgt( \last( \bar \chi)).
    \end{cases}
  \end{multline*}
  In the beginning of the game, all memory paths are initialized to be
  empty.

  In the expression above, the strategy $\theta_2^{ s, u}$ is
  constructed from the strategies $\theta_2^{ s, t}$ and $\theta_2^{ t,
    u}$ by using the answer to the move of Player~1 in one of the games
  as an emulated Player-1 move in the other.  The paths $\bar \pi$,
  $\bar \chi$ are constructed from the configuration $( \pi, \chi)$ of
  the $( s, u)$-game and are only kept in memory so that we can see
  whether Player~1 is playing an edge prolonging $\bar \pi$ or $\bar
  \chi$.  The pair $( \pi', \rho_1')$ is the configuration in the $( s,
  t)$-game we are emulating, and $( \rho_2', \chi')$ is the $( t,
  u)$-configuration.  The path $\bar \rho'= \bar \rho_1'= \bar \rho_2'$
  is common for the paths $( \bar \pi', \bar \rho_1')$, $( \bar \rho_2',
  \bar \chi')$ constructed from $( \pi', \rho_1')$ and $( \rho_2',
  \chi')$.

  If Player~1 has played an edge $e$ prolonging $\bar \pi$ (first case
  above), we compute an answer move $( e', m)= \theta_2^{ s, t}(
  \pi'\cdot e, \rho_1', m)$ to this in the $( s, t)$-game.  This answer
  is then used to emulate a Player-1 move in the $( t, u)$-game, and the
  answer $\theta_2^{ t, u}( \rho_2'\cdot e', \chi', m)$ to this is what
  Player~2 plays in the $( s, u)$-game.  The memory is updated
  accordingly.  If on the other hand, Player~1 has played an edge $e$
  prolonging $\bar \chi$, we play in the $( t, u)$-game first and use
  the answer $( e', m)= \theta_2^{ t, u}( \rho_2'\cdot e, \chi', m)$ in
  the $( s, t)$-game to compute $\theta_2^{ s, t}( \pi'\cdot e',
  \rho_1', m)$.  Figure~\ref{qltbt.fi:triangle_scramble} gives an illustration
  of how the configurations are updated during the game; note that
  uniformity of $\Theta_1'$ is necessary for being able to emulate
  Player-1 moves from one game in another.

  \begin{figure}[tpb]
    \centering
    \begin{tikzpicture}[scale=.9]
      \tikzstyle{every node}=[font=\small]
      \node at (0,.2) {$s$};
      \node at (3,.2) {$t$};
      \node at (6,.2) {$u$};
      \begin{scope}[draw=white,every node/.style=left]
        \draw (0,0) -- node {$\pi$} (-.2,-1);
        \draw (5.8,-1) -- node {$\pi$} (6,-2);
        \draw (6,-2) -- node {$\pi$} (5.8,-3);
        \draw (-.2,-3) -- node {$\pi$} (0,-4);
      \end{scope}
      \begin{scope}[draw=white,every node/.style=right]
        \draw (6,0) -- node {$\chi$} (5.8,-1);
        \draw (-.2,-1) -- node {$\chi$} (0,-2);
        \draw (0,-2) -- node {$\chi$} (-.2,-3);
        \draw (5.8,-3) -- node {$\chi$} (6,-4);
      \end{scope}
      \begin{scope}[draw=yellow,every node/.style=right]
        \draw[thick] (0,0) -- node {$\pi'$} (-.2,-1);
        \draw[thick] (2.4,-1) -- node {$\pi'$} (2.2,-2);
        \draw[thick] (2.2,-2) -- node {$\pi'$} (2.4,-3);
        \draw[thick] (-.2,-3) -- node {$\pi'$} (0,-4);
      \end{scope}
      \begin{scope}[draw=yellow!50!red,every node/.style=left]
        \draw[thick] (2.2,0) -- node {$\rho_1'$} (2.4,-1);
        \draw[thick] (-.2,-1) -- node {$\rho_1'$} (0,-2);
        \draw[thick] (0,-2) -- node {$\rho_1'$} (-.2,-3);
        \draw[thick] (2.4,-3) -- node {$\rho_1'$} (2.2,-4);
      \end{scope}
      \begin{scope}[draw=red!70!white,every node/.style=right]
        \draw[thick] (3.8,0) -- node {$\rho_2'$} (4,-1);
        \draw[thick] (5.8,-1) -- node {$\rho_2'$} (6,-2);
        \draw[thick] (6,-2) -- node {$\rho_2'$} (5.8,-3);
        \draw[thick] (4,-3) -- node {$\rho_2'$} (3.8,-4);
      \end{scope}
      \begin{scope}[draw=red!30!black,every node/.style=left]
        \draw[thick] (6,0) -- node {$\chi'$} (5.8,-1);
        \draw[thick] (4,-1) -- node {$\chi'$} (3.8,-2);
        \draw[thick] (3.8,-2) -- node {$\chi'$} (4,-3);
        \draw[thick] (5.8,-3) -- node {$\chi'$} (6,-4);
      \end{scope}
    \end{tikzpicture}
    \caption{\label{qltbt.fi:triangle_scramble}%
      Configuration update in the game used for showing the triangle
      inequality}
  \end{figure}
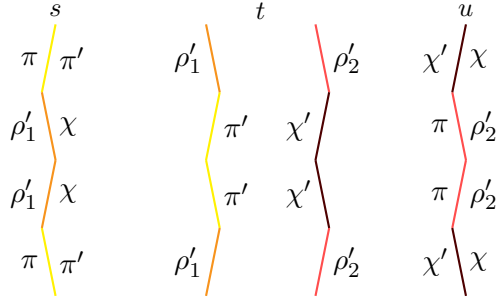

  Take now any $\theta_1^{ s, u}\in \Theta_1'$, let $(\bar \pi, \bar
  \chi)= \outcome( \theta_1^{ s, u}, \theta_2^{ s, u})( s, u)$, and let
  $\bar \rho'$ be the corresponding memory path.  By
  Lemma~\ref{qltbt.le:construct_paths} there exist $\theta_1^{ s, t},
  \theta_1^{ t, u}\in \Theta_1'$ for which $\trace d( \tr{ \bar
    \pi}, \tr{ \bar \rho'})= \util( \theta_1^{ s, t}, \theta_2^{ s, t})(
  s, t)$ and $\trace d( \tr{ \bar \rho'}, \tr{ \bar \chi})= \util(
  \theta_1^{ t, u}, \theta_2^{ t, u})( t, u)$.  Using
  Equation~\eqref{qltbt.eq:suptheta} we have
  \begin{align*}
    \util( \theta_1^{ s, u}, \theta_2^{ s, u})( s, u) &= \trace d( \tr{
      \bar \pi}, \tr{ \bar \chi}) \\
    &\le \trace d( \tr{ \bar \pi}, \tr{ \bar \rho})+ \trace d( \tr{ \bar
      \rho}, \tr{ \bar \chi}) \\
    &< v( s, t)+ v( t, u)+ \varepsilon
  \end{align*}
  and hence also $\inf_{ \theta_2\in \Theta_2} \util( \theta_1^{ s, u},
  \theta_2)( s, u)< v( s, t)+ v( t, u)+ \varepsilon$.  As the choice of
  $\theta_1^{ s, u}$ was arbitrary, this implies
  \begin{equation*}
    \sup_{ \theta_1\in
      \Theta_1'} \inf_{ \theta_2\in \Theta_2} \util( \theta_1, \theta_2)(
    s, u)\le v( s, t)+ v( t, u)+ \varepsilon\,,
  \end{equation*}
  and as also $\varepsilon$ was chosen arbitrarily, we have $v( s, u)\le v(
  s, t)+ v( t, u)$. \qed
\end{proof}

Next we show a \emph{transfer principle} which allows us to generalize
counterexamples regarding the equivalences in the qualitative
linear-time--branching-time spectrum~\cite{inbook/hpa/Glabbeek01} to the
qualitative setting.  We will make use of this principle later to show
that all distances we introduce are topologically inequivalent.

\begin{lemma}
  \label{qltbt.th:transfer_princ}
  Let $\Theta_1', \Theta_1''\subseteq \Theta_1$, and assume $\trace d$
  to be separating.  If there exist states $s, t\in S$ for which
  $v( \disc{ \trace d}, \Theta_1')( s, t)= 0$ and
  $v( \disc{ \trace d}, \Theta_1'')( s, t)= \infty$, then
  $v( \trace d, \Theta_1')$ and $v( \trace d, \Theta_1'')$ are
  topologically inequivalent.
\end{lemma}

\begin{proof}
  By $v( \disc{ \trace d}, \Theta_1')( s, t)= 0$, we know that for any
  $\theta_1\in \Theta_1'$ there exists $\theta_2\in \Theta_2$ for
  which $( \bar \pi, \bar \rho)= \outcome( \theta_1, \theta_2)( s, t)$
  satisfy $\tr{ \bar \pi}= \tr{ \bar \rho}$, hence also
  $v( \trace d, \Theta_1')( s, t)= 0$.  Conversely, as $\trace d$ is
  separating, $v( \trace d, \Theta_1'')( s, t)= 0$ would imply that
  also $v( \disc{ \trace d}, \Theta_1'')( s, t)= 0$, hence we must
  have $v( \trace d, \Theta_1'')( s, t)\ne 0$, entailing topological
  inequivalence. \qed
\end{proof}

\section{The Distance Spectrum}
\label{qltbt.se:spectrum}

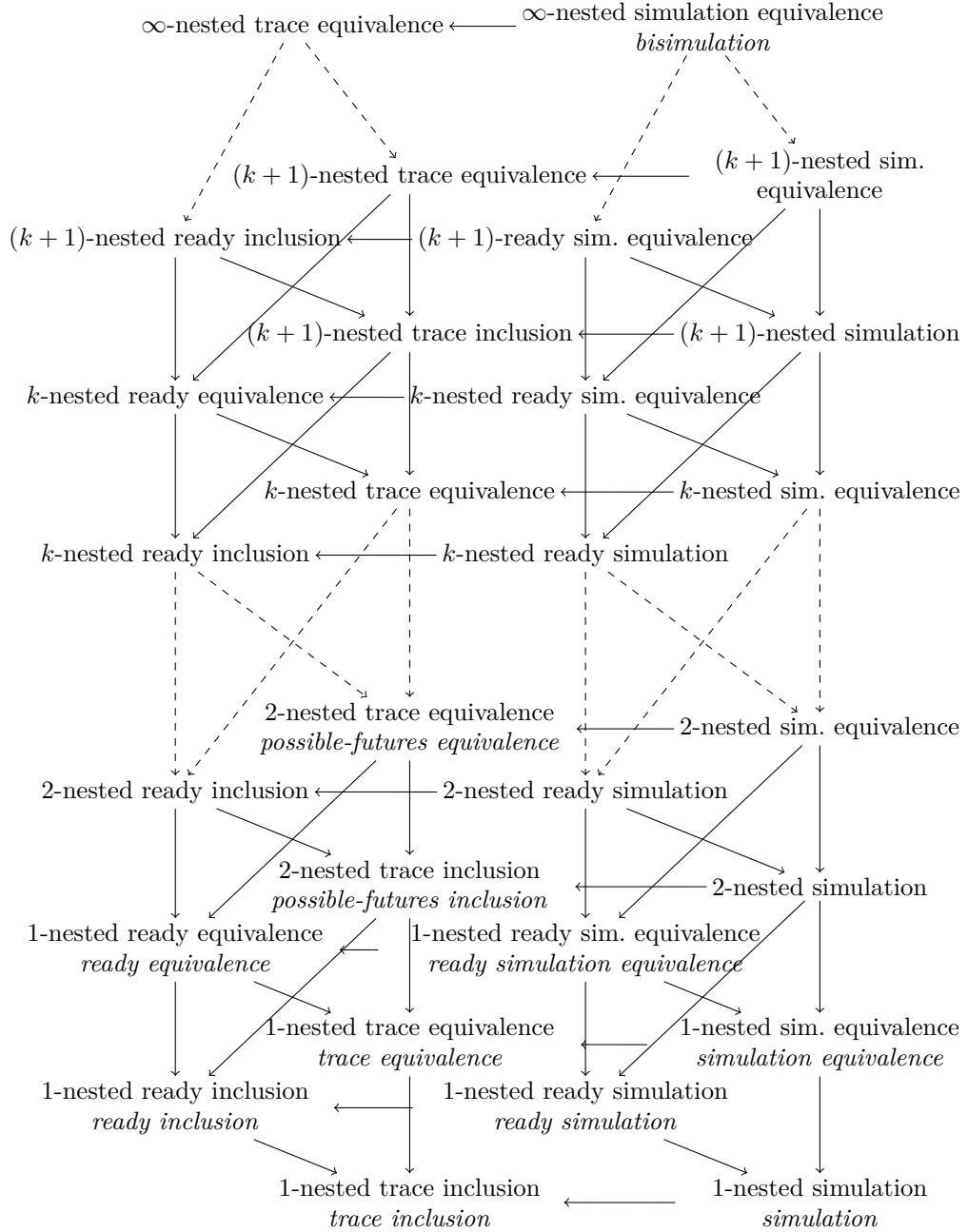
\begin{figure}[tbp]
  \centering
  \begin{tikzpicture}[->,xscale=.81,yscale=1.1]
    \tikzstyle{every node}=[font=\small,text badly centered, inner sep=2pt]
    \node (traceeq) at (0,.3) {$\infty$-nested trace equivalence};
    \node (k+1-r-trace) at (-2,-2.4) {$( k+ 1)$-nested ready
      inclusion};
    \node (k+1-traceeq) at (2,-1.6) {$( k+ 1)$-nested trace equivalence};
    \node (k-r-traceeq) at (-2,-4.4) {$k$-nested ready equivalence};
    \node (k+1-trace) at (2,-3.6) {$( k+ 1)$-nested trace inclusion};
    \node (k-r-trace) at (-2,-6.4) {$k$-nested ready inclusion};
    \node (k-traceeq) at (2,-5.6) {$k$-nested trace equivalence};
    \node (2-r-trace) at (-2,-9.4) {$2$-nested ready inclusion};
    \node [text width=11.6em] (2-traceeq) at (2,-8.6) {$2$-nested trace
      equivalence \\ \emph{possible-futures equivalence}};
    \node [text width=11.4em] (1-r-traceeq) at (-2,-11.4) {$1$-nested ready
      equivalence \\ \emph{ready equivalence}};
    \node [text width=11.5em] (2-trace) at (2,-10.6) {$2$-nested trace
      inclusion \\ \emph{possible-futures inclusion}};
    \node [text width=11em] (1-r-trace) at (-2,-13.4) {$1$-nested ready
      inclusion \\ \emph{ready inclusion}};
    \node [text width=11.9em] (1-traceeq) at (2,-12.6) {$1$-nested trace
      equivalence \\ \emph{trace equivalence}};
    \node [text width=10.7em] (1-trace) at (2,-14.6) {$1$-nested trace
      inclusion \\ \emph{trace inclusion}};
    \node [text width=13em] (bisim) at (7,.3) {$\infty$-nested
      simulation equivalence \\ \emph{bisimulation}};
    \node (k+1-r-sim) at (5,-2.4) {$( k+ 1)$-ready sim.~equivalence};
    \node [text width=9em] (k+1-simeq) at (9,-1.6) {$( k+ 1)$-nested
      sim. equivalence};
    \node (k-r-simeq) at (5,-4.4) {$k$-nested ready sim.~equivalence};
    \node (k+1-sim) at (9,-3.6) {$( k+ 1)$-nested simulation};
    \node (k-r-sim) at (5,-6.4) {$k$-nested ready simulation};
    \node (k-simeq) at (9,-5.6) {$k$-nested sim.~equivalence};
    \node (2-r-sim) at (5,-9.4) {$2$-nested ready simulation};
    \node (2-simeq) at (9,-8.6) {$2$-nested sim.~equivalence};
    \node [text width=14.5em] (1-r-simeq) at (5,-11.4) {$1$-nested ready
      sim.~equivalence \\ \emph{ready simulation equivalence}};
    \node (2-sim) at (9,-10.6) {$2$-nested simulation};
    \node [text width=12em] (1-r-sim) at (5,-13.4) {$1$-nested ready
      simulation \\ \emph{ready simulation}};
    \node [text width=12em] (1-simeq) at (9,-12.6) {$1$-nested
      sim.~equivalence \\ \emph{simulation equivalence}};
    \node [text width=10em] (1-sim) at (9,-14.6) {$1$-nested
      simulation \\ \emph{simulation}};
    %
    \path (bisim) edge (traceeq);
    \path (k+1-r-sim) edge (k+1-r-trace);
    \path (k+1-simeq) edge (k+1-traceeq);
    \path (k-r-simeq) edge (k-r-traceeq);
    \path (k+1-sim) edge (k+1-trace);
    \path (k-r-sim) edge (k-r-trace);
    \path (k-simeq) edge (k-traceeq);
    \path (2-simeq) edge (2-traceeq);
    \path (2-r-sim) edge (2-r-trace);
    \path (2-sim) edge (2-trace);
    \path (1-r-simeq) edge (1-r-traceeq);
    \path (1-simeq) edge (1-traceeq);
    \path (1-r-sim) edge (1-r-trace);
    \path (1-sim) edge (1-trace);
    \path [dashed] (traceeq) edge (k+1-r-trace);
    \path [dashed] (traceeq) edge (k+1-traceeq);
    \path (k+1-r-trace) edge (k-r-traceeq);
    \path (k+1-r-trace) edge (k+1-trace);
    \path (k+1-traceeq) edge (k-r-traceeq);
    \path (k+1-traceeq) edge (k+1-trace);
    \path (k-r-traceeq) edge (k-r-trace);
    \path (k-r-traceeq) edge (k-traceeq);
    \path (k+1-trace) edge (k-r-trace);
    \path (k+1-trace) edge (k-traceeq);
    \path [dashed] (k-r-trace) edge (2-r-trace);
    \path [dashed] (k-r-trace) edge (2-traceeq);
    \path [dashed] (k-traceeq) edge (2-r-trace);
    \path [dashed] (k-traceeq) edge (2-traceeq);
    \path (2-r-trace) edge (1-r-traceeq);
    \path (2-r-trace) edge (2-trace);
    \path (2-traceeq) edge (1-r-traceeq);
    \path (2-traceeq) edge (2-trace);
    \path (1-r-traceeq) edge (1-r-trace);
    \path (1-r-traceeq) edge (1-traceeq);
    \path (2-trace) edge (1-r-trace);
    \path (2-trace) edge (1-traceeq);
    \path (1-r-trace) edge (1-trace);
    \path (1-traceeq) edge (1-trace);
    \path [dashed] (bisim) edge (k+1-r-sim);
    \path [dashed] (bisim) edge (k+1-simeq);
    \path (k+1-r-sim) edge (k-r-simeq);
    \path (k+1-r-sim) edge (k+1-sim);
    \path (k+1-simeq) edge (k-r-simeq);
    \path (k+1-simeq) edge (k+1-sim);
    \path (k-r-simeq) edge (k-r-sim);
    \path (k-r-simeq) edge (k-simeq);
    \path (k+1-sim) edge (k-r-sim);
    \path (k+1-sim) edge (k-simeq);
    \path [dashed] (k-r-sim) edge (2-r-sim);
    \path [dashed] (k-r-sim) edge (2-simeq);
    \path [dashed] (k-simeq) edge (2-r-sim);
    \path [dashed] (k-simeq) edge (2-simeq);
    \path (2-r-sim) edge (1-r-simeq);
    \path (2-r-sim) edge (2-sim);
    \path (2-simeq) edge (1-r-simeq);
    \path (2-simeq) edge (2-sim);
    \path (1-r-simeq) edge (1-r-sim);
    \path (1-r-simeq) edge (1-simeq);
    \path (2-sim) edge (1-r-sim);
    \path (2-sim) edge (1-simeq);
    \path (1-r-sim) edge (1-sim);
    \path (1-simeq) edge (1-sim);
  \end{tikzpicture}
  \caption{\label{qltbt.fi:spectrum}%
    The quantitative linear-time--branching-time spectrum.  The nodes
    are the different system distances introduced in this chapter, and
    an edge $d_1\longrightarrow d_2$ or $d_1\dashrightarrow d_2$
    indicates that $d_1( s, t)\ge d_2( s, t)$ for all states $s$, $t$,
    and that $d_1$ and $d_2$ in general are topologically
    inequivalent.}
\end{figure}

In this section we introduce the distances depicted in
Figure~\ref{qltbt.fi:spectrum} and show their relationship.  Note
again that the results obtained here are independent of the particular
trace distance considered.  Throughout this section, we fix a LTS
$( S, T)$ and a trace distance
$\trace d: \KK^\infty\times \KK^\infty\to \Realnn\cup\{ \infty\}$.

\subsection{Branching Distances}

If the switching counter in the game introduced in Section~\ref{qltbt.se:game}
is unbounded, Player~1 can choose at any move whether to prolong the
previous choice or to switch paths, hence this resembles the
bisimulation game~\cite{DBLP:conf/banff/Stirling95}.

\begin{definition}
  The \emph{bisimulation distance} between $s$ and $t$ is $\bisim d( s,
  t)= v( s, t)$.
\end{definition}

\begin{theorem}
  \label{qltbt.th:bisim}
  For $\trace d= \disc{ \trace d}$ the discrete trace distance,
  $\disc {\bisim d}( s, t)= 0$ iff $s$ and $t$ are bisimilar.
\end{theorem}

\begin{proof}
  By discreteness of $\disc{ \trace d}$, we have $\disc{ \bisim
    d}( s, t)= 0$ iff it holds that for all $\theta_1\in
  \Theta_1$ there exists $\theta_2\in \Theta_2$ for which $\util(
  \theta_1, \theta_2)( s, t)= 0$.  Hence for each reachable Player-1
  configuration $( \pi, \rho, m)$ with $\theta_1( \pi, \rho, m)=( e',
  m')$, we have $\theta_2( \pi\cdot e', \rho, m')= (e'', m')$ with $\tr{
    e'}= \tr{ e''}$, \ie~Player~2 matches the labels chosen by Player~1
  precisely, implying that $s$ and $t$ are bisimilar.  The proof of the
  other direction is trivial. \qed
\end{proof}

We can restrict the strategies available to Player~1 by allowing only a
pre-defined finite number of switches:
\begin{equation*}
  \nsim k \Theta_1=  \{ \theta_1\in \Theta_1 \mid \text{if }
  \theta_1( \pi, \rho, m)=( e', m') \text{ is defined, then } m'\le k-
  1\}
\end{equation*}
In the so-defined $k$-nested simulation game, Player~1 is
only allowed to switch paths $k-1$ times during the game.  Note that
$\nsim 1 \Theta_1= \Theta_1^0$ is the set of non-switching strategies.

\begin{definition}
  The \emph{$k$-nested simulation distance} from $s$ to $t$, for $k\in
  \Natp$, is $\nsim k d( s, t)= v( \nsim k \Theta_1)( s, t)$.  The
  \emph{$k$-nested simulation equivalence distance} between $s$ and $t$ is
  $\nsimeq k d( s, t)= \max( v( \nsim k \Theta_1)( s, t), v( \nsim k
  \Theta_1)( t, s))$.
\end{definition}

\begin{theorem}
  \label{qltbt.th:nestedsim}
  For $\trace d= \disc{ \trace d}$ the discrete trace distance,
  \begin{itemize}
  \item $\disc{ \nsim k d}( s, t)= 0$ iff there is a
    $k$-nested simulation from $s$ to $t$,
  \item $\disc{ \nsimeq k d}( s, t)= 0$ iff there is a
    $k$-nested simulation equivalence between $s$ and $t$.
  \end{itemize}
\end{theorem}

Especially, $\disc{ \nsim 1 d}$ corresponds to the usual
\emph{simulation} preorder, and $\disc{ \nsim 2 d}$ to
\emph{two-nested simulation}.  Similarly, $\disc{ \nsimeq 1 d}$ is
\emph{similarity}, and $\disc{ \nsimeq 2 d}$ is \emph{two-nested
  simulation equivalence}.  We refer
to~\cite{DBLP:journals/iandc/GrooteV92,DBLP:journals/jacm/HennessyM85}
for definitions and discussion of two-nested and $k$-nested
simulation.

\begin{proof}
  This is similar to the proof of Theorem~\ref{qltbt.th:bisim}: If $\disc{
    \nsim k d}( s, t)= 0$, then any $\theta_1\in \nsim k \Theta_1$ has a
  counter-strategy $\theta_2\in \Theta_2$ which matches the labels
  chosen by Player~1 precisely, implying $k$-nested simulation from $s$
  to $t$.  The other direction is again trivial. \qed
\end{proof}

\begin{theorem}
  \label{qltbt.th:spectrum-1}
  For all $k, \ell\in \Natp$ with $k< \ell$ and all $s, t\in S$,
  \begin{equation*}
    \nsimeq k d( s, t)\le \nsim \ell d( s, t)\le \nsimeq \ell
    d( s, t)\le \bisim d( s, t).
  \end{equation*}
  If the trace distance $\trace d$ is separating, then all distances
  above are topologically inequivalent.
\end{theorem}

\begin{proof}
  The first part of the theorem follows from the inclusions
  $\nsimeq k \Theta_1\subseteq \nsim \ell \Theta_1\subseteq \nsimeq
  \ell \Theta_1\subseteq \Theta_1$ and
  Lemma~\ref{qltbt.le:strat-restrict}.  Topological inequivalence
  follows from Lemma~\ref{qltbt.th:transfer_princ} and the fact that
  for the discrete relations corresponding to the distances above
  (obtained by letting $\trace d= \disc{ \trace d}$), the inequalities
  are strict~\cite{inbook/hpa/Glabbeek01}. \qed
\end{proof}

As a variation of $k$-nested simulation, we can consider strategies
which allow Player~1 to switch paths $k$ times during the game, but at
the last switch, he may only pose \emph{one} transition as a challenge,
to which Player~2 must answer, and then the game finishes:
\begin{equation*}
  \nrsim k \Theta_1=  \{ \theta_1\in \Theta_1 \mid \text{if }
  \theta_1( \pi, \rho, m) \text{ is defined, then } m\le k- 1\}
\end{equation*}
Hence after his $k$'th switch, Player~1 has no more moves available, and
the game finishes after the answer move of Player~2.  Again, we allow
Player~1 to increase the switch counter without actually switching
paths.

\begin{definition}
  The \emph{$k$-nested ready simulation distance} from $s$ to $t$, for $k\in
  \Natp$, is $\nrsim k d( s, t)= v( \nrsim k \Theta_1)( s, t)$.  The
  \emph{$k$-nested ready simulation equivalence distance} between $s$ and
  $t$ is $\nrsimeq k d( s, t)= \max( v( \nrsim k \Theta_1)( s, t),
  v( \nrsim k \Theta_1)( t, s))$.
\end{definition}

For the discrete case, it seems only $k= 1$ has been considered; the
proof is similar to the one of Theorem~\ref{qltbt.th:bisim}.

\begin{theorem}
  For $\trace d= \disc{ \trace d}$ the discrete trace distance,
  \begin{itemize}
  \item $\disc{ \nrsim 1 d}( s, t)= 0$ iff there is a ready
    simulation from $s$ to $t$,
  \item $\disc{ \nrsimeq 1 d}( s, t)= 0$ iff $s$ and $t$
    are ready simulation equivalent.
  \end{itemize}
\end{theorem}

The next theorem finishes our work on the right half of
Figure~\ref{qltbt.fi:spectrum}.

\begin{theorem}
  \label{qltbt.th:spectrum-2}
  For all $k, \ell\in \Natp$ with $k< \ell$ and all $s, t\in S$,
  \begin{gather*}
    \nsim k d( s, t)\le \nrsim k d( s, t)\le \nsim \ell
    d( s, t)\,, \\
    \nsimeq k d( s, t)\le \nrsimeq k d( s, t)\le \nsimeq \ell d( s,
    t)\,.
  \end{gather*}
  Additionally, $\nrsim k d$ and $\nsimeq k d$ are incomparable, and
  also $\nrsimeq k d$ and $\nsim{( k+ 1)} d$ are incomparable.  If the
  trace distance $\trace d$ is separating, then all distances above
  are topologically inequivalent.
\end{theorem}

\begin{proof}
  Like in the proof of Theorem~\ref{qltbt.th:spectrum-1}, the inequalities follow
  from strategy set inclusions and topological inequivalence from
  Lemma~\ref{qltbt.th:transfer_princ}.  The incomparability results follow
  from the corresponding results for $\disc{ \trace d}$ and
  Lemma~\ref{qltbt.th:transfer_princ}. \qed
\end{proof}

\subsection{Linear Distances}

Above we have introduced the distances in the right half of the
quantitative linear-time--branching-time spectrum in
Figure~\ref{qltbt.fi:spectrum} and shown the relations claimed in the diagram.
To develop the left half, we need the notion of \emph{blind} strategies.
For any subset $\Theta_1'\subseteq \Theta_1$ we define the set of blind
$\Theta_1'$-strategies by
\begin{multline*}
  \blind \Theta_1'=\{ \theta_1\in \Theta_1'\mid
  \forall \pi, \rho, \rho', m: \theta_1( \pi, \rho, m)= \theta_1( \pi,
  \rho', m), \\ \text{ or } \theta_{ 1}( \pi, \rho, m)=( e, m+ 1) \text{
    and } \tgt( \last( \rho))\ne \tgt( \last( \rho'))\}.
\end{multline*}
Hence in such a blind strategy, either the edge chosen by Player~1 does not
depend on the choices of Player~2, or the switch counter is increased, in
which case the Player-1 choice only depends on the target of the last choice
of Player~2 (note that this dependency is necessary if Player~1 wants to
switch paths).

Now we can define, for $s, t\in S$ and $k\in \Natp$,
\begin{itemize}
\item the \emph{$\infty$-nested trace equivalence distance}: $\ntraceeq
  \infty d( s, t)= v( \blind \Theta_1)( s, t)$,
\item the \emph{$k$-nested trace distance}: $\ntrace k d( s, t)=
  v( \nsim k{ \blind \Theta_1})( s, t)$,
\item the \emph{$k$-nested trace equivalence distance}:\\\mbox{} \hfill
  $\ntraceeq k d( s, t)= \max( v( \nsim k{ \blind \Theta_1})( s,
  t), v( \nsim k{ \blind \Theta_1})( t, s))$,
\item the \emph{$k$-nested ready distance}: $\nrtrace k d( s, t)=
  v( \nrsim k{ \blind \Theta_1})( s, t)$, and
\item the \emph{$k$-nested ready equivalence distance}:\\\mbox{} \hfill
  $\nrtraceeq k d( s, t)= \max( v( \nrsim k{ \blind \Theta_1})( s,
  t), v( \nrsim k{ \blind \Theta_1})( t, s))$.
\end{itemize}

Our approach is justified by the following lemma which shows that the
($1$-nested) trace distance from $s$ to $t$ is precisely the Hausdorff
distance between the sets of traces available from $s$ and $t$,
respectively.

\begin{lemma}
  \label{qltbt.le:trace_is_trace}
  For $s, t\in S$, $\ntrace 1 d( s, t)= \sup_{ \sigma\in \tracesfrom s} \inf_{
    \tau\in \tracesfrom t} \trace d( \sigma, \tau)$.
\end{lemma}

\begin{proof}
  We have $\ntrace 1 d( s, t)= v( \blind \Theta_1^0)( s, t)$, with
  $\blind \Theta_1^0=\{ \theta_1\in \Theta_1^0\mid \forall \pi, \rho,
  \rho', m: \theta_1( \pi, \rho, m)= \theta_1( \pi, \rho', m)\}$.
  Hence, and as strategies in $\Theta_1^0$ are non-switching, every
  strategy $\theta_1\in \blind \Theta_1^0$ gives rise to precisely one
  trace $\sigma= \sigma( \theta_1)\in \tracesfrom s$ independently of Player-2
  strategy $\theta_2\in \Theta_2$.  Conversely, by
  Lemma~\ref{qltbt.le:construct_paths} (noticing that indeed, we have
  constructed a blind Player-1 strategy in the proof of that lemma),
  every trace $\sigma\in \tracesfrom s$ is generated by a strategy
  $\theta_1\in \blind \Theta_1^0$ with $\sigma= \sigma( \theta_1)$.

  We can finish the proof by showing that for all $\theta_1\in \blind
  \Theta_1^0$,
  \begin{equation*}
    \inf_{ \theta_2\in \Theta_2} \trace d( \sigma( \theta_1), \tr{ \bar
      \rho( \theta_1, \theta_2)})= \inf_{ \tau\in \tracesfrom t} \trace d(
    \sigma( \theta_1), \tau)\,.
  \end{equation*}
  But again using Lemma~\ref{qltbt.le:construct_paths}, we see that
  any $\tau\in \tracesfrom t$ is generated by a strategy
  $\theta_2\in \Theta_2$, hence this is clear. \qed
\end{proof}

Using the discrete trace distance, we recover the following standard
relations~\cite{inbook/hpa/Glabbeek01}.  The theorem follows by
Lemma~\ref{qltbt.le:trace_is_trace} and arguments similar to the ones used in
the proofs of the corresponding theorems in the preceding
section.  We refer
to~\cite{DBLP:conf/focs/RoundsB81,DBLP:journals/jacm/HennessyM85} for
definitions and discussion of possible-futures inclusion and
equivalence.

\begin{theorem}
  For $\trace d= \disc{ \trace d}$ the discrete trace
  distance and $s, t\in S$ we have
  \begin{itemize}
  \item $\disc{ \ntrace 1 d}( s, t)= 0$ iff there is a
    trace inclusion from $s$ to $t$,
  \item $\disc{ \ntraceeq 1 d}( s, t)= 0$ iff $s$ and $t$
    are trace equivalent,
  \item $\disc{ \ntrace 2 d}( s, t)= 0$ iff there is a
    possible-futures inclusion from $s$ to $t$,
  \item $\disc{ \ntraceeq 2 d}( s, t)= 0$ iff $s$ and $t$
    are possible-futures equivalent,
  \item $\disc{ \nrtrace 1 d}( s, t)= 0$ iff there is a
    readiness inclusion from $s$ to $t$,
  \item $\disc{ \nrtraceeq 1 d}( s, t)= 0$ iff $s$ and $t$
    are ready equivalent.
  \end{itemize}
\end{theorem}

The following theorem entails all relations in the left side of
Figure~\ref{qltbt.fi:spectrum}; the right-to-left arrows follow from
the strategy set inclusions $\blind \Theta_1'\subseteq \Theta_1'$ for
any $\Theta_1'\subseteq \Theta_1$ and
Lemma~\ref{qltbt.le:strat-restrict}.  As with
Theorems~\ref{qltbt.th:spectrum-1} and~\ref{qltbt.th:spectrum-2}, the
theorem follows by strategy set inclusion,
Lemma~\ref{qltbt.th:transfer_princ}, and corresponding results for the
discrete relations.

\begin{theorem}
  For all $k, \ell\in \Natp$ with $k< \ell$ and $s, t\in S$,
  \begin{align*}
    & \ntraceeq k d( s, t)\le \ntrace \ell d( s, t)\le
    \ntraceeq \ell d( s, t)\le \ntraceeq \infty d( s, t), \\
    & \ntrace k d( s, t)\le \nrtrace k d( s, t)\le \ntrace
    \ell d( s, t), \\
    & \ntraceeq k d( s, t)\le \nrtraceeq k d( s, t)\le \ntraceeq
    \ell d( s, t).
  \end{align*}
  Additionally, $\nrtrace k d$ and $\ntraceeq k d$ are incomparable,
  and also $\nrtraceeq k d$ and $\ntrace{( k+ 1)} d$ are incomparable.
  If the trace distance $\trace d$ is separating, then all distances
  above are topologically inequivalent.
\end{theorem}

\section{Recursive Characterizations}
\label{qltbt.se:recurse}

We now turn our attention to an important special case in which the
given trace distance has a specific recursive characterization; we show
that in this case, all distances in the spectrum can be characterized as
least fixed points.  We will see in Section~\ref{qltbt.se:examples_rec} that
this can be applied to all examples of trace distances mentioned in
Section~\ref{qltbt.se:example_distances}.

Note that all theorems require the LTS in question to be finitely
branching; this is a standard assumption which goes back
to~\cite{DBLP:conf/banff/Stirling95}.  In most cases it may be relaxed
to \emph{compact branching} in the sense
of~\cite{journals/anyas/Breugel96}, but to keep things simple, we do
not do this here.

\subsection{Fixed-Point Characterizations}
\label{qltbt.se:fixedp}

Let $L$ be a complete lattice with order $\sqsubseteq$ and bottom and
top elements $\bot$, $\top$.  Let
$f: \KK^\infty\times \KK^\infty\to L$,
$g: L\to \Realnn\cup\{ \infty\}$ and $F: \KK\times \KK\times L\to L$
such that $\trace d= g\circ f$, $g$ is monotone,
$F( x, y, \cdot): L\to L$ is monotone for all $x, y\in \KK$, and
\begin{equation}
  \label{qltbt.eq:trace_dist_rec}
  f( \sigma, \tau)=
  \begin{cases}
    F( \sigma_0, \tau_0, f( \sigma^1, \tau^1)) &\text{if } \sigma, \tau\ne
    \emptyseq, \\
    \top &\text{if } \sigma= \emptyseq, \tau\ne \emptyseq \text{ or }
    \sigma\ne \emptyseq, \tau= \emptyseq, \\
    \bot &\text{if } \sigma= \tau= \emptyseq
  \end{cases}
\end{equation}
for all $\sigma, \tau\in \KK^\infty$.

We hence assume that $\trace d$ has a recursive characterization
(using $F$) on top of a complete (but otherwise arbitrary) lattice $L$
which we introduce between $\KK^\infty$ and $\Realnn\cup\{ \infty\}$
to serve as a \emph{memory}.  Below we will work with different
endofunctions $I$ on the set of mappings
$( \Natp\cup\{ \infty\})\times\{ 1, 2\}\to L^{ S\times S}$ which are
parametrized by the number $m$ of switches in $\Natp\cup\{ \infty\}$
which Player~1 has left, and a value $p\in\{ 1, 2\}$ which keeps track
of whether Player~1 currently is building the left or the right path.

\begin{theorem}
  \label{qltbt.th:iter-1}
  The endofunction $I$ on $( \Natp\cup\{ \infty\})\times\{ 1, 2\}\to L^{
    S\times S}$ defined by
  \begin{align*}
    I( h_{ m, p})( s, t)=
    \begin{cases}
      \max
      \begin{cases}
        \adjustlimits \sup_{ s\tto{ x} s'} \inf_{ t\tto{ y} t'} F( x, y, h_{
        m, 1}( s', t')) \\
        \adjustlimits \sup_{ t\tto{ y} t'} \inf_{ s\tto{ x} s'} F( x, y, h_{
        m- 1, 2}( s', t'))
      \end{cases}
      &\text{if } m\ge 2, p= 1 \\
      \adjustlimits \sup_{ s\tto{ x} s'} \inf_{ t\tto{ y} t'} F( x, y, h_{
      m, 1}(s', t'))
      &\text{if } m= 1, p= 1 \\
      \max
      \begin{cases}
        \adjustlimits \sup_{ t\tto{ y} t'} \inf_{ s\tto{ x} s'} F( x, y, h_{
        m, 2}( s', t')) \\
        \adjustlimits \sup_{ s\tto{ x} s'} \inf_{ t\tto{ y} t'} F( x, y, h_{
        m- 1, 1}( s', t'))
      \end{cases}
      &\text{if } m\ge 2, p= 2 \\
      \adjustlimits \sup_{ t\tto{ y} t'} \inf_{st\tto{ x} s'} F( x, y, h_{
      m, 2}( s', t')) &\text{if } m= 1, p= 2
    \end{cases}
  \end{align*}
  has a least fixed point $h^*:( \Natp\cup\{ \infty\})\times\{ 1, 2\}\to
  L^{ S\times S}$, and if the LTS $( S, T)$ is finitely branching, then
  $\nsim k d= g\circ h^*_{ k, 1}$, $\nsimeq k d= g\circ \max( h^*_{ k,
    1}, h^*_{ k, 2})$ for all $k\in \Natp\cup\{ \infty\}$.
\end{theorem}

Hence $I$ iterates the function $h$ over the branching structure of $(
S, T)$, computing all nested branching distances at the same time.  Note
the specialization of this to simulation and bisimulation distance,
where we have the following fixed-point equations, using $h_{ 1, 1}^*=
\nsim 1 h$ and $h_{ \infty, 1}^*= \bisim h$:
\begin{align*}
  \nsim 1 h( s, t) &= \adjustlimits \sup_{ s\tto{ x} s'} \inf_{ t\tto{ y} t'} F( x,
  y, \nsim 1 h( s', t')) 
  \\
  \bisim h( s, t) &= \max
  \begin{cases}
    \adjustlimits \sup_{ s\tto{ x} s'} \inf_{ t\tto{ y} t'} F( x,
    y, \bisim h( s', t')) \\
    \adjustlimits \sup_{ t\tto{ y} t'} \inf_{ s\tto{ x} s'} F( x,
    y, \bisim h( s', t'))
  \end{cases}
\end{align*}

\begin{proof}
  The lattice of mappings $( \Natp\cup\{ \infty\})\times\{ 1, 2\}\to L^{
    S\times S}$ with the point-wise partial order is complete, and $I$
  is monotone because $F$ is, so by Tarski's fixed-point theorem, $I$
  has indeed a least fixed point $h^*$.  To show that $\nsim k d= g\circ
  h^*_{ k, 1}$ for all $k$, we pull back $\nsim k d$ along $g$: Define
  $w: ( \Natp\cup\{ \infty\})\times\{ 1, 2\}\to L^{ S\times S}$ by
  \begin{align*}
    w_{ k, 1}( s, t) &= \adjustlimits \sup_{ \theta_1\in \nsim k \Theta_1}
    \inf_{ \theta_2\in \Theta_2} f( \tr{ \outcome( \theta_1, \theta_2)( s,
      t)}) \\
    w_{ k, 2}( s, t) &= \adjustlimits \sup_{ \theta_1\in \nsim k \Theta_1}
    \inf_{ \theta_2\in \Theta_2} f( \tr{ \outcome( \theta_1, \theta_2)( t,
      s)})
  \end{align*}
  then $\nsim k d= g\circ f( k, 1)$ for all $k$ by monotonicity of $g$.
  We will be done once we can show that $w= h^*$.

  We first show that $w$ is a fixed point for $I$.  Let $s, t\in S$, then
  (assuming $k\ge 2$)
  \begin{align*}
    I( w_{ k, 1}&)( s, t) \\[-2ex]
    &= \max
    \begin{cases}
      \adjustlimits \sup_{ s\tto{ x} s'} \inf_{ t\tto{ y} t'} F( x, y, w_{
        k, 1}( s', t')) \\
      \adjustlimits \sup_{ t\tto{ y} t'} \inf_{ s\tto{ x} s'} F( x, y, w_{
        k- 1, 2}( s', t'))
    \end{cases} \\
    &= \max
    \begin{cases}
      \adjustlimits \sup_{ s\tto{ x} s'} \inf_{ t\tto{ y} t'} F( x, y,
      \adjustlimits \sup_{ \theta_1\in \nsim k \Theta_1} \inf_{ \theta_2\in
        \Theta_2} f( \tr{ \outcome( \theta_1, \theta_2)( s',
        t')})) \\
      \adjustlimits \sup_{ t\tto{ y} t'} \inf_{ s\tto{ x} s'} F( x, y,
      \adjustlimits \sup_{ \theta_1\in \nsim{( k- 1)} \Theta_1} \inf_{
        \theta_2\in \Theta_2} f( \tr{ \outcome( \theta_1, \theta_2)( t',
        s')}))
    \end{cases} \\
    &= \max
    \begin{cases}
      \adjustlimits \sup_{ s\tto{ x} s'} \inf_{ t\tto{ y} t'} \adjustlimits
      \sup_{ \theta_1\in \nsim k \Theta_1} \inf_{ \theta_2\in \Theta_2} F(
      x, y, f( \tr{ \outcome( \theta_1, \theta_2)( s',
        t')})) \\
      \adjustlimits \sup_{ t\tto{ y} t'} \inf_{ s\tto{ x} s'} \adjustlimits
      \sup_{ \theta_1\in \nsim{( k- 1)} \Theta_1} \inf_{ \theta_2\in
        \Theta_2} F( x, y, f( \tr{ \outcome( \theta_1, \theta_2)( t',
        s')}))
    \end{cases} \\
    &= \max
    \begin{cases}
      \begin{aligned}
        & \adjustlimits \sup_{ s\tto{ x} s'} \inf_{ t\tto{ y} t'}
        \adjustlimits \sup_{ \theta_1\in \nsim k \Theta_1} \inf_{
          \theta_2\in \Theta_2} \\
        &\qquad\qquad f( x\cdot \tr{ \outcome[ 1]( \theta_1, \theta_2)( s',
          t')}, y\cdot \tr{ \outcome[ 2]( \theta_1, \theta_2)( s', t')})
      \end{aligned}
      \\
      \begin{aligned}
        &\adjustlimits \sup_{ t\tto{ y} t'} \inf_{ s\tto{ x} s'}
        \adjustlimits \sup_{ \theta_1\in \nsim{( k- 1)} \Theta_1} \inf_{
          \theta_2\in \Theta_2} \\
        &\qquad\qquad f( x\cdot \tr{ \outcome[ 1]( \theta_1, \theta_2)( t',
          s')}, y\cdot \tr{ \outcome[ 2]( \theta_1, \theta_2)(
          t', s')})\,,
      \end{aligned}
    \end{cases}
  \end{align*}
  the next-to-last step by monotonicity of $F$.  By uniformity, the
  choices of $t\tto{ y} t'$ and $\theta_1\in \nsim k \Theta_1$ do not
  depend on each other, so the corresponding inf and sup can be
  exchanged, whence
  \begin{align*}
    I( w_{ k, 1})( s, t) &= \max
    \begin{cases}
      \begin{aligned}
        & \displaystyle \multiadjustlimits{ \sup_{ s\tto{ x} s'}
          \sup_{ \theta_1\in \nsim k
            \Theta_1} \inf_{ t\tto{ y} t'} \inf_{ \theta_2\in \Theta_2}} \\
        &\qquad f( x\cdot \tr{ \outcome[ 1]( \theta_1,
          \theta_2)( s', t')}, y\cdot \tr{ \outcome[ 2]( \theta_1,
          \theta_2)( s', t')})
      \end{aligned}
      \\
      \begin{aligned}
        &\displaystyle \multiadjustlimits{ \sup_{ t\tto{ y} t'} \sup_{
            \theta_1\in \nsim{( k-
              1)} \Theta_1} \inf_{ s\tto{ x} s'} \inf_{ \theta_2\in \Theta_2}} \\
        &\qquad f( x\cdot \tr{ \outcome[ 1]( \theta_1,
          \theta_2)( t', s')}, y\cdot \tr{ \outcome[ 2]( \theta_1,
          \theta_2)( t', s')})
      \end{aligned}
    \end{cases} \\
    &= \max
    \begin{cases}
      \adjustlimits \sup_{ \theta_1\in \nsim k \Theta_{ 1, \textup{ns}}}
      \inf_{ \theta_2\in \Theta_2} f( \tr{ \outcome( \theta_1, \theta_2)( s,
        t)}) \\
      \adjustlimits \sup_{ \theta_1\in \nsim k \Theta_{ 1, \textup{s}}}
      \inf_{ \theta_2\in \Theta_2} f( \tr{ \outcome( \theta_1, \theta_2)( s,
        t)})
    \end{cases} \\
    &= w_{ k, 1}( s, t).
  \end{align*}
  In the last max expression,
  $\nsim k \Theta_{ 1, \textup{ns}}\subseteq \nsim k \Theta_1$ is the
  subset of Player-1 strategies $\theta_1$ which do not switch from
  the configuration $( s, t, 0)$, \ie~for which
  $\src( \theta_{ 1, 1}( s, t, 0))= s$, and
  $\nsim k \Theta_{ 1, \textup{s}}= \nsim k \Theta_1\setminus \nsim k
  \Theta_{ 1, \textup{ns}}$ consists of the strategies which do switch
  from $( s, t, 0)$.  The other cases in the definition of
  $I$---$I( w_{ 1, 1})$, $I( w_{ 1, 2})$, and $I( w_{ k, 2})$ for
  $k\ge 2$---can be shown similarly, and we can conclude that
  $I( w_{ k, p})= w_{ k, p}$ for all $k\in \Natp\cup\{ \infty\}$,
  $p\in\{ 1, 2\}$.

  To show that $w$ is the least fixed point for $I$, let $\bar h:(
  \Natp\cup\{ \infty\})\times\{ 1, 2\}\to L^{ S\times S}$ be such that
  $I( \bar h)= \bar h$.  We prove that $w\le \bar h$, and again we show
  only the case $w_{ k, 1}\le \bar h_{ k, 1}$ for $k\ge 2$.  Note first
  that as the LTS $( S, T)$ is finitely branching, we can use the
  equation for $I( \bar h_{ k, 1})( s, t)$ to conclude that for all $s,
  t\in S$,
  \begin{align}
    & \text{for any $s\tto{ x} s'$ there is $t\tto{ y} t'$ such that $F( x,
      y, \bar h_{ k, 1}( s', t'))\le I( \bar h_{ k, 1})( s,
      t)$,} \label{qltbt.eq:strat-exist:1} \\
    & \text{for any $t\tto{ y} t'$ there is $s\tto{ x} s'$ such that $F( x,
      y, \bar h_{ k- 1, 2}( s', t'))\le I( \bar h_{ k, 1})( s,
      t)$.} \label{qltbt.eq:strat-exist:2}
  \end{align}

  Now let $\theta_1\in \nsim k \Theta_1$; the proof will be finished
  once we can find $\theta_2\in \Theta_2$ for which
  $f( \tr{ \outcome( \theta_1, \theta_2)( s, t)})\le \bar h_{ k, 1}(
  s, t)$.  Let $( \pi\cdot e, \rho, m)\in \Conf_2$ and write
  $s= \tgt( \last( \pi))$, $t= \tgt( \last( \rho))$.  Assume first
  that $e=( s, x, s')$, let $t= \tgt( \last( \rho))$ and
  $e=( t, y, t')$ an edge which satisfies the inequality
  of~\eqref{qltbt.eq:strat-exist:1}, and define
  $\theta_2( \pi\cdot e, \rho, m)=( e', m)$.  For the so-defined
  Player-2 strategy $\theta_2$ we have
  $f( \tr{ \outcome( \theta_1, \theta_2)( s, t)})\le \sup_{ s\tto{ x}
    s'} \inf_{ t\tto{ y} t'} F( x, y, \bar h_{ k, 1}( s', t'))\le I(
  \bar h_{ k, 1})( s, t)= \bar h_{ k, 1}( s, t)$ for all $s, t\in S$.
  The case $e=( t, y, t')$ is shown similarly,
  using~\eqref{qltbt.eq:strat-exist:2} instead.  \qed
\end{proof}

The fixed-point characterization for the ready simulation distances is
similar (and so is its proof, which we hence omit):

\begin{theorem}
  \label{qltbt.th:iter-2}
  The endofunction $I$ on $( \Natp\cup\{ \infty\})\times\{ 1, 2\}\to L^{
    S\times S}$ defined by
  \begin{align*}
    I( h_{ m, p})( s, t)=
    \begin{cases}
      \max
      \begin{cases}
        \adjustlimits \sup_{ s\tto{ x} s'} \inf_{ t\tto{ y} t'} F( x, y, h_{
        m, 1}( s', t')) \\
        \adjustlimits \sup_{ t\tto{ y} t'} \inf_{ s\tto{ x} s'} F( x, y, h_{
        m- 1, 2}( s', t'))
      \end{cases}
      &\text{if } m\ge 2, p= 1 \\
      \max
      \begin{cases}
        \adjustlimits \sup_{ s\tto{ x} s'} \inf_{ t\tto{ y} t'} F( x, y, h_{
        m, 1}(s', t')) \\
        \adjustlimits \sup_{ t\tto{ y} t'} \inf_{ s\tto{ x} s'} f( x, y)
      \end{cases}
      &\text{if } m= 1, p= 1 \\
      \max
      \begin{cases}
        \adjustlimits \sup_{ t\tto{ y} t'} \inf_{ s\tto{ x} s'} F( x, y, h_{
        m, 2}( s', t')) \\
        \adjustlimits \sup_{ s\tto{ x} s'} \inf_{ t\tto{ y} t'} F( x, y, h_{
        m- 1, 1}( s', t'))
      \end{cases}
      &\text{if } m\ge 2, p= 2 \\
      \max
      \begin{cases}
        \adjustlimits \sup_{ t\tto{ y} t'} \inf_{st\tto{ x} s'} F( x, y, h_{
        m, 2}( s', t')) \\
        \adjustlimits \sup_{ s\tto{ x} s'} \inf_{ t\tto{ y} t'} f( x, y)
      \end{cases}
      &\text{if } m= 1, p= 2
    \end{cases}
  \end{align*}
  has a least fixed point $h^*:( \Natp\cup\{ \infty\})\times\{ 1, 2\}\to
  L^{ S\times S}$, and if the LTS $( S, T)$ is finitely branching, then
  $\nrsim k d= g\circ h^*_{ k, 1}$, $\nrsimeq k d= g\circ \max( h^*_{ k,
    1}, h^*_{ k, 2})$ for all $k\in \Natp\cup\{ \infty\}$.
\end{theorem}

For the \emph{linear} distances, we extend $F$ to a function $\KK^n\times
\KK^n\times L\to L$, for $n\in \Nat$, by
\begin{equation*}
  F( \emptyseq, \emptyseq, \alpha)= \alpha, \qquad F( x\cdot \sigma, y\cdot
  \tau, \alpha)= F( x, y, F( \sigma, \tau, \alpha)).
\end{equation*}
We also extend the $\tto{ x}$ relation to finite traces so we can write
$s\tto{ \sigma} s'$ below, by letting $s\tto{ \emptyseq} s$ for all
$s\in S$ and $s\tto{ x\cdot \sigma} s'$ iff $s\tto{ x}
s''\tto{ \sigma} s'$ for some $s''\in S$.  We write $s\tto{ \sigma}$ if
there is a (finite or infinite) trace $\sigma$ from $s$.  The proofs of
the below theorems are similar to the one of Theorem~\ref{qltbt.th:iter-1}.

\begin{theorem}
  The endofunction $I$ on $( \Natp\cup\{ \infty\})\times\{ 1, 2\}\to L^{
    S\times S}$ defined by
  \begin{align*}
    I( h_{ m, p})( s, t)=
    \begin{cases}
      \max
      \begin{cases}
        \adjustlimits \sup_{ s\tto{ \sigma}\;} \inf_{ t\tto{ \tau}\;} f( \sigma,
        \tau) \\
        \adjustlimits \sup_{ s\tto{ \sigma} s'} \inf_{ t\tto{ \tau} t'} F(
        \sigma, \tau, h_{ m- 1, 1}( s', t')) \\
        \adjustlimits \sup_{ s\tto{ \sigma} s'} \inf_{ t\tto{ \tau} t'} F(
        \sigma, \tau, h_{ m- 1, 2}( s', t'))
      \end{cases}
      &\text{if } m\ge 2, p= 1 \\
      \adjustlimits \sup_{ s\tto{ \sigma}\;} \inf_{ t\tto{ \tau}\;} f( \sigma,
      \tau) &\text{if } m= 1, p= 1 \\
      \max
      \begin{cases}
        \adjustlimits \sup_{ t\tto{ \tau}\;} \inf_{s\tto{ \sigma}\;} f( \sigma,
        \tau) \\
        \adjustlimits \sup_{ t\tto{ \tau} t'} \inf_{ s\tto{ \sigma} s'} F(
        \sigma, \tau, h_{ m- 1, 2}( s', t')) \\
        \adjustlimits \sup_{ t\tto{ \tau} t'} \inf_{ s\tto{ \sigma} s'} F(
        \sigma, \tau, h_{ m- 1, 1}( s', t'))
      \end{cases}
      &\text{if } m\ge 2, p= 2 \\
      \adjustlimits \sup_{ t\tto{ \tau}\;} \inf_{s\tto{ \sigma}\;} f( \sigma,
      \tau) &\text{if } m= 1, p= 2
    \end{cases}
  \end{align*}
  has a least fixed point $h^*:( \Natp\cup\{ \infty\})\times\{ 1, 2\}\to
  L^{ S\times S}$, and if the LTS $( S, T)$ is finitely branching, then
  $\ntrace k d= g\circ h^*_{ k, 1}$, $\ntraceeq k d= g\circ \max(h^*_{
    k, 1}, h^*_{ k, 2})$ for all $k\in \Natp\cup\{ \infty\}$.
\end{theorem}

\begin{theorem}
  \label{qltbt.th:iter-4}
  The endofunction $I$ on $( \Natp\cup\{ \infty\})\times\{ 1, 2\}\to L^{
    S\times S}$ defined by
  \begin{align*}
    I( h_{ m, p})( s, t)=
    \begin{cases}
      \max
      \begin{cases}
        \adjustlimits \sup_{ s\tto{ \sigma}\;} \inf_{ t\tto{ \tau}\;} f(
        \sigma, \tau) \\
        \adjustlimits \sup_{ s\tto{ \sigma} s'} \inf_{ t\tto{ \tau} t'} F(
        \sigma, \tau, h_{ m- 1, 1}( s', t')) \\
        \adjustlimits \sup_{ s\tto{ \sigma} s'} \inf_{ t\tto{ \tau} t'} F(
        \sigma, \tau, h_{ m- 1, 2}( s', t'))
      \end{cases}
      &\text{if } m\ge 2, p= 1 \\
      \max
      \begin{cases}
        \adjustlimits \sup_{ s\tto{ \sigma}\;} \inf_{ t\tto{ \tau}\;} f(
        \sigma, \tau) \\
        \adjustlimits \sup_{ s\tto{ \sigma} s'} \inf_{ t\tto{ \tau} t'}
        \adjustlimits \sup_{ s'\tto{ x} s''} \inf_{ t'\tto{ y}  t''} f(
        \sigma\cdot x, \tau\cdot y) \\
        \adjustlimits \sup_{ s\tto{ \sigma} s'} \inf_{ t\tto{ \tau} t'}
        \adjustlimits \sup_{ t'\tto{ y} t''} \inf_{ s'\tto{ x}  s''} f(
        \sigma\cdot x, \tau\cdot y)
      \end{cases}
      &\text{if } m= 1, p= 1 \\
      \max
      \begin{cases}
        \adjustlimits \sup_{ t\tto{ \tau}\;} \inf_{ s\tto{ \sigma}\;} f(
        \sigma, \tau) \\
        \adjustlimits \sup_{ t\tto{ \tau} t'} \inf_{ s\tto{ \sigma} s'} F(
        \sigma, \tau, h_{ m- 1, 2}( s', t')) \\
        \adjustlimits \sup_{ t\tto{ \tau} t'} \inf_{ s\tto{ \sigma} s'} F(
        \sigma, \tau, h_{ m- 1, 1}( s', t'))
      \end{cases}
      &\text{if } m\ge 2, p= 2 \\
      \max
      \begin{cases}
        \adjustlimits \sup_{ t\tto{ \tau}\;} \inf_{s\tto{ \sigma}\;} f(
        \sigma, \tau) \\
        \adjustlimits \sup_{ t\tto{ \tau} t'} \inf_{ s\tto{ \sigma} s'}
        \adjustlimits \sup_{ t'\tto{ y} t''} \inf_{ s'\tto{ x}  s''} f(
        \sigma\cdot x, \tau\cdot y) \\
        \adjustlimits \sup_{ t\tto{ \tau} t'} \inf_{ s\tto{ \sigma} s'}
        \adjustlimits \sup_{ s'\tto{ x} s''} \inf_{ t'\tto{ y}  t''} f(
        \sigma\cdot x, \tau\cdot y)
      \end{cases}
      &\text{if } m= 1, p= 2
    \end{cases}
  \end{align*}
  has a least fixed point $h^*:( \Natp\cup\{ \infty\})\times\{ 1, 2\}\to
  L^{ S\times S}$, and if the LTS $( S, T)$ is finitely branching, then
  $\nrtrace k d= g\circ h^*_{ k, 1}$, $\nrtraceeq k d= g\circ \max(
  h^*_{ k, 1}, h^*_{ k, 2})$ for all $k\in \Natp\cup\{ \infty\}$.
\end{theorem}

The fixed-point characterizations above immediately lead to iterative
\emph{semi-algorithms} for computing the respective distances: to
compute for example simulation distance, we can initialize
$\nsim 1 h( s, t)= 0$ for all states $s, t\in S$ and then iteratively
apply the above equality.  This assumes the LTS $( S, T)$ to be
finitely branching and uses Kleene's fixed-point theorem and
continuity of $F$.  However, this computation is only guaranteed to
converge to simulation distance in finitely many steps in case the
lattice $L^{ S\times S}$ is \emph{finite}; otherwise, the procedure
might not terminate.

\subsection{Relation Families}

Below we show that both simulation and bisimulation distance admit a
relational characterization akin to the one of the standard Boolean
notions.  Using switching counters like we did in the previous section,
this can easily be generalized to give relational characterizations to
all distances in this chapter.

\begin{theorem}
  \label{qltbt.th:relation_fam}
  If the LTS $( S, T)$ is finitely branching, then
  $\nsim 1 d( s, t)\le \varepsilon$ iff there exists a relation family
  $\cal R=\{ R_\alpha\subseteq S\times S\mid \alpha\in L\}$ for which
  $( s, t)\in R_\beta\in R$ for some $\beta$ with
  $g( \beta)\le \varepsilon$, and such that for any $\alpha\in L$ and
  for all $( s', t')\in R_\alpha\in \cal R$,
  \begin{itemize}
  \item for all $s'\tto{ x} s''$, there exists $t'\tto{ y} t''$ such
    that $( s'', t'')\in R_{ \alpha'}\in \cal R$ for some
    $\alpha'\in L$ with $F( x, y, \alpha')\sqsubseteq \alpha$.
  \end{itemize}
  Similarly, $\bisim d( s, t)\le \varepsilon$ iff there exists a
  relation family $\cal R=\{ R_\alpha\subseteq S\times S\mid \alpha\in L\}$
  for which $( s, t)\in R_\beta\in \cal R$ for some $\beta$ with $g(
  \beta)\le \varepsilon$, and such that for any $\alpha\in L$ and for all
  $( s', t')\in R_\alpha\in \cal R$,
  \begin{itemize}
  \item for all $s'\tto{ x} s''$, there exists $t'\tto{ y} t''$ such
    that $( s'', t'')\in R_{ \alpha'}\in \cal R$ for some $\alpha'\in L$ with
    $F( x, y, \alpha')\sqsubseteq \alpha$;
  \item for all $t'\tto{ y} t''$, there exists $s'\tto{ x} s''$ such
    that $( s'', t'')\in R_{ \alpha'}\in \cal R$ for some
    $\alpha'\in L$ with $F( x, y, \alpha')\sqsubseteq \alpha$.
  \end{itemize}
\end{theorem}

\begin{proof}
  We only show the proof for simulation distance; for bisimulation
  distance it is analogous.  Assume first that
  $\nsim 1 d( s, t)\le \varepsilon$, then we have $h: S\times S\to L$
  for which $g( h( s, t))\le \varepsilon$ and
  \begin{equation*}
    h( s', t')= \adjustlimits \sup_{ s'\tto{ x} s''} \inf_{ t'\tto{ y}
      t''} F( x, y, h( s'', t''))
  \end{equation*}
  for all $s', t'\in S$.  Let $\beta= h( s, t)$, and define a relation
  family $\cal R=\{ R_\alpha\mid \alpha\in L\}$ by
  $R_\alpha=\{( s', t')\mid h( s', t')\sqsubseteq \alpha\}$.  Let
  $\alpha\in L$ and $( s', t')\in R_\alpha$, then
  $\sup_{ s'\tto{ x} s''} \inf_{ t'\tto{ y} t''} F( x, y, h( s'',
  t''))= h( s', t')\sqsubseteq \alpha$, and as $( S, T)$ is finitely
  branching, this implies that for all $s'\tto{ x} s''$ there is
  $t'\tto{ y} t''$ and $\alpha'= h( s'', t'')$ such that
  $( s'', t'')\in R_{ \alpha'}$ and
  $F( x, y, \alpha')\sqsubseteq \alpha$.

  For the other direction, assume a relation family as in the theorem
  and define $h: S\times S\to L$ by $h( s', t')= \inf\{ \alpha\mid(
  s', t'\in R_\alpha\}$.  Then $( s, t)\in R_\beta$ implies that $h(
  s, t)\sqsubseteq \beta$ and hence $g( h( s, t))\le \varepsilon$.  Let
  $s', t'\in S$, then $( s', t')\in R_{ h( s', t')}$, hence for all
  $s'\tto{ x} s''$ there is $t'\tto{ y} t''$ and $\alpha'\in L$ for
  which $F( x, y, \alpha')\sqsubseteq h( s', t')$ and $( s'', t'')\in
  R_{ \alpha'}$, implying $h( s'', t'')\sqsubseteq \alpha'$ and hence
  $F( x, y, h( s'', t''))\sqsubseteq h( s', t')$.  Collecting the
  pieces, we get $I( h)( s', t')= \sup_{ s'\tto{ x} s''} \inf_{
    t'\tto{ y} t''} F( x, y, h( s'', t''))\sqsubseteq h( s', t')$,
  hence $h$ is a pre-fixed point for $I$.  But then $h^*\sqsubseteq
  h$, hence $\nsim 1 d( s, t)= g( h^*( s, t))\le g( h( s, t))\le
  \varepsilon$. \qed
\end{proof}

\section{Recursive Characterizations for Example Distances}
\label{qltbt.se:examples_rec}

We show that the considerations in Section~\ref{qltbt.se:recurse} apply to all
the example distances we have introduced in
Section~\ref{qltbt.se:example_distances}.  We apply Theorem~\ref{qltbt.th:iter-1} to
derive fixed-point formulae for corresponding simulation distances, but
of course all other distances in the quantitative
linear-time--branching-time spectrum have similar characterizations.

Let $d$ be a hemimetric on $\KK$, then for all $\sigma, \tau\in \KK^\infty$
and $0< \lambda\le 1$,
\begin{align*}
  \PWDIS d( \sigma, \tau) &=
  \begin{cases}
    \max( d( \sigma_0, \tau_0), \lambda\, \PWDIS d( \sigma^1, \tau^1))
    &\text{if } \sigma, \tau\ne \emptyseq, \\
    \infty &\hspace*{-6.2em}\text{if } \sigma= \emptyseq, \tau\ne
    \emptyseq \text{ or }
    \sigma\ne \emptyseq, \tau= \emptyseq, \\
    0 &\hspace*{-6.2em}\text{if } \sigma= \tau= \emptyseq,
  \end{cases} \\
  \ACCDIS d( \sigma, \tau) &=
  \begin{cases}
    d( \sigma_0, \tau_0)+ \lambda\, \ACCDIS d( \sigma^1, \tau^1)
    &\text{if } \sigma, \tau\ne \emptyseq, \\
    \infty &\hspace*{-4em}\text{if } \sigma= \emptyseq, \tau\ne
    \emptyseq \text{ or }
    \sigma\ne \emptyseq, \tau= \emptyseq, \\
    0 &\hspace*{-4em}\text{if } \sigma= \tau= \emptyseq,
  \end{cases}
\end{align*}
hence we can apply the iteration theorems with lattice
$L= \Realnn\cup\{ \infty\}$, $g= \id$ the identity function, and the
recursion function $F$ given like the formulae above.  Using
Theorem~\ref{qltbt.th:iter-1} we can derive the following fixed-point
expressions for simulation distance:
\begin{align*}
  \nsim 1{ \PWDIS d}( s, t) &= \adjustlimits \sup_{ s\tto{ x} s'} \inf_{
    t\tto{ y} t'} \max( d( x, y), \lambda\, \nsim 1{ \PWDIS d}( s',
  t')) \\
  \nsim 1{ \ACCDIS d}( s, t) &= \adjustlimits \sup_{ s\tto{ x} s'}
  \inf_{ t\tto{ y} t'} ( d( x, y)+ \lambda\, \nsim 1{ \ACCDIS d}( s',
  t'))
\end{align*}
Incidentally, these are exactly the expressions introduced
in~\cite{DBLP:journals/tse/AlfaroFS09} and in previous chapters.

Also note that if $S$ is finite with $| S|= n$, then undiscounted
point-wise distance $\PWDIS[ 1] d$ can only take on the finitely many
values $\{ d( x, y)\mid( s, x, s'),( t, y, t')\in T\}$, hence the
fixed-point algorithm given by Kleene's theorem converges in at most
$n^2$ steps.  This algorithm is used
in~\cite{DBLP:journals/tse/AlfaroFS09, DBLP:conf/qest/DesharnaisLT08,
  DBLP:journals/tcs/LarsenFT11}.  For undiscounted accumulating distance
$\ACCDIS[ 1] d$, it can be shown~\cite{DBLP:journals/tcs/LarsenFT11} that with
$D= \max\{ d( x, y)\mid( s, x, s'),( t, y, t')\in T\}$, distance is
either infinite or bounded above by $2n^2D$, hence the $\ACCDIS[ 1] d$
algorithm either converges in at most $2n^2D$ steps or diverges.

For the limit-average distance $\ACCAVG d$, we let $L=( \Realnn\cup\{
\infty\})^\Nat$, $g( h)= \liminf_j h( j)$, and $f( \sigma, \tau)( j)=
\frac1{ j+ 1} \sum_{ i= 0}^j d( \sigma_i, \tau_i)$ the $j$'th average.
The intuition is that $L$ is used for ``remembering'' how long in the
traces we have progressed with the computation.  With $F$ given by $F(
x, y, h)( n)= \frac1{ n+ 1} d( x, y)+ \frac{ n}{ n+ 1} h( n-1)$ it can
be shown that~\eqref{qltbt.eq:trace_dist_rec} holds, giving the following
fixed-point expression for limit-average simulation distance (which to
the best of our knowledge is new):
\begin{equation*}
  \nsim 1 h_n( s, t)= \adjustlimits \sup_{ s\tto{ x} s'} \inf_{ t\tto{ y}
    t'} \big( \tfrac1{ n+ 1} d( x, y)+ \tfrac n{ n+ 1} \nsim 1 h_{ n-
    1}( s', t')\big)
\end{equation*}

For the maximum-lead distance, we let
$L=( \Realnn\cup\{ \infty\})^\Real$, the lattice of mappings from
leads to maximum leads.  Using the notation from
Section~\ref{qltbt.se:example_distances}, we let $g( h)= h( 0)$ and
$f( \sigma, \tau)( \delta)= \max(| \delta|, \sup_j| \delta+ \sum_{ i=
  0}^j \sigma^w_i- \sum_{ i= 0}^j \tau^w_j|)$ the maximum-lead
distance between $\sigma$ and $\tau$ assuming that $\sigma$ already
has a lead of $\delta$ over $\tau$.  With
$F( x, y, h)( \delta)= \max(| \delta|, h( \delta+ x- y))$ it can be
shown that~\eqref{qltbt.eq:trace_dist_rec} holds, and then the
fixed-point expression for maximum-lead simulation distance becomes
the one given in Chapter~\ref{ch:wtsjlap}:
\begin{equation*}
  \nsim 1 h( \delta)( s, t)= \adjustlimits \sup_{ s\tto{ x} s'} \inf_{
    t\tto{ y} t'} \max(| \delta|, \nsim 1 h( s', t')( \delta+ x- y))
\end{equation*}
Again it can be shown~\cite{DBLP:conf/formats/HenzingerMP05} that for $S$
finite with $| S|= n$ and $D= \max\{ d( x, y)\mid( s, x, s'),( t, y,
t')\in T\}$, the iterative algorithm for computing maximum-lead distance
either converges in at most $2n^2D$ steps or diverges.

Regarding Cantor distance, a useful recursive formulation is
\begin{equation*}
  f( \sigma, \tau)( n)=
  \begin{cases}
    f( \sigma^1, \tau^1)( n+ 1) &\text{if } \sigma_0= \tau_0, \\
    n &\text{otherwise},
  \end{cases}
\end{equation*}
which iteratively counts the number of matching symbols in $\sigma$ and
$\tau$.  Here we use $L=( \Realnn\cup\{ \infty\})^\Nat$ and $g( h)=
\frac1{ h( 0)}$; note that the order on $L$ has to be reversed for $g$
to be monotone.  The fixed-point expression for Cantor simulation
distance becomes
\begin{equation*}
  \nsim 1 h_n( s, t)= \max( n, \adjustlimits \sup_{ s\tto{ x} s'} \inf_{
    t\tto{ x} t'} \nsim 1 h_{ n+ 1}( s', t'))
\end{equation*}
but as the order on $L$ is reversed, the sup now means that Player~1 is
trying to \emph{minimize} this expression, and Player~2 tries to
maximize it.  Hence Player~2 tries to find maximal matching
\emph{subtrees}; the corresponding Cantor simulation equivalence
distance between $s$ and $t$ hence is the inverse of the maximum depth
of matching subtrees under $s$ and $t$.  The Cantor bisimulation
distance in turn is the same as the inverse of \emph{bisimulation
  depth}~\cite{DBLP:journals/jacm/HennessyM85}.

\chapter[Weighted Modal Transition Systems][Weighted Modal Transition
Systems]{Weighted Modal Transition Systems\footnote{This chapter is
    based on the journal paper~\cite{DBLP:journals/fmsd/BauerFJLLT13}
    published in Formal Methods in System Design.}}
\label{ch:weightedmodal}

In this chapter we lift the accumulating distance to \emph{modal
  specifications}, a specification formalism which permits incremental
and compositional design.  To this end, we replace the
\emph{refinement} relation of standard modal specifications by a
refinement \emph{distance}.  We then show that our quantitative
generalization does \emph{not} admit any notions of determinization or
conjunction, but that structural composition and quotient do satisfy
the expected quantitative properties.

\section{Weighted Modal Transition Systems}
\label{weightedmodal.se:wmts}

In this section we present the formalism we use for implementations
and specifications.  As implementations we choose the model of
\emph{weighted transition systems}, \ie~labeled transition systems
with integer weights at transitions.  Specifications both have a
\emph{modal} dimension, specifying discrete behavior which \emph{must}
be implemented and behavior which \emph{may} be present in
implementations, and a \emph{quantitative} dimension, specifying
intervals of weights on each transition within are permissible for an
implementation.

Let $\II=\big\{[ x, y]\bigmid x\in \Int\cup\{ -\infty\}, y\in \Int\cup\{
\infty\}, x\le y\big\}$ be the set of closed extended-integer intervals
and let $\Sigma$ be a finite set of actions.  Our set of
\emph{specification labels} is $\Spec= \Sigma\times \II$, pairs of actions
and intervals.  The set of \emph{implementation labels} is defined as
$\Impl \Spec= \Sigma\times\big\{[ x, x]\bigmid x\in \Int\big\}\approx
\Sigma\times \Int$. Hence a specification imposes labels and integer
intervals which constrain the possible weights of an implementation.

We define a partial order on $\II$ (representing inclusion of intervals)
by $[ x, y] \labpre [ x', y']$ if $x'\le x$ and $y\le y'$, and we
extend this order to specification labels by $( a, I)\labpre( a',
I')$ if $a= a'$ and $I\labpre I'$.  The partial order on $\Spec$ is
hence a \emph{refinement} order; if $k_1\labpre k_2$ for $k_1,
k_2\in \Spec$, then no more implementation labels are contained in $k_1$
than in $k_2$.

Specifications and implementations are defined as follows:

\begin{definition}
  A \emph{weighted modal transition system} (WMTS) is a quadruple
  $( S, s^0, \mmayto, \mmustto)$ consisting of a set of states $S$
  with an initial state $s^0\in S$ and \textit{must} ($\mmustto$) and
  \textit{may} ($\mmayto$) transition relations
  $\mmustto, \mmayto\subseteq S\times \Spec\times S$ such that for every
  $(s,k,s') \in \mmustto$ there is $(s,\ell,s') \in \mmayto$ where
  $k \labpre \ell$.  A WMTS is an \emph{implementation} if
  $\mmustto= \mmayto\subseteq S\times \Impl \Spec\times S$.
\end{definition}

Note the natural requirement that any required (\textit{must}) behavior
is also allowed (\textit{may}) above, and that implementations
correspond to standard integer-weighted transition systems, where all
optional behavior and positioning in the intervals has been decided on.

A WMTS is \emph{finite} if $S$ and $\mmayto$ (and hence also $\mmustto$)
are finite sets, and it is \emph{deterministic} if it holds that for all
$s\in S$, $a\in \Sigma$, $\big( s,( a, I_1), t_1\big),\big( s,( a,
I_2), t_2\big)\in \mmayto$ imply $I_1= I_2$ and $t_1= t_2$.  Hence a
deterministic specification allows at most one transition under each
discrete action from every state.  In the rest of the paper we will
write $s\mayto{ k} s'$ for $( s, k, s')\in \mmayto$ and similarly for
$\mmustto$, and we will always write $S=( S, s^0, \mmayto, \mmustto)$ or
$S_i=( S_i, s^0_i, \mmayto_i, \mmustto_i)$ for WMTS and $I=( I, i^0,
\mmustto)$ for implementations.  Note that an implementation is just a
usual integer-weighted transition system.

Our theory will work with infinite WMTS, though we will require them
to be \emph{compactly branching}.  This is a natural generalization of
the standard requirement on systems to be \emph{finitely branching}
which was first used in~\cite{journals/anyas/Breugel96}; see
Def.~\ref{weightedmodal.de:comp} below.

The implementation semantics of a specification is given through modal
refinement, as follows:

\begin{definition}
  A \emph{modal refinement} of WMTS $S_1$, $S_2$
  is a relation $R\subseteq S_1\times S_2$ such that for any $(s_1,
  s_2)\in R$
  \begin{itemize}
  \item whenever $s_1\mayto{ k_1}_1 t_1$ for some $k_1\in \Spec$, $t_1\in
    S_1$, then there exists $s_2\mayto{ k_2}_2 t_2$ for some $k_2\in
    \Spec$, $t_2\in S_2$, such that $k_1 \labpre k_2$ and $( t_1,
    t_2)\in R$,
  \item whenever $s_2\mustto{ k_2}_2 t_2$ for some $k_2\in \Spec$,
    $t_2\in S_2$, then there exists $s_1\mustto{ k_1}_1 t_1$ for some
    $k_1\in \Spec$, $t_1\in S_1$, such that $k_1\labpre k_2$ and $(
    t_1, t_2)\in R$.
  \end{itemize}
  We write $S_1\mr S_2$ if there is a modal refinement relation $R$
  for which $( s^0_1, s^0_2)\in R$.
\end{definition}

Hence in such a modal refinement, behavior which is required in $S_2$ is
also required in $S_1$, no more behavior is allowed in $S_1$ than in
$S_2$, and the quantitative requirements in $S_1$ are refinements of the
ones in $S_2$.  The implementation semantics of a specification can then
be defined as the set of all implementations which are also refinements:

\begin{definition}
  The \emph{implementation semantics} of a WMTS $S$ is the set
  $\llbracket S\rrbracket=\{ I \mid I \mr S ~\text{and}~ I ~\text{is
    an implementation}\}$.
\end{definition}

This conforms with the intuition developed above: if
$I\in\llbracket S\rrbracket$, then any (reachable) behavior
$i\mustto{ a, x} j$ in $I$ must be allowed by a matching transition
$s\mayto{ a,[ l, r]} t$ in $S$ with $l\le x\le r$; correspondingly,
any (reachable) required behavior $s\mustto{ a,[ l, r]} t$ in $S$ must
be implemented by a matching transition $i\mustto{ a, x} j$ in $I$
with $l\le x\le r$.

\section{Thorough and Modal Refinement Distances}
\label{weightedmodal.se:distances}

For the quantitative specification formalism we have introduced in the
last section, the standard Boolean notions of satisfaction and
refinement are too fragile.  To be able to reason not only whether a
given quantitative implementation satisfies a given quantitative
specification, but also to what extent, we introduce a notion of
\emph{distance} between both implementations and specifications.

We first define the distance between \emph{implementations}; for this we
introduce a distance on implementation labels by
\begin{equation}
  \label{weightedmodal.eq:d_imp}
  d_{ \Impl \Spec}\big(( a_1, x_1),( a_2, x_2)\big)=\left\{
  \begin{array}{cl}
    \infty &\quad\text{if } a_1\ne a_2, \\
    | x_1- x_2| &\quad\text{if } a_1= a_2.
  \end{array}
  \right.
\end{equation}
In the rest of the chapter, let $\lambda\in \Real$ with
$0 < \lambda< 1$ be a \emph{discounting factor}.

\begin{definition}
  \label{weightedmodal.de:acc.dist}
  The \emph{implementation distance} $d: I_1\times I_2\to
  \Realnn\cup\{ \infty\}$ between the states of implementations $I_1$
  and $I_2$ is the least fixed point of the equations
  \begin{equation*}
    d( i_1, i_2)= \max\left\{
      \begin{aligned}
        &\adjustlimits \sup_{ i_1\tto{ k_1}_1 j_1} \inf_{ i_2\tto{
          k_2}_2 j_2} d_{ \Impl \Spec}( k_1, k_2)+ \lambda d( j_1, j_2), \\
        &\adjustlimits \sup_{ i_2\tto{ k_2}_2 j_2} \inf_{ i_1\tto{
          k_1}_1 j_1} d_{ \Impl \Spec}( k_1, k_2)+ \lambda d( j_1, j_2).
      \end{aligned}
    \right.
  \end{equation*}
  We define $d(I_1,I_2) = d(i_1^0, i_2^0)$.
\end{definition}

\begin{lemma}
  \label{weightedmodal.le:impdistmet}
  The implementation distance is well-defined, and is a pseudometric.
\end{lemma}

\begin{proof}
  This is precisely the accumulating bisimulation distance from
  Chapter~\ref{ch:qltbt}, so the statement follows from
  Proposition~\ref{qltbt.pr:hemi}.  See also the proof of
  Lemma~\ref{simdistax.lem:acc:contraction}.
\end{proof}

We remark that besides this accumulating distance, other interesting
system distances may be defined depending on the application at hand,
\cf~Chapter~\ref{ch:qltbt}.  We concentrate here on this distance and
leave a generalization to other distances for the next chapter.

\begin{figure}[tpb]
  \begin{center}
    \begin{tikzpicture}[->,>=stealth',shorten >=1pt,auto,node
      distance=2.0cm,initial text=,scale=0.9,transform shape]
      \tikzstyle{every node}=[font=\small] \tikzstyle{every
        state}=[fill=white,shape=circle,inner sep=.5mm,minimum size=6mm]
      \path[use as bounding box] (0,-0.5) rectangle (12,1.3); 
      \node[state,initial] (s) at (0,1) {$i_1$};
      \node[state] (s1) at (2,1) {$j_1$};
      \node[state] (s2) at (1,0) {$k_1$};
      \path (s) edge [solid] node [above,sloped] {$3$} (s1);
      \path (s) edge [solid] node [above,sloped] {$7$} (s2);
      \path (s2) edge [solid] node [above,sloped] {$6$} (s1);
      \path (s2) edge [solid,loop below] node [below,sloped] {$9$} (s2);
      
      \node[state,initial] (t) at (3,0) {$i_2$};
      \node[state] (t1) at (5,0) {$j_2$};
      \path (t) edge [solid] node [above,sloped] {$6$} (t1);
      \path (t) edge [solid,loop below] node [below,sloped] {$7$} (t);
      
      \node (text) at (9.4,0) 
      {\begin{minipage}{6.8cm}
          $d(j_1,j_2) = 0$ \\
          $d(i_1,j_2) = \infty$ \\
          $d(j_1,i_2) = \infty$  \\
          $d(k_1,j_2) = \infty$  \\[-4mm]
          $d(k_1,i_2) = \max\{2+.9\, d(k_1,i_2),\ 
          .9\overbrace{d(j_1,j_2)}^0 \}$ \\
          $d(i_1,i_2) = \max\{3+.9\underbrace{d(j_1,j_2)}_0,\ 
          .9\, d(k_1,i_2) \}$ \\
        \end{minipage}
      };
    \end{tikzpicture}
  \end{center}
  \caption{Two weighted transition systems with branching distance
    $d( I_1, I_2)= 18$.}
  \label{weightedmodal.fig:distance}
\end{figure}
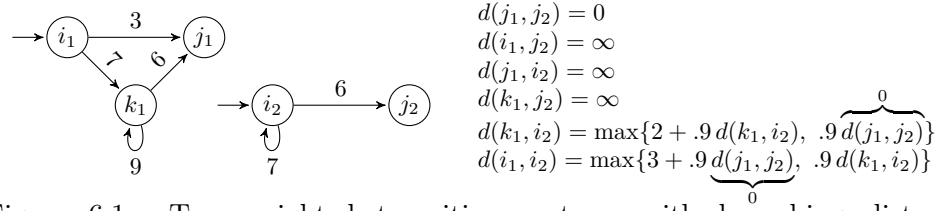

\begin{example}
  Consider the two implementations $I_1$ and $I_2$ in
  Figure~\ref{weightedmodal.fig:distance} with a single action (elided for simplicity)
  and with discounting factor $\lambda = .9$.  The equations in the
  illustration have already been simplified by removing all expressions
  that evaluate to $\infty$.  What remains to be done is to compute the
  least fixed point of the equation $d(k_1,i_2) = \max\big\{ 2 +
  .9\, d(k_1,i_2), 0\big\}$.  Clearly $0$ is not a fixed point, and
  solving the equation $d( k_1, i_2)= 2+ .9\,d( k_1, i_2)$
  gives $d(k_1,i_2) = 20$. Hence $d(i_1,i_2) = \max\{3,
  .9\cdot 20\} = 18$.
\end{example}

Note that the interpretation of the distance between two
implementations depends entirely on the application one has in mind;
but it can easily be shown that the distance between two
implementations is zero iff they are \emph{weighted bisimilar}.  The
intuition is then that the smaller the distance, the closer the
implementations are to being bisimilar.

To lift the implementation distance to specifications, we need first to
consider the distance between \emph{sets} of implementations.  Given
implementation sets $\mathcal I_1, \mathcal I_2\subseteq \Imp$, we define
\begin{equation*}
  d( \mathcal I_1, \mathcal I_2)= \adjustlimits \sup_{ I_1\in
    \mathcal I_1} \inf_{ I_2\in \mathcal I_2} d( I_1, I_2)
\end{equation*}
Note that in case $\mathcal I_2$ is finite, we have that for all
$\epsilon\ge 0$, $d( \mathcal I_1, \mathcal I_2)\le \epsilon$ if
and only if for each implementation $I_1\in \mathcal I_1$ there exists
$I_2\in \mathcal I_2$ for which $d( I_1, I_2)\le \epsilon$, hence
this is quite a natural notion of distance.  Especially, $d(
\mathcal I_1, \mathcal I_2)= 0$ if $\mathcal I_1$ is a subset of
$\mathcal I_2$ up to bisimilarity.  For infinite $\mathcal I_2$, we have
the slightly more complicated property that $d( \mathcal I_1,
\mathcal I_2)\le \epsilon$ iff for all $\delta> 0$ and any
$I_1\in \mathcal I_1$, there is $I_2\in \mathcal I_2$ for which $d(
I_1, I_2)\le \epsilon+ \delta$.

We lift this distance to specifications as follows:

\begin{definition}
  The \emph{thorough refinement distance} between WMTS $S_1$ and $S_2$
  is defined as $\thd( S_1, S_2)= d\big(\llbracket S_1\rrbracket,
  \llbracket S_2\rrbracket\big)$.  We write $S_1\thr^\epsilon S_2$ if
  $\thd( S_1, S_2)\le \epsilon$.
\end{definition}

\begin{lemma}
  The thorough refinement distance is a hemimetric.
\end{lemma}

\begin{proof}
  To show that $\thd( S, S)= 0$ is trivial, and the triangle
  inequality $\thd( S_1, S_2)+ \thd( S_2, S_3)\ge \thd( S_1, S_3)$
  follows like in the proof of~\cite[Lemma~3.72]{book/AliprantisB07}.
\end{proof}

Indeed this permits us to measure incompatibility of specifications;
intuitively, if two specifications have thorough distance $\epsilon$,
then any implementation of the first specification can be matched by an
implementation of the second up to $\epsilon$.  Also observe the special
case where $S_1= I_1$ is an implementation: then $\thd( I_1, S_2)=
\inf_{ I_2\in \llbracket S_2\rrbracket} d( I_1, I_2)$, which
measures how close $I_1$ is to satisfy the specification $S_2$.

We now proceed to introduce \emph{modal} refinement distance as an
overapproximation of thorough refinement distance.

First we generalize the distance on implementation labels from
Equation~\eqref{weightedmodal.eq:d_imp} to specification labels, again using a
Hausdorff-type construction.  For $k, \ell \in \Spec$ we define
\begin{equation*}
  d_\Spec( k, \ell)= \adjustlimits \sup_{ k'\labpre k, k'\in \Impl{}\,}
  \inf_{\, \ell'\labpre \ell, \ell'\in \Impl{}} d_{ \Impl \Spec}( k', \ell').
\end{equation*}
Note that $d_\Spec$ is asymmetric, and that $d_\Spec( k, \ell)= 0$ if
and only if $k\labpre \ell$.  Also, $d_\Spec( k, \ell)= d_{ \Impl
  \Spec}( k, \ell)$ for all $k, \ell\in \Impl \Spec$.  In more elementary
terms, we can express $d_\Spec$ as follows:
\begin{align*}
  d_\Spec\big(( a_1, I_1),( a_2, I_2)\big) &= \infty \quad \text{if }
  a_1\ne a_2 \\
  d_\Spec\big(( a,[ x_1, y_1]),( a,[ x_2, y_2])\big) &= \max(
  x_2 - x_1, y_1 - y_2, 0)
\end{align*}

\begin{definition}
  \label{weightedmodal.de:acc.mo.dist}
  Let $S_1$, $S_2$ be WMTS.  The \emph{modal refinement distance}
  $\md: S_1\times S_2\to \Realnn\cup\{ \infty\}$ from states of
  $S_1$ to states of $S_2$ is the least fixed point of the equations
  \begin{equation*}
    \md( s_1, s_2)= \max\left\{
      \begin{aligned}
        &\adjustlimits \sup_{ s_1\mayto{ k_1}_1 t_1} \inf_{ s_2\mayto{
          k_2}_2 t_2} d_\Spec( k_1, k_2)+ \lambda \md( t_1, t_2)\,, \\
        &\adjustlimits \sup_{ s_2\mustto{ k_2}_2 t_2} \inf_{
          s_1\mustto{ k_1}_1 t_1} d_\Spec( k_1, k_2)+ \lambda \md(
        t_1, t_2)\,.
      \end{aligned}
    \right.
  \end{equation*}
  We define $\md(S_1,S_2) = \md(s_1^0, s_2^0)$, and we write
  $S_1\mr^\epsilon S_2$ if $\md( S_1, S_2)\le \epsilon$.
\end{definition}

\begin{lemma}
  The modal refinement distance is well-defined, and is a hemimetric.
\end{lemma}

\begin{proof}
  Like in the proof of Lemma~\ref{weightedmodal.le:impdistmet}, the argument for
  existence of a unique least fixed point to the defining equations is
  that they define a contraction.  The triangle inequality can again be
  shown inductively, and the property $\md( s, s)= 0$ is clear. \qed
\end{proof}

We can now give a precise definition of compact branching:

\begin{definition}
  \label{weightedmodal.de:comp}
  A WMTS $S$ is said to be \emph{compactly branching} if the sets $\{(
  s', k)\mid s\mayto{ k} s'\}, \{( s', k)\mid s\mustto{ k} s'\}\subseteq
  S\times \Spec$ are compact under the symmetrized product distance $\bar
  d_\textup{\textsf{m}}\times \bar d_\Spec$ for every $s\in S$.
\end{definition}

The notion of compact branching was first introduced, for a formalism
of \emph{metric transition systems},
in~\cite{journals/anyas/Breugel96}.  It is a natural generalization of
the standard requirement on transition systems to be \emph{finitely
  branching} to a distance setting; we will need it for the property
that continuous functions defined on the sets $\{( s', k)\mid s\mayto{
  k} s'\}, \{( s', k)\mid s\mustto{ k} s'\}\subseteq S\times \Spec$, for
some $s\in S$, attain their infimum and supremum, see
Lemma~\ref{weightedmodal.le:family} and its proof below.

Thus, we shall henceforth assume all our WMTS to be compactly
branching.  The following lemma sets up some sufficient conditions for
this to be the case.

\begin{lemma}
  Let $S$ be a WMTS and define the sets $L_i( s, a)$, $U_i( s, a)$ for
  all $s\in S$, $a\in \Sigma$ and $i\in\{ 1, 2\}$ by
  \begin{align*}
    L_1( s, a) &= \{ l\mid s\mayto{ a,[ l, r]} s'\}, &
    L_2( s, a) &= \{ l\mid s\mustto{ a,[ l, r]} s'\}, \\
    U_1( s, a) &= \{ r\mid s\mayto{ a,[ l, r]} s'\}, 
    & U_2( s, a) &= \{ r\mid s\mustto{ a,[ l, r]} s'\}.
  \end{align*}
  Then $S$ is compactly branching if
  \begin{itemize}
  \item for all $s\in S$, any Cauchy sequence $( s'_n)_{ n\in \Nat}$
    in $\{ s'\mid s\mayto{} s'\}$ (with pseudometric
    $\bar d_\textup{\textsf{m}}$) has
    $\lim_{ n\to \infty} s_n\in\{ s'\mid s\mayto{} s'\}$, and
    likewise, any Cauchy sequence $( s'_n)_{ n\in \Nat}$ in
    $\{ s'\mid s\mustto{} s'\}$ has
    $\lim_{ n\to \infty} s_n\in\{ s'\mid s\mustto{} s'\}$, and
  \item for all $s\in S$, $a\in \Sigma$ and $i\in\{ 1, 2\}$, $L_i( s, a)$ is
    finite or $-\infty\in L_i( s, a)$, and $U_i( s, a)$ is finite or
    $\infty\in U_i( s, a)$.
  \end{itemize}
\end{lemma}

Note that the first property mimicks (and generalizes) standard
properties of finite branching and \emph{saturation},
\cf~\cite[Sect.~3.3]{DBLP:journals/toplas/Sangiorgi09}.  The intuition
is that if $s$ has (either \textit{may} or \textit{must}) transitions to
a converging sequence of states, then it also has a transition to the
limit.

\begin{proof}
  The first condition implies that the sets
  $\{ s'\in S\mid s\mayto{} s'\}$ and $\{ s'\in S\mid s\mustto{} s'\}$
  are compact in the pseudometric $\bar d_\textup{\textsf{m}}$ for all
  $s\in S$.  By Tychonoff's theorem, products of compact sets are
  compact, so we need only show that the second condition implies that
  the sets $\{ k\in \Spec\mid s\mayto{ k} s'\}$ and
  $\{ k\in \Spec\mid s\mustto{ k} s'\}$ are compact in the
  pseudometric $\bar d_\Spec$ for every $s\in S$.

  Let $s\in S$.  By definition of $d_\Spec$, the sets $\{ k\mid s\mayto{
  k} s'\}$, $\{ k\mid s\mustto{ k} s'\}$ fall into connected components
  $\{ I\mid s\mayto{ a, I} s'\}$, $\{ I\mid s\mustto{ a, I} s'\}$ for
  all $a\in \Sigma$, hence the former are compact iff all the
  latter are.  These in turn are compact iff the four sets
  $L_i$, $U_i$ in the lemma, collecting lower and upper bounds of
  intervals, are compact.  Now interval bounds are extended integers, so
  a sequence in $L_i$ or $U_i$ converges iff it is eventually
  stable or goes towards $-\infty$ or $\infty$.  If the sets are finite,
  eventual stability is the only option; if they are infinite, they need
  to include the limit points $-\infty$ (for the lower interval bounds
  in $L_i$) or $\infty$ (for the upper interval bounds in $U_i$). \qed
\end{proof}

We extend the notion of relation families from revious chapters to
modal refinement distance.  We define a \emph{modal refinement family}
as an $\Realnn$-indexed family of relations
$\cal R=\{ R_\epsilon\subseteq S_1\times S_2\mid \epsilon\ge 0\}$ such
that for any $\epsilon$ and any $(s_1, s_2)\in R_\epsilon\in \cal R$,
\begin{itemize}
\item whenever $s_1\mayto{ k_1} t_1$ for some $k_1 \in \Spec$,
  $t_1 \in S_1$, then there exists $s_2\mayto{ k_2} t_2$ for some
  $k_2 \in \Spec$, $t_2 \in S_2$, such that
  $d_\Spec( k_1, k_2)\le \epsilon$ and
  $( t_1, t_2)\in R_{ \epsilon'}\in \cal R$ for some
  $\epsilon'\le \lambda^{ -1}\big( \epsilon- d_\Spec( k_1, k_2)\big)$,
\item whenever $s_2\mustto{ k_2} t_2$ for some $k_2 \in \Spec$,
  $t_2 \in S_2$, then there exists $s_1\mustto{ k_1} t_1$ for some
  $k_1 \in \Spec$, $t_1 \in S_1$, such that
  $d_\Spec( k_1, k_2)\le \epsilon$ and
  $( t_1, t_2)\in R_{ \epsilon'}\in \cal R$ for some
  $\epsilon'\le \lambda^{ -1}\big( \epsilon- d_\Spec( k_1, k_2)\big)$.
\end{itemize}
Note that modal refinement families are
\begin{itemize}
\item \emph{upward closed} in the sense that $( s_1, s_2)\in R_\epsilon$
  implies that $( s_1, s_2)\in R_{ \epsilon'}$ for all $\epsilon'\ge
  \epsilon$, and
\item \emph{downward compact} in the sense that for any set $E\subseteq
  \Realnn$, if $(s_1,s_2) \in R_\epsilon$ for all $\epsilon\in E$, then
  also $(s_1,s_2) \in R_{\inf E}$.  This property follows from the
  assumption that our WMTS are compactly branching.
\end{itemize}

Following the proof strategy developed in previous chapters for
implementations, we can show the following characterization of modal
refinement distance by modal refinement families:

\begin{lemma}
  \label{weightedmodal.le:family}
  $S_1\mr^\epsilon S_2$ iff there is a modal refinement family $\cal R$
  with $( s_1^0, s_2^0)\in R_\epsilon \in \cal R$.
\end{lemma}

\begin{proof}
  First, assume that $S_1\mr^\epsilon S_2$,
  \ie~$\md(s^0_1,s^0_2) \le \epsilon$, and define a relation family
  $\cal R = \{ R_\delta \mid \delta \ge 0 \}$ by
  $R_\delta = \{ (s_1,s_2) \in S_1 \times S_2 \mid \md(s_1,s_2) \le
  \delta \}$ for all $\delta \ge 0$, then
  $(s^0_1,s^0_2) \in R_\epsilon$ holds by assumption.  We show that
  $\cal R$ is a modal refinement family.  Let $(s_1,s_2) \in R_\delta$
  for some $\delta \ge 0$, then by definition we know that
  $\md(s_1,s_2) \le \delta$. Assume $s_1 \mayto{k_1}_1 t_1$. From
  $\md(s_1,s_2) \le \delta$ we can infer that
  \begin{equation*}
    \inf_{ s_2\mayto{ k_2}_2 t_2} d_\Spec( k_1, k_2)+ \lambda
    \md( t_1, t_2) \le \delta\,.
  \end{equation*}
  Hence, because $S_2$ is compactly branching, there exists a
  may-transition $s_2 \mayto{k_2} t_2$ such that
  $d_\Spec( k_1, k_2) \le \delta$ and
  $\md(t_1,t_2) \le \lambda^{-1}(\delta - d_\Spec(k_1,k_2))$. The latter
  implies that $(t_1,t_2) \in R_{\delta'}$ for some
  $\delta' \le \lambda^{-1}(\delta - d_\Spec(k_1,k_2))$ which was to be
  shown. The argument for the other assertion for must-transitions is
  symmetric. This proves that there is a modal refinement family
  $\cal R$ such that $(s^0_1,s^0_2) \in R_\epsilon \in \cal R$.

  For the reverse direction, assume that $( s_1^0, s_2^0)\in R_\epsilon
  \in \cal R$ for some modal refinement family $\cal R = \{R_\epsilon \mid
  \epsilon \ge 0 \}$. We prove that $(s_1,s_2) \in R_\delta$, for some
  $\delta \ge 0$, implies $\md(s_1,s_2) \le \delta$. The claim $S_1
  \mr^\epsilon S_2$ then follows from the assumption $(s^0_1,s^0_2)
  \in R_\epsilon$.

  To this end, observe that the space of functions $\Delta = [S_1 \times
  S_2 \to \Realnn \cup \{\infty\}]$ forms a complete lattice, when the
  partial order $\le_\Delta$ is defined such that for $f, f' \in \Delta$, $f
  \le_\Delta f'$ iff $f(s_1,s_2) \le f'(s_1,s_2)$ for all $s_1\in S_1$, $s_2\in
  S_2$. Moreover, since $\max, \sup, \inf$ and $+$ are monotone, the
  function $D$ defined for all $f \in \Delta$ by
  \begin{equation*}
    D(f) = \max
      \begin{cases}
        &\adjustlimits \sup_{ s_1\mayto{ k_1}_1 t_1} \inf_{ s_2\mayto{
          k_2}_2 t_2} d_\Spec( k_1, k_2)+ \lambda f( t_1, t_2)\,, \\
        &\adjustlimits \sup_{ s_2\mustto{ k_2}_2 t_2} \inf_{
          s_1\mustto{ k_1}_1 t_1} d_\Spec( k_1, k_2)+ \lambda f(
        t_1, t_2)
      \end{cases}
  \end{equation*}
  is a monotone endofunction on $\Delta$, hence by Tarski's fixed
  point theorem, $D$ has a least fixed point. Now define
  $h(s_1,s_2) = \inf\{\delta \mid (s_1,s_2) \in R_\delta \in \cal
  R\}$; since $R_\delta$ is downward compact, we have
  $(s_1,s_2)\in R_{h(s_1,s_2)}$. By showing that $h$ is a pre-fixed
  point of $D$, \ie~that $D(h) \le_\Delta h$, we get that
  $(s_1,s_2)\in R_\delta$ implies that $\md(s_1,s_2) \le \delta$,
  since $h(s_1,s_2) \le \delta$ and $\md(s_1,s_2) \le h(s_1,s_2)$.

  Since $(s_1,s_2)\in R_{h(s_1,s_2)}$, every $s_1\mayto{k_1}s'_1$ can be
  matched by some $s_2\mayto{k_2} s'_2$ such that $d_\Spec(k_1,k_2) +
  \lambda \delta' \le h(s_1,s_2)$ for some $\delta'$ where $(s'_1,s'_2)
  \in R_{\delta'}$, implying $h(s'_1,s'_2) \le \delta'$, but then also
  $d_\Spec(k_1,k_2) + \lambda h(s'_1,s'_2) \le h(s_1,s_2)$. Similarly,
  every $s_2\mustto{k_2} s'_2$ has a match $s_1 \mustto{k_1} s'_1$ such
  that $d_\Spec(k_1,k_2) + \lambda h(s'_1,s'_2) \le h(s_1,s_2)$. Hence we
  have $D(h) \le_\Delta h$ which was to be shown. \qed
\end{proof}

The next theorems show that modal refinement distance indeed
overapproximates thorough refinement distance, and that it is exact
for deterministic WMTS.  Note that nothing general can be said about
the precision of the overapproximation in the nondeterministic case;
the standard counterexample given for the Boolean case
in~\cite{DBLP:journals/tcs/BenesKLS09} shows that there exist WMTS for
which $\thd(S_1,S_2) = 0$ but $\md(S_1,S_2) = \infty$.

\begin{theorem}
  \label{weightedmodal.th:dtledm}
  For WMTS $S_1$, $S_2$ we have $\thd( S_1, S_2)\le \md( S_1,
  S_2)$.
\end{theorem}

\begin{proof}
  If $\md( S_1, S_2)=\infty$, we have nothing to prove.  Otherwise,
  let $\cal R=\{ R_\epsilon\subseteq S_1\times S_2\mid \epsilon\ge 0\}$ be a
  modal refinement family which witnesses $\md( S_1, S_2)$, \ie~such
  that $( s_1^0, s_2^0)\in R_{ \md( S_1, S_2)}$, and let
  $I_1\in\llbracket S_1\rrbracket$.  We have to expose $I_2\in\llbracket
  S_2\rrbracket$ for which $d( I_1, I_2)\le \md( S_1, S_2)$.

  Let $\tilde R\subseteq I_1\times S_1$ be a witness for $I_1\mr S_1$,
  define
  $R'_\epsilon= \tilde R\circ R_\epsilon\subseteq I_1\times S_2$ for
  all $\epsilon\ge 0$, and let
  $\cal R'=\{ R'_\epsilon\mid \epsilon\ge 0\}$.  The states of
  $I_2=( I_2, i_2^0, \Impl \Spec, \mto_{ I_2})$ are $I_2= S_2$ with
  $i_2^0= s_2^0$, and the transitions we define as follows:

  For any $i_1\tto{ k_1'}_{ I_1} j_1$ and any $s_2\in S_2$ for which $(
  i_1, s_2)\in R'_\epsilon\in \cal R'$ for some $\epsilon$, we have
  $s_2\mayto{ k_2}_2 t_2$ in $S_2$ with $d_\Spec( k_1', k_2)\le \epsilon$
  and $( j_1, t_2)\in R'_{ \epsilon'}\in \cal R'$ for some $\epsilon'\le
  \lambda^{ -1}\big( \epsilon- d_\Spec( k_1', k_2)\big)$.  Write $k_1'=(
  a_1', x_1')$ and $k_2=\big( a_2,[ x_2, y_2]\big)$, then we must have
  $a_1'= a_2$.  Let
  \begin{equation}
    x_2'=
    \begin{cases}
      x_2 &\text{if } x_1'< x_2, \\
      x_1' &\text{if } x_2\le x_1'\le y_2, \\
      y_2 &\text{if } x_1'> y_2
    \end{cases}
    \label{weightedmodal.eq:x_2'_app}
  \end{equation}
  and $k_2'=( a_2, x_2')$, and put $s_2\tto{ k_2'}_{ I_2} t_2$ in $I_2$.
  Note that
  \begin{equation}
    d_\Spec( k_1', k_2')= d_\Spec( k_1', k_2).
    \label{weightedmodal.eq:dk1'k2'=dk1'k2_app}
  \end{equation}

  Similarly, for any $s_2\mustto{ k_2}_2 t_2$ in $S_2$ and any
  $i_1\in I_1$ with $( i_1, s_2)\in R'_\epsilon\in \cal R'$ for some
  $\epsilon$, we have $i_1\tto{ k_1'}_{ I_1} j_1$ with
  $d_\Spec( k_1', k_2)\le \epsilon$ and
  $( j_1, t_2)\in R'_{ \epsilon'}\in \cal R'$ for some
  $\epsilon'\le \lambda^{ -1}\big( \epsilon- d_\Spec( k_1', k_2)\big)$.
  Write $k_1'=( a_1', x_1')$ and $k_2=( a_2,[ x_2, y_2])$, define
  $x_2'$ as in~\eqref{weightedmodal.eq:x_2'_app} and
  $k_2'=( a_2, x_2')$, and put $s_2\tto{ k_2'}_{ I_2} t_2$ in $I_2$.

  We show that the identity relation $\id_{ S_2}=\{( s_2, s_2)\mid
  s_2\in S_2\}\subseteq S_2\times S_2$ witnesses $I_2\mr S_2$.  Let
  first $s_2\tto{ k_2'}_{ I_2} t_2$; we must have used one of the two
  constructions above for creating this transition.  In the first
  case, we have $s_2\mayto{ k_2}_2 t_2$ with $k_2'\labpre k_2$,
  and in the second case, we have $s_2\mustto{ k_2}_2 t_2$, hence also
  $s_2\mayto{ k_2}_2 t_2$, with the same property.  For a transition
  $s_2\mustto{ k_2}_2 t_2$ on the other hand, we have introduced
  $s_2\tto{ k_2'}_{ I_2} t_2$ in the second construction above, with
  $k_2'\labpre k_2$.

  We also want to show that the family $\cal R'$ is a witness for
  $d( I_1, I_2)\le \md( S_1, S_2)$.  We have
  $( i_1^0, s_2^0)\in R'_{ \md( S_1, S_2)}= \tilde R\circ R_{ \md(
    S_1, S_2)}$, so let $( i_1, s_2)\in R'_\epsilon\in \cal R'$ for some
  $\epsilon\ge 0$.  For any $i_1\tto{ k_1'}_{ I_1} j_1$ we have
  $s_2\mayto{ k_2}_2 t_2$ and $s_2\tto{ k_2'}_{ I_2} t_2$ by the first
  part of our construction above, with
  $d_\Spec( k_1', k_2')= d_\Spec( k_1', k_2)\le \epsilon$ because
  of~\eqref{weightedmodal.eq:dk1'k2'=dk1'k2_app}, and also
  $( j_1, t_2)\in R'_{ \epsilon'}\in \cal R'$ for some
  $\epsilon'\le \lambda^{ -1}\big( \epsilon- d_\Spec( k_1', k_2)\big)$.
  For any $s_2\tto{ k_2'}_{ I_2} t_2$, we must have used one of the
  constructions above to introduce this transition, and both give us
  $i_1\tto{ k_1'}_{ I_1} j_1$ with $d_\Spec( k_1', k_2')\le \epsilon$
  and $( j_1, t_2)\in R'_{ \epsilon'}\in \cal R'$ for some
  $\epsilon'\le \lambda^{ -1}\big( \epsilon- d_\Spec( k_1',
  k_2)\big)$. \qed
\end{proof}

The fact that modal refinement only equals thorough refinement for
deterministic specifications is well-known from the theory of modal
transition systems~\cite{DBLP:conf/avmfss/Larsen89}, and the special
case of $S_2$ deterministic is important, as it can be
argued~\cite{DBLP:conf/avmfss/Larsen89} that deterministic
specifications are sufficient for applications.

\begin{theorem}
  \label{weightedmodal.th:det-dteqdm}
  If $S_2$ is deterministic, then $\thd( S_1, S_2)= \md( S_1, S_2)$.
\end{theorem}

\begin{proof}
  If $\thd( S_1, S_2)= \infty$, we are done by
  Theorem~\ref{weightedmodal.th:dtledm}.  Otherwise, let
  $\cal R=\{ R_\epsilon\mid \epsilon\ge 0\}$ be the smallest relation
  family for which
  \begin{itemize}
  \item $( s_1^0, s_2^0)\in R_{ \thd( S_1, S_2)}$ and
  \item whenever we have $( s_1, s_2)\in R_\epsilon\in \cal R$, $s_1\mayto{
    a, I_1}_1 t_1$, and $s_2\mayto{ a, I_2}_2 t_2$, then $( t_1, t_2)\in
    R_{ \lambda^{ -1}( \epsilon- d_\Spec(( a, I_1),( a, I_2)))}\in \cal R$.
  \end{itemize}
  We show below that $\cal R$ is well-defined (also that
  $\epsilon- d_\Spec\big(( a, I_1),( a, I_2)\big)\ge 0$ in all cases)
  and a modal refinement family.  We will use the convenient notation
  $( s_1, S_1)$ for the WMTS $S_1$ with initial state $s_1^0$ replaced
  by $s_1$, similarly for $( s_2, S_2)$.

  We first show inductively that for any pair of states $( s_1,
  s_2)\in R_\epsilon\in \cal R$ we have $\thd\big(( s_1, S_1),( s_2,
  S_2)\big)\le \epsilon$.  This is obviously the case for $s_1= s_1^0$
  and $s_1= s_2^0$, so assume now that $( s_1, s_2)\in R_\epsilon\in
  \cal R$ is such that $\thd\big(( s_1, S_1),( s_2, S_2)\big)\le
  \epsilon$ and let $s_1\mayto{ a, I_1}_1 t_1$, $s_2\mayto{ a, I_2}_2
  t_2$.
  Let $( q_1', P_1')\in\llbracket( t_1, S_1)\rrbracket$ and $x_1\in
  I_1$.

  There is an implementation $( p_1, P_1)\in\llbracket( s_1,
  S_1)\rrbracket$ for which $p_1\tto{ a, x_1} q_1$ and such that $( q_1,
  P_1)\mr( q_1', P_1')$.  Now
  \begin{equation*}
    \thd\big(( p_1, P_1),( s_2, S_2)\big)\le \thd\big(( p_1, P_1),(
    s_1, S_1)\big)+ \thd\big(( s_1, S_1),( s_2, S_2)\big)\le \epsilon,
  \end{equation*}
  hence we must have $s_2\mayto{ a_2', I_2'}_2 t_2'$ with $d_\Spec\big((
  a, x_1),( a_2', I_2')\big)\le \epsilon$.  But then $a_2'= a$, hence
  by determinism of $S_2$, $I_2= I_2'$ and $t_2= t_2'$.

  The above considerations hold for any $x_1\in I_1$, hence
  $d_\Spec\big(( a, I_1),( a, I_2)\big)\le \epsilon$.  Thus
  $\epsilon- d_\Spec\big(( a, I_1),( a, I_2)\big)\ge 0$, and the
  definition of $\cal R$ above is justified.  Now let $x_2\in I_2$
  such that
  $d_\Spec\big(( a, x_1),( a, x_2)\big)= d_\Spec\big(( a, x_1),( a,
  I_2)\big)$, then there is an implementation
  $( p_2, P_2)\in\llbracket( s_2, S_2)\rrbracket$ for which
  $p_2\tto{ a, x_2} q_2$, and
  \begin{align*}
    d\big(( q_1', P_1'),( q_2, P_2)\big) &\le \lambda^{ -1}\big(
    \epsilon- d_\Spec(( a, x_1),( a, x_2))\big) \\
    &= \lambda^{ -1}\big( \epsilon- d_\Spec(( a, I_1),( a, I_2))\big)\,,
  \end{align*}
  which, as $( q_1', P_1')\in\llbracket( t_1, S_1)\rrbracket$ was
  chosen arbitrarily, entails
  \begin{equation*}
    \thd\big(( s_1, S_1),( s_2,
    S_2)\big)\le \lambda^{ -1}\big( \epsilon- d_\Spec(( a, I_1),( a,
    I_2))\big)\,.
  \end{equation*}

  We are ready to show that $\cal R$ is a modal refinement family.
  Let $( s_1, s_2)\in R_\epsilon\in \cal R$ for some $\epsilon$, and
  assume $s_1\mayto{ a,I_1}_1 t_1$.  Let $x\in I_1$, then there is
  $( p, P^x)\in\llbracket( s_1, S_1)\rrbracket$ with a transition
  $p\tto{ m} q$.  Now
  $\thd\big(( p, P^x),( s_2, S_2)\big)\le \epsilon$, hence we have a
  transition $s_2\mayto{ a, I_2^x}_2 t_2^x$ with
  $d_\Spec\big(( a, x),( a, I_2^x)\big)\le \epsilon$.  Also for any
  other $x'\in I_1$ we have a transition
  $s_2\mayto{ a, I_2^{ x'}}_2 t_2^{ x'}$ with
  $d_\Spec\big(( a, x'),( a, I_2^{ x'})\big)\le \epsilon$, hence by
  determinism of $S_2$, $I_2^x= I_2^{ x'}$ and $t_2^x= t_2^{ x'}$.  It
  follows that there is a unique transition $s_2\mayto{ a, I_2} t_2$,
  and as $d_\Spec\big(( a, x),( a, I_2)\big)\le \epsilon$ for all
  $x\in I_1$, we have
  $d_\Spec\big(( a, I_1),( a, I_2)\big)\le \epsilon$, and
  $( t_1, t_2)\in R_{ \lambda^{ -1}( \epsilon- d_\Spec(( a, I_1),( a,
    I_2)))}\in \cal R$ by definition.

  Now assume $s_2\mustto{ a,I_2}_2 t_2$. Let $( p_1,
  P_1)\in\llbracket( s_1, S_1)\rrbracket$, then we have $( p_2,
  P_2)\in\llbracket( s_2, S_2)\rrbracket$ with $d\big(( p_1, P_1),(
  p_2, P_2)\big)\le \epsilon$.  Now any $( p_2, P_2)\in\llbracket(
  s_2, S_2)\rrbracket$ has $p_2\tto{ a, x_2} q_2$ with $x_2\in I_2$,
  thus there is also $p_1\tto{ a, x_1} q_1$ with $d_\Spec\big(( a, x_1),(
  a, x_2)\big)\le \epsilon$ and $d\big(( q_1, P_1),( q_2,
  P_2)\big)\le \lambda^{ -1}\big( \epsilon- d_\Spec(( a, x_1),( a,
  x_2))\big)$.  This in turn implies that $s_1\mustto{ a, I_1}_1 t_1$
  for some $x_1\in I_1$.  We will be done once we can show $d_\Spec\big((
  a, I_1),( a, I_2)\big)\le \epsilon$, so assume to the contrary
  that there is $x_1'\in I_1$ with $d_\Spec\big(( a, x_1'),( a,
  I_2)\big)> \epsilon$.  Then there must be an implementation $( p_1',
  P_1')\in\llbracket( s_1, S_1)\rrbracket$ with $p_1'\tto{ a, x_1'}
  q_1'$, hence a transition $s_2\mayto{ a, I_2'}_2 t_2'$ with
  $d_\Spec\big(( a, x_1'),( a, I_2')\big)\le \epsilon$.  But $I_2'=
  I_2$ by determinism of $S_2$, a contradiction. \qed
\end{proof}

\section{Relaxation}
\label{weightedmodal.se:relax}

We introduce here a notion of \emph{relaxation} which is specific to the
quantitative setting.  Intuitively, relaxing a specification means to
weaken the quantitative constraints, while the discrete demands on which
transitions may or must be present in implementations are kept.  A
similar notion of \emph{strengthening} may be defined, but we do not use
this here.

\begin{definition}
  \label{weightedmodal.def:relax}
  For WMTS $S$, $S'$ and $\epsilon\ge 0$, $S'$ is an
  \emph{$\epsilon$-relaxation} of $S$ if $S\mr S'$ and
  $S'\mr^\epsilon S$.
\end{definition}

Hence the quantitative constraints in $S'$ may be more permissive than
the ones in $S$, but no new discrete behavior may be introduced.  Also
note that any implementation of $S$ is also an implementation of $S'$,
and no implementation of $S'$ is further than $\epsilon$ away from an
implementation of $S$.  The following proposition relates specifications
to relaxed specifications:

\begin{proposition}
  \label{weightedmodal.pr:wide-two}
  If $S_1'$ and $S_2'$ are $\epsilon$-relaxations of $S_1$ and $S_2$,
  respectively, then $\md( S_1, S_2)- \epsilon\le \md( S_1,
  S_2')\le \md( S_1, S_2)$ and $\md( S_1, S_2)\le \md( S_1',
  S_2)\le \md( S_1, S_2)+ \epsilon$.
\end{proposition}

\begin{proof}
  By the triangle inequality we have
  \begin{align*}
    \md( S_1, S_2') &\le \md( S_1, S_2)+ \md( S_2, S_2'), \\
    \md( S_1, S_2) &\le \md( S_1, S_2')+ \md( S_2', S_2), \\
    \md( S_1, S_2) &\le \md( S_1, S_1')+ \md( S_1', S_2), \\
    \md( S_1', S_2) &\le \md( S_1', S_1)+ \md( S_1,
    S_2). \qed
  \end{align*}
\end{proof}

On the syntactic level, we can introduce the following \emph{widening}
operator which relaxes all quantitative constraints in a systematic
manner.  We write $I\pm \delta=[ x- \delta, y+ \delta]$ for an interval
$I=[ x, y]$ and $\delta\in \Nat$.

\begin{definition}\label{weightedmodal.def:widening}
  Given $\delta\in \Nat$, the \emph{$\delta$-widening} of a WMTS $S$
  is the WMTS $S^{ +\delta}$ with transitions $s\mayto{ a, I\pm
  \delta} t$ in $S^{ +\delta}$ for all $s\mayto{ a, I} t$ in $S$, and
  $s\mustto{ a, I\pm \delta} t$ in $S^{ +\delta}$ for all $s\mustto{
  a, I} t$ in $S$.
\end{definition}

Widening and relaxation are related as follows; note also that as
widening is a global operation whereas relaxation may be achieved
entirely locally, not all relaxations may be obtained as widenings.

\begin{proposition}
  \label{weightedmodal.pr:wide-propt}
  The $\delta$-widening of any WMTS $S$ is a
  $\frac \delta{ 1- \lambda}$-relaxation.
\end{proposition}

\begin{proof}
  For the first claim, the identity relation $\id_S=\{( s, s)\mid s\in
  S\}\subseteq S\times S$ is a witness for $S\mr S^{ +\delta}$: if
  $s\mayto{ k} t$, then by construction $s\mayto{ k_2}_{ +\delta} t$
  with $k\labpre k_2$, and if $s\mustto{ k_2}_{ +\delta} t$, then
  again by construction $s\mustto{ k} t$ for some $k\labpre k_2$.

  Now to prove $\md( S^{ +\delta}, S)\le( 1- \lambda)^{ -1}
  \delta$, we define a family of relations $\cal R=\{ R_\epsilon\mid
  \epsilon\ge 0\}$ by $R_\epsilon= \emptyset$ for $\epsilon<( 1-
  \lambda)^{ -1} \delta$ and $R_\epsilon= \id_S$ for $\epsilon\ge {(
    1- \lambda)^{ -1} \delta}$.  We show that $\cal R$ is a modal
  refinement family.

  Let $( s, s)\in R_\epsilon$ for some $\epsilon\ge {( 1- \lambda)^{
      -1} \delta}$, and assume $s\mayto{ k_2}_{ +\delta} t$.  By
  construction there is a transition $s\mayto{ k} t$ with $d_\Spec( k_2,
  k)\le \delta\le \epsilon$.  Now
  \begin{equation*}
    \frac1\lambda\Big( \epsilon- d_\Spec( k_2, k)\Big)\ge \frac1\lambda\Big(
    \frac \delta{ 1- \lambda}- \delta\Big)= \frac \delta{ 1-
      \lambda}\ge \epsilon
  \end{equation*}
  and $( t, t)\in R_\epsilon$, which settles this part of the proof.
  The other direction, starting with a transition $s\mustto{ k} t$, is
  similar.  \qed
\end{proof}

There is also an implementation-level notion which corresponds to
relaxation:

\begin{definition}
  The \emph{$\epsilon$-extended implementation semantics}, for
  $\epsilon\ge 0$, of a WMTS $S$ is $\llbracket S\rrbracket^{
    +\epsilon}=\big\{ I\bigmid I\mr^\epsilon S, I \text{
    implementation}\big\}$.
\end{definition}

\begin{proposition}
  \label{weightedmodal.pr:semwide}
  If $S'$ is an $\epsilon$-relaxation of $S$, then $\llbracket
  S'\rrbracket\subseteq \llbracket S\rrbracket^{ +\epsilon}$.
\end{proposition}

\begin{proof}
  If $I\in \llbracket S'\rrbracket$, then $\md( I, S')= 0$, hence
  $\md( I, S)\le \epsilon$ by Proposition~\ref{weightedmodal.pr:wide-two}, which in
  turn implies that $I\in \llbracket S\rrbracket^{+\epsilon}$. \qed
\end{proof}

The example in Figure~\ref{weightedmodal.fi:sem-wide-counterex} shows that there are
WMTS $S$, $S'$ such that $S'$ is an $\epsilon$-relaxation of $S$ but the
inclusion $\llbracket S'\rrbracket\subseteq \llbracket S\rrbracket^{
  +\epsilon}$ is strict.  Indeed, for $\delta= 1$ and $\lambda= .9$, we
have $I\in \llbracket S\rrbracket^{ +( 1- \lambda)^{ -1} \delta}$, but
$I\notin \llbracket S^{ +\delta}\rrbracket$.

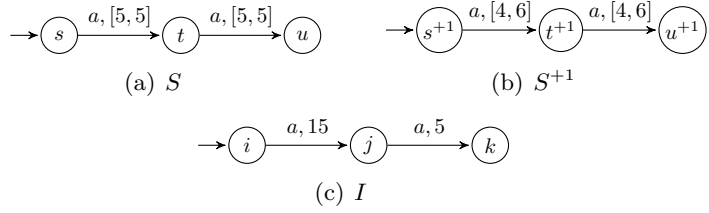
\begin{figure}[tpb]
  \centering
  \subbottom[$S$]{
    \begin{tikzpicture}[->,>=stealth',shorten >=1pt,auto,node
      distance=2.0cm,initial text=,scale=0.8,transform shape]
      \tikzstyle{every node}=[font=\small] \tikzstyle{every
        state}=[fill=white,shape=circle,inner sep=.5mm,minimum size=6mm]
      \node[initial,state] (i) at (0,0) {$s$};
      \node[state] (j) at (2,0) {$t$};
      \node[state] (k) at (4,0) {$u$};
      \path (i) edge node [above] {$a,[ 5, 5]$} (j);
      \path (j) edge node [above] {$a,[ 5, 5]$} (k);
    \end{tikzpicture}}
  \quad
  \subbottom[$S^{ +1}$]{
    \begin{tikzpicture}[->,>=stealth',shorten >=1pt,auto,node
      distance=2.0cm,initial text=,scale=0.8,transform shape]
      \tikzstyle{every node}=[font=\small] \tikzstyle{every
        state}=[fill=white,shape=circle,inner sep=.5mm,minimum size=6mm]
      \node[initial,state] (i) at (0,0) {$s^{ +1}$};
      \node[state] (j) at (2,0) {$t^{ +1}$};
      \node[state] (k) at (4,0) {$u^{ +1}$};
      \path (i) edge node [above] {$a,[ 4, 6]$} (j);
      \path (j) edge node [above] {$a,[ 4, 6]$} (k);
    \end{tikzpicture}}
  \quad
  \subbottom[$I$]{
    \begin{tikzpicture}[->,>=stealth',shorten >=1pt,auto,node
      distance=2.0cm,initial text=,scale=0.8,transform shape]
      \tikzstyle{every node}=[font=\small] \tikzstyle{every
        state}=[fill=white,shape=circle,inner sep=.5mm,minimum size=6mm]
      \node[initial,state] (i) at (0,0) {$i$};
      \node[state] (j) at (2,0) {$j$};
      \node[state] (k) at (4,0) {$k$};
      \path (i) edge node [above] {$a, 15$} (j);
      \path (j) edge node [above] {$a, 5$} (k);
    \end{tikzpicture}}
  \caption{WMTS $S$ and implementation $I$ for which $I\in \llbracket
    S\rrbracket^{ +( 1- \lambda)^{ -1} \delta}$, for $\delta= 1$ and
    $\lambda= .9$ (thus $( 1- \lambda)^{ -1} \delta= 10$), but $I\notin
    \llbracket S^{ +\delta}\rrbracket$, so that $\llbracket S^{
      +\delta}\rrbracket\subsetneq \llbracket
    S\rrbracket^{ +( 1- \lambda)^{ -1} \delta}$, even though $S^{
      +\delta}$ is a $( 1- \lambda)^{ -1} \delta$-relaxation of $S$.}
  \label{weightedmodal.fi:sem-wide-counterex}
\end{figure}

\section{Limitations of the Quantitative Approach}
\label{weightedmodal.se:hull}

In this section we turn our attention towards some of the standard
operators for specification theories; determinization and logical
conjunction. In the standard Boolean setting, there is indeed a
determinization operator which derives the smallest deterministic
overapproximation of a specification, which is useful because it
enables checking thorough refinement,
\cf~Theorem~\ref{weightedmodal.th:det-dteqdm}. Quite surprisingly, we show that in
the quantitative setting, there are problems with these notions which
do not appear in the Boolean theory.  More specifically, we show that
there is no determinization operator which always yields a
smallest deterministic overapproximation, and there is no
conjunction operator which acts as a greatest lower bound.

\begin{theorem}
  \label{weightedmodal.th:no-dethull}
  There is no unary operator $\mathcal D$ on WMTS for which it holds
  that
  \begin{enumerate}[$(\ref{weightedmodal.th:no-dethull}.1)$]
  \item \label{weightedmodal.en:no-dethull:det}
    $\mathcal D( S)$ is deterministic for any WMTS $S$,
  \item
    \label{weightedmodal.en:no-dethull:ub}
    $S\mr \mathcal D( S)$ for any WMTS $S$,
  \item
    \label{weightedmodal.en:no-dethull:lub}
    $S\mr^\epsilon D$ implies $\mathcal D( S)\mr^\epsilon D$ for
    any WMTS $S$, any deterministic WMTS $D$, and any $\epsilon\ge 0$.
  \end{enumerate}
\end{theorem}

\begin{proof}
  There is a determinization operator $\mathcal D'$ on WMTS which
  satisfies Properties~$(\ref{weightedmodal.th:no-dethull}.\ref{weightedmodal.en:no-dethull:det})$
  and $(\ref{weightedmodal.th:no-dethull}.\ref{weightedmodal.en:no-dethull:ub})$ above and a weaker
  version of Property~$(\ref{weightedmodal.th:no-dethull}.\ref{weightedmodal.en:no-dethull:lub})$
  with $\epsilon= 0$:
  \begin{enumerate}[$(\ref{weightedmodal.th:no-dethull}.1')$]
  \item[$(\ref{weightedmodal.th:no-dethull}.\ref{weightedmodal.en:no-dethull:lub}')$] $S\mr D$
    implies $\mathcal D'( S)\mr D$ for any WMTS $S$ and any
    deterministic WMTS $D$.
  \end{enumerate}
  This $\mathcal D'$ can be defined as follows: For a WMTS $S =
  (S,s_0,\mmayto,\mmustto)$,
  \begin{equation*}
    \mathcal D'(S) = \big( 2^S \setminus
    \{\emptyset\},\{s_0\},\mmayto_d,\mmustto_d\big),
  \end{equation*}
  where $2^S$ is the power set of $S$ and the transition relations
  $\mmayto_d$ and $\mmustto_d$ are defined as follows: Let $\mathcal T
  \in( 2^S\setminus \{\emptyset\})$ be a state in $\mathcal
  D'(S)$. For every maximal, nonempty set $L_a \subseteq \{ I \mid
  \exists s \in \mathcal T : s \mayto{ a, I}\}$ for some $a \in
  \Sigma$, we have $\mathcal T \mayto{ a, \bigcup L_a}_d \mathcal T_a$
  where $\mathcal T_a = \{ s' \in S \mid \exists s \in \mathcal T, I
  \in L_a : s\mayto{ a, I} s' \}$ and $\bigcup L_a$ is the smallest
  interval containing all intervals from $L_a$. If, moreover, for each
  $s \in \mathcal T$ we have $s \mustto{ a, I} s'$ for some $s' \in
  \mathcal T_a$ and some $I \in L_a$, then $\mathcal T \mustto{ a,
    \bigcup L_a}_d \mathcal T_a$. It is straightforward to prove that
  $\mathcal D'$ satisfies the expected properties.

  Assume now that there is an operator $\mathcal D$ as in the theorem.
  Then for any WMTS $S$, $S\mr \mathcal D'( S)$ and thus
  $\mathcal D( S)\mr \mathcal D'( S)$
  by~$(\ref{weightedmodal.th:no-dethull}.\ref{weightedmodal.en:no-dethull:lub})$,
  and $S\mr \mathcal D( S)$ and hence
  $\mathcal D'( S)\mr \mathcal D( S)$
  by~$(\ref{weightedmodal.th:no-dethull}.\ref{weightedmodal.en:no-dethull:lub}')$.
  We finish the proof by showing that the operator $\mathcal D'$ does
  not
  satisfy~$(\ref{weightedmodal.th:no-dethull}.\ref{weightedmodal.en:no-dethull:lub})$.
  The example in Figure~\ref{weightedmodal.fig:cex} shows a WMTS $S$
  and a deterministic WMTS $D$ for which
  $\md\big( \mathcal D'( S), D\big)= 3+ 3\lambda$ and
  $\md( S, D)= \max( 3, 3\lambda)= 3$, hence
  $\md\big( \mathcal D'( S), D\big)\not\le \md( S, D)$. \qed
\end{proof}

\begin{figure}[btp]
  \centering 
  \subbottom[$S$]{
    \begin{tikzpicture}[->,>=stealth',shorten >=1pt,auto,node
      distance=2.0cm,initial text=,scale=0.8,transform shape]
      \tikzstyle{every node}=[font=\small] \tikzstyle{every
        state}=[fill=white,shape=circle,inner sep=.5mm,minimum size=6mm]
      \node[initial,state] (i) at (0,0) {$s_0$};
      \node[state] (s) at (2,1) {$s_1$};
      \node[state] (t) at (2,-1) {$s_2$};
      \node[state] (spost) at (4,1) {$s_3$};
      \node[state] (tpost) at (4,-1) {$s_4$};
      \path (i) edge [densely dashed] node [above,sloped] {$a, [3,3]$} (s);
      \path (i) edge [densely dashed] node [above,sloped] {$a, [5,6]$} (t);
      \path (t) edge [densely dashed] node [above,sloped] {$a,
        [0,0]$} (tpost); 
      \path (s) edge [densely dashed] node [above,sloped] {$a,
        [3,3]$} (spost); 
    \end{tikzpicture}}
  \hspace{2mm}
  \subbottom[$\mathcal D'(S)$]{
    \raisebox{8mm}{
      \begin{tikzpicture}[->,>=stealth',shorten >=1pt,auto,node
        distance=2.0cm,initial text=,scale=0.8,transform shape]
        \tikzstyle{every node}=[font=\small] \tikzstyle{every
          state}=[fill=white,shape=circle,inner sep=.5mm,minimum size=6mm]
        \node[initial,state,shape=rectangle,rounded corners] (i) at
        (0,0) {$\{s_0\}$}; 
        \node[state,shape=rectangle,rounded corners] (s) at (3,0)
        {$\{s_1,s_2\}$}; 
        \node[state,shape=rectangle,rounded corners] (t) at (6,0)
        {$\{s_3,s_4\}$}; 
        \path (i) edge [densely dashed] node [above,sloped] {$a,
          [3,6]$} (s); 
        \path (s) edge [densely dashed] node [above,sloped] {$a,
          [0,3]$} (t); 
      \end{tikzpicture}}}
  \hspace{2mm}
  \subbottom[$D$]{
    \begin{tikzpicture}[->,>=stealth',shorten >=1pt,auto,node
      distance=2.0cm,initial text=,scale=0.8,transform shape]
      \tikzstyle{every node}=[font=\small] \tikzstyle{every
        state}=[fill=white,shape=circle,inner sep=.5mm,minimum size=6mm]
      \node[initial,state] (i) at (0,0) {$d_0$};
      \node[state] (s) at (2,0) {$d_1$};
      \node[state] (t) at (4,0) {$d_2$};
      \path (i) edge [densely dashed] node [above,sloped] {$a, [2,3]$} (s);
      \path (s) edge [densely dashed] node [above,sloped] {$a, [0,0]$} (t);
    \end{tikzpicture}
  }
  \caption{Counter-example for Theorem~\ref{weightedmodal.th:no-dethull}: $\md\big(
    \mathcal D'( S), D\big)= 3+ 3\lambda$ and $\md( S, D)= \max( 3,
    3\lambda)= 3$, hence $\md\big( \mathcal D'( S), D\big)\not\le
    \md( S, D)$.}
  \label{weightedmodal.fig:cex}
\end{figure}

Likewise, the greatest-lower-bound property of logical conjunction in
the Boolean setting ensures that the set of implementations of a
conjunction of specifications is precisely the intersection of the
implementation sets of the two specifications.  Conjoining two WMTS
naturally involves a partial label conjunction operator $\owedge$. We let
$( a_1, I_1)\owedge( a_2, I_2)$ be undefined if $a_1\ne a_2 $, and
otherwise
\begin{align*}
  \big( a,[ x_1, y_1]\big)\owedge\big( a,[ x_2, y_2]\big) &= 
  \begin{cases}
    \big( a,[ \max(x_1,x_2), \min(y_1, y_2)]\big) \\
    &\hspace*{-7em}\text{if }
    \max(x_1,x_2) \le \min(y_1, y_2),\\
    \text{undefined} &\hspace*{-5em}\text{otherwise}.
\end{cases}
\end{align*}

Before we show that such a conjunction operator for WMTS does not
exist in general, we need to define a \emph{pruning operator} which
removes inconsistent states that naturally arise when conjoining two
WMTS.  The intuition is that if a WMTS $S_1$ requires a behavior $s_1
\mustto{ k_1}_1$ for which there is no may transition $s_2 \mayto{
  k_2}_2$ such that $k_1 \owedge k_2$ is defined, then the state
$(s_1,s_2)$ in the conjunction is \emph{inconsistent} and will have to
be pruned away, together with all \textit{must} transitions leading to
it.  In the definition below, $\pre^*$ denotes the reflexive,
transitive closure of $\pre$.

\begin{definition}
  \label{weightedmodal.de:prune}
  For a WMTS $S$, let $\pre: 2^S\to 2^S$ be given by $\pre( B)=\{ s\in
  S\mid s\mustto{ k} t\in B \text{ for some } k\}$. Let $B
  \subseteq S$ be a set of \emph{inconsistent} states. If $s^0\notin
  \pre^*( B )$, then the \emph{pruning of $S$
    w.r.t.~$B$} is defined by $\rho^B( S)=( S_\rho,
  s^0, \mmayto_\rho, \mmustto_\rho)$ where $S_\rho= S\setminus \pre^*(
  B )$, $\mmayto_\rho= \mmayto\cap\big( S_\rho\times \Spec \times
  S_\rho\big)$ and $\mmustto_\rho= \mmustto\cap\big( S_\rho\times \Spec
  \times S_\rho\big)$.
\end{definition}

\begin{theorem}
  \label{weightedmodal.th:no-conj}
  There is no partial binary operator $\wedge$ on WMTS for which it
  holds that, for all WMTS $S$, $S_1$, $S_2$ such that $S_1$ and $S_2$
  are deterministic,
  \begin{enumerate}[$(\ref{weightedmodal.th:no-conj}.1)$]
  \item
    \label{weightedmodal.en:no-conj:lb}
    whenever $S_1 \wedge S_2$ is defined, then
    $S_1\wedge S_2\mr S_1$ and $S_1\wedge S_2\mr S_2$,
  \item 
    \label{weightedmodal.en:no-conj:def}
    whenever $S\mr S_1$ and
    $S\mr S_2$, then $S_1 \wedge S_2$ is defined and $S\mr S_1\wedge S_2$,
  \item
    \label{weightedmodal.en:no-conj:glb}
    for any $\epsilon\ge 0$, there exist $\epsilon_1\ge 0$ and
    $\epsilon_2\ge 0$ such that if $S_1 \wedge S_2$ is defined,
    $S\mr^{ \epsilon_1} S_1$ and $S\mr^{ \epsilon_2} S_2$, then
    $S\mr^\epsilon S_1\wedge S_2$.
  \end{enumerate}
\end{theorem}

\begin{figure}[tpb]
  \centering
  \subbottom[$S$]{
    \begin{tikzpicture}[->,>=stealth',shorten >=1pt,auto,node
      distance=2.0cm,initial text=,scale=0.8,transform shape]
      \tikzstyle{every node}=[font=\small] \tikzstyle{every
        state}=[fill=white,shape=circle,inner sep=.5mm,minimum size=6mm]
      \node[initial,state] (s) at (0,0) {$s$};
      \node[state] (t) at (2,0) {$t$};
      \path (s) edge [densely dashed] node [above] {$a,[ 1, 2]$} (t);
    \end{tikzpicture}}
  \qquad
  \subbottom[$S_1$]{
    \begin{tikzpicture}[->,>=stealth',shorten >=1pt,auto,node
      distance=2.0cm,initial text=,scale=0.8,transform shape]
      \tikzstyle{every node}=[font=\small] \tikzstyle{every
        state}=[fill=white,shape=circle,inner sep=.5mm,minimum size=6mm]
      \node[initial,state] (s) at (0,0) {$s_1$};
      \node[state] (t) at (2,0) {$t_1$};
      \path (s) edge [densely dashed] node [above] {$a,[ 0, 1]$} (t);
    \end{tikzpicture}}
  \\
  \subbottom[$S_2$]{
    \begin{tikzpicture}[->,>=stealth',shorten >=1pt,auto,node
      distance=2.0cm,initial text=,scale=0.8,transform shape]
      \tikzstyle{every node}=[font=\small] \tikzstyle{every
        state}=[fill=white,shape=circle,inner sep=.5mm,minimum size=6mm]
      \node[initial,state] (s) at (0,0) {$s_2$};
      \node[state] (t) at (2,0) {$t_2$};
      \path (s) edge [densely dashed] node [above] {$a,[ 2, 3]$} (t);
    \end{tikzpicture}}
  \qquad
  \subbottom[$S_1\wedge S_2$]{
    \begin{tikzpicture}[->,>=stealth',shorten >=1pt,auto,node
      distance=2.0cm,initial text=,scale=0.8,transform shape]
      \tikzstyle{every node}=[font=\small] \tikzstyle{every
        state}=[fill=white,shape=circle,inner sep=.5mm,minimum size=6mm]
      \node[initial,state,shape=rectangle,rounded corners] (s) at
      (0,0) {$( s_1, s_2)$};
    \end{tikzpicture}}
  \caption{Counter-example for Theorem~\ref{weightedmodal.th:no-conj}: $\md( S,
    S_1)= \md( S, S_2)= 1$, but $\md( S, S_1\wedge S_2)= \infty$.}
  \label{weightedmodal.fi:conjunct-counterex}
\end{figure}
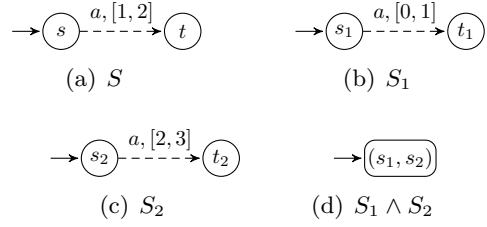

\begin{proof}
  We follow the same strategy as in the proof of
  Theorem~\ref{weightedmodal.th:no-dethull}. One can define a partial conjunction
  operator $\wedge'$ defined for WMTS which satisfies
  Properties~$(\ref{weightedmodal.th:no-conj}.\ref{weightedmodal.en:no-conj:lb})$
  and~$(\ref{weightedmodal.th:no-conj}.\ref{weightedmodal.en:no-conj:def})$ as follows: For
  deterministic WMTS $S_1$ and $S_2$, $S_1 \wedge' S_2 =
  \rho^B(S_1 \times S_2,(s_1^0,s_2^0),\mmayto,\mmustto)$ where
  the transition relations $\mmayto$ and $\mmustto$ and the set
  $B \subseteq S_1 \times S_2$ of inconsistent states are
  defined by the following rules:
  \begin{gather*}
    \frac{s_1 \mustto{k_1} s_1'\quad s_2 \mayto{k_2} s_2' \quad k_1
      \owedge k_2 \text{ defined}}{(s,t)\mustto{k_1 \owedge k_2}
      (s_1',s_2')}%
    \qquad \frac{s_1 \mayto{ k_1 } s_1'\quad s_2 \mustto{k_2 } s_2'
      \quad k_1 \owedge k_2 \text{
        defined}}{(s_1,s_2)\mustto{ k_1 \owedge k_2}(s_1',s_2')} \\
    \frac{s_1 \mayto{ k_1 } s_1'\quad s_2 \mayto{ k_2} s_2' \quad k_1
      \owedge k_2 \text{ defined}}{(s_1,s_2)\mayto{ k_1\owedge
      k_2}(s_1',s_2')} \\
    \frac{s_1 \mustto{k_1} \quad \big(k_1\owedge k_2 \text{ undefined for
        any $k_2$ such that } s_2\mayto{k_2}\big)} {(s_1,s_2)\in
      B} \\
    \frac{s_2\mustto{k_2} \quad \big(k_1\owedge k_2 \text{ undefined for
        any $k_1$ such that } s_1\mayto{k_1} \big)}{(s_1,s_2) \in
      B}
  \end{gather*}
  
  Using these properties, one can see that for all deterministic WMTS
  $S_1$ and $S_2$, $S_1\wedge S_2\mr S_1\wedge' S_2$ and $S_1\wedge'
  S_2\mr S_1\wedge S_2$.  The WMTS depicted in
  Figure~\ref{weightedmodal.fi:conjunct-counterex} then show that
  Property~$(\ref{weightedmodal.th:no-conj}.\ref{weightedmodal.en:no-conj:glb})$ cannot hold: here,
  $\md( S, S_1)= \md( S, S_2)= 1$, but $\md( S, S_1\wedge S_2)=
  \infty$. \qed
\end{proof}

The counterexamples used in the proofs of Theorems~\ref{weightedmodal.th:no-dethull}
and~\ref{weightedmodal.th:no-conj} are quite general and apply to a large class of
distances, rather than only to the accumulating distance discussed in
this paper.  Hence it can be argued that what we have exposed here is a
fundamental limitation of any quantitative approach to modal
specifications.

\section{Structural Composition and Quotient}
\label{weightedmodal.se:parcomp}

In this section we show that in our quantitative setting, notions of
structural composition and quotient can be defined which obey the
properties expected of such operations.  In particular, structural
composition satisfies independent
implementability~\cite{inbook/natosec/AlfaroH95}, hence the refinement
distance between structural composites can be bounded by the distances
between their respective components.

First we define partial synchronization operators $\oplus$ and $\ominus$
on specification labels which will be used for synchronizing
transitions.  We let $( a_1, I_1)\oplus( a_2, I_2)$ and $( a_1,
I_1)\ominus( a_2, I_2)$ be undefined if $a_1\ne a_2 $, and otherwise
\begin{align*}
  \big( a,[ x_1, y_1]\big)\oplus\big( a,[ x_2, y_2]\big) &= \big( a,[
  x_1+ x_2, y_1+ y_2]\big)\,, \\
  \big( a,[ x_1, y_1]\big)\ominus\big( a,[ x_2, y_2]\big) &=
    \begin{cases}
      \text{undefined} &\quad\text{if } x_1- x_2> y_1- y_2\,, \\
      \big( a,[ x_1- x_2, y_1- y_2]\big) &\quad\text{if } x_1- x_2\le
      y_1- y_2\,.
    \end{cases}
\end{align*}
Note that we use CSP-style synchronization, but other types of
synchronization can easily be defined.  Also, defining $\oplus$ to add
intervals (and $\ominus$ to subtract them) is only one particular
choice; depending on the application, one can also \eg~let $\oplus$ be
intersection of intervals or some other operation.  It is not difficult
to see that these alternative synchronization operators would lead to
properties similar to those we show here.

\begin{definition}
  \label{weightedmodal.de:comp-quot}
  Let $S_1$ and $S_2$ be WMTS.  The \emph{structural composition} of $S_1$
  and $S_2$ is $S_1\| S_2=\big( S_1\times S_2,( s_1^0, s_2^0), \Spec,
  \mmayto, \mmustto\big)$ with transitions given as follows:
    \begin{gather*}
      \dfrac{ s_1\mayto{ k_1}_1 t_1 \quad s_2\mayto{ k_2}_2 t_2 \quad
        k_1\oplus k_2 \text{ def.}}{( s_1, s_2)\mayto{ k_1\oplus
        k_2}( t_1, t_2)}%
      \qquad \dfrac{ s_1\mustto{ k_1}_1 t_1 \quad s_2\mustto{ k_2}_2 t_2
        \quad k_1\oplus k_2 \text{ def.}} {( s_1, s_2)\mustto{
        k_1\oplus k_2}( t_1, t_2)}
    \end{gather*}
    The \emph{quotient} of $S_1$ by $S_2$ is $S_1/ S_2=
    \rho^{B}\big( S_1\times S_2\cup\{ u\},( s_1^0, s_2^0), \Spec,
    \mmayto, \mmustto\big)$ with transitions and the set of inconsistent
    states given as follows:
    \begin{gather*}
    \dfrac{%
      s_1\mayto{ k_1}_1 t_1 \quad s_2\mayto{ k_2}_2 t_2 \quad
      k_1\ominus k_2 \text{ def.}}{%
      ( s_1, s_2)\mayto{ k_1\ominus k_2}( t_1, t_2)} \qquad
    \dfrac{%
      s_1\mustto{ k_1}_1 t_1 \quad s_2\mustto{ k_2}_2 t_2 \quad
      k_1\ominus k_2 \text{ def.}}{%
      ( s_1, s_2)\mustto{ k_1\ominus k_2}( t_1, t_2)} \\
    \dfrac{%
      s_1\mustto{ k_1}_1 t_1 \quad \forall s_2\mustto{ k_2}_2 t_2:
      k_1\ominus k_2 \text{ undefined}}{%
      ( s_1, s_2) \in B} \\
    \dfrac{%
      k\in \Spec \quad \forall s_2\mayto{ k_2}_2 t_2: k\oplus k_2
      \text{ undefined}}{%
      ( s_1, s_2)\mayto{ k} u} \qquad \dfrac{%
      k\in \Spec}{%
      u\mayto{ k} u}
  \end{gather*}
\end{definition}

Note that during the quotient construction inconsistent states can
arise which are then recursively removed using the pruning operator
$\rho$, see Definition~\ref{weightedmodal.de:prune}.  After a
technical lemma, the next theorem shows that structural composition is
well-behaved with respect to modal refinement distance in the sense
that the distance between the composed systems is bounded by the
distances of the individual systems.  Note also the special case in
the theorem of $S_1\mr S_2$ and $S_3\mr S_4$ implying
$S_1\| S_3\mr S_2\| S_4$.

\begin{lemma}
  \label{weightedmodal.le:synch-cont}
  For $k_1, k_2, k_3, k_4\in \Spec$ with $k_1\oplus k_3$ and $k_2\oplus
  k_4$ defined, we have $d_\Spec( k_1\oplus k_3, k_2\oplus k_4)\le
  d_\Spec( k_1, k_2)+ d_\Spec( k_3, k_4)$.
\end{lemma}

\begin{proof}
  Let 
  $k_i=\big( a,[ x_i, y_i]\big)$ for all $i$.  We have
  \begin{align*}
    d_\Spec( k_1, k_2)+{} &d_\Spec( k_3, k_4) \\
    &= \max( x_2 - x_1, y_1 -
    y_2,0)+ \max( x_4 - x_3, y_3 - y_4,0) \\
    &\ge \max\big(( x_2 - x_1)+( x_4 - x_3),(
    y_1 - y_2)+( y_3 - y_4), 0\big)\\
    &= \max\big(( x_2+ x_4)-(
    x_1+ x_3),( y_1+ y_3)-( y_2+ y_4),0\big) \\
    &= d_\Spec( k_1\oplus k_3, k_2\oplus k_4). \quad \qed
  \end{align*}
\end{proof}

\begin{theorem}[Independent implementability]
  \label{weightedmodal.th:indepimp}
  For WMTS $S_1$, $S_2$, $S_3$, $S_4$ we have $\md\big( S_1\| S_3,
  S_2\| S_4\big)\le \md( S_1, S_2)+ \md( S_3, S_4)$.
\end{theorem}

\begin{proof}
  If $\md( S_1, S_2)= \infty$ or $\md( S_3, S_4)= \infty$, we have
  nothing to prove.  Otherwise, let
  $\cal R^1=\{ R^1_\epsilon\subseteq S_1\times S_2\mid \epsilon\ge
  0\}$,
  $\cal R^2=\{ R^2_\epsilon\subseteq S_3\times S_4\mid \epsilon\ge
  0\}$ be witnesses for $\md( S_1, S_2)$ and $\md( S_3, S_4)$,
  respectively; hence
  $( s_1^0, s_2^0)\in R^1_{ \md( S_1, S_2)}\in \cal R^1$ and
  $( s_3^0, s_4^0)\in R^2_{ \md( S_3, S_4)}\in \cal R^2$.  Define
  \begin{multline*}
    R_\epsilon=\big\{\big(( s_1, s_3),( s_2, s_4)\big)\in S_1\times
    S_3\times S_2\times S_4\bigmid \\
    ( s_1, s_2)\in R^1_{ \epsilon_1}\in R^1,( s_3, s_4)\in R^2_{
      \epsilon_2}\in R^2, \epsilon_1+ \epsilon_2\le\epsilon \big\}
  \end{multline*}
  for all $\epsilon\ge 0$ and let
  $\cal R=\{ R_\epsilon\mid \epsilon\ge 0\}$.  We show that $\cal R$
  witnesses
  $\md\big( S_1\| S_3, S_2\| S_4\big)\le \md( S_1, S_2)+ \md( S_3,
  S_4).$

  We have
  $\big(( s_1^0, s_3^0),( s_2^0, s_4^0)\big)\in R_{ \md( S_1, S_2)+
    \md( S_3, S_4)}\in \cal R$.  Now let
  \begin{equation*}
    \big(( s_1, s_3),( s_2, s_4)\big)\in R_\epsilon\in \cal R
  \end{equation*}
  for some $\epsilon$, then $( s_1, s_2)\in R^1_{ \epsilon_1}\in \cal R^1$
  and $( s_3, s_4)\in R^2_{ \epsilon_2}\in \cal R^2$ for some
  $\epsilon_1+ \epsilon_2\le \epsilon$.

  Assume $( s_1, s_3)\mayto{ k_1\oplus k_3}( t_1, t_3)$, then
  $s_1\mayto{ k_1}_1 t_1$ and $s_3\mayto{ k_3}_3 t_3$.  By
  $( s_1, s_2)\in R^1_{ \epsilon_1}\in \cal R^1$, we have
  $s_2\mayto{ k_2}_2 t_2$ with $d_\Spec( k_1, k_2)\le \epsilon_1$ and
  $( t_1, t_2)\in R^1_{ \epsilon_1'}\in \cal R^1$ for some
  $\epsilon_1'\le \lambda^{ -1}\big( \epsilon_1- d_\Spec( k_1,
  k_2)\big)$; similarly, $s_4\mayto{ k_4}_4 t_4$ with
  $d_\Spec( k_3, k_4)\le \epsilon_2$ and
  $( t_3, t_4)\in R^2_{ \epsilon_2'}\in \cal R^2$ for some
  $\epsilon_2'\le \lambda^{ -1}\big( \epsilon_2- d_\Spec( k_3,
  k_4)\big)$.  Let $\epsilon'= \epsilon_1'+ \epsilon_2'$, then the sum
  $k_2\oplus k_4$ is defined, and
  \begin{align*}
    \epsilon' &\le \lambda^{ -1}\big( \epsilon_1+ \epsilon_2-( d_\Spec( k_1,
    k_2)+ d_\Spec( k_3, k_4))\big) \\
    &\le \lambda^{ -1}\big( \epsilon- d_\Spec( k_1\oplus k_3, k_2\oplus
    k_4)\big)
  \end{align*}
  by Lemma~\ref{weightedmodal.le:synch-cont}.  We have
  $( s_2, s_4)\mayto{ k_2\oplus k_4}( t_2, t_4)$,
  $d_\Spec( k_1\oplus k_3, k_2\oplus k_4)\le \epsilon_1+ \epsilon_2\le
  \epsilon$ again by Lemma~\ref{weightedmodal.le:synch-cont}, and
  $\big(( t_1, t_3),( t_2, t_4)\big)\in R_{ \epsilon'}\in \cal R$.
  The reverse direction, starting with a transition
  $( s_2, s_4)\mustto{ k_2\oplus k_4}( t_2, t_4)$, is similar. \qed
\end{proof}

Again after a technical lemma, the next theorem expresses the fact that
quotient is a partial inverse to structural composition.  Intuitively,
the theorem shows that the quotient $S_1/ S_2$ is maximal among
all WMTS $S_3$ with respect to any distance $S_2\| S_3\mr^\epsilon
S_1$; note the special case of $S_3\mr S_1/ S_2$ iff
$S_2\| S_3\mr S_1$.

\begin{lemma}
  \label{weightedmodal.le:plusminus}
  If $k_1, k_2, k_3\in \Spec$ are such that $k_1\ominus k_2$ and $k_2\oplus
  k_3$ are defined, then $d_\Spec( k_3, k_1\ominus k_2)= d_\Spec(
  k_2\oplus k_3, k_1)$.
\end{lemma}

\begin{proof}
  We can write $k_i=\big( a,[ x_i, y_i]\big)$ for some $a\in \Sigma$.
  Then
  \begin{align*}
    d_\Spec( k_3, k_1\ominus k_2) &= \max\big(( x_1- x_2) - x_3,
    y_3 - ( y_1- y_2), 0\big) \\
    &= \left\{
      \begin{array}{cl}
        x_1- x_2- x_3 &\quad\text{if}\quad
        \begin{aligned}[t]
          x_1- x_2- x_3 &\ge 0, \\
          x_1- x_2- x_3 &\ge y_3- y_1+ y_2;
        \end{aligned} \\
        y_3- y_1+ y_2 &\quad\text{if}\quad
        \begin{aligned}[t]
          y_3- y_1+ y_2 &\ge 0, \\
          y_3- y_1+ y_2 &\ge x_1- x_2- x_3;
        \end{aligned} \\
        0 &\quad\text{if}\quad
        \begin{aligned}[t]
          x_1- x_2- x_3 &\le 0, \\
          y_3- y_1+ y_2 &\le 0.
        \end{aligned}
      \end{array}
    \right.
  \end{align*}
  Similarly,
  \begin{align*}
    d_\Spec( k_2\oplus k_3, k_1) &= \max\big( x_1 - ( x_2+ x_3),( y_2+
    y_3)- y_1, 0\big) \\
    &= \left\{
      \begin{array}{cl}
        x_1- x_2- x_3 &\quad\text{if}\quad
        \begin{aligned}[t]
          x_1- x_2- x_3 &\ge 0, \\
          x_1- x_2- x_3 &\ge y_2+ y_3- y_1 ;
        \end{aligned} \\
        y_2+ y_3- y_1 &\quad\text{if}\quad
        \begin{aligned}[t]
          y_2+ y_3- y_1 &\ge 0, \\
          y_2+ y_3- y_1 &\ge x_1- x_2- x_3;
        \end{aligned} \\
        0 &\quad\text{if}\quad
        \begin{aligned}[t]
          x_1- x_2- x_3 &\le 0, \\
          y_2+ y_3- y_1 &\le 0.
        \end{aligned}
      \end{array}
    \right.
  \end{align*}
  \mbox{}
\end{proof}

\begin{theorem}[Soundness and maximality of quotient]
  \label{weightedmodal.th:soundmaxquot}
  Let $S_1$, $S_2$ and $S_3$ be locally consistent WMTS such that
  $S_2$ is deterministic and $S_1/ S_2$ is defined.  If
  $\md( S_3, S_1/ S_2)< \infty$, then $\md( S_3,
  S_1/ S_2)= \md( S_2\| S_3, S_1)$.
\end{theorem}

\begin{proof}
  To avoid confusion, we write $\mmayto_/$ and
  $\mmustto_/$ for transitions in $S_1/ S_2$ and
  $\mmayto_\|$ and $\mmustto_\|$ for transitions in $S_2\| S_3$.  The
  inequality $\md( S_3, S_1/ S_2)\ge \md( S_2\| S_3, S_1)$
  is trivial if $\md( S_2\| S_3, S_1)= \infty$, so assume the
  opposite and let $\cal R^1=\big\{ R^1_\epsilon\subseteq S_3\times\big(
  S_1\times S_2\cup\{ u\}\big)\bigmid \epsilon\ge 0\big\}$ be a
  witness for $\md( S_3, S_1/ S_2)$.  Define
  $R^2_\epsilon=\big\{\big(( s_2, s_3), s_1\big)\bigmid\big(s_3,( s_1,
  s_2)\big)\in R^1_\epsilon\big\}\subseteq S_2\times S_3\times S_1$
  for all $\epsilon\ge 0$, and let $\cal R^2=\{ R^2_\epsilon\mid
  \epsilon\ge 0\}$.  Certainly $\big(( s_2^0, s_3^0), s_1^0\big)\in
  R^2_{ \md( S_3, S_1/ S_2)}\in \cal R^2$, so let now $\big(( s_2,
  s_3), s_1\big)\in R^2_\epsilon\in \cal R^2$ for some $\epsilon\ge 0$.

  Assume $( s_2, s_3)\mayto{ k_2\oplus k_3}_\|( t_2, t_3)$, then also
  $s_2\mayto{ k_2}_2 t_2$ and $s_3\mayto{ k_3}_3 t_3$.  We have
  $\big( s_3,( s_1, s_2)\big)\in R^1_\epsilon$, so there is
  $( s_1, s_2)\mayto{ k_1\ominus k_2'}_/( t_1, t_2')$ for which
  $d_\Spec( k_3, k_1\ominus k_2')= d_\Spec( k_2'\oplus k_3, k_1)\le
  \epsilon$ and such that
  $\big( t_3,( t_1, t_2')\big)\in R^1_{ \epsilon'}\in \cal R^1$, hence
  $\big(( t_2', t_3), t_1\big)\in R^2_{ \epsilon'}\in \cal R^2$, for some
  $\epsilon'\le \lambda^{ -1}\big( \epsilon- d_\Spec( k_2'\oplus k_3,
  k_1)\big)$.  By definition of quotient we must have
  $s_1\mayto{ k_1}_1 t_1$ and $s_2\mayto{ k_2'}_2 t_2'$, and by
  determinism of $S_2$, $k_2'= k_2$ and $t_2'= t_2$.

  Assume $s_1\mustto{ k_1}_1 t_1$.  We must have a transition
  $s_2\mustto{ k_2}_2 t_2$ for which $k_1\ominus k_2$ is defined.
  Hence $( s_1, s_2)\mustto{ k_1\ominus k_2}_/( t_1, t_2)$.
  This in turn implies that there is $s_3\mustto{ k_3}_3 t_3$ for
  which
  $d_\Spec( k_3, k_1\ominus k_2)= d_\Spec( k_2\oplus k_3, k_1)\le
  \epsilon$ and such that
  $\big( t_3,( t_1, t_2)\big)\in R^1_{ \epsilon'}\in \cal R^1$, hence
  $\big(( t_2, t_3), t_1\big)\in R^2_{ \epsilon'}\in \cal R^2$, for some
  $\epsilon'\le \lambda^{ -1}\big( \epsilon- d_\Spec( k_2\oplus k_3,
  k_1)\big)$, and by definition of parallel composition,
  $( s_2, s_3)\mustto{ k_2\oplus k_3}_\|( t_2, t_3)$.

  To show that $\md( S_3, S_1/ S_2)\le \md( S_2\| S_3, S_1)$,
  let
  $\cal R^2=\{ R^2_\epsilon\subseteq S_2\times S_3\times S_1\mid
  \epsilon\ge 0\}$ be a witness for $\md( S_2\| S_3, S_1)$, define
  $R^1_\epsilon=\big\{\big( s_3,( s_1, s_2)\big)\bigmid\big(( s_2,
  s_3), s_1\big)\in R^2_\epsilon\big\}\cup\big\{( s_3, u)\bigmid
  s_3\in S_3\big\}$ for all $\epsilon\ge 0$, and let
  $\cal R^1=\{ R^1_\epsilon\mid \epsilon\ge 0\}$, then
  $\big( s_3^0,( s_1^0, s_2^0)\big)\in R^1_{ \md( S_2\| S_3, S_1)}\in
  \cal R^1$.

  For any $( s_3, u)\in R^1_\epsilon$ for some $\epsilon\ge 0$, any
  transition $s_3\mayto{ k_3}_3 t_3$ can be matched by
  $u\mayto{ k_3}_/ u$, and then $( t_3, u)\in R^1_0$.  Let now
  $\big( s_3,( s_1, s_2)\big)\in R^1_\epsilon\in \cal R^1$ for some
  $\epsilon\ge 0$, and assume $s_3\mayto{ k_3}_3 t_3$.  If
  $k_2\oplus k_3$ is undefined for all transitions
  $s_2\mayto{ k_2}_2 t_2$, then by definition
  $( s_1, s_2)\mayto{ k_3} u$, and again $( t_3, u)\in R^1_0$.  If
  there is a transition $s_2\mayto{ k_2}_2 t_2$ such that
  $k_2\oplus k_3$ is defined, then also
  $( s_2, s_3)\mayto{ k_2\oplus k_3}_\|( t_2, t_3)$.  Hence we have
  $s_1\mayto{ k_1}_1 t_1$ with
  $d_\Spec( k_2\oplus k_3, k_1)\le \epsilon$, implying that
  $( s_1, s_2)\mayto{ k_1\ominus k_2}_/( t_1, t_2)$. Hence
  $d_\Spec( k_3, k_1\ominus k_2)= d_\Spec( k_2\oplus k_3, k_1)\le
  \epsilon$.  Also,
  $\big(( t_2, t_3), t_1\big)\in R^2_{ \epsilon'}\in \cal R^2$, hence
  $\big( t_3,( t_1, t_2)\big)\in R^1_{ \epsilon'}\in \cal R^1$, for some
  $\epsilon'\le \lambda^{ -1}\big( \epsilon- d_\Spec( k_3, k_1\ominus
  k_2)\big)$.

  Assume $( s_1, s_2)\mustto{ k_1\ominus k_2}_/( t_1, t_2)$,
  hence we have $s_1\mustto{ k_1}_1 t_1$ and $s_2\mustto{ k_2}_2
  t_2$.  It follows that $( s_2, s_3)\mustto{ k_2'\oplus k_3}_\|(
  t_2', t_3)$ with $d_\Spec( k_2'\oplus k_3, k_1)= d_\Spec( k_3,
  k_1\ominus k_2')\le \epsilon$ and such that $\big(( t_2', t_3),
  t_1\big)\in R^2_{ \epsilon'}\in R^2$, hence $\big( t_3,( t_1,
  t_2')\big)\in R^1_{ \epsilon'}\in R^1$, for some $\epsilon'\le
  \lambda^{ -1}\big( \epsilon- d_\Spec( k_3, k_1\ominus k_2')\big)$.  By
  definition of parallel composition we must have $s_2\mustto{ k_2'}_2
  t_2'$ and $s_3\mustto{ k_3}_3 t_3$, and by determinism of $S_2$,
  $k_2'= k_2$ and $t_2'= t_2$. \qed
\end{proof}

\begin{figure}[tpb]
  \centering
  \subbottom[$S_1$]{
    \begin{tikzpicture}[->,>=stealth',shorten >=1pt,auto,node
      distance=2.0cm,initial text=,scale=0.8,transform shape]
      \tikzstyle{every node}=[font=\small] \tikzstyle{every
        state}=[fill=white,shape=circle,inner sep=.5mm,minimum size=6mm]
      \node[initial,state] (s1) at (0,0) {$s_1$};
      \node[state] (s2) at (2,0) {$t_1$};
      \path (s1) edge [densely dashed] node [above] {$a,[ 0, 0]$} (s2);
    \end{tikzpicture}}
  \hspace{2mm}
  \subbottom[$S_2$]{
    \begin{tikzpicture}[->,>=stealth',shorten >=1pt,auto,node
      distance=2.0cm,initial text=,scale=0.8,transform shape]
      \tikzstyle{every node}=[font=\small] \tikzstyle{every
        state}=[fill=white,shape=circle,inner sep=.5mm,minimum size=6mm]
      \node[initial,state] (s1) at (0,0) {$s_2$};
      \node[state] (s2) at (2,0) {$t_2$};
      \path (s1) edge [densely dashed] node [above] {$a,[ 0, 1]$} (s2);
    \end{tikzpicture}}
  \hspace{2mm}
  \subbottom[$S_3$]{
    \begin{tikzpicture}[->,>=stealth',shorten >=1pt,auto,node
      distance=2.0cm,initial text=,scale=0.8,transform shape]
      \tikzstyle{every node}=[font=\small] \tikzstyle{every
        state}=[fill=white,shape=circle,inner sep=.5mm,minimum size=6mm]
      \node[initial,state] (s1) at (0,0) {$s_3$};
      \node[state] (s2) at (2,0) {$t_3$};
      \path (s1) edge [densely dashed] node [above] {$a,[ 0, 0]$} (s2);
    \end{tikzpicture}}
  \hspace{2mm}
  \subbottom[$S_2\| S_3$]{
    \begin{tikzpicture}[->,>=stealth',shorten >=1pt,auto,node
      distance=2.0cm,initial text=,scale=0.8,transform shape]
      \tikzstyle{every node}=[font=\small] \tikzstyle{every
        state}=[fill=white,shape=circle,inner sep=.5mm,minimum size=6mm]
      \node[initial,state,shape=rectangle,rounded corners] (s1) at
      (0,0) {$( s_2, s_3)$}; 
      \node[state,shape=rectangle,rounded corners] (s2) at (3,0)
      {$( t_2, t_3)$};
      \path (s1) edge [densely dashed] node [above] {$a,[ 0, 1]$} (s2);
    \end{tikzpicture}}
  \hspace{2mm}
  \subbottom[$S_1/ S_2$]{
    \begin{tikzpicture}[->,>=stealth',shorten >=1pt,auto,node
      distance=2.0cm,initial text=,scale=0.8,transform shape]
      \tikzstyle{every node}=[font=\small] \tikzstyle{every
        state}=[fill=white,shape=circle,inner sep=.5mm,minimum size=6mm]
      \node[initial,state,shape=rectangle,rounded corners] (s1) at
      (0,0) {$( s_1, s_2)$};
    \end{tikzpicture}}
  \caption{WMTS for which $\md( S_2\| S_3, S_1)\ne \md( S_3,
    S_1/ S_2)= \infty$.}
  \label{weightedmodal.fi:quotient-counterex}
\end{figure}

The example of Figure~\ref{weightedmodal.fi:quotient-counterex} shows
that the condition $\md( S_3, S_1/ S_2)< \infty$ in
Theorem~\ref{weightedmodal.th:soundmaxquot} is necessary.  Here
$\md( S_2\| S_3, S_1)= 1$, but $\md( S_3, S_1/ S_2)= \infty$ because
of inconsistency between the transitions $s_1\mayto{ a,[ 0, 0]}_1 t_1$
and $s_2\mayto{ a,[ 0, 1]}_2 t_2$ for which $k_1\ominus k_2$ is
defined.

\medskip As a practical application, we notice that \emph{relaxation} as
defined in Section~\ref{weightedmodal.se:relax} can be useful when computing
quotients.  The quotient construction in Definition~\ref{weightedmodal.de:comp-quot}
introduces inconsistent states (which afterwards are pruned) whenever
there is a \textit{must} transition $s_1\mustto{ k_1}_1 s_1'$ such that
$k_1\ominus k_2$ is undefined for all transitions $s_2\mustto{ k_2}_2
s_2'$.  Looking at the definition of $\ominus$, we see that this is the
case if $k_1=( a_1,[ x_1, y_1])$ and $k_2=( a_2,[ x_2, y_2])$ are such
that $a_1\ne a_2$ or $x_1- x_2> y_1- y_2$.  In the first case, the
inconsistency is of a \emph{structural} nature and cannot be dealt with;
but in the second case, it may be avoided by \emph{enlarging} $k_1$:
decreasing $x_1$ or increasing $y_1$ so that now, $x_1- x_2\le y_1-
y_2$.

Enlarging quantitative constraints is exactly the intuition of
relaxation, thus in practical cases where we get a quotient
$S_1/ S_2$ which is ``too inconsistent'', we may be able to
solve this problem by constructing a suitable $\epsilon$-relaxation
$S_1'$ of $S_1$.  Theorems~\ref{weightedmodal.th:indepimp} and~\ref{weightedmodal.th:soundmaxquot}
can then be used to ensure that also $S_1'/ S_2$ is a
relaxation of $S_1/ S_2$.

\section{Conclusion}

We have shown in this chapter that within the quantitative
specification framework of weighted modal transition systems,
refinement and implementation distances provide a useful tool for
robust compositional reasoning.  Note that these distances permit us
not only to reason about differences between implementations and from
implementations to specifications, but they also provide a means by
which we can compare specifications directly at the abstract level.

We have shown that for some of the ingredients of our specification
theory, namely structural composition and quotient, our formalism is a
conservative extension of the standard Boolean notions.  We have also
noted however, that for determinization and logical conjunction, the
properties of the Boolean notions are not preserved, and that this
seems to be a fundamental limitation of any reasonable quantitative
specification theory.  We will have more to say about this in the next
chapter.

\chapter[General Quantitative Specification Theories with Modal
Transition Systems][General Quantitative Specification
Theories]{General Quantitative Specification Theories with Modal
  Transition Systems\footnote{This chapter is based on the journal
    paper~\cite{DBLP:journals/acta/FahrenbergL14} published in Acta
    Informatica.}}
\label{ch:wm2}

This chapter combines the work of the two previous chapters.  It uses
the general theory of linear and branching distances developed in
Chapter~\ref{ch:qltbt} to introduce general refinement distances
between \emph{structured modal transition systems}.  It then proceeds
to consider quantitative properties of structural composition,
quotient, and conjunction, and finishes with a logical
characterization of quantitative refinement using Hennessy-Milner
logic.

\section{Structured Modal Transition Systems}
\label{wm2.se:smts}

We work with a poset $\Spec$ of \emph{specification labels} with a
partial order $\labpre$
and denote by $\Spec^\infty= \Spec^*\cup \Spec^\omega$ the set of finite and
infinite traces over $\Spec$.  In applications, $\Spec$ may be used to
model data about the behavior of a system; for specifications this may
be considered as legal parameters of operation, whereas for
implementations it may be thought of as observed
information.

The partial order $\labpre$ is meant to model
\emph{refinement} of data; if $k\labpre \ell$, then $k$ is
more refined (leaves fewer choices) than $\ell$.
The set $\Imp= \{ k\in \Spec\mid
k'\labpre k\Longrightarrow k'= k\}$ is called the set of
\emph{implementation labels}; these are the data which cannot be refined
further.  We let $\llbracket k\rrbracket=\{ k'\in \Imp\mid
k'\labpre k\}$ and assume that
$\llbracket k\rrbracket\ne \emptyset$ for all $k\in \Spec$.

When $k\not\labpre \ell$, we want to be able to quantify the
impact of this difference in data on the systems in question, thus
circumventing the fragility of the theory.  To this end, we introduce a
general notion of distance on sequences of data following
the approach laid out in Chapter~\ref{ch:qltbt}.

\subsection{Trace distances}

In order to build a framework for specification distances which is
general enough to cover the distances commonly used, we introduce a
notion of abstract trace distance which factors through a lattice on
which it has a recursive characterization.  We will show in
Section~\ref{wm2.se:examples} that this indeed covers the common
scenarios; see also Section~\ref{qltbt.se:examples_rec}.

Let $M$ be an arbitrary set and $\LL=( \Realnn\cup\{ \infty\})^M$ the
set of functions from $M$ to the extended non-negative real line.  Then
$\LL$ is a complete lattice with partial order
$\mathord{\sqsubseteq_\LL}$ given by $\alpha\sqsubseteq_\LL \beta$ if
and only if $\alpha( x)\le \beta( x)$ for all $x\in M$, and with an
addition $\oplus_\LL$ given by $( \alpha\oplus_\LL \beta)( x)= \alpha(
x)+ \beta( x)$.  The bottom element of $\LL$ is also the zero of
$\oplus_\LL$ and given by $\bot_\LL( x)= 0$, and the top element is
$\top_\LL( x)= \infty$.  We also define a metric on $\LL$ by $d_\LL(
\alpha, \beta)= \sup_{ x\in M}| \alpha( x)- \beta( x)|$.

Intuitively, the lattice $\LL$ serves as a memory for more elaborate
trace distances such as for example the limit-average distance, see
Section~\ref{wm2.se:examples}.  For simpler distances, it will suffice
to let $M=\{ *\}$ be the one-point set and thus
$\LL= \Realnn\cup\{ \infty\}$.  We extend the notions of hemimetrics,
pseudometrics and metrics to mappings $d: X\times X\to \LL$, by
replacing in their defining properties $0$ by $\bot_\LL$ and $+$ by
$\oplus_\LL$.

Let $d: \Imp\times \Imp\to \LL$ be a hemimetric on implementation
labels.  We extend $d$ to $\Spec$ by $d( k, \ell)= \sup_{ m\in
  \llbracket k\rrbracket} \inf_{ n\in \llbracket \ell\rrbracket} d( m,
n)$.  Hence also this distance is \emph{asymmetric}; the intuition is
that any label in $\llbracket k\rrbracket$ has to be matched as good as
possible in $\llbracket \ell\rrbracket$.  Note that this is the
Hausdorff hemimetric associated with $d$ on implementation labels.

We will assume given an abstract \emph{trace distance} $\tdl:
\Spec^\infty\times \Spec^\infty\to \LL$ which is a hemimetric and has a
recursive expression using a \emph{distance iterator} function $F:
\Imp\times \Imp\times \LL\to \LL$, see below.  This will allow us to
recover many of the system distances found in the literature, while
preserving key results.  We will need to assume that $F$ satisfies the
following properties:
\begin{enumerate}[(1)]
\item $F$ is \emph{continuous} in the first two coordinates: $F( \cdot,
  n, \alpha)$ and $F( m, \cdot, \alpha)$ are continuous functions
  $\Imp\to \LL$ for all $\alpha\in \LL$.
\item $F$ is \emph{monotone} in the third coordinate: $F( m, n, \cdot):
  \LL\to \LL$ is monotone for all $m, n\in \Imp$.
\item $F$ extends $d$: for all $m, n\in \Imp$, $F( m, n, \bot_\LL)= d(
  m, n)$.
\item \label{wm2.en:F.indid} Indiscernibility of identicals: $F( m, m,
  \alpha)= \alpha$ for all $m\in \Imp$.
\item \label{wm2.en:F.triangle} An extended triangle inequality: for
  all $m, n, o\in \Imp$ and $\alpha, \beta\in \LL$,
  $F( m, n, \alpha)\oplus_\LL F( n, o, \beta)\sqsupseteq_\LL F( m, o,
  \alpha\oplus_\LL \beta)$.
\end{enumerate}
Note how the last two axioms are a generalization of the standard
axioms for hemimetrics.

We extend $F$ to specification labels by defining
\begin{equation*}
  F( k, \ell, \alpha)= \adjustlimits \sup_{ m\in \llbracket k\rrbracket}
  \inf_{ n\in \llbracket \ell\rrbracket} F( m, n, \alpha)\,.
\end{equation*}
Then also the extended $F: \Spec\times \Spec\times \LL\to \LL$ is
continuous in the first two and monotone in the third coordinates.
Additionally, we assume that sets of implementation labels are
\emph{closed} with respect to $F$ in the sense that for all $k, \ell\in
\Spec$ and $\alpha\in \LL$ with $F( k, \ell, \alpha)\ne \top_\LL$, there
are $m\in \llbracket k\rrbracket$, $n\in \llbracket \ell\rrbracket$ with
$F( m, \ell, \alpha)= F( k, n, \alpha)= F( k, \ell, \alpha)$.  Note that
this implies that the sets $\llbracket k\rrbracket$ are closed under the
hemimetric $d$ on $\Spec$.

Axioms~\eqref{wm2.en:F.indid} and~\eqref{wm2.en:F.triangle} for $F$
above now imply that for the extension, the following hold:
\begin{enumerate}[(1$'$)]
  \setcounter{enumi}3
\item \label{wm2.en:F.indid'} For all $k, \ell\in \Spec$ with
  $k\labpre \ell$ and all $\alpha\in \LL$, $F( k, \ell,
  \alpha)= \alpha$.
\item \label{wm2.eq:triangleF} For all $k, \ell, m\in \Spec$ and
  $\alpha, \beta\in \LL$,
  \begin{equation*}
    F( k, \ell, \alpha)\oplus_\LL F( \ell, m, \beta)\sqsupseteq_\LL F(
    k, m, \alpha\oplus_\LL \beta)\,.
  \end{equation*}
\end{enumerate}

Let $\emptyseq\in \Spec^\infty$ denote the empty sequence, and for any
sequence $\sigma\in \Spec^\infty$, denote by $\sigma_0$ its first element
and by $\sigma^1$ the tail of the sequence with the first element
removed.  We assume that $\tdl$ has a recursive characterization, using
$F$, as follows:
\begin{equation}
  \label{wm2.eq:trdist}
  \tdl( \sigma, \tau)=
  \begin{cases}
    F( \sigma_0, \tau_0, \tdl( \sigma^1, \tau^1)) &\text{if }
    \sigma, \tau\ne \emptyseq, \\
    \top_\LL &\text{if } \sigma= \emptyseq, \tau\ne \emptyseq \text{ or }
    \sigma\ne \emptyseq, \tau= \emptyseq, \\
    \bot_\LL &\text{if } \sigma= \tau= \emptyseq.
  \end{cases}
\end{equation}

We remark that a recursive characterization such as the one above is
quite natural.  Not only does it cover all commonly used trace distances
(see the examples in the next section), but recursion is central to
computing, and any trace distance without a recursive characterization
would strike us as being quite artificial.  It is precisely this
recursive characterization which allows us to lift the trace distance to
\emph{states} of specifications in Definition~\ref{wm2.de:modrefdist} below,
see also Chapter~\ref{ch:qltbt}.

In applications (see below), the lattice $\LL$ comes equipped with a
homomorphism $\eval: \LL\to \Realnn\cup\{ \infty\}$ for which
$\eval( \bot_\LL)= 0$.  The actual trace distance of interest is then the
composition $\td= \eval\circ \tdl$.  The triangle inequality for $F$
implies the usual triangle inequality for $\td$:
$\td( \sigma, \tau)+ \td( \tau, \chi)\le \td(
\sigma, \chi)$ for all $\sigma, \tau, \chi\in \Spec^\infty$, hence
$\td$ is a hemimetric on $\Spec^\infty$.

We need to work with distances which factor through $\LL$, instead of
plainly taking values in $\Realnn\cup\{ \infty\}$, because some
distances which are useful in practice, as the ones in
Examples~\ref{wm2.ex:limavgdist} and~\ref{wm2.ex:maxleaddist} below, have no
recursive characterization using $\LL= \Realnn\cup\{ \infty\}$.  Whether
the theory works for more general intermediate lattices than $\LL=(
\Realnn\cup\{ \infty\})^M$ is an open question; we have had no occasion
to use more general lattices in practice.

\subsection{Examples}
\label{wm2.se:examples}

To give an application to the framework laid out above, we show here a
few examples of specification labels and trace distances and how they
fit into the framework.  For a much more comprehensive application of
the theory see Chapter~\ref{ch:qltbt}.

\begin{example}
  \label{wm2.ex:mfcs}
  \label{wm2.ex:accdist}
  A good example of a set of specification labels is given by
  $\Spec= \Sigma\times \II$, where $\Sigma$ is a finite set of
  discrete labels and
  $\II=\{[ l, r]\mid l\in \Int\cup\{ -\infty\}, r\in \Int\cup\{
  \infty\}, l\le r\}$ is the set of extended-integer intervals.  The
  partial order is defined by $( a,[ l, r])\labpre( a',[ l', r'])$ iff
  $a= a'$, $l'\le l$ and $r'\ge r$.  Hence refinement is given by
  restricting intervals, so that
  $\Imp= \Sigma\times\{[ x, x]\mid x\in \Int\}\approx \Sigma\times
  \Int$.

  The implementation label distance is given by
  \begin{equation*}
    d(( a, x),( a', x'))=
    \begin{cases}
      | x- x'| &\text{if } a= a', \\
      \infty &\text{otherwise},
    \end{cases}
  \end{equation*}
  so that for specification labels $( a,[ l, r])$, $( a',[ l', r'])$,
  \begin{align*}
    d(( a,[ l, r]),( a',[ l', r'])) &= \adjustlimits \sup_{
      m\in\llbracket( a,[ l, r])\rrbracket} \inf_{ n\in\llbracket( a',[
      l', r'])\rrbracket} d( m, n) \\
    &=
    \begin{cases}
      \max( l'- l, r- r', 0) &\text{if } a= a', \\
      \infty &\text{otherwise}.
    \end{cases}
  \end{align*}

  Now let $\LL= \Realnn\cup\{ \infty\}$, $\eval= \id$, and
  \begin{equation*}
    F( m, n, \alpha)= d( m, n)+ \lambda \alpha
  \end{equation*}
  for some fixed \emph{discounting factor} $\lambda\in \Real$ with
  $0< \lambda< 1$, then
  $\td( \sigma, \tau)= \sum_j \lambda^j d( \sigma_j, \tau_j)$ for
  implementation traces $\sigma$, $\tau$ of equal length.  This is the
  accumulating distance which we have used in
  Chapter~\ref{ch:weightedmodal} to develop a specification theory; we
  will continue this example below to show how it fits in our present
  context.
\end{example}

\begin{example}
  \label{wm2.ex:pwdist}
  Using the same setting as above, with $\Spec= \Sigma\times \II$,
  $( a,[ l, r])\labpre( a',[ l', r'])$ iff $a= a'$, $l'\le l$ and
  $r'\ge r$, and $d(( a, x),( a', x'))=| x- x'|$ if $a= a'$ and
  $\infty$ otherwise, we can instantiate $F$ to a \emph{point-wise}
  instead of accumulating distance.

  Let again $\LL= \Realnn\cup\{ \infty\}$ and $\eval= \id$, but now
  \begin{equation*}
    F( m, n, \alpha)= \max( d( m, n), \alpha)\,.
  \end{equation*}
  Then $\td( \sigma, \tau)= \sup_j d( \sigma_j, \tau_j)$ for
  implementation traces $\sigma$, $\tau$ of equal length, hence
  measuring the biggest individual difference between the traces'
  symbols.  We will also continue this example below to show how to
  develop a specification theory based on the point-wise distance.
\end{example}

\begin{example}
  \label{wm2.ex:limavgdist}
  Again with the same instantiations of $\Imp$ and $\Spec$ as above, we
  can introduce \emph{limit-average} distance.  Here we let $\LL=(
  \Realnn\cup\{ \infty\})^\Nat$, $\eval: \LL\to \Realnn\cup\{ \infty\}$
  given by $\eval( \alpha)= \liminf_j \alpha( j)$, and
  \begin{equation*}
    F( m, n, \alpha)(
    j)= \frac1{ j+ 1} d( m, n)+ \frac j{ j+ 1} \alpha( j- 1)\,,
  \end{equation*}
  then
  $\td( \sigma, \tau)= \eval( \tdl( \sigma, \tau))= \liminf_j \frac1{
    j+ 1}\sum_{ i= 0}^j d( \sigma_j, \tau_j)$ for traces of equal
  length.  We show below how this distance, in the framework of the
  present paper, gives a limit-average specification theory.
\end{example}

\begin{example}
  Examples~\ref{wm2.ex:mfcs} to~\ref{wm2.ex:limavgdist} above are in a
  sense agnostic to the precise structure of implementation and
  specification labels.  Indeed, the definitions only use the label
  distance $d: \Imp\times \Imp\to \Realnn\cup\{ \infty\}$, hence
  $\Imp$ (and $\Spec$) can be any set.  In particular, the theory put
  forward here works equally well in a \emph{multi-weighted} setting
  as for example in~\cite{DBLP:conf/ictac/FahrenbergJLS11}, where
  $\Imp= \Sigma\times \Int^k$ and $\Spec= \Sigma\times \II^k$ for some
  $k\in \Nat$.
\end{example}

\begin{example}
  \label{wm2.ex:maxleaddist}
  With the same instantiations of $\Imp$ and $\Spec$ as in
  Examples~\ref{wm2.ex:mfcs} to~\ref{wm2.ex:limavgdist}, we can
  introduce a distance which, instead of accumulating individual label
  differences, measures the long-run difference between
  \emph{accumulated labels}.  This \emph{maximum-lead} distance is
  especially useful for real-time systems and has been considered
  in~\cite{DBLP:conf/formats/HenzingerMP05,
    DBLP:journals/jlp/ThraneFL10}, see also Chapter~\ref{ch:wtsjlap}.
  Unlike Examples~\ref{wm2.ex:mfcs} to~\ref{wm2.ex:limavgdist}, it
  does not use the distance $d$ on implementation labels in the
  definition of the trace distance; rather it accumulates the labels
  itself before taking the distance.

  Let $\LL=( \Realnn\cup\{ \infty\})^\Real$ and
  define $F: \Imp\times \Imp\times \LL\to \LL$ by
  \begin{equation*}
    F(( a, x),( a', x'), \alpha)( \delta)=
    \begin{cases}
      \infty &\text{if } a\ne a', \\
      \max(| \delta+ x- x'|, \alpha( \delta+ x- x')) &\text{if } a= a'.
    \end{cases}
  \end{equation*}
  Define $\eval: \LL\to \Realnn\cup\{ \infty\}$ by $\eval( \alpha)= \alpha( 0)$;
  the maximum-lead distance assuming the lead is zero.  It can then be
  shown that for implementation traces $\sigma=(( a_0, x_0),( a_1,
  x_1),\dots)$, $\tau=(( a_0, y_0),( a_1, y_1),\dots)$,
  \begin{equation*}
    \td(
    \sigma, \tau)= \eval( \tdl( \sigma, \tau))= \sup_m\Big| \sum_{ i= 0}^m x_i-
    \sum_{ i= 0}^m y_i\Big|
  \end{equation*}
  is precisely the maximum-lead distance.
\end{example}

\begin{example}
  \label{wm2.ex:clock}
  Specification labels different from the ones above can for example
  be \emph{clock constraints}, or
  \emph{zones}~\cite{DBLP:journals/tcs/AlurD94}.  For a finite set
  $\Sigma$, let $\Spec= \Phi( \Sigma)$ be the set of closed clock
  constraints over $\Sigma$ given by
  \begin{equation*}
    \Phi( \Sigma) \ni \phi ::= a\le k\mid a\ge k\mid \phi_1\wedge \phi_2
    \qquad( a\in \Sigma, k\in \Nat, \phi_1, \phi_2\in \Phi( \Sigma)).
  \end{equation*}
  Clock constraints have a natural partial order given by 
  $\phi\labpre \phi'$ iff $\phi\Longrightarrow \phi'$.
  Implementation labels are then clock constraints which impose a
  precise value for each $a\in \Sigma$, which can be seen as functions $u:
  \Sigma\to \Nat$.  The natural distance between such \emph{discrete clock
    valuations} is $d( u, u')= \max_{ a\in \Sigma}| u( a)- u'( a)|$, and on
  top of this, any interesting trace distance can be imposed using our
  framework.
\end{example}

\subsection{Structured Modal Transition Systems}

\begin{definition}
  \qquad A \emph{structured modal transition system} (SMTS) is a tuple
  $( S, s_0, \mmayto_S, \mmustto_S)$ consisting of a set $S$ of
  states, an initial state $s_0\in S$, and \must and \may transitions
  $\mmustto_S, \mmayto_S\subseteq S\times \Spec\times S$ for which it
  holds that for all $s\mustto{ k}_S s'$ there is
  $s\mayto{ \ell}_S s'$ with $k\labpre \ell$.
\end{definition}

The last condition is one of \emph{consistency}: everything which is
required is also allowed.  If no confusion can arise, we will omit the
subscripts $S$ on the \must and \may transitions; we will also sometimes
identify an SMTS $( S, s_0, \mmayto_S, \mmustto_S)$ with its state set
$S$.

Intuitively, a \may transition $s\mayto{ k} s'$ specifies that an
implementation $I$ of $S$ is \emph{permitted} to have a corresponding
transition $i\mustto{ m} i'$, for any $m\in \llbracket k\rrbracket$,
whereas a \must transition $s\mustto{ \ell} s'$ postulates that $I$ is
\emph{required} to implement at least one corresponding transition
$i\mustto{ n} i'$ for some $n\in \llbracket \ell\rrbracket$.  We will
make this precise below.

An SMTS $S$ is an \emph{implementation} if $\mmustto_S=
\mmayto_S\subseteq S\times \Imp\times S$; hence in an implementation,
all optional behavior has been resolved, and all data has been refined
to implementation labels.

\begin{definition}
  \label{wm2.de:deterministic}
  An SMTS $( S, s_0, \mmayto_S, \mmustto_S)$ is
  \emph{$\LL$-deterministic}, for a given lattice $\LL$, if it holds for
  all $s\in S$, $s\mayto{ k_1} s_1$, $s\mayto{ k_2} s_2$ for which there
  is $k\in \Spec$ with $d( k, k_1)\ne \top_\LL$ and $d( k, k_2)\ne
  \top_\LL$ that $k_1= k_2$ and $s_1= s_2$.
\end{definition}

Note that for the Boolean label distance given by $d( k, k')= \bot_\LL$
if $k= k'$ and $\top_\LL$ otherwise, the above definition reduces to the
property that if $k_1= k_2$, then also $s_1= s_2$, hence
$\LL$-determinism is a generalization of usual determinism.  In our
quantitative case, we need to be more restrictive: not only do we not
allow distinct transitions from $s$ with the same label, but we forbid
distinct transitions with labels which have a common quantitative
refinement.  Despite of this, we will generally omit the $\LL$ and say
deterministic instead of $\LL$-deterministic.

\begin{example}
  For the label distance $d(( a, x),( a', x'))=| x- x'|$ if $a= a'$ and
  $\infty$ otherwise of Examples~\ref{wm2.ex:accdist} to~\ref{wm2.ex:limavgdist}
  and~\ref{wm2.ex:maxleaddist}, the above condition that there exist $k\in
  \Spec$ with $d( k, k_1)\ne \top_\LL$ and $d( k, k_2)\ne \top_\LL$ is
  equivalent, with $k_1=( a_1, I_1)$ and $k_2=( a_2, I_2)$, to saying
  that $a_1= a_2$, hence our notion of determinism agrees with the one
  of the previous chapter.
\end{example}

A \emph{modal refinement} of SMTS $S$, $T$ is a relation $R\subseteq
S\times T$ such that for any $( s, t)\in R$,
\begin{itemize}
\item whenever $s\mayto{ k}_S s'$, then also $t\mayto{ \ell}_T t'$ for
  some $k\labpre \ell$ and $( s', t')\in R$,
\item whenever $t\mustto{ \ell}_T t'$, then also $s\mustto{ k}_S s'$ for
  some $k\labpre \ell$ and $( s', t')\in R$.
\end{itemize}
Thus any behavior which is permitted in $S$ is also permitted in $T$,
and any behavior required in $T$ is also required in $S$.  We write
$S\mr T$ if there is a modal refinement $R\subseteq S\times T$ with $(
s_0, t_0)\in R$.

The \emph{implementation semantics} of a SMTS $S$ is the set $\llbracket
S\rrbracket=\{ I\mr S\mid I\text{ is an implementation}\}$, and we
write $S\thr T$ if $\llbracket S\rrbracket\subseteq \llbracket
T\rrbracket$, saying that $S$ \emph{thoroughly refines} $T$.  It follows
by reflexivity of $\mr$ that $S\mr T$ implies $S\thr T$, hence
modal refinement is a \emph{syntactic over-approximation} of thorough
refinement.

It can be shown for standard modal transition systems that $S\thr T$
does not imply $S\mr T$ unless $T$ is deterministic,
see~\cite{DBLP:journals/tcs/BenesKLS09} and
Theorem~\ref{weightedmodal.th:det-dteqdm} in the previous chapter.  We
shall provide a quantitative generalization of this result in
Theorem~\ref{wm2.th:thorough_vs_modal} below.  Also, modal refinement
for MTS can be decided in polynomial time, whereas deciding thorough
refinement is EXPTIME-complete~\cite{DBLP:journals/tcs/BenesKLS09}.
Intuitively, thorough refinement---inclusion of implementation
sets---is the relation one really is interested in, but modal
refinement provides a useful over-approximation.

\section{Refinement Distances}
\label{wm2.se:refdist}

We define two distances between SMTS, one at the syntactic and one at
the semantic level.

\subsection{Modal and thorough refinement distance}

\begin{definition}
  \label{wm2.de:modrefdist}
  The \emph{modal refinement distance} $\md: S\times T\to \LL$ between
  the states of SMTS $S$, $T$ is defined to be the least fixed point to
  the equations
  \begin{equation*}
    \md( s, t)= \max\left\{
      \begin{aligned}
        & \adjustlimits \sup_{ s\,\mayto{ k}_S\, s'\,} \inf_{ \, t\,\mayto{
          \ell}_T\, t'} F( k, \ell, \md( s', t'))\,, \\
        & \adjustlimits \sup_{ t\,\mustto{ \ell}_T\, t'\,} \inf_{ \, s\,\mustto{
          k}_S\, s'} F( k, \ell, \md( s', t'))\,.
      \end{aligned}
    \right.
  \end{equation*}
  We let $\md(S, T) = \md(s_0, t_0)$, and we write $S\mr^\alpha T$ if
  $\md( S, T)\sqsubseteq_\LL \alpha$.
\end{definition}

\begin{lemma}
  The modal refinement distance is well-defined and a hemimetric.  Also,
  $S\mr T$ implies $\md( S, T)= \bot_\LL$.
\end{lemma}

\begin{proof}
  Let $I: \LL^{ S\times T}\to \LL^{ S\times T}$ be the endofunction
  defined by
  \begin{equation*}
    I( h)( s, t)= \max\left\{
      \begin{aligned}
        & \adjustlimits \sup_{ s\,\mayto{ k}_S\, s'\,} \inf_{ \, t\,\mayto{
          \ell}_T\, t'} F( k, \ell, h( s', t')), \\
        & \adjustlimits \sup_{ t\,\mustto{ \ell}_T\, t'\,} \inf_{ \, s\,\mustto{
          k}_S\, s'} F( k, \ell, h( s', t')).
      \end{aligned}
    \right.
  \end{equation*}
  The lattice $\LL^{ S\times T}$ is complete because $\LL$ is, and $I$
  is monotone because $F( k, \ell,\cdot): \LL\to \LL$ is.  By an
  application of Tarski's fixed point
  theorem~\cite{journals/pjmath/Tarski55}, $I$ has a unique least fixed
  point which hence defines $\md$.

  The property that $\md( S, S)= 0$ for all SMTS $S$ is clear, and the
  triangle inequality $\md( S, T)\oplus_\LL \md( T, U)\sqsupseteq_\LL
  \md( S, U)$ can be shown inductively.

  To show the last claim, assume $s\mr t$.  Then for any $s\mayto{ k}
  s'$ there is $t\mayto{ \ell} t'$ for which $k\labpre \ell$,
  hence $F( k, \ell, \alpha)= \alpha$ for all $\alpha\in \LL$ by
  Axiom~$(\ref{wm2.en:F.indid'}')$.  Similarly for \must transitions, so the
  fixed point equations simplify to
  \begin{equation*}
    \md( s, t)= \max\Big( \adjustlimits \sup_{ s\mayto{} s'} \inf_{
      t\mayto{} t'} \md( s', t'), \adjustlimits \sup_{ t\mustto{} t'} \inf_{
      s\mustto{} s'} \md( s', t')\Big),
  \end{equation*}
  the least fixed point of which is $\md( s, t)= \bot_\LL$. \qed
\end{proof}

One can also define a \emph{linear distance} between states, analogous
to trace inclusion.  This is given by
\begin{equation*}
  \tdl( s, t)= \max\Big( \adjustlimits \sup_{ \sigma\in \tracesfrom{
      s}} \inf_{ \tau\in \tracesfrom{ t}} \tdl( \sigma, \tau),
  \adjustlimits \sup_{ \tau\in \tracesfrom{ t}} \inf_{ \sigma\in
    \tracesfrom{ s}} \tdl( \sigma, \tau)\Big),
\end{equation*}
where $\tracesfrom{ s}$ denotes the set of (\may or \must) traces
emanating from $s$.  It can then be shown that
$\tdl( s, t)\sqsubseteq_\LL \md( s, t)$ for all $s, t\in S$, see Chapter~\ref{ch:qltbt}.

\begin{definition}
  The \emph{thorough refinement distance} from an SMTS $S$ to an
  SMTS $T$ is
  \begin{equation*}
    \thd( S, T)= \adjustlimits \sup_{ I\in \llbracket S\rrbracket} \inf_{
      J\in \llbracket T\rrbracket} \md( I, J),
  \end{equation*}
  and we write $S\thr^\alpha T$ if $\thd( S, T)\sqsubseteq_\LL \alpha$.
\end{definition}

\begin{lemma}
  The thorough refinement distance is a hemimetric, and $S\thr T$
  implies $\thd( S, T)= \bot_\LL$.
\end{lemma}

\begin{proof}
  The equality $\thd( S, S)= \bot_\LL$ is clear, and the triangle
  inequality $\thd( S, T)+ \thd( T, U)\ge \thd( S, U)$ follows like in the
  proof of~\cite[Lemma~3.72]{book/AliprantisB07}.  If $S\thr T$,
  then $\llbracket S\rrbracket\subseteq \llbracket T\rrbracket$ implies
  $\thd( S, T)= \bot_\LL$. \qed
\end{proof}

\subsection{Refinement families}

As is the case for ordinary
(bi)simulation~\cite{DBLP:conf/tcs/Park81}, there is a dual
\emph{relational} notion of refinement distance which is useful.
Before we can introduce this, we need a notion similar to the
\emph{finite branching} assumption one needs to make for the case of
bisimulation, \cf~\cite{book/Milner89}, see also
Chapter~\ref{ch:weightedmodal}.

\begin{definition}
  \label{wm2.de:comp}
  A SMTS $S$ is said to be \emph{compactly branching} if the sets
  $\{( s', k)\mid s\mayto{ k} s'\}, \{( s', k)\mid s\mustto{ k}
  s'\}\subseteq S\times \Spec$ are compact under the symmetrized
  product distance $\bar d_\textup{\textsf{m}}\times \bar d$ for every
  $s\in S$.
\end{definition}

Recall that the pseudometric $\bar d_\textup{\textsf{m}}\times \bar d$
is given by
$\bar d_\textup{\textsf{m}}\times \bar d(( s, k),( s', k'))= \bar
d_\textup{\textsf{m}}( s, s')+ \bar d( k, k')= \max( \md( s, s'), \md(
s', s))+ \max( d( k, k'), d( k', k))$.  As in
Chapter~\ref{ch:weightedmodal}, we will need compactness of the sets
$\{( s', k)\mid s\mayto{ k} s'\}, \{( s', k)\mid s\mustto{ k}
s'\}\subseteq S\times \Spec$ for the property that continuous
functions defined on them attain their infimum and supremum, see
Lemma~\ref{wm2.le:family} and its proof below.

The notion of compact branching was first introduced, for a formalism
of \emph{metric transition systems},
in~\cite{journals/anyas/Breugel96}.  It is a natural generalization of
finite branching to a distance setting; we shall henceforth assume all
our SMTS to be compactly branching.

\begin{definition}
  A \emph{modal refinement family} from $S$ to $T$, for SMTS $S$, $T$,
  is an $\LL$-indexed family of relations $\cal R=\{ R_\alpha\subseteq
  S\times T\mid \alpha\in \LL\}$ with the property that for all
  $\alpha\in \LL$ and all $( s, t)\in R_\alpha$,
  \begin{itemize}
  \item whenever $s\mayto{ k}_S s'$, then there is $\beta\in \LL$ and $(
    s', t')\in R_\beta$ for which $t\mayto{ \ell}_T t'$ and $F( k, \ell,
    \beta)\sqsubseteq_\LL \alpha$,
  \item whenever $t\mustto{ \ell}_T t'$, then there is $\beta\in \LL$
    and $( s', t')\in R_\beta$ for which $s\mustto{ k}_S s'$ and $F( k,
    \ell, \beta)\sqsubseteq_\LL \alpha$.
  \end{itemize}
\end{definition}

Compact branching implies that refinement families $\cal R$ are
\emph{closed} in the sense that for all $s\in S$, $t\in T$,
$( s, t)\in R_{ \inf\{ \alpha\mid( s, t)\in R_\alpha\in \cal R\}}\in
\cal R$.  Also note how this definition is a common refinement of the
notions of relaton families from Chapters~\ref{ch:qltbt}
and~\ref{ch:weightedmodal}.

\begin{lemma}
  \label{wm2.le:family}
  For all SMTS $S$, $T$ and $\alpha\in \LL$, $S\mr^\alpha T$ if and
  only if there is a modal refinement family $\cal R$ from $S$ to $T$
  with $( s_0, t_0)\in R_\alpha$.
\end{lemma}

We say that a modal refinement family as in the lemma \emph{witnesses}
$S\mr^\alpha T$; this is of course the same as saying that it
witnesses $\md( S, T)\sqsubseteq_\LL \alpha$, which we sometimes shorten
to say that it witnesses $\md( S, T)$.

\begin{proof}
  Assume first that $S\mr^\alpha T$, thus we know that
  $\md(S,T) \sqsubseteq_\LL \alpha$. 
  We have to show that there is a modal refinement family $\cal R$
  from $S$ to $T$ with $(s_0,t_0) \in R_\alpha$.
  Define a family $\cal R=\{ R_{ \alpha'}\subseteq S\times T\mid \alpha'\in
  \LL\}$ by
  \begin{equation*}
    R_{ \alpha'}=\{( s, t)\mid \md( s, t)\sqsubseteq_\LL \alpha'\}
  \end{equation*}
  for every $\alpha' \in \LL$; note that $\cal R$ is closed in the sense
  above.  Now let $\beta\in \LL$ and $( s, t)\in R_\beta$.
  \begin{itemize}
  \item Assume $s\mayto{ k} s'$. By $\md( s, t)\sqsubseteq_\LL \beta$ and
    the definition of $\md(s,t)$ it follows that $\inf_{ \smash{ t
        \mayto{\ell} t'}} F(k,\ell,\md(s',t')) \sqsubseteq_\LL \beta$.
    As $T$ is compactly branching and $F$ continuous, the set $\{
    F(k,\ell,\md(s',t'))\mid t \mayto{\ell} t'\}$ is compact, hence
    there exists a transition $t \mayto{\ell} t'$ such that
    $F(k,\ell,\md(s',t')) \sqsubseteq_\LL \beta$.
  \item Assume $t \mustto{\ell} t'$. By $\md( s, t)\sqsubseteq_\LL \beta$
    and the definition of $\md(s,t)$ it follows that $\inf_{ \smash{ s
        \mustto{k} s'}} F( k,\ell,\md(s',t')) \sqsubseteq_\LL \beta$.
    Again $\{ F(k,\ell,\md(s',t'))\mid s \mustto{k} s'\}$ is a compact
    set, whence there exists a transition $s \mustto{k} s'$ such that
    $F(k,\ell,\md(s',t')) \sqsubseteq_\LL \beta$.
  \end{itemize}

  For the other direction, assume a refinement family $\cal R$ from
  $S$ to $T$ with $(s_0,t_0) \in R_\alpha$.  Define
  $h: S\times T\to \LL$ by
  $h( s, t)= \inf\{ \alpha\mid( s, t)\in R_\alpha\}$.  Then
  $( s, t)\in R_\beta$ implies that $h( s, t)\sqsubseteq_\LL \beta$.
  Let $s\in S$ and $t\in T$, then $( s, t)\in R_{ h( s, t)}$ because
  $\cal R$ is closed, hence for all $s\mayto{ k} s'$ there is
  $t\mayto{ \ell} t'$ and $\alpha'\in \LL$ for which
  $F( k, \ell, \alpha')\sqsubseteq_\LL h( s, t)$ and
  $( s', t')\in R_{ \alpha'}$, implying
  $h( s', t')\sqsubseteq_\LL \alpha'$ and hence
  $F( k, \ell, h( s', t'))\sqsubseteq_\LL h( s, t)$ by monotonicity
  and transitivity.  Similarly, for all $t\mustto{ \ell} t'$ there is
  $s\mustto{ k} s'$ with
  $F( k, \ell, h( s', t'))\sqsubseteq_\LL h( s, t)$.  Hence $h$ is a
  pre-fixed point for the equations in the definition of $\md$,
  implying that $\md( s, t)\sqsubseteq_\LL h( s, t)$ for all $s\in S$,
  $t\in T$, thus especially $\md( s_0, t_0)\sqsubseteq_\LL \alpha$,
  because $(s_0,t_0) \in R_\alpha$ implies
  $h(s_0,t_0) \sqsubseteq_\LL \alpha$ and
  $\md(s_0,t_0) \sqsubseteq_\LL h(s_0,t_0)$. \qed
\end{proof}

\subsection{Modal distance bounds thorough distance}

The next theorem shows that the modal refinement distance
overapproximates the thorough one, and that it is exact for
deterministic SMTS.  This is similar to the situation for standard
modal transition systems~\cite{DBLP:conf/avmfss/Larsen89};
note~\cite{DBLP:conf/avmfss/Larsen89} that deterministic
specifications generally suffice for applications.

\begin{theorem}
  \label{wm2.th:thorough_vs_modal}
  For all SMTS $S$, $T$, $\thd( S, T)\sqsubseteq_\LL \md( S, T)$.  If
  $T$ is deterministic, then $\thd( S, T)= \md( S, T)$.
\end{theorem}

The counterexample for the Boolean version of the second result given
in~\cite{DBLP:journals/tcs/BenesKLS09} also works in our setting, to
show that there exist (necessarily nondeterministic) SMTS $S$, $T$ for
which $\thd( S, T)= \bot_\LL$, but $\md( S, T)= \top_\LL$.

\begin{proof}
  For the first claim, if $\md( S, T)= \top_\LL$, we have nothing to
  prove.  Otherwise, let
  $\cal R=\{ R_\alpha\subseteq S\times T\mid \alpha\in \LL\}$ be a
  modal refinement family which witnesses $\md( S, T)$, then
  $( s_0, t_0)\in R_{ \md( S, T)}$.  Let
  $I\in \llbracket S\rrbracket$; we will expose
  $J\in \llbracket T\rrbracket$ for which
  $\md( I, J)\sqsubseteq_\LL \md( S, T)$.

  Let $\tilde R\subseteq I\times S$ be a witness for $I\mr S$, define
  $R'_\alpha= \tilde R\circ R_\alpha\subseteq I\times T$ for all
  $\alpha\in \LL$, and let $\cal R'=\{ R'_\alpha\mid \alpha\in \LL\}$.
  We let the states of $J$ be $J= T$, with $j_0= t_0$, and define
  $\mmayto_J= \mmustto_J$ as follows:

  For any $i\mustto{ m}_I i'$ and any $t\in T$ for which
  $( i, t)\in R'_\alpha\in \cal R'$ for some $\alpha\in \LL$,
  $\alpha\ne \top_\LL$, we have $t\mayto{ \ell}_T t'$ with
  $( i', t')\in R'_\beta\in \cal R'$ for some $\beta\in \LL$ with
  $F( m, \ell, \beta)\sqsubseteq_\LL \alpha$.  As
  $\llbracket \ell\rrbracket$ is closed under $F$, there is
  $n\in \llbracket \ell\rrbracket$ for which
  $F( m, n, \beta)= F( m, \ell, \beta)$, and we add a transition
  $t\mustto{ n}_J t'$ to $J$.

  Similarly, for any $t\mustto{ \ell}_T t'$ and any $i\in I$ for which
  $( i, t)\in R'_\alpha\in \cal R'$ for some $\alpha\in \LL$, $\alpha\ne
  \top_\LL$, we have $i\mustto{ m}_I i'$ with $( i', t')\in R'_\beta$ for
  some $\beta\in \LL$ with $F( m, \ell, \beta)\sqsubseteq_\LL \alpha$.
  Using again closedness of $\llbracket \ell\rrbracket$, we find $n\in
  \llbracket \ell\rrbracket$ for which $F( m, n, \beta)= F( m, \ell,
  \beta)$ and add a transition $t\mustto{ n}_J t'$ to $J$.

  We show that the identity relation $\{( t, t)\mid t\in T\}\subseteq
  J\times T$ witnesses $J\mr T$.  Let first $t\mustto{ n}_J t'$; we
  must have used one of the two constructions above for creating this
  transition.  In the first case, there is $t\mayto{ \ell}_T t'$ with
  $n\in \llbracket \ell\rrbracket$, and in the second case, there is
  $t\mustto{ \ell}_T t'$, hence also $t\mayto{ \ell'}_T t'$ with
  $\ell\labpre \ell'$, thus $n\in \llbracket
  \ell\rrbracket\subseteq \llbracket \ell'\rrbracket$.  Now let
  $t\mustto{ \ell}_T t'$, then the second construction above has
  introduced $t\mustto{ n}_J t'$ with $n\in \llbracket \ell\rrbracket$.

  To finish the proof, we show that the family $\cal R'$ is a witness
  for $\md( I, J)\sqsubseteq_\LL \md( S, T)$.  First,
  $( i_0, s_0)\in \tilde R$ and $( s_0, t_0)\in R_{ \md( S, T)}$ imply
  $( i_0, t_0)\in R'_{ \md( S, T)}$.  Let
  $( i, t)\in R'_\alpha\in \cal R'$ for some $\alpha\in \LL$,
  $\alpha\ne \top_\LL$, and assume first $i\mustto{ m}_I i'$.  Then
  $t\mayto{ \ell}_T t'$ and $t\mustto{ n}_J t'$ by the first part of
  our above construction, and $( i', t')\in R'_\beta$ with
  $F( m, n, \beta)\sqsubseteq_\LL F( m, \ell, \beta)\sqsubseteq_\LL
  \alpha$.  For the converse, a transition $t\mustto{ n}_J t'$ must
  have been introduced above, and in both cases, $i\mustto{ m}_I i'$
  with $( i', t')\in R'_\beta$ and
  $F( m, n, \beta)\sqsubseteq_\LL F( m, \ell, \beta)\sqsubseteq_\LL
  \alpha$.

  Now to the proof of the second assertion of the theorem.  If
  $\thd( S, T)= \top_\LL$, we are done.  Otherwise we inductively
  construct a relation family
  $\cal R=\{ R_\alpha\subseteq S\times T\mid \alpha\in \LL\}$ which
  satisfies $\thd(( s, S),( t, T))\sqsubseteq \alpha$ for any
  $( s, t)\in R_\alpha$, as follows: Begin by letting
  $R_\alpha=\{( s_0, t_0)\}$ for all
  $\alpha\sqsupseteq_\LL \thd( S, T)$, and let now
  $( s, t)\in R_\alpha$ with
  $\thd(( s, S),( t, T))\sqsubseteq \alpha\ne \top_\LL$.

  Let $s\mayto{ k}_S s'$ and $t\mayto{ \ell}_T t'$ such that $d( k,
  \ell)\ne \top_\LL$.  Let $( i', I')\in \llbracket( s', S)\rrbracket$
  and $m\in \llbracket k\rrbracket$, then there is $( i, I)\in
  \llbracket( s, S)\rrbracket$ for which $i\mustto{ m}_I i''$ and $(
  i'', I)\mr( i', I')$.  By the triangle inequality we have $\thd(( i,
  I),( t, T))\sqsubseteq_\LL \thd(( i, I),( s, S))\oplus_\LL \thd(( s, S),(
  t, T))\sqsubseteq_\LL \alpha$, hence there is $t\mayto{ \ell'} t''$ for
  which $d( m, \ell')\sqsubseteq_\LL \alpha$.  But we also have $d( m,
  \ell)\sqsubseteq_\LL d( m, k)\oplus_\LL d( k, \ell)= d( k, \ell)\ne
  \top_\LL$, so by determinism of $T$ it follows that $\ell= \ell'$ and
  $t'= t''$.

  As $m\in \llbracket k\rrbracket$ was chosen arbitrarily above, we have
  $d( m, \ell)\sqsubseteq_\LL \alpha$ for all $m\in \llbracket
  k\rrbracket$, hence $d( k, \ell)= F( k, \ell, \bot_\LL)\sqsubseteq
  \alpha$.  Let $B=\{ \beta'\in \LL\mid F( k, \ell, \beta')\sqsubseteq_\LL
  \alpha\}$ and $\beta= \sup B$, then $F( k, \ell, \beta)\sqsubseteq_\LL
  \alpha$ as $\bot_\LL\in S$.  Add $( s', t')$ to $R_\gamma$ for all
  $\gamma\sqsupseteq_\LL \beta$.

  We miss to show that $\thd(( s', S),( t', T))\sqsubseteq_\LL \beta$.
  By $\thd(( s, S),( t, T))\sqsubseteq_\LL \alpha$ we must have
  $( j, J)\in \llbracket( t, T)\rrbracket$, $j\mustto{ n}_J j'$, and
  an element $\beta'\in \LL$ for which
  $\md(( i', I'),( j', J))\sqsubseteq_\LL \beta'$ and
  $F( m, n, \beta')\sqsubseteq_\LL \alpha$.  Then
  \begin{equation*}
    F( k, \ell, \beta')= F( m, \ell, \beta')\sqsubseteq_\LL F( m, n,
    \beta')\sqsubseteq_\LL \alpha\,,
  \end{equation*}
  hence $\beta'\in B$, implying that
  $\thd(( s', S),( t', T))\sqsubseteq_\LL \beta'\sqsubseteq_\LL
  \beta$.

  We show that $\cal R$ is a refinement family which witnesses
  $\md( S, T)$.  Let $( s, t)\in R_\alpha\in \cal R$ for some
  $\alpha\in \LL$ and assume $s\mayto{ k}_S s'$.  Let
  $m\in \llbracket k\rrbracket$, then there is
  $( i, I)\in \llbracket( s, S)\rrbracket$ with $i\mustto{ m}_I i'$.
  As $\thd(( i, I),( t, T))\sqsubseteq_\LL \alpha$, this implies that
  there is $t\mayto{ \ell}_T t'$ with
  $d( m, \ell)\sqsubseteq_\LL \alpha$.  Also for any other
  $m'\in \llbracket k\rrbracket$ we have $t\mayto{ \ell'}_T t''$ with
  $d( m, \ell')\sqsubseteq_\LL \alpha$, hence $\ell= \ell'$ and
  $t'= t''$ by determinism.  As $m$ was chosen arbitrarily, we have
  $d( m, \ell)\sqsubseteq \alpha$ for all
  $m\in \llbracket k\rrbracket$, hence
  $d( k, \ell)= F( k, \ell, \bot_\LL)\sqsubseteq \alpha$.  By
  construction of $\cal R$, $( s', t')\in R_\beta$ for
  $\beta= \sup\{ \beta'\in \LL\mid F( k, \ell, \beta')\sqsubseteq_\LL
  \alpha\}$.

  Now assume $t\mustto{ \ell}_T t'$.  Let $( i, I)\in \llbracket( s,
  S)\rrbracket$, then we have $( j, J)\in \llbracket( t, T)\rrbracket$
  with $\md(( i, I),( j, J))\sqsubseteq_\LL \alpha$.  We must have
  $j\mustto{ n}_J j'$ with $n\in \llbracket \ell\rrbracket$, hence there
  are $i\mustto{ m}_I i'$ and $\beta'\in \LL$ with $\md(( i', I),( j',
  J))\sqsubseteq_\LL \beta'$ and $F( m, n, \beta')\sqsubseteq_\LL \alpha$.

  The above considerations hold for all
  $( i, I)\in \llbracket( s, S)\rrbracket$, hence there is $k\in \LL$
  with $m\in \llbracket k\rrbracket$, $s\mustto{ k}_S s'$, and
  $F( k, \ell, \beta')= F( m, \ell, \beta')$.  But then
  $F( k, \ell, \beta')\sqsubseteq_\LL F( m, n, \beta')\sqsubseteq_\LL
  \alpha$, hence by construction of $\cal R$, $( s', t')\in R_\beta$
  for
  $\beta= \sup\{ \beta'\in \LL\mid F( k, \ell, \beta')\sqsubseteq_\LL
  \alpha\}$. \qed
\end{proof}

\subsection{Quantitative relaxation}

In a quantitative framework, it can be useful to be able to \emph{relax}
and \emph{strengthen} specifications during the development process.
Which precise relaxations and strengthenings one wishes to apply will
depend on the actual application, but we can here show three general
relaxations which differ from each other in the \emph{level} of the
theory at which they are applied.  For $\alpha\in \LL$ and SMTS $S$, $T$,
\begin{itemize}
\item $T$ is an \emph{$\alpha$-widening} of $S$ if there is a relation
  $R\subseteq S\times T$ for which $( s_0, t_0)\in R$ and such that for
  all $( s, t)\in R$, $s\mayto{ k}_S s'$ if and only if $t\mayto{
  \ell}_T t'$, and $s\mustto{ k}_S s'$ if and only if $t\mustto{ \ell}_T
  t'$, for $k\labpre \ell$, $d( \ell, k)\sqsubseteq_\LL
  \alpha$, and $( s', t')\in R$;
\item $T$ is an \emph{$\alpha$-relaxation} of $S$ if $S\mr T$ and
  $T\mr^\alpha S$;
\item the \emph{$\alpha$-extended implementation semantics} of $S$ is
  \begin{equation*}
    \llbracket S\rrbracket^{ +\alpha}=\{ I\mr^\alpha S\mid I\text{
      implementation}\}.
  \end{equation*}
\end{itemize}

All three notions have also been introduced for the special case of
integer weights in Section~\ref{weightedmodal.se:relax}; but note that
the notion of widening presented here is more synthetic than the one
of the previous chapter.

The notion of $\alpha$-widening is entirely \emph{syntactic}: up to
unweighted bisimulation, $T$ is the same as $S$, but transition labels
in $T$ can be $\alpha$ ``wider'' than in $S$ (hence also $S\mr T$).
The second notion, $\alpha$-relaxation, works at the level of
semantics of specifications, whereas the last notion is at
implementation level.  A priori, there is no relation between the
syntactic and semantic notions, even though one can be established in
some special cases.

\begin{example}
  For the accumulated distance with discounting factor $\lambda$, any
  $\alpha$-widening is also an $\frac{\alpha}{1-\lambda}$-relaxation,
  see Proposition~\ref{weightedmodal.pr:wide-propt}.  This is due to
  the fact that for traces $\sigma, \tau\in \Spec^\infty$ with
  $d( \sigma_j, \tau_j)\le \alpha$ for all $j$, we have
  $\sum_j \lambda^j d( \sigma_j, \tau_j)\le \sum_j \lambda^j \alpha\le
  ( 1- \lambda)^{ -1} \alpha$ by convergence of the geometric series.

  For the point-wise distance, it is easy to see that any
  $\alpha$-widening is also an $\alpha$-relaxation, and the same holds
  for the limit-average distance:
  \begin{equation*}
    \liminf_j \tfrac1{ j+ 1}
    \sum_{ i= 0}^j d( \sigma_i, \tau_i)\le \liminf_j \tfrac1{ j+ 1}\, j
    \alpha= \alpha
  \end{equation*}
\end{example}

\begin{example}
  For the maximum-lead distance on the other hand, it is easy to expose
  cases of $\alpha$-widenings which are \emph{not} $\beta$-relaxations
  for any $\beta$.  One example consists of two one-state SMTS $S$, $T$
  with loops $s_0\mustto{a, 1} s_0$ and $t_0\mustto{a, [0,2]} t_0$;
  then $T$ is an $\alpha$-widening of $S$ for $\alpha( \delta)=| \delta+
  1|$, but $\md( T, S)= \top_\LL$.
\end{example}

\begin{proposition}
  If $T$ is an $\alpha$-relaxation of $S$, then $\llbracket
  T\rrbracket\subseteq \llbracket S\rrbracket^{ +\alpha}$. \qedhere
\end{proposition}

It can be shown for special cases that the inclusion in the
proposition is strict, see Section~\ref{weightedmodal.se:relax}; for
its proof one only needs the fact that
$\md( I, S)\sqsubseteq_\LL \md( I, T)\oplus_\LL \md( T,
S)\sqsubseteq_\LL \alpha$ for all $I\in \llbracket T\rrbracket$.

Also of interest is the relation between relaxations of different
specifications.  An easy application of the triangle inequality for
$\md$ shows that the distance between relaxations is bounded by the sum
of the relaxation constants and the unrelaxed systems' distances:

\begin{proposition}
  \label{wm2.pr:relax}
  Let $T$ be an $\alpha$-relaxation of $S$ and $T'$ an
  $\alpha'$-relaxation of $S'$.  Then $\md( T, T')\sqsubseteq_\LL
  \alpha\oplus_\LL \md( S, S')$ and $\md( T', T)\sqsubseteq_\LL
  \alpha'\oplus_\LL \md( S', S)$. \qedhere
\end{proposition}

\section{Structural Composition and Quotient}
\label{wm2.se:compquot}

We now introduce the different operations on SMTS which make up a
specification theory.  Firstly, we are interested in composing
specifications $S$, $S'$ into a specification $S\|S'$ by synchronizing on
shared actions.
Secondly, we need a quotient operator which solves equations of the form
$S\| X\equiv T$, that is, the quotient synthesizes the most general
specification $T \oslash S$ which describes all SMTS $X$ satisfying the
above equation.

\subsection{Structural composition}
\label{wm2.se:compo}

To structurally compose SMTS, we assume given a generic partial
\emph{label composition} operator $\obar: \Spec\times \Spec\parto
\Spec$ which specifies which labels can synchronize,
\cf~\cite{inbook/WinskelN95}.  We will need to assume the following
property:
\begin{itemize}
\item for all $\ell, \ell'\in \Spec$, $( \exists k\in \Spec: d( k,
  \ell)\ne \top_\LL, d( k, \ell')\ne \top_\LL)\liff( \exists m\in \Spec:
  \ell\obar m, \ell'\obar m\text{ are defined})$.
\end{itemize}
This operator permits to compose labels at transitions which are
executed in parallel; 
the property required relates composability to distances in such a way
that two labels have a common quantitative refinement if and only if
they have a common synchronization.  This is quite natural and holds for
all our examples, and is needed to relate determinism to composition in
the proof of Theorem~\ref{wm2.th:quotient} below.

Additionally, we must assume that there exists a function $P: \LL\times
\LL\to \LL$ which allows us to infer bounds on distances on synchronized
labels.  We assume that $P$ is monotone in both coordinates, has $P(
\bot_\LL, \bot_\LL)= \bot_\LL$, $P( \alpha, \top_\LL)= P( \top_\LL, \alpha)=
\top_\LL$ for all $\alpha\in \LL$, and that
\begin{equation*}
  F( k\obar k', \ell\obar \ell', P( \alpha, \alpha'))\sqsubseteq_\LL P(
  F( k, \ell, \alpha), F( k', \ell', \alpha'))
\end{equation*}
for all $k, \ell, k', \ell'\in \Spec$ and $\alpha, \alpha'\in \LL$ for
which $k\obar k'$ and $\ell\obar \ell'$ are defined.
Hence $d( k\obar k', \ell\obar \ell')\sqsubseteq_\LL P( d( k, \ell), d(
k', \ell'))$ for all such $k, \ell, k', \ell'\in \Spec$,
thus $P$ indeed bounds distances of synchronized labels.

Intuitively, $P$ gives us a \emph{uniform bound} on label composition:
distances between composed labels can be bounded above using $P$ and the
individual labels' distances.

\begin{definition}
  The \emph{structural composition} of two SMTS $S$ and $T$ is the SMTS
  $S\| T=( S\times T,( s_0, t_0), \mmayto_{ S\| T}, \mmustto_{ S\| T})$
  with transitions defined as follows:
  \begin{equation*}
    \dfrac{ s\mayto{ k}_S s' \quad t\mayto{ \ell}_T t' \quad k\obar \ell
      \text{ defined}} {( s, t)\mayto{ k\obar \ell}_{ S\| T}( s',
      t')}\qquad \dfrac{ s\mustto{ k}_S s' \quad t\mustto{ \ell}_T t'
      \quad k\obar \ell \text{ defined}} {( s, t)\mustto{ k\obar \ell}_{
        S\| T}( s', t')}
  \end{equation*}
\end{definition}

The next theorem shows that structural composition supports
\emph{quantitative independent implementability}: the distance between
structural compositions can bounded above using $P$ and the distances
between the individual components.

\begin{theorem}
  \label{wm2.th:struct}
  \quad For all SMTS $S$, $T$, $S'$ and $T'$ with $\md( S\| S', T\| T')\ne
  \top_\LL$, $\md( S\| S', T\| T')\sqsubseteq_\LL P( \md( S, T), \md(
  S', T'))$.
\end{theorem}

\begin{proof}
  Let $\cal R=\{ R_\alpha\subseteq S\times T\mid \alpha\in \LL\}$,
  $\cal R'=\{ R'_\alpha\subseteq S'\times T'\mid \alpha\in \LL\}$ be
  witnesses for $\md( S, T)$ and $\md( S', T')$, respectively, and
  define
  \begin{multline*}
    R^\|_\beta= \{(( s, s'),( t, t'))\in S\times S'\times T\times
    T'\mid
    \\
    \exists \alpha, \alpha'\in \LL:( s, t)\in R_\alpha\in \cal R,( s',
    t')\in R'_{ \alpha'}\in \cal R', P( \alpha,
    \alpha')\sqsubseteq_\LL \beta\}
  \end{multline*}
  for all $\beta\in \LL$.  We show that
  $\cal R^\|=\{ R^\|_\beta\mid \beta\in L\}$ witnesses
  $\md( S\| S', T\| T')\sqsubseteq_\LL P( \md( S, T), \md( S', T'))$.

  First,
  $(( s_0, s'_0),( t_0, t'_0))\in R^\|_{ P( \md( S, T), \md( S',
    T'))}$.  Let now $\beta\in \LL\setminus\{ \top_\LL\}$ and
  $(( s, s'),( t, t'))\in R^\|_\beta\in \cal R^\|$, then we have
  $\alpha, \alpha'\in \LL\setminus\{ \top_\LL\}$ with
  $( s, t)\in R_\alpha\in \cal R$,
  $( s', t')\in R'_{ \alpha'}\in \cal R'$, and
  $P( \alpha, \alpha')\sqsubseteq_\LL \beta$.

  Let $(s, s')\mayto{ k\obar k'}_{ S\| S'}( \bar s, \bar s')$, then
  $s\mayto{ k}_S \bar s$ and $s'\mayto{ k'}_{ S'} \bar s'$.  As
  $( s, t)\in R_\alpha\in \cal R$, we have $t\mayto{ \ell}_T \bar t$
  and $\bar \alpha\in \LL$ with
  $( \bar s, \bar t)\in R_{ \bar \alpha}\in \cal R$ and
  $F( k, \ell, \bar \alpha)\sqsubseteq_\LL \alpha$.  Similarly,
  $( s', t')\in R'_{ \alpha'}\in \cal R'$ implies that there is
  $t'\mayto{ \ell'}_{ T'} \bar t'$ and $\bar \alpha'\in \LL$ with
  $( \bar s', \bar t')\in R'_{ \bar \alpha'}\in \cal R'$ and
  $F( k', \ell', \bar \alpha')\sqsubseteq_\LL \alpha'$.

  Now if the composition $\ell\obar \ell'$ is undefined, then $\md( S\|
  S', T\| T')= \top_\LL$.  If it is defined, then
  we have $( t, t')\mayto{ \ell\obar \ell'}_{ T\| T'}( \bar t, \bar t')$
  by definition of $S\| S'$.  Also, $( \bar t, \bar t')\in R^\|_{ P(
    \bar \alpha, \bar \alpha')}\in \cal R^\|$ and
  \begin{equation*}
    F( k\obar k', \ell\obar
    \ell', P( \bar \alpha, \bar \alpha'))\sqsubseteq_\LL P( F( k, \ell,
    \bar \alpha), F( k', \ell', \bar \alpha'))\sqsubseteq_\LL P( \alpha,
    \alpha')\,.
  \end{equation*}
  The reverse direction, assuming a transition $( t, t')\mustto{
  \ell\obar \ell'}_{ T\| T'}( \bar t, \bar t')$, is \mbox{similar}. \qed
\end{proof}

\begin{example}
  One popular label synchronization operator for the set
  $\Spec= \Sigma\times \II$ from our examples, also used in
  Chapter~\ref{ch:weightedmodal}, is given by adding interval
  boundaries, \viz
  \begin{equation*}
    ( a,[ l, r])\obar( a',[ l', r'])=
    \begin{cases}
      ( a,[ l+ l', r+ r']) &\text{if } a= a', \\
      \text{undefined} &\text{otherwise}.
    \end{cases}
  \end{equation*}
  It can then be shown that
  \begin{equation}
    \label{wm2.eq:lemaddadd}
    d( k\obar k', \ell\obar \ell')\le d( k, \ell)+ d( k', \ell')
  \end{equation}
  for all $k, \ell, k', \ell'\in \Spec$ for which $k\obar k'$ and
  $\ell\obar \ell'$ are defined.

  For the accumulating distance, \eqref{wm2.eq:lemaddadd}~implies that
  $\obar$ is bounded above by $P( \alpha, \alpha')= \alpha+ \alpha'$:
  \begin{align*}
    F( k\obar k', \ell\obar \ell', \alpha+ \alpha') %
    &= d( k\obar k', \ell\obar \ell')+ \lambda( \alpha+ \alpha') \\
    &\le d( k, \ell)+ \lambda \alpha+ d( k', \ell')+ \lambda \alpha' \\
    &= F( k, \ell, \alpha)+ F( k', \ell', \alpha')
  \end{align*}
  Theorem~\ref{wm2.th:struct} thus specializes to
  Theorem~\ref{weightedmodal.th:indepimp}:
  $\md( S\| S', T\| T')\le \md( S, T)+ \md( S',
  T')$ for all SMTS $S, T, S', T'$.
\end{example}

\begin{example}
  Also for the point-wise distance, a bound is given by $P( \alpha,
  \alpha')= \alpha+ \alpha'$:
  \begin{align*}
    F( k\obar k', \ell\obar \ell', \alpha+ \alpha') %
    &= \max( d( k\obar k', \ell\obar \ell'), \alpha+ \alpha') \\
    &\le \max( d( k, \ell)+ d( k', \ell'), \alpha+ \alpha') \\
    &\le \max( d( k, \ell), \alpha)+ \max( d( k', \ell'), \alpha') \\
    &= F( k, \ell, \alpha)+ F( k', \ell', \alpha')\,,
  \end{align*}
  the last inequality because of distributivity of addition over
  maximum.  Thus also here, $\md( S\| S', T\| T')\le \md(
  S, T)+ \md( S', T')$ for all SMTS $S, T, S', T'$.
\end{example}

\begin{example}
  For the limit-average distance, a similar bound $P( \alpha, \alpha')=
  \alpha\oplus_\LL \alpha'$ works: For all $j\in \Nat$,
  \begin{align*}
    F( k\obar k', \ell\obar \ell', \alpha\oplus_\LL \alpha')( j)
    \hspace*{-5em}& \\
    &= \tfrac1{ j+ 1} d( k\obar k', \ell\obar \ell')+ \tfrac j{ j+ 1}(
    \alpha( j- 1)+ \alpha'( j- 1)) \\
    &\le \tfrac1{ j+ 1} d( k, \ell)+ \tfrac j{ j+ 1} \alpha( j- 1)+
    \tfrac1{ j+ 1} d( k', \ell')+ \tfrac j{ j+ 1} \alpha'( j- 1) \\
    &= F( k, \ell, \alpha)( j)+ F( k', \ell', \alpha')( j)\,.
  \end{align*}
  Hence also for the limit-average distance, we have
  $\md( S\| S', T\| T')\le \md( S, T)+ \md( S',
  T')$ for all SMTS $S, T, S', T'$.
\end{example}

\begin{example}
  In a real-time setting, a label synchronization operator which uses
  intersection of intervals instead of addition has been
  used~\cite{conf/fit/FahrenbergL12, DBLP:conf/icfem/BertrandLPR09}.
  That is,
  \begin{equation*}
    ( a,[ l, r])\obar( a',[ l', r'])=
    \begin{cases}
      ( a,[ \max( l, l'), \min( r, r')]) \hspace*{-5em} & \\
      &\text{if } a= a', \max( l, l')\le \min(
      r, r'), \\
      \text{undefined} &\text{otherwise}.
    \end{cases}
  \end{equation*}
  We show that for the maximum-lead distance, $\obar$ is bounded above
  by $P( \alpha, \alpha')= \max( \alpha, \alpha')$, that is,
  \begin{equation*}
    F( k\obar k', \ell\obar \ell', \max( \alpha, \alpha'))( d)\le \max(
    F( k, \ell, \alpha)( d), F( k', \ell', \alpha')( d))\,.
  \end{equation*}
  Applying the definition of $F$, we see that this is equivalent to
  \begin{multline*}
    \adjustlimits \sup_{ p\in \llbracket k\obar k'\rrbracket} \inf_{
      q\in \llbracket \ell\obar \ell'\rrbracket} \max(| d+ p- q|, \max(
    \alpha( d+ p- q), \alpha'( d+ p- q))) \\
    \le \max \left\{
      \begin{aligned}
        & \adjustlimits \sup_{ m\in \llbracket k\rrbracket} \inf_{ n\in
          \llbracket \ell\rrbracket} \max(| d+ m-
        n|, \alpha( d+ m- n)) \\
        & \adjustlimits \sup_{ m'\in \llbracket k'\rrbracket} \inf_{ n'\in
          \llbracket \ell'\rrbracket} \max(| d+ m'- n'|, \alpha'( d+ m'-
        n'))\,;
      \end{aligned}
    \right.
  \end{multline*}
  note that we are abusing notation by identifying $p=( a, x)$ with
  $x$ etc.  This inequality in turn is equivalent to
  \begin{equation*}
    \max \left\{
      \begin{aligned}
        & \adjustlimits \sup_{ p\in \llbracket k\obar k'\rrbracket}
        \inf_{
          q\in \llbracket \ell\obar \ell'\rrbracket}| d+ p- q| \\
        & \adjustlimits \sup_{ p\in \llbracket k\obar k'\rrbracket}
        \inf_{
          q\in \llbracket \ell\obar \ell'\rrbracket} \alpha( d+ p- q) \\
        & \adjustlimits \sup_{ p\in \llbracket k\obar k'\rrbracket}
        \inf_{ q\in \llbracket \ell\obar \ell'\rrbracket} \alpha'( d+ p-
        q)
      \end{aligned}
    \right.
        \le \max \left\{
      \begin{aligned}
        & \adjustlimits \sup_{ m\in \llbracket k\rrbracket} \inf_{ n\in
          \llbracket \ell\rrbracket}| d+ m- n| \\
        & \adjustlimits \sup_{ m\in \llbracket k\rrbracket} \inf_{ n\in
          \llbracket \ell\rrbracket} \alpha( d+ m- n) \\
        & \adjustlimits \sup_{ m'\in \llbracket k'\rrbracket} \inf_{ n'\in
          \llbracket \ell'\rrbracket}| d+ m'- n'| \\
        & \adjustlimits \sup_{ m'\in \llbracket k'\rrbracket} \inf_{ n'\in
          \llbracket \ell'\rrbracket} \alpha'( d+ m'- n')\,.
      \end{aligned}
    \right.
  \end{equation*}
  In this expression, the first line on the left-hand side is bounded
  by the right-hand side's first line, the second line on the left by
  the second line on the right, and the left-hand side's last line by
  the last line of the right-hand side, so that altogether, it holds.
  Theorem~\ref{wm2.th:struct} then translates to
  $\md( S\| S', T\| T')\le \max\big( \md( S, T), \md( S', T')\big)$.
\end{example}

\subsection{Quotient}

For \emph{quotients} of SMTS, we need a partial label operator
$\oslash: \Spec\times \Spec\parto \Spec$ for which it holds that
\begin{itemize}
\item for all $k, \ell, m\in \Spec$, $\ell\oslash k$ is defined and
  $m\labpre \ell\oslash k$ if and only if $k\obar m$ is
  defined and $k\obar m\labpre \ell$;
\item for all $\ell, \ell'\in \Spec$, $( \exists k\in \Spec: d( k,
  \ell)\ne \top_\LL, d( k, \ell')\ne \top_\LL)\liff( \exists m\in \Spec:
  m\oslash \ell, m\oslash \ell'\text{ are defined})$.
\end{itemize}
The first condition ensures that $\oslash$ is adjoint to $\obar$, and
the second relates it to distances just as we did for $\obar$ above.
Extending the first condition, we say that
\begin{itemize}
\item $\oslash$ is \emph{quantitatively well-behaved} if it holds for
  all $k, \ell, m\in \Spec$ that $\ell\oslash k$ is defined and $d( m,
  \ell\oslash k)\ne \top_\LL$ if and only if $k\obar m$ is defined and
  $d( k\obar m, \ell)\ne \top_\LL$, and in that case, $F( m,
  \ell\oslash k, \alpha)\sqsupseteq_\LL F( k\obar m, \ell, \alpha)$ for
  all $\alpha\in \LL$, and that
\item $\oslash$ is \emph{quantitatively exact} if the inequality can be
  sharpened to $F( m, \ell\oslash k, \alpha)= F( k\obar m, \ell,
  \alpha)$.
\end{itemize}
Both of these are useful quantitative generalization of the adjunction
between $\oslash$ and $\obar$; we will see examples below of
quantitatively exact and quantitatively well-behaved label quotients.

In the definition of quotient below, we denote by $\rho_B( S)$ the
\emph{pruning} of a SMTS $S$ with respect to the states in $B\subseteq
S$, see Section~\ref{weightedmodal.se:hull}.

\begin{definition}
  For SMTS $S$, $T$, the \emph{quotient} of $T$ by $S$ is the SMTS
  $T/ S= \rho_B( T\times S\cup\{ u\},( t_0, s_0), \mmayto_{
    T/ S}, \mmustto_{ T/ S})$ given as follows (if it
  exists):
  \begin{gather*}
    \dfrac{%
      t\mayto{ \ell}_T t' \quad s\mayto{ k}_S s' \quad \ell\oslash k
      \text{ defined}}{%
      ( t, s)\mayto{ \ell\oslash k}_{ T/ S}( t', s')} \qquad
    \dfrac{%
      t\mustto{ \ell}_T t' \quad s\mustto{ k}_S s' \quad \ell\oslash k
      \text{ defined}}{%
      ( t, s)\mustto{ \ell\oslash k}_{ T/ S}( t', s')} \\
    \dfrac{%
      t\mustto{ \ell}_T t' \quad \forall s\mustto{ k}_S s': \ell\oslash
      k \text{ undefined}}{%
      ( t, s)\in B} \\
    \dfrac{%
      m\in \Spec \quad \forall s\mayto{ k}_S s': k\obar m \text{
        undefined}}{%
      ( t, s)\mayto{ m}_{ T/ S} u} \qquad \dfrac{%
      m\in \Spec}{%
      u\mayto{ m}_{ T/ S} u}
  \end{gather*}
\end{definition}

In the above definition, $u$ is a new \emph{universal} state from which
everything is allowed and nothing required (last SOS rule).  This state
is reached from a quotient state $( t, s)$ under label $m$ whenever
there is no \may~transition from $s$ with whose label $m$ can synchronize
(next-to-last SOS rule), because in that case, any transition in the
quotient will be canceled in the structural composition
(\cf~Theorem~\ref{wm2.th:quotient} below), and we need the quotient to be
maximal.  Similarly, if $t$ specifies a \must~transition under a label
$\ell$ which cannot be matched by any transition from $s$, then the
quotient state $( t, s)$ is \emph{inconsistent}; hence we add it to $B$
and remove it when pruning.

The next theorem shows that under certain standard conditions, quotient
is \emph{sound} and \emph{maximal} with respect to structural
composition.

\begin{theorem}
  \label{wm2.th:quotient}
  Let $S$, $T$, $X$ be
  SMTS such that $S$ is deterministic and $T/ S$ exists.  Then
  $X\mr T/ S$ if and only if $S\| X\mr T$.  Also,
  \begin{itemize}
  \item if $\oslash$ is quantitatively well-behaved, then $\md( X,
    T/ S)\sqsupseteq_\LL \md( S\| X, T)$;
  \item if $\oslash$ is quantitatively exact and $\md( X, T/
    S)\ne \top_\LL$, then $\md( X, T/ S)= \md( S\| X, T)$.
  \end{itemize}
\end{theorem}

The (Boolean) property that $X\mr T/ S$ iff $S\| X\mr T$
implies \emph{uniqueness} of
quotient~\cite{DBLP:conf/models/FahrenbergLW11}.  For the quantitative
generalizations, the property induced by a well-behaved $\oslash$
means that distances to the quotient bound distances of structural
compositions, which can be useful in further calculations; similarly
for exact $\oslash$.  Note that uniqueness implies that if a certain
instantiation of our framework admits a quotient which is not
quantitatively well-behaved, there is no hope that one can find
another one which is.

\begin{proof}
  The proof that $X\mr T/ S$ if and only if $S\| X\mr T$ is
  in~\cite{DBLP:journals/mscs/BauerJLLS12}.  For the other properties,
  assume first $\oslash$ to be quantitatively well-behaved; we show
  that $\md( S\| X, T)\sqsubseteq_\LL \md( X, T/ S)$.  If
  $\md( X, T/ S)= \top_\LL$, there is nothing to prove, so
  assume $\md( X, T/ S)\ne \top_\LL$ and let
  $\cal R=\{ R_\alpha\subseteq X\times( T\times S\cup\{ u\})\}$ be a
  witness for $\md( X, T/ S)$.  Define
  $R'_\alpha=\{(( s, x), t)\mid( x,( t, s))\in R_\alpha\}\subseteq
  S\times X\times T$ for all $\alpha\in \LL$ and collect these to a
  family $\cal R'=\{ R'_\alpha\mid \alpha\in \LL\}$.  We show that
  $\cal R'$ is a witness for
  $\md( S\| X, T)\sqsubseteq_\LL \md( X, T/ S)$.

  We have
  $(( s_0, x_0), t_0)\in R'_{ \md( X, T/ S)}\in \cal R'$, so
  let $\alpha\in \LL$ and $(( s, x), t)\in R'_\alpha\in \cal R'$, and
  assume first that $( s, x)\mayto{ k\obar m}_{ S\| X}( s', x')$.
  Then $s\mayto{ k}_S s'$ and $x\mayto{ m}_X x'$ by definition of
  $S\| X$.  Now $( x,( t, s))\in R_\alpha\in \cal R$ implies that there is
  $( t, s)\mayto{ \ell\oslash k'}_{ T/ S}( t', s'')$ and
  $\alpha'\in \LL$ for which
  $F( m, \ell\oslash k', \alpha')\sqsubseteq_\LL \alpha$ and
  $( x',( t', s''))\in R_{ \alpha'}\in \cal R$.  But then also
  $(( s'', x'), t')\in R'_{ \alpha}\in \cal R'$, hence $k'\obar m$ is
  defined and
  $F( k'\obar m, \ell, \alpha')\sqsubseteq_\LL F( m, \ell\oslash k',
  \alpha')\sqsubseteq_\LL \alpha$.

  Now $k\obar m$ and $k'\obar m$ being defined implies that there is
  $k''$ for which $d( k'', k)\ne \top_\LL$ and
  $d( k'', k')\ne \top_\LL$, and by definition of $T/ S$,
  $s\mayto{ k'}_S s''$.  As $S$ is deterministic, this implies $k= k'$
  and $s'= s''$.  Hence $(( s', x'), t')\in R'_{ \alpha'}\in \cal R'$
  and $F( k\obar m, \ell, \alpha')\sqsubseteq_\LL \alpha$.

  Assume now that $t\mustto{ \ell}_T t'$.  We must have
  $s\mustto{ k}_S s'$ for which $\ell\oslash k$ is defined, for
  otherwise $( t, s)\in B$ and hence $( t, s)$ would have been pruned
  in $T/ S$.  Thus
  $( t, s)\mustto{ \ell\oslash k}_{ T/ S}( t', s')$, which by
  $( x,( t, s))\in R_\alpha\in \cal R$ implies that there is
  $x\mustto{ m}_X x'$ and $\alpha'\in \LL$ for which
  $F( m, \ell\oslash k, \alpha')\sqsubseteq_\LL \alpha$ and
  $( x',( t', s'))\in R_{ \alpha'}\in \cal R$, hence
  $(( s', x'), t')\in R'_{ \alpha'}\in \cal R'$.  But then $k\obar m$
  is defined and
  $F( k\obar m, \ell, \alpha')\sqsubseteq_\LL F( m, \ell\oslash k,
  \alpha')\sqsubseteq_\LL\alpha$, and
  $( s, x)\mustto{ k\obar m}_{ S\| X}( s', x')$.

  In order to prove the theorem's last claim, let $\oslash$ be
  quantitatively exact.  To show that
  $\md( X, T/ S)\sqsubseteq_\LL \md( S\| X, T)$, assume that
  $\md( S\| X, T)\ne \top_\LL$ (otherwise there is nothing to prove),
  let
  $\cal R=\{ R_\alpha\subseteq S\times X\times T\mid \alpha\in \LL\}$
  be a witness for $\md( S\| X, T)$, and define
  $R'_\alpha=\{( x,( t, s))\mid(( s, x), t)\in R_\alpha\}\cup\{( x,
  u)\mid x\in X\}\subseteq X\times( T\times S\cup\{ u\})$ for all
  $\alpha\in \LL$.  We show that
  $\cal R'=\{ R'_\alpha\mid \alpha\in \LL\}$ is a witness for
  $\md( X, T/ S)\sqsubseteq_\LL \md( S\| X, T)$.

  We have $( x_0,( t_0, s_0))\in R'_{ \md( S\| X, T)}\in \cal R'$.
  Let $\alpha\in \LL$, $( x, u)\in R'_\alpha\in \cal R'$ and
  $x\mayto{ m}_X x'$, then also $u\mayto{ m}_{ T/ S} u$,
  $F( m, m, \bot_\LL)\sqsubseteq \alpha$, and
  $( x', u)\in R'_{ \bot_\LL}\in \cal R'$.  Now let
  $( x,( t, s))\in R'_\alpha\in \cal R'$ and $x\mayto{ m}_X x'$.  If
  $k\obar m$ is undefined for all $s\mayto{ k}_S s'$, then by
  definition of $T/ S$, $( t, s)\mayto{ m}_{ T/ S} u$,
  $F( m, m, \bot_\LL)\sqsubseteq \alpha$, and
  $( x', u)\in R'_{ \bot_\LL}\in \cal R'$.

  If there is a transition $s\mayto{ k}_S s'$ for which $k\obar m$ is
  defined (by determinism there can be at most one), then also
  $( s, x)\mayto{ k\obar m}_{ S\| X}( s', x')$.  As
  $(( s, x), t)\in R_\alpha\in \cal R$, we must have
  $t\mayto{ \ell} t'$ and $\alpha'\in \LL$ with
  $F( k\obar m, \ell, \alpha')\sqsubseteq_\LL \alpha$ and
  $(( s', x'), t')\in R_{ \alpha'}\in \cal R$, hence
  $( x',( t', s'))\in R'_{ \alpha'}\in \cal R'$.  Then $\ell\oslash k$ is
  defined and $F( m, \ell\oslash k, \alpha')\sqsubseteq_\LL \alpha$,
  and by definition of $T/ S$,
  $( t, s)\mayto{ \ell\oslash k}_{ T/ S}( t', s')$.

  Now assume that $( t, s)\mustto{ \ell\oslash k}_{ T/ S}( t',
  s')$, then $t\mustto{ \ell}_T t'$ and $s\mustto{ k}_S s'$ by
  definition of $T/ S$.  By $(( s, x), t)\in R_\alpha\in \cal R$, we
  have $( s, x)\mustto{ k'\obar m}_{ S\| X}( s'', x')$ and $\alpha'\in
  \LL$ with $F( k'\obar m, \ell, \alpha')\sqsubseteq_\LL \alpha$ and $((
  s'', x'), t')\in R_{ \alpha'}\in \cal R$.  This in turn implies that
  $s\mustto{ k'}_S s''$ and $x\mustto{ m}_X s'$ by definition of $S\|
  X$.  We also see that $\ell\oslash k'$ is defined, which by
  determinism of $S$ entails $k= k'$ and $s'= s''$.  Hence $F( k\obar m,
  \ell, \alpha')\sqsubseteq_\LL \alpha$ and $( x',( t', s'))\in R'_{
    \alpha'}\in \cal R'$. \qed
\end{proof}

\begin{example}
  For the label synchronization operator for $\Spec= \Sigma\times \II$
  given by adding interval boundaries, a quotient can be defined by
  \begin{equation*}
    ( a',[ l', r'])\oslash( a,[ l, r])=
    \begin{cases}
      ( a,[ l'- l, r'- r]) &\text{if } a= a' \text{ and } l'- l\le r'-
      r, \\
      \text{undefined} &\text{otherwise}.
    \end{cases}
  \end{equation*}
  It can then be shown that
  $d( m, \ell\oslash k)= d( k\obar m, \ell)$ for all
  $k, \ell, m\in \Spec$ for which both $\ell\oslash k$ and $k\obar m$
  are defined, see Chapter~\ref{ch:weightedmodal}.  From this it
  easily follows that both for the accumulating, the point-wise, and
  the limit-average distance, $\oslash$ is quantitatively exact,
  hence for all three distances, Theorem~\ref{wm2.th:quotient}
  specializes to the theorem that
  $\md( X, T/ S)= \md( S\| X, T)$ for all SMTS
  $S$, $T$, $X$ for which $S$ is deterministic, $T/ S$ exists
  and $\md( X, T/ S)\ne \infty$.  For the accumulating
  distance, this is Theorem~\ref{weightedmodal.th:soundmaxquot}.
\end{example}

\begin{figure}[tbp]
  \centering
  \begin{tikzpicture}[scale=.7]
    \tikzstyle{every node}=[font=\small,anchor=base]
    \tikzset{rangebar/.style={color=gray,very thick}}

    \begin{scope}
      \node (l) at (0,0) {$l$};
    \node (l') at (1,0) {$l'$};
    \node (r) at (2,0) {$r$};
    \node (r') at (3,0) {$r'$};
    \path[rangebar] (l.south west) edge (r.south east);
    \path[rangebar] ($(l'.south west)+(0,-1.5mm)$) edge ($(r'.south)+(0,-1.5mm)$);
    \path[rangebar] ($(l'.south west)+(0,-4mm)$) edge ($(r'.south)+(0mm,-4mm)$)
    edge [dotted] ($(r'.south west)+(10mm,-4mm)$);
    \end{scope}
    
    \begin{scope}[xshift=6cm]
      \node (l) at (0,0) {$l$};
    \node (l') at (1,0) {$l'$};
    \node (r') at (2,0) {$r'$};
    \node (r) at (3,0) {$r$};
    \path[rangebar] (l.south west) edge (r.south east);
    \path[rangebar] ($(l'.south west)+(0,-1.5mm)$) edge ($(r'.south)+(0,-1.5mm)$);
    \path[rangebar] ($(l'.south west)+(0,-4mm)$) edge ($(r'.south)+(0mm,-4mm)$);
    \end{scope}

    \begin{scope}[xshift=12cm]
      \node (l) at (0,0) {$l$};
    \node (r) at (1,0) {$r$};
    \node (l') at (2,0) {$l'$};
    \node (r') at (3,0) {$r'$};
    \path[rangebar] (l.south west) edge (r.south east);
    \path[rangebar] ($(l'.south west)+(0,-1.5mm)$) edge ($(r'.south)+(0,-1.5mm)$);
    \end{scope}

    \begin{scope}[yshift=-1.6cm]
      \begin{scope}
        \node (l') at (0,0) {$l'$};
        \node (l) at (1,0) {$l$};
        \node (r) at (2,0) {$r$};
        \node (r') at (3,0) {$r'$};
        \path[rangebar] (l.south west) edge (r.south east);
        \path[rangebar] ($(l'.south west)+(0,-1.5mm)$) edge
        ($(r'.south)+(0,-1.5mm)$);
        \path[rangebar] ($(l'.south
        west)+(-8mm,-4mm)$)  edge[dotted]  ($(l'.south
        west)+(0mm,-4mm)$);
        \path[rangebar]   ($(l'.south
        west)+(0mm,-4mm)$) edge ($(r'.south)+(0mm,-4mm)$) edge [dotted]
        ($(r'.south west)+(10mm,-4mm)$);
      \end{scope}
    
    \begin{scope}[xshift=6cm]
      \node (l') at (0,0) {$l'$}; \node (l) at (1,0) {$l$}; \node (r')
      at (2,0) {$r'$}; \node (r) at (3,0) {$r$}; \path[rangebar]
      (l.south west) edge (r.south east); \path[rangebar] ($(l'.south
      west)+(0,-1.5mm)$) edge ($(r'.south)+(0,-1.5mm)$);
      \path[rangebar] ($(l'.south west)+(-8mm,-4mm)$) edge[dotted] ($(l'.south west)+(0mm,-4mm)$); 
      \path[rangebar] ($(l'.south west)+(0mm,-4mm)$) edge
      ($(r'.south)+(0mm,-4mm)$);
    \end{scope}

    \begin{scope}[xshift=12cm]
      \node (l') at (0,0) {$l'$}; \node (r') at (1,0) {$r'$}; \node (l)
      at (2,0) {$l$}; \node (r) at (3,0) {$r$}; \path[rangebar]
      (l.south west) edge (r.south east); \path[rangebar] ($(l'.south
      west)+(0,-1.5mm)$) edge ($(r'.south)+(0,-1.5mm)$);
    \end{scope}
  \end{scope}
  \end{tikzpicture}
  \caption{%
    \label{fi:quotient-mecs}
    Quotient $[ l', r']\oslash[ l, r]$ of intervals in
    Example~\ref{ex:quot-mecs}, six cases.  Top bar: $[ l, r]$; middle
    bar: $[ l', r']$; bottom bar: quotient.  Note that for the two
    cases on the right, quotient is undefined.  }
\end{figure}
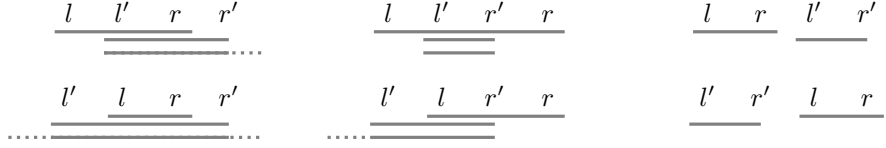

\begin{example}
  \label{ex:quot-mecs}
  For the variant of the operator $\obar$ which uses intersection of
  intervals instead of addition, a quotient can be defined as follows:
  \begin{equation*}
    ( a',[ l', r'])\oslash( a,[ l, r])=
    \begin{cases}
      \text{undefined} &\text{if } a\ne a'\,, \\
      ( a,[ l', \infty]) &\text{if } a= a'\text{ and } l< l'\le r\le
      r'\,,
      \\
      ( a,[ l', r']) &\text{if } a= a'\text{ and } l< l'\le r'< r\,, \\
      \text{undefined} &\text{if } a= a'\text{ and } l\le r< l'\le r'\,, \\
      ( a,[ 0, \infty]) &\text{if } a= a'\text{ and } l'\le l\le
      r\le r'\,, \\
      ( a,[ 0, r']) &\text{if } a= a'\text{ and } l'\le l\le r< r'\,,
      \\
      \text{undefined} &\text{if } a= a'\text{ and } l'\le r'< l\le
      r\,.
    \end{cases}
  \end{equation*}
  The intuition is that to obtain the maximal solution $[ p, q]$ to an
  equation $[ l, r]\obar[ p, q]\sqsubseteq_\Spec[ l', r']$, whether
  $p$ and $q$ must restrain the interval in the intersection, or can
  be $0$ and $\infty$, respectively, depends on the position of
  $[ l, r]$ relative to $[ l', r']$,
  \cf~Figure~\ref{fi:quotient-mecs}.

  It can be shown~\cite{conf/fit/FahrenbergL12} that for the
  maximum-lead distance, this variant of $\oslash$ is quantitatively
  well-behaved, but not quantitatively exact.
  Theorem~\ref{wm2.th:quotient} hence translates to the fact that for
  all SMTA $S$, $T$, $X$ for which $S$ is deterministic and $T/ S$
  exists, $\md( X, T/ S)\ge \md( S\| X, T)$.
\end{example}

\section{Conjunction}
\label{wm2.se:conj}

Conjunction of SMTS can be used to merge two specifications into one.
Let $\oland: \Spec\times \Spec\parto \Spec$ be a partial label operator
for which it holds that
\begin{itemize}
\item for all $k, \ell\in \Spec$, if $k\oland \ell$ is defined, then
  $k\oland \ell\labpre k$, $k\oland \ell\labpre
  \ell$, and
\item for all $\ell, \ell'\in \Spec$, $( \exists k\in \Spec: d( k,
  \ell)\ne \top_\LL, d( k, \ell')\ne \top_\LL)\liff( \exists m\in \Spec:
  \ell\oland m, \ell'\oland m\text{ are defined})$.
\end{itemize}
The first requirement above ensures that conjunction acts as a lower
bound, and the second one relates it to distances such that two labels
have a common quantitative refinement if and only if they have a common
conjunction.  One also usually wants conjunction to be a \emph{greatest}
lower bound; we say that $\oland$ is \emph{conjunctively compositional}
if it holds for all $k, \ell, m\in \Spec$ for which $m\labpre
k$ and $m\labpre \ell$ that also $k\oland \ell$ is defined
and $m\labpre k\oland \ell$.

As a quantitative generalization, and analogously to what we did for
structural composition, we say that $\oland$ is \emph{conjunctively
  bounded} by a function $C: \LL\times \LL\to \LL$ if $C$ is monotone in
both coordinates, has $C( \bot_\LL, \bot_\LL)= \bot_\LL$, $C( \alpha,
\top_\LL)= C( \top_\LL, \alpha)= \top_\LL$ for all $\alpha\in \LL$, and if
it holds for all $k, \ell, m\in \Spec$ for which $d( m, k)\ne \top_\LL$
and $d( m, \ell)\ne \top_\LL$ that $k\oland \ell$ is defined and
\begin{equation*}
  F( m, k\oland \ell, C( \alpha, \alpha'))\sqsubseteq_\LL C( F( m, k,
  \alpha), F( m, \ell, \alpha'))
\end{equation*}
for all $\alpha, \alpha'\in \LL$.  Note that this implies that $d( m,
k\oland \ell)\sqsubseteq_\LL C( d( m, k), d( m, \ell))$, hence
conjunctive boundedness implies conjunctive
compositionality.  Like $P$ for structural composition, $C$ gives a
uniform bound on label conjunction.

\begin{definition}
  The \emph{conjunction} of two SMTS $S$ and $T$ is the SMTS $S\wedge
  T=\rho_B( S\times T,( s_0, t_0), \mmayto_{ S\wedge T}, \mmustto_{
    S\wedge T})$ given as follows:
  \begin{gather*}
    \dfrac{%
      s\mustto{ k}_S s'\quad t\mayto{ \ell}_T t'\quad k\oland
      \ell\text{ defined}%
    }{%
      ( s, t)\mustto{ k\oland \ell}_{ S\wedge T}( s', t')%
    } \qquad%
    \dfrac{%
      s\mayto{ k}_S s'\quad t\mustto{ \ell}_T t'\quad k\oland
      \ell\text{ defined}%
    }{%
      ( s, t)\mustto{ k\oland \ell}_{ S\wedge T}( s', t')%
    } \\%
    \dfrac{%
      s\mayto{ k}_S s'\quad t\mayto{ \ell}_T t'\quad k\oland \ell\text{
        defined}%
    }{%
      ( s, t)\mayto{ k\oland \ell}_{ S\wedge T}( s', t')%
    } \\%
    \dfrac{%
      s\mustto{ k}_S s'\quad \forall t\mayto{ \ell}_T t': k\oland
      \ell\text{ undef.}%
    }{%
      ( s, t)\in B} \qquad%
    \dfrac{%
      t\mustto{ \ell}_T t'\quad \forall s\mayto{ k}_S s': k\oland
      \ell\text{ undef.}%
    }{%
      ( s, t)\in B}
  \end{gather*}
\end{definition}

Note that like for quotient, conjunction of SMTS may give inconsistent
states which need to be pruned after.  As seen in the last two SOS
rules above, this is the case when one SMTS specifies a \must
transition with which the other SMTS cannot synchronize; then, the
demand on implementations would be that they simultaneously
\emph{must} and \emph{cannot} have a transition, which of course is
unsatisfiable.

The next theorem shows the precise conditions under which conjunction is
a greatest lower bound.  Note that the greatest-lower-bound condition
$U\mr S$, $U\mr T\limpl U\mr S\wedge T$ entails uniqueness.

\begin{theorem}
  \label{wm2.th:conj}
  Let $S$, $T$, $U$ be 
  SMTS.  If $S\wedge T$ is defined, then $S\wedge T\mr S$ and $S\wedge
  T\mr T$.  If, additionally, $S$ or $T$ are deterministic, then:
  \begin{itemize}
  \item If $\oland$ is conjunctively compositional, $U\mr S$, and
    $U\mr T$, then $S\wedge T$ is defined and $U\mr S\wedge T$.
  \item If $\oland$ is conjunctively bounded by $C$, $\md( U, S)\ne
    \top_\LL$, and $\md( U, T)\ne\top_\LL$, then $S\wedge T$ is defined
    and $\md( U, S\wedge T)\sqsubseteq_\LL C( \md( U, S), \md( U, T))$.
  \end{itemize}
\end{theorem}

\begin{proof}
  The proof of the two first claims is
  in~\cite{DBLP:journals/mscs/BauerJLLS12}.  For the third claim, let
  $\cal R=\{ R_\alpha\subseteq U\times S\mid \alpha\in \LL\}$ and
  $\cal R'=\{ R'_\alpha\subseteq U\times T\mid \alpha\in \LL\}$ be relation
  families witnessing $\md( U, S)$ and $\md( U, T)$, respectively,
  define
  $R^\wedge_\beta=\{( u,( s, t))\mid \exists \alpha, \alpha'\in \LL:(
  u, s)\in R_\alpha,( u, t)\in R'_{ \alpha'}, C( \alpha,
  \alpha')\sqsubseteq_\LL \beta\}\subseteq U\times S\times T$ for all
  $\beta\in \LL$, and let
  $\cal R^\wedge=\{ R^\wedge_\beta\mid \beta\in \LL\}$.  We show that
  $\cal R^\wedge$ is a witness for
  $\md( U, S\wedge T)\sqsubseteq_\LL C( \md( U, S), \md( U, T))$.

  We have $( u_0,( s_0, t_0))\in R^\wedge_{ C( \md( U, S), \md( U,
    T))}\in \cal R^\wedge$.  Let $\beta\in \LL\setminus\{ \bot_\LL\}$ and $(
  u,( s, t))\in R^\wedge_\beta\in \cal R^\wedge$, then we have $\alpha,
  \alpha'\in \LL\setminus\{ \bot_\LL\}$ with $( u, s)\in R_\alpha\in \cal R$,
  $( u, t)\in R'_{ \alpha'}\in \cal R'$, and $C( \alpha,
  \alpha')\sqsubseteq_\LL \beta$.

  Assume $u\mayto{ m}_U u'$, then there exist $s\mayto{ k}_S s'$ and
  $\bar \alpha\in \LL$ for which
  $( u', s')\in R_{ \bar \alpha}\in \cal R$ and
  $F( m, k, \bar \alpha)\sqsubseteq_\LL \alpha$, and similarly
  $t\mayto{ \ell}_T t'$ and $\bar \alpha'$ with
  $( u', t')\in R'_{ \bar \alpha'}\in \cal R'$ and
  $F( m, \ell, \bar \alpha')\sqsubseteq \alpha'$.  Then
  $d( m, k)\ne \top_\LL$ and $d( m, \ell)\ne \top_\LL$, so by
  conjunctive boundedness $k\oland \ell$ is defined, and
  $( s, t)\mayto{ k\oland \ell}_{ S\wedge T}( s', t')$ by definition
  of $S\wedge T$.  Also,
  $( u',( s', t'))\in R^\wedge_{ C( \bar \alpha, \bar \alpha')}\in
  \cal R^\wedge$ and
  $F( m, k\oland \ell, C( \bar \alpha, \bar \alpha'))\sqsubseteq_\LL
  C( F( m, k, \bar \alpha), F( m, \ell, \bar \alpha'))\sqsubseteq_\LL
  C( \alpha, \alpha')$.

  Assume $( s, t)\mustto{ k\oland \ell}_{ S\wedge T}( s', t')$, then
  $s\mustto{ k}_S s'$ and $t\mustto{ \ell}_T t'$ by definition of
  $S\wedge T$.  We can without loss of generality postulate that $T$
  is deterministic.  The fact that $( u, s)\in R_\alpha\in \cal R$
  implies that there are $u\mustto{ m}_U u'$ and $\bar \alpha\in \LL$
  for which $( u', s')\in R_{ \bar \alpha}\in \cal R$ and
  $F( m, k, \bar \alpha)\sqsubseteq_\LL \alpha$.  We must also have
  $u\mayto{ m'}_U u'$ for some $m'\sqsupseteq_\Spec m$, and then
  $( u, t)\in R'_{ \bar \alpha}\in \cal R'$ implies that there exist
  $t\mayto{ \ell'}_T t''$ and $\bar \alpha'\in \LL$ with
  $( u', t'')\in R'_{ \bar \alpha'}\in \cal R'$ and
  $F( m', \ell', \bar \alpha')\sqsubseteq_\LL \alpha'$.

  The triangle inequality for $F$ gives
  \begin{equation*}
    F( m, \ell', \bar
    \alpha')\sqsubseteq_\LL F( m, m', \bot_\LL)\oplus_\LL F( m', \ell', \bar
    \alpha')\sqsubseteq_\LL \alpha'\,,
  \end{equation*}
  hence $d( m, \ell')\ne \top_\LL$.  Together with
  $d( m, k)\ne \top_\LL$, conjunctive boundedness allows us to
  conclude that $k\oland \ell'$ is defined, but then both
  $k\oland \ell$ and $k\oland \ell'$ are defined, hence by
  determinism of $T$, $\ell= \ell'$ and $t'= t''$. \qed
\end{proof}

\begin{example}
  For the set $\Spec= \Sigma\times \II$ from our examples, the unique
  compositional conjunction operator on is given, on labels, by
  intersection of intervals:
  \begin{equation*}
    ( a,[ l, r])\oland( a',[ l', r'])=
    \begin{cases}
      ( a,[ \max( l, l'), \min( r, r')]) \hspace*{-5em}& \\
      &\text{if } a= a',
      \max( l, l')\le \min( r, r'), \\
      \text{undefined} &\text{otherwise}.
    \end{cases}
  \end{equation*}
  We can easily show that $\oland$ is \emph{not} conjunctively
  bounded: with $m=( a,[ 2, 2])$, $k=( a,[ 0, 1])$ and
  $\ell=( a,[ 3, 4])$, we have $d( m, k)= d( m, \ell)= 1$, but
  $k\oland \ell$ is not defined.  Noting that this statement does not
  involve the distance iterator $F$, we conclude that neither
  accumulating, point-wise nor limit-average distance admit a bounded
  conjunction operator.  For the accumulating distance, this statement
  is Theorem~\ref{weightedmodal.th:no-conj}.
\end{example}


To deal with the problem that, as in the above example, conjunction may
not be conjunctively bounded, we introduce another, weaker, property
which ensures some compatibility of conjunction with distances.  We say
that $\oland$ is \emph{relaxed conjunctively bounded} by a function $C:
\LL\times \LL\to \LL$ if $C$ is monotone in both coordinates, has $C(
\bot_\LL, \bot_\LL)= \bot_\LL$, $C( \alpha, \top_\LL)= C( \top_\LL,
\alpha)= \top_\LL$ for all $\alpha\in \LL$, and such that
for all $k, \ell\in \Spec$ for which there is $m\in \Spec$ with $d( m,
k)\ne \top_\LL$ and $d( m, \ell)\ne \top_\LL$, there exist $k', \ell'\in
\Spec$ with $k'\oland \ell'$ defined, $k\labpre k'$,
$\ell\labpre \ell'$, $d( k', k)\ne \top_\LL$, and $d( \ell',
\ell)\ne \top_\LL$, such that for all $m'\in \Spec$, $\alpha, \alpha'\in
\LL$,
\begin{equation}
  \label{wm2.eq:relconb}
  F( m', k'\oland \ell', C( \alpha, \alpha'))\sqsubseteq_\LL C( F( m',
  k, \alpha), F( m', \ell, \alpha')).
\end{equation}

The following theorem shows that relaxed boundedness of $\oland$
entails a similar property for SMTS conjunction.

\begin{theorem}
  \label{wm2.th:relax_conj}
  Let $S$, $T$ be SMTS with $S$ or $T$ deterministic and $\oland$
  relaxed conjunctively bounded by $C$.  If there is an SMTS $U$ for
  which $\md( U, S)\ne \top_\LL$ and $\md( U, T)\ne \top_\LL$, then
  there exist $\beta$- and $\gamma$-widenings $S'$ of $S$ and $T'$ of
  $T$ such that $S'\wedge T'$ is defined, and
  $\md( U', S'\wedge T')\sqsubseteq_\LL C( \md( U', S), \md( U', T))$
  for all SMTS $U'$.
\end{theorem}

\begin{proof}
  We start by constructing $S'$ and $T'$, almost as in the proof of the
  third claim of Theorem~\ref{wm2.th:conj}.  The states of $S'$ and $T'$
  will be the same as for $S$ and $T$, and we start by letting $\beta=
  \bot_\LL$, $\gamma= \bot_\LL$.

  Let $U$ fulfill $\md( U, S)\ne \top_\LL$ and
  $\md( U, T)\ne \top_\LL$, let
  $\cal R=\{ R_\alpha\subseteq U\times S\mid \alpha\in \LL\}$ and
  $\cal R'=\{ R'_\alpha\subseteq U\times T\mid \alpha\in \LL\}$ be
  relation families witnessing $\md( U, S)$ and $\md( U, T)$,
  respectively, define
  $R^\wedge_\eta=\{( u,( s, t))\mid \exists \alpha, \alpha'\in \LL:(
  u, s)\in R_\alpha,( u, t)\in R'_{ \alpha'}, C( \alpha,
  \alpha')\sqsubseteq_\LL \eta\}\subseteq U\times S\times T$ for all
  $\eta\in \LL$, and let
  $\cal R^\wedge=\{ R^\wedge_\eta\mid \eta\in \LL\}$.

  Now let $\eta\in \LL\setminus\{ \top_\LL\}$ and
  $( u,( s, t))\in R^\wedge_\eta\in \cal R^\wedge$, then we have
  $\alpha, \alpha'\in \LL\setminus\{ \bot_\LL\}$ with
  $( u, s)\in R_\alpha\in \cal R$,
  $( u, t)\in R'_{ \alpha'}\in \cal R'$, and
  $C( \alpha, \alpha')\sqsubseteq_\LL \eta$.  Let $u\mayto{ m}_U u'$,
  then also $s\mayto{ k}_S s'$ and $t\mayto{ \ell}_T t'$, and there
  are $\bar \alpha, \bar \alpha'\in \LL\setminus\{ \top_\LL\}$ with
  $F( m, k, \bar \alpha)\sqsubseteq_\LL \alpha$ and
  $F( m, \ell, \bar \alpha')\sqsubseteq_\LL \alpha'$.  Hence
  $d( m, k)\ne \top_\LL$ and $d( m, \ell)\ne \top_\LL$, and by relaxed
  conjunctive boundedness we have $k', \ell'\in \Spec$ with
  $k\labpre k'$, $\ell\labpre \ell'$,
  $d( k', k)\ne \top_\LL$, $d( \ell', \ell)\ne \top_\LL$, and
  $k'\oland \ell'$ defined.  We add the transitions
  $s\mayto{ k'}_{ S'} s'$, $t\mayto{ \ell'}_{ T'} t'$ to $S'$ and $T'$
  and update $\beta:= \max( \beta, d( k', k))$,
  $\gamma:= \max( \gamma, d( \ell', \ell))$.

  As the sets $\{ k\in \Spec\mid s\mayto{ k}_S s'\}$, $\{ \ell\in
  \Spec\mid t\mayto{ \ell}_T t'\}$ are compact, the above process
  converges to some $\beta, \gamma\ne \top_\LL$.  The \must transitions
  we just copy from $S$ to $S'$ and from $T$ to $T'$, and then $S'$ is a
  $\beta$-widening of $S$ and $T'$ is a $\gamma$-widening of $T$.

  We must show that $S'$ and $T'$ satisfy the properties claimed.  By
  construction $S'\wedge T'$ is defined, so let $U'$ be an SMTS with
  $\md( U', S)\ne \top_\LL$ and $\md( U', T)\ne \top_\LL$ (otherwise
  we have nothing to prove).  We must show that
  $\md( U', S'\wedge T')\sqsubseteq_\LL C( \md( U', S), \md( U', T))$.
  Let $\cal R=\{ R_\alpha\subseteq U'\times S\mid \alpha\in \LL\}$ and
  $\cal R'=\{ R'_\alpha\subseteq U'\times T\mid \alpha\in \LL\}$ be
  relation families witnessing $\md( U', S)$ and $\md( U', T)$,
  respectively, define
  $R^{ \wedge'}_\eta=\{( u',( s, t))\mid \exists \alpha, \alpha'\in
  \LL:( u', s)\in R_\alpha,( u', t)\in R'_{ \alpha'}, C( \alpha,
  \alpha')\sqsubseteq_\LL \eta\}\subseteq U'\times S\times T$ for all
  $\eta\in \LL$, and let
  $\cal R^{ \wedge'}=\{ R^{ \wedge'}_\eta\mid \eta\in \LL\}$.

  We have
  $( u_0',( s_0, t_0))\in R^{ \wedge'}_{ C( \md( U', S), \md( U',
    T))}\in \cal R^{ \wedge'}$.  Let
  $\eta\in \LL\setminus\{ \bot_\LL\}$ and
  $( u',( s, t))\in R^{ \wedge'}_\eta$, then we have
  $\alpha, \alpha'\in \LL\setminus\{ \bot_\LL\}$ with
  $( u', s)\in R_\alpha\in \cal R$,
  $( u', t)\in R'_{ \alpha'}\in \cal R'$, and
  $C( \alpha, \alpha')\sqsubseteq_\LL \eta$.  Let
  $u'\mayto{ m}_{ U'} u''$, then also $s\mayto{ k}_S s'$ and
  $t\mayto{ \ell}_T t'$, and there are
  $\bar \alpha, \bar \alpha'\in \LL\setminus\{ \top_\LL\}$ with
  $( u'', s')\in R_{ \bar \alpha}$,
  $( u'', t')\in R'_{ \bar \alpha'}$,
  $F( m, k, \bar \alpha)\sqsubseteq_\LL \alpha$, and
  $F( m, \ell, \bar \alpha')\sqsubseteq_\LL \alpha'$.

  By construction of $S'$ and $T'$, we have $s\mayto{ k'}_{ S'} s'$
  and $t\mayto{ \ell'}_{ T'} t'$ with $k\labpre k'$,
  $\ell\labpre \ell'$, $d( k', k)\sqsubseteq_\LL \beta$, and
  $d( \ell', \ell)\sqsubseteq_\LL \gamma$, and such that $k'\oland
  \ell'$ is defined.  Also, $( u'',( s', t'))\in R^{ \wedge'}_{ C(
    \bar \alpha, \bar \alpha')}$ and
  \begin{equation*}
    F( m, k'\oland
    \ell', C( \bar \alpha, \bar \alpha'))\sqsubseteq_\LL
    C( F( m, k, \bar \alpha), F( m, \ell, \bar
    \alpha'))\sqsubseteq_\LL C( \alpha, \alpha').
  \end{equation*}

  The other direction of the proof, starting with a transition $( s,
  t)\mustto{ k\oland \ell}_{ S'\wedge T'}( s', t')$, is an exact copy
  of the corresponding part of the proof of Theorem~\ref{wm2.th:conj}. \qed
\end{proof}

\begin{example}
  For the set $\Spec= \Sigma\times \II$ from our examples, the following
  lemma shows a one-step version of relaxed conjunctive boundedness.

  \begin{lemma}
    \label{wm2.le:ex-relconb}
    For all $k, \ell\in \Spec$ for which there is $m\in \Spec$ with $d(
    m, k)\ne \infty$ and $d( m, \ell)\ne \infty$, there exist $k',
    \ell'\in \Spec$ with $k\labpre k'$, $\ell\labpre
    \ell'$, $d( k', k)\ne \infty$, $d( \ell', \ell)\ne \infty$, and
    $k'\oland \ell'$ defined, and then $d( m', k'\oland \ell')\le
    \max( d( m', k), d( m', \ell))$ for all $m'\in \Spec$.
  \end{lemma}

  \begin{proof}
    Let $k, \ell\in \Spec$ such that there is $m\in \Spec$ with $d( m,
    k)\ne \infty$ and $d( m, \ell)\ne \infty$.  This implies that
    $k=( a,[ l, r])$ and $\ell=( a,[ l', r'])$ for some $a\in \Sigma$,
    $k, l, k', l'\in \Int\cup\{ -\infty, \infty\}$.  Without loss of
    generality we can assume that $l\le l'$.

    If $r\ge l'$, then $k\oland \ell=( a,[ l', r])$ is defined, and we
    take $k'= k$, $\ell'= \ell$.  Now let $m'=( a',[ l'', r''])\in
    \Spec$.  If $a'\ne a$, the property to prove is trivially true.  If
    $a'= a$, then we have
    \begin{align*}
      d( m', k'\oland \ell') &= \max( 0, l'- l'', r''- r)\,, \\
      d( m', k) &= \max( 0, l- l'', r''- r)\,, \\
      d( m', \ell) &= \max( 0, l'- l'', r''- r')\,.
    \end{align*}
    Thus we need to show that
    \begin{equation*}
      \max( 0, l'- l'', r''- r)\le \max( 0, l- l'', r''- r, l'- l'', r''-
      r')\,,
    \end{equation*}
    which is clear as all left-hand terms also appear on the right-hand
    side.

    In case $r< l'$, we let $k'=( a,[ l, l'])$ and $\ell'=( a,[ r,
    r'])$.  Then $k\labpre k'$, $\ell\labpre \ell'$,
    and $k'\oland \ell'=( a,[ r, l'])$ is defined.  Also, $d( k', k)=
    d( \ell', \ell)= l'- r\ne \infty$.

    Let $m'=( a',[ l'', r''])$ as before, then the case $a'\ne a$ is
    again trivial.  We have
    \begin{equation*}
      d( m', k'\oland \ell')= \max( 0, r- l'', r''- l')\,,
    \end{equation*}
    so we need to show that
    \begin{align*}
      \max( 0, r- l'', r''- l') &\le \max( 0, l- l'', r''- r, l'- l'',
      r''- r') \\
      &= \max( 0, r''- r, l'- l'')\,,
    \end{align*}
    where the equality follows from $l\le l'$, hence $l- l''\le l'-
    l''$, and $r\le r'$, hence $r''- r'\le r''- r$.  But $0\le l'- r$,
    $r- l''< l'= l''$, and $r''- l'< r''- r$ because of $r< l'$, so the
    inequality follows. \qed
  \end{proof}

  For the accumulating distance, it then follows that $\oland$ is
  relaxed conjunctively bounded by $C( \alpha, \alpha')= \alpha+
  \alpha'$: Using the notation from Lemma~\ref{wm2.le:ex-relconb}, we need
  to show~\eqref{wm2.eq:relconb}, \ie~that $d( m', l'\oland \ell')+
  \lambda( \alpha+ \alpha')\le d( m', k)+ \lambda \alpha+ d( m', \ell)+
  \lambda \alpha'$, which however is clear by $d( m', k'\oland
  \ell')\le \max( d( m', k), d( m', \ell))\le d( m', k)+ d( m', \ell)$.
\end{example}

\begin{example}
  For the pointwise distance, $\oland$ is relaxed conjunctively bounded
  by $C( \alpha, \alpha')= \max( \alpha, \alpha')$:
  \eqref{wm2.eq:relconb}~is then equivalent to $\max( d( m', k'\oland
  \ell'), \alpha, \alpha')\le \max( d( m', k), d( m', \ell), \alpha,
  \alpha')$, which follows from Lemma~\ref{wm2.le:ex-relconb}.
\end{example}

\begin{example}
  For the limit-average distance, $\oland$ is relaxed conjunctively
  bounded by $C( \alpha, \alpha')= \alpha\oplus_\LL \alpha'$: Again
  using the notation from Lemma~\ref{wm2.le:ex-relconb}, we need to
  show~\eqref{wm2.eq:relconb}, so we need to see that for all $j\in
  \Natp$,
  \begin{multline*}
    \tfrac1{ j+ 1} d( m', k'\oland \ell')+ \tfrac j{ j+ 1} \alpha( j- 1)+
    \tfrac j{ j+ 1} \alpha'( j- 1) \\
    \le \tfrac1{ j+ 1} d( m', k)+ \tfrac j{ j+
      1} \alpha( j- 1)+ \tfrac1{ j+ 1} d( m', \ell)+ \tfrac j{ j+ 1}
    \alpha'( j- 1)\,.
  \end{multline*}
  This follows again from
  $d( m', k'\oland \ell')\le \max( d( m', k), d( m', \ell))\le d( m',
  k)+ d( m', \ell)$.
\end{example}

\begin{example}
  Also for the maximum-lead distance, $\oland$ is relaxed
  conjunctively bounded by
  $C( \alpha, \alpha')= \alpha\oplus_\LL \alpha'$.  To see this, we
  again use the notation from Lemma~\ref{wm2.le:ex-relconb}.  We need
  to show that for all $\alpha, \alpha'\in \LL$ and all $d\in \Real$,
  \begin{multline}
    \label{eq:ml-relconb}
    \smash[b]{ \adjustlimits \sup_{ l''\le z\le r''} \inf_{ r\le w\le
        l'} \max(| d+ z- w|, \alpha( d+ z- w)+ \alpha'( d+ z- w))} \\
      \le
    \begin{aligned}[t]
      & \adjustlimits \sup_{ l''\le z\le r''} \inf_{ l\le x\le r} \max(|
      d+ z- x|, \alpha( d+ z- x)) \\
      &+ \adjustlimits \sup_{ l''\le z\le r''} \inf_{ l'\le y\le r'}
      \max(| d+ z- y|, \alpha'( d+ z- y))
    \end{aligned}
  \end{multline}
  Now for all $z\in[ l'', r'']$, we have
  \begin{equation*}
    \smash[b]{ \inf_{ r\le w\le l'}| d+ z- w|\le \inf_{ l\le x\le r}| d+ z- x|+
      \inf_{ l'\le y\le r'}| d+ z- y|}
  \end{equation*}
  and
  \begin{multline*}
    \smash[b]{\inf_{ r\le w\le l'}}( \alpha( d+ z- w)+ \alpha'( d+ z- w)) \\
    \le \inf_{ l\le x\le r} \alpha( d+ z- x)+ \inf_{ l'\le y\le r'}
    \alpha'( d+ z- y)\,;
  \end{multline*}
  both can be shown by simply considering all cases of the placement of
  the infima.  But then also
  \begin{align*}
    \max( \inf_{ r\le w\le l'}| d+ z- w|, & \inf_{ r\le w\le l'}( \alpha(
    d+ z- w)+ \alpha'( d+ z- w))) \\
    &\le \max\left\{
      \begin{aligned}
        & \inf_{ l\le x\le r}| d+ z- x|+ \inf_{ l'\le y\le r'}| d+ z- y|
        \\
        & \inf_{ l\le x\le r} \alpha( d+ z- x)+ \inf_{ l'\le y\le r'}
        \alpha'( d+ z- y))
      \end{aligned}
    \right. \\
    &\le
    \begin{aligned}[t]
      & \max( \inf_{ l\le x\le r}| d+ z- x|, \inf_{ l\le x\le r} \alpha(
      d+ z- x)) \\
      &+ \max( \inf_{ l'\le y\le r'}| d+ z- y|, \inf_{ l'\le y\le r'}
      \alpha'( d+ z- y))\,,
    \end{aligned}
  \end{align*}
  the last inequality by distributivity of $+$ over $\max$.  As this
  holds for all $z$, we have proven~\eqref{eq:ml-relconb}.
\end{example}

\section{Logical Characterizations}

We show that quantitative refinement admits a logical characterization.
Our results extend the logical characterization of modal transition
systems in~\cite{DBLP:conf/avmfss/Larsen89}.  Our logic $\mathcal L$ is
the smallest set of expressions generated by the following abstract
syntax:
\begin{equation*}
  \phi,\phi_1,\phi_2 := \ltrue \mid \lfalse
  \mid  \langle \ell\rangle \phi \mid [ \ell] \phi \mid
  \phi_1 \wedge \phi_2
  \mid \phi_1 \vee \phi_2 \qquad (\ell\in \Spec)
\end{equation*}
The semantics of a formula $\phi\in \mathcal L$ is a mapping $\wsem \phi:
S\to \LL$ given inductively as follows:
\begin{gather*}
  \begin{aligned}
    \wsem \ltrue s &= \bot &\qquad\qquad \wsem \lfalse s &= \top \\
    \wsem{( \phi_1\wedge \phi_2)} s &= \max( \wsem{ \phi_1} s, \wsem \phi_2
    s) &\qquad \wsem{( \phi_1\vee \phi_2)} s &= \min( \wsem{ \phi_1} s,
    \wsem{ \phi_2} s)
  \end{aligned}
  \\
  \begin{aligned}
    \wsem{ \langle \ell\rangle \phi} s &= \inf\{ F( k, \ell, \wsem \phi
    t)\mid s\mustto{ k} t, d( k, \ell)\ne \top_\LL\} \\
    \wsem{ [ \ell] \phi} s &= \sup\{ F( k, \ell, \wsem \phi t)\mid
    s\mayto{ k} t, d( k, \ell)\ne \top_\LL\}
  \end{aligned}
\end{gather*}
For a SMTS $S$ we write $\wsem \phi S= \wsem \phi s_0$.

The below theorems express the fact that $\mathcal L$ is
\emph{quantitatively sound} for refinement distance, \ie~the value of
a formula in a specification is bounded by its value in any other
specification together with their distance, and that the
disjunction-free fragment of $\mathcal L$ is \emph{quantitatively
  implementation complete}, \ie~the value of any disjunction-free
formula in a specification $S$ is bounded above by its value in any
implementation of $S$.  Note that disjunction-freeness is a very
common assumption in this context,
\cf~\cite{DBLP:conf/avmfss/Larsen89, DBLP:conf/atva/BenesCK11}.

\begin{theorem}
  \label{wm2.th:l-sound}
  For all $\phi\in \mathcal L$ and all SMTS $S$, $T$, $\wsem \phi
  S\sqsubseteq_\LL \wsem \phi T\oplus_\LL \md( S, T)$.
\end{theorem}

\begin{proof}
  Structural induction.  The claim obviously holds for $\phi= \ltrue$ and
  $\phi= \lfalse$; if $\phi= \phi_1\wedge \phi_2$, then $\wsem{ \phi_i}
  s_1\sqsubseteq_\LL \wsem{ \phi_i} s_2\oplus_\LL \md( s_1, s_2)$ for $i=
  1, 2$ imply that also $\max( \wsem{ \phi_1} s_1, \wsem{ \phi_2}
  s_1)\sqsubseteq_\LL \max( \wsem{ \phi_1} s_2, \wsem{ \phi_2}
  s_2)\oplus_\LL \md( s_1, s_2)$, and similarly for $\phi= \phi_1\vee
  \phi_2$.

  For the case $\phi=\langle \ell\rangle \phi'$, there is nothing to
  prove if there are no transitions $s_2\mustto{}_2$ or if $\md( s_1,
  s_2)= \top_\LL$.  Let thus $s_2\mustto{ k_2}_2 t_2$, then there exist
  $s_1\mustto{ k_1}_1 t_1$ with $F( k_1, k_2, \md( t_1,
  t_2))\sqsubseteq_\LL \md( s_1, s_2)$.  Now by induction hypothesis,
  $\wsem{ \phi'} t_1\sqsubseteq_\LL \md( t_1, t_2)\oplus_\LL \wsem{ \phi'}
  t_2$, and then, using the triangle inequality,
  \begin{align*}
    F( k_1, \ell, \wsem{ \phi'} t_1) &\sqsubseteq_\LL F( k_1, k_2, \md(
    t_1, t_2))\oplus_\LL F( k_2, \ell, \wsem{ \phi'} t_2) \\
    &\sqsubseteq_\LL \md( s_1, s_2)\oplus_\LL F( k_2, \ell, \wsem{ \phi'}
    t_2)\,.
  \end{align*}
  As $s_2\mustto{ k_2}_2 t_2$ was arbitrary, this entails
  \begin{multline*}
    \inf\{ F( k_1, \ell, \wsem{ \phi'} t_1)\mid s_1\mustto{ k_1}_1
    t_1\} \\ \sqsubseteq_\LL \inf\{ F( k_2, \ell, \wsem{ \phi'} t_2)\mid
    s_1\mustto{ k_2}_2 t_2\}\oplus_\LL \md( s_1, s_2)\,.
  \end{multline*}

  For the case $\phi=[ \ell] \phi'$ the proof is similar: We have
  nothing to prove if $\md( s_1, s_2)= \top_\LL$ or if there are no
  transitions $s_1\mayto{ k_1}_1 t_1$ with $F( k_1, \ell, \wsem{ \phi'}
  t_1)\ne \top_\LL$, so assume there is such a transition.  Then we also
  have $s_2\mayto{ k_2}_2 t_2$ with $F( k_1, k_2, \md( t_1,
  t_2))\sqsubseteq_\LL \md( s_1, s_2)$, and
  \begin{multline*}
    F( k_1, \ell, \wsem{ \phi'} t_1)\sqsubseteq_\LL F( k_1, k_2, \md(
    t_1, t_2))\oplus_\LL F( k_2, \ell, \wsem{ \phi'} t_2) \\
    \sqsubseteq_\LL \md( s_1, s_2)\oplus_\LL F( k_2, \ell, \wsem{ \phi'}
    t_2)\,.
  \end{multline*}
\end{proof}

\medskip

\begin{theorem}
  \label{wm2.th:l-complete}
  For all disjunction-free formulae $\phi\in \mathcal L$ and all SMTS
  $S$,
  $\wsem \phi S= \sup_{ I\in \llbracket S\rrbracket} \wsem \phi I$.
\end{theorem}

\begin{proof}
  Theorem~\ref{wm2.th:l-sound} entails $\wsem \phi I\sqsubseteq_\LL \wsem
  \phi S\oplus_\LL \md( I, S)= \wsem \phi S$ for all $I\in \llbracket
  S\rrbracket$, hence also $\sup_{ I\in \llbracket S\rrbracket} \wsem
  \phi I\sqsubseteq_\LL \wsem \phi S$.  To show that $\wsem \phi
  S\sqsubseteq_\LL \sup_{ I\in \llbracket S\rrbracket} \wsem \phi I$ we
  use structural induction on $\phi$.  If $\phi= \ltrue$, both sides
  are $\bot_\LL$, and if $\phi= \lfalse$, both sides are $\top_\LL$,
  so the induction base is clear.

  The case $\phi= \phi_1\wedge \phi_2$ is also clear: By hypothesis,
  $\wsem{ \phi_1} S\sqsubseteq_\LL \sup_{ I\in \llbracket S\rrbracket}
  \wsem{ \phi_1} I$ and similarly for $\phi_2$, hence
  \begin{align*}
    \wsem \phi S= \max( \wsem{ \phi_1} S, \wsem{ \phi_2} S) &\sqsubseteq_\LL
    \max( \sup_{ I\in \llbracket S\rrbracket} \wsem{ \phi_1} I, \sup_{
      I\in \llbracket S\rrbracket} \wsem{ \phi_2} I) \\ &= \sup_{ I\in
      \llbracket S\rrbracket} \max( \wsem{ \phi_1} I, \wsem{ \phi_2} I)\,.
  \end{align*}

  For the case $\phi= \langle \ell\rangle \phi'$, we are done if $\wsem
  \phi S= \bot_\LL$.  Otherwise, let $\alpha\sqsubset_\LL \wsem \phi S$; we
  want to expose $I\in \llbracket S\rrbracket$ for which
  $\alpha\sqsubset_\LL \wsem \phi I$.  Start by letting $I=\{ i_0\}$ and
  $\mmustto_I= \emptyset$.

  Now for each transition $s_0\mustto{ k}_S t$, we have
  $\alpha\sqsubset_\LL F( k, \ell, \wsem{ \phi'} t)$, so (assuming for the
  moment that $\wsem{ \phi'} t\ne \bot_\LL$) there is
  $\alpha_k'\sqsubset_\LL \wsem{ \phi'} t$ for which $F( k, \ell,
  \alpha_k')\sqsupset_\LL \alpha$.  By induction hypothesis, there is
  $J\in \llbracket t, S\rrbracket$ for which $\alpha_k'\sqsubset_\LL
  \wsem{ \phi'} J$; let $n\in \llbracket k\rrbracket$ such that $F( n,
  \ell, \wsem{ \phi'} J)= F( k, \ell, \wsem{ \phi'} J)$, and add $J$
  together with a transition $i_0\mustto{ n}_I j_0$ to $I$.  In case
  $\wsem{ \phi'} t= \bot_\LL$, we just take an arbitrary $J\in \llbracket
  t, S\rrbracket$.

  For the so-constructed implementation $I$ we have
  \begin{align}
    \wsem \phi I &= \inf\{ F( m, \ell, \wsem{ \phi'} j\mid i_0\mustto{
    m}_I j\} \notag\\
    &= \inf\{ F( k, \ell, \wsem{ \phi'} J)\mid s_0\mustto{ k}_S t, J\in
    \llbracket t, S\rrbracket, \wsem{ \phi'} t= \bot_\LL \text{ or }
    \alpha_k'\sqsubset_\LL \wsem{ \phi'} J\} \notag\\
    &\sqsupset_\LL \inf(\{ F( k, \ell, \alpha_k')\mid s_0\mustto{ k}_S
    t\}\cup\{ F( k, \ell, \wsem{ \phi'} t)\})\sqsupseteq_\LL
    \alpha\,, \label{wm2.eq:fibragen}
  \end{align}
  the strict inequality in~\eqref{wm2.eq:fibragen} because $S$ is compactly
  branching.

  For the case $\phi=[ \ell] \phi'$, let again $\alpha\sqsubset_\LL \wsem
  \phi S$, and let $I\in \llbracket S\rrbracket$ be any implementation
  (there exists one because of local consistency of $S$).  If $F( k,
  \ell, \wsem{ \phi'} t)= \top_\LL$ for all $s_0\mayto{ k}_S t$, then
  $\wsem \phi S= \sup \emptyset= \bot_\LL$ and we are done.  Otherwise let
  $s_0\mayto{ k}_S t$ be such that $\wsem \phi S= F( k, \ell, \wsem{
    \phi'} t)$, which exists because $S$ is compactly branching.  Then
  $\alpha\sqsubset_\LL F( k, \ell, \wsem{ \phi'} t)$, so (assuming that
  $\wsem{ \phi'} t\ne \bot_\LL$) we have $\alpha_k'\sqsubset_\LL \wsem{
    \phi'} t$ with $F( k, \ell, \alpha_k')\sqsupset_\LL \alpha$.

  Let $J\in \llbracket t, S\rrbracket$ such that $\alpha_k'\sqsubset_\LL
  \wsem{ \phi'} J$, let $n\in \llbracket k\rrbracket$ such that $F( n,
  \ell, \wsem{ \phi'} J)= F( k, \ell, \wsem{ \phi'} J)$, and add $J$
  together with a transition $i_0\mustto{ n}_I j_0$ to $I$.  Then
  \begin{align*}
    \wsem \phi I &= \sup\{ F( m, \ell, \wsem{ \phi'} n)\mid i_0\mustto{
    m}_I j\} \\
    &\sqsupseteq_\LL F( n, \ell, \wsem{ \phi'} J)= F( k, \ell, \wsem{
      \phi'} J)\sqsupseteq_\LL F( k, \ell, \alpha_k')\sqsupset_\LL
    \alpha\,.
  \end{align*}
  In case $\wsem{ \phi'} t= \bot_\LL$ instead, we again take an arbitrary
  $J\in \llbracket t, S\rrbracket$, and then $\wsem \phi I\sqsupseteq_\LL
  F( k, \ell, \wsem{ \phi'} t)\sqsupset_\LL \alpha$. \qed
\end{proof}

\chapter[Logical vs.\ Behavioral Specifications][Logical vs.\
Behavioral Specifications]{Logical vs.\ Behavioral
  Specifications\footnote{This chapter is based on the journal
    paper~\cite{DBLP:journals/iandc/BenesFKLT20} published in
    Information and Computation.}}
\label{ch:dmts}

In this chapter we depart from the quantitative setting of this thesis
and introduce \emph{disjunctive} modal transition systems.  We show
that this generalization of MTS is closely related to other
specification formalisms, \viz acceptance automata and the modal
$\nu$-calculus and that it admits both disjunction and conjunction as
well as a general notion of quotient which was unavailable for MTS.

\section{Specification Formalisms}
\label{sec:spec}

In this section we introduce the four specification formalisms with
which this chapter is concerned.  For the rest of the chapter, we fix
a~finite alphabet $\Sigma$.  In each of the formalisms, the semantics
of a~specification is a~set of implementations, in our case always
a~set of \emph{(finite) labeled transition systems} (LTS) over
$\Sigma$, \ie~structures $\mcalI=( S, s^0, \omust)$ consisting of
a~finite set $S$ of \emph{states}, an~initial state $s^0 \in S$, and
a~\emph{transition relation} $\omust\subseteq S\times \Sigma\times S$.

\subsection{Disjunctive Modal Transition Systems}

\begin{definition}
  A \emph{disjunctive modal transition system} (DMTS) is a structure
  $\mcalD=( S, S^0, \omay, \omust)$ consisting of finite sets $S\supseteq
  S^0$ of states and initial states, a \emph{may}-transition relation
  $\omay\subseteq S\times \Sigma\times S$, and a \emph{disjunctive
    must}-transition relation $\omust\subseteq S\times 2^{ \Sigma\times
    S}$.  It is assumed that for all $( s, N)\in \omust$ and all $( a,
  t)\in N$, $( s, a, t)\in \omay$; furthermore, if $(s, \emptyset) \in \omust$
  then there are no $a$ and $t$ such that $(s, a, t)\in \omay$.
\end{definition}

As customary, we write $s\DMTSmay a t$ instead of $( s, a, t)\in \omay$,
$s\DMTSmust{} N$ instead of $( s, N)\in \omust$, $s\DMTSmay a$ if there exists
$t$ for which $s\DMTSmay a t$, and $s\notmay a$ if there does not.

The intuition is that may-transitions $s\DMTSmay a t$ specify which
transitions are permitted in an implementation, whereas
a~must-transition $s\DMTSmust{} N$ stipulates a disjunctive requirement: at
least one of the choices $( a, t)\in N$ has to be implemented.  A DMTS $(
S, S^0, \omay, \omust)$ is an \emph{implementation} if $S^0=\{ s^0\}$ is
a~singleton and $\omust=\{( s,\{( a, t)\})\mid s\DMTSmay a t\}$.

DMTS were introduced in~\cite{DBLP:conf/lics/LarsenX90} in the context
of equation solving.  They are a natural extension of the modal
transition systems (MTS) of previous chapters.  We say that a~DMTS
$( S, S^0, \omay, \omust)$ is a~MTS if $S^0 = \{ s^0 \}$ is
a~singleton and for all $s\DMTSmust{} N$ it holds that $N=\{( a, t)\}$
is also a~singleton.  When speaking about MTS, we usually write
$s \DMTSmust{a} t$ instead of $s \DMTSmust{} \{(a,t)\}$.

An LTS $( S, s^0, \omust)$ can be translated to a DMTS implementation
\linebreak
$( S, S^0, \omay,$ $\omust')$ by setting $S^0=\{ s^0\}$,
$\omay=\omust$ and $\omust'=\{( s,\{( a, t)\})\mid s\DMTSmust a t\}$.
This defines an embedding of LTS into DMTS whose image is precisely
the set of DMTS implementations.

\begin{definition}\label{def:dmts_mr}
  Let $\mcalD_1=( S_1, S^0_1, \omay_1, \omust_1)$ and
  $\mcalD_2=( S_2, S^0_2, \omay_2, \omust_2)$ be DMTS.  A relation
  $R\subseteq S_1\times S_2$ is a \emph{modal refinement} if for all
  $( s_1, s_2)\in R$ the following conditions hold:
  \begin{itemize}
  \item for all $s_1\DMTSmay a_1 t_1$ there is $t_2\in S_2$ with
    $s_2\DMTSmay a_2 t_2$ and $( t_1, t_2)\in R$, and
  \item for all $s_2\DMTSmust{}_2 N_2$ there is $s_1\DMTSmust{}_1 N_1$ such
    that for each $( a, t_1)\in N_1$ there is $( a, t_2)\in N_2$ with
    $( t_1, t_2)\in R$.
  \end{itemize}
  We say that $\mcalD_1$ \emph{modally refines} $\mcalD_2$, denoted $\mcalD_1\mr
  \mcalD_2$, whenever there exists a modal refinement $R$ such that for all
  $s^0_1\in S^0_1$, there exists $s^0_2\in S^0_2$ for which $( s^0_1,
  s^0_2)\in R$.
\end{definition}

We write $\mcalD_1\mreq \mcalD_2$ if $\mcalD_1\mr \mcalD_2$ and $\mcalD_2\mr \mcalD_1$.
For states $s_1\in S_1$, $s_2\in S_2$, we write $s_1\mr s_2$ if
$( S_1,\{ s_1\}, \omay_1, \omust_1)\mr( S_2,\{ s_2\}, \omay_2,
\omust_2)$.
Sometimes we will refer to the last property of a modal refinement
relation,
$\forall s^0_1\in S^0_1: \exists s^0_2\in S^0_2:( s^0_1, s^0_2)\in R$,
as being \emph{initialised}.
Note that modal refinement is reflexive and transitive, \ie~a preorder
on DMTS.

The \emph{set of implementations} of a DMTS $\mcalD$ is $\sem \mcalD=\{ \mcalI\mr
\mcalD\mid \mcalI~\text{implement-}\linebreak[4]
\text{ation}\}$.  This is, thus, the set of all LTS which satisfy the
specification given by the DMTS $\mcalD$.  We say that $\mcalD_1$
\emph{thoroughly refines} $\mcalD_2$, and write $\mcalD_1\DMTStr \mcalD_2$, if $\sem{
  \mcalD_1}\subseteq \sem{ \mcalD_2}$.  We write $\mcalD_1\DMTStreq \mcalD_2$ if
$\mcalD_1\DMTStr \mcalD_2$ and $\mcalD_2\DMTStr \mcalD_1$.  For states $s_1\in S_1$, $s_2\in
S_2$, we write $\sem{ s_1}= \sem{( S_1,\{ s_1\}, \omay_1, \omust_1)}$
and $s_1\DMTStr s_2$ if $\sem{ s_1}\subseteq \sem{ s_2}$.

The proposition below, which follows directly from transitivity of modal
refinement, shows that modal refinement is \emph{sound} with respect to
thorough refinement; in the context of specification theories, this is
what one would expect, and we only include it for completeness of
presentation.  It can be shown that modal refinement is also
\emph{complete} for \emph{deterministic}
DMTS~\cite{DBLP:conf/atva/BenesCK11}, but we will not need this
here.

\begin{proposition}
  \label{pr:mrvstr}
  For all DMTS $\mcalD_1$, $\mcalD_2$, $\mcalD_1\mr \mcalD_2$ implies $\mcalD_1\DMTStr
  \mcalD_2$. \noproof
\end{proposition}

\subsection{The Modal $\nu$-Calculus}

We recall the syntax and semantics of the modal $\nu$-calculus, the
fragment of the modal $\mu$-calculus~\cite{unpub/ScottB69,
  DBLP:journals/tcs/Kozen83} with only greatest fixed points.  Instead
of an explicit greatest fixed point operator, we use the
representation by equation systems in Hennessy-Milner logic developed
in~\cite{DBLP:journals/tcs/Larsen90}.

For a finite set $X$ of variables, let $\HML( X)$ be the set of
\emph{Hennessy-Milner formulae}, generated by the abstract syntax
\begin{equation*}
  \HML(
  X)\ni \phi\Coloneqq \ltrue\mid \lfalse\mid x\mid \langle a\rangle
  \phi\mid[ a] \phi\mid \phi\land \phi\mid \phi\lor \phi\,,
\end{equation*}
for $a\in \Sigma$ and $x\in X$.

A \emph{declaration} is a mapping $\Delta: X\to \HML( X)$; we recall
the greatest fixed point semantics of declarations
from~\cite{DBLP:journals/tcs/Larsen90}.  For an LTS
$( S, s^0, \omust)$, an~\emph{assignment} is a mapping
$\sigma: X\to 2^S$.  The set of assignments forms a~complete lattice
with order $\sigma_1\sqsubseteq \sigma_2$ iff
$\sigma_1( x)\subseteq \sigma_2( x)$ for all $x\in X$ and least upper
bound
$\big(\bigsqcup_{ i\in I} \sigma_i\big)( x)= \bigcup_{ i\in I}
\sigma_i( x)$.

The semantics of a formula is a subset of $S$, given relative to an
assignment~$\sigma$, defined as follows: $\lsem \ltrue \sigma= S$, $\lsem
\lfalse \sigma= \emptyset$, $\lsem x \sigma= \sigma( x)$, $\lsem{ \phi\land
  \psi} \sigma= \lsem \phi\sigma\cap \lsem \psi \sigma$, $\lsem{
  \phi\lor \psi} \sigma= \lsem \phi\sigma\cup \lsem \psi \sigma$, and
\begin{align*}
  \lsem{\langle a\rangle \phi} \sigma &= \{ s\in S\mid \exists s\DMTSmust a
  s': s'\in \lsem \phi \sigma\}, \\
  \lsem{[ a] \phi} \sigma &= \{ s\in S\mid \forall s\DMTSmust a s': s'\in
  \lsem \phi \sigma\}.
\end{align*}
The semantics of a declaration $\Delta$ is then the assignment defined
by
\begin{equation*}
  \lsem \Delta= \bigsqcup\{ \sigma: X\to 2^S\mid \forall x\in X:
  \sigma( x)\subseteq \lsem{ \Delta( x)} \sigma\};
\end{equation*}
the greatest (post)fixed point of $\Delta$.

A \emph{$\nu$-calculus expression} is a structure
$\mcalN=( X, X^0, \Delta)$, with $X^0\subseteq X$ sets of variables and
$\Delta: X\to \HML( X)$ a declaration.  We say that an LTS
$\mcalI=( S, s^0, \omust)$ \emph{implements} (or models) the expression,
and write $\mcalI\models \mcalN$, if there is $x^0\in X^0$ such that
$s^0\in \lsem \Delta( x^0)$.  We write $\sem \mcalN$ for the set of
implementations (models) of a $\nu$-calculus expression $\mcalN$.  As for
DMTS, we write $\sem x= \sem{( X,\{ x\}, \Delta)}$ for $x\in X$, and
thorough refinement of expressions and variables is defined
accordingly.\footnote{Any $\nu$-calculus expression is thoroughly
  equivalent to one with precisely one initial state, \ie~$X^0$ a
  singleton, however this is not true for $\nu$-calculus expressions
  in normal form as defined below.  See also
  Section~\ref{se:initial}.}

We are now going to introduce a~\emph{normal form} for $\nu$-calculus
expressions. The purpose of this normal form is twofold. One is to allow
us to define modal refinement for $\nu$-calculus, an analogue to the
DMTS modal refinement that can be seen as a~sound approximation of the
logical implication, \cf Proposition~\ref{pr:mrvstr}.  The second
purpose is to facilitate a~simple translation between DMTS and
$\nu$-calculus expressions, see Section~\ref{se:dmtsnu}
below.

\begin{lemma}
  \label{le:hmlnormal}
  For any $\nu$-calculus expression $\mcalN_1=( X_1, X^0_1, \Delta_1)$,
  there exists another expression $\mcalN_2=( X_2, X^0_2, \Delta_2)$ with
  $\sem{ \mcalN_1}= \sem{ \mcalN_2}$ and such that for any $x\in X$,
  $\Delta_2( x)$ is of the form
  \begin{equation}
    \label{eq:hmlnormal}
    \Delta_2( x)= \bigland_{ i\in I}\Big( \biglor_{ j\in
      J_i} \langle a_{ ij}\rangle  x_{ ij}\Big)\land \bigland_{ a\in
      \Sigma}[ a] \Big( \biglor_{ j\in J_a} y_{ a, j}\Big)
  \end{equation}
  for finite (possibly empty) index sets $I$, $J_i$, $J_a$, for $i\in I$
  and $a\in \Sigma$, and all $x_{ ij}, y_{ a, j}\in X_2$.  Additionally,
  for all $i\in I$ and $j\in J_i$, there exists $j'\in J_{ a_{ ij}}$ for
  which $x_{ ij}\DMTStr y_{ a_{ ij}, j'}$.
  Also if at least one of the $J_i = \emptyset$ then $J_a = \emptyset$ for all
  $a$.
\end{lemma}

Remark that this normal form includes a semantic check ($x_{ ij}\DMTStr
y_{ a_{ ij}, j'}$), so it is not entirely syntactic.

\begin{proof}
  It is shown in~\cite{DBLP:journals/tcs/BoudolL92} that any
  Hennessy-Milner formula is equivalent to one in so-called \emph{strong
    normal form}, \ie~of the form $\biglor_{ i\in I}( \bigland_{ j\in
    J_i} \langle a_{ ij}\rangle \phi_{ ij}\land \bigland_{ a\in \Sigma}[
  a] \psi_{ i, a})$ for HML formulas $\phi_{ ij}$, $\psi_{ i, a}$ which
  are also in strong normal form, and such that for all $i, j$,
  $\phi_{ ij}\DMTStr \psi_{ i, a_{ ij}}$.

  We can replace the $\phi_{ ij}$, $\psi_{ i, a}$ by (new) variables
  $x_{ ij}$, $y_{ i, a}$ and add declarations
  $\Delta_2( x_{ ij})= \phi_{ ij}$,
  $\Delta_2( y_{ i, a})= \psi_{ i, a}$ to arrive at an expression in
  which all formulae are of the form
  $\Delta_2( x)= \biglor_{ i\in I}( \bigland_{ j\in J_i} \langle a_{
    ij}\rangle x_{ ij}\land \bigland_{ a\in \Sigma}[ a] y_{ i, a})$
  and such that for all $i, j$, $x_{ ij}\DMTStr y_{ i, a_{ ij}}$.

  Now for each such formula, replace (recursively) $x$ by new
  variables $\{ \tilde x^i\mid i\in I\}$, similarly for $y$, and set
  $\Delta_2( \tilde x^i)= \bigland_{ j\in J_i}\langle a_{
    ij}\rangle\big( \biglor_k \tilde x_{ij}^k\big)\land \bigland_{
    a\in \Sigma}[ a]\big( \biglor_k \tilde y_{ i, a}^k\big)$.
  Using initial variables $X_2^0=\{ \tilde x^i\mid x\in X_1^0\}$, the
  so-constructed $\nu$-calculus expression is equivalent to the
  original one.  We know that for all $i, j$,
  $\bigvee_k \tilde x_{ ij}^k\DMTStr \bigvee_k \tilde y_{ i, a_{ ij}}^k$,
  hence for all $i, j, k$ there exists $k'$ such that
  $\tilde x_{ ij}^k\DMTStr \tilde y_{ i, a_{ ij}}^{ k'}$. We can thus rename
  variables and apply the distributivity of $\langle \cdot\rangle$ over
  $\bigvee$.

  Finally, if at least one of the $J_i = \emptyset$ then $\Delta_2$ is
  false and we can simply set $J_a = \emptyset$ for all $a$ without changing
  the semantics of $\Delta_2$.
  \qed
\end{proof}

As this is a type of \emph{conjunctive normal form}, it is clear that
translating a~$\nu$-calculus expression into normal form may incur an
exponential blow-up.

We introduce some notation for $\nu$-calculus expressions in normal form
which will make our life easier later.  Let $\mcalN=( X, X^0, \Delta)$ be
such an expression and $x\in X$, with $\Delta( x)= \bigland_{ i\in
  I}\big( \biglor_{ j\in J_i} \langle a_{ ij}\rangle x_{ ij}\big)\land
\bigland_{ a\in \Sigma}[ a] \big( \biglor_{ j\in J_a} y_{ a, j}\big)$ as
in the lemma.  Define $\Diamond( x)=\{\{( a_{ ij}, x_{ ij})\mid j\in
J_i\}\mid i\in I\}$ and, for each $a\in \Sigma$, $\Box^a( x)=\{ y_{ a,
  j}\mid j\in J_a\}$.  Note that now,
\begin{equation*}
  \Delta( x)= \bigland_{ N\in
    \Diamond(x)} \Big( \biglor_{( a, y)\in N} \langle a\rangle y\Big)
  \land \bigland_{ a\in \Sigma}[ a]\Big( \biglor_{ y\in \Box^a( x)}
  y\Big)\,.
\end{equation*}

\begin{definition}
  Let $\mcalN_1=( X_1, X^0_1, \Delta_1)$, $\mcalN_2=( X_2, X^0_2, \Delta_2)$
  be $\nu$-calculus expressions in normal form and $R\subseteq X_1\times
  X_2$.  The relation $R$ is a \emph{modal refinement} if it holds for
  all $( x_1, x_2)\in R$ that
  \begin{itemize}
  \item for all $a\in \Sigma$ and every $y_1\in
    \Box^a_1( x_1)$, there is $y_2\in \Box^a_2( x_2)$ for which $( y_1,
    y_2)\in R$, and
  \item for all $N_2\in \Diamond_2( x_2)$ there is $N_1\in \Diamond_1(
    x_1)$ such that for each $( a, y_1)\in N_1$, there exists $( a,
    y_2)\in N_2$ with $( y_1, y_2)\in R$.
  \end{itemize}
  We say that $\mcalN_1$ \emph{modally refines} $\mcalN_2$, denoted $\mcalN_1\mr
  \mcalN_2$, whenever there exists a~modal refinement $R$ such that for
  every $x^0_1\in X^0_1$ there exists $x^0_2\in X^0_2$ for which $(
  x^0_1, x^0_2)\in R$.
\end{definition}

We say that a $\nu$-calculus expression $( X, X^0, \Delta)$ in normal
form is an \emph{implementation} if $X^0=\{ x^0\}$ is a singleton,
$\Diamond( x)=\{\{( a, y)\}\mid y\in \Box^a( x), a\in \Sigma\}$ and
$\Box^a(x)= \emptyset$ for all $a\notin \Sigma$, for all $x\in X$.

We can translate an LTS $( S, s^0, \omust)$ to a $\nu$-calculus
expression $( S, S^0, \Delta)$ in normal form by setting $S^0=\{ s^0\}$
and $\Diamond( s)=\{\{( a, t)\}\mid s\DMTSmust a t\}$ and $\Box^a( s)=\{
t\mid s\DMTSmust a t\}$ for all $s\in S$, $a\in \Sigma$.  Like for DMTS,
this defines an embedding of LTS into the modal $\nu$-calculus whose
image are precisely the $\nu$-calculus implementations.

We will show below in Theorem~\ref{th:models=ref-nu} that for any LTS
$\mcalI$ and any $\nu$-calculus expression $\mcalN$ in normal form,
$\mcalI\models \mcalN$ iff $\mcalI\mr \mcalN$, hence the fixed-point semantics
of~\cite{DBLP:journals/tcs/Larsen90} and our refinement semantics agree.
As a corollary of this result, we get that modal refinement is a~sound
approximation to logical implication, \ie that $\mcalN_1 \mr \mcalN_2$ implies that
for all implementations $\mcalI$, $(\mcalI \models \mcalN_1) \Rightarrow (\mcalI \models
\mcalN_2)$.

\subsection{Nondeterministic Acceptance Automata}

\begin{definition}
  A~\emph{nondeterministic acceptance automaton} (\NAA) is a structure
  $\mcalA=( S, S^0, \Tran)$, with $S\supseteq S^0$ finite sets of states
  and initial states and $\Tran: S\to 2^{ 2^{ \Sigma\times S}}$ an
  assignment of \emph{transition constraints}.
\end{definition}

Acceptance automata were first introduced
in~\cite{report/irisa/Raclet07} (see
also~\cite{DBLP:journals/entcs/Raclet08}, where a slightly different
language-based approach is taken), based on the notion of acceptance
trees in~\cite{DBLP:journals/jacm/Hennessy85}; however, these are
\emph{deterministic}.  We extend the formalism to a nondeterministic
setting here.  The following notion of modal refinement was introduced
in~\cite{DBLP:conf/atva/BenesKLMS11}.

\begin{definition}\label{def:aaref}
  Let $\mcalA_1=( S_1, S^0_1, \Tran_1)$ and $\mcalA_2=( S_2, S^0_2, \Tran_2)$
  be \NAA.  A relation $R\subseteq S_1\times S_2$ is a \emph{modal
    refinement} if it holds for all $( s_1, s_2)\in R$ and all $M_1\in
  \Tran_1( s_1)$ that there exists $M_2\in \Tran_2( s_2)$ such that
  \begin{equation}
    \label{eq:aaref}
    \begin{aligned}
      & \forall( a, t_1)\in M_1: \exists( a, t_2)\in M_2:( t_1, t_2)\in
      R\,, \\
      & \forall( a, t_2)\in M_2: \exists( a, t_1)\in M_1:( t_1, t_2)\in
      R\,.
    \end{aligned}
  \end{equation}
  We say that $\mcalA_1$ \emph{modally refines} $\mcalA_2$, and write
  $\mcalA_1\mr \mcalA_2$, whenever there exists a modal refinement $R$ such
  that for all $s^0_1\in S^0_1$, there exists $s^0_2\in S^0_2$ for which
  $( s^0_1, s^0_2)\in R$.
\end{definition}

An \NAA\ is an \emph{implementation} if $S^0=\{ s^0\}$ is a singleton and,
for all $s\in S$, $\Tran( s)=\{ M\}$ is a singleton.  An LTS $( S, s^0,
\omust)$ can be translated to an \NAA\ by setting $S^0=\{ s^0\}$ and
$\Tran( s)=\{\{( a, t)\mid s\DMTSmust a t\}\}$.  This defines an
embedding of LTS into \NAA\ whose image are precisely the
\NAA\ implementations.  As for DMTS, we write $\sem \mcalA$ for the set
of implementations of an \NAA\ $\mcalA$, and through refinement and
equivalence are defined accordingly.

\subsection{Hybrid Modal Logic}

As our fourth specification formalism, we introduce a hybrid modal
logic, closely related to the Boolean modal transition systems
of~\cite{DBLP:conf/atva/BenesKLMS11} and hybrid in the sense
of~\cite{book/Prior68, DBLP:journals/igpl/Blackburn00}: it contains
nominals, and the semantics of a nominal is given as all sets which
contain the nominal.

For a finite set $X$ of nominals, let $\mcalL( X)$ be the set of formulae
generated by the abstract syntax
$\mcalL( X)\ni \phi\coloneqq \ltrue\mid \lfalse\mid \langle a\rangle
x\mid \neg \phi\mid \phi\land \phi$,
for $a\in \Sigma$ and $x\in X$.  The semantics of a formula is a set
of subsets of $\Sigma\times X$, given as follows:
$\lsem{ \ltrue}= 2^{ \Sigma\times X}$, $\lsem{ \lfalse}= \emptyset$,
$\lsem{ \neg \phi}= 2^{ \Sigma\times X}\setminus \lsem{ \phi}$,
$\lsem{ \langle a\rangle x}=\{ M\subseteq \Sigma\times X\mid( a, x)\in
M\}$,
and $\lsem{ \phi\land \psi}= \lsem{ \phi}\cap \lsem{ \psi}$.  We also
define disjunction
$\phi_1\lor \phi_2= \neg( \neg \phi_1\land \neg \phi_2)$.

An \emph{$\mcalL$-expression} is a structure $\mcalE=( X, X^0, \Phi)$
consisting of finite sets $X^0\subseteq X$ of variables and a mapping
$\Phi: X\to \mcalL( X)$.  Such an expression is an \emph{implementation} if
$\lsem{ \Phi( x)}=\{ M\}$ is a singleton for each $x\in X$.

\begin{definition}
  Let $\mcalE_1=( X_1, X^0_1, \Phi_1)$ and $\mcalE_2=( X_2, X^0_2, \Phi_2)$ be
  $\mcalL$-expressions.  A relation $R\subseteq X_1\times X_2$ is a
  \emph{modal refinement} if it holds for all $( x_1, x_2)\in R$ and all
  $M_1\in \lsem{ \Phi_1( x_1)}$ that there exists $M_2\in \lsem{ \Phi_2(
    x_2)}$ such that
  \begin{itemize}
  \item $\forall( a, y_1)\in M_1: \exists( a, y_2)\in M_2:( y_1, y_2)\in
    R$,
  \item $\forall( a, y_2)\in M_2: \exists( a, y_1)\in M_1:( y_1, y_2)\in
    R$.
  \end{itemize}
  We say that $\mcalE_1$ \emph{modally refines} $\mcalE_2$, denoted $\mcalE_1\mr
  \mcalE_2$, whenever there exists a~modal refinement $R$ such that for all
  $x^0_1\in X^0_1$, there exists $x^0_2\in X^0_2$ for which $( x^0_1,
  x^0_2)\in R$.
\end{definition}

We can translate an LTS $( S, s^0, \omust)$ to an $\mcalL$-expression
$( S, S^0, \Phi)$ by setting $S^0=\{ s^0\}$ and
$\Phi( s)= \bigland_{ s\DMTSmust a t} \langle a\rangle t\land \bigland_{
  s\not\DMTSmust b u} \neg \langle b\rangle u$.
This defines an embedding of LTS into $\mcalL$-expressions whose image
are precisely the $\mcalL$-implementations.  As for DMTS, we write
$\sem \mcalE$ for the set of implementations of an $\mcalL$-expression
$\mcalE$, and through refinement and equivalence are defined accordingly.

\begin{remark}
  As all our four specification formalisms have the same type of
  implementations, labeled transition systems, we can use thorough
  refinement and equivalence cross-formalism.  As an example example,
  for a given DMTS $\mcalD$ and an \NAA\ $\mcalA$, the expression
  $\mcalD\DMTStr \mcalA$ is valid.  We will use these types of thorough
  refinement and equivalence in many places throughout the paper.
\end{remark}

\section{Structural Equivalence}
\label{sec:struct}

\begin{figure}
  \hfill
  \xymatrix@R=12ex@C=4em{%
    & \text{DMTS} \ar@<.4ex>[r]^\da \ar@<.4ex>[d]^\dn
    & \text{\NAA} \ar@<.4ex>[l]^\ad \ar@<.4ex>[d]^\al
    \\ \text{$\nu$-calculus} \ar@<.4ex>[r]
    & \text{$\nu$-normal form} \ar@<.4ex>[l]
    \ar@<.4ex>[u]^\nd
    & \text{hybrid modal logic} \ar@<.4ex>[u]^\la
  }\hfill\mbox{}
  \caption{%
    \label{fi:translations}
    Six translations between specification formalisms}
\end{figure}
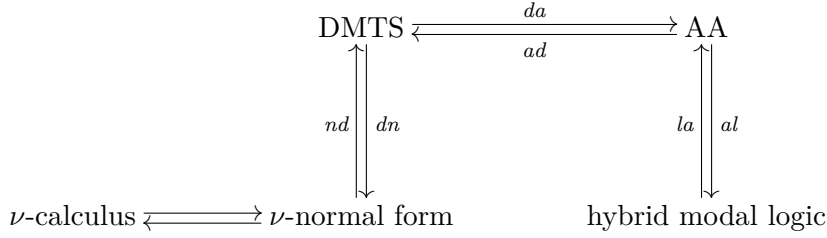

We proceed to show that our four specification formalisms are
structurally equivalent.  To this end, we shall expose six translations
between them, see Figure~\ref{fi:translations}.  Section~\ref{se:dmtsnu}
is concerned with $\dn$ and $\nd$, Section~\ref{se:aahyb} with $\al$ and
$\la$, and Section~\ref{se:dmtsaa} with $\da$ and $\ad$.  We show in
Theorems~\ref{th:dmtsnu}, \ref{th:aahyb} and~\ref{th:dmtsaa} that all
six translations preserve and reflect modal refinement.

\subsection{DMTS vs.\ the Modal $\nu$-Calculus}
\label{se:dmtsnu}

Our first two translations are rather straight-forward.  For a DMTS
$\mcalD=( S, S^0, \omay, \omust)$ and all $s\in S$, define
$\Diamond(s)=\{ N \mid s \DMTSmust{} N\}$ and, for each $a\in \Sigma$,
$\Box^a(s)=\{ t \mid s \DMTSmay{a} t\}$.  Then, let
\begin{equation}
  \label{eq:dmtstonu}
  \Delta( s)= \bigland_{N \in \Diamond(s)}\Big( \biglor_{( a, t)\in N}
  \langle a\rangle t\Big)\land \bigland_{ a\in \Sigma}[ a]\Big(
  \biglor_{t \in \Box^a(s)} t\Big)
\end{equation}
and define the (normal-form) $\nu$-calculus expression $\dn( \mcalD)=( S,
S^0, \Delta)$.

Note how the formula precisely expresses that we demand at least one of
every choice of disjunctive must-transitions (first part) and permit all
may-transitions (second part); this is similar to the
\emph{characteristic formulae} of~\cite{DBLP:conf/avmfss/Larsen89}.

Conversely, for a $\nu$-calculus expression $\mcalN=( X, X^0, \Delta)$ in
normal form, let
\begin{gather*}
  \omay=\{( x, a, y)\in X\times \Sigma\times X\mid y\in \Box^a( x)\}, \\
  \omust=\{( x, N)\mid x\in X, N\in \Diamond( x)\}.
\end{gather*}
and define the DMTS $\nd( \mcalN)=( X, X^0, \omay, \omust)$.  Note how
this is a simple syntactic translation from diamonds to disjunctive
must-transitions and from boxes to may-transitions.  Also, the two
translations are inverse to each other: $\dn( \nd( \mcalN))= \mcalN$ and
$\nd( \dn( \mcalD))= \mcalD$.

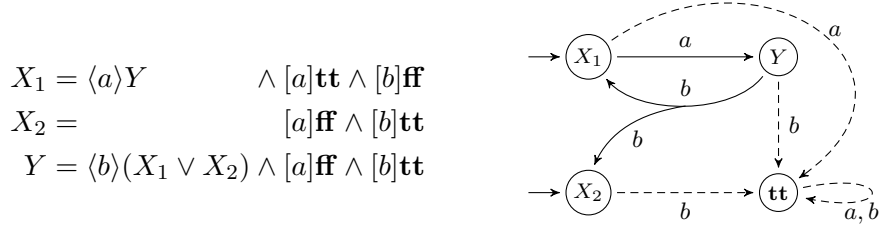
\begin{figure}
  \centering
  \begin{minipage}{.4\linewidth}
    \begin{alignat*}{2}
      X_1 &= \langle a\rangle Y &\vphantom{x} \land [a] \ltrue \land [b]
      \lfalse \\
      X_2 &= &[a]\lfalse \land [b]\ltrue\\
      Y   &= \langle b \rangle (X_1 \lor X_2) &\vphantom{x} \land [a]
      \lfalse \land [b]\ltrue
    \end{alignat*}
  \end{minipage}
  \qquad
  \begin{tikzpicture}[baseline=(bas),font=\footnotesize,->,>=stealth',
    yscale=1.8,xscale=2.5, state/.style={shape=circle,draw,font=\scriptsize,
      inner sep=.5mm,outer sep=0.8mm,minimum size=.5cm,initial text=,
      initial distance=1ex}]
    \path[use as bounding box] (-1,-1.2) rectangle (1,.5);
    \coordinate (bas) at (0,-0.4);
    \node[state, initial left] (s0) at (-.5,0) {$X_1$};
    \node[state] (s1) at (.5,-1) {$\ltrue$};
    \node[state, initial left] (s2) at (-.5,-1) {$X_2$};
    \node[state] (s3) at (.5,0) {$Y$};
    \coordinate (t1) at (-.5,-.3);
    \coordinate (t2) at (0,-.6);
    \path (s0) edge[densely dashed,bend left,looseness=2,out=90,in=90]
    node[right] {$a$} (s1);
    \path (s0) edge node[above] {$a$} (s3);
    \path (s3) edge [bend left=60]
    node[inner sep=0,outer sep=0,minimum size=0,pos=0.5,name=XX] {}
    node[above] {$b$} (s0);
    \path (XX) edge[bend right] node[below,pos=0.4] {$b$} (s2);
    \path (s2) edge[densely dashed] node[below] {$b$} (s1);
    \path (s3) edge[densely dashed] node[right] {$b$} (s1);
    \path (s1) edge[loop right,densely dashed] node[below]
    {$a, b$\,\,\,\,\,} (s1);
  \end{tikzpicture}
  \caption{%
    \label{fig:nu-to-dmts}
    $\nu$-calculus expression in normal form and its DMTS translation,
    \cf~Example~\ref{ex:nu-to-dmts}.  The state corresponding to
    $\lfalse$ is inconsistent and not shown}
\end{figure}

\begin{example}\label{ex:nu-to-dmts}
  Consider the $\nu$-calculus formula
  \begin{equation*}
    X = \big(\langle a\rangle (\langle
    b\rangle X \land [a] \lfalse ) \land [b] \lfalse \big) \lor [a]\lfalse\,.
  \end{equation*}
  Converting the formula into the normal form of
  Lemma~\ref{le:hmlnormal} yields the result given in
  Figure~\ref{fig:nu-to-dmts} (left), where both $X_1$ and $X_2$ are
  initial variables. The resulting DMTS is illustrated in
  Figure~\ref{fig:nu-to-dmts} (right). \qed
\end{example}

The following theorem follows easily:

\begin{theorem}
  \label{th:dmtsnu}
  For all DMTS $\mcalD_1$, $\mcalD_2$, $\mcalD_1\mr \mcalD_2$ iff
  $\dn( \mcalD_1)\mr \dn( \mcalD_2)$.  For all $\nu$-calculus
  expressions $\mcalN_1$, $\mcalN_2$ in normal form,
  $\mcalN_1\mr \mcalN_2$ iff $\nd( \mcalN_1)\mr \nd(
  \mcalN_2)$. \noproof
\end{theorem}

As a consequence, we can now show that the fixed-point semantics and
our refinement semantics for the modal $\nu$-calculus agree:

\begin{theorem}
  \label{th:models=ref-nu}
  For any LTS $\mcalI$ and any $\nu$-calculus expression $\mcalN$ in
  normal form, $\mcalI\models \mcalN$ iff $\mcalI\mr \mcalN$.
\end{theorem}

\begin{proof}
  We show that $\mcalI\mr \mcalD$ iff $\mcalI\models \dn( \mcalD)$ for any DMTS
  $\mcalD$; the claim then follows because $\mcalI\mr \mcalN$ iff $\mcalI\mr \nd(
  \mcalN)$ iff $\mcalI\models \dn( \nd( \mcalN))= \mcalN$.

  Write $\mcalI=( I, i^0, \omust_I)$, $\mcalD=( S, S^0, \omay, \omust)$, and
  $\dn( \mcalD)=( S, S^0, \Delta)$.

  We start with the only-if part. The proof is done by coinduction.  We
  define the assignment $\sigma : S \to 2^I$ as follows: $\sigma(t) =
  \{j \in I \mid j \mr t\}$.  We need to show that for every $s \in S$,
  $\sigma(s) \subseteq \lsem{ \Delta(s)} \sigma$.  Let $i \in
  \sigma(s)$.

  As $i \mr s$, we know that (1) $\forall s\DMTSmust{} N: \exists i\DMTSmust a_I
  j,( a, t)\in N: j\mr t$ and (2) $\forall i\DMTSmust a_I j: \exists s\DMTSmay a
  t: j\mr t$.

  Due to (1), we see that for all $N\in \Diamond( s)$, there is $i\DMTSmust
  a_I j$ and $( a, t)\in N$ such that $j\in \sigma( t)$ and $i \in
  \lsem{ \langle a\rangle t} \sigma$.  Hence $i\in \bigland_{ N\in
    \Diamond( s)} \biglor_{( a, t)\in N} \lsem{ \langle a\rangle t)}
  \sigma= \lsem{ \bigland_{ N\in \Diamond( s)} \biglor_{( a, t)\in N}
    \langle a\rangle t)} \sigma$.

  Due to (2), it holds that for every $a\in \Sigma$ and every
  $i\DMTSmust a_I j$, there is $t\in \Box^a( s)$ such that
  $j\in \sigma( t)\subseteq \lsem{ \biglor_{t \in \Box^a( s)}
    t}\sigma$.
  Hence
  $i\in \bigland_{ a\in \Sigma} \lsem{[ a]( \biglor_{t \in \Box^a( s)}
    t)} \sigma= \lsem{ \bigland_{ a\in \Sigma}[ a]( \biglor_{t \in
      \Box^a( s)} t)} \sigma$.
  Altogether, we have shown that $i\in \lsem{ \Delta( s)} \sigma$.

  Clearly, there is $s^0 \in S^0$ such that $i^0 \in
  \sigma(s^0)$. Therefore, $\mcalI \models \dn( \mcalD)$.

  For the other direction, define a~relation $R\subseteq I\times S$ by
  $R=\{( j, t)\mid j\models t\}$.  We show that $R$ satisfies the
  conditions of modal refinement.

  Let $( i, s)\in R$. As $i\models s$, we know that (1) $\forall N\in
  \Diamond( s): \exists( a, t)\in N: i\models \langle a\rangle t$ and
  (2) $\forall a\in \Sigma: i\models[ a]( \biglor_{ t\in \Box^a( s)}
  t)$.

  By (1), we know that for all $s\DMTSmust{} N$, there is $( a, t)\in N$
  and $i\DMTSmust a_I j$ such that $j\models t$.  By (2), it holds that for
  all $i\DMTSmust a_I j$, there is $s\DMTSmay a t$ so that $j\models t$.  We
  have shown that $i\mr s$.

  Clearly, there is $s^0\in S^0$ for which $( i^0, s^0)\in R$, hence
  $\mcalI\mr \mcalD$. \qed
\end{proof}

\subsection{\NAA\ vs.\ Hybrid Modal Logic}
\label{se:aahyb}

Also the translations between \NAA\ and our hybrid modal logic are
straight-forward.  For an \NAA\ $\mcalA=( S, S^0, \Tran)$ and all $s\in S$,
let
\begin{equation*}
  \Phi( s)= \biglor_{ M\in \Tran( s)} \Big( \bigland_{( a, t)\in M}
  \langle a\rangle t\land \bigland_{( b, u)\notin M} \neg \langle
  b\rangle u\Big)
\end{equation*}
and define the $\mcalL$-expression $\al( \mcalA)=( S, S^0, \Phi)$.

For an $\mcalL$-expression $\mcalE=( X, X^0, \Phi)$ and all $x\in X$, let
$\Tran( x)= \lsem{ \Phi( x)}$ and define the \NAA\ $\la( \mcalE)=( X, X^0,
\Tran)$.

\begin{theorem}
  \label{th:aahyb}
  For all \NAA\ $\mcalA_1$, $\mcalA_2$, $\mcalA_1\mr \mcalA_2$ iff
  $\al( \mcalA_1)\mr \al( \mcalA_2)$.  For all $\mcalL$-expressions
  $\mcalE_1$, $\mcalE_2$, $\mcalE_1\mr \mcalE_2$ iff
  $\la( \mcalE_1)\mr \la( \mcalE_2)$.
\end{theorem}

\begin{proof}
  We show that for any \NAA\ $( S, S^0, \Tran)$ and any $\mcalL$-expression
  $( S, S^0, \Phi)$, $\Tran( s)= \lsem{ \Phi( s)}$ for every $s\in S$,
  for both translations.  For the second one, $\la$, this is clear by
  definition, and for the first,
  \begin{align*}
    \lsem{ \Phi( s)} &= \lsem{ \biglor_{ M\in \Tran( s)} \big(
      \bigland_{( a, t)\in M} \langle a\rangle t\land \bigland_{( b,
        u)\notin M} \neg \langle b\rangle u\big)} \\
    &= \bigcup_{ M\in \Tran( s)} \big( \bigcap_{( a, t)\in M} \{ M'\mid(
    a, t)\in M'\} \cap \bigcap_{( b, u)\notin M} \{ M'\mid( b, u)\notin
    M'\}\big) \\
    &= \bigcup_{ M\in \Tran( s)} \big( \{ M'\mid \forall( a, t)\in M:(
    a, t)\in M'\} \\[-3ex]
    & \hspace*{14em} \cap\{ M'\mid \forall( b, u)\notin M:( b, u)\notin
    M'\}\big) \\[1ex]
    &= \bigcup_{ M\in \Tran( s)} \big( \{ M'\mid M\subseteq M'\}\cap\{
    M'\mid M'\subseteq M\}\big) \\
    &= \bigcup_{ M\in \Tran( s)}\{ M\}= \Tran( s)
  \end{align*}
  as was to be shown. \qed
\end{proof}

\subsection{DMTS vs.\ \NAA}
\label{se:dmtsaa}

The translations between DMTS and \NAA\ are somewhat more intricate.
For a DMTS $\mcalD=( S, S^0, \omay, \omust)$ and all $s\in S$, let
\begin{equation*}
  \Tran(s)=\{ M\subseteq \Sigma\times S\mid \forall (a,t)\in M: s\DMTSmay{a}
  t, \forall s\DMTSmust{} N: N\cap M\ne \emptyset\}
\end{equation*}
and define the \NAA\ $\da( \mcalD)=( S, S^0, \Tran)$.

For an \NAA\ $\mcalA=( S, S^0, \Tran)$, define the DMTS $\ad( \mcalA)=( D, D^0,
\omay, \omust)$ as follows:
\begin{align*}
  D &= \{ M\in \Tran( s)\mid s\in S\} \cup \{ \lightning \}\\
  D^0 &= \{ M^0\in \Tran( s^0)\mid s^0\in S^0\} \cup \{ \lightning \mid
	\text{if } \exists s^0 \in S^0 : \Tran(s^0) = \emptyset \}\\
  \omust &= \big\{\big( M,\{( a, M')\mid M'\in \Tran(
  t)\}\big)\bigmid( a, t)\in M, \Tran(t) \ne \emptyset\big\} \cup \phantom{x}\\
 &\quad\ \big\{\big(M, \{(a, \lightning)\}\big) \mid (a, t) \in M,
   \Tran(t) = \emptyset \big\} \cup \big\{\big(\lightning, \emptyset\big)\big\}\\
  \omay &= \{( M, a, M')\mid \exists M\DMTSmust{} N: ( a, M')\in N\}
\end{align*}
Note that the state spaces of $\mcalA$ and $\ad( \mcalA)$ are not the same;
the one of $\ad( \mcalA)$ may be exponentially larger.

\begin{theorem}
  \label{th:dmtsaa}
  For all DMTS $\mcalD_1$, $\mcalD_2$, $\mcalD_1\mr \mcalD_2$ iff
  $\da( \mcalD_1)\mr \da( \mcalD_2)$.  For all \NAA\ $\mcalA_1$,
  $\mcalA_2$, $\mcalA_1\mr \mcalA_2$ iff
  $\ad( \mcalA_1)\mr \ad( \mcalA_2)$.
\end{theorem}

\begin{proof}
  There are four parts to this proof, two implications to show the
  first claim of the theorem and two implications for the second
  claim.

  \proofpara{$\mcalD_1\mr \mcalD_2$ implies $\da( \mcalD_1)\mr \da( \mcalD_2)$:}

  Write $\mcalD_1=( S_1, S^0_1, \omay_1, \omust_1)$,
  $\mcalD_2=( S_2, S^0_2, \omay_2, \omust_2)$.  We have a modal
  refinement relation (in the DMTS sense) $R\subseteq S_1\times S_2$.
  Now let $( s_1, s_2)\in R$ and $M_1\in \Tran_1( s_1)$, and define
  \begin{equation*}
    M_2=\{( a, t_2)\mid s_2\DMTSmay{a}_2 t_2, \exists( a, t_1)\in M_1:( t_1,
    t_2)\in R\}.
  \end{equation*}

  We prove that $M_2\in \Tran_2( s_2)$.  First we notice that by
  construction, indeed $s_2\DMTSmay{a}_2 t_2$ for all $( a, t_2)\in M_2$.
  Now let $s_2\DMTSmust{}_2 N_2$; we need to show that
  $N_2\cap M_2\ne \emptyset$.

  By DMTS refinement, we have $s_1\DMTSmust{}_1 N_1$ such that $\forall( a,
  t_1)\in N_1: \exists( a, t_2)\in N_2:( t_1, t_2)\in R$.  We know that
  $N_1\cap M_1\ne \emptyset$, so let $( a, t_1)\in N_1\cap M_1$.  Then
  there also is $( a, t_2)\in N_2$ with $( t_1, t_2)\in R$.  But $( a,
  t_2)\in N_2$ implies $s_2\DMTSmay{a}_2 t_2$, hence $( a, t_2)\in M_2$.

  Now the condition
  \begin{equation*}
    \forall( a, t_2)\in M_2: \exists( a, t_1)\in M_1:( t_1, t_2)\in
    R
  \end{equation*}
  in the definition of \NAA\ refinement is satisfied by construction.  For
  the inverse condition, let $( a, t_1)\in M_1$, then $s_1\DMTSmay{a}_1
  t_1$, so by DMTS refinement, there is $t_2\in S_2$ with $s_2\DMTSmay{a}_2
  t_2$ and $( t_1, t_2)\in R$, whence $( a, t_2)\in M_2$ by construction.

  \proofpara{$\da( \mcalD_1)\mr \da( \mcalD_2)$ implies $\mcalD_1\mr \mcalD_2$:}

  Let $R\subseteq S_1\times S_2$ be a~modal refinement relation in the
  \NAA\ sense and $( s_1, s_2)\in R$.
  Let $s_1\DMTSmay{a}_1 t_1$ and
  $M_1=\{( a, t_1)\}\cup \bigcup_{ s_1\DMTSmust{}_1 N_1} N_1$, then
  $M_1\in \Tran_1( s_1)$ by construction.  As $R$ is a modal
  refinement, this implies that there is $M_2\in \Tran_2( s_2)$ and
  $( a, t_2)\in M_2$ with $( t_1, t_2)\in R$, but then also
  $s_2\DMTSmay{a}_2 t_2$ as was to be shown.

  Let $s_2\DMTSmust{}_2 N_2$ and assume, for the sake of contradiction, that
  there is no $s_1\DMTSmust{}_1 N_1$ for which $\forall( a, t_1)\in N_1:
  \exists( a, t_2)\in N_2:( t_1, t_2)\in R$ holds.  Then for each
  $s_1\DMTSmust{}_1 N_1$, there is an element $( a_{ N_1}, t_{ N_1})\in N_1$
  for which there is no $( a_{ N_1}, t_2)\in N_2$ with $( t_{ N_1},
  t_2)\in R$.

  Let $M_1=\{( a_{ N_1}, t_{ N_1})\mid s_1\DMTSmust{}_1 N_1\}$, then
  $M_1\in \Tran_1( s_1)$ by construction.  Hence we have
  $M_2\in \Tran_2( s_2)$ satisfying the conditions in the definition
  of \NAA\ refinement.  By construction of $\Tran_2( s_2)$,
  $s_2\DMTSmust{}_2 N_2$ and $N_2\cap M_2\ne \emptyset$, so let
  $( a, t_2)\in N_2\cap M_2$.  Then there exists $( a, t_1)\in M_1$
  for which $( t_1, t_2)\in R$, in contradiction to the definition of
  $M_1$.

  \proofpara{$\mcalA_1\mr \mcalA_2$ implies $\ad( \mcalA_1)\mr \ad( \mcalA_2)$:}

  Write $\mcalA_1=( S_1, S^0_1, \Tran_1)$,
  $\mcalA_2=( S_2, S^0_2, \Tran_2)$, with DMTS translations
  $( D_1, D^0_1, \omust_1, \omay_1)$,
  $( D_2, D^0_2, \omust_2, \omay_2)$.  We have a modal refinement
  relation (in the \NAA\ sense) $R\subseteq S_1\times S_2$.  Define
  $R'\subseteq D_1\times D_2$ by
  \begin{multline*}
    R'= \{(\lightning, M_2) \mid M_2 \in D_2 \}\\
    \cup \{( M_1, M_2)\mid \exists( s_1, s_2)\in R: M_1\in \Tran_1( s_1),
    M_2\in \Tran( s_2), \\
    \begin{aligned}
      & \forall( a, t_1)\in M_1: \exists( a, t_2)\in M_2:( t_1, t_2)\in
      R, \\
      & \forall( a, t_2)\in M_2: \exists( a, t_1)\in M_1:( t_1, t_2)\in
      R \}.
    \end{aligned}
  \end{multline*}
  We show that $R'$ is a modal refinement in the DMTS sense.  Let $(
  M_1, M_2)\in R'$.

  If $M_1 = \lightning$ then $R'$ trivially satisfies the modal refinement
  conditions as there is no $\lightning \DMTSmay{a}$ transition and every
  $M_2 \DMTSmust{} N_2$ is matched by $\lightning \DMTSmust{} \emptyset$. Let us henceforth
  assume that $M_1 \ne \lightning$.

  Let $M_2\DMTSmust{}_2 \{(a, \lightning)\}$. By construction of $\omust_2$, there is
  $(a, t_2) \in M_2$ with $\Tran_2(t_2) = \emptyset$. Then $(M_1, M_2) \in R'$
  implies that there must be $(a, t_1) \in M_1$ for which $(t_1, t_2) \in R$.
  As $R$ is a \NAA\ refinement, this means that $\Tran_1(t_1) = \emptyset$
  and thus $M_1 \DMTSmust{}_1 \{(a, \lightning)\}$. Clearly, $(a, \lightning)$ is matched
  by $(a, \lightning)$ and $(\lightning, \lightning) \in R'$.

  Let $M_2\DMTSmust{}_2 N_2 \ne \{(a, \lightning)\}$.
  By construction of $\omust_2$, there is $( a,
  t_2)\in M_2$ such that $N_2=\{( a, M_2')\mid M_2'\in \Tran_2( t_2)\}
  \ne \emptyset$.
  Then $( M_1, M_2)\in R'$ implies that there must be $( a, t_1)\in M_1$
  for which $( t_1, t_2)\in R$.

  If $\Tran_1(t_1) = \emptyset$, we know that $M_1 \DMTSmust{}_1 \{(a, \lightning)\}$.
  We can then match $(a, \lightning)$ with arbitrary $(a, M'_2) \in N_2$ as
  $(\lightning, M'_2) \in R'$.

  If $\Tran_1(t_1) \ne \emptyset$, we can define $N_1=\{( a, M_1')\mid
  M_1'\in \Tran_1( t_1)\}$, whence $M_1\DMTSmust{}_1 N_1$.
  We show that $\forall( a, M_1')\in N_1: \exists( a, M_2')\in N_2:(
  M_1', M_2')\in R'$.  Let $( a, M_1')\in N_1$, then $M_1'\in \Tran_1(
  t_1)$.  From $( t_1, t_2)\in R$ we hence get $M_2'\in \Tran_2( t_2)$,
  and then $( a, M_2')\in N_2$ by construction of $N_2$ and $( M_1',
  M_2')\in R'$ due to the conditions of \NAA\ refinement (applied to $(
  t_1, t_2)\in R$).

  Let $M_1\DMTSmay{a}_1 \lightning$. We then have $M_1\DMTSmust{} \{(a, \lightning)\}$
  which means that there is $(a, t_1) \in M_1$ with $\Tran_1(t_1) = \emptyset$.
  By $(M_1, M_2) \in R'$ we get $(a, t_2) \in M_2$, hence
  $M_2 \DMTSmust{}_2 N_2$ with some $(a, M_2') \in N_2$ and $M_2 \DMTSmay{a} M_2'$.
  By definition of $R'$, $(\lightning, M_2') \in R'$.

  Let $M_1\DMTSmay{a}_1 M_1' \ne \lightning$, then we have $M_1\DMTSmust{}_1 N_1$ for which
  $(a, M_1')\in N_1$ and by construction of $\omay_1$.  This in turn implies
  that there must be $( a, t_1)\in M_1$ such that $N_1=\{( a, M_1'')\mid
  M_1''\in \Tran_1( t_1)\} \ne \emptyset$,
  and then by $( M_1, M_2)\in R'$, we get $(
  a, t_2)\in M_2$ for which $( t_1, t_2)\in R$.
  Due to the conditions of \NAA\ refinement, $\Tran_2(t_2) \ne \emptyset$.
  Let $N_2=\{( a,
  M_2')\mid M_2'\in \Tran_2( t_2)\}$, then $M_2\DMTSmust{}_2 N_2$ and hence
  $M_2\DMTSmay{a}_2 M_2'$ for all $( a, M_2')\in N_2$.
  Furthermore, the argument above shows that there is $( a,
  M_2')\in N_2$ for which $( M_1', M_2')\in R'$.

  We miss to show that $R'$ is initialised.
  If $\lightning \in D_1^0$, then $(\lightning, M_2^0) \in R'$ for any $M_2^0 \in D_2^0$.
  If $\lightning \ne M_1^0\in D_1^0$, then
  we have $s_1^0\in S_1^0$ with $M_1^0\in \Tran_1( s_1^0)$.  As $R$ is
  initialised, this entails that there is $s_2^0\in S_2^0$ with $(
  s_1^0, s_2^0)\in R$, which gives us $M_2^0\in \Tran_2( s_2^0)$ which
  satisfies the \NAA\ refinement conditions, whence $( M_1^0, M_2^0)\in R'$.

  \proofpara{$\ad( \mcalA_1)\mr \ad( \mcalA_2)$ implies $\mcalA_1\mr \mcalA_2$:}

  Let $R\subseteq D_1\times D_2$ be a~modal refinement relation in the
  DMTS sense and define $R'\subseteq S_1\times S_2$ by
  \begin{equation*}
    R'=\{( s_1, s_2)\mid \forall M_1\in \Tran_1( s_1): \exists M_2\in
    \Tran_2( s_2):( M_1, M_2)\in R\}\,;
  \end{equation*}
  we will show that $R'$ is an \NAA\ modal refinement.

  Observe first that $(M_1, \lightning) \in R$ implies $M_1 = \lightning$ as
  $\lightning$ is the only state in $D_1$ that has a must transition
  $\DMTSmust{}_1 \emptyset$. We shall occasionally refer to this observation in
  the following.

  Let $( s_1, s_2)\in R'$ and $M_1\in \Tran_1( s_1)$, then by
  construction of $R'$, we have $M_2\in \Tran_2( s_2)$ with $( M_1,
  M_2)\in R$.

  Let $( a, t_2)\in M_2$; we need to find $( a, t_1)\in M_1$ such that
  $( t_1, t_2)\in R'$.

  If $\Tran_2(t_2) = \emptyset$, then $M_2 \DMTSmust{}_2 \{(a, \lightning)\}$.
  From $(M_1, M_2) \in R$ we get $M_1 \DMTSmust{}_1 N_1$ such that
  all $(a, M_1') \in N_1$ are matched by $(a, \lightning)$ with $(M_1', \lightning) \in R$.
  By the observation above, this means that $M_1' = \lightning$ and thus
  $N_1 = \{(a, \lightning)\}$ due to the construction. This means that
  there exists $(a, t_1) \in M_1$ with $\Tran_1(t_1) = \emptyset$.
  Hence $(t_1, t_2) \in R'$.

  If $\Tran_2(t_2) \ne \emptyset$,
  define $N_2=\{( a, M_2')\mid M_2'\in \Tran_2( t_2)\}$, then
  $M_2\DMTSmust{}_2 N_2$.  Now $( M_1, M_2)\in R$ implies that there must
  be $M_1\DMTSmust{}_1 N_1$ satisfying
  $\forall( a, M_1')\in N_1: \exists( a, M_2')\in N_2:( M_1', M_2')\in
  R$.
  If $N_1 = \{(a, \lightning)\}$, we have $(a, t_1) \in M_1$ with
  $\Tran_1(t_1) = \emptyset$ and thus trivially $(t_1, t_2) \in R'$.
  Otherwise, we have $( a, t_1)\in M_1$ such that
  $N_1=\{( a, M_1')\mid M_1'\in \Tran_1( t_1)\}$; we only miss to show
  that $( t_1, t_2)\in R'$.

  Let $M_1'\in \Tran_1( t_1)$, then $( a, M_1')\in N_1$, hence there
  is $( a, M_2')\in N_2$ with $( M_1', M_2')\in R$, but
  $( a, M_2')\in N_2$ also entails $M_2'\in \Tran_2( t_2)$; thus
  $(t_1, t_2) \in R'$.

  Let now $( a, t_1)\in M_1$; we need to find $( a, t_2)\in M_2$ such that
  $( t_1, t_2)\in R'$.

  If $\Tran_1(t_1) = \emptyset$, then $M_1 \DMTSmust{}_1 \{(a, \lightning)\}$
  and $M_1 \DMTSmay{a} \lightning$. By modal refinement, we have $M_2 \DMTSmay{a} M_2'$
  with $(\lightning, M_2') \in R$. In any case (whether $M_2' = \lightning$ or not),
  this means that there exists some $(a, t_2) \in M_2$ and trivially
  $(t_1, t_2) \in R'$.

  In case $\Tran_1( t_1)\ne \emptyset$, define
  $N_1=\{( a, M_1')\mid M_1'\in \Tran_1( t_1)\}$, then
  $N_1\ne \emptyset$ and $M_1\DMTSmust{}_1 N_1$.  Now let
  $( a, M_1')\in N_1$, then $M_1\DMTSmay{a}_1 M_1'$, hence we have
  $M_2\DMTSmay{a}_2 M_2'$ for some $( M_1', M_2')\in R$ by modal
  refinement.  Note that $M_1' \ne \lightning$ implies
  $M_2' \ne \lightning$ due to the observation above.
  By construction of $\omay_2$, this implies that there
  is $M_2\DMTSmust{}_2 N_2$ with $( a, M_2')\in N_2$, and we have
  $( a, t_2)\in M_2$ for which
  $N_2=\{( a, M_2'')\mid M_2''\in \Tran_2( t_2)\}$.  We show that
  $( t_1, t_2)\in R'$.

  Let $M_1''\in \Tran_1( t_1)$, then $( a, M_1'')\in N_1$, thus
  $M_1\DMTSmay a_1 M_1''$, so that there is $M_2\DMTSmay a_2 M_2''$ with $(
  M_1'', M_2'')\in R$.  By construction of $\omay_2$, there is
  $M_2\DMTSmust{}_2 N_2'$ with $( a, M_2'')\in N_2'$, hence also $M_2''\in
  \Tran_2( t_2)$.

  We miss to show that $R'$ is initialised.  Let $s^0_1\in S^0_1$; if
  $\Tran_1( s^0_1)= \emptyset$, then trivially $( s^0_1, s^0_2)\in R'$
  for any $s^0_2\in S^0_2$.  If $\Tran_1( s^0_1)\ne \emptyset$, then
  there is $M^0_1\in \Tran_1( s^0_1)$.  As $R$ is initialised, this
  gets us $M^0_2\in D_2$ with $( M^0_1, M^0_2)\in R$, but
  $M^0_2\in \Tran_2( s^0_2)$ for some $s^0_2\in S^0_2$, and then
  $( s^0_1, s^0_2)\in R'$.
\end{proof}

\begin{corollary}
  For all DMTS $\mcalD$, $\nu$-calculus expressions $\mcalN$, \NAA\
  $\mcalA$, and $\mcalL$-ex\-pressions $\mcalE$,
  $\dn( \mcalD)\DMTStreq \da( \mcalD)\DMTStreq \mcalD$,
  $\nd( \mcalN)\DMTStreq \mcalN$,
  $\ad( \mcalA)\DMTStreq \al( \mcalA)\DMTStreq \mcalA$, and
  $\la( \mcalE)\DMTStreq \mcalE$. \noproof
\end{corollary}

\subsection{Translation Complexity}

We have shown that our four specification formalisms are structurally
equivalent, which will be useful for us from a theoretical point of
view.  From a practical point of view however, some of the translations
may incur exponential blow-ups, hence care has to be taken.  On the
other hand, all our translations can be implemented in an on-the-fly
manner, only creating states when necessary.

We already noticed that the translation of $\nu$-calculus expressions
into normal form may incur an exponential blow-up, so this also affects
our translation from the modal $\nu$-calculus to DMTS.  When considering
only normal-form expressions, the translations to and from DMTS incur no
blow-ups.

When translating from \NAA\ to our hybrid modal logic, we see that, due to
the complementation $( b, u)\notin M$, the length of a formula $\Phi(
s)$ is quadratic in the representation of $\Tran( s)$.  For the reverse
translation, the number of $\Tran$ constraints can be exponential in the
number of states.  The worst case is $\Phi( s)= \ltrue$, which gets
translated to $\Tran( s)= 2^{ 2^{ \Sigma\times S}}$.

\begin{remark}
  There is a direct translation from DMTS to hybrid modal logic: for a
  DMTS $\mcalD=( S, S^0, \omay, \omust)$, define
  $\dl( \mcalD)=( S, S^0, \Phi)$ with
  \begin{equation*}
    \Phi( s)= \bigland_{ s\DMTSmust{} N} \biglor_{( a, t)\in
      N} \langle a\rangle t\land \bigland_{ s\not\DMTSmay a u} \neg \langle
    a\rangle u
  \end{equation*}
  for all $s\in S$.  This translation is again quadratic, and
  $\da( \mcalD)= \la( \dl( \mcalD))$.
\end{remark}

The translations between DMTS and \NAA\ may involve exponential blow-ups
both ways.  For the first translation, we can see this by considering
the one-state DMTS $(\{ s\},\{ s\}, \omay, \omust)$ with $\omay=\{( s,
a, s)\mid a\in \Sigma\}$ and $\omust= \emptyset$.  Then $\Tran( s)= 2^{
  2^{ \Sigma\times S}}$.

The fact that also the translation from \NAA\ to DMTS may be exponential
in space is evident from the definition.  To see that this blow-up is
unavoidable, we expose a special property of the $\Tran$-sets arising in
the DMTS-to-\NAA\ translation.

\begin{lemma}
  \label{le:dmtstobfsspecial}
  Let $\mcalD=( S, S^0, \omay, \omust)$ be a DMTS and $s\in S$.  For all
  $M_1, M_2\in \Tran( s)$ and all $M\subseteq \Sigma\times S$ with
  $M_1\subseteq M\subseteq M_1\cup M_2$, also $M\in \Tran( s)$.
\end{lemma}

\begin{proof}
  For $i=1,2$, since $M_i \in \Tran(s)$, we know that
  \begin{itemize}
  \item for all $(a,t) \in M_i$, $s \DMTSmay a t$, and
  \item for all $s\DMTSmust{} N$, there is $(a,t) \in M_i \cap N$.
  \end{itemize}

  Now as $M \subseteq M_1 \cup M_2$, it directly follows that for all
  $(a,t) \in M$, we have $s\DMTSmay a t$. Moreover, since $M_1 \subseteq M$,
  we also have that for all $s\DMTSmust{} N$, there exists $(a,t) \in M \cap
  N$.  As a~consequence, $M \in \Tran(s)$. \qed
\end{proof}

Using this, we can show the following.

\begin{proposition}
  \label{pr:bfstodmtsblowup}
  There exists a one-state \NAA\ $\mcalA$ for which any DMTS $\mcalD\DMTStreq \mcalA$
  has at least $2^{ n- 1}$ states, where $n$ is the size of the alphabet
  $\Sigma$.
\end{proposition}

\begin{proof}
  Let $\Sigma=\{ a_1,\dots, a_n\}$ and $\mcalA=(\{ s^0\},\{ s^0\}, \Tran)$
  the \NAA\ with $\Tran( s^0)=\{ M\subseteq \Sigma\times\{ s^0\}\mid
  \exists k:| M|= 2k\}$ the transition constraint containing all
  disjunctive choices of even cardinality.  Let $\mcalD=( T, T^0, \omay,
  \omust)$ be a DMTS with $\mcalD\DMTStreq \mcalA$; we claim that $\mcalD$ must have
  at least $2^{ n- 1}$ initial states.

  Assume, for the purpose of contradiction, that
  $T^0=\{ t^0_1,\dots, t^0_m\}$ with $m< 2^{ n- 1}$.  As
  $\mcalD\DMTStreq \mcalA$, we must have
  $\bigcup_{ i= 1}^m \Tran( t^0_i)=\{ M\subseteq \Sigma\times T\mid
  \exists k:| M|= 2k\}$,
  so that there is an index $j\in\{ 1,\dots, m\}$ for which
  $\Tran( t^0_j)=\{ M_1, M_2\}$ contains two different disjunctive
  choices from $\Tran( s^0)$.  By Lemma~\ref{le:dmtstobfsspecial},
  also $M\in \Tran( t^0_j)$ for any $M$ with
  $M_1\subseteq M\subseteq M_1\cup M_2$.  But $M_1\cup M_2$ has
  greater cardinality than $M_1$, so that there will be an
  $M\in \Tran( t^0_j)$ with odd cardinality. \qed
\end{proof}

\begin{figure}
  \hfill
  \xymatrix@R=12ex@C=4em{%
    & \text{DMTS} \ar@<.4ex>[r]^{\text{X}} \ar@<.4ex>[d]^{\text{L}}
    \ar[dr]_{\text{Q}}
    & \text{\NAA} \ar@<.4ex>[l]^{\text{X}} \ar@<.4ex>[d]^{\text{Q}}
    \\ \text{$\nu$-calculus} \ar@<.4ex>[r]^{\text{X}}
    & \text{$\nu$-normal form} \ar@<.4ex>[l]^{\text{L}}
    \ar@<.4ex>[u]^{\text{L}}
    & \text{hybrid modal logic} \ar@<.4ex>[u]^{\text{X}}
  }\hfill\mbox{}
  \caption{%
    \label{fi:translations-comp}
    Complexity of the translations between specification formalisms.
    ``L'' stands for linear (no blow-up), ``Q'' for quadratic blow-up,
    and ``X'' for exponential blow-up}
\end{figure}
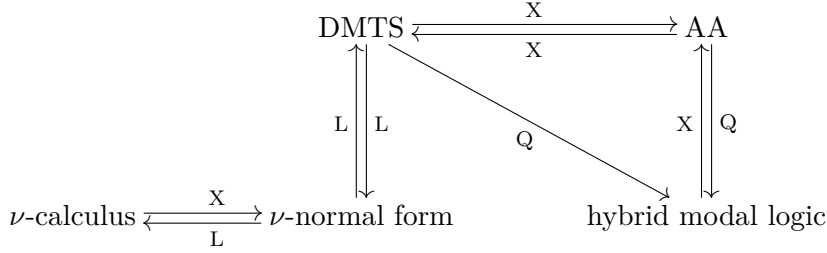

Figure~\ref{fi:translations-comp} sums up the translation complexities.

\subsection{Initial States}
\label{se:initial}

We finish this section with a justification for why we allow our
specifications to have several (or possibly zero) initial states.  The
first lemma shows that for \NAA, and up to \emph{thorough} refinement,
this is inessential; due to their close relationship, this also holds
for $\mcalL$-expressions.

\begin{lemma}
  \label{le:aa-single}
  For any \NAA\ $\mcalA_1$, there is an \NAA\ $\mcalA_2=( S_2, S^0_2, \Tran_2)$ with
  $S^0_2=\{ s^0_2\}$ a singleton and $\mcalA_1\DMTStreq \mcalA_2$.
\end{lemma}

\begin{proof}
  Write $\mcalA_1=( S_1, S^0_1, \Tran_1)$.  If $S^0_1= \emptyset$, we can
  let $S_2= S^0_2=\{ s^0_2\}$ and $\Tran_2( s^0_2)=
  \emptyset$.
  Otherwise, we let $S_2= S_1\cup\{ s^0_2\}$, where $s^0_2$ is a new
  state, and $\Tran_2( s)= \Tran_1( s)$ for $s\in S_1$,
  $\Tran_2( s^0_2)= \bigcup_{ s_1^0\in S_1^0} \Tran_1( s_1^0)$.  Let
  $R=\id_{ S_1}\cup\{( s^0_1, s^0_2)\mid s^0_1\in S^0_1\}$, then $R$
  is easily seen to be a modal refinement showing $\mcalA_1\mr \mcalA_2$.

  We show that $\mcalA_2\DMTStr \mcalA_1$.  Let $\mcalI=( S, s^0, \omust)\in \sem{
    \mcalA_2}$, then we have a modal refinement $R_2\subseteq S\times S_2$,
  \ie~such that for all $( s, s_2)\in R_2$, there exists $M_2\in
  \Tran_2( s_2)$ for which
  \begin{equation}
    \label{eq:ltsmraa}
    \begin{gathered}
      \forall s\DMTSmust a t: \exists( a, t_2)\in M_2:( t, t_2)\in R_2\,, \\
      \forall( a, t_2)\in M_2: \exists s\DMTSmust a t:( t, t_2)\in R_2\,.
    \end{gathered}
  \end{equation}
  Now $( s^0, s_2^0)\in R_2$ implies that there must be $M_2\in \Tran_2(
  s_2^0)$ for which~\eqref{eq:ltsmraa} holds, but by definition of
  $\Tran_2( s_2^0)$, this entails that there is $s_1^0\in S_1^0$ for
  which $M_2\in \Tran_1( s_1^0)$.  Define $R_1\subseteq S\times S_1$ by
  \begin{equation*}
    R_1=\{( s, s_2)\mid( s, s_2)\in R_2, s_2\ne s^0_2\}\cup\{( s^0,
    s_1^0)\}\,,
  \end{equation*}
  then $R_1$ is a modal refinement showing $\mcalI\mr \mcalA_1$. \qed
\end{proof}

In order to show that the above statement does \emph{not} hold for DMTS,
we expose a special property of DMTS with single initial states,
\cf~\cite[Example~7.8]{DBLP:journals/tcs/BenesKLS09}.  Recall that for
LTS $\mcalI_1=( S_1, s_1^0, \omust_1)$, $\mcalI_2=( S_2, s_2^0, \omust_2)$,
their \emph{nondeterministic sum} is given by $\mcalI_1+ \mcalI_2=( S, s^0,
\omust)$ with $S=S_1\cup S_2\cup\{ s^0\}$ (with the unions disjoint),
where $s^0$ is a new state,
and transitions $s\DMTSmust a t$ iff $s\DMTSmust a_1 t$ or $s\DMTSmust a_2 t$
together with $s^0\DMTSmust a t$ for all $t$ with $s_1^0\DMTSmust a_1 t$ or
$s_2^0 \DMTSmust a_2 t$.

\begin{lemma}
  \label{le:dmts1sumclosed}
  If $\mcalD=( S,\{ s^0\}, \omay, \omust)$ is a DMTS with a single initial
  state and $\mcalI_1, \mcalI_2\in \sem \mcalD$, then also $\mcalI_1+ \mcalI_2\in \sem
  \mcalD$.
\end{lemma}

\begin{proof}
  Let $i^0_1$ and $i^0_2$ be the initial states of $\mcalI_1$ and $\mcalI_2$,
  respectively. Let further $i^0$ be the initial state of $\mcalI_1 + \mcalI_2$.
  Assume that we have modal refinements $R_1$ and $R_2$ such that
  $(i^0_1,s^0) \in R_1$ and $(i^0_2,s^0) \in R_2$. Let
  $R = R_1 \cup R_2 \cup \{ (i^0,s^0) \}$. Clearly, $R$ is a~modal
  refinement witnessing $\mcalI_1+\mcalI_2 \mr \mcalD$.\qed
\end{proof}

Now let $\mcalD$ be the DMTS, with \emph{two} initial states, depicted in
Figure~\ref{fi:dmts2i}, then $\sem \mcalD=\{ \mcalI_1, \mcalI_2\}$ as also seen in
Figure~\ref{fi:dmts2i}, but $\mcalI_1+ \mcalI_2\notin \sem \mcalD$.  Hence $\mcalD$ is
not thoroughly equivalent to any DMTS with a single initial state.

Applying the construction from the proof of Lemma~\ref{le:aa-single}
to the \NAA\ generated by the DMTS in Figure~\ref{fi:dmts2i} gives an
\NAA\ $\mcalA_2=( S_2,\{ s^0_2\}, \Tran_2)$ with
$\Tran_2( s^0_2)=\{\{( a, s_1)\},\{( b, t_1)\}\}$ (where $s_1$ and
$t_1$ are the target states of the $a$ and $b$ transitions in $\mcalD$,
respectively).  This specifies an \emph{exclusive disjunction}: one of
$a$ and $b$ has to be implemented, but not both.  This also serves to
show that Lemma~\ref{le:dmts1sumclosed} does not hold for \NAA.

\begin{figure}
  \centering
  \begin{tikzpicture}[->, >=stealth', font=\footnotesize,
    state/.style={ shape= circle, draw, initial text=, inner
      sep= .5mm, minimum size= 2mm}, scale= 1]
    \begin{scope}
      \node[state,initial] (i1) at (0,0) {};
      \node[state,initial] (i2) at (0,-1) {};
      \node[state] (a) at (2,0) {};
      \node[state] (b) at (2,-1) {};
      \node at (1,-2) {$\mcalD$};
      \path (i1) edge node[above] {$a$} (a);
      \path (i2) edge node[above] {$b$} (b);
    \end{scope}
    \begin{scope}[xshift=12em]
      \node[state,initial] (i1) at (0,0) {};
      \node[state] (a) at (2,0) {};
      \node at (1,-2) {$\mcalI_1$};
      \path (i1) edge node[above] {$a$} (a);
    \end{scope}
    \begin{scope}[xshift=24em]
      \node[state,initial] (i2) at (0,-1) {};
      \node[state] (b) at (2,-1) {};
      \node at (1,-2) {$\mcalI_2$};
      \path (i2) edge node[above] {$b$} (b);
    \end{scope}
  \end{tikzpicture}
  \caption{%
    \label{fi:dmts2i}
    DMTS $\mcalD$ with two initial states and its two implementations
    $\mcalI_1$, $\mcalI_2$}
\end{figure}
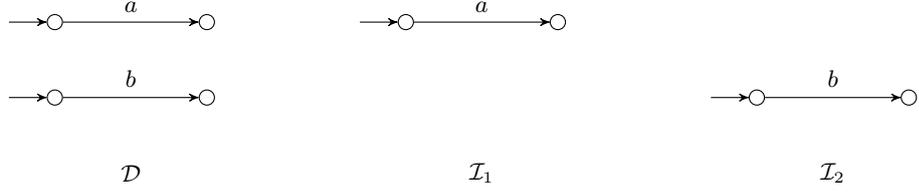

\begin{corollary}
  There is a DMTS $\mcalD_1$ for which there is no DMTS $\mcalD_2=( S_2,
  S_2^0,$ $\omay, \omust)$ with $S_2^0=\{ s_2^0\}$ a singleton and
  $\mcalD_1\DMTStreq \mcalD_2$. \noproof
\end{corollary}

Due to their close relationship with DMTS, this property also holds for
$\nu$-calculus expressions \emph{in normal form}: there exist
$\nu$-calculus expressions which are not equivalent to any normal-form
$\nu$-calculus expression with a single initial variable.  (Of course,
omitting ``normal form'' would make this statement invalid; as
disjunction is part of the syntax, any $\nu$-calculus expression is
thoroughly equivalent to one with only one initial variable.)

We also remark that the above argument can easily be extended to show
that for any $k\in \Nat$, there exists a DMTS with $k+ 1$ initial states
which is not thoroughly equivalent to any DMTS with at most $k$ initial
states.

Using again the example in Figure~\ref{fi:dmts2i}, we can also show that
the statement in Lemma~\ref{le:aa-single} does \emph{not} hold when
thorough equivalence is replaced by modal equivalence.  Let
$\mcalA_1= \da( \mcalD)$, with initial states $s^0$ and $t^0$, be the \NAA\
translation of the DMTS in Figure~\ref{fi:dmts2i} and assume that there
exists an \NAA\ $\mcalA_2$ with single initial state $s_2^0$ for which
$\mcalA_2\mr \mcalA_1$.  Then there is a modal refinement $R$ with
$( s_2^0, s^0),( s_2^0, t^0)\in R$.  Let $M_2\in \Tran_2( s_2^0)$,
then by $( s_2^0, s^0)\in R$, there must be some $( a, t_2)\in M_2$
with $( t_2, s_1)\in R$.  By $( s_2^0, t^0)\in R$, this implies that
there must be $( a, t_1)\in \Tran_1( t^0)$, a contradiction.

\section{Specification Theory}\label{sec:theory}

Behavioral specifications typically come equipped with operations which
allow for \emph{compositional reasoning}, \viz~conjunction, composition
and quotient, \cf~\cite{DBLP:conf/fase/BauerDHLLNW12}.  On deterministic
MTS, these operations can be given easily using simple structural
operational rules.  For non-deterministic systems this is significantly
harder.

We remark that composition and quotient operators are well-known from
some logics, such as, \eg~linear~\cite{DBLP:journals/tcs/Girard87} or
spatial logic~\cite{DBLP:journals/iandc/CairesC03}, and were extended to
quite general contexts~\cite{DBLP:conf/icalp/CardelliLM11}.  However,
whereas these operators are part of the formal syntax in those logics,
for us they are simply operations on logical expressions (or DMTS, or
\NAA).  Consequently, composition is generally only a sound
over-approximation of the semantic composition.

Given the structural equivalence of DMTS, the modal $\nu$-calculus, \NAA,
and our hybrid modal logic exposed in the previous section, it suffices
to introduce the operations for \emph{one} of the four types of
specifications.  On the other hand, we will often state properties for
\emph{all} four types of specifications at the same time, letting $\mcalS$
stand for a specification of any type.

\subsection{Disjunction and Conjunction}
\label{se:discon}

Disjunction of specifications is easily defined as we allow multiple
initial states.  For DMTS $\mcalD_1=( S_1, S_1^0, \omay_1, \omust_1)$,
$\mcalD_2=( S_2, S_2^0, \omay_2, \omust_2)$, we can hence define $\mcalD_1\lor
\mcalD_2= (S_1 \cup S_2, S^0_1 \cup S^0_2, \omay_1 \cup \omay_2,
\omust_1\cup\omust_2)$ (with all unions disjoint).  Similar definitions
are available for the other types of specifications, and disjunction
commutes with the translations.

Conjunction for DMTS is an extension of the construction
from~\cite{DBLP:conf/atva/BenesCK11} for multiple initial states.
Given two DMTS $\mcalD_1=( S_1, S^0_1, \omay_1, \omust_1)$,
$\mcalD_2=( S_2, S^0_2, \omay_2, \omust_2)$, we define
$\mcalD_1\land \mcalD_2=( S, S^0, \omay, \omust)$ with $S= S_1\times S_2$,
$S^0= S^0_1\times S^0_2$, and
\begin{itemize}
\item $( s_1, s_2)\DMTSmay a( t_1, t_2)$ iff $s_1\DMTSmay a_1 t_1$ and $s_2\DMTSmay
  a_2 t_2$,
\item for all $s_1\DMTSmust{}_1 N_1$, $( s_1, s_2)\DMTSmust{} \{( a,( t_1,
  t_2))\mid( a, t_1)\in N_1, s_2\DMTSmay a_2 t_2\}$,
\item for all $s_2\DMTSmust{}_2 N_2$, $( s_1, s_2)\DMTSmust{} \{( a,( t_1,
  t_2))\mid( a, t_2)\in N_2, s_1\DMTSmay a_1 t_1\}$.
\end{itemize}

For \NAA, conjunction can be defined using auxiliary projection
functions
$\pi_i: 2^{ \Sigma\times S_1\times S_2}\to 2^{ \Sigma\times S_i}$
given by
\begin{align*}
  \pi_1( M)=&\{( a, s_1)\mid \exists s_2\in
  S_2:( a, s_1, s_2)\in M\}\,,\\
  \pi_2( M)=&\{( a, s_2)\mid \exists s_1\in S_1:( a, s_1, s_2)\in M\}\,.
\end{align*}
Then for \NAA\ $\mcalA_1=( S_1, S^0_1, \Tran_1)$,
$\mcalA_2=( S_2, S^0_2, \Tran_2)$, we let
$\mcalA_1\land \mcalA_2=( S, S^0, \Tran)$, with $S= S_1\times S_2$,
$S^0= S^0_1\times S^0_2$ and
$\Tran(( s_1, s_2))=\{ M\subseteq \Sigma\times S\mid \pi_1( M)\in
\Tran_1( s_1), \pi_2( M)\in \Tran_2( s_2)\}$.

We can also define conjunction for $\mcalL$-expressions, using similar
auxiliary mappings on formulae.  For sets $X_1$, $X_2$ and
$i\in\{ 1, 2\}$, we define
$\rho_i: \mcalL( X_i)\to \mcalL( X_1\times X_2)$\label{pg:rho-L}
inductively, by
\begin{itemize}
\item $\rho_i( \ltrue)= \ltrue$, $\rho_i( \lfalse)= \lfalse$, $\rho_i(
  \neg \phi)= \neg \rho_i( \phi)$, $\rho_i( \phi_i\land \phi_2)= \rho_i(
  \phi_i)\land \rho_i( \phi_2)$,
\item $\rho_1( \langle a\rangle x_1)= \biglor_{ x_2\in X_2} \langle
  a\rangle( x_1, x_2)$,
\item $\rho_2( \langle a\rangle x_2)= \biglor_{ x_1\in X_1} \langle
  a\rangle( x_1, x_2)$.
\end{itemize}
Then, for $\mcalL$-expressions $\mcalE_1=( X_1, X_1^0, \Phi_1)$,
$\mcalE_2=( X_2, X_2^0, \Phi_2)$, we let
$\mcalE_1\land \mcalE_2=( X_1\times X_2, X_1^0\times X_2^0, \Phi)$ with
$\Phi(( x_1, x_2))= \rho_1( \Phi_1( x_1))\land \rho_2( \Phi_2( x_2))$.

\begin{lemma}
  For all DMTS $\mcalD_1$, $\mcalD_2$, \NAA\ $\mcalA_1$, $\mcalA_2$, and
  $\mcalL$-expressions $\mcalE_1$, $\mcalE_2$, $\da( \mcalD_1\land \mcalD_2)= \da(
  \mcalD_1)\land \da( \mcalD_2)$,
  $\al( \mcalA_1\land \mcalA_2)= \al( \mcalA_1)\land
  \al( \mcalA_2)$, and $\la( \mcalE_1\land \mcalE_2)= \la( \mcalE_1)\land \la(
  \mcalE_2)$.
\end{lemma}

Note that above makes no statement about the $\ad$ translation; due to
the change of state space during the translation, equality does not
hold here.

\begin{proof}
  The last two claims follow easily once one notices that for
  $i\in\{ 1, 2\}$ and all $M$, $M\models \rho_i( \phi)$ iff
  $\pi_i( M)\models \phi$.  To show the first claim, let
  $\mcalD_1=( S_1, S^0_1, \omay_1, \omust_1)$ and
  $\mcalD_2=( S_2, S^0_2, \omay_2, \omust_2)$ be DMTS, with \NAA\
  translations $\da( \mcalD_1)=( S_1, S^0_1, \Tran_1)$ and
  $\da( \mcalD_2)=( S_2, S^0_2, \Tran_2)$.  Write
  $\mcalA^\land= \da( \mcalD_1\land \mcalD_2)$ and
  $\mcalA_\land= \da( \mcalD_1)\land \da( \mcalD_2)$; we show that
  $\mcalA^\land= \mcalA_\land$.

  First, remark that $\mcalA^\land$ and $\mcalA_\land$ have precisely the same
  state space $S_1\times S_2$ and initial states $S_1^0\times S_2^0$. We
  now show that they have the same transition constraints. Let
  $\Tran_\land$ (resp.\ $\Tran^\land$) be the transition constraints
  mapping of $\mcalA_\land$ (resp.\ $\mcalA^\land$).  Let $( s_1, s_2)\in
  S_1\times S_2$ and $M\in \Tran_\land( s_1, s_2)$.

  By construction of $\Tran_\land$, there must be $M_1\in \Tran_1( s_1)$
  and $M_2\in \Tran_2( s_2)$ such that $\pi_1( M)= M_1$ and $\pi_2( M)=
  M_2$. We show that $M\in \Tran^\land( s_1, s_2)$.  Let $( a,( t_1,
  t_2))\in M$. Since $\pi_1( M)= M_1$ and $\pi_2( M)= M_2$, we have $(
  a, t_1)\in M_1$ and $( a, t_2)\in M_2$. As a consequence, there are
  transitions $s_1\DMTSmay{a} t_1$ and $s_2\DMTSmay{a} t_2$ in $\mcalD_1$ and
  $\mcalD_2$, respectively.  Thus, by construction, there is a transition
  $( s_1, s_2)\DMTSmay{a}( t_1, t_2)$ in $\mcalD_1\land \mcalD_2$.

  Let $( s_1, s_2)\DMTSmust{} N$ in $\mcalD_1\land \mcalD_2$.  By construction,
  $N$ is such that either (1) there exists $N_1$ such that $s_1\DMTSmust{}
  N_1$ in $\mcalD_1$ and $N=\{( a,( t_1, t_2))\mid( a, t_1)\in N_1,( s_1,
  s_2)\DMTSmay a( t_1, t_2)\}$, or (2) there exists $N_2$ such that
  $s_2\DMTSmust{} N_2$ in $\mcalD_2$ and $N=\{( a,( t_1, t_2))\mid( a, t_2)\in
  N_2,( s_1, s_2)\DMTSmay a( t_1, t_2)\}$. Assume that (1) holds (case (2)
  being symmetric).  Since $M_1\in \Tran_1( s_1)$, there must be $( a,
  t_1)\in N_1\cap M_1$.  As $\pi_1( M)= M_1$, there must be $t_2\in S_2$
  such that $( a,( t_1, t_2))\in M$. As a~consequence, there is $( a,(
  t_1, t_2))\in M \cap N$.

  We have shown that $M\in \Tran^\land( s_1, s_2)$.  Similarly, we can
  show that for all $M\in \Tran^\land( s_1, s_2)$, we also have $M\in
  \Tran_\land( s_1, s_2)$. We can thus conclude that $\Tran^\land =
  \Tran_\land$, hence $\mcalA^\land= \mcalA_\land$. \qed
\end{proof}

\begin{theorem}
  \label{th:condis}
  For all specifications $\mcalS_1$, $\mcalS_2$, $\mcalS_3$,
  \begin{itemize}
  \item $\mcalS_1\lor \mcalS_2\mr \mcalS_3$ iff $\mcalS_1\mr \mcalS_3$ and $\mcalS_2\mr
    \mcalS_3$,
  \item $\mcalS_1\mr \mcalS_2\land \mcalS_3$ iff $\mcalS_1\mr \mcalS_2$ and $\mcalS_1\mr
    \mcalS_3$,
  \item $\sem{ \mcalS_1\lor \mcalS_2}= \sem{ \mcalS_1}\cup \sem{ \mcalS_2}$, and
    $\sem{ \mcalS_1\land \mcalS_2}= \sem{ \mcalS_1}\cap \sem{ \mcalS_2}$.
  \end{itemize}
\end{theorem}

\begin{proof}
  The proof falls into six parts.  We show the second statement of the
  theorem separately for DMTS and for \NAA.  For the other formalisms,
  the statement then follows by structural equivalence.

  \proofpara{$\mcalS_1\lor \mcalS_2\mr \mcalS_3$ iff $\mcalS_1\mr \mcalS_3$ and $\mcalS_2\mr
    \mcalS_3$:}

  Any modal refinement $R\subseteq( S_1\cup S_2)\times S_3$ splits
  into two refinements $R_1\subseteq S_1\times S_3$,
  $R_2\subseteq S_2\times S_3$ and vice versa.

  \proofpara{$\mcalS_1\mr \mcalS_2\land \mcalS_3$ is implied by
    $\mcalS_1\mr \mcalS_2$ and $\mcalS_1\mr \mcalS_3$:}

  We first show this proof for DMTS.  Let
  $\mcalS_i=( S_i, S^0_i, \omay_i, \omust_i)$, for $i= 1, 2, 3$, be DMTS
  and $R_2\subseteq S_1\times S_2$, $R_3\subseteq S_1\times S_3$ modal
  refinements and define
  $R=\{( s_1,( s_2, s_3))\mid( s_1, s_2)\in R_1,( s_1, s_3)\in
  R_3\}\subseteq S_1\times( S_2\times S_3)$.  Then $R$ is initialised.

  Now let $( s_1,( s_2, s_3))\in R$, then $( s_1, s_2)\in R_2$ and $(
  s_1, s_3)\in R_3$.  Assume that $s_1\DMTSmay{ a}_1 t_1$, then by $\mcalS_1\mr
  \mcalS_2$, we have $s_2\DMTSmay{ a}_2 t_2$ with $( t_1, t_2)\in R_2$.
  Similarly, by $\mcalS_1\mr \mcalS_3$, we have $s_3\DMTSmay{ a}_3 t_3$ with $(
  t_1, t_3)\in R_3$.  But then also $( t_1,( t_2, t_3))\in R$, and $(
  s_2, s_3)\DMTSmay{ a}( t_2, t_3)$ by definition.

  Assume that $( s_2, s_3)\DMTSmust{} N$.  Without loss of generality we
  can assume that there is $s_2\DMTSmust{}_2 N_2$ such that
  $N=\{( a,( t_2, t_3))\mid( a, t_2)\in N_2, s_3\DMTSmay{ a}_3 t_3\}$.  By
  $S_1\mr S_2$, we have $s_1\DMTSmust{}_1 N_1$ such that
  \begin{equation}
    \label{eq:conj_proof_N1}
    \forall( a, t_1)\in N_1: \exists( a, t_2)\in N_2:( t_1, t_2)\in
    R_2
  \end{equation}

  Let $( a, t_1)\in N_1$, then also $s_1\DMTSmay{ a}_1 t_1$, so by
  $S_1\mr S_3$, there is $s_3\DMTSmay{ a}_3 t_3$ with
  $( t_1, t_3)\in R_3$.  By~\eqref{eq:conj_proof_N1}, we also have
  $( a, t_2)\in N_2$ such that $( t_1, t_2)\in R_2$, but then
  $( a,( t_2, t_3))\in N$ and $( t_1,( t_2, t_3))\in R$.

  \proofpara{$\mcalS_1\mr \mcalS_2\land \mcalS_3$ implies $\mcalS_1\mr \mcalS_2$ and
    $\mcalS_1\mr \mcalS_3$:}

  Let $R\subseteq S_1\times( S_2\times S_3)$ be a (DMTS) modal
  refinement.  We show that $\mcalS_1\mr \mcalS_2$, the proof of
  $\mcalS_1\mr \mcalS_3$ being entirely analogous.  Define
  $R_2=\{( s_1, s_2)\mid \exists s_3\in S_3:( s_1,( s_2, s_3))\in
  R\}\subseteq S_1\times S_2$, then $R_2$ is initialised.

  Let $( s_1, s_2)\in R_2$, then we must have $s_3\in S_3$ such that $(
  s_1,( s_2, s_3))\in R$.  Assume that $s_1\DMTSmay{ a}_1 t_1$, then also $(
  s_2, s_3)\DMTSmay{ a}( t_2, t_3)$ and $( t_1,( t_2, t_3))\in R$.  By
  construction we have $s_2\DMTSmay{ a}_2 t_2$ and $s_3\DMTSmay{ a}_3 t_3$.
  Moreover, by definition of $R_2$, $( t_1, t_2)\in R_2$.

  Assume that $s_2\DMTSmust{}_2 N_2$, then by construction of
  $\mcalS_2\land \mcalS_3$,
  $( s_2, s_3)\DMTSmust{} N=\{( a,( t_2, t_3))\mid( a, t_2)\in N_2,
  s_3\DMTSmay{ a}_3 t_3\}$.
  By $\mcalS_1\mr \mcalS_2\land \mcalS_3$, we have $s_1\DMTSmust{}_1 N_1$ such that
  $\forall( a, t_1)\in N_1: \exists( a,( t_2, t_3))\in N:( t_1,( t_2,
  t_3))\in R$.

  Let $( a, t_1)\in N_1$, then we have $( a,( t_2, t_3))\in N$ for
  which $( t_1,( t_2, t_3))\in R$.  By construction of $N$, this
  implies that there are $( a, t_2)\in N_2$ and $s_3\DMTSmay{ a}_3 t_3$.
  Moreover, by definition of $R_2$, $( t_1, t_2)\in R$.

  \proofpara{$\mcalA_1\mr \mcalA_2\land \mcalA_3$ is implied by
    $\mcalA_1\mr \mcalA_2$ and $\mcalA_1\mr \mcalA_3$:}

  Now we show the second part of the lemma for \NAA.  Let
  $\mcalA_i=( \mcalS_i, \mcalS^0_i, \Tran_i)$, for $i= 1, 2, 3$, be \NAA and
  $R_2\subseteq S_1\times S_2$, $R_3\subseteq S_1\times S_3$ modal
  refinements and define
  $R=\{( s_1,( s_2, s_3))\mid( s_1, s_2)\in R_1,( s_1, s_3)\in
  R_3\}\subseteq S_1\times( S_2\times S_3)$.  Then $R$ is
  initialised.

  Let $( s_1,( s_2, s_3))\in R$, then $( s_1, s_2)\in R_2$ and
  $( s_1, s_3)\in R_3$.  Let $M_1\in \Tran_1( s_1)$, then we have
  $M_2\in \Tran_2( s_2)$ and $M_3\in \Tran_3( s_3)$ such that the
  pairs $M_1, M_2$ and $M_1, M_3$ verify the conditions~\eqref{eq:aaref}
  in Definition~\ref{def:aaref}.  Let $M=\{( a,( t_2, t_3))\mid( a,
  t_2)\in M_2,( a, t_3)\in M_3\}\subseteq \Sigma\times S$.

  We show that $\pi_2( M)= M_2$ and, similarly, $\pi_3( M)= M_3$:  It
  is clear that $\pi_2( M)\subseteq M_2$, so let $( a, t_2)\in M_2$.
  By the refinement $R_2$, there is $( a, t_1)\in M_1$, so by the
  refinement $R_3$, there is $( a, t_3)\in M_3$, but then $( a,
  t_2)\in \pi_2( M)$.  Using $\pi_2( M)= M_2$ and $\pi_3( M)= M_3$, we
  can now conclude that $M\in \Tran(( s_2, s_3))$.

  Let $( a, t_1)\in M_1$, then we have $( a, t_2)\in M_2$ and $( a,
  t_3)\in M_3$ such that $( t_1, t_2)\in R_2$ and $( t_1, t_3)\in
  R_3$.  But then $( t_1,( t_2, t_3))\in R$ and $( a,( t_2, t_3))\in
  M$.

  Let $( a,( t_2, t_3))\in M$, then $( a, t_2)\in M_2$ and $( a,
  t_3)\in M_3$.  By the refinements $R_2$ and $R_3$, we have $( a,
  t_1)\in M_1$ such that $( t_1, t_2)\in R_2$ and $( t_1, t_3)\in
  R_3$, but then also $( t_1,( t_2, t_3))\in R$.

  \proofpara{$\mcalA_1\mr \mcalA_2\land \mcalA_3$ implies $\mcalA_1\mr \mcalA_2$ and
    $\mcalA_1\mr \mcalA_3$:}

  Let $R\subseteq S_1\times( S_2\times S_3)$ be a (\NAA) modal
  refinement.  We show that $\mcalA_1\mr \mcalA_2$; the proof of
  $\mcalA_1\mr \mcalA_3$ is similar.  Define
  $R_2=\{( s_1, s_2)\mid \exists s_3\in S_3:( s_1,( s_2, s_3))\in
  R\}\subseteq S_1\times S_2$, then $R_2$ is initialised.

  Let $( s_1, s_2)\in R_2$, then there is $t_3\in S_3$ with $( t_1,(
  t_2, t_3))\in R$.  Let $M_1\in \Tran_1( s_1)$, then we have $M\in
  \Tran(( s_2, s_3))$ such that the pair $M_1, M$ satisfies
  conditions~\eqref{eq:aaref}.  Let $M_2= \pi_2( M)\in \Tran_2( s_2)$.

  Let $( a, t_1)\in M_1$, then there is $( a,( t_2, t_3))\in M$ such
  that $( t_1,( t_2, t_3))\in R$; hence $( t_1, t_2)\in R$ and $( a,
  t_2)\in M_2$.

  Let $( a, t_2)\in M_2$, then there is $t_3\in S_3$ such that $( a,(
  t_2, t_3))\in M$.  But then we also have $( a, t_1)\in M_1$ such
  that $( t_1,( t_2, t_3))\in R$, thus $( t_1, t_2)\in R_2$.

  \proofpara{$\sem{ \mcalS_1\lor \mcalS_2}= \sem{ \mcalS_1}\cup \sem{ \mcalS_2}$
    and $\sem{ \mcalS_1\land \mcalS_2}= \sem{ \mcalS_1}\cap \sem{ \mcalS_2}$:}

  $\sem{ \mcalS_1\land \mcalS_2}= \sem{ \mcalS_1}\cap \sem{ \mcalS_2}$ is clear
  from what we just proved: for all implementations $\mcalI$,
  $\mcalI\mr \mcalS_1\land \mcalS_2$ iff $\mcalI\mr \mcalS_1$ and $\mcalI\mr \mcalS_2$.
  For the other part, it is clear by construction that for any
  implementation $\mcalI$, any witness $R$ for $\mcalI\mr \mcalS_1$ is also a
  witness for $\mcalI\mr \mcalS_1\lor \mcalS_2$, and similarly for $\mcalS_2$,
  hence
  $\sem{ \mcalS_1}\cup \sem{ \mcalS_2}\subseteq \sem{ \mcalS_1\lor \mcalS_2}$.

  To show the other inclusion, we note that an initialised refinement
  $R$ witnessing $\mcalI\mr \mcalS_1\lor \mcalS_2$ must relate the initial state
  of $\mcalI$ either to an initial state of $\mcalS_1$ or to an initial state
  of $\mcalS_2$.  In the first case, and by disjointness, $R$~witnesses
  $\mcalI\mr \mcalS_1$, in the second, $\mcalI\mr \mcalS_2$.  Note how it is
  essential here that \emph{implementations} have but \emph{one} initial
  state; this part of the proof would break down if we were to allow
  several initial states for implementations.  \qed
\end{proof}

\begin{corollary}
  \label{co:distlat}
  With operations $\lor$ and $\land$, each of our four classes of
  specifications forms a bounded distributive lattice up to $\mreq$.
\end{corollary}

\begin{proof}
  The bottom elements (up to $\mreq$) in the lattices are given by
  specifications with empty initial state sets.  The top elements are
  the DMTS $(\{ s^0\},\{ s^0\},\{( s^0, a, s^0)\mid a\in \Sigma\},
  \emptyset)$ and its respective translations.  The other lattice
  properties follow from Theorem~\ref{th:condis}.

  We miss to verify distributivity.  Let $\mcalA_i=( S_i, S_i^0, \Tran_i)$,
  for $i= 1, 2, 3$, be \NAA.  The set of states of both $\mcalA_1\land(
  \mcalA_2\lor \mcalA_3)$ and $( \mcalA_1\land \mcalA_2)\lor( \mcalA_1\land \mcalA_3)$ is
  $S_1\times( S_2\cup S_3)= S_1\times S_2\cup S_1\times S_3$, and one
  easily sees that the identity relation is a two-sided modal
  refinement.  Things are similar for the other distributive law. \qed
\end{proof}

\subsection{Composition}

The composition operator for a specification theory is to mimic, at
specification level, the parallel composition of implementations.  That
is to say, if $\|$ is a composition operator for implementations (LTS),
then the goal is to extend~$\|$ to specifications such that for all
specifications $\mcalS_1$, $\mcalS_2$,
\begin{equation}
  \label{eq:compcomplete}
  \sem{ \mcalS_1\| \mcalS_2}= \big\{ \mcalI_1\| \mcalI_2\mid \mcalI_1\in
  \sem{ \mcalS_1}, \mcalI_2\in \sem{ \mcalS_2}\big\}.
\end{equation}

For simplicity, we use CSP-style synchronisation for parallel
composition of LTS, however, our results readily carry over to other
types of composition.  Analogously to the situation for
MTS~\cite{DBLP:journals/tcs/BenesKLS09}, we have the following negative
result:

\begin{figure}
  \centering
  \begin{tikzpicture}[->, >=stealth', font=\footnotesize,
    state/.style={ shape= circle, draw, initial text=, inner
      sep= .5mm, minimum size= 2mm}, scale= 1, on grid]
    \node[state,label=left:{$s$},initial above] (s) {};
    \node[state,right=3 of s,label=left:{$t$},initial above] (t) {};
    \node[state,below left=1 and 0.5 of s] (s1) {};
    \node[state,below right=1 and 0.5 of s] (s2) {};
    \node[state,below left=1 and 0.5 of t] (t1) {};
    \node[state,below right=1 and 0.5 of t] (t2) {};

    \path
    (s) edge[must,loop right] node {$a$} (s)
    (t) edge[must,loop right] node {$a$} (t)
    (s) edge[must,swap] 	node[pos=0.75] {$b$} (s1)
    (s) edge[may]  		node[pos=0.75] {$c$} (s2)
    (t) edge[may,swap] 	node[pos=0.75] {$b$} (t1)
    (t) edge[must]  	node[pos=0.75] {$c$} (t2)
    ;
  \end{tikzpicture}
  \caption{%
    \label{fi:composition_cex}
    DMTS $\mcalS$ and $\mcalT$ whose composition cannot be captured
    precisely}
\end{figure}
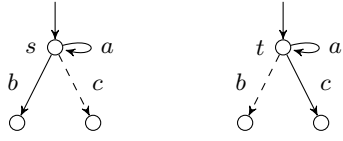

\begin{theorem}
  \label{th:composition}
  There is no operator $\|$, for any of our specification formalisms,
  which satisfies~\eqref{eq:compcomplete}.
\end{theorem}

\begin{proof}
  We show that there exist DMTS $\mcalS$ and $\mcalT$
  such that there is no DMTS $\mcalD$ with $\sem{\mcalD}= \sem \mcalS\| \sem
  \mcalT\coloneqq \{ \mcalI\| \mcalJ\mid \mcalI\in \sem{\mcalS}, \mcalJ\in \sem{ \mcalT}\}$.
  They are given in Figure~\ref{fi:composition_cex}; $\mcalS$ has initial
  state~$s$, while $\mcalT$ has initial state~$t$.  Note that in fact,
  $\mcalS$ and $\mcalT$ are MTS, \ie~no disjunctive must transitions are used.

  We make the following observations about implementations of $\mcalS$
  and $\mcalT$.  They always admit one or more infinite runs labeled
  with $a$'s with one-step $b$ or $c$ branches. Moreover, all infinite
  runs in these implementations are of this form. To each infinite
  $a$-run of an implementation we assign its signature, that is a~word
  over $2^{\{b,c\}}$ that describes which one-step branches are
  available at each step. This means that every implementation of
  $\mcalS$ has runs with signatures from
  $\{\{b\},\{b,c\}\}^\omega$, while every implementation of
  $\mcalT$ has runs with signatures from
  $\{\{c\},\{b,c\}\}^\omega$.

  We now construct an implementation state space as illustrated
  in~Figure~\ref{fi:composition_cex_impl}.  Consider the implementations
  $\mcalI_1, \mcalI_2, \dotsc$ that share the same state space and have the
  initial state $i_1, i_2, \dotsc$, respectively.  The implementation
  $\mcalI_n$ has only one $a$-run with the signature $\emptyset^n \{b,c\}
  \emptyset^\omega$.  Note that $\mcalI_n$ is the composition of an
  implementation of $\mcalS$ that has only one $a$-run with the signature
  $\{b\}^n \{b,c\} \{b\}^\omega$ and an implementation of $\mcalT$ that has
  only one $a$-run with the signature $\{c\}^n \{b,c\} \{c\}^\omega$.

  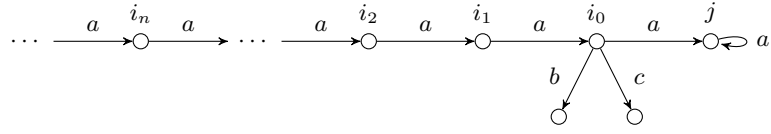
\begin{figure}[tbp]
    \centering
    \begin{tikzpicture}[->, >=stealth', font=\footnotesize,
      state/.style={ shape= circle, draw, initial text=, inner
        sep= .5mm, minimum size= 2mm}, scale= 1, on grid]
      \node[state,label=above:{$j$}] (j) {};
      \node[state,label=above:{$i_0$},left=1.5 of j] (i0) {};
      \node[state,below left=1 and 0.5 of i0] (b) {};
      \node[state,below right=1 and 0.5 of i0] (c) {};
      \node[state,label=above:{$i_1$},left=1.5 of i0] (i1) {};
      \node[state,label=above:{$i_2$},left=1.5 of i1] (i2) {};
      \node[left=1.5 of i2] (dots) {$\cdots$};
      \node[state,label=above:{$i_n$},left=1.5 of dots] (in) {};
      \node[left=1.5 of in] (dots2) {$\cdots$};

      \path
      (j) edge[must,loop right] node {$a$} (j)
      (i0) edge[must,swap] node[pos=0.75] {$b$} (b)
      (i0) edge[must] node[pos=0.75] {$c$} (c)
      (i0) edge[must] node {$a$} (j)
      (i1) edge[must] node {$a$} (i0)
      (i2) edge[must] node {$a$} (i1)
      (dots) edge[must] node {$a$} (i2)
      (in) edge[must] node {$a$} (dots)
      (dots2) edge[must] node {$a$} (in)
      ;
    \end{tikzpicture}
    \caption{%
      \label{fi:composition_cex_impl}
      Implementation state space in the proof
      of~Theorem~\ref{th:composition}}
  \end{figure}

  Assume now that there exists a~DMTS $\mcalD$ with $\sem{\mcalD}= \sem \mcalS\|
  \sem \mcalT$.  As all $\mcalI_n$ belong to $\sem{\mcalD}$ and there is only
  a~finite number of initial states of $\mcalD$, there has to be at least
  one initial state of $\mcalD$, say $d$, such that there exists a~modal
  refinement $R$ containing both $(i_k,d)$ and $(i_l,d)$ for some
  numbers $k < l$. Let $\mcalD_d$ be created from $\mcalD$ by changing the set
  of initial states to the singleton~$\{d\}$. As both $\mcalI_k\mr \mcalD_d$
  and $\mcalI_l\mr \mcalD_d$ and $\mcalD_d$ has only one initial state, we know
  by Lemma~\ref{le:dmts1sumclosed} that also $\mcalI_{kl} = \mcalI_k + \mcalI_l\mr
  \mcalD_d$. The unfolding of this implementation is illustrated in
  Figure~\ref{fi:composition_cex_impl_sum}.

  \begin{figure}
    \centering
    \begin{tikzpicture}[->, >=stealth', font=\footnotesize,
      state/.style={ shape= circle, draw, initial text=, inner
        sep= .5mm, minimum size= 2mm}, scale= 1, on grid]
      \node[state,label=above:{$\mcalI_{kl}$},initial] (I) {};
      \node[state,above right=1 and 1.5 of I] (ik1) {};
      \node[state,below right=1 and 1.5 of I] (il1) {};
      \node[right=1.5 of ik1] (ik2) {$\cdots$};
      \node[right=1.5 of il1] (il2) {$\cdots$};
      \node[state,right=1.5 of ik2] (ik3) {};
      \node[state,right=1.5 of il2] (il3) {};
      \node[right=1.5 of ik3] (ik4) {$\cdots$};
      \node[right=1.5 of il3] (il4) {$\cdots$};
      \node[state,right=1.5 of ik4] (ik5) {};
      \node[state,right=1.5 of il4] (il5) {};
      \node[right=1.5 of ik5] (ik6) {$\cdots$};
      \node[right=1.5 of il5] (il6) {$\cdots$};

      \node[state,above left=1 and 0.5 of ik3] (bk) {};
      \node[state,above right=1 and 0.5 of ik3] (ck) {};
      \node[state,below left=1 and 0.5 of il5] (bl) {};
      \node[state,below right=1 and 0.5 of il5] (cl) {};

      \path
      (I) edge[must,swap] node {$a$} (ik1)
      (I) edge[must] node {$a$} (il1)
      (ik1) edge[must,swap] node {$a$} (ik2)
      (il1) edge[must] node {$a$} (il2)
      (ik2) edge[must,swap] node {$a$} (ik3)
      (il2) edge[must] node {$a$} (il3)
      (ik3) edge[must,swap] node {$a$} (ik4)
      (il3) edge[must] node {$a$} (il4)
      (ik4) edge[must,swap] node {$a$} (ik5)
      (il4) edge[must] node {$a$} (il5)
      (ik5) edge[must,swap] node {$a$} (ik6)
      (il5) edge[must] node {$a$} (il6)

      (ik3) edge[must] node[pos=0.75] {$b$} (bk)
      (ik3) edge[must,swap] node[pos=0.75] {$c$} (ck)
      (il5) edge[must,swap] node[pos=0.75] {$b$} (bl)
      (il5) edge[must] node[pos=0.75] {$c$} (cl)

      ;
    \end{tikzpicture}
    \caption{%
      \label{fi:composition_cex_impl_sum}
      The nondeterministic sum of $\mcalI_k$ and $\mcalI_l$, unfolded}
  \end{figure}
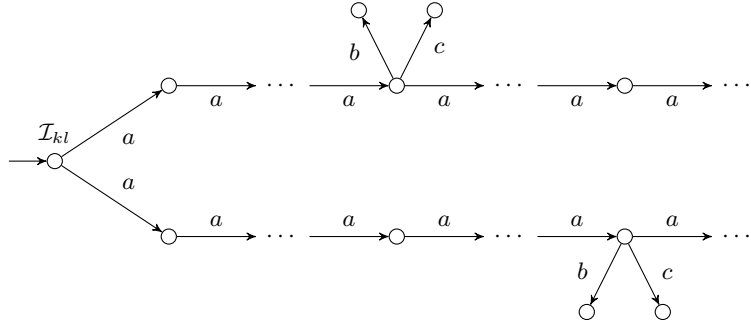

  We now argue that $\mcalI_{kl}\notin \sem \mcalS\| \sem \mcalT$.  We actually show
  that it cannot even be bisimilar to any $\mcalI \| \mcalJ$ with
  $\mcalI\in\sem{\mcalS}$ and $\mcalJ\in\sem{\mcalT}$.  Let us assume that there
  exist such $\mcalI$ and $\mcalJ$. We make the following observations:
  \begin{itemize}
  \item $\mcalI$ has to contain at least one $a$-run with signature
    $\{b\}^k\{b,c\}\{b\}^\omega$. Otherwise, it would be impossible to
    create the $\mcalI_k$ part of $\mcalI_{kl}$.
  \item $\mcalJ$ has to contain at least one $a$-run with signature
    $\{c\}^l\{b,c\}\{c\}^\omega$. Otherwise, it would be impossible to
    create the $\mcalI_l$ part of $\mcalI_{kl}$.
  \end{itemize}
  However, these observations mean that $\mcalI\| \mcalJ$ contains at least
  one $a$-run with signature $\emptyset^k\{ c\}\emptyset^{ l- k- 1}\{
  b\}\emptyset^\omega$. It is thus not bisimilar to $\mcalI_{kl}$. \qed
\end{proof}

Given that we cannot have~\eqref{eq:compcomplete}, the revised goal is
to have a \emph{sound} composition operator for which the right-to-left
inclusion holds in~\eqref{eq:compcomplete}.  For \NAA\ $\mcalA_1=( S_1,
S^0_1, \Tran_1)$, $\mcalA_2=( S_2, S^0_2, \Tran_2)$, we define $\mcalA_1\|
\mcalA_2=( S, S^0, \Tran)$ with $S= S_1\times S_2$, $S^0= S^0_1\times
S^0_2$, and for all $( s_1, s_2)\in S$, $\Tran( s_1, s_2)=\{ M_1\|
M_2\mid M_1\in \Tran_1( s_1), M_2\in \Tran_2( s_2)\}$, where $M_1\|
M_2=\{( a,( t_1, t_2))\mid( a, t_1)\in M_1,( a, t_2)\in M_2\}$.
Composition for DMTS is defined using the translations to and from \NAA;
note that this may incur an exponential blow-up.

\begin{lemma}
  \label{le:bfs||prop}
  Up to $\mreq$, the operator $\|$ on \NAA\ is associative and commutative,
  distributes over $\lor$, and has unit $\mathsf{U}$, where $\mathsf{U}$
  is the LTS $(\{ s\}, s, \omust)$ with $s\DMTSmust a s$ for all $a\in
  \Sigma$.
\end{lemma}

\begin{proof}
  Associativity and commutativity are clear.  To show distributivity
  over $\lor$, let $\mcalA_i=( S_i, S_i^0, \Tran_i)$, for $i= 1, 2, 3$, be
  \NAA. We prove that $\mcalA_1\|( \mcalA_1\lor \mcalA_3)\mreq \mcalA_1\| \mcalA_2\lor
  \mcalA_1\| \mcalA_3$; right-distributivity will follow by commutativity.
  The state spaces of both sides are $S_1\times S_2\cup S_1\times S_3$,
  and it is easily verified that the identity relation is a two-sided
  modal refinement.

  For the claim that $\mcalA\| \mathsf{U}\mreq \mcalA$ for all \NAA\ $\mcalA=( S,
  S^0, \Tran)$, let $u$ be the unique state of $\mathsf{U}$ and define
  $R=\{(( s, u), s)\mid s\in S\}\subseteq S\times \mathsf{U} \times S$.
  We show that $R$ is a two-sided modal refinement.  Let $(( s, u),
  s)\in R$ and $M\in \Tran( s, u)$, then there must be $M_1\in \Tran(
  s)$ for which $M= M_1\|( \Sigma\times\{ u\})$.  Thus $M_1=\{( a,
  t)\mid( a,( t, u))\in M\}$.  Then any element of $M$ has a
  corresponding one in $M_1$, and vice versa, and their states are
  related by $R$.  For the other direction, let $M_1\in \Tran( s)$, then
  $M= M_1\|( \Sigma\times\{ u\})=\{( a,( t, u))\mid( a, t)\in M_1\}\in
  \Tran( s, u)$, and the same argument applies. \qed
\end{proof}

The next theorem is one of \emph{independent implementability}, as it
ensures that a~composition of refinements is a refinement of
compositions:

\begin{theorem}
  \label{th:indimp}
  For all specifications $\mcalS_1$, $\mcalS_2$, $\mcalS_3$, $\mcalS_4$, $\mcalS_1\mr
  \mcalS_3$ and $\mcalS_2\mr \mcalS_4$ imply $\mcalS_1\| \mcalS_2\mr \mcalS_3\| \mcalS_4$.
\end{theorem}

\begin{proof}
  Let $\mcalS_1\mr \mcalS_3$ and $\mcalS_2\mr \mcalS_4$, then $\mcalS_1\lor \mcalS_3\mreq
  \mcalS_3$ and $\mcalS_2\lor \mcalS_4\mreq \mcalS_4$.  By distributivity,
  \begin{align*}
    \mcalS_3\| \mcalS_4 &\mreq ( \mcalS_1\lor \mcalS_3)\|( \mcalS_2\lor \mcalS_4) \\
    &\mreq \mcalS_1\| \mcalS_2\lor \mcalS_1\| \mcalS_4\lor \mcalS_3\| \mcalS_2\lor
      \mcalS_3\| \mcalS_4\,,
  \end{align*}
  thus
  \begin{equation*}
    \mcalS_1\| \mcalS_2\lor \mcalS_1\| \mcalS_4\lor \mcalS_3\| \mcalS_2\mr \mcalS_3\|
    \mcalS_4\,.
  \end{equation*}
  But
  \begin{equation*}
    \mcalS_1\| \mcalS_2\mr \mcalS_1\| \mcalS_2\lor \mcalS_1\| \mcalS_4\lor \mcalS_3\|
    \mcalS_2\,,
  \end{equation*}
  finishing the argument. \qed
\end{proof}

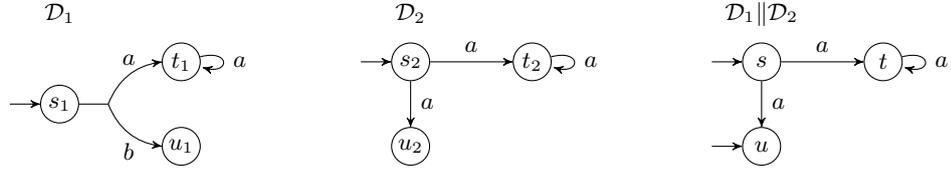
\begin{figure}
  \centering
  \begin{tikzpicture}[->, >=stealth', font=\footnotesize,
    state/.style={shape=circle, draw, initial text=,inner
      sep=.3mm,minimum size=5mm}, scale=.8]
    \begin{scope}
      \node at (0,1.5) {$\mcalD_1$};
      \node[state,initial] (X) at (0,0) {$s_1$};
      \node[state] (Z) at (2,.7) {$t_1$};
      \node[state] (Y) at (2,-.7) {$u_1$};
      \coordinate (XX) at (.8,0);
      \path[-] (X) edge (XX);
      \path (XX) edge[bend left] node[above]{$a$} (Z);
      \path (XX) edge[bend right] node[below]{$b$} (Y);
      \path (Z) edge[loop right] node[right]{$a$} (Z);
    \end{scope}
    \begin{scope}[xshift=15em]
      \node at (0,1.5) {$\mcalD_2$};
      \node[state,initial] (X) at (0,.7) {$s_2$};
      \node[state] (Z) at (2,.7) {$t_2$};
      \node[state] (Y) at (0,-.7) {$u_2$};
      \path (X) edge node[above]{$a$} (Z);
      \path (X) edge node[right]{$a$} (Y);
      \path (Z) edge[loop right] node[right]{$a$} (Z);
    \end{scope}
    \begin{scope}[xshift=30em]
      \node at (0,1.5) {$\mcalD_1\| \mcalD_2$};
      \node[state, initial] (X) at (0,.7) {$s$};
      \node[state] (Z) at (2,.7) {$t$};
      \node[state, initial] (Y) at (0,-.7) {$u$};
      \path (X) edge node[above]{$a$} (Z);
      \path (X) edge node[right]{$a$} (Y);
      \path (Z) edge[loop right] node[right]{$a$} (Z);
    \end{scope}
  \end{tikzpicture}
  \caption{%
    \label{fi:comp}
    Two DMTS and the reachable parts of the DMTS translation of their
    composition.  Here, $s=\{( a,( t_1, t_2)),( a,( t_1, u_2))\}$,
    $t=\{( a,( t_1, t_2))\}$ and $u= \emptyset$}
\end{figure}

\begin{example}
  An example of composition is shown in Figure~\ref{fi:comp}.  Here the
  DMTS translation of $\mcalD_1\| \mcalD_2$ has two initial states; it can be
  shown that no DMTS with a single initial state is thoroughly
  equivalent.
\end{example}

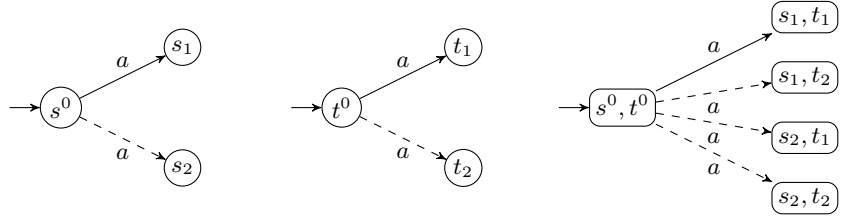
\begin{figure}
  \centering
  \begin{tikzpicture}[->, >=stealth', font=\footnotesize,
    state/.style={ shape= circle, draw, initial text=, inner
      sep= .5mm, minimum size= 2mm}, scale= .8, on grid]
    \begin{scope}
      \node[state,initial] (s) at (0,0) {$s^0$};
      \node[state] (s1) at (2,1) {$s_1$};
      \node[state] (s2) at (2,-1) {$s_2$};
      \path (s) edge[must] node[above] {$a$} (s1);
      \path (s) edge[may] node[below] {$a$} (s2);
    \end{scope}
    \begin{scope}[xshift=12em]
      \node[state,initial] (s) at (0,0) {$t^0$};
      \node[state] (s1) at (2,1) {$t_1$};
      \node[state] (s2) at (2,-1) {$t_2$};
      \path (s) edge[must] node[above] {$a$} (s1);
      \path (s) edge[may] node[below] {$a$} (s2);
    \end{scope}
    \begin{scope}[xshift=24em, state/.style={ shape= rectangle, rounded
        corners, draw, initial text=, inner sep= .8mm, minimum size=
        2mm}]
      \node[state,initial] (st) at (0,0) {$s^0, t^0$};
      \node[state] (s1t1) at (3,1.5) {$s_1, t_1$};
      \node[state] (s1t2) at (3,.5) {$s_1, t_2$};
      \node[state] (s2t1) at (3,-.5) {$s_2, t_1$};
      \node[state] (s2t2) at (3,-1.5) {$s_2, t_2$};
      \path (st) edge[must] node[above] {$a$} (s1t1);
      \path (st) edge[may] node[below] {$a$} (s1t2);
      \path (st) edge[may] node[below] {$a$} (s2t1);
      \path (st) edge[may] node[below] {$a$} (s2t2);
    \end{scope}
  \end{tikzpicture}
  \caption{%
    \label{fi:mtscompvsaa}
    Two MTS and their MTS composition according
    to~\cite{DBLP:conf/avmfss/Larsen89}}
\end{figure}

Remark that \NAA\ composition is more precise than the composition for
MTS introduced in~\cite{DBLP:conf/avmfss/Larsen89}.  The MTS composition
is given by the following rules: $(s_1,s_2) \DMTSmay{a} (t_1,t_2)$ whenever
$s_1 \DMTSmay{a} t_1$ and $s_2 \DMTSmay{a} t_2$, $(s_1,s_2) \DMTSmust{a} (t_1,t_2)$
whenever $s_1 \DMTSmust{a} t_1$ and $s_2 \DMTSmust{a} t_2$.  The difference
between the two compositions is illustrated in
Figure~\ref{fi:mtscompvsaa}. The figure shows two MTS and their MTS
composition; for their \NAA\ composition,
\begin{multline}
  \label{eq:mtscompvsaa}
  \Tran( s^0, t^0)=\big\{\{( a,( s_1, t_1))\}, \\
  \{( a,( s_1, t_1)),( a,( s_1, t_2))\},\{( a,( s_1, t_1)),( a,( s_2,
  t_1))\}, \\
  \{( a,( s_1, t_1)),( a,( s_1, t_2)),( a,( s_2, t_1)),( a,( s_2,
  t_2))\}\big\}\,.
\end{multline}
The \NAA\ translation of their MTS composition has eight transition
constraints instead of four; note how the four constraints
in~\eqref{eq:mtscompvsaa} precisely correspond to the four
implementation choices for $s^0$ and $t^0$.

It can easily be shown that generally, \NAA\ composition is a refinement
of MTS composition.  The following lemma shows a~stronger relationship,
namely that the MTS composition is a~conservative approximation of the
\NAA composition.

\begin{lemma}\label{lem:mts-naa-comp}
  Let $\mcalM_1$, $\mcalM_2$, $\mcalM_3$ be MTS and let $\pM$ and $\pA$ be the
  MTS and \NAA composition, respectively. It holds that
  $\mcalM_1 \pM \mcalM_2 \mr \mcalM_3$ iff $\mcalM_1 \pA \mcalM_2 \mr \mcalM_3$.
\end{lemma}

\begin{proof}
  Let $\mcalM_i = (S^0_i,\{s^0_i\},\omay_i,\omust_i)$ be MTS for
  $i = 1,2,3$.  In the following, we use the notation $s_1\pA s_2$ to
  denote the states $(s_1,s_2)$ of $\mcalM_1\pA\mcalM_2$ and similarly for
  $\pM$.

  For an~MTS translated into \NAA, the $\Tran$ sets have a~special
  structure, namely, for all states $s$, $\Tran(s)$ always has
  a~maximal element $\{ (a,t) \mid s \DMTSmay{a} t \}$ and a~minimal
  element $\{ (a,t) \mid s \DMTSmust{a} t \}$ (with respect to~set
  inclusion; \cf~Lemma~\ref{le:dmtstobfsspecial} for the similar
  property for DMTS).  Furthermore, we note that
  $\Tran(s_1\pA s_2) \subseteq \Tran(s_1\pM s_2)$ and, moreover,
  $\Tran(s_1\pA s_2)$ also has a~minimal and a~maximal element and
  these elements correspond to the minimal and maximal element of
  $\Tran(s_1\pM s_2)$.

  The fact that $\mcalM_1\pA\mcalM_2 \mr \mcalM_1\pM\mcalM_2$ follows from the
  observation that $\Tran(s_1\pA s_2) \subseteq \Tran(s_1\pM s_2)$ for
  all $( s_1, s_2)$.  This proves the `only-if' part of the
  lemma.

  To prove the `if' part of the lemma, we let
  $R = \{ (s_1\pM s_2, s_3) \mid s_1\pA s_2 \mr s_3 \}$ and show that
  it is a~modal refinement relation witnessing
  $\mcalM_1\pM\mcalM_2\mr\mcalM_3$.  Let $(s_1\pM s_2, s_3) \in R$.
  \begin{itemize}
  \item Let $s_1\pM s_2 \DMTSmay{a} t_1\pM t_2$. Then $(a,t_1\pM t_2)$
    belongs to the maximal element of $\Tran(s_1\pM s_2)$, which is
    also in $\Tran(s_1\pA s_2)$. Due to $s_1\pA s_2 \mr s_3$ we have
    some $N \in \Tran(s_3)$ with $(a,t_3) \in N$ such that
    $t_1\pA t_2 \mr t_3$. Thus $s_3 \DMTSmay{a} t_3$ and
    $(t_1\pM t_2, t_3) \in R$.
  \item Let $s_3 \DMTSmust{a} t_3$. Then all elements of $\Tran(s_3)$
    contain $(a,t_3)$. If we now chose the minimal element
    $M \in \Tran(s_1\pA s_2)$ then it has to contain $(a,t_1\pA t_2)$
    such that $(t_1\pA t_2\mr t_3)$.  This means that
    $s_1\pM s_2 \DMTSmust{a} t_1\pM t_2$ and
    $(t_1\pM t_2, t_3) \in R$. \qed
  \end{itemize}
\end{proof}

\subsection{Quotient}

The quotient operator for a specification theory is used to synthesise
specifications for components of a composition.  Hence it is to have the
property, for all specifications $\mcalS$, $\mcalS_1$ and all implementations
$\mcalI_1$, $\mcalI_2$, that
\begin{equation}
  \label{eq:quotient}
  \mcalI_1\in \sem{ \mcalS_1} \text{ and } \mcalI_2\in \sem{ \mcalS/ \mcalS_1}
  \text{ imply } \mcalI_1\| \mcalI_2\in \sem \mcalS.
\end{equation}
Furthermore, $\mcalS/ \mcalS_1$ is to be as permissive as possible.

\subsubsection{Quotient for MTS}

Before we describe the general construction of the quotient,
we start with a~simpler construction that works for the important
special case of MTS.
However, MTS are not closed under quotient,
\cf~\cite[Thm.~5.5]{DBLP:conf/concur/Larsen90}; we show that the
quotient of two MTS will generally be a DMTS.

Recall that MTS have only one initial state and all their must
transitions are singletons.
Let $\mcalM_1 = ( S_1, s_1^0, \omay_1, \omust_1)$ and $\mcalM_2 = ( S_2,
s_2^0, \omay_2, \omust_2)$ be MTS.  We define $\mcalM_1/ \mcalM_2=( S, s^0,
\omay, \omust)$ with $S=2^{ S_1\times S_2}$, $s^0=\{( s_1^0, s_2^0)\}$,
and the transition relations given as follows.

For $s = \{(s^1_1, s^1_2), \ldots, (s^n_1,s^n_2)\} \in S$ we say that
$a \in \Sigma$ is \emph{permissible from $s$} if for all
$i = 1, \ldots, n$ either $s^i_1 \DMTSmay{a}$ or $s^i_2 \notmay{a}$.

For $a$ permissible from $s$ and $i \in \{1, \ldots, n\}$, let
$\{t^{i,1}_2, \ldots, t^{i,m_i}_2\} = \{ t_2 \in S_2 \mid s^i_2
\DMTSmay{a} t_2\}$
be an enumeration of the possible states in $S_2$ after an
$a$-transition from $s^i_2$.  We then define the set of \emph{possible
  transitions} from $s$ under $a$ as
$\postra[a]{s}=\big\{\{(t_1^{ i, j}, t_2^{ i, j}) \mid i= 1,\dots, n,
j= 1,\dots, m_i\}\bigmid \forall i, j: s_1^i \DMTSmay{a} t_1^{i,j}\big\}$.

The transitions of $s$ are now given as follows:
for every $a$ permissible from $s$ and every $t \in \postra[a]{s}$,
let $s \DMTSmay{a} t$. Furthermore,
for every $s^i_1 \DMTSmust{a} t_1$ let
$s \DMTSmust{} \{ (a,M) \in \{a\} \times \postra[a]{s} \mid
\exists t_2 : (t_1,t_2) \in M, s^i_2 \DMTSmust{a} t_2 \}$.

Note that as a special case we obtain
$\emptyset \DMTSmay{a} \emptyset$ for all $a \in \Sigma$
and there are no must transitions from $\emptyset$.

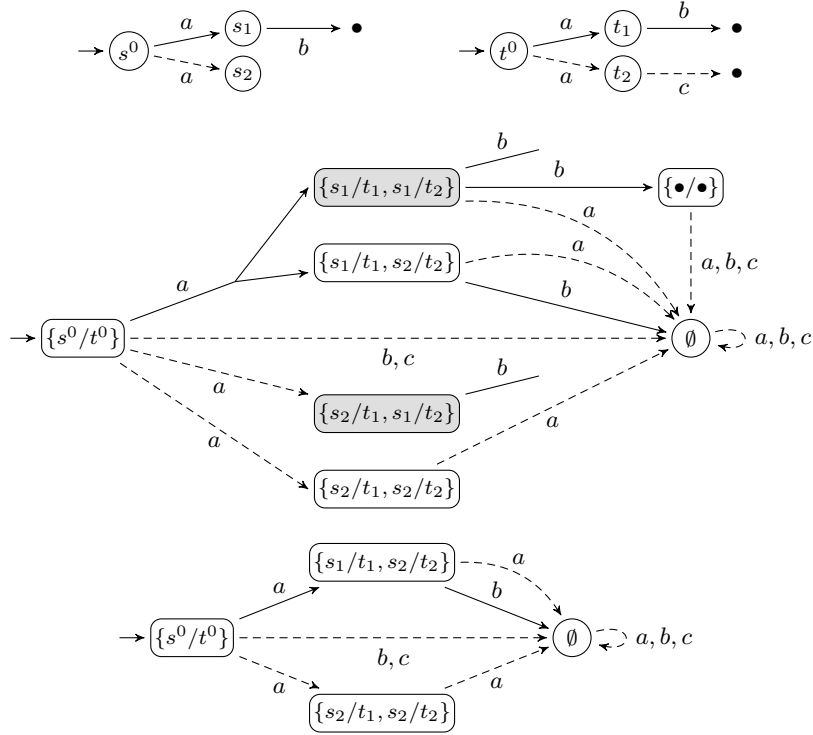
\begin{figure}[t]
  \centering
  \begin{tikzpicture}[x=1.5cm,y=0.6cm,font=\footnotesize,
    ->,>=stealth',
    state/.style={shape=circle,draw,font=\scriptsize,inner sep=.5mm,outer
      sep=0.8mm, minimum size=0.3cm,initial text=,initial
      distance=2ex}]
    \begin{scope}
      \node[state,initial] (s) at (0,0) {$t^0$};
      \node[state] (s1) at (1,0.5) {$t_1$};
      \path (s)	edge [->] node[above]{$a$}	(s1);
      \node[state] (s2) at (1,-0.5) {$t_2$};
      \path (s)	edge [->,densely dashed] node[below]{$a$}	(s2);
      \node (end) at (2,0.5) {$\bullet$};
      \path (s1)	edge [->] node[above]{$b$}	(end);
      \node (end) at (2,-0.5) {$\bullet$};
      \path (s2)	edge [->,densely dashed] node[below]{$c$}	(end);
    \end{scope}
    \begin{scope}[xshift=-5cm]
      \node[state,initial] (s) at (0,0) {$s^0$};
      \node[state] (s1) at (1,0.5) {$s_1$};
      \path (s)	edge [->] node[above]{$a$}	(s1);
      \node[state] (s2) at (1,-0.5) {$s_2$};
      \path (s)	edge [->,densely dashed] node[below]{$a$}	(s2);
      \node (end) at (2,0.5) {$\bullet$};
      \path (s1) edge [->] node[below]{$b$}	(end);
    \end{scope}
  \end{tikzpicture}\medskip

  \begin{tikzpicture}[x=4cm,y=1cm,font=\footnotesize,
    ->,>=stealth',
    state/.style={shape=circle,draw,font=\scriptsize,inner sep=.5mm,outer
      sep=0.8mm, minimum size=0.5cm,initial text=,initial
      distance=2ex,rectangle,rounded corners}]
    \node[state,initial] (s) at (0,0) {$\{s^0/ t^0\}$};
    \node[state,fill=gray!25] (s11) at (1,2) {$\{s_1/ t_1,s_1/ t_2\}$};
    \node[state] (s1) at (1,1) {$\{s_1/ t_1,s_2/ t_2\}$};
    \coordinate (amust) at (0.5,0.75);
    \path (s) edge[-] node[above]{$a$} (amust);
    \path (amust) edge[->] (s11.west);
    \path (amust) edge[->] (s1);
    \node[state] (end2) at (2,2) {$\{\bullet/\bullet\}$};
    \path (s11) edge[->] node[above] {$b$} (end2);
    \coordinate (bmust) at (1.5,2.5);
    \path (s11) edge[-] node[above]{$b$} (bmust);
    \node[state,fill=gray!25] (s21) at (1,-1) {$\{s_2/ t_1,s_1/ t_2\}$};
    \path (s)	edge [->,densely dashed] node[below]{$a$}	(s21);
    \coordinate (bmust2) at (1.5,-0.5);
    \path (s21) edge[-] node[above]{$b$} (bmust2);
    \node[state] (s2) at (1,-2) {$\{s_2/ t_1,s_2/ t_2\}$};
    \path (s)	edge [->,densely dashed] node[below]{$a$}	(s2.west);
    \node[state,shape=circle] (end) at (2,0) {$\emptyset$};
    \path (s11.350)	edge [->,bend left,densely dashed] node[above]{$a$}	(end);
    \path (s1)	edge [->] node[above]{$b$}	(end);
    \path (s1.east)	edge [->,bend left,densely dashed]
    node[above]{$a$}	(end);
    \path (s2)	edge [->,densely dashed] node[below]{$a$}	(end);
    \path (s)	edge [->,densely dashed] node[below]{$b,c$}	(end);
    \path (end)	edge [loop right,->,densely dashed]
    node[right]{$a, b, c$} (end);
    \path (end2) edge[->,densely dashed] node[right]{$a,b,c$} (end);
  \end{tikzpicture}\bigskip

   \begin{tikzpicture}[x=2.5cm,y=1cm,font=\footnotesize,
   ->,>=stealth',
   state/.style={shape=circle,draw,font=\scriptsize,inner sep=.5mm,outer
   	sep=0.8mm, minimum size=0.5cm,initial text=,initial
   	distance=2ex,rectangle,rounded corners}]
   \node[state,initial] (s) at (0,0) {$\{s^0/ t^0\}$};
   \node[state] (s1) at (1,1) {$\{s_1/ t_1,s_2/ t_2\}$};
   \path (s)	edge [->] node[above]{$a$}	(s1);
   \node[state] (s2) at (1,-1) {$\{s_2/ t_1,s_2/ t_2\}$};
   \path (s)	edge [->,densely dashed] node[below]{$a$}	(s2);
   \node[state,shape=circle] (end) at (2,0) {$\emptyset$};
   \path (s1)	edge [->] node[above]{$b$}	(end);
   \path (s1.east)	edge [->,bend left,densely dashed]
   node[above]{$a$}	(end);
   \path (s2)	edge [->,densely dashed] node[below]{$a$}	(end);
   \path (s)	edge [->,densely dashed] node[below]{$b,c$}	(end);
   \path (end)	edge [loop right,->,densely dashed]
   node[right]{$a, b, c$} (end);
   \end{tikzpicture}
  \caption{%
    \label{fig:mtsquotient}
    Two nondeterministic MTS, their quotient, and its simplification by pruning}
\end{figure}

\begin{example}
  We illustrate the construction on an example.  Let $S$ and $T$ be
  the MTS on the top of Figure~\ref{fig:mtsquotient}.  We construct
  $S/ T$, displayed below; this can be further simplified into the
  system on the bottom.

  First we construct the may-successors of $s^0/ t^0$. Both $b$ and $c$ are admissible due to $t^0\notmay{b}$ and $t^0\notmay{c}$ and thus with $\postra_b(s^0/t^0)=\postra_c(s^0/t^0)=\{\emptyset\}$.
  Consequently, the only successor here is $\emptyset$.
  Further, $a$ is also admissible due to $s^0\DMTSmay{a}$.
  For may-transitions under $a$,
  we have to consider all mappings of successors of $t^0$ to
  successors of $s^0$, namely $\{s_1/ t_1,s_1/ t_2\}$,
  $\{s_1/ t_1,s_2/ t_2\}$, $\{s_2/ t_1,s_1/ t_2\}$, and
  $\{s_2/ t_1,s_2/ t_2\}$. Besides, since there is a must-transition from
  $s^0$ (to $s_1$), we create a disjunctive must-transition to all
  successors that can be used to yield this must-transition
  $s^0\DMTSmust{a}s_1$ when composed
  with the must-transition $t^0\DMTSmust{a}t_1$. These are all
  successors where $t_1$ is mapped to $s_1$, hence the first
  two.

  Further, $\{s_1/ t_1,s_2/ t_2\}$ is obliged to have a must under $b$ so
  that it refines $s_1$ when composed with $t_1$, but cannot have any
  $c$ in order to match $s_2$ when composed with $t_2$. Similarly,
  $\{s_2/ t_1,s_2/ t_2\}$ has neither $c$ nor $b$.

  The first and third successor of $s^0/ t^0$ deserve special attention.
  Firstly, $\{s_1/ t_1,s_1/ t_2\}$ has (apart from a may-transition under $a$) transitions under $b$: may to $\{\bullet/\bullet\}$ and two musts.
  Both musts are due to $s_1$, but the first one because of $s_1$ in $s_1/ t_1$ (leading to $\{\bullet/\bullet\}$) and the second one because of $s_1$ in $s_1/ t_2$ (leading to an empty disjunction because $t_2\not\DMTSmust{}$).
  The empty disjunction is drawn as a line not branching anywhere.
  Note that it is very different from a may-transition to $\emptyset$ and cannot be implemented.
  States with such a transition are drawn in gray here and called \emph{inconsistent}.

  Secondly, $\{s_2/ t_1,s_1/ t_2\}$ is inconsistent for the same reason:
  it requires to refine $s_1$ by a composition with $t_2$.
  As $t_2$ has no must under $b$, the composition has none
  either, hence the must of $s_1$ can never be matched.
\end{example}

We have seen that the construction may produce empty must-disjunctions and thus also inconsistent states, i.e., states $s$ such that $s\DMTSmust{}\emptyset$.
Since inconsistent states have no implementations, their presence in the system is useless and we can remove them from the system using the procedure of \emph{pruning}.
This procedure produces a more readable system that has the same set of implementations and, moreover, is modally refining the original system.
The procedure is standard for MTS, see e.g.\ \cite{DBLP:journals/mscs/BauerJLLS12};
here we describe its straightforward adaptation for DMTS.
The procedure exhaustively repeats the following:
if there is an inconsistent state $s$, then remove it together with all its outgoing transitions (both may and must) and incoming may-transitions, and each remaining must-transition $t\DMTSmust{}D$ is modified into $t\DMTSmust{}D\setminus(\Sigma\times\{s\})$;
intuitively, we are removing the incoming must-branches (not the whole transitions).
This may of course turn other states inconsistent and thus the procedure is repeated until there are no more inconsistent states.

\begin{example}
  When we apply pruning to the quotient in the previous example, we
  obtain the system on the bottom of Figure~\ref{fig:mtsquotient}.
  Here the gray inconsistent states are removed and the disjunctive
  must from $\{s^0/ t^0\}$ leads only to $\{s_1/ t_1,s_2/ t_2\}$.

  Now it is easy to see that $T\|( S/ T)\mreq S$ in this case.
\end{example}

Recall from Lemma~\ref{lem:mts-naa-comp} that the MTS composition
is a~conservative approximation to the \NAA composition.
This means that the following theorem holds regardless of which of
the two compositions is used.

\begin{theorem}
  \label{thm:mts-quotient}
  For all MTS specifications $\mcalM_1$, $\mcalM_2$ and $\mcalM_3$,
  $\mcalM_1\| \mcalM_2\mr \mcalM_3$ iff
  $\mcalM_2 \mr \mcalM_3/\mcalM_1$.
\end{theorem}

\begin{proof}
  In this proof only, let $\|$ denote MTS composition.

  Write $\mcalM_i = ( S_i, s_i^0, \omay_i, \omust_i)$ for $i= 1, 2, 3$.
  We use the following notation to help distinguish states of
  $\mcalM_1\|\mcalM_2$ and $\mcalM_3/\mcalM_1$. The states of $\mcalM_1\|\mcalM_2$
  are denoted by $s_1\|s_2$ instead of $(s_1,s_2)$ while the states
  of $\mcalM_3/\mcalM_1$ are denoted by $\{s_3/ s_1,\ldots\}$
  instead of $\{(s_3,s_1),\ldots\}$.
  We also note that for states of $\mcalM_3/\mcalM_1$, $s \supseteq t$
  implies $s \mr t$ due to the construction.

  Now assume that $\mcalM_2 \mr \mcalM_3/\mcalM_1$ and let $R =
  \{(s_1\|s_2,s_3)\mid s_2 \mr \{s_3/ s_1\} \}$.  We show that $R$ is
  a~witness for $\mcalM_1\|\mcalM_2\mr\mcalM_3$, \ie that it satisfies the
  conditions of Definition~\ref{def:dmts_mr}.  Let $(s_1\|s_2,s_3) \in
  R$.
  \begin{itemize}
  \item Let $s_1\|s_2 \DMTSmay{a} t_1\|t_2$. As $s_2 \mr \{s_3 / s_1\}$
    this means that $\{ s_3 / s_1 \} \DMTSmay{a} \{t^1_3/ t^1_1,\ldots,
    t^k_3/ t^k_1\} = t$ and $t_2 \mr t$.  Due to the construction of
    $\{s_3 / s_1\}$, we know that there is an index $j$ for which
    $t^j_1 = t_1$ and $s_3 \DMTSmay{a} t^j_3$.  Let $t_3 = t^j_3$.  As $t
    \supseteq \{t_3/ t_1\}$, $t \mr \{t_3 / t_1\}$. Therefore,
    since $t_2 \mr t$, we have $t_2 \mr \{t_3 / t_1\}$ and thus
    $(t_1\parallel t_2 ,t_3) \in R$.
  \item Let $s_3 \DMTSmust{a} t_3$. This means that
    $\{ s_3 / s_1 \} \DMTSmust{} U =
      \{ (a, u) \in \{a\} \times \postra[a]{\{s_3/ s_1\}} \mid
         \exists t_1 : t_3 / t_1 \in u, s_1 \DMTSmust{a} t_1 \}$.
    As $s_2 \mr \{ s_3 / s_1 \}$, we know that $s_2 \DMTSmust{a} t_2$ and
    $t_2 \mr u$ for some $(a,u) \in U$.  Due to the construction of
    $U$ we know that there exists $t_1$ such that $t_3 / t_1 \in u$
    with $s_1 \DMTSmust{a} t_1$. Thus $s_1 \| s_2 \DMTSmust{a} t_1 \| t_2$.
    Again, as $u \supseteq \{t_3/ t_1\}$, $t_2 \mr \{t_3/ t_1\}$.
    Therefore, $(t_1\parallel t_2,t_3) \in R$.
  \end{itemize}

  Assume, for the other direction of the proof, that $\mcalM_1 \parallel
  \mcalM_2 \mr \mcalM_3$.  Define
  \begin{equation*}
    R = \{ (s_2, \{ s^1_3 / s^1_1, \ldots, s^n_3 / s^n_1 \})
    \mid \forall i = 1, \ldots, n : s^i_1 \parallel s_2 \mr s^i_3 \}\,;
  \end{equation*}
  note that $(s_2,\emptyset) \in R$ for all $s_2 \in S_2$.  We show that
  $R$ is a~witness for $\mcalM_2 \mr \mcalM_3/\mcalM_1$.  Let $(s_2,s) \in R$
  with $s = \{ s^1_3 / s^1_1, \ldots, s^n_3 / s^n_1 \}$.
  \begin{itemize}
  \item Let $s_2 \DMTSmay{a} t_2$. If there is no $i$ such that
    $s^i_1 \DMTSmay{a}$ then $s \DMTSmay{a} \emptyset$ and
    $(t_2,\emptyset) \in R$.  Otherwise, for each
    $i \in \{1, \ldots, n\}$
    and each $j \in \{1, \ldots, m_i\}$ such that
    $s^i_1 \DMTSmay{a} t^{i,j}_1$ consider that we have
    $s^i_1 \| s_2 \DMTSmay{a} t^{i,j}_1 \| t_2$ and as
    $s^i_1 \| s_2 \mr s^i_3$ we also have a~corresponding
    $s^i_3 \DMTSmay{a} t^{i,j}_3$ with $t^{i,j}_1 \| t_2 \mr t^{i,j}_3$.
    We fix these $t^{i,j}_3$ for each $i$ and $j$.  Let
    $t = \{ t^{i,j}_3 / t^{i,j}_1 \mid i \in \{1,\ldots,n\}, j \in
    \{1,\ldots,m_i\}\}$.  Clearly, $s \DMTSmay{a} t$ and $(t_2,t) \in R$.
  \item Let $s \DMTSmust{} U$ (note that this means that $s \ne \emptyset$)
    and let $s^i_3 \DMTSmust{a} t^i_3$ be
    the corresponding must transition in the construction.  As
    $s^i_1 \| s_2 \mr s^i_3$, this means that $s_2 \DMTSmust{a} t_2$ and
    $s^i_1 \DMTSmust{a} t^i_1$ such that $t^i_1 \| t_2 \mr t^i_3$.  This
    also means that $s_2 \DMTSmay{a} t_2$. We thus build $t$ as we did in
    the previous case where for $i$, $j$ such that $t^{i, j}_1 = t^i_1$
    we choose the corresponding $t^{i, j}_3$ to be $t^i_3$.
    Clearly $(t_2,t) \in R$. \qed
  \end{itemize}
\end{proof}

\subsubsection{Quotient for \NAA}

We now introduce the general quotient operator for \NAA. The
construction is similar to the previous one, with the notions of
permissibility and $\postra[a]{s}$ adapted to the more general setting.

Let $\mcalA_1=( S_1, S_1^0, \Tran_1)$, $\mcalA_2=( S_2, S_2^0, \Tran_2)$ be
\NAA\ and define $\mcalA_1/ \mcalA_2=( S, S^0, \Tran)$, with
$S= 2^{ S_1\times S_2}$.  To define the set of initial states, let us
first enumerate the initial states of $\mcalA_2$ as follows:
$S_2^0=\{ s_2^{ 0, 1},\dotsc, s_2^{ 0, p}\}$.  The set of initial
states is given by all possible assignments of states from $S_1^0$ to
states of~$S_2^0$, formally:
$S^0= \big\{\{( s_1^{ 0, q}, s_2^{ 0, q})\mid q\in\{1,\dotsc, p\}\}\bigmid
\forall q: s_1^{ 0, q}\in S_1^0\big\}$.

The assignment of transition constraints $\Tran$ is given as follows.
Let $\Tran( \emptyset)= 2^{ \Sigma\times\{ \emptyset\}}$.  For $s=\{(
s_1^1, s_2^1),\dots,( s_1^n, s_2^n)\}\in S$, say that $a\in \Sigma$ is
\emph{permissible from $s$} if it holds for all $i= 1,\dots, n$ that
there is $M_1\in \Tran_1( s_1^i)$ and $t_1\in S_1$ for which $( a,
t_1)\in M_1$, or else there is no $M_2\in \Tran_2( s_2^i)$ and $t_2\in S_2$
with $(a, t_2) \in M_2$.

Let us now fix a nonempty
$s=\{( s_1^1, s_2^1),\dotsc,( s_1^n, s_2^n)\}\in S$.  We introduce
some notation that we are going to use throughout this construction
and the following proof to denote the successor states of $s_2^i$ for
all $i$.  For each $s_2^i$ let
$\Tran_2( s_2^i)=\{ M_2^{ i, 1},\dotsc, M_2^{ i, m_i}\}$ be a fixed
enumeration of $\Tran_2(s_2^i)$. For each $a$ permissible from $s$ and
for each $M_2^{ i, j}$ we further fix an enumeration of all states in
$M_2^{ i, j}$ after an $a$-transition:
$\{ t_2\in S_2\mid( a, t_2)\in M_2^{ i, j}\}=\{ t^{ i, j, a,
  1},\dotsc, t^{ i, j, a, r_{ i, j, a}}\}$.
This means that $t^{ i, j, a, k}$ is the $k$th state (out of
$r_{ i, j, a}$) with an $a$-transition in $M_2^{ i, j}$, which is the
$j$th member of $\Tran_2( s_2^i)$, where $s_2^i$ is the state in the
$i$th pair in $s$.

For $a$ permissible from $s$, we define
\begin{multline*}
  \postra[a]{s}= \big\{\{( t_1^{ i, j, a, k}, i, j, t_2^{ i, j, a,
    k})\mid i= 1,\dotsc, n, j= 1,\dotsc, m_i, k= 1,\dotsc, r_{ i, j,
    a}\} \\
  {}\bigmid \forall i, j, k: \exists M_1\in \Tran_1( s_1^i):( a, t_1^{ i, j, a,
    k})\in M_1 \big\},
\end{multline*}
the set of all sets of possible assignments of next-$a$ states from
$s_1^i$ to next-$a$ states from $s_2^i$. Note that unlike the case of
MTS quotient, we also keep the indices $i$, $j$ in the assignments.
We further define
$\postra{s} = \{ (a, x) \mid a\text{ permissible from }s, x \in
\postra[a]{s} \}$.

To deal with the elements of $\postra[a]{s}$ we define the following
auxiliary operations. The first operation $\ell$ allows us to
``forget'' the indices $i$, $j$ and is defined as follows:
$\ell(x) = \{ (t_1, t_2) \mid \exists i, j : (t_1, i, j, t_2) \in x
\}$.
The operation can be naturally lifted to subsets of $\postra{s}$ as
follows: $\ell(N) = \{ (a, \ell(x)) \mid (a, x) \in N \}$ where
$N \subseteq \postra{s}$.

The second operation is a type of projection that given the two
indices $i$, $j$ and the state $t_2^{i, j, a, k}$ produces the
next-$a$ state of $s_1$ assigned in the given element $x$ of
$\postra[a]{s}$. Note that the projection is defined uniquely:
$x \upharpoonright (i, j, t_2^{i,j,a,k}) = t_1^{i,j,a,k}$ where
$x = \{\ldots, (t_1^{i,j,a,k}, i, j, t_2^{i,j,a,k}), \ldots \}$.  The
projection operation can also be lifted to subsets $N$ of $\postra{s}$
and sets $M_2 \in \Tran_2(s^2_i)$ as follows:
$N \upharpoonright (i, j, M_2) = \{(a, x \upharpoonright (i, j, t_2)) \mid (a, t_2) \in
M_2\}$.
Note that the result of this operation is then a set of elements of
the form $(a,t_3)$ where $t_3$ is an $a$-successor of $s_3$.

Having the two auxiliary operations, we can then finally define
\begin{multline}
  \label{eq:qtran}
  \Tran(s) = \{ \ell(N) \mid N \subseteq \postra{s},
  \forall i \in \{1, \ldots, n\}, j \in \{1, \ldots, m_j\} : \\
  N \upharpoonright (i, j, M_2^{i,j}) \in \Tran_3(s_3^i) \}\,.
\end{multline}

\begin{theorem}
  \label{th:quotient}
  For all specifications $\mcalS_1$, $\mcalS_2$, $\mcalS_3$,
  $\mcalS_1\| \mcalS_2\mr \mcalS_3$ iff
  $\mcalS_2\mr \mcalS_3/ \mcalS_1$.
\end{theorem}

\begin{proof}
  We show the proof for \NAA.
  Let $\mcalA_1=( S_1, S^0_1, \Tran_1)$ and $\mcalA_2=( S_2, S^0_2, \Tran_2)$,
  $\mcalA_3=( S_3, S^0_3, \Tran_3)$; we show that
  $\mcalA_1\| \mcalA_2\mr \mcalA_3$ iff $\mcalA_2\mr \mcalA_3/ \mcalA_1$.

  We use the notation introduced in the proof of
  Theorem~\ref{thm:mts-quotient}, \ie~$s_1\|s_2$ instead of
  $(s_1,s_2)$ when speaking about states of~$\mcalA_1\|\mcalA_2$ and
  $\{s_3/ s_1, \ldots\}$ instead of $\{(s_3,s_1),\ldots\}$ when
  speaking about states of $\mcalA_3/\mcalA_1$.  We further note that by
  construction, $s\supseteq t$ implies $s\mr t$ for all
  $s, t\in 2^{S_3\times S_1}$.

  Now assume that $\mcalA_2\mr \mcalA_3/ \mcalA_1$ and let $R=\{( s_1\| s_2,
  s_3)\mid s_2\mr \{s_3/ s_1\}\}$; we show that $R$ is a witness for
	$\mcalA_1\| \mcalA_2\mr \mcalA_3$.

  Let $( s_1\| s_2, s_3)\in R$ and $M_\|\in \Tran_\|( s_1\| s_2)$.  Then
  $M_\|= M_1\| M_2$ with $M_1\in \Tran_1( s_1)$ and $M_2\in \Tran_2(
  s_2)$.  As $s_2\mr \{s_3/ s_1\}$, we can pair $M_2$ with an $M_/\in
  \Tran_/( \{s_3/ s_1\})$, such that the conditions
  in~\eqref{eq:aaref} are satisfied (see Definition~\ref{def:aaref}).

  Note that, as $M_1 \in \Tran_1(s_1)$, there exists some $j$ such that
  $M_1 = M_1^{1,j}$, \ie~$M_1$ is the $j$th element of $\Tran(s_1)$ in the
  enumeration as described in the construction of the quotient.
  Further note that $M_\| = \ell(N)$ for some
  $N \subseteq \postra{\{s_3/ s_1\}}$ satisfying the conditions
  in~\eqref{eq:qtran}.

  We now define $M_3 = N \upharpoonright (1, j, M_1)$ and show that~\eqref{eq:aaref}
  holds for the pair $M_\|, M_3$:
  \begin{itemize}
  \item Let $( a, t_1\| t_2)\in M_\|$, then there are
    $( a, t_1)\in M_1$ and $( a, t_2)\in M_2$. By~\eqref{eq:aaref}
    applied to the pair $M_2$, $M_/$, there is $( a, t)\in M_/$
    such that $t_2\mr t$. This means that there is $(a, x)\in N$ such
    that $t = \ell(x)$. Due to the construction of
    $\postra{\{s_3/ s_1\}}$ there has to be some $t_3$ such that
    $(t_3, 1, j, t_1) \in x$.  Due to the definition of $M_3$, this
    means that $(a, t_3) \in M_3$.  We also know that
    $t \supseteq \{t_3/ t_1\}$, hence $t \mr \{t_3/ t_1\}$, and
    together with $t_2 \mr t$ we get $t_2 \mr \{t_3/ t_1\}$. Thus
    $(t_1\|t_2,t_3) \in R$.
  \item Let $( a, t_3)\in M_3$. This means that there is some
    $(a, x) \in N$ with $(t_3, 1, j, t_1) \in x$ and
    $(a,t_1) \in M_1$.  Therefore, $\ell(x) = t \in M_/$ with
    $t_3/ t_1 \in t$.  Due to~\eqref{eq:aaref} applied to $M_2$,
    $M_/$, there has to be a~corresponding $(a, t_2) \in M_2$ such
    that $t_2 \mr t$. Thus $(a, t_1\| t_2) \in M_\|$ and again by
    $t \supseteq \{t_3/ t_1\}$ we have $t_2 \mr \{t_3/ t_1\}$ and
    hence $(t_1\| t_2, t_3) \in R$.
  \end{itemize}

  It remains to show that $R$ is initialised. Let
  $\hat s_1^0 \| \hat s_2^0$ be an initial state of $\mcalA_1\|\mcalA_2$.
  By $\mcalA_2 \mr \mcalA_3/\mcalA_1$ we know that $\hat s_2^0 \mr s^0$ for
  some $s^0 \in S^0$. We then take $\hat s_3^0 \in S_3^0$ such that
  $\hat s_3^0/\hat s_1^0 \in s^0$ (there has to be exactly
  one
  due to the definition of $s^0$).  We then
  have $s^0 \supseteq \{ \hat s_3^0/\hat s_1^0 \}$ and thus
  $s^0 \mr \{ \hat s_3^0/\hat s_1^0 \}$. This means that
  $\hat s_2^0 \mr \{ \hat s_3^0/\hat s_1^0 \}$ and hence
  $(\hat s_1^0 \| \hat s_2^0, \hat s_3^0) \in R$.

  Assume, for the other direction of the proof, that $\mcalA_1\| \mcalA_2\mr
  \mcalA_3$.  Define $R\subseteq S_2\times 2^{ S_3\times S_1}$ by
  \begin{equation*}
    R=\{( s_2,\{ s_3^1/ s_1^1,\dotsc, s_3^n/ s_1^n\})\mid \forall i=
    1,\dotsc, n: s_1^i\| s_2\mr s_3^i\}\,;
  \end{equation*}
  we show that $R$ is a witness for $\mcalA_2\mr \mcalA_3/ \mcalA_1$.  We
  first note that $(s_2, \emptyset) \in R$ for all $s_2$.
  Let now
  $(s_2, s)\in R$, with nonempty
  $s=\{ s_3^1/ s_1^1,\dotsc, s_3^n/ s_1^n\}$, and
  $M_2\in \Tran_2( s_2)$.

  Note that for every $M_1^{i,j} \in \Tran(s_1^i)$ we can build
  $M_\|^{i,j} = M_1^{i,j} \| M_2$, and as $s_1^i\|s_2 \mr s_3^i$,
  there has to be a~corresponding $M_3 \in \Tran(s_3^i)$ satisfying
  the conditions of~\eqref{eq:aaref}. We fix such $M_3$ for every $i$,
  $j$ and denote it by $M_3^{i,j}$.

  We are going to build a~subset $N$ of $\postra{s}$. To that end, we
  first define an auxiliary notion of an \emph{adequate} element of
  $\postra{s}$ with respect to $(a, t_2) \in M_2$ as follows.  Let
  $(a, x) \in \postra{s}$. We say that $(a, x)$ is adequate
  w.r.t.~$(a,t_2)$ if for every $i \in \{1, \ldots, n\}$,
  $j \in \{1, \ldots, m_i\}$, and $k \in \{1, \ldots, r_{i,j,a}\}$,
  the projection $t_3= x\upharpoonright (i, j, t_1^{i,j,k})$ satisfies
  $(a, t_3) \in M_3^{i,j}$ and $t_1^{i, j, k} \| t_2 \mr t_3$.

  Clearly, if we have $(a, x)$ adequate w.r.t.~$(a, t_2)$, then
  $(t_2, \ell(x)) \in R$.

  We can now define
  \begin{equation*}
    N=\{( a, x)\in \postra{s}\mid \exists( a, t_2)\in M_2:( a,
    x)\text{ is adequate w.r.t.~}(a,t_2)\}\,.
  \end{equation*}

  We first need to show that $\ell(N) \in \Tran_/(s)$.  Let
  $i \in \{1, \ldots, n\}$ and $j \in \{1, \ldots, m_j\}$.  We want to
  show that $N \upharpoonright (i, j, M_1^{i,j}) = M_3^{i,j}$.  Let first
  $(a, t_3) \in N \upharpoonright (i, j, M_1^{i,j})$. This means that there is
  some $(a, x) \in N$ with $x = (t_3, i, j, t_1^{i,j,k})$.  Due to the
  definition of $N$, $(a, t_3) \in M_3^{i,j}$.  Let now
  $(a,t_3) \in M_3^{i,j}$.  Recall that the pair $M_1^{i,j}\|M_2$,
  $M_3^{i,j}$ satisfies~\eqref{eq:aaref}. This means that for
  $(a,t_3) \in M_3^{i,j}$ there exists
  $(a,t_1^{i,j,k}\|t_2) \in M_1^{i,j}\|M_2$ such that
  $t_1^{i, j, k} \| t_2 \mr t_3$.  Hence
  $(a, (t_3, i, j, t_1^{i,j,k}))$ is adequate w.r.t.~$(a,t_2)$ and
  thus $(a, t_3) \in N \upharpoonright (i, j, M_1^{i,j})$.

  We now show that the pair $M_2$, $M = \ell(N)$ satisfies the
  conditions of~\eqref{eq:aaref}.
  \begin{itemize}
  \item Let $(a,t_2) \in M_2$. We need to show that there exists
    $(a,x) \in N$ adequate w.r.t.~$(a,t_2)$. Recall that the pair
    $M_1^{i,j} \| M_2$, $M_3^{i,j}$ satisfies~\eqref{eq:aaref} for
    every $i$, $j$.  For every $(a, t_1^{i,j,k}\|t_2)$ there thus has
    to be $(a,t_3) \in M_3^{i,j}$ with $t_1^{i, j, k} \| t_2 \mr t_3$.
    We fix such $t_3$ for every $i$, $j$, $k$ and denote it by
    $t_3^{i,j,k}$. We then set
    $x = \{ (t_3^{i,j,k}, i, j, t_1^{i,j,k} \mid i \in \{1, \ldots,
    n\}, j \in \{1, \ldots, m_i\}, k \in \{1, \ldots, r_{i,j,a} \}$.
    Clearly $(a,x) \in N$ and $(a,x)$ is adequate w.r.t.~$(a,t_2)$.
    As noted above, we have $(t_2, \ell(x)) \in R$.
  \item Let $(a, t) \in M$. This means that there is some
    $(a, x) \in N$ such that $t = \ell(x)$. Due to the definition of
    $N$, there exists $(a, t_2) \in M_2$ such that $(a, x)$ is
    adequate w.r.t.~$(a,t_2)$.  Again, as noted above, we have
    $(t_2, t) = (t_2, \ell(x)) \in R$.
  \end{itemize}

  It remains to show that $R$ is initialised. Let $s_2^0$ be an
  initial state of $\mcalA_2$. By $\mcalA_1\|\mcalA_2\mr\mcalA_3$ we know that for
  every $s_1^{ 0, q} \in S_1^0$ (recall the enumeration of initial
  states in the construction of the quotient) there exists
  $s_3^0 \in S_3^0$ such that $s_1^{ 0, q}\| s_2^0\mr s_3^0$. Let us
  fix for every $q$ such $s_3^0$ and denote it by $s_3^{ 0, q}$. Let
  then $s^0=\{ s_3^{ 0, q}/ s_1^{ 0, q}\mid q\in\{ 1,\dotsc, p\}\}$.
  Clearly, $s^0$ is an initial state of $\mcalA_3/\mcalA_1$ and
  $(s_2^0, s^0) \in R$. \qed
\end{proof}

As a corollary, we get~\eqref{eq:quotient}: If $\mcalI_2\in \sem{ \mcalS/
  \mcalS_1}$, \ie~$\mcalI_2\mr \mcalS/ \mcalS_1$, then $\mcalS_1\| \mcalI_2\mr \mcalS$,
which using $\mcalI_1\mr \mcalS_1$ and Theorem~\ref{th:indimp} implies
$\mcalI_1\| \mcalI_2\mr \mcalS_1\| \mcalI_2\mr \mcalS$.  The reverse implication in
Theorem~\ref{th:quotient} implies that $\mcalS/ \mcalS_1$ is as permissive
as possible.

\begin{corollary}
  \label{th:resilat}
  With operations $\land$, $\lor$, $\|$ and $/$, each of our four
  classes of specifications forms a commutative residuated lattice up to
  $\mreq$.
\end{corollary}

\begin{proof}
  We have already seen in Corollary~\ref{co:distlat} that the class of
  \NAA\ forms a lattice, up to $\mreq$, under $\land$ and $\lor$, and by
  Theorem~\ref{th:quotient}, $/$ is the residual, up to $\mreq$, of
  $\|$.  All other properties (such as distributivity of $\|$ over
  $\lor$ or $\mcalN\| \bot\mreq \bot$) follow. \qed
\end{proof}

\section{Related Work}\label{sec:rw}

The modal $\nu$-calculus is equivalent to the Hennessy-Milner logic
with \emph{greatest} fixed points, which arises from Hennessy-Milner
logic (HML)~\cite{DBLP:journals/jacm/HennessyM85} by introducing
variables and greatest fixed points. If also \emph{least} fixed points
are allowed, one arrives at the full \emph{modal
  $\mu$-calculus}~\cite{unpub/ScottB69, DBLP:conf/focs/Pratt81,
  DBLP:journals/tcs/Kozen83}.  Janin and Walukiewicz have
in~\cite{DBLP:conf/mfcs/JaninW95} introduced an automata-like
representation for the modal $\mu$-calculus which seems related to our
\NAA.

DMTS have been proposed as solutions to algebraic process equations in
Larsen and Xinxin's~\cite{DBLP:conf/lics/LarsenX90} and further
investigated also as a specification
formalism~\cite{DBLP:conf/concur/Larsen90,
  DBLP:conf/atva/BenesCK11}. The DMTS formalism is a member of the modal
transition systems (MTS) family and as such has also received attention
recently. The MTS formalisms have proven to be useful in
practice. Industrial applications started with
Bruns'~\cite{DBLP:journals/scp/Bruns97} where MTS have been used for an
air-traffic system at Heathrow airport. Besides, MTS classes are
advocated as an appropriate base for interface theories by Raclet \etal
in~\cite{DBLP:conf/acsd/RacletBBCP09} and for product line theories in
Nyman's~\cite{thesis/Nyman08}. Further, an MTS based software
engineering methodology for design via merging partial descriptions of
behavior has been established by Uchitel and Chechik
in~\cite{DBLP:conf/sigsoft/UchitelC04} and methods for supervisory
control of MTS shown by Darondeau \etal in~\cite{Darondeau2010a}.  Tool
support is quite extensive, \eg~\cite{DBLP:journals/fmsd/BorjessonLS95,
  DBLP:conf/eclipse/DIppolitoFFU07, DBLP:conf/atva/BauerML11, motras}.

Over the years, many extensions of MTS have been proposed, surveyed in
more detail in~\cite{DBLP:conf/birthday/Kretinsky17,
  conf/sifakis/FahrenbergLLT14, conf/sifakis/FahrenbergLT14}. While
MTS can only specify whether or not a particular transition is
required, some extensions equip MTS with more general abilities to
describe what \emph{combinations} of transitions are possible. These
include DMTS~\cite{DBLP:conf/lics/LarsenX90}, Fecher and Schmidt's
1-MTS~\cite{DBLP:journals/jlp/FecherS08} allowing to express exclusive
disjunction, OTS~\cite{DBLP:conf/memics/BenesK10} capable of
expressing positive Boolean combinations, and Boolean
MTS~\cite{DBLP:conf/atva/BenesKLMS11} covering all Boolean
combinations. The last one is closely related to our \NAA
as well as hybrid modal logic~\cite{book/Prior68,
  DBLP:journals/igpl/Blackburn00}.  Our results show that all these
formalisms are at most as expressive as DMTS.

Larsen has shown in~\cite{DBLP:conf/avmfss/Larsen89} that any finite
acyclic MTS is equivalent to a HML formula (without recursion or fixed
points), the \emph{characteristic formula} of the given MTS,
\cf~\eqref{eq:dmtstonu}.  Conversely, Boudol and Larsen show
in~\cite{DBLP:journals/tcs/BoudolL92} that any consistent and
\emph{prime} HML formula is equivalent to a~MTS.\footnote{A HML formula
  is \emph{prime} if implying a disjunction means implying one of the
  alternatives.}
Here we extend these results to $\nu$-calculus formulae, and show that
any such formula is equivalent to a DMTS, solving a problem left open
in~\cite{DBLP:conf/lics/LarsenX90}.  Hence the modal $\nu$-calculus
supports full compositionality and decomposition in the sense
of~\cite{DBLP:conf/concur/Larsen90}.  This finishes some of the work
started in~\cite{DBLP:conf/avmfss/Larsen89, DBLP:journals/tcs/BoudolL92,
  DBLP:conf/concur/Larsen90}.  Recently, the graphical representability
of a variant of alternating simulation called covariant-contravariant
simulation has been studied in~\cite{DBLP:journals/scp/AcetoFFIP13}.

Quotients are related to \emph{decomposition} of processes and
properties, an issue which has received considerable attention through
the years.  In~\cite{DBLP:conf/lics/LarsenX90}, a solution to
bisimulation $C(X)\sim P$ for a given process $P$ and context $C$ is
provided (as a~DMTS). This solves the quotienting problem $P/ C$ for
the special case where both $P$ and $C$ are processes.
This is extended in~\cite{DBLP:conf/icalp/LarsenX90} to the setting
where the context $C$ can have several holes and $C(X_1,\ldots,X_n)$
must satisfy a $\nu$-calculus property $Q$.  However, $C$ remains to be
a process context, not a~specification context. Our \emph{specification}
context allows for arbitrary specifications, representing infinite sets
of processes and process equations.  Other extensions use infinite
conjunctions~\cite{DBLP:journals/tcs/FokkinkGW06}, probabilistic
processes~\cite{DBLP:conf/concur/GeblerF12} or processes with continuous
time and space~\cite{DBLP:conf/icalp/CardelliLM11}.

Quotient operators, or \emph{guarantee} or \emph{multiplicative
  implication} as they are called there, are also well-known from
various logical formalisms.  Indeed, the algebraic properties of our
parallel composition $\|$ and quotient $/$ resemble closely those of
multiplicative conjunction $\&$ and implication $\multimap$ in
\emph{linear logic}~\cite{DBLP:journals/tcs/Girard87}, and of spatial
conjunction and implication in \emph{spatial
  logic}~\cite{DBLP:journals/iandc/CairesC03} and \emph{separation
  logic}~\cite{DBLP:conf/lics/Reynolds02, DBLP:conf/csl/OHearnRY01}.
For these and other logics, proof systems have been developed which
allow one to reason about expressions containing these operators.
In these logics, $\&$ and $\multimap$
are first-class operators on par with the other logical operators, and
their semantics are defined as certain sets of processes.  In contrast,
for \NAA\ and hence, via the translations, also for $\nu$-calculus, $\|$
and $/$ are \emph{derived} operators, and we provide constructions to
reduce any expression which contains them, to one which does not.  This
is important from the perspective of reuse of components and useful in
industrial applications. To the best of our knowledge, there are no
other such \emph{reductions} of quotient for the synchronisation type of
composition in the context of specifications.

\section{Conclusion}\label{sec:conclusion}

In this chapter we have introduced a~general specification framework
whose basis consists of four different but equally expressive
formalisms: one of a~graphical behavioral kind (DMTS), one
logic-based ($\nu$-calculus) and two intermediate languages between
the former two (\NAA\ and hybrid modal logic).  We have shown their
structural equivalence.

The established connection implies several consequences.  On the one
hand, it allows for a graphical representation of $\nu$-calculus.
Further, composition on DMTS can be transferred to the modal
$\nu$-calculus, hence turning it into a modal process algebra.  On the
other hand, such a correspondence identifies a~class of modal transition
systems with a natural expressive power and provides another
justification of this formalism. Further, this class is closed under
both conjunction and disjunction, a requirement raised by
component-based design methods.  However, it is not closed under
complement and difference.\footnote{Previous results on
  difference~\cite{DBLP:journals/sosym/SassolasCU11} are incorrect due
  to a mistake in~\cite{DBLP:conf/sigsoft/FischbeinU08} on conjunction
  of MTS, see~\cite[p.~36]{thesis-MU14}.}  Nevertheless, since DMTS are
closed under conjunction, disjunction and composition, we still have a
positive Boolean process algebra.

Altogether, we have shown that the framework possesses a rich algebraic
structure that includes logical (conjunction, disjunction) and
behavioral operations (parallel composition and quotient) and forms a
complete specification theory in the sense
of~\cite{DBLP:conf/concur/Larsen90, DBLP:conf/fase/BauerDHLLNW12}.

Moreover, the construction of the quotient solves an open problem in
the area of MTS. All attempts to find the quotient for variants of MTS
so far have been limited to the much simpler deterministic
case~\cite{DBLP:journals/entcs/Raclet08}. Here we have given the first
solution to the quotient on nondeterministic specifications: first, a
quotient construction for MTS, and then a quotient for general DMTS.
Due to the established correspondence, the quotient can be applied
also to $\nu$-calculus formulae.  We remark that all our translations
and constructions are based on a new \emph{normal form} for
$\nu$-calculus expressions, and that turning a $\nu$-calculus
expression into normal form may incur an exponential blow-up.
However, the translations and constructions preserve the normal form,
so that this translation only need be applied once in the beginning.

\chapter[Compositionality for Quantitative
Specifications][Compositionality for Quantitative
Specifications]{Compositionality for Quantitative
  Specifications\footnote{This chapter is based on the journal
    paper~\cite{DBLP:journals/soco/FahrenbergKLT18} published in Soft
    Computing.}}
\label{ch:dmts2}

This chapter continues and finishes the work of Chapter~\ref{ch:wm2}.
It extends the quantitative theory of that chapter to the disjunctive
modal transition systems (DMTS) of Chapter~\ref{ch:dmts} and shows
that also in the quantitative setting, DMTS are closely related to
acceptance automata and the modal $\nu$-calculus.  The quantitative
theory of DMTS is shown to be rather pleasant, with better properties
than for pure MTS.

\section{Structured Labels}
\label{se:structlabels}

Let $\Sigma$ be a poset with partial order $\labpre$.  We think of
$\labpre$ as \emph{label refinement}, so that if $a\labpre b$, then $a$
is less permissive (more restricted) than $b$.

\begin{definition}
  A label $a\in \Sigma$ is an \emph{implementation label} if $b\labpre
  a$ implies $b= a$ for all $b\in \Sigma$.  The set of implementation
  labels is denoted $\Gamma$, and for $a\in \Sigma$, we let $\impl
  a=\{ b\in \Gamma\mid b\labpre a\}$ denote the set of its
  implementations.
\end{definition}

Hence $a$ is an implementation label iff $a$ cannot be further
refined.  Note that $a\labpre b$ implies $\impl a\subseteq \impl b$
for all $a, b\in \Sigma$.

\begin{example}
  \label{ex:labelsets}
  A trivial but important example of our label structure is the
  \emph{discrete} one in which label refinement $\labpre$ is equality
  (and $\Gamma= \Sigma$).  This is equivalent to the ``standard'' case
  of \emph{unstructured} labels.

  A typical label set in quantitative applications consists of a
  discrete component and real-valued weights.  For specifications,
  weights are replaced by (closed) weight \emph{intervals}, so that
  $\Sigma= U\times\{[ l, r]\mid l\in \Real\cup\{ -\infty\}, r\in
  \Real\cup\{ \infty\}, l\le r\}$ for a finite set $U$,
  \cf~\cite{DBLP:journals/fmsd/BauerFJLLT13,
    DBLP:journals/mscs/BauerJLLS12}.  Label refinement is given by $(
  u_1,[ l_1, r_1])\labpre( u_2,[ l_2, r_2])$ iff $u_1= u_2$ and $[ l_1,
  r_1]\subseteq[ l_2, r_2]$, so that labels are more refined if they
  specify smaller intervals; thus, $\Gamma= U\times\{[ x, x]\mid x\in
  \Real\}\approx U\times \Real$.

  For a quite general setting, we can instead start with an arbitrary
  set $\Gamma$ of implementation labels, let $\Sigma= 2^\Gamma$, the
  powerset, and $\mathord\labpre= \mathord\subseteq$ be subset
  inclusion.  Then $\impl a= a$ for all $a\in \Sigma$.  (Hence we
  identify implementation labels with one-element subsets of $\Sigma$.)
  \qed
\end{example}

\subsection{Label operations}

Specification theories come equipped with several standard operations
that make compositional software design
possible~\cite{DBLP:conf/fase/BauerDHLLNW12}: conjunction for merging
viewpoints covering different system's
aspects~\cite{DBLP:conf/sigsoft/UchitelC04,
  DBLP:conf/concur/Ben-DavidCU13}, structural composition for running
components in parallel, and quotient to synthesize missing parts of
systems~\cite{DBLP:conf/lics/LarsenX90}. In order to provide them
for DMTS, we first need the respective atomic operations on their action
labels.

We hence assume that $\Sigma$ comes equipped with a partial conjunction,
\ie~an operator $\oland: \Sigma\times \Sigma\parto \Sigma$ for which it
holds that
\begin{enumerate}[(1)]
\item \label{en:oland.lb} if $a_1\oland a_2$ is defined, then $a_1\oland
  a_2\labpre a_1$ and $a_1\oland a_2\labpre a_2$, and
\item \label{en:oland.glb} if $a_3\labpre a_1$ and $a_3\labpre a_2$, then
  $a_1\oland a_2$ is defined and $a_3\labpre a_1\oland a_2$.
\end{enumerate}
Note that by these properties, any two partial conjunctions on $\Sigma$
have to agree on elements for which they are both defined.

\begin{example}
  \label{ex:conjunction}
  For discrete labels, the unique conjunction operator is given by
  \begin{equation*}
    a_1\oland a_2=
    \begin{cases}
      a_1 &\text{if } a_1= a_2\,,\\
      \text{undef.} &\text{otherwise}\,.
    \end{cases}
  \end{equation*}
  Indeed, by property~\eqref{en:oland.glb}, $a_1\oland a_2$ must be
  defined for $a_1= a_2$, and by~\eqref{en:oland.lb}, if $a_1\oland
  a_2= a_3$ is defined, then $a_3= a_1$ and $a_3= a_2$.

  For labels in $U\times\{[ l, r]\mid l, r\in \Real, l\le r\}$,
  the unique conjunction is
  \begin{equation*}
    ( u_1,[ l_1, r_1])\oland( u_2,[ l_2, r_2])=
    \begin{cases}
      \text{undef.} \qquad\text{if } u_1\ne u_2\text{ or }[ l_1,
      r_1]\cap[ l_2, r_2]= \emptyset\,,\\
      ( u_1,[ l_1, r_1]\cap[ l_2, r_2]) \qquad\text{otherwise}\,.
    \end{cases}
  \end{equation*}
  To see uniqueness, let $a_i=( u_i,[ l_i, r_i])$ for $i= 1, 2, 3$.
  Using property~\eqref{en:oland.glb}, we see that $a_1\oland a_2$
  must be defined when $u_1= u_2$ and $[ l_1, r_1]\cap[ l_2, r_2]\ne
  \emptyset$, and by~\eqref{en:oland.glb}, if $a_1\oland a_2= a_3$ is
  defined, then $u_3= u_1$ and $u_3= u_2$, and $[ l_3, r_3]\subseteq[
  l_1, r_1]$, $[ l_3, r_3]\subseteq[ l_2, r_2]$ imply $[ l_1,
  r_1]\cap[ l_2, r_2]\ne \emptyset$.

  Finally, for the case of specification labels as sets of
  implementation labels, the unique conjunction is $a_1\oland a_2=
  a_1\cap a_2$. \qed
\end{example}

For structural composition and quotient of specifications, we assume a
partial \emph{label synchronization} operator $\mathord{\obar}:
\Sigma\times \Sigma\parto \Sigma$ which specifies how to compose
labels.  We assume $\obar$ to be associative and commutative, with the
following technical property which we shall need later: For all $a_1,
a_2, b_1, b_2\in \Sigma$ with $a_1\labpre a_2$ and $b_1\labpre b_2$,
$a_1\obar b_1$ is defined iff $a_2\obar b_2$ is, and if both are
defined, then $a_1\obar b_1\labpre a_2\obar b_2$.

\begin{example}
  \label{ex:composition}
  For discrete labels, the conjunction of Example~\ref{ex:conjunction}
  is the same as CSP-style composition, \ie~$a\obar b= a$ if $a= b$ and
  undefined otherwise, but other compositions can easily be defined.

  For labels in $U\times\{[ l, r]\mid l, r\in \Real, l\le r\}$, several
  useful label synchronization operators may be defined for different
  applications.  One is given by \emph{addition} of intervals, \ie
  \begin{equation*}
    ( u_1,[ l_1, r_1])\obarplus( u_2,[ l_2, r_2])=
    \begin{cases}
      \text{undef.} &\text{if } u_1\ne u_2\,,\\
      ( u_1,[ l_1+ l_2, r_1+ r_2]) &\text{otherwise}\,,
    \end{cases}
  \end{equation*}
  for example modeling computation time of actions on a single
  processor.  Another operator, useful in scheduling, uses maximum
  instead of addition:
  \begin{equation*}
    ( u_1,[ l_1, r_1])\obarmax( u_2,[ l_2, r_2])=
    \begin{cases}
      \text{undef.} &\text{if } u_1\ne u_2\,,\\
      ( u_1,[ \max( l_1, l_2), \max( r_1, r_2)])
      &\text{otherwise}\,.
    \end{cases}
  \end{equation*}

  For set-valued specification labels, we may take any synchronization
  operator $\obar$ given on implementation labels $\Gamma$ and lift it
  to one on $\Sigma$ by $a_1\obar a_2=\{ b_1\obar b_2\mid b_1\in \impl{
    a_1}, b_2\in \impl{ a_2}\}$. \qed
\end{example}

\section{Specification Formalisms}
\label{se:formalisms}

In this section we introduce the specification formalisms which we use
in the rest of the paper.  The universe of models for our specifications
is the one of standard \emph{labeled transition systems}.  For
simplicity of exposition, we work only with \emph{finite} specifications
and implementations, but most of our results extend to the infinite (but
finitely branching) case.

A \emph{labeled transition system} (LTS) is a structure $\mcalI=( S, s^0,
\omust)$ consisting of a finite set $S$ of states, an initial state
$s^0\in S$, and a transition relation $\omust \subseteq S\times
\Gamma\times S$.  We usually write $\smash{s\DMTSmust a t}$ instead of $( s,
a, t)\in \omust$.  Note that transitions are labeled with
\emph{implementation} labels.

\subsection{Disjunctive Modal Transition Systems}

A \emph{disjunctive modal transition system} (DMTS) is a structure
$\mcalD=( S, S^0,$ $\omay, \omust)$ consisting of finite sets
$S\supseteq S^0$ of states and initial states, respectively,
may-tran\-sitions $\omay\subseteq S\times \Sigma\times S$, and
disjunctive must-transitions
$\omust\subseteq S\times 2^{ \Sigma\times S}$.  It is assumed that for
all $( s, N)\in \omust$ and $( a, t)\in N$ there is
$( s, b, t)\in \omay$ with $a\labpre b$.

Note that we allow multiple (or zero) initial states.  We write
$\smash{s\DMTSmay a t}$ instead of $( s, a, t)\in \omay$ and $s\DMTSmust{} N$
instead of $( s, N)\in \omust$.

A DMTS $( S, S^0, \omay, \omust)$ is an \emph{implementation} if
$\omay\subseteq S\times \Gamma\times S$, $\omust=\{( s,\{( a,
t)\})\mid \smash{s\DMTSmay a t}\}$, and $S^0=\{ s^0\}$ is a singleton; DMTS
implementations are hence isomorphic to LTS.

DMTS were introduced in~\cite{DBLP:conf/lics/LarsenX90} in the context
of equation solving, or \emph{quotient} of specifications by processes
and are used \eg~in~\cite{DBLP:conf/atva/BenesCK11} for LTL model
checking.  They are a natural extension of \emph{modal} transition
systems~\cite{DBLP:conf/lics/LarsenT88}, which are DMTS in which all
disjunctive must-transitions $s\DMTSmust{} N$ lead to singletons $N=\{( a,
t)\}$; in fact, DMTS are the closure of MTS under
quotient~\cite{DBLP:conf/lics/LarsenX90}.

We introduce a notion of modal refinement of DMTS with structured
labels.  For discrete labels, it coincides with the classical
definition~\cite{DBLP:conf/lics/LarsenX90}.

\begin{definition}
  Let $\mcalD_1=( S_1, S^0_1, \omay_1, \omust_1)$ and $\mcalD_2=( S_2, S^0_2,
  \omay_2, \omust_2)$ be DMTS.  A relation $R\subseteq S_1\times S_2$ is
  a \emph{modal refinement} if it holds for all $( s_1, s_2)\in R$ that
  \begin{itemize}
  \item for all $s_1\DMTSmay{ a_1}_1 t_1$ there is $s_2\DMTSmay{ a_2}_2 t_2$
    such that $a_1\labpre a_2$ and $( t_1, t_2)\in R$, and
  \item for all $s_2\DMTSmust{}_2 N_2$ there is $s_1\DMTSmust{}_1 N_1$ such that
    for all $( a_1, t_1)\in N_1$ there is $( a_2, t_2)\in N_2$ with
    $a_1\labpre a_2$ and $( t_1, t_2)\in R$.
  \end{itemize}
  $\mcalD_1$ \emph{refines} $\mcalD_2$, denoted $\mcalD_1\mr \mcalD_2$, if there
  exists an \emph{initialized} modal refinement $R$, \ie~one for which
  it holds that for every $s_1^0\in S_1^0$ there is $s_2^0\in S_2^0$ for
  which $( s_1^0, s_2^0)\in R$.
\end{definition}

Note that this definition reduces to the one
of~\cite{DBLP:conf/lics/LarsenX90, DBLP:conf/atva/BenesCK11} for
discrete labels (\cf~Example~\ref{ex:labelsets}).

We write $\mcalD_1\mreq \mcalD_2$ if $\mcalD_1\mr \mcalD_2$ and $\mcalD_2\mr \mcalD_1$.
The \emph{implementation semantics} of a DMTS $\mcalD$ is $\impl \mcalD=\{
\mcalI\mr \mcalD\mid \mcalI~\text{implementation}\}$.  This is, thus, the set of
all LTS which satisfy the specification given by the DMTS $\mcalD$.  We say
that $\mcalD_1$ \emph{thoroughly refines} $\mcalD_2$, and write $\mcalD_1\DMTStr
\mcalD_2$, if $\impl{ \mcalD_1}\subseteq \impl{ \mcalD_2}$.

The below proposition, which follows directly from transitivity of modal
refinement, shows that modal refinement is \emph{sound} with respect to
thorough refinement; in the context of specification theories, this is
what one would expect.  It can be shown that modal refinement is also
\emph{complete} for \emph{deterministic}
DMTS~\cite{DBLP:journals/tcs/BenesKLS09}, but we will not need this
here.

\begin{proposition}\label{prop:mrtr}
  \label{soco.pr:mrvstr}
  For all DMTS $\mcalD_1$, $\mcalD_2$, $\mcalD_1\mr \mcalD_2$ implies $\mcalD_1\DMTStr
  \mcalD_2$. \noproof
\end{proposition}

\subsection{Acceptance automata}

An \emph{acceptance automaton} (\NAA) is a structure
$\mcalA=( S, S^0, \Tran)$, with $S\supseteq S^0$ finite sets of states
and initial states and $\Tran: S\to 2^{ 2^{ \Sigma\times S}}$ an
assignment of \emph{transition constraints}.
The intuition is that a transition constraint $\Tran( s)=\{ M_1,\dots,
M_n\}$ specifies a disjunction of $n$ choices $M_1,\dots, M_n$ as to
which transitions from $s$ have to be implemented.

An \NAA is an \emph{implementation} if $S^0=\{ s^0\}$ is a singleton and
it holds for all $s\in S$ that $\Tran( s)=\{ M\}\subseteq 2^{
  \Gamma\times S}$ is a singleton; hence \NAA implementations are
isomorphic to LTS.
Acceptance automata were first introduced
in~\cite{report/irisa/Raclet07}, based on the notion of acceptance trees
in~\cite{DBLP:journals/jacm/Hennessy85}; however, there they are
restricted to be \emph{deterministic}.  We employ no such restriction
here.

Let $\mcalA_1=( S_1, S^0_1, \Tran_1)$ and $\mcalA_2=( S_2, S^0_2, \Tran_2)$ be
\NAA.  A relation $R\subseteq S_1\times S_2$ is a \emph{modal refinement}
if it holds for all $( s_1, s_2)\in R$ and all $M_1\in \Tran_1( s_1)$
that there exists $M_2\in \Tran_2( s_2)$ such that
\begin{equation}
  \label{soco.eq:aaref}
  \begin{aligned}
    &\forall( a_1, t_1)\in M_1: \exists( a_2, t_2)\in M_2:
    a_1\labpre a_2,( t_1, t_2)\in R\,,\\
    &\forall( a_2, t_2)\in M_2: \exists( a_1, t_1)\in M_1:
    a_1\labpre a_2,( t_1, t_2)\in R\,.
  \end{aligned}
\end{equation}
The definition reduces to the one of~\cite{report/irisa/Raclet07} in
case labels are discrete.  We will write $M_1\labpre_R M_2$ if $M_1$,
$M_2$, $R$ satisfy~\eqref{soco.eq:aaref}.

In Chapter~\ref{ch:dmts} we have introduced translations between DMTS
and \NAA.  For a DMTS $\mcalD=( S, S^0, \omay, \omust)$ and $s\in S$,
let
$\Tran(s)=\{ M\subseteq \Sigma\times S\mid \forall (a,t)\in M:
s\DMTSmay{a} t, \forall s\DMTSmust{} N: N\cap M\ne \emptyset\}$ and
define the \NAA $\da( \mcalD)=( S, S^0, \Tran)$.
For an \NAA $\mcalA=( S, S^0, \Tran)$, define the DMTS $\ad( \mcalA)=( D, D^0,
\omay, \omust)$ by
\begin{align*}
  D &= \{ M\in \Tran( s)\mid s\in S\}\,, \\
  D^0 &= \{ M^0\in \Tran( s^0)\mid s^0\in S^0\}\,, \\
  \omust &= \big\{\big( M,\{( a, M')\mid M'\in \Tran(
  t)\}\big)\bigmid( a, t)\in M\big\}\,, \\
  \omay &= \{( M, a, M')\mid \exists M\DMTSmust{} N: ( a, M')\in N\}\,.
\end{align*}

Theorem~\ref{th:dmtsaa} is easily extended to our case of structured
labels:

\begin{theorem}
  \label{th:dmtsvsaa-bool}
  For all DMTS $\mcalD_1$, $\mcalD_2$ and \NAA $\mcalA_1$, $\mcalA_2$, $\mcalD_1\mr
  \mcalD_2$ iff $\da( \mcalD_1)\mr \da( \mcalD_2)$ and $\mcalA_1\mr \mcalA_2$ iff $\ad(
  \mcalA_1)\mr \ad( \mcalA_2)$.  \noproof
\end{theorem}

This structural equivalence will allow us to freely translate forth
and back between DMTS and \NAA in the rest of the paper.  Note,
however, that the state spaces of $\mcalA$ and $\ad( \mcalA)$ are not
the same; the one of $\ad( \mcalA)$ may be exponentially larger.
Proposition~\ref{pr:bfstodmtsblowup} shows that this blow-up is
unavoidable.

From a practical point of view, DMTS are a somewhat more useful
specification formalism than \NAA.  This is because they are usually more
compact and easily drawn and due to their close relation to the modal
$\nu$-calculus, see below.

\subsection{The Modal $\nu$-Calculus}

The modal $\nu$-calculus~\cite{DBLP:journals/deds/FeuilladeP07} is the
maximal-fixed point fragment of the modal
$\mu$-calculus~\cite{DBLP:journals/tcs/Kozen83}, \ie~the modal
$\mu$-calculus without negation and without the minimal fixed point
operator.  This is also sometimes called \emph{Henn\-essy-Milner logic
  with maximal fixed points} and represented using equation systems in
Hennessy-Milner logic with variables,
see~\cite{DBLP:journals/tcs/Larsen90, books/AcetoILS07}.  We will use
this representation below.  In Chapter~\ref{ch:dmts} we have
introduced translations between DMTS and the modal $\nu$-calculus,
showing that for discrete labels, these formalisms are
\emph{structurally equivalent}.

For a finite set $X$ of variables, let $\HML( X)$ be the set of
\emph{Hennessy-Milner formulae}, generated by the abstract syntax $\HML(
X)\ni \phi\Coloneqq \ltrue\mid \lfalse\mid x\mid \langle a\rangle \phi\mid[
a] \phi\mid \phi\land \phi\mid \phi\lor \phi$, for $a\in \Sigma$ and
$x\in X$.  A \emph{$\nu$-calculus expression} is a structure $\mcalN=( X,
X^0, \Delta)$, with $X^0\subseteq X$ sets of variables and $\Delta: X\to
\HML( X)$ a \emph{declaration}.

We recall the greatest fixed point semantics of $\nu$-calculus
expressions from \cite{DBLP:journals/tcs/Larsen90}, but extend it to
structured labels.  Let $( S, S^0, \omust)$ be an LTS, then an
\emph{assignment} is a mapping $\sigma: X\to 2^S$.  The set of
assignments forms a complete lattice with order
$\sigma_1\sqsubseteq \sigma_2$ iff
$\sigma_1( x)\subseteq \sigma_2( x)$ for all $x\in X$ and lowest upper
bound
$\big(\bigsqcup_{ i\in I} \sigma_i\big)( x)= \bigcup_{ i\in I}
\sigma_i( x)$.

The semantics of a formula in $\HML( X)$ is a function from
assignments to subsets of $S$ defined as follows: $\lsem \ltrue \sigma=
S$, $\lsem \lfalse \sigma= \emptyset$, $\lsem x \sigma= \sigma( x)$,
$\lsem{ \phi\land \psi} \sigma= \lsem \phi\sigma\cap \lsem \psi
\sigma$, $\lsem{ \phi\lor \psi} \sigma= \lsem \phi\sigma\cup \lsem
\psi \sigma$, and
\begin{align*}
  \lsem{\langle a\rangle \phi} \sigma &= \{ s\in S\mid \exists s\DMTSmust b
  t: b \in \impl{a}, t\in \lsem \phi \sigma\}, \\
  \lsem{[ a] \phi} \sigma &= \{ s\in S\mid \forall s\DMTSmust b t: b \in
  \impl{a} \limpl t\in \lsem \phi \sigma\}.
\end{align*} 
The semantics of a declaration $\Delta$ is then the assignment defined
by
\begin{equation*}
  \lsem \Delta= \bigsqcup\{ \sigma: X\to 2^S\mid \forall x\in X:
  \sigma( x)\subseteq \lsem{ \Delta( x)} \sigma\};
\end{equation*}
the greatest (pre)fixed point of~$\Delta$.

An LTS $\mcalI=( S, s^0, \omust)$ \emph{implements} (or models) the
expression $\mcalN$, denoted $\mcalI\models \mcalN$, if there is $x^0\in X^0$
such that $s^0\in \lsem \Delta( x^0)$.

In Chapter~\ref{ch:dmts} we have introduced another semantics for
$\nu$-calculus expressions, which is given by a notion of refinement,
like for DMTS and \NAA.  For this we recall the normal form for
$\nu$-calculus expressions:

\begin{lemma}
  \label{soco.le:hmlnormal}
  For any $\nu$-calculus expression $\mcalN_1=( X_1, X^0_1, \Delta_1)$,
  there exists another $\mcalN_2=( X_2, X^0_2, \Delta_2)$ with
  $\sem{ \mcalN_1}= \sem{ \mcalN_2}$ and such that for any $x\in X$,
  $\Delta_2( x)$ is of the form
  \begin{equation*}
    \Delta_2( x)= \bigland_{ i\in I}\big( \biglor_{ j\in
      J_i} \langle a_{ ij}\rangle  x_{ ij}\big)\land \bigland_{ a\in
      \Sigma}[ a] \big( \biglor_{ j\in J_a} y_{ a, j}\big)
  \end{equation*}
  for finite (possibly empty) index sets $I$, $J_i$, $J_a$ and all $x_{
    ij}, y_{ a, j}\in X_2$.  \noproof
\end{lemma}

As this is a type of \emph{conjunctive normal form}, it is clear that
translating a $\nu$-calculus expression into normal form may incur an
exponential blow-up.

We introduce some notation for $\nu$-calculus expressions in normal
form.  Let $\mcalN=( X, X^0, \Delta)$ be such an expression and
$x\in X$, with
\begin{equation*}
  \Delta( x)= \bigland_{ i\in I}\Big( \biglor_{ j\in J_i} \langle a_{
    ij}\rangle x_{ ij}\Big)\land \bigland_{ a\in \Sigma}[ a] \Big(
  \biglor_{ j\in J_a} y_{ a, j}\Big)  
\end{equation*}
as in the lemma.  Define
$\Diamond( x)=\{\{( a_{ ij}, x_{ ij})\mid j\in J_i\}\mid i\in I\}$
and, for each $a\in \Sigma$, $\Box^a( x)=\{ y_{ a, j}\mid j\in J_a\}$.
Intuitively, $\Diamond( x)$ collects all
$\langle a\rangle$-requirements from $x$, whereas $\Box^a( x)$
specifies the disjunction of $[ a]$-properties which must hold from
$x$.  Note that now,
\begin{equation}
  \label{eq:boxdiatodelta}
  \Delta( x)= \bigland_{ N\in
    \Diamond(x)} \Big( \biglor_{( a, y)\in N} \langle a\rangle y\Big)
  \land \bigland_{ a\in \Sigma}[ a]\Big( \biglor_{ y\in \Box^a( x)}
  y\Big)\,.
\end{equation}

Let $\mcalN_1=( X_1, X^0_1, \Delta_1)$, $\mcalN_2=( X_2, X^0_2, \Delta_2)$ be
$\nu$-calculus expressions in normal form and $R\subseteq X_1\times
X_2$.  The relation $R$ is a \emph{modal refinement} if it holds for all
$( x_1, x_2)\in R$ that
\begin{itemize}
\item for all $a_1\in \Sigma$ and $y_1\in \Box_1^{ a_1}( x_1)$ there is
  $a_2 \in \Sigma$ and $y_2\in \Box_2^{ a_2}( x_2)$ with $a_1\labpre
  a_2$ and $( y_1, y_2)\in R$, and
\item for all $N_2\in \Diamond_2( x_2)$ there is $N_1\in
  \Diamond_1(x_1)$ such that for all $( a_1, y_1)\in N_1$ there exists
  $( a_2,y_2)\in N_2$ with $a_1 \labpre a_2$ and $( y_1, y_2)\in R$.
\end{itemize}

We say that a $\nu$-calculus expression $( X, X^0, \Delta)$ in normal
form is an \emph{implementation} if $X^0=\{ x^0\}$ is a singleton,
$\Diamond( x)=\{\{( a, y)\}\mid y\in \Box^a( x), a\in \Sigma\}$ and
$\Box^a(x)= \emptyset$ for all $a\notin \Gamma$, for all $x\in X$.  

We can translate a LTS $( S, S^0, \omust)$ to a $\nu$-calculus
expression $( S, S^0, \Delta)$ in normal form by setting $\Diamond(
s)=\{\{( a, t)\}\mid s\DMTSmust a t\}$ and $\Box^a( s)=\{ t\mid s\DMTSmust a
t\}$ for all $s\in S$, $a\in \Sigma$.  This defines a bijection
between LTS and $\nu$-calculus implementations, hence, like for DMTS
and \NAA, an embedding of LTS into the modal $\nu$-calculus.

We have shown in Chapter~\ref{ch:dmts} that for discrete labels, the
refinement semantics and the fixed point semantics of the modal
$\nu$-calculus agree; the proof can easily be extended to our case of
structured labels:

\begin{theorem}
  For any LTS $\mcalI$ and any $\nu$-calculus expression $\mcalN$ in normal
  form, $\mcalI\models \mcalN$ iff $\mcalI\mr \mcalN$.  \noproof
\end{theorem}

For a DMTS $\mcalD=( S, S^0, \omay, \omust)$ and all $s\in S$, let
$\Diamond(s)=\{ N \mid s \DMTSmust{} N\}$ and, for each $a\in \Sigma$,
$\Box^a(s)=\{ t \mid s \DMTSmay{a} t\}$.  Define the (normal-form)
$\nu$-calculus expression $\ddh( \mcalD)=( S, S^0, \Delta)$, with $\Delta$
given as in~\eqref{eq:boxdiatodelta}.
For a $\nu$-calculus expression $\mcalN=( X, X^0, \Delta)$ in normal form,
let $\omay=\{( x, a, y)\in X\times \Sigma\times X\mid y\in \Box^a(
x)\}$, $\omust=\{( x, N)\mid x\in X, N\in \Diamond( x)\}$ and define the
DMTS $\hd( \mcalN)=( X, X^0, \omay, \omust)$.  Given that these
translations are entirely syntactic, the following theorem is not a
surprise:

\begin{theorem}
  \label{th:dmtsvsnu-bool}
  For DMTS $\mcalD_1$, $\mcalD_2$ and $\nu$-calculus expressions $\mcalN_1$,
  $\mcalN_2$, $\mcalD_1\mr \mcalD_2$ iff $\ddh( \mcalD_1)\mr \ddh( \mcalD_2)$ and
  $\mcalN_1\mr \mcalN_2$ iff $\hd( \mcalN_1)\mr \hd( \mcalN_2)$.  \noproof
\end{theorem}

\section{Specification theory}
\label{se:specth}

Structural specifications typically come equipped with operations
which permit \emph{compositional reasoning}, \viz conjunction,
structural composition, and quotient,
\cf~\cite{DBLP:conf/fase/BauerDHLLNW12}.  On \emph{deterministic} MTS,
these operations can be given easily using simple structural
operational rules (for such semantics of weighted systems, see for
instance~\cite{DBLP:journals/iandc/KlinS13}).  For non-deterministic
specifications this is significantly harder;
in~\cite{DBLP:conf/concur/BenesDFKL13} it is shown that DMTS and \NAA
permit these operations and, additionally but trivially, disjunction.
Here we show how to extend these operations on non-deterministic
systems to our setting with structured labels.

We remark that structural composition and quotient operators are
well-known from some logics, such as,
\eg~linear~\cite{DBLP:journals/tcs/Girard87} or spatial
logic~\cite{DBLP:journals/iandc/CairesC03}, see
also~\cite{DBLP:conf/icalp/CardelliLM11} for a stochastic extension.
However, whereas these operators are part of the formal syntax in
those logics, for us they are simply operations on logical expressions
(or DMTS, or \NAA).

Given the equivalence of DMTS, \NAA and the modal $\nu$-calculus
exposed in the previous section, we will often state properties for
all three types of specifications at the same time, letting $\mcalS$
stand for any of the three types.  For definitions and proofs, we are
free to use the type of specification which is most well suited for
the context; we will use DMTS for the logical operations
(Section~\ref{se:disjconj}) and \NAA for the structural operations
(Sections~\ref{se:comp} and~\ref{se:quot}).

\subsection{Disjunction and conjunction}
\label{se:disjconj}

Disjunction of specifications is easily defined, as we allow for
multiple initial states.  For two DMTS $\mcalD_1=( S_1, S_1^0, \omay_1,
\omust_1)$ and $\mcalD_2=( S_2, S_2^0, \omay_2, \omust_2)$, we can hence
define $\mcalD_1\lor \mcalD_2= (S_1 \cup S_2, S^0_1 \cup S^0_2, \omay_1 \cup
\omay_2, \omust_1\cup\omust_2)$ (with all unions disjoint).

For conjunction, we let $\mcalD_1\land \mcalD_2=( S_1\times S_2, S_1^0\times
S_2^0, \omay, \omust)$, with
\begin{itemize}
\item $( s_1, s_2)\DMTSmay{ a_1\oland a_2}( t_1, t_2)$ whenever $s_1\DMTSmay{
    a_1}_1 t_1$, $s_2\DMTSmay{ a_2}_2 t_2$ and $a_1\oland a_2$ is defined,
\item for all $s_1\DMTSmust{} N_1$, $( s_1, s_2)\DMTSmust{} \{( a_1\oland a_2,(
  t_1, t_2))\mid( a_1, t_1)\in N_1, \smash{s_2\DMTSmay{ a_2}_2 t_2},$
  $a_1\oland a_2\text{ defined}\}$,
\item for all $s_2\DMTSmust{} N_2$, $( s_1, s_2)\DMTSmust{} \{( a_1\oland a_2,(
  t_1, t_2))\mid( a_2, t_2)\in N_2, \smash{s_1\DMTSmay{ a_1}_1 t_1},$
  $a_1\oland a_2\text{ defined}\}$.
\end{itemize}

The following theorem generalizes Theorem~\ref{th:condis} to
structured labels.  Also its proof is a generalization, but our
structured labels do introduce some extra difficulties.

\begin{theorem}
  \label{soco.th:condis}
  For all specifications $\mcalS_1$, $\mcalS_2$, $\mcalS_3$,
  \begin{itemize}
  \item $\mcalS_1\lor \mcalS_2\mr \mcalS_3$ iff $\mcalS_1\mr \mcalS_3$ and $\mcalS_2\mr
    \mcalS_3$,
  \item $\mcalS_1\mr \mcalS_2\land \mcalS_3$ iff $\mcalS_1\mr \mcalS_2$ and $\mcalS_1\mr
    \mcalS_3$,
  \item $\sem{ \mcalS_1\lor \mcalS_2}= \sem{ \mcalS_1}\cup \sem{ \mcalS_2}$, and
    $\sem{ \mcalS_1\land \mcalS_2}= \sem{ \mcalS_1}\cap \sem{ \mcalS_2}$.
  \end{itemize}
\end{theorem}

\begin{proof}
  The proof that $\mcalS_1\lor \mcalS_2\mr \mcalS_3$ iff $\mcalS_1\mr \mcalS_3$ and
  $\mcalS_2\mr \mcalS_3$ is trivial: any modal refinement $R\subseteq( S_1\cup
  S_2)\times S_3$ splits into two refinements $R_1\subseteq S_1\times
  S_3$, $R_2\subseteq S_2\times S_3$ and vice versa.

  For the proof of the second claim, which we show for DMTS, we prove
  the back direction first.  Let $R_2\subseteq S_1\times S_2$,
  $R_3\subseteq S_1\times S_3$ be initialized (DMTS) modal refinements
  which witness $\mcalS_1\mr \mcalS_2$ and $\mcalS_1\mr \mcalS_3$, respectively.
  Define $R=\{( s_1,( s_2, s_3))\mid( s_1, s_2)\in R_2,( s_1, s_3)\in
  R_3\}\subseteq S_1\times( S_2\times S_3)$, then $R$ is initialized.

  Now let $( s_1,( s_2, s_3))\in R$, then $( s_1, s_2)\in R_2$ and $(
  s_1, s_3)\in R_3$.  Assume that $s_1\DMTSmay{ a_1}_1 t_1$, then by
  $\mcalS_1\mr \mcalS_2$, we have $s_2\DMTSmay{ a_2}_2 t_2$ with $a_1\labpre a_2$
  and $( t_1, t_2)\in R_2$.  Similarly, by $\mcalS_1\mr \mcalS_3$, we have
  $s_3\DMTSmay{ a_3}_3 t_3$ with $a_1\labpre a_3$ and $( t_1, t_3)\in R_3$.
  But then also $a_1\labpre a_2\oland a_3$ and $( t_1,( t_2, t_3))\in
  R$, and $( s_2, s_3)\DMTSmay{ a_2\oland a_3}( t_2, t_3)$ by definition.

  Assume that $( s_2, s_3)\DMTSmust{} N$.  Without loss of generality we can
  assume that there is $s_2\DMTSmust{}_2 N_2$ such that $N=\{( a_2\oland
  a_3,( t_2, t_3))\mid( a_2, t_2)\in N_2, s_3\DMTSmay{ a_3}_3 t_3\}$.  By
  $S_1\mr S_2$, we have $s_1\DMTSmust{}_1 N_1$ such that $\forall( a_1,
  t_1)\in N_1: \exists( a_2, t_2)\in N_2: a_1\labpre a_2,( t_1, t_2)\in
  R_2$.

  Let $( a_1, t_1)\in N_1$, then also $s_1\DMTSmay{ a_1}_1 t_1$, so by
  $S_1\mr S_3$, there is $s_3\DMTSmay{ a_3}_3 t_3$ with $a_1\labpre a_3$ and
  $( t_1, t_3)\in R_3$.  By the above, we also have $( a_2, t_2)\in N_2$
  such that $a_1\labpre a_2$ and $( t_1, t_2)\in R_2$, but then $(
  a_2\oland a_3,( t_2, t_3))\in N$, $a_1\labpre a_2\land a_3$, and $(
  t_1,( t_2, t_3))\in R$.

  For the other direction of the second claim, let $R\subseteq
  S_1\times( S_2\times S_3)$ be an initialized (DMTS) modal refinement
  which witnesses $\mcalS_1\mr \mcalS_2\land \mcalS_3$.  We show that $S_1\mr
  S_2$, the proof of $S_1\mr S_3$ being entirely analogous.  Define
  $R_2=\{( s_1, s_2)\mid \exists s_3\in S_3:( s_1,( s_2, s_3))\in
  R\}\subseteq S_1\times S_2$, then $R_2$ is initialized.

  Let $( s_1, s_2)\in R_2$, then we must have $s_3\in S_3$ such that $(
  s_1,( s_2, s_3))\in R$.  Assume that $s_1\DMTSmay{ a_1}_1 t_1$, then also
  $( s_2, s_3)\DMTSmay{ a}( t_2, t_3)$ for some $a$ with $a_1\labpre a$ and
  $( t_1,( t_2, t_3))\in R$.  By construction we have $s_2\DMTSmay{ a_2}_2
  t_2$ and $s_3\DMTSmay{ a_3}_3 t_3$ such that $a= a_2\oland a_3$, but then
  $a_1\labpre a_2\oland a_3\labpre a_2$ and $( t_1, t_2)\in R_2$.

  Assume that $s_2\DMTSmust{}_2 N_2$, then by construction we have $( s_2,
  s_3)\DMTSmust{} N=\{( a_2\oland a_3,( t_2, t_3))\mid( a_2, t_2)\in N_2,
  s_3\DMTSmay{ a_3}_3 t_3\}$.  By $\mcalS_1\mr \mcalS_2\land \mcalS_3$, there is
  $s_1\DMTSmust{}_1 N_1$ such that $\forall( a_1, t_1)\in N_1: \exists( a,(
  t_2, t_3))\in N: a_1\labpre a,( t_1,( t_2, t_3))\in R$.

  Let $( a_1, t_1)\in N_1$, then we have $( a,( t_2, t_3))\in N$ for
  which $a_1\labpre a$ and $( t_1,( t_2, t_3))\in R$.  By construction
  of $N$, this implies that there are $( a_2, t_2)\in N_2$ and $s_3\DMTSmay{
    a_3}_3 t_3$ such that $a= a_2\oland a_3$, but then $a_1\labpre
  a_2\oland a_3\labpre a_2$ and $( t_1, t_2)\in R$.

  As to the last claims of the theorem, $\sem{ \mcalS_1\land \mcalS_2}= \sem{
    \mcalS_1}\cap \sem{ \mcalS_2}$ is clear from what we just proved: for all
  implementations $\mcalI$, $\mcalI\mr \mcalS_1\land \mcalS_2$ iff $\mcalI\mr \mcalS_1$
  and $\mcalI\mr \mcalS_2$.  For the other part, it is clear by construction
  that for any implementation $\mcalI$, any witness $R$ for $\mcalI\mr \mcalS_1$
  is also a witness for $\mcalI\mr \mcalS_1\lor \mcalS_2$, and similarly for
  $\mcalS_2$, hence $\sem{ \mcalS_1}\cup \sem{ \mcalS_2}\subseteq \sem{ \mcalS_1\lor
    \mcalS_2}$.

  To show that also $\sem{ \mcalS_1}\cup \sem{ \mcalS_2}\supseteq \sem{
    \mcalS_1\lor \mcalS_2}$, we note that an initialized refinement $R$
  witnessing $\mcalI\mr \mcalS_1\lor \mcalS_2$ must relate the initial state of
  $\mcalI$ either to an initial state of $\mcalS_1$ or to an initial state of
  $\mcalS_2$.  In the first case, and by disjointness, $R$ witnesses
  $\mcalI\mr \mcalS_1$, in the second, $\mcalI\mr \mcalS_2$.  \qed
\end{proof}

With bottom and top elements given by $\bot=( \emptyset, \emptyset,
\emptyset)$ and $\top=(\{ s\},\{ s\},$ $\Tran_\top)$ with $\Tran_\top( s)=
2^{ 2^{ \Sigma\times\{ s\}}}$, our classes of specifications form
\emph{bounded distributive lattices} up to $\mreq$.

\subsection{Structural composition}
\label{se:comp}

For \NAA $\mcalA_1=( S_1, S_1^0, \Tran_1)$, $\mcalA_2=( S_2, S_2^0, \Tran_2)$,
their \emph{structural composition} is $\mcalA_1\| \mcalA_2=( S_1\times S_2,
S_1^0\times S_2^0, \Tran)$, with $\Tran(( s_1, s_2))=\{ M_1\obar M_2\mid
M_1\in \Tran_1( s_1), M_2\in \Tran_2( s_2)\}$ for all $s_1\in S_1$,
$s_2\in S_2$, where $M_1\obar M_2=\{( a_1\obar a_2,( t_1, t_2))\mid(
a_1, t_1)\in M_1,$ $( a_2, t_2)\in M_2, a_1\obar a_2\text{ defined}\}$.

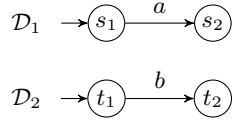
\begin{figure}
  \centering
  \begin{tikzpicture}[->, >=stealth', font=\footnotesize,
    state/.style={shape=circle, draw, initial text=,inner
      sep=.5mm,minimum size=2mm}, yscale=.5, xscale=.7]
    \node at (-1.5,0) {$\mcalD_1$};
    \node at (-1.5,-2) {$\mcalD_2$};
    \node[state, initial] (1) at (0,0) {$s_1$};
    \node[state] (2) at (2,0) {$s_2$};
    \node[state, initial] (3) at (0,-2) {$t_1$};
    \node[state] (4) at (2,-2) {$t_2$};
    \path (1) edge node[above]{$a$} (2);
    \path (3) edge node[above]{$b$} (4);
  \end{tikzpicture}
  \caption{%
    \label{fi:conjvscomp}
    Two simple DMTS}
\end{figure}

Remark a subtle difference between conjunction and structural
composition, which we expose for discrete labels and CSP-style
composition: for the DMTS $\mcalD_1$, $\mcalD_2$ shown in
Figure~\ref{fi:conjvscomp}, both $\mcalD_1\land \mcalD_2$ and $\mcalD_1\| \mcalD_2$
have only one state, but $\Tran( s_1\land t_1)= \emptyset$ and $\Tran(
s_1\| t_1)=\{ \emptyset\}$, so that $\mcalD_1\land \mcalD_2$ is inconsistent,
whereas $\mcalD_1\| \mcalD_2$ is not.

This definition extends the structural composition defined for modal
transition systems, with structured labels,
in~\cite{DBLP:journals/acta/FahrenbergL14}.  For DMTS specifications
(and hence also for $\nu$-calculus expressions), the back translation
from \NAA to DMTS entails an exponential explosion.

\begin{theorem}
  \label{th:comp}
  Up to $\mreq$, the operator $\|$ is associative, commutative and
  monotone.
\end{theorem}

\begin{proof}
  Associativity and commutativity are clear by associativity and
  commutativity of $\obar$.  Monotonicity is equivalent to the assertion
  that (up to $\mreq$) $\|$ distributes over the least upper bound
  $\lor$; one easily sees that for all specifications $\mcalS_1$, $\mcalS_2$,
  $\mcalS_3$, the identity is a two-sided modal refinement $\mcalS_1\|(
  \mcalS_2\lor \mcalS_3)\mreq \mcalS_1\| \mcalS_2\lor \mcalS_1\| \mcalS_3$.  \qed
\end{proof}

\begin{corollary}[Independent implementability]
  \label{co:indimp}
  For all specifications $\mcalS_1$, $\mcalS_2$, $\mcalS_3$, $\mcalS_4$, $\mcalS_1\mr
  \mcalS_3$ and $\mcalS_2\mr \mcalS_4$ imply $\mcalS_1\| \mcalS_2\mr \mcalS_3\| \mcalS_4$.
  \noproof
\end{corollary}

\subsection{Quotient}
\label{se:quot}

Because of non-determinism, we have to use a power set construction for
the quotient, as opposed to conjunction and structural composition where
product is sufficient. For \NAA $\mcalA_3=( S_3, S_3^0, \Tran_3)$, $\mcalA_1=(
S_1, S_1^0, \Tran_1)$, the quotient is $\mcalA_3 / \mcalA_1 = ( S,\{ s^0\},
\Tran)$, with $S= 2^{S_3 \times S_1}$ and $s^0=\{( s_3^0, s_1^0) \mid
s_3^0\in S_3^0, s_1^0\in S_1^0\}$.  States in $S$ will be written $\{
s_3^1 / s_1^1,\dots, s_3^n / s_1^n)\}$.
Intuitively, this denotes that such state when composed with $s_1^i$
conforms to $s_3^i$ for each $i$; we call this \emph{consistency} here.
 
We now define $\Tran$.  First, $\Tran( \emptyset)= 2^{ \Sigma\times\{
  \emptyset\}}$, so $\emptyset$ is universal.  For any other state $s=\{
s_3^1 / s_1^1,\dots, s_3^n / s_1^n\} \in S$, its set of
\emph{permissible labels} is defined by
\begin{multline*}
  \PermL(s) = \big\{ a_2 \in \Sigma \bigmid \forall i =1,\dotsc,n : \forall
  (a_1,t_1) \in\in \Tran_1(s_1^i) : \\
  \exists (a_3,t_3) \in\in \Tran_3(s_3^i) : a_1\obar a_2\labpre a_3
  \big\}\,,
\end{multline*}
that is, a label is permissible iff it cannot violate consistency.  Here
we use the notation $x\in \in z$ as a shortcut for $\exists y: x\in y\in
z$.

Now for each $a \in \PermL(s)$ and each $i \in \{1,\dots,n\}$, let
$\{t_1 \in S_1 \mid (a,t_1) \in\in
\Tran_1(t_1^i)\}=\{t_1^{i,1},\dots,t_1^{i,m_i}\}$ be an enumeration of
all the possible states in $S_1$ after an $a$-transition.  Then we
define the set of all sets of possible assignments of next-$a$ states
from $s_3^i$ to next-$a$ states from $s_1^i$:
\begin{multline*}
  \postra[a]{s} = \big\{ \{ (t_3^{i,j},t_1^{i,j}) \mid i=1,\dots,n,
  j=1,\dots,m_i\} \\ \bigmid \forall i: \forall j : (a,t_3^{i,j}) \in\in
  \Tran_3(s_3^i) \big\}
\end{multline*}
These are all possible next-state assignments which preserve
consistency.  Now let $\postra{s} = \bigcup_{a \in
  \PermL(s)}\postra[a]{s}$ and define
\begin{multline*}
  \Tran(s) = \big\{ M \subseteq \postra{s} \bigmid \forall
  i=1,\dots,n:
  \\
  \forall M_1 \in \Tran_1(s_1^i): \exists M_3\in \Tran_3( s_3^i): M
  \triangleright M_1 \labpre_R M_3 \big\}\,,
\end{multline*}
where $M \triangleright M_1 = \{ (a_1 \obar a, t_3^i) \mid (a,\{ t_3^1
/ t_1^1,\dots, t_3^k / t_1^k)\}) \in M, (a_1,t_1^i) \in M_1\}$, to
guarantee consistency no matter which element of $\Tran_1(s_1^i)$, $s$
is composed with.

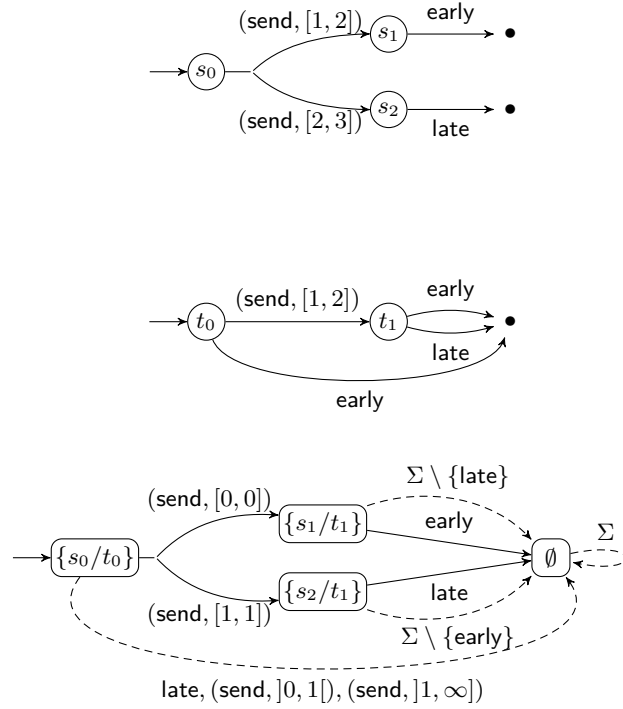
\begin{figure}
  \centering
  \begin{tikzpicture}[->, >=stealth', font=\footnotesize, xscale=2,
    yscale=1, state/.style={shape=circle, draw, initial text=,inner
      sep=.5mm,minimum size=3mm,initial distance=1.5ex}]
    \begin{scope}
      \node[state,initial] (s) at (0,0) {$s_0$};
      \node[inner sep=0,outer sep=0] (ss) at (0.3,0){};
      \node[state] (s1) at (1.2,0.5) {$s_1$};
      \path (s) edge[-] (ss);
      \path (ss)	edge [->,bend left] node[above]{$(\send,[1,2])$}
      (s1);
      \node[state] (s2) at (1.2,-0.5) {$s_2$};
      \path (ss)	edge [->,bend right]
      node[below]{$(\send,[2,3])$}	(s2);
      \node (end) at (2,0.5) {$\bullet$};
      \path (s1) edge [->] node[above]{$\early$}	(end);
      \node (end2) at (2,-0.5) {$\bullet$};
      \path (s2) edge [->] node[below]{$\late$}	(end2);
    \end{scope}
    \begin{scope}[yshift=-20ex]
      \node[state,initial] (s) at (0,0) {$t_0$};
      \node[state] (s1) at (1.2,0) {$t_1$};
      \path (s)	edge [->] node[above]{$(\send,[1,2])$}	(s1);
      \node (end) at (2,0) {$\bullet$};
      \path (s1)	edge [->,bend left] node[above]{$\early$}	(end);
      \path (s1)	edge [->, bend right] node[below]{$\late$}	(end);
      \path (s)  edge [->, bend right=80pt] node[below]{$\early$}  (end);
    \end{scope}
  \end{tikzpicture} 
  
  \bigskip

  \begin{tikzpicture}[->, >=stealth', font=\footnotesize, xscale=2,
    yscale=1.5, state/.style={shape=circle,rectangle,rounded
      corners,draw,initial text=,inner sep=.5mm,minimum size=5mm,initial
      distance=1.5ex}]
    \begin{scope}
      \node[state,initial] (s) at (0,0) {$\{s_0/t_0\}$};
      \node[inner sep=0,outer sep=0] (ss) at (0.4,0){};
      \node[state] (s1) at (1.5,0.3) {$\{s_1/t_1\}$};
      \node[state] (s11) at (1.5,-0.3) {$\{s_2/t_1\}$};
      \node[state] (s2) at (3,0) {$\emptyset$};
      \path (s2) edge[loop right, densely dashed]
      node[above,pos=.2]{$\Sigma$} (s2);
      \path (s) edge[-](ss);
      \path (ss) edge [->,bend left] node[above]{$(\send,[0,0])$}  (s1);
      \path (ss) edge [->,bend right] node[below]{$(\send,[1,1])$}  (s11);
      \path (s) edge [->,out=-90,in=-90,densely dashed,bend right=120pt]
      node[below]{$\late,(\send,\mathopen] 0,
        1\mathclose[),( \send,\mathopen] 1,\infty])$} (s2);
      \path (s1)  edge [->,densely dashed, bend left=40pt]
      node[above]{$\Sigma\setminus\{\late\}$} (s2);
      \path (s1)  edge [->] node[above]{$\early$} (s2);
      \path (s11)  edge [->,densely dashed, bend right=40pt]
      node[below]{$\Sigma\setminus\{\early\}$} (s2);
      \path (s11)  edge [->] node[below]{$\late$} (s2);
    \end{scope}
  \end{tikzpicture} 
  \caption{%
    \label{fi:quotientex}
    Two DMTS (top and center) and their quotient (bottom)}
\end{figure}

\begin{example}
  Figure~\ref{fi:quotientex} shows two simple specifications and their
  quotient under $\obarplus$, \ie~using addition of intervals for
  label synchronization (see Example~\ref{ex:composition}).  During
  the construction and the translation back to DMTS, many states were
  eliminated as they were inconsistent (their $\Tran$-set was
  empty). For instance, there is no may transition to state
  $\{s_2/t_2\}$, because when it is composed with $t_2$ there is no
  guarantee of a $\late$-transition, hence no guarantee to refine
  $s_2$.

  Note that in order to have a finite representation of the quotient,
  we have to extend the label set to allow intervals which are not
  closed; for instance, the may-transition
  $( \send,\mathopen] 1, \infty])$ from $\{ s_0/ t_0\}$ to $\emptyset$
  comprises the fact that $\postra[ a]{\{ s_0/ t_0\}}= \emptyset$ for
  all $a=( \send,[ x, \infty])$ with $x> 1$.  This can be formalized
  by introducing a (partial) \emph{label quotient} operator
  $\mathord{\oslash}: \Sigma\times \Sigma\parto \Sigma$ which is
  adjoint to label synchronization $\obar$, see
  Chapter~\ref{ch:wm2}. \qed
\end{example}

\begin{theorem}
  \label{th:quotient-bool}
  For all specifications $\mcalS_1$, $\mcalS_2$, $\mcalS_3$, $\mcalS_1\| \mcalS_2\mr
  \mcalS_3$ iff $\mcalS_2\mr \mcalS_3/ \mcalS_1$.
\end{theorem}

\begin{proof}
  We show the proof for \NAA; for DMTS and $\nu$-calculus expressions it
  will follow through the translations.  Let $\mcalA_1=( S_1, S^0_1,
  \Tran_1)$, $\mcalA_2=( S_2, S^0_2, \Tran_2)$, $\mcalA_3=( S_3, S^0_3,
  \Tran_3)$; we show that $\mcalA_1\| \mcalA_2\mr \mcalA_3$ iff $\mcalA_2\mr
  \mcalA_3/ \mcalA_1$.

  We assume that the elements of $\Tran_1( s_1)$ are pairwise disjoint
  for each $s_1\in S_1$; this can be achieved by, if necessary,
  splitting states.

  First we note that by construction, $s\supseteq t$ implies $s\mr t$
  for all $s, t\in S$.

  Assume that $\mcalA_2\mr \mcalA_3/ \mcalA_1$ and let $R=\{( s_2, s_3/
  s_1)\mid s_2\mr s_3/ s_1\}$ be the witnessing refinement relation.
  Let $R'=\{( s_1\| s_2, s_3)\mid( s_2, s_3/ s_1)\in R\}$ (for
  readability, we abuse notation here and write $( s_1\| s_2, s_3)$
  instead of $( s_1, s_2, s_3)$); we show that $R'$ is a witness for
  $\mcalA_1\| \mcalA_2\mr \mcalA_3$.

  Let $( s_1\| s_2, s_3)\in R'$ and $M_\|\in \Tran_\|( s_1\| s_2)$.  Then
  $M_\|= M_1\| M_2$ with $M_1\in \Tran_1( s_1)$ and $M_2\in \Tran_2(
  s_2)$.  As $s_2\mr s_3/ s_1$, we can pair $M_2$ with a set $M_/\in
  \Tran_/( s_3/ s_1)$ such that $M_2\labpre_R M_/$.

  Let $M_3= M_/\triangleright M_1$.  We show that $M_\|\labpre_{ R'}
  M_3$:
  \begin{itemize}
  \item Let $( a, t_1\| t_2)\in M_\|$, then there are $a_1, a_2\in
    \Sigma$ with $a= a_1\obar a_2$ and $( a_1, t_1)\in M_1$, $( a_2,
    t_2)\in M_2$.  By $M_2\labpre_R M_/$, there is $( a_2', t)\in
    M_/$ such that $a_2\labpre a_2'$ and $t_2\mr t$.  Note that
    $a_3= a_1\obar a_2'$ is defined and $a\labpre a_3$.  Write $t=\{
    t_3^1/ t_1^1,\dots, t_3^n/ t_1^n\}$.  By construction, there
    is an index $i$ for which $t_1^i= t_1$, hence $( a_3, t_3^i)\in
    M_3$.  Also, $t\supseteq\{ t_3^i/ t_1^i\}$, hence $t_2\mr
    t_3^i/ t_1^i$ and consequently $( t_1\| t_2, t_3)\in R'$.
  \item Let $( a_3, t_3)\in M_3$, then there are $( a_2', t)\in M_/$
    and $( a_1, t_1)\in M_1$ such that $a_3= a_1\obar a_2'$ and
    $t_3/ t_1\in t$.  By $M_2\labpre_R M_/$, there is $( a_2,
    t_2)\in M_2$ for which $a_2\labpre a_2'$ and $t_2\mr t$.  Note
    that $a= a_1\obar a_2$ is defined and $a\labpre a_3$.  Thus $( a,
    t_1\| t_2)\in M_\|$, and by $t\supseteq\{ t_3/ t_1\}$, $t_2\mr
    t_3/ t_1$.
  \end{itemize}

  Assume, for the other direction of the proof, that $\mcalA_1\| \mcalA_2\mr
  \mcalA_3$ and let $R=\{( s_1\| s_2, s_3)\mid s_1\| s_2\mr s_2\}$ (again
  abusing notation) be the witnessing refinement relation.  Define
  $R'\subseteq S_2\times 2^{ S_3\times S_1}$ by
  \begin{equation*}
    R= \big\{( s_2,\{ s_3^1/ s_1^1,\dotsc, s_3^n/ s_1^n\})
    \bigmid \forall i= 1,\dotsc, n:( s_1^i\| s_2, s_3^i)\in R\}\,;
  \end{equation*}
  we show that $R'$ is a witness for $\mcalA_2\mr \mcalA_3/ \mcalA_1$.  Let $(
  s_2, s)\in R'$, with $s=\{ s_3^1/ s_1^1,\dotsc, s_3^n/ s_1^n\}$,
  and $M_2\in \Tran_2( s_2)$.

  For every $i= 1,\dotsc, n$, write the set $\Tran_1( s_1^i)=\{ M_1^{ i,
    1},\dotsc, M_1^{ i, m_i}\}$.  By assumption, $M_1^{ i, j_1}\cap
  M_1^{ i, j_2}= \emptyset$ for $j_1\ne j_2$, hence every $( a_1,
  t_1)\in \in \Tran_1( s_1^i)$ is contained in a unique $M_1^{ i,
    \delta_i( a_1, t_1)}\in \Tran_1( s_1^i)$.

  For every $j= 1,\dotsc, m_i$, let $M^{ i, j}= M_1^{ i, j}\| M_2\in
  \Tran_\|( s_1^i\| s_2)$.  By $s_1^i\| s_2\mr s_3^i$, we have $M_3^{
    i, j}\in \Tran_3( s_3^i)$ such that $M^{ i, j}\labpre_R M_3^{ i, j}$.

  Now define
  \begin{multline}
    M= \big\{ ( a_2, t)\bigmid \exists( a_2, t_2)\in M_2: \forall t_3/
    t_1\in t: \exists i, a_1, a_3:
    (a_1, t_1)\in \in \Tran_1( s_1^i), \\
    ( a_3, t_3)\in M_3^{ i, \delta_i( a_1, t_1)}, a_1\obar a_2\labpre
    a_3, t_1\| t_2\mr t_3\}\,.
    \label{eq:quotproof.M}
  \end{multline}
  We need to show that $M\in \Tran_/( s)$.

  Let $i\in\{ 1,\dots, n\}$ and $M_1^{ i, j}\in \Tran_1( s_1^i)$; we
  claim that $M\triangleright M_1^{ i, j}\labpre_{ R'} M_3^{ i, j}$.
  Let $( a_3, t_3)\in M\triangleright M_1^{ i, j}$, then $a_3=
  a_1\obar a_2$ for some $a_1, a_2$ such that $t_3/ t_1\in t$, $(
  a_1, t_1)\in M_1^{ i, j}$ and $( a_2, t)\in M$.  By disjointness,
  $j= \delta_i( a_1, t_1)$, hence by definition of $M$, $( a_3,
  t_3)\in M_3^{ i, j}$ as was to be shown.

  For the reverse inclusion, let $( a_3, t_3)\in M_3^{ i, j}$.  By
  $M^{ i, j}\labpre_R M_3^{ i, j}$ and definition of $M^{ i, j}$,
  there are $( a_1, t_1)\in M_1^{ i, j}$ and $( a_2, t_2)\in M_2$ for
  which $a_1\obar a_2\labpre a_3$ and $t_1\| t_2\mr t_3$.  Thus $j=
  \delta_i( a_1, t_1)$, so that there must be $( a_2, t)\in M$ for
  which $t_3/ t_1\in t$, but then also $( a_1\obar a_2, t_3)\in
  M\triangleright M_1^{ i, j}$.

  We show that $M_2\labpre_{ R'} M$.
  \begin{itemize}
  \item Let $( a_2, t_2)\in M_2$.  For every $i= 1,\dotsc, n$ and
    every $( a_1, t_1)\in \in \Tran_1( t_1^i)$, we can use $M^{ i,
      j}\labpre_R M_3^{ i, j}$ and choose an element $( \eta_i( a_1,
    t_1), \tau_i( a_1, t_1))\in M_3^{ i, \delta_i( a_1, t_1)}$ for
    which $t_1\| t_2\mr \tau_i(a_1, t_1)$ and $a_1\obar a_2\labpre
    \eta_i( a_1, t_1)$.
    Let $t=\{ \tau_i( a_1, t_1)/ t_1\mid i= 1,\dotsc, n,( a_1, t_1)\in
    \in \Tran_1( t_1^i)\}$, then $( a_2, t)\in M$ and $( t_2, t)\in R'$.
  \item Let $( a_2, t)\in M$, then we have $( a_2, t_2)\in M_2$
    satisfying the conditions in~\eqref{eq:quotproof.M}.  Hence $t_1\|
    t_2\mr t_3$ for all $t_3/ t_1\in t$, so that $( t_2, t)\in
    R'$. \qed
  \end{itemize}
\end{proof}

\section{Robust Specification Theories}
\label{se:quant}

We proceed to lift the results of the previous sections to a
\emph{quantitative} setting, where the Boolean notions of modal and
thorough refinement are replaced by refinement \emph{distances}.  We
have shown in previous chapters that a good setting for quantitative
analysis is given by the one of \emph{recursively specified trace
  distances} on an abstract complete lattice $\LL$.  In order to
extend this to specification theories, we enrich $\LL$ with an
addition like in Chapter~\ref{ch:wm2} and require it to be a
(commutative) quantale, see below.

Denote by $\Sigma^\infty= \Sigma^*\cup \Sigma^\omega$ the set of finite
and infinite traces over $\Sigma$.

\subsection{Recursively specified trace distances}

Recall that a \emph{(commutative) quantale} consists of a complete
lattice $( \LL, \sqsubseteq_\LL)$ and a commutative, associative
addition operation $\oplus_\LL$ which distributes over arbitrary
suprema; we denote by $\bot_\LL$, $\top_\LL$ the bottom and top elements
of $\LL$.  We call a function $d: X\times X\to \LL$, for a set $X$ and a
quantale $\LL$, an \emph{$\LL$-hemimetric} if it satisfies $d( x, x)=
\bot_\LL$ for all $x\in X$ and $d( x, z)\sqsubseteq_\LL d( x,
y)\oplus_\LL d( y, z)$ for all $x, y, z\in X$.

$\LL$-hemimetrics are generalizations of distances: for $\LL=
\Realnn\cup\{ \infty\}$ the extended real line, an $( \Realnn\cup\{
\infty\})$-hemimetric is simply an extended hemimetric, \ie~a function
$d: X\times X\to \Realnn\cup\{ \infty\}$ which satisfies $d( x, x)= 0$
for all $x\in X$ and the triangle inequality $d( x, z)\le d( x, y)+ d(
y, z)$ for all $x, y, z\in X$.  If $d$ also is symmetric,
\ie~satisfies $d( x, y)= d( y, x)$ for all $x, y\in X$, then $d$ is
usually called a \emph{pseudometric}.  If $d$ also satisfies the
principle of \emph{separability}, or indiscernibility of identicals,
\ie~such that $d( x, y)= 0$ implies $x= y$, it is called a
\emph{metric}.

A \emph{recursive trace distance specification} $( \LL, \eval, \tdl,
F)$ consists of a quantale $\LL$, a quantale morphism $\eval: \LL\to
\Realnn\cup\{ \infty\}$, an $\LL$-hemimetric $\tdl:
\Sigma^\infty\times \Sigma^\infty\to \LL$ (called \emph{lifted trace
  distance}), and a \emph{distance iterator} function $F: \Sigma\times
\Sigma\times \LL\to \LL$.  For our purposes, $F$ must be monotone in
the third and anti-monotone in the second coordinate and satisfy an
extended triangle inequality: for all $a, b, c\in \Sigma$ and $\alpha,
\beta\in \LL$, $F( a, b, \alpha)\oplus_\LL F( b, c,
\beta)\sqsupseteq_\LL F( a, c, \alpha\oplus_\LL \beta)$.

$F$ is to specify $\tdl$ recursively in the sense that for all $a, b\in
\Sigma$ and all $\sigma, \tau\in \Sigma^\infty$ (and with ``$.$''
denoting concatenation),
\begin{equation}
  \label{eq:rec}
  \tdl( a. \sigma, b. \tau)= F( a, b, \tdl( \sigma, \tau))\,.
\end{equation}
The \emph{trace distance} associated with such a distance specification
is $\td: \Sigma^\infty\times \Sigma^\infty\to \Realnn$ given by $\td=
\eval\circ \tdl$.

Note that $\tdl$ specializes to a distance on labels (because
$\Sigma\subseteq \Sigma^\infty$); we require that this is compatible
with label refinement in the sense that $a\labpre b$ implies $\tdl( a,
b)= \bot_\LL$.  Then \eqref{eq:rec}~implies that whenever $a\labpre b$,
then $F( a, b, \bot_\LL)= \tdl( a, b)= \bot_\LL$.  As an inverse
property, we say that $F$ is \emph{recursively separating} if $F( a, b,
\alpha)= \bot_\LL$ implies that $a\labpre b$ and $\alpha= \bot_\LL$.

\begin{example}
  \label{ex:distances}
  We have shown in previous chapters that all commonly used trace
  distances obey recursive characterizations as above.  We give a few
  examples, all of which are recursively separating:

  The \emph{point-wise} distance
  from~\cite{DBLP:journals/tcs/AlfaroFHMS05} has
  $\LL= \Realnn\cup\{ \infty\}$, $\eval= \id$ and
  \begin{equation*}
    \tdl( a. \sigma,
    b. \tau)= \max( d( a, b), \tdl( \sigma, \tau))\,,
  \end{equation*}
  where $d: \Sigma\times \Sigma\to \Realnn\cup\{ \infty\}$ is a
  hemimetric on labels.  For the label set $\Sigma= U\times\{[ l, r]\mid
  l\in \Real\cup\{ -\infty\}, r\in \Real\cup\{ \infty\}, l\le r\}$ from
  Example~\ref{ex:labelsets}, one useful example of such a hemimetric is
  $d(( u_1,[ l_1, r_1]),( u_2,[ l_2, r_2]))= \sup_{ x_1\in[ l_1, r_1]}$
  $\inf_{ x_2\in[ l_2, r_2]}| x_1- x_2|= \max( l_2- l_1, r_1- r_2, 0)$ if
  $u_1= u_2$ and $\infty$ otherwise.

  The \emph{discounting} distance, also used
  in~\cite{DBLP:journals/tcs/AlfaroFHMS05}, again uses $\LL=
  \Realnn\cup\{ \infty\}$ and $\eval= \id$, but
  \begin{equation*}
    \tdl( a. \sigma,
    b. \tau)= d( a, b)+ \lambda \tdl( \sigma, \tau)
  \end{equation*}
  for a constant $\lambda\in[ 0, 1\mathclose[$.

  For the \emph{limit-average} distance used
  in~\cite{DBLP:journals/tcs/CernyHR12} and other papers, $\LL=(
  \Realnn\cup\{ \infty\})^\Nat$, $\eval( \alpha)= \liminf_{ j\in \Nat}
  \alpha( j)$, and
  \begin{equation*}
    \tdl( a. \sigma, b. \tau)( j)= \tfrac 1{ j+ 1}d( a,
    b)+ \tfrac j{ j+ 1} \tdl( \sigma, \tau)( j- 1)\,.
  \end{equation*}
  It is clear that limit-average distance has no recursive
  specification which uses $\LL= \Realnn\cup\{ \infty\}$ as for the
  other distances above.  Intuitively, the quantale
  $( \Realnn\cup\{ \infty\})^\Nat$ has to be used to memorize how many
  symbols one has seen in the sequences $\sigma$, $\tau$.  This and
  other examples show that using general quantales in recursive trace
  distance specifications instead of simply
  $\LL= \Realnn\cup\{ \infty\}$ is necessary.

  The \emph{discrete} trace distance is given by $\td( \sigma, \tau)= 0$
  if $\sigma\labpre \tau$ and $\infty$ otherwise (here we have extended
  $\labpre$ to traces in the obvious way).  It has a recursive
  characterization with $\LL= \Realnn\cup\{ \infty\}$, $\eval= \id$, and
  $\td( a. \sigma, b. \tau)= \td( \sigma, \tau)$ if $a\labpre b$ and
  $\infty$ otherwise. \qed
\end{example}

For the rest of this paper, we fix a recursively specified trace
distance.

\subsection{Refinement distances}

We lift the notions of modal refinement, for all our formalisms, to
distances.  Conceptually, this is done by replacing ``$\forall$''
quantifiers by ``$\sup$'' and ``$\exists$'' by ``$\inf$'' in the
definitions, and then using the distance iterator to introduce a
recursive functional whose least fixed point is the distance.

\begin{definition}
  \quad The \emph{lifted refinement distance} on the states of DMTS
  $\mcalD_1=( S_1,$ $S^0_1, \omay_1, \omust_1)$ and $\mcalD_2=( S_2, S^0_2,
  \omay_2, \omust_2)$ is the least fixed point to the equations
  \begin{equation*}
    \mdl( s_1, s_2)= \max
    \begin{cases}
      &\hspace*{-1em} \adjustlimits \sup_{ s_1\DMTSmay{ a_1} t_1}
      \inf_{
        s_2\DMTSmay{ a_2} t_2} F( a_1, a_2, \mdl( t_1, t_2))\,, \\
      &\hspace*{-1em} \multiadjustlimits{ \sup_{ s_2\DMTSmust{} N_2}
        \inf_{ s_1\DMTSmust{} N_1} \sup_{( a_1, t_1)\in N_1} \inf_{(
          a_2, t_2)\in N_2}} F( a_1, a_2, \mdl( t_1, t_2))\,,
    \end{cases}
  \end{equation*}
  for $s_1\in S_1$, $s_2\in S_2$.  For \NAA $\mcalA_1=( S_1, S^0_1,
  \Tran_1)$, $\mcalA_2=( S_2, S^0_2, \Tran_2)$, the right-hand side is
  replaced by
  \begin{equation*}
    \adjustlimits \sup_{ M_1\in \Tran_1( s_1)} \inf_{ M_2\in \Tran_2(
      s_2)} \max
    \begin{cases}
      &\hspace*{-1em} \adjustlimits \sup_{( a_1, t_1)\in M_1} \inf_{(
        a_2, t_2)\in M_2} F( a_1, a_2, \mdl( t_1, t_2))\,, \\
      &\hspace*{-1em} \adjustlimits \sup_{( a_2, t_2)\in M_2} \inf_{(
        a_1, t_1)\in M_1} F( a_1, a_2, \mdl( t_1, t_2))\,,
    \end{cases}
  \end{equation*}
  and for $\nu$-calculus expressions $\mcalN_1=( X_1, X^0_1, \Delta_1)$,
  $\mcalN_2=( X_2, X^0_2, \Delta_2)$ in normal form, it is
  \begin{equation*}
    \max
    \begin{cases}
      & \hspace*{-1em} \adjustlimits \sup_{ a_1\in \Sigma, y_1\in
        \Box^{ a_1}_1( x_1)\,} \inf_{ \, a_2\in \Sigma, y_2\in \Box^{
          a_2}_2(
        x_2)} F( a_1, a_2, \mdl( y_1, y_2))\,, \\
      &\hspace*{-1em} \multiadjustlimits{ \sup_{ N_2\in\Diamond_2( x_2)}
      \inf_{ N_1\in \Diamond_1( x_1)} \sup_{( a_1,
        y_1)\in N_1} \inf_{( a_2, y_2)\in N_2}} F( a_1, a_2, \mdl( y_1,
      y_2))\,.
    \end{cases}
  \end{equation*}
\end{definition}

Using Tarski's fixed point theorem, one easily sees that the lifted
refinement distances are indeed well-defined.  (Here one needs
monotonicity of $F$ in the third coordinate, together with the fact that
$\sup$ and $\inf$ are monotonic.)

Note that we define the distances using \emph{least} fixed points, as
opposed to the \emph{greatest} fixed point definition of standard
refinement.  Informally, this is because our order is reversed: we are
not interested in maximizing refinement relations, but in
\emph{minimizing} refinement distance.

The lifted refinement distance between specifications is defined by
\begin{equation*}
  \mdl( \mcalS_1, \mcalS_2)= \adjustlimits \sup_{ s^0_1\in S^0_1} \inf_{
    s^0_2\in S^0_2} \mdl( s^0_1, s^0_2)\,.
\end{equation*}
Analogously to thorough refinement, there is also a \emph{lifted
  thorough refinement distance}, given by $\thdl( \mcalS_1, \mcalS_2)= \sup_{
  \mcalI_1\in \sem{ \mcalS_1}} \inf_{ \mcalI_2\in \sem{ \mcalS_2}}$ $\mdl( \mcalI_1,
\mcalI_2)$.

Using the $\eval$ function, one gets distances $\md= \eval\circ \mdl$
and $\thd= \eval\circ \thdl$, with values in $\Realnn\cup\{ \infty\}$,
which will be the ones one is interested in for concrete applications.

We recall the notion of \emph{refinement family} from
Chapter~\ref{ch:wm2} and extend it to specifications.  We give the
definition for \NAA only; for DMTS and the modal $\nu$-calculus it is
similar.

\begin{definition}
  \label{de:reffam}
  A \emph{refinement family} from $\mcalA_1$ to $\mcalA_2$, for \NAA $\mcalA_1=(
  S_1, S^0_1, \Tran_1)$, $\mcalA_2=( S_2, S^0_2, \Tran_2)$, is an
  $\LL$-indexed family of relations $R=\{ R_\alpha\subseteq S_1\times
  S_2\mid \alpha\in \LL\}$ with the property that for all $\alpha\in
  \LL$ with $\alpha\ne \top_\LL$, all $( s_1, s_2)\in R_\alpha$, and all
  $M_1\in \Tran_1( s_1)$, there is $M_2\in \Tran_2( s_2)$ such that
  \begin{itemize}
  \item $\forall( a_1, t_1)\in M_1: \exists( a_2, t_2)\in M_2, \beta\in
    \LL:( t_1, t_2)\in R_\beta, F( a_1, a_2, \beta)\sqsubseteq \alpha$,
  \item $\forall( a_2, t_2)\in M_2: \exists( a_1, t_1)\in M_1, \beta\in
    \LL:( t_1, t_2)\in R_\beta, F( a_1, a_2, \beta)\sqsubseteq \alpha$.
  \end{itemize}
\end{definition}

\begin{lemma}
  For all \NAA $\mcalA_1=( S_1, S^0_1, \Tran_1)$, $\mcalA_2=( S_2, S^0_2,
  \Tran_2)$, there exists a refinement family $R$ from $\mcalA_1$ to
  $\mcalA_2$ such that for all $s^0_1\in S^0_1$, there is $s^0_2\in S^0_2$
  for which $( s^0_1, s^0_2)\in R_{ \mdl( \mcalA_1, \mcalA_2)}$.
\end{lemma}

We say that a refinement family as in the lemma \emph{witnesses} $\mdl(
\mcalA_1, \mcalA_2)$.

\begin{proof}
  Define $R$ by $R_\alpha=\{( s_1, s_2)\mid \mdl( s_1,
  s_2)\sqsubseteq_\LL \alpha\}$.  First, as $( s^0_1, s^0_2)\in R_{
    \mdl( s^0_1, s^0_2)}$ for all $s^0_1\in S^0_1$, $s^0_2\in S^0_2$, it
  is indeed the case that for all $s^0_1\in S^0_1$, there is $s^0_2\in
  S^0_2$ for which
  \begin{equation*}
    ( s^0_1, s^0_2)\in R_{ \mdl( \mcalA_1, \mcalA_2)}= R_{
      \max_{ s^0_1\in S^0_1} \min_{ s^0_2\in S^0_2} \mdl( s^0_1, s^0_2)}\,.
  \end{equation*}

  Now let $\alpha\in \LL$ with $\alpha\ne \top_\LL$ and $( s_1, s_2)\in
  R_\alpha$.  Let $M_1\in \Tran_1( s_1)$.  We have $\mdl( s_1,
  s_2)\sqsubseteq_\LL \alpha$, hence there is $M_2\in \Tran_2( s_2)$
  such that
  \begin{equation*}
    \alpha\sqsupseteq_\LL \max
    \begin{cases}
      &\hspace*{-1em} \adjustlimits \sup_{( a_1, t_1)\in M_1} \inf_{(
        a_2, t_2)\in M_2} F( a_1, a_2, \mdl( t_1, t_2))\,, \\
      &\hspace*{-1em} \adjustlimits \sup_{( a_2, t_2)\in M_2} \inf_{(
        a_1, t_1)\in M_1} F( a_1, a_2, \mdl( t_1, t_2))\,.
    \end{cases}
  \end{equation*}
  But this entails that for all $( a_1, t_1)\in M_1$, there is $( a_2,
  t_2)\in M_2$ and $\beta= \mdl( t_1, t_2)$ with $F( a_1, a_2,
  \beta)\sqsubseteq_\LL \alpha$, and that for all $( a_2, t_2)\in M_2$,
  there is $( a_1, t_1)\in M_1$ and $\beta= \mdl( t_1, t_2)$ such that
  $F( a_1, a_2, \beta)\sqsubseteq_\LL \alpha$. \qed
\end{proof}

The following quantitative extension of Theorems~\ref{th:dmtsvsaa-bool}
and~\ref{th:dmtsvsnu-bool} shows that our translations preserve and
reflect refinement distances.

\begin{theorem}
  \label{th:trans-moddist}
  For all DMTS $\mcalD_1, \mcalD_2$, all \NAA $\mcalA_1$, $\mcalA_2$ and all
  $\nu$-calculus expressions $\mcalN_1$, $\mcalN_2$:
  \begin{align*}
    \mdl( \mcalD_1, \mcalD_2) &= \mdl( \da( \mcalD_1), \da( \mcalD_2)) \\
    \mdl( \mcalA_1, \mcalA_2) &= \mdl( \ad( \mcalA_1), \ad( \mcalA_2)) \\
    \mdl( \mcalD_1, \mcalD_2) &= \mdl( \ddh( \mcalD_1), \ddh( \mcalD_2)) \\
    \mdl( \mcalN_1, \mcalN_2) &= \mdl( \hd( \mcalN_1), \hd( \mcalN_2))
  \end{align*}
\end{theorem}

\begin{proof}
  \mbox{}

  \noindent \underline{$\mdl( \da( \mcalD_1), \da(
    \mcalD_2))\sqsubseteq_\LL \mdl( \mcalD_1, \mcalD_2)$:}

  Let $\mcalD_1=( S_1, S^0_1, \omay_1, \omust_1)$ and $\mcalD_2=( S_2, S^0_2,$
  $\omay_2, \omust_2)$ be DMTS.  There exists a DMTS refinement family
  $R=\{ R_\alpha\subseteq S_1\times S_2\mid \alpha\in \LL\}$ such that
  for all $s^0_1\in S^0_1$, there is $s^0_2\in S^0_2$ with $( s^0_1,
  s^0_2)\in R_{ \mdl( \mcalD_1, \mcalD_2)}$.  We show that $R$ is an \NAA
  refinement family.

  Let $\alpha\in \LL$ and $( s_1, s_2)\in R_\alpha$.  Let $M_1\in
  \Tran_1( s_1)$ and define
  \begin{multline*}
    M_2= \big\{ ( a_2, t_2)\mid s_2\DMTSmay{ a_2}_2 t_2, \exists( a_1, t_1)\in
    M_1: \exists \beta\in \LL: \\
    ( t_1, t_2)\in R_\beta, F( a_1, a_2, \beta)\sqsubseteq_\LL \alpha\big\}\,.
  \end{multline*}
  The condition
  \begin{equation*}
    \forall( a_2, t_2)\in M_2: \exists( a_1, t_1)\in M_1, \beta\in
    \LL: ( t_1, t_2)\in R_\beta, F( a_1, a_2, \beta)\sqsubseteq \alpha
  \end{equation*}
  is satisfied by construction.  For the inverse condition, let $( a_1,
  t_1)\in M_1$, then $s_1\DMTSmay{ a_1}_1 t_1$, and as $R$ is a DMTS
  refinement family, this implies that there is $s_2\DMTSmay{ a_2}_2 t_2$
  and $\beta\in \LL$ for which $( t_1, t_2)\in R_\beta$ and $F( a_1,
  a_2, \beta)\sqsubseteq_\LL \alpha$, so that $( a_2, t_2)\in M_2$ by
  construction.

  We are left with showing that $M_2\in \Tran_2( s_2)$.  First we notice
  that by construction, indeed $s_2\DMTSmay{a_2}_2 t_2$ for all $( a_2,
  t_2)\in M_2$.  Now let $s_2\DMTSmust{} N_2$; we need to show that $N_2\cap
  M_2\ne \emptyset$.

  We have $s_1\DMTSmust{} N_1$ such that $\forall( a_1, t_1)\in N_1:
  \exists( a_2, t_2)\in N_2, \beta\in \LL:( t_1, t_2)\in R_\beta, F(
  a_1, a_2, \beta)\sqsubseteq_\LL \alpha$.  We know that $N_1\cap M_1\ne
  \emptyset$, so let $( a_1, t_1)\in N_1\cap M_1$.  Then there is $(
  a_2, t_2)\in N_2$ and $\beta\in \LL$ such that $( t_1, t_2)\in
  R_\beta$ and $F( a_1, a_2, \beta)\sqsubseteq_\LL \alpha$.  But $( a_2,
  t_2)\in N_2$ implies $s_2\DMTSmay{ a_2}_2 t_2$, hence $( a_2, t_2)\in
  M_2$.

  \smallskip \noindent \underline{$\mdl( \mcalD_1, \mcalD_2)\sqsubseteq_\LL
    \mdl( \da( \mcalD_1), \da( \mcalD_2))$:}

  Let $\mcalD_1=( S_1, S^0_1, \omay_1, \omust_1)$ and $\mcalD_2=( S_2, S^0_2,$
  $\omay_2, \omust_2)$ be DMTS.  There exists an \NAA refinement family
  $R=\{ R_\alpha\subseteq S_1\times S_2\mid \alpha\in \LL\}$ such that
  for all $s^0_1\in S^0_1$, there is $s^0_2\in S^0_2$ for which $( s^0_1,
  s^0_2)\in R_{ \mdl( \da( \mcalD_1), \da( \mcalD_2))}$.  We show that $R$ is
  a DMTS refinement family.  Let $\alpha\in \LL$ and $( s_1, s_2)\in
  R_\alpha$.

  Let $s_1\DMTSmay{ a_1}_1 t_1$, then we cannot have $s_1\DMTSmust{} \emptyset$.
  Let $M_1=\{( a_1, t_1)\}\cup \bigcup\{ N_1\mid s_1\DMTSmust{} N_1\}$, then
  $M_1\in \Tran_1( s_1)$ by construction.  This implies that there is
  $M_2\in \Tran_2( s_2)$, $( a_2, t_2)\in M_2$ and $\beta\in \LL$ such
  that $( t_1, t_2)\in R_\beta$ and $F( a_1, a_2, \beta)\sqsubseteq_\LL
  \alpha$, but then also $s_2\DMTSmay{ a_2} t_2$ as was to be shown.

  Let $s_2\DMTSmust{} N_2$ and assume, for the sake of contradiction, that
  there is no $s_1\DMTSmust{} N_1$ for which $\forall( a_1, t_1)\in N_1:
  \exists( a_2, t_2)\in N_2, \beta\in \LL:( t_1, t_2)\in R_\beta, F(
  a_1, a_2, \beta)\sqsubseteq_\LL \alpha$ holds.  Then for each
  $s_1\DMTSmust{} N_1$, there is an element $( a_{ N_1}, t_{ N_1})\in N_1$
  such that $\exists( a_2, t_2)\in N_2, \beta\in \LL:( t_{ N_1}, t_2)\in
  R_\beta, F( a_{ N_1}, a_2,$ $\beta)\sqsubseteq_\LL \alpha$ does
  \emph{not} hold.

  Let $M_1=\{( a_{ N_1}, t_{ N_1})\mid s_1\DMTSmust{} N_1\}$, then $M_1\in
  \Tran_1( s_1)$ by construction.  Hence we have $M_2\in \Tran_2( s_2)$
  such that $\forall( a_2, t_2)\in M_2: \exists( a_1, t_2)\in M_1,
  \beta\in \LL:( t_1, t_2)\in R_\beta, F( a_1, a_2, \beta)\sqsubseteq
  \alpha$.  Now $N_2\cap M_2\ne \emptyset$, so let $( a_2, t_2)\in
  N_2\cap M_2$, then there is $( a_1, t_1)\in M_1$ and $\beta\in \LL$
  such that $( t_1, t_2)\in R_\beta$ and $F( a_1, a_2,
  \beta)\sqsubseteq_\LL \alpha$, in contradiction to how $M_1$ was
  constructed.

  \smallskip \noindent \underline{$\mdl( \ad( \mcalA_1), \ad(
    \mcalA_2))\sqsubseteq_\LL \mdl( \mcalA_1, \mcalA_2)$:}

  Let $\mcalA_1=( S_1, S^0_1, \Tran_1)$, $\mcalA_2=( S_2, S^0_2, \Tran_2)$ be
  \NAA, with DMTS translations $\ad( \mcalA_1)=( D_1, D^0_1, \omust_1,$
  $\omay_1)$, $\ad( \mcalA_2)=( D_2, D^0_2, \omust_2, \omay_2)$.  There is
  an \NAA refinement family $R=\{ R_\alpha\subseteq S_1\times S_2\mid
  \alpha\in \LL\}$ such that for all $s^0_1\in S^0_1$, there is
  $s^0_2\in S^0_2$ with $( s^0_1, s^0_2)\in R_{ \mdl( \mcalA_1, \mcalA_2)}$.

  Define a relation family $R'=\{ R'_\alpha\subseteq D_1\times D_2\mid
  \alpha\in \LL\}$ by
  \begin{align*}
    R'_\alpha &= \big\{ ( M_1, M_2)\bigmid \exists( s_1, s_2)\in
    R_\alpha:
    M_1\in \Tran_1( s_1), M_2\in \Tran( s_2), \\
    &\hspace{8em}
    \begin{aligned}
      & \forall( a_1, t_1)\in M_1: \exists( a_2, t_2)\in M_2, \beta\in
      \LL: \\
      &\hspace*{6em} ( t_1, t_2)\in R_\beta, F( a_1, a_2,
      \beta)\sqsubseteq_\LL \alpha\,, \\
      & \forall( a_2, t_2)\in M_2: \exists( a_1, t_1)\in M_1, \beta\in
      \LL: \\
      &\hspace*{6em} ( t_1, t_2)\in R_\beta, F( a_1, a_2,
      \beta)\sqsubseteq_\LL \alpha\big\}\,.
    \end{aligned}
  \end{align*}
  We show that $R'$ is a witness for $\mdl( \ad( \mcalA_1), \ad(
  \mcalA_2))\sqsubseteq_\LL \mdl( \mcalA_1, \mcalA_2)$.  Let $\alpha\in \LL$ and
  $( M_1, M_2)\in R'_\alpha$.

  Let $M_2\DMTSmust{}_2 N_2$.  By construction of $\omust$, there is $( a_2,
  t_2)\in M_2$ such that $N_2=\{( a_2, M_2')\mid M_2'\in \Tran_2(
  t_2)\}$.  Then $( M_1, M_2)\in R'_\alpha$ implies that there must be
  $( a_1, t_1)\in M_1$ and $\beta\in \LL$ such that $( t_1, t_2)\in
  R_\beta$ and $F( a_1, a_2, \beta)\sqsubseteq_\LL \alpha$.  Let
  $N_1=\{( a_1, M_1')\mid M_1'\in \Tran_1( t_1)\}$, then $M_1\DMTSmust{}_1
  N_1$.

  We show that $\forall( a_1, M_1')\in N_1: \exists( a_2, M_2')\in N_2:(
  M_1', M_2')\in R'_\beta$: Let $( a_1, M_1')\in N_1$, then $M_1'\in
  \Tran_1( t_1)$.  From $( t_1, t_2)\in R_\beta$ we get $M_2'\in
  \Tran_2( t_2)$ such that
  \begin{align*}
    & \forall( b_1, u_1)\in M_1': \exists( b_2, u_2)\in M_2', \gamma\in
    \LL: ( u_1, u_2)\in R_\gamma, F( b_1, b_2,
    \gamma)\sqsubseteq_\LL \beta\,, \\
    & \forall( b_2, u_2)\in M_2': \exists( b_1, u_1)\in M_1', \gamma\in
    \LL: ( u_1, u_2)\in R_\gamma, F( b_1, b_2,
    \gamma)\sqsubseteq_\LL \beta\,,
  \end{align*}
  hence $( M_1', M_2')\in R'_\beta$; also, $( a_2, M_2')\in N_2$ by
  construction of $N_2$.

  Let $M_1\DMTSmay{ a_1}_1 M_1'$, then we have $M_1\DMTSmust{}_1 N_1$ for which
  $( a_1, M_1')\in N_1$ by construction of $\omay_1$.  This in turn
  implies that there must be $( a_1, t_1)\in M_1$ such that $N_1=\{(
  a_1, M_1'')\mid M_1''\in \Tran_1( t_1)\}$.  By $( M_1, M_2)\in
  R'_\alpha$, we get $( a_2, t_2)\in M_2$ and $\beta\in \LL$ such that
  $( t_1, t_2)\in R_\beta$ and $F( a_1, a_2, \beta)\sqsubseteq_\LL
  \alpha$.  Let $N_2=\{( a_2, M_2')\mid M_2'\in \Tran_2( t_2)\}$, then
  $M_2\DMTSmust{}_2 N_2$ and hence $M_2\DMTSmay{ a_2}_2 M_2'$ for all $( a_2,
  M_2')\in N_2$.  By the same arguments as above, there is $( a_2,
  M_2')\in N_2$ for which $( M_1', M_2')\in R'_\beta$.

  We miss to show that $R'$ is initialized.  Let $M_1^0\in D_1^0$, then
  we have $s_1^0\in S_1^0$ with $M_1^0\in \Tran_1( s_1^0)$.  As $R$ is
  initialized, this entails that there is $s_2^0\in S_2^0$ with $(
  s_1^0, s_2^0)\in R_{ \mdl( \mcalA_1, \mcalA_2)}$, which gives us $M_2^0\in
  \Tran_2( s_2^0)$ which satisfies the conditions in the definition of
  $R'_{ \mdl( \mcalA_1, \mcalA_2)}$, whence $( M_1^0, M_2^0)\in R'_{ \mdl(
    \mcalA_1, \mcalA_2)}$.

  \smallskip \noindent \underline{$\mdl( \mcalA_1, \mcalA_2)\sqsubseteq_\LL
    \mdl( \ad( \mcalA_1), \ad( \mcalA_2))$:}

  Let $\mcalA_1=( S_1, S^0_1, \Tran_1)$, $\mcalA_2=( S_2, S^0_2, \Tran_2)$ be
  \NAA, with DMTS translations $\ad( \mcalA_1)=( D_1, D^0_1, \omust_1,$
  $\omay_1)$, $\ad( \mcalA_2)=( D_2, D^0_2, \omust_2, \omay_2)$.  There is
  a DMTS refinement family $R=\{ R_\alpha\subseteq D_1\times D_2\mid
  \alpha\in \LL\}$ such that for all $M_1^0\in D_1^0$, there exists
  $M_2^0\in D_2^0$ with $( M_1^0, M_2^0)\in R_{ \mdl( \ad( \mcalA_1), \ad(
    \mcalA_2))}$.

  Define a relation family $R'=\{ R'_\alpha\subseteq S_1\times S_2\mid
  \alpha\in \LL\}$ by
  \begin{equation*}
    R'_\alpha= \big\{ ( s_1, s_2)\bigmid \forall M_1\in \Tran_1( s_1):
    \exists M_2\in \Tran_2( s_2):( M_1, M_2)\in R_\alpha\big\}\,;
  \end{equation*}
  we will show that $R'$ is a witness for $\mdl( \mcalA_1,
  \mcalA_2)\sqsubseteq_\LL \mdl( \ad( \mcalA_1), \ad( \mcalA_2))$.

  Let $\alpha\in \LL$, $( s_1, s_2)\in R'_\alpha$ and $M_1\in \Tran_1(
  s_1)$, then by construction of $R'$, we have $M_2\in \Tran_2( s_2)$
  with $( M_1, M_2)\in R_\alpha$.

  Let $( a_2, t_2)\in M_2$ and define $N_2=\{( a_2, M_2')\mid M_2'\in
  \Tran_2( t_2)\}$, then $M_2\DMTSmust{}_2 N_2$.  Now $( M_1, M_2)\in
  R_\alpha$ implies that there must be $M_1\DMTSmust{}_1 N_1$ satisfying
  $\forall( a_1, M_1')\in N_1: \exists( a_2, M_2')\in N_2, \beta\in
  \LL:( M_1', M_2')\in R_\beta,$ $F( a_1, a_2, \beta)\sqsubseteq_\LL
  \alpha$.  We have $( a_1, t_1)\in M_1$ such that $N_1=\{( a_1,
  M_1')\mid M_1'\in \Tran_1( t_1)\}$; we only miss to show that $( t_1,
  t_2)\in R'_\beta$ for some $\beta\in \LL$ for which $F( a_1, a_2,
  \beta)\sqsubseteq_\LL \alpha$.  Let $M_1'\in \Tran_1( t_1)$, then $(
  a_1, M_1')\in N_1$, hence there is $( a_2, M_2')\in N_2$ and $\beta\in
  \LL$ such that $( M_1', M_2')\in R_\beta$ and $F( a_1, a_2,
  \beta)\sqsubseteq \alpha$, but $( a_2, M_2')\in N_2$ also entails
  $M_2'\in \Tran_2( t_2)$.

  Let $( a_1, t_1)\in M_1$ and define $N_1=\{( a_1, M_1')\mid M_1'\in
  \Tran_1( t_1)\}$, then $M_1\DMTSmust{}_1 N_1$.  Now let $( a_1, M_1')\in
  N_1$, then $M_1\DMTSmay{ a_1}_1 M_1'$, hence we have $M_2\DMTSmay{ a_2}_2
  M_2'$ and $\beta\in \LL$ such that $( M_1', M_2')\in R_\beta$ and $F(
  a_1, a_2, \beta)\sqsubseteq_\LL \alpha$.  By construction of
  $\omay_2$, this implies that there is $M_2\DMTSmust{}_2 N_2$ with $( a_2,
  M_2')\in N_2$, and we have $( a_2, t_2)\in M_2$ for which $N_2=\{(
  a_2, M_2'')\mid M_2''\in \Tran_2( t_2)\}$.  Now if $M_1''\in \Tran_1(
  t_1)$, then $( a_1, M_1'')\in N_1$, hence there is $( a_2, M_2'')\in
  N_2$ with $( M_1'', M_2'')\in R_\beta$, but $( a, M_2'')\in N_2$ also
  gives $M_2''\in \Tran_2( t_2)$.

  We miss to show that $R'$ is initialized.  Let $s^0_1\in S^0_1$ and
  $M^0_1\in \Tran_1( s^0_1)$.  As $R$ is initialized, this gets us
  $M^0_2\in D_2$ with $( M^0_1, M^0_2)\in R_{ \mdl( \ad( \mcalA_1), \ad(
    \mcalA_2))}$, but $M^0_2\in \Tran_2( s^0_2)$ for some $s^0_2\in S^0_2$,
  and then $( s^0_1, s^0_2)\in R'_{ \mdl( \ad( \mcalA_1), \ad( \mcalA_2))}$.

  \smallskip \noindent \underline{$\mdl( \ddh( \mcalD_1), \ddh(
    \mcalD_2))\sqsubseteq_\LL \mdl( \mcalD_1, \mcalD_2)$:}

  Let $\mcalD_1=( S_1, S_1^0, \omay_1, \omust_1)$ and $\mcalD_2=( S_2, S_2^0,$
  $\omay_2, \omust_2)$ be DMTS, with $\nu$-calculus
  translations
  $\ddh( \mcalD_1)=( S_1, S_1^0, \Delta_1)$ and $\ddh( \mcalD_2)=( S_2, S_2^0,
  \Delta_2)$.  There is a DMTS refinement family $R=\{ R_\alpha\subseteq
  S_1\times S_2\mid \alpha\in \LL\}$ such that for all $s_1^0\in S_1^0$,
  there exists $s_2^0\in S_2^0$ for which $( s_1^0, s_2^0)\in R_{ \mdl(
    \mcalD_1, \mcalD_2)}$.

  Let $\alpha\in \LL$, $( s_1, s_2)\in R_\alpha$, $a_1\in \Sigma$, and
  $t_1\in \Box^{ a_1}_1( s_1)$.  Then $s_1\DMTSmay{ a_1}_1 t_1$, hence we
  have $s_2\DMTSmay{ a_2}_2 t_2$ and $\beta\in \LL$ with $( t_1, t_2)\in
  R_\beta$ and $F( a_1, a_2, \beta)\sqsubseteq_\LL \alpha$, but then
  also $t_2\in \Box^{ a_2}_2( s_2)$.

  Let $N_2\in \Diamond_2( s_2)$, then also $s_2\DMTSmust{}_2 N_2$, so that
  there must be $s_1\DMTSmust{}_1 N_1$ such that $\forall( a_1, t_1)\in N_1:
  \exists( a_2, t_2)\in N_2, \beta\in \LL:( t_1, t_2)\in R_\beta, F(
  a_1, a_2, \beta)\sqsubseteq_\LL \alpha$, but then also $N_1\in
  \Diamond_1( s_1)$.

  \smallskip \noindent \underline{$\mdl( \mcalD_1, \mcalD_2)\sqsubseteq_\LL
    \mdl( \ddh( \mcalD_1), \ddh( \mcalD_2))$:}

  Let $\mcalD_1=( S_1, S_1^0, \omay_1, \omust_1)$ and $\mcalD_2=( S_2, S_2^0,$
  $\omay_2, \omust_2)$ be DMTS, with $\nu$-calculus
  translations
  $\ddh( \mcalD_1)=( S_1, S_1^0, \Delta_1)$ and $\ddh( \mcalD_2)=( S_2, S_2^0,
  \Delta_2)$.  There is a $\nu$-calculus refinement family $R=\{
  R_\alpha\subseteq S_1\times S_2\mid \alpha\in \LL\}$ such that for all
  $s_1^0\in S_1^0$, there exists $s_2^0\in S_2^0$ for which $( s_1^0,
  s_2^0)\in R_{ \mdl( \mcalD_1, \mcalD_2)}$.

  Let $\alpha\in \LL$ and $( s_1, s_2)\in R_\alpha$, and assume that
  $s_1\DMTSmay{ a_1}_1 t_1$.  Then $t_1\in \Box^{ a_1}_1( s_1)$, so that
  there is $a_2\in \Sigma$, $t_2\in \Box^{ a_2}_2( s_2)$ and $\beta\in
  \LL$ for which $( t_1, t_2)\in R_\beta$ and $F( a_1, a_2,
  \beta)\sqsubseteq_\LL \alpha$, but then also $s_2\DMTSmay{ a_2}_2 t_2$.

  Assume that $s_2\DMTSmust{}_2 N_2$, then $N_2\in \Diamond_2( s_2)$.  Hence
  there is $N_1\in \Diamond_1( s_1)$ so that $\forall( a_1, t_1)\in N_1:
  \exists( a_2, t_2)\in N_2, \beta\in \LL:( t_1, t_2)\in R_\beta, F(
  a_1, a_2, \beta)\sqsubseteq_\LL \alpha$, but then also $s_1\DMTSmust{}_1
  N_1$.

  \smallskip \noindent \underline{$\mdl( \hd( \mcalN_1), \hd(
    \mcalN_2))\sqsubseteq_\LL \mdl( \mcalN_1, \mcalN_2)$:}

  Let $\mcalN_1=( X_1, X_1^0, \Delta_1)$, $\mcalN_2=( X_2, X_2^0, \Delta_2)$
  be $\nu$-calculus expressions in normal form, with DMTS translations
  $\hd( \mcalN_1)=( X_1, X_1^0, \omay_1, \omust_1)$ and $\hd( \mcalN_2)=( X_2,
  X_2^0, \omay_2, \omust_2)$.  There is a $\nu$-calculus refinement
  family $R=\{ R_\alpha\subseteq X_1\times X_2\mid \alpha\in \LL\}$ such
  that for all $x_1^0\in X_1^0$, there is $x_2^0\in X_2^0$ for which $(
  x_1^0, x_2^0)\in R_{ \mdl( \mcalN_1, \mcalN_2)}$.

  Let $\alpha\in \LL$ and $( x_1, x_2)\in R_\alpha$, and assume that
  $x_1\DMTSmay{ a_1}_1 y_1$.  Then $y_1\in \Box_1^{ a_1}( x_1)$, hence there
  are $a_2\in \Sigma$, $y_2\in \Box_2^{ a_2}$ and $\beta\in \LL$ such
  that $( y_1, y_2)\in R_\beta$ and $F( a_1, a_2, \beta)\sqsubseteq_\LL
  \alpha$, but then also $x_2\DMTSmay{ a_2}_2 y_2$.

  Assume that $x_2\DMTSmust{}_2 N_2$, then $N_2\in \Diamond_2( x_2)$.  Hence
  there must be $N_1\in \Diamond_1( x_1)$ such that $\forall( a_1,
  y_1)\in N_1: \exists( a_2, y_2)\in N_2, \beta\in \LL:( y_1, y_2)\in
  R_\beta, F( a_1, a_2, \beta)\sqsubseteq_\LL \alpha$, but then also
  $x_1\DMTSmust{}_1 N_1$.

  \smallskip \noindent \underline{$\mdl( \mcalN_1, \mcalN_2)\sqsubseteq_\LL
    \mdl( \hd( \mcalN_1), \hd( \mcalN_2))$:}

  Let $\mcalN_1=( X_1, X_1^0, \Delta_1)$, $\mcalN_2=( X_2, X_2^0, \Delta_2)$
  be $\nu$-calculus expressions in normal form, with DMTS translations
  $\hd( \mcalN_1)=( X_1, X_1^0, \omay_1, \omust_1)$ and $\hd( \mcalN_2)=( X_2,
  X_2^0, \omay_2, \omust_2)$.  There is a DMTS refinement family $R=\{
  R_\alpha\subseteq X_1\times X_2\mid \alpha\in \LL\}$ such that for all
  $x_1^0\in X_1^0$, there is $x_2^0\in X_2^0$ for which $( x_1^0,
  x_2^0)\in R_{ \mdl( \mcalN_1, \mcalN_2)}$.

  Let $\alpha\in \LL$, $( x_1, x_2)\in R_\alpha$, $a_1\in \Sigma$, and
  $y_1\in \Box^{ a_1}_1( x_1)$.  Then $x_1\DMTSmay{ a_1}_1 y_1$, hence we
  have $x_2\DMTSmay{ a_2}_2 y_2$ and $\beta\in \LL$ so that $( y_1, y_2)\in
  R_\beta$ and $F( a_1, a_2, \beta)\sqsubseteq_\LL \alpha$, but then
  also $y_1\in \Box^{ a_2}_2( x_2)$.

  Let $N_2\in \Diamond_2( x_2)$, then also $x_2\DMTSmust{}_2 N_2$.  Hence we
  must have $x_1\DMTSmust{}_1 N_1$ with $\forall( a_1, y_1)\in N_1: \exists(
  a_2, y_2)\in N_2, \beta\in \LL:( y_1, y_2)\in R_\beta, F( a_1, a_2,
  \beta)\sqsubseteq_\LL \alpha$, but then also $N_1\in \Diamond_1(
  x_1)$.  \qed
\end{proof}

\subsection{Properties}

We sum up some important properties of our distances.

\begin{proposition}
  \label{pr:dist-prop}
  For all specifications $\mcalS_1$, $\mcalS_2$, $\mcalS_1\mr \mcalS_2$ implies
  $\mdl( \mcalS_1, \mcalS_2)= \bot_\LL$, and $\mcalS_1\DMTStr \mcalS_2$ implies $\thdl(
  \mcalS_1, \mcalS_2)= \bot_\LL$.  If $F$ is recursively separating, then
  $\mdl( \mcalS_1, \mcalS_2)= \bot_\LL$ implies $\mcalS_1\mr \mcalS_2$.
\end{proposition}

\begin{proof}
  We show the proposition for \NAA.  First, if $\mcalA_1\mr \mcalA_2$, with
  $\mcalA_1=( S_1, S_1^0, \Tran_1)$, $\mcalA_2=( S_2, S_2^0, \Tran_2)$, then
  there is an initialized refinement relation $R\subseteq S_1\times
  S_2$, \ie~such that for all $( s_1, s_2)\in R$ and all $M_1\in
  \Tran_1( s_1)$, there is $M_2\in \Tran_2( s_2)$ for which
  \begin{itemize}
  \item $\forall( a_1, t_1)\in M_1: \exists( a_2, t_2)\in M_2:
    a_1\labpre a_2,( t_1, t_2)\in R$ and
  \item $\forall( a_2, t_2)\in M_2: \exists( a_1, t_1)\in M_1:
    a_1\labpre a_2,( t_1, t_2)\in R$.
  \end{itemize}
  Defining $R'=\{ R'_\alpha\mid \alpha\in \LL\}$ by $R'_\alpha= R$ for
  all $\alpha\in \LL$, we see that $R'$ is an initialized refinement
  family which witnesses $\mdl( \mcalA_1, \mcalA_2)= \bot_\LL$.

  We have shown that $\mcalA_1\mr \mcalA_2$ implies $\mdl( \mcalA_1, \mcalA_2)=
  \bot_\LL$.  Now if $\mcalA_1\DMTStr \mcalA_2$ instead, then for all
  $\mcalI\in \sem{ \mcalA_1}$, also $\mcalI\in \sem{ \mcalA_2}$, hence $\thdl( \mcalA_1,
  \mcalA_2)= \bot_\LL$.

  To show the last property, assume $F$ to be recursively separating.
  Define $R\subseteq S_1\times S_2$ by $R=\{( s_1, s_2)\mid \mdl( s_1,
  s_2)= \bot_\LL\}$; we show that $R$ is a witness for $\mcalA_1\mr
  \mcalA_2$.  By $\mdl( \mcalA_1, \mcalA_2)= \bot_\LL$, $R$ is initialized.

  Let $( s_1, s_2)\in R$ and $M_1\in \Tran_1( s_1)$, then there is
  $M_2\in \Tran_2( s_2)$ such that
  \begin{align*}
    & \forall( a_1, t_1)\in M_1: \exists( a_2, t_2)\in M_2,
    \beta_1\in \LL: \\
    &\hspace*{10em} \mdl( t_1, t_2)\sqsubseteq_\LL \beta_1, F( a_1, a_2,
    \beta_1)= \bot_\LL\,, \\
    & \forall( a_2, t_2)\in M_2: \exists( a_1, t_1)\in M_1,
    \beta_1\in \LL: \\
    &\hspace*{10em} \mdl( t_1, t_2)\sqsubseteq_\LL \beta_1, F( a_1, a_2,
    \beta_1)= \bot_\LL\,.
  \end{align*}
  As $F$ is recursively separating, we must have $a_1\labpre a_2$ in
  both these equations and $\beta_1= \beta_2= \bot_\LL$.  But then $(
  t_1, t_2)\in R$, hence $R$ is indeed a witness for $\mcalA_1\mr
  \mcalA_2$. \qed
\end{proof}

\begin{proposition}
  \label{pr:dist-prop-hemi}
  The functions $\mdl$ and $\thdl$ are $\LL$-hemi\-metrics, and $\md$,
  $\thd$ are hemimetrics.
\end{proposition}

\begin{proof}
  We show the proof for \NAA.  The properties that $\mdl( \mcalA, \mcalA)=
  \bot_\LL$ and $\thdl( \mcalA, \mcalA)= \bot_\LL$ follow from
  Proposition~\ref{pr:dist-prop}.

  We show the triangle inequality for $\mdl$.  The triangle inequality
  for $\thdl$ will then follow from standard arguments used to show
  that the Hausdorff metric satisfies the triangle inequality, see for
  example \cite[Lemma 3.72]{book/AliprantisB07}.  Let
  $\mcalA_1=( S_1, S^0_1, \Tran_1)$,
  $\mcalA_2=( S_2, S^0_2, \Tran_2)$, $\mcalA_3=( S_3, S^0_3, \Tran_3)$
  be \NAA and
  $R^1=\{ R^1_\alpha\subseteq S_1\times S_2\mid \alpha\in \LL\}$,
  $R^2=\{ R^2_\alpha\subseteq S_2\times S_3\mid \alpha\in \LL\}$
  refinement families such that
  $\forall s_1^0\in S_1^0: \exists s_2^0\in S_2^0:( s_1^0, s_2^0)\in
  R^1_{ \mdl( \mcalA_1, \mcalA_2)}$ and
  $\forall s_2^0\in S_2^0: \exists s_3^0\in S_3^0:( s_2^0, s_3^0)\in
  R^2_{ \mdl( \mcalA_2, \mcalA_3)}$.

  Define $R= \{ R_\alpha\subseteq S_1\times S_3\mid \alpha\in \LL\}$ by
  \begin{multline*}
    R_\alpha= \big\{( s_1, s_3)\bigmid \exists \alpha_1, \alpha_2\in
    \LL, s_2\in S_2: \\
    ( s_1, s_2)\in R^1_{ \alpha_1},( s_2, s_3)\in R^2_{ \alpha_2},
    \alpha_1\oplus_\LL \alpha_2= \alpha\big\}\,.
  \end{multline*}
  Then for all $s_1^0\in S_1^0$ there is $s_3^0\in S_3^0$ such that
  $( s_1^0, s_3^0)\in R_{ \mdl( \mcalA_1, \mcalA_2)\oplus_\LL \mdl(
    \mcalA_2, \mcalA_3)}$; we show that $R$ is a refinement family
  from $\mcalA_1$ to $\mcalA_2$.

  Let $\alpha\in \LL$ and $( s_1, s_3)\in R_\alpha$, then we have
  $\alpha_1, \alpha_2\in \LL$ and $s_2\in S_2$ such that
  $\alpha_1\oplus_\LL \alpha_2= \alpha$, $( s_1, s_2)\in R^1_{
    \alpha_1}$ and $( s_2, s_3)\in R^2_{ \alpha_2}$.  Let $M_1\in
  \Tran_1( s_1)$, then we have $M_2\in \Tran_2( s_2)$ such that
  \begin{align}
    \notag & \forall( a_1, t_1)\in M_1: \exists( a_2, t_2)\in M_2,
    \beta_1\in \LL: \\
    \label{eq:tri.1-2}
    &\hspace*{10em} ( t_1, t_2)\in R^1_{ \beta_1}, F( a_1, a_2,
    \beta_1)\sqsubseteq_\LL \alpha_1\,, \\
    \notag & \forall( a_2, t_2)\in M_2: \exists( a_1, t_1)\in M_1,
    \beta_1\in \LL: \\
    \label{eq:tri.2-1}
    &\hspace*{10em} ( t_1, t_2)\in R^1_{ \beta_1}, F( a_1, a_2,
    \beta_1)\sqsubseteq_\LL \alpha_1\,.
  \end{align}
  This in turn implies that there is $M_3\in \Tran_3( s_3)$ with
  \begin{align}
    \notag & \forall( a_2, t_2)\in M_2: \exists( a_3, t_3)\in M_3,
    \beta_2\in \LL: \\
    \label{eq:tri.2-3}
    &\hspace*{10em} ( t_2, t_3)\in R^2_{ \beta_2}, F( a_2, a_3,
    \beta_2)\sqsubseteq_\LL \alpha_2\,, \\
    \notag & \forall( a_3, t_3)\in M_3: \exists( a_2, t_2)\in M_2,
    \beta_2\in \LL: \\
    \label{eq:tri.3-2}
    &\hspace*{10em} ( t_2, t_3)\in R^2_{ \beta_2}, F( a_2, a_3,
    \beta_2)\sqsubseteq_\LL \alpha_2\,.
  \end{align}

  Now let $( a_1, t_1)\in M_1$, then we get $( a_2, t_2)\in M_2$, $(
  a_3, t_3)\in M_3$ and $\beta_1, \beta_2\in \LL$ as
  in~\eqref{eq:tri.1-2} and~\eqref{eq:tri.2-3}.  Let $\beta=
  \beta_1\oplus_\LL \beta_2$, then $( t_1, t_3)\in R_\beta$, and by the
  extended triangle inequality for $F$, $F( a_1, a_3,
  \beta)\sqsubseteq_\LL F( a_1, a_2, \beta_1)\oplus_\LL F( a_2, a_3,
  \beta_2)\sqsubseteq_\LL \alpha_1\oplus_\LL \alpha_2= \alpha$.

  Similarly, given $( a_3, t_3)\in M_3$, we can apply~\eqref{eq:tri.3-2}
  and~\eqref{eq:tri.2-1} to get $( a_1, t_1)\in M_1$ and $\beta\in \LL$
  such that $( t_1, t_3)\in R_\beta$ and $F( a_1, a_3,
  \beta)\sqsubseteq_\LL \alpha$.

  We have shown that $\mdl$ and $\tdl$ are $\LL$-hemimetrics.  Using
  monotonicity of the $\eval$ function, it follows that $\md$ and $\td$
  are hemimetrics.  \qed
\end{proof}

\begin{proposition}
  \label{pr:dist-prop.disc}
  For the \emph{discrete} distances,
  $\md^\textup{disc}( \mcalS_1, \mcalS_2)= 0$ if
  $\mcalS_1\mr \mcalS_2$ and $\infty$ otherwise.  Similarly,
  $\thd^\textup{disc}( \mcalS_1, \mcalS_2)= 0$ if
  $\mcalS_1\DMTStr \mcalS_2$ and $\infty$ otherwise.
\end{proposition}

\begin{proof}
  We show the proposition for \NAA.  We already know that, also for the
  discrete distances, $\mcalA_1\mr \mcalA_2$ implies $\md( \mcalA_1, \mcalA_2)= 0$
  and that $\mcalA_1\DMTStr \mcalA_2$ implies $\thd( \mcalA_1, \mcalA_2)= 0$.  We show
  that $\md( \mcalA_1, \mcalA_2)= 0$ implies $\mcalA_1\mr \mcalA_2$.  Let $R=\{
  R_\alpha\subseteq S_1\times S_2\mid \alpha\in \LL\}$ be a refinement
  family such that $\forall s_1^0\in S_1^0: \exists s_2^0\in S_2^0:(
  s_1^0, s_2^0)\in R_0$.  We show that $R_0$ is a witness for $\mcalA_1\mr
  \mcalA_2$; it is clearly initialized.

  Let $( s_1, s_2)\in R_0$ and $M_1\in \Tran_1( s_1)$, then we have
  $M_2\in \Tran_2( s_2)$ such that
  \begin{equation}
    \label{eq:disctradistproof}
    \begin{aligned}
      & \forall( a_1, t_1)\in M_1: \exists( a_2, t_2)\in M_2, \beta\in
      \LL: ( t_1, t_2)\in R_\beta, F( a_1, a_2, \beta)= 0\,, \\
      & \forall( a_2, t_2)\in M_2: \exists( a_1, t_1)\in M_1, \beta\in
      \LL: ( t_1, t_2)\in R_\beta, F( a_1, a_2, \beta)= 0\,.
    \end{aligned}
  \end{equation}
  Using the definition of the distance, we see that the condition $F(
  a_1, a_2, \beta)= 0$ is equivalent to $a_1\labpre a_2$ and $\beta= 0$,
  hence~\eqref{eq:disctradistproof} degenerates to
  \begin{equation*}
    \begin{aligned}
      & \forall( a_1, t_1)\in M_1: \exists( a_2, t_2)\in M_2:( t_1,
      t_2)\in R_0, a_1\labpre a_2\,, \\
      & \forall( a_2, t_2)\in M_2: \exists( a_1, t_1)\in M_1:( t_1,
      t_2)\in R_0, a_1\labpre a_2\,,
    \end{aligned}
  \end{equation*}
  which are exactly the conditions for $R_0$ to be a modal refinement.

  Again by definition, we see that for any \NAA $\mcalA_1$, $\mcalA_2$, either
  $\md( \mcalA_1, \mcalA_2)= 0$ or $\md( \mcalA_1, \mcalA_2)= \infty$, hence
  $\mcalA_1\not\mr \mcalA_2$ implies that $\md( \mcalA_1, \mcalA_2)= \infty$.

  To show the last part of the proposition, we notice that
  \begin{align*}
    \thd( \mcalA_1, \mcalA_2) &= \adjustlimits \sup_{ \mcalI_1\in \sem{ \mcalA_1}}
    \inf_{ \mcalI_2\in \sem{ \mcalA_2}} \md( \mcalI_1, \mcalI_2) \\
    &=
    \begin{cases}
      0 &\text{if } \forall \mcalI_1\in \sem{ \mcalA_1}: \exists \mcalI_2\in
      \sem{ \mcalA_2}: \mcalI_1\mr \mcalI_2\,, \\
      \infty &\text{otherwise}\,,
    \end{cases} \\
    &=
    \begin{cases}
      0 &\text{if } \sem{ \mcalA_1}\subseteq \sem{ \mcalA_2}\,, \\
      \infty &\text{otherwise}\,.
    \end{cases}
  \end{align*}
  Hence $\thd( \mcalA_1, \mcalA_2)= 0$ if $\mcalA_1\DMTStr \mcalA_2$ and $\thd( \mcalA_1,
  \mcalA_2)= \infty$ otherwise.  \qed
\end{proof}

As a quantitative analogy to the implication from (Boolean) modal
refinement to thorough refinement (see Proposition~\ref{prop:mrtr}), the
next theorem shows that thorough refinement distance is bounded above by
modal refinement distance.  Note that for the discrete trace distance
(and using Proposition~\ref{pr:dist-prop.disc}), this is equivalent to
the Boolean statement.

\begin{theorem}
  \label{th:mdl-vs-tdl}
  For all specifications $\mcalS_1$, $\mcalS_2$, $\thdl( \mcalS_1,
  \mcalS_2)\sqsubseteq_\LL \mdl( \mcalS_1, \mcalS_2)$.
\end{theorem}

\begin{proof}
  We prove the statement for \NAA; for DMTS and $\nu$-calculus
  expressions it then follows from Theorem~\ref{th:trans-moddist}.

  Let $\mcalA_1=( S_1, S^0_1, \Tran_1)$, $\mcalA_2=( S_2, S^0_2, \Tran_2)$.
  We have a refinement family $R=\{ R_\alpha\subseteq S_1\times S_2\mid
  \alpha\in \LL\}$ such that for all $s^0_1\in S^0_1$, there is
  $s^0_2\in S^0_2$ with $( s^0_1, s^0_2)\in R_{ \mdl( \mcalA_1, \mcalA_2)}$.
  Let $\mcalI=( S, S^0, T)\in \sem{ \mcalA_1}$, \ie~$\mcalI\mr \mcalA_1$.

  Let $R^1\subseteq S\times S_1$ be an initialized modal refinement,
  and define a relation family $R^2=\{ R^2_\alpha\subseteq S\times
  S_2\mid \alpha\in \LL\}$ by $R^2_\alpha= R^1\circ R_\alpha=\{( s,
  s_2)\mid \exists s_1\in S:( s, s_1)\in R^1,( s_1, s_2)\in
  R_\alpha$.  We define a LTS $\mcalI_2=( S_2, S^0_2, T_2)$ as follows:

  For all $\alpha\in \LL$ with $\alpha\ne \top_\LL$ and $( s, s_2)\in
  R^2_\alpha$: We must have $s_1\in S_1$ with $( s, s_1)\in R^1$ and $(
  s_1, s_2)\in R_\alpha$.  Then there is $M_1\in \Tran_1( s_1)$ such that
  \begin{itemize}
  \item for all $s\DMTSmust{ a} t$, there is $( a, t_1)\in M_1$ with $( t,
    t_1)\in R_1$,
  \item for all $( a_1, t_1)\in M_1$, there is $s\DMTSmust{ a} t$ with $( t,
    t_1)\in R_1$.
  \end{itemize}
  This in turn implies that there is $M_2\in \Tran_2( s_2)$ satisfying
  the conditions in Definition~\ref{de:reffam}.  For all $( a_2, t_2)\in
  M_2$: add a transition $s_2\DMTSmust{ a_2} t_2$ to $T_2$.

  We show that the identity relation $\{( s_2, s_2)\mid s_2\in S_2\}$ is
  a witness for $\mcalI_2\mr \mcalA_2$.  Let $s_2\in S_2$ and $s_2\DMTSmust{ a_2}
  t_2$.  By construction, there is an $M_2\in \Tran_2( s_2)$ with $(
  a_2, t_2)\in M_2$, and for all $( a_2', t_2')\in M_2$, $s_2\DMTSmust{
    a_2'} t_2'$.

  We show that $R^2$ is a witness for $\mdl( \mcalI, \mcalI_2)$; clearly,
  $R^2$ is initialized.  Let $\alpha\in \LL$ with $\alpha\ne \top_\LL$
  and $( s, s_2)\in R^2_\alpha$, then there is $s_1\in S_1$ with $( s,
  s_1)\in R^1$ and $( s_1, s_2)\in R_\alpha$.  We also have $M_1\in
  \Tran_1( s_1)$ such that
  \begin{itemize}
  \item for all $s\DMTSmust{ a} t$, there is $( a, t_1)\in M_1$ with $( t,
    t_1)\in R^1$,
  \item for all $( a, t_1)\in M_1$, there is $s\DMTSmust{ a} t$ with $( t,
    t_1)\in R^1$
  \end{itemize}
  and thus $M_2\in \Tran_2( s_2)$ satisfying the conditions in
  Definition~\ref{de:reffam}.

  Let $s\DMTSmust{ a} t$, then there is $( a, t_1)\in M_1$ with $( t,
  t_1)\in R^1$, hence also $( a_2, t_2)\in M_2$ and $\beta\in \LL$ with
  $( t_1, t_2)\in R_\beta$ and $F( a, a_2, \beta)\sqsubseteq_\LL
  \alpha$.  But then $( t, t_2)\in R^2_\beta$, and $s_2\DMTSmust{ a_2} t_2$
  by construction.

  Let $s_2\DMTSmust{ a_2} t_2$.  By construction, there is an $M_2\in
  \Tran_2( s_2)$ with $( a_2, t_2)\in M_2$.  This implies that there is
  $M_1\in \Tran_1( s_1)$, $\beta\in \LL$ and $( a_1, t_1)\in M_1$ with
  $( t_1, t_2)\in R_\beta$ and $F( a_1, a_2, \beta)\sqsubseteq \alpha$.
  But then there is also $s\DMTSmust{ a_1} t$ with $( t, t_1)\in R^1$, hence
  $( t, t_2)\in R^2_\beta$. \qed
\end{proof}

\subsection{Disjunction and conjunction}

In order to generalize the properties of Theorem~\ref{soco.th:condis} to our
quantitative setting, we introduce a notion of relaxed implementation
semantics:

\begin{definition}
  The \emph{$\alpha$-relaxed implementation semantics} of $\mcalS$, for a
  specification $\mcalS$ and $\alpha\in \LL$, is
  \begin{equation*}
    \sem{ \mcalS}^\alpha=\{
    \mcalI\text{ implementation}\mid \mdl( \mcalI, \mcalS)\sqsubseteq
    \alpha\}\,.
  \end{equation*}
\end{definition}

Hence, $\sem \mcalS^\alpha$ comprises all labeled transition systems which
are implementations of $\mcalS$ \emph{up to $\alpha$}.  Note that by
Proposition~\ref{pr:dist-prop} and for $F$ recursively separating,
$\sem{ \mcalS}^{ \bot_\LL}= \sem \mcalS$.

\begin{theorem}
  \label{th:condis-q}
  For all specifications $\mcalS_1$, $\mcalS_2$, $\mcalS_3$ and $\alpha\in \LL$,
  \begin{itemize}
  \item $\mdl( \mcalS_1\lor \mcalS_2, \mcalS_3)= \max( \mdl( \mcalS_1, \mcalS_3), \mdl(
    \mcalS_2, \mcalS_3))$,
  \item $\mdl( \mcalS_1, \mcalS_2\land \mcalS_3)\sqsupseteq_\LL \max( \mdl(
    \mcalS_1, \mcalS_2), \mdl( \mcalS_1, \mcalS_3))$,
  \item $\sem{ \mcalS_1\lor \mcalS_2}^\alpha= \sem{ \mcalS_1}^\alpha\cup \sem{
      \mcalS_2}^\alpha$, and
  \item $\sem{ \mcalS_1\land \mcalS_2}^\alpha\subseteq \sem{ \mcalS_1}^\alpha\cap
    \sem{ \mcalS_2}^\alpha$.
  \end{itemize}
\end{theorem}

\begin{proof}
  We show the proof for DMTS.

  The proof that $\mdl( \mcalD_1\lor \mcalD_2, \mcalD_3)= \max( \mdl( \mcalD_1,
  \mcalD_3),$ $\mdl( \mcalD_2, \mcalD_3))$ is trivial: any refinement family
  witnessing $\mdl( \mcalD_1\lor \mcalD_2, \mcalD_3)$ splits into two families
  witnessing $\mdl( \mcalD_1, \mcalD_3)$ and $\mdl( \mcalD_2, \mcalD_3)$ and vice
  versa.

  To show that $\mdl( \mcalD_1, \mcalD_2\land \mcalD_3)\sqsupseteq_\LL \max(
  \mdl( \mcalD_1, \mcalD_2),$ $\mdl( \mcalD_1, \mcalD_3))$, let $R=\{
  R_\alpha\subseteq S_1\times( S_2\times S_3)\mid \alpha\in \LL\}$ be a
  witness for $\mdl( \mcalD_1, \mcalD_2\land \mcalD_3)$ and define $R^2=\{
  R^2_\alpha\subseteq S_1\times S_2\mid \alpha\in \LL\}$ by
  $R^2_\alpha=\{( s_1, s_2)\mid \exists s_3\in S_3:( s_1,( s_2, s_3))\in
  R_\alpha\}$ for all $\alpha\in \LL$.

  Let $s_1^0\in S_1^0$, then we have $( s_2^0, s_3^0)\in S_2^0\times
  S_3^0$ so that $( s_1^0,( s_2^0, s_3^0))\in R_{ \mdl( \mcalD_1,
    \mcalD_2\land \mcalD_3)}$, hence also $( s_1^0, s_2^0)\in R^2_{ \mdl(
    \mcalD_1, \mcalD_2\land \mcalD_3)}$.

  Let $\alpha\in \LL$ and $( s_1, s_2)\in R^2_\alpha$, then we have
  $s_3\in S_3$ for which $( s_1,( s_2, s_3))\in R_\alpha$.  Assume first
  that $s_1\DMTSmay{ a_1} t_1$, then there is $( s_2, s_3)\DMTSmay a( t_2, t_3)$
  and $\beta\in \LL$ such that $F( a_1, a, \beta)\sqsubseteq_\LL \alpha$
  and $( t_1,( t_2, t_3))\in R_\beta$, hence $( t_1, t_2)\in
  R^2_\beta$.  By construction of $\mcalD_2\land \mcalD_3$, there are
  $s_2\DMTSmay{ a_2} t_2$ and $s_3\DMTSmay{ a_3} t_3$ such that $a= a_2\oland
  a_3$, but then by anti-monotonicity, $F( a_1, a_2,
  \beta)\sqsubseteq_\LL F( a_1, a, \beta)\sqsubseteq \alpha$.

  Now assume $s_2\DMTSmust{} N_2$, then, by
  construction,
  $( s_2, s_3)\DMTSmust{} N=\{( a_2\oland a_3,( t_2, t_3))\mid( a_2,
  t_2)\in N_2,$ $s_3\DMTSmay{ a_3}_3 t_3\}$.  Hence we have $s_1\DMTSmust{}_1
  N_1$ such that $\forall( a_1, t_1)\in N_1: \exists( a,( t_2,
  t_3))\in N, \beta\in \LL: F( a_1, a, \beta)\sqsubseteq_\LL \alpha,(
  t_1,( t_2, t_3))\in R_\beta$.

  Let $( a_1, t_1)\in N_1$, then we have $( a,( t_2, t_3))\in N$ and
  $\beta\in \LL$ for which $F( a_1, a, \beta)\sqsubseteq_\LL \alpha$ and
  $( t_1,( t_2, t_3))\in R_\beta$, hence $( t_1, t_2)\in R^2_\beta$.  By
  construction of $N$, this implies that there are $( a_2, t_2)\in N_2$
  and $s_3\DMTSmay{ a_3}_3 t_3$ such that $a= a_2\oland a_3$, but then by
  anti-monotonicity, $F( a_1, a_2, \beta)\sqsubseteq_\LL F( a_1, a,
  \beta)\sqsubseteq \alpha$.

  We have shown that
  $\mdl( \mcalD_1, \mcalD_2\land \mcalD_3)\sqsubseteq_\LL \mdl(
  \mcalD_1, \mcalD_2)$.  The proof of
  $\mdl( \mcalD_1, \mcalD_2\land \mcalD_3)\sqsubseteq_\LL \mdl(
  \mcalD_1, \mcalD_3)$ is entirely analogous.

  The inclusion $\sem{ \mcalD_1\land \mcalD_2}^\alpha\subseteq \sem{
    \mcalD_1}^\alpha\cap \sem{ \mcalD_2}^\alpha$ is clear now: If $\mcalI\in
  \sem{ \mcalD_1\land \mcalD_2}^\alpha$, \ie~$\mdl( \mcalI, \mcalD_1\land
  \mcalD_2)\sqsubseteq_\LL \alpha$, then also $\mdl( \mcalI,
  \mcalD_1)\sqsubseteq_\LL \alpha$ and $\mdl( \mcalI, \mcalD_2)\sqsubseteq_\LL
  \alpha$, thus $\mcalI\in \sem{ \mcalD_1}^\alpha\cap \sem{ \mcalD_2}^\alpha$.

  To show that $\sem{ \mcalD_1\lor \mcalD_2}^\alpha= \sem{ \mcalD_1}^\alpha\cup
  \sem{ \mcalD_2}^\alpha$, one notices, like in the proof of
  Theorem~\ref{soco.th:condis}, that for any LTS $\mcalI$, any refinement family
  witnessing $\mdl( \mcalI, \mcalD_1)$ or $\mdl( \mcalI, \mcalD_2)$ is also a
  witness for $\mdl( \mcalI, \mcalD_1\lor \mcalD_2)$ and vice versa. \qed
\end{proof}

The below example shows why the inclusions above cannot be replaced by
equalities.  To sum up, disjunction is quantitatively sound and
complete, whereas conjunction is only quantitatively sound.

\begin{figure}
  \centering
  \begin{tikzpicture}[->, >=stealth', font=\footnotesize,
    state/.style={shape=circle, draw, initial text=,inner
      sep=.5mm,minimum size=2mm}, yscale=1, xscale=1]
    \begin{scope}
      \node at (-1,0) {$\mcalI$};
      \node[state, initial] (s) at (0,0) {};
      \node[state] (t) at (2,0) {};
      \path (s) edge node [below] {$a, 2$} (t);
    \end{scope}
    \begin{scope}[xshift=15em]
      \node at (-1,0) {$\mcalD_1$};
      \node[state, initial] (s) at (0,0) {};
      \node[state] (t) at (2,0) {};
      \path (s) edge [densely dashed] node [below] {$a,[ 0, 1]$} (t);
    \end{scope}
    \begin{scope}[yshift=-10ex]
      \node at (-1,0) {$\mcalD_2$};
      \node[state, initial] (s) at (0,0) {};
      \node[state] (t) at (2,0) {};
      \path (s) edge [densely dashed] node [below] {$a,[ 3, 4]$} (t);
    \end{scope}
    \begin{scope}[yshift=-10ex, xshift=15em]
      \node at (-1,0) {$\mcalD_1\land \mcalD_2$};
      \node[state, initial] (s) at (.5,0) {};
    \end{scope}
  \end{tikzpicture}
  \caption{%
    \label{fi:ex-conj}
    LTS $\mcalI$ together with DMTS $\mcalD_1$, $\mcalD_2$ and their conjunction.
    For the point-wise or discounting distances, $\md( \mcalI, \mcalD_1)= \md(
    \mcalI, \mcalD_2)= 1$, but $\md( \mcalI, \mcalD_1\land \mcalD_2)= \infty$}
\end{figure}

\begin{example}
  \label{ex:no-qconj}
  For the point-wise or discounting distances, the DMTS in
  Figure~\ref{fi:ex-conj} are such that $\md( \mcalI, \mcalD_1)= 1$ and $\md(
  \mcalI, \mcalD_2)= 1$, but $\md( \mcalI, \mcalD_1\land \mcalD_2)= \infty$.  Hence
  $\md( \mcalI, \mcalD_1\land \mcalD_2)\ne \max( \md( \mcalI, \mcalD_1), \md( \mcalI,
  \mcalD_2))$, and $\mcalI\in \sem{ \mcalD_1}^1\cap \sem{ \mcalD_2}^1$, but
  $\mcalI\notin \sem{ \mcalD_1\land \mcalD_2}^1$. \qed
\end{example}

\subsection{Structural composition and quotient}

We proceed to devise a quantitative generalization of the properties of
structural composition and quotient exposed in Section~\ref{se:specth}.
To this end, we need to use a \emph{uniform composition bound} on
labels:

Let $P: \LL\times \LL\to \LL$ be a function which is monotone in both
coordinates, has $P( \alpha, \bot_\LL)= P( \bot_\LL, \alpha)= \alpha$
and $P( \alpha, \top_\LL)= P( \top_\LL, \alpha)= \top_\LL$ for all
$\alpha\in \LL$.  We require that for all $a_1, b_1, a_2, b_2\in
\Sigma$ and $\alpha, \beta\in \LL$ with $F( a_1, a_2, \alpha)\ne
\top_\LL$ and $F( b_1, b_2, \beta)\ne \top_\LL$, $a_1\obar b_1$ is
defined iff $a_2\obar b_2$ is, and if both are defined, then
\begin{equation}
  \label{eq:synchbound}
  F( a_1\obar b_1, a_2\obar b_2, P( \alpha, \beta))
  \sqsubseteq_\LL P( F( a_1, a_2, \alpha), F( b_1, b_2, \beta))\,.
\end{equation}

Note that \eqref{eq:synchbound}~implies that
\begin{equation}
  \label{eq:synchbound-s}
  \tdl( a_1\obar a_2,
  b_1\obar b_2)\sqsubseteq_\LL P( \tdl( a_1, b_1), \tdl( a_2, b_2))\,.
\end{equation}
Hence $P$ provides a \emph{uniform bound} on distances between
synchronized labels, and \eqref{eq:synchbound} extends this property
so that it holds recursively.  Also, this is a generalization of the
condition that we imposed on $\obar$ in Section~\ref{se:structlabels};
it is shown in~\cite[p.~18]{DBLP:journals/acta/FahrenbergL14} that it
holds for all common label synchronizations.

Remark that $P$ can be understood as a (generalized) \emph{modulus of
  continuity}~\cite{encmath/continuity} for the partial function $f:
\Sigma\times \Sigma\parto \Sigma$ given by label synchronization $f(
a, b)= a\obar b$: with that notation, \eqref{eq:synchbound-s} asserts
that the distance from $f( a_1, a_2)$ to $f( b_1, b_2)$ is bounded by
$P$ applied to the distance from $( a_1, a_2)$ to $( b_1, b_2)$.

The following theorems show that composition is uniformly continuous
(\ie~a quantitative generalization of independent implementability;
Corollary~\ref{co:indimp}) and that quotient preserves and reflects
refinement distance (a quantitative generalization of
Theorem~\ref{th:quotient-bool}).

\begin{theorem}[Independent implementability]
  \label{th:indimp-q}
  For all specifications $\mcalS_1$, $\mcalS_2$, $\mcalS_3$,
  $\mcalS_4$,
  $\mdl( \mcalS_1\| \mcalS_2, \mcalS_3\| \mcalS_4)\sqsubseteq_\LL P(
  \mdl( \mcalS_1, \mcalS_3), \mdl( \mcalS_2, \mcalS_4))$.
\end{theorem}

\begin{proof}
  We show the proof for \NAA.  For $i= 1, 2, 3, 4$, let $\mcalA_i=( S_i,
  S_i^0, \Tran_i)$.  Let $R^1=\{ R^1_\alpha\subseteq S_1\times S_3\mid
  \alpha\in \LL\}$, $R^2=\{ R^2_\alpha\subseteq S_2\times S_4\mid
  \alpha\in \LL\}$ be refinement families such that $\forall s_1^0\in
  S_1^0: \exists s_3^0\in S_3^0:( s_1^0, s_3^0)\in R^1_{ \mdl( \mcalA_1,
    \mcalA_3)}$ and $\forall s_2^0\in S_2^0: \exists s_4^0\in S_4^0:(
  s_2^0, s_4^0)\in R^2_{ \mdl( \mcalA_2, \mcalA_4)}$.  Define $R=\{
  R_\alpha\subseteq( S_1\times S_2)\times( S_3\times S_4)\mid \alpha\in
  \L\}$ by
  \begin{multline*}
    R_\alpha= \big\{ ((s_1, s_2),( s_3, s_4))\bigmid \exists \alpha_1,
    \alpha_2\in \LL: \\
    ( s_1, s_3)\in R^1_{ \alpha_1},( s_2, s_4)\in R^2_{ \alpha_2}, P(
    \alpha_1, \alpha_2)\sqsubseteq_\LL \alpha\big\}\,,
  \end{multline*}
  then it is clear that $\forall( s_1^0, s_2^0)\in S_1^0\times S_2^0:
  \exists( s_3^0, s_4^0)\in S_3^0\times S_4^0:(( s_1^0, s_2^0),( s_3^0,
  s_4^0))\in R_{ P( \mdl( \mcalA_1, \mcalA_3), \mdl( \mcalA_2, \mcalA_4))}$.  We
  show that $R$ is a refinement family from $\mcalA_1\| \mcalA_2$ to $\mcalA_3\|
  \mcalA_4$.

  Let $\alpha\in \LL$ and $(( s_1, s_2),( s_3, s_4))\in R_\alpha$, then
  we have $\alpha_1, \alpha_2\in \LL$ with $( s_1, s_3)\in R^1_{
    \alpha_1}$, $( s_2, s_4)\in R^2_{ \alpha_2}$ and $P( \alpha_1,
  \alpha_2)\sqsubseteq_\LL \alpha$.  Let $M_{ 12}\in \Tran(( s_1,
  s_2))$, then there must be $M_1\in \Tran_1( s_1)$, $M_2\in \Tran_2(
  s_2)$ for which $M_{ 12}= M_1\obar M_2$.  Thus we also have $M_3\in
  \Tran_3( s_3)$ and $M_4\in \Tran_4( s_4)$ such that
  \begin{align}
    \notag & \forall( a_1, t_1)\in M_1: \exists( a_3, t_3)\in M_3,
    \beta_1\in \LL: \\
    \label{eq:comp.1-3}
    &\hspace*{10em} ( t_1, t_3)\in R^1_{ \beta_1}, F( a_1, a_3,
    \beta_1)\sqsubseteq_\LL \alpha_1\,, \\
    \notag & \forall( a_3, t_3)\in M_3: \exists( a_1, t_1)\in M_1,
    \beta_1\in \LL: \\
    \label{eq:comp.3-1}
    &\hspace*{10em} ( t_1, t_3)\in R^1_{ \beta_1}, F( a_1, a_3,
    \beta_1)\sqsubseteq_\LL \alpha_1\,, \\
    \notag & \forall( a_2, t_2)\in M_2: \exists( a_4, t_4)\in M_4,
    \beta_2\in \LL: \\
    \label{eq:comp.2-4}
    &\hspace*{10em} ( t_2, t_4)\in R^2_{ \beta_2}, F( a_2, a_4,
    \beta_2)\sqsubseteq_\LL \alpha_2\,, \\
    \notag & \forall( a_4, t_4)\in M_4: \exists( a_2, t_2)\in M_2,
    \beta_2\in \LL: \\
    \label{eq:comp.4-2}
    &\hspace*{10em} ( t_2, t_4)\in R^2_{ \beta_2}, F( a_2, a_4,
    \beta_2)\sqsubseteq_\LL \alpha_2\,.
  \end{align}

  Let $M_{ 34}= M_3\obar M_4$, then $M_{ 34}\in \Tran(( s_3, s_4))$.
  Let $( a_{ 12},( t_1, t_2))\in M_{ 12}$, then there are $( a_1,
  t_1)\in M_1$ and $( a_2, t_2)\in M_2$ for which $a_{ 12}= a_1\obar
  a_2$.  Using~\eqref{eq:comp.1-3} and~\eqref{eq:comp.2-4}, we get $(
  a_3, t_3)\in M_3$, $( a_4, t_4)\in M_4$ and $\beta_1, \beta_2\in \LL$
  such that $( t_1, t_3)\in R^1_{ \beta_1}$, $( t_2, t_4)\in R^2_{
    \beta_2}$, $F( a_1, a_3, \beta_1)\sqsubseteq_\LL \alpha_1$, and $F(
  a_2, a_4, \beta_2)\sqsubseteq_\LL \alpha_2$.

  Let $a_{ 34}= a_3\obar a_4$ and $\beta= P( \beta_1, \beta_2)$, then we
  have $( a_{ 34},( t_3, t_4))\in M_{ 34}$.  Also, $( t_1, t_3)\in R^1_{
    \beta_1}$ and $( t_2, t_4)\in R^2_{ \beta_2}$ imply that $(( t_1,
  t_2),( t_3, t_4))\in R_\beta$, and
  \begin{align*}
    F( a_{ 12}, a_{ 34}, \beta) &= F( a_1\obar a_2, a_3\obar a_4, P(
    \beta_1, \beta_2)) \\
    &\sqsubseteq P( F( a_1, a_3, \beta_1), F( a_2, a_4,
    \beta_2)) \\
    &\sqsubseteq_\LL P( \alpha_1, \alpha_2)\sqsubseteq_\LL \alpha\,.
  \end{align*}
  We have shown that for all $( a_{ 12},( t_1, t_2))\in M_{ 12}$, there
  exists $( a_{ 34},( t_3, t_4))\in M_{ 34}$ and $\beta\in \LL$ such
  that $(( t_1, t_2),$ $( t_3, t_4))\in R_\beta$ and $F( a_{ 12}, a_{
    34}, \beta)\sqsubseteq_\LL \alpha$.  To show the reverse property,
  starting from an element $( a_{ 34},( t_3, t_4))\in M_{ 34}$, we can
  proceed entirely analogous,
  using~\eqref{eq:comp.3-1}
  and~\eqref{eq:comp.4-2}.  \qed
\end{proof}

\begin{theorem}
  \label{th:quot-q}
  For all specifications $\mcalS_1$, $\mcalS_2$, $\mcalS_3$,
  $\mdl( \mcalS_1\| \mcalS_2, \mcalS_3)= \mdl( \mcalS_2, \mcalS_3/
  \mcalS_1)$.
\end{theorem}

\begin{proof}
  We show the proof for \NAA.  Let $\mcalA_i=( S_i, S^0_i,$ $\Tran_i)$
  for $i=1,\dotsc,3$; we prove that
  $\mdl( \mcalA_1\| \mcalA_2, \mcalA_3)= \mdl( \mcalA_2, \mcalA_3/
  \mcalA_1)$.

  We assume that the elements of $\Tran_1( s_1)$ are pairwise disjoint
  for each $s_1\in S_1$; this can be achieved by, if necessary,
  splitting states.

  Define $R=\{ R_\alpha\subseteq S_1\times S_2\times S_3\mid \alpha\in
  \LL\}$ by $R_\alpha=\{( s_1\| s_2, s_3)\mid \mdl( s_2, s_3/
  s_1)\sqsubseteq_\LL \alpha\}$.  (We again abuse notation and write
  $( s_1\| s_2, s_3)$ instead of $( s_1, s_2, s_3)$.)  We show that
  $R$ is a witness for $\mdl( \mcalA_1\| \mcalA_2, \mcalA_3)$.

  Let $s_1^0\| s_2^0\in S_1^0\times S_2^0$, then there is $s_3^0/
  s_1^0\in s^0$ for which it holds that $\mdl( s_2^0, s_3^0/
  s_1^0)\sqsubseteq_\LL \mdl( \mcalA_2, \mcalA_3/ \mcalA_1)$, hence $( s_1^0\|
  s_1^0, s_3^0)\in R_{ \mdl( \mcalA_2, \mcalA_3/ \mcalA_1)}$.

  Let $\alpha\in \LL\setminus\{ \top_\LL\}$, $( s_1\| s_2, s_3)\in
  R_\alpha$ and $M_\|\in \Tran_\|( s_1\| s_2)$.  Then $M_\|= M_1\| M_2$
  with $M_1\in \Tran_1( s_1)$ and $M_2\in \Tran_2( s_2)$.  As $\mdl(
  s_2, s_3/ s_1)\sqsubseteq_\LL \alpha$, we can pair $M_2$ with an
  $M_/\in \Tran_/( s_3/ s_1)$, \ie~such that the conditions in
  Definition~\ref{de:reffam} are satisfied.

  Let $M_3= M_/\triangleright M_1$.  We show that the conditions in
  Definition~\ref{de:reffam} are satisfied for the pair $M_\|, M_3$:
  \begin{itemize}
  \item Let $( a, t_1\| t_2)\in M_\|$, then there are $a_1, a_2\in
    \Sigma$ with $a= a_1\obar a_2$ and $( a_1, t_1)\in M_1$, $( a_2,
    t_2)\in M_2$.  Hence there is $( a_2', t)\in M_/$ and $\beta\in
    \LL$ such that $F( a_2, a_2', \beta)\sqsubseteq_\LL \alpha$ and
    $\mdl( t_2, t)\sqsubseteq_\LL \beta$.

    Note that $a_3= a_1\obar a_2'$ is defined and $F( a, a_3,
    \beta)\sqsubseteq \alpha$.  Write $t=\{ t_3^1/ t_1^1,\dots,
    t_3^n/ t_1^n\}$.  By construction, there is an index $i$ for which
    $t_1^i= t_1$, hence $( a_3, t_3^i)\in M_3$.  Also, $t\supseteq\{
    t_3^i/ t_1^i\}$, hence $\mdl( t_2, t_3^i/ t_1^i)\sqsubseteq
    \beta$ and consequently $( t_1\| t_2, t_3)\in R_\beta$.

  \item Let $( a_3, t_3)\in M_3$, then there are $( a_2', t)\in M_/$
    and $( a_1, t_1)\in M_1$ such that $a_3= a_1\obar a_2'$ and $t_3/
    t_1\in t$.  Hence there are $( a_2, t_2)\in M_2$ and $\beta\in \LL$
    for which $F( a_2, a_2', \beta)\sqsubseteq_\LL \alpha$ and $\mdl(
    t_2, t)\sqsubseteq_\LL \beta$.  Note that $a= a_1\obar a_2$ is
    defined and $F( a, a_3, \beta)\sqsubseteq_\LL \alpha$.  Thus $( a,
    t_1\| t_2)\in M$, and by $t\supseteq\{ t_3/ t_1\}$, $\mdl( t_2,
    t_3/ t_1)\sqsubseteq \beta$.
  \end{itemize}

  Assume, for the other direction of the proof, that $\mcalA_1\| \mcalA_2\mr
  \mcalA_3$.  Define $R=\{ R_\alpha\subseteq S_2\times 2^{ S_3\times
    S_1}\mid \alpha\in \LL\}$ by
  \begin{equation*}
    R_\alpha= \big\{ ( s_2,\{ s_3^1/ s_1^1,\dotsc, s_3^n/ s_1^n\})\bigmid
    \forall i= 1,\dotsc, n:
    \mdl( s_1^i\| s_2, s_3^i)\sqsubseteq_\LL \alpha\big\}\,;
  \end{equation*}
  we show that $R$ is a witness for $\mdl( \mcalA_2, \mcalA_3/ \mcalA_1)$.

  Let $s_2^0\in S_2^0$.  We know that for every $s_1^0\in S_1^0$, there
  exists $\sigma( s_1^0)\in S_3^0$ such that $\mdl( s_1^0\| s_2^0,
  s_3^0)\sqsubseteq_\LL \mdl( \mcalA_1\| \mcalA_2, \mcalA_3)$.  By
  $s^0\supseteq\{ \sigma( s_1^0)/ s_1^0\mid s_1^0\in S_1^0\}$, we see
  that $( s_2^0, s^0)\in R_{ \mdl( \mcalA_1\| \mcalA_2, \mcalA_3)}$.

  Let $\alpha\in \LL\setminus\{ \top_\LL\}$ and $( s_2, s)\in R_\alpha$,
  with $s=\{ s_3^1/ s_1^1,\dotsc, s_3^n/ s_1^n\}$, and $M_2\in
  \Tran_2( s_2)$.

  For every $i= 1,\dotsc, n$, let us write $\Tran_1( s_1^i)=\{ M_1^{ i,
    1},\dotsc, M_1^{ i, m_i}\}$.  By assumption, $M_1^{ i, j_1}\cap
  M_1^{ i, j_2}= \emptyset$ for $j_1\ne j_2$, hence every $( a_1,
  t_1)\in \in \Tran_1( s_1^i)$ is contained in a unique $M_1^{ i,
    \delta_i( a_1, t_1)}\in \Tran_1( s_1^i)$.

  For every $j= 1,\dotsc, m_i$, let $M^{ i, j}= M_1^{ i, j}\| M_2\in
  \Tran_\|( s_1^i\| s_2)$.  From $\mdl( s_1^i\| s_2, s_3^i)\sqsubseteq_\LL
  \alpha$ we have $M_3^{ i, j}\in \Tran_3( s_3^i)$ such that the
  conditions in Definition~\ref{de:reffam} hold for the pair $M^{ i, j},
  M_3^{ i, j}$.

  Now define
  \begin{multline}
    \label{eq:quotproof2.M}
    M= \big\{( a_2, t)\bigmid \exists( a_2, t_2)\in M_2: \forall t_3/
    t_1\in t: \exists i, a_1, a_3, \beta: \\
    ( a_1, t_1)\in \in \Tran_1(
    s_1^i), ( a_3, t_3)\in M_3^{ i, \delta_i( a_1, t_1)}, \\
    F( a_1\obar a_2, a_3, \beta)\sqsubseteq_\LL \alpha, \mdl( t_1\|
    t_2, t_3)\sqsubseteq_\LL \beta\big\}\,.
  \end{multline}
  We need to show that $M\in \Tran_/( s)$.

  Let $i\in\{ 1,\dots, n\}$ and $M_1^{ i, j}\in \Tran_1( s_1^i)$; we
  claim that $M\triangleright M_1^{ i, j}\labpre_R M_3^{ i, j}$.  Let $(
  a_3, t_3)\in M\triangleright M_1^{ i, j}$, then $a_3= a_1\obar a_2$
  for some $a_1, a_2$ such that $t_3/ t_1\in t$, $( a_1, t_1)\in M_1^{
    i, j}$ and $( a_2, t)\in M$.  By disjointness, $j= \delta_i( a_1,
  t_1)$, hence by definition of $M$, $( a_3, t_3)\in M_3^{ i, j}$ as was
  to be shown.

  For the reverse inclusion, let $( a_3, t_3)\in M_3^{ i, j}$.  By
  definition of $M^{ i, j}$, there are $( a_1, t_1)\in M_1^{ i, j}$, $(
  a_2, t_2)\in M_2$ and $\beta\in \LL$ for which $F( a_1\obar a_2, a_3,
  \beta)\sqsubseteq_\LL \alpha$ and $\mdl( t_1\| t_2,
  t_3)\sqsubseteq_\LL \beta$.  Thus $j= \delta_i( a_1, t_1)$, so that
  there must be $( a_2, t)\in M$ for which $t_3/ t_1\in t$, but then
  also $( a_1\obar a_2, t_3)\in M\triangleright M_1^{ i, j}$.

  We show that the pair $M_2, M$ satisfies the conditions of
  Definition~\ref{de:reffam}.
  \begin{itemize}
  \item Let $( a_2, t_2)\in M_2$.  For every $i= 1,\dotsc, n$ and every
    $( a_1, t_1)\in \in \Tran_1( t_1^i)$, we can use
    Definition~\ref{de:reffam} applied to the pair $M_1^{ i, \delta_i(
      a_1, t_1)}\| M_2, M_3^{ i, \delta_i( a_1, t_1)}$ to choose an
    element $( \eta_i( a_1, t_1), \tau_i( a_1, t_1))\in M_3^{ i,
      \delta_i( a_1, t_1)}$ and $\beta_i( a_1, t_1)\in \LL$ for which
    $\mdl( t_1\| t_2, \tau_i(a_1, t_1))\sqsubseteq_\LL \beta_i( a_1,
    t_1)$ and $F( a_1\obar a_2, \eta_i( a_1, t_1),$ $\beta_i( a_1,
    t_1))\sqsubseteq_\LL \alpha$.  Let $t=\{ \tau_i( a_1, t_1)/
    t_1\mid i= 1,\dotsc, n,( a_1, t_1)\in \in \Tran_1( t_1^i)\}$, then
    $( a_2, t)\in M$ and $( t_2, t)\in R_\beta$.
  \item Let $( a_2, t)\in M$, then we have $( a_2, t_2)\in M_2$
    satisfying the conditions in~\eqref{eq:quotproof2.M}.  Hence for all
    $t_3/ t_1\in t$, there are $i$, $a_1$, $a_3$, and $\beta( t_3/
    t_1)$ such that $( a_3, t_3)\in M_3^{ i, \delta_i( a_1, t_1)}$, $F(
    a_1\obar a_2, a_3, \beta( t_3/ t_1))\sqsubseteq_\LL \alpha$,
    and
    $\mdl( t_1\| t_2, t_3)\sqsubseteq_\LL \beta( t_3/ t_1)$.  Let
    $\beta= \sup\{ \beta( t_3/ t_1)\mid t_3/ t_1\in t\}$, then
    $\mdl( t_1\| t_2, t_3)\sqsubseteq_\LL \beta$ for all $t_3/ t_1\in
    t$, hence $( t_2, t)\in R_\beta$. \qed
  \end{itemize}
\end{proof}

\section{Conclusion}

We have presented a framework for compositional and iterative design and
verification of systems which supports quantities and system and action
refinement.  Moreover, it is robust, in that it uses distances to
measure quantitative refinement and the operations preserve
distances.

The framework is very general.  It can be applied to a large variety of
quantities (energy, time, resource consumption etc.) and implement the
robustness notions associated with them.  It is also agnostic with
respect to the type of specifications used, as it applies equally to
behavioral and logical specifications.  This means that logical and
behavioral quantitative specifications can be freely combined in
quantitative system development.

\chapter{References}

\bibliography{bib}

\end{document}